\documentclass[a4paper,11pt,twoside]{book}

\usepackage[T1]{fontenc}		% Codifica dei font; T1 codifica di output dell'italiano e di lingue occidentali
\usepackage[utf8]{inputenc}		% Codifica degli input; abilita l'utilizzo di caratteri accentati
\usepackage[english]{babel}		% Set the languages; the last one is the main language
\usepackage{geometry}			% Allow to change the margin
\usepackage{setspace}			% Allow to change the interline with \begin{onehalfspaceng} ecc.
\usepackage{fancyhdr}			% Fancy style for page layout
\usepackage{afterpage} 			% Allow to load blank pages with the command '\afterpage{\null\thispagestyle{empty}\clearpage}' 
\usepackage{url}

\usepackage{color}			% Allow to use colors
\usepackage{xcolor}			% manage colors
\usepackage{enumerate}			% Numbered lists
\usepackage{enumitem}			% Allow to manage the style of the list
\usepackage{graphicx}			% Images
\usepackage{epstopdf}			% allows to convert eps images to pdf for use with pdflatex
\usepackage{pstool}			% Allows to use psfrag with pdflatex (no latex compiler needed)
\usepackage{psfrag}
\usepackage{subfig}		% manage subfigures
\usepackage{array}			% Tables
\usepackage{tabularx}			% allow to set the width of the whole table
\usepackage[bottom]{footmisc}		% To attach footnote at the end of the page
\usepackage{booktabs}			% is a must for professional-looking layout
\usepackage{longtable}			% is very popular for multi-page tables
\usepackage{caption}			% Captions
\usepackage{float}			% Manage float objects (image, table, ecc.)
\usepackage{subfloat}			% permette la sub numerazione degli oggetti mobili (immagini, tabelle, ecc.)
\usepackage{rotating}			% permette di ruotare immagini, tabelle, ecc. di 90° o di 270°
\usepackage{rotfloat}			% costruisce un ponte tra i pacchetti float e rotating
\usepackage{amsmath} 			% Equations
\usepackage{mathrsfs}
\usepackage{amssymb}    % Mathematical Symbols
\usepackage{amstext}
\usepackage{amsthm}
\usepackage{mathtools}
\usepackage{tikz-cd}
\usepackage{cancel}                     % Simplification Cancellation
\usepackage{mleftright}
\usepackage{listings}                   % Codes
\usepackage[intoc]{nomencl}
\usepackage[acronym,toc,nomain,nopostdot,nonumberlist]{glossaries}
\usepackage[square, numbers, comma, sort&compress]{natbib}
\usepackage{verbatim}                   % Needed for the "comment" environment to make LaTeX comments
\usepackage{wrapfig}                 % Per inserire le figure contornate da testo
\usepackage{blindtext}
\usepackage{listings}
\newtheorem{theorem}{Theorem}[section]

\newtheorem*{remark*}{Remark}
\newtheorem*{example}{Example}

\newtheorem{theorem*}{Theorem}

\newtheorem{corollary}[theorem]{Corollary}
\newtheorem{lemma}[theorem]{Lemma}
\newtheorem{proposition}[theorem]{Proposition}

\theoremstyle{definition}
\newtheorem{definition}[theorem]{Definition}

\fancypagestyle{plain}{%                                        % modifica dello stile predefinito plain
                        \fancyhf{}                              % cancella tutti i campi di  intestazione e pie di pagina
                        \fancyfoot[R]{\thepage}                 % mette al centro il numero di pagina
                        \renewcommand{\headrulewidth}{0pt}      % spessore della linea di separazione in alto (0 per eliminare la linea)
                        \renewcommand{\footrulewidth}{0.4pt}    % spessore della linea di separazione in basso (0 per eliminare la linea)
                        }                                       % modifica dello stile predefinito plain
\definecolor{sapienza}{RGB}{130,36,51} % example \definecolor{name}{model}{color-spec}
\definecolor{cust1}{RGB}{85,85,85}
\definecolor{cust2}{RGB}{212,212,212}
\makenomenclature % Generate the nomenclature
\makeglossaries % Generate the glossary
\usepackage{calligra}
\usepackage{pdfpages}

\usepackage[pdftex]{hyperref}			% Crea collegamenti ipertestuali rendendo cliccabili i riferimenti
\usepackage{cleveref}                   %Pacchetto per \Cref
\pdfoutput=1
\begin{document}
%   FRONT MATTER                                            %
\frontmatter 	   % Begin Roman style (i, ii, iii, iv...) page numbering
\pagestyle{empty}  % No headers or footers for the following pages

%   TITLE PAGE

\IfFileExists{\jobname-frn.pdf}{}{%
\immediate\write18{pdflatex \jobname-frn}}
%\lstinputlisting{output-frn.log}

%\afterpage{\null\thispagestyle{empty}\clearpage}
%\clearpage
\includepdf[lastpage=1]{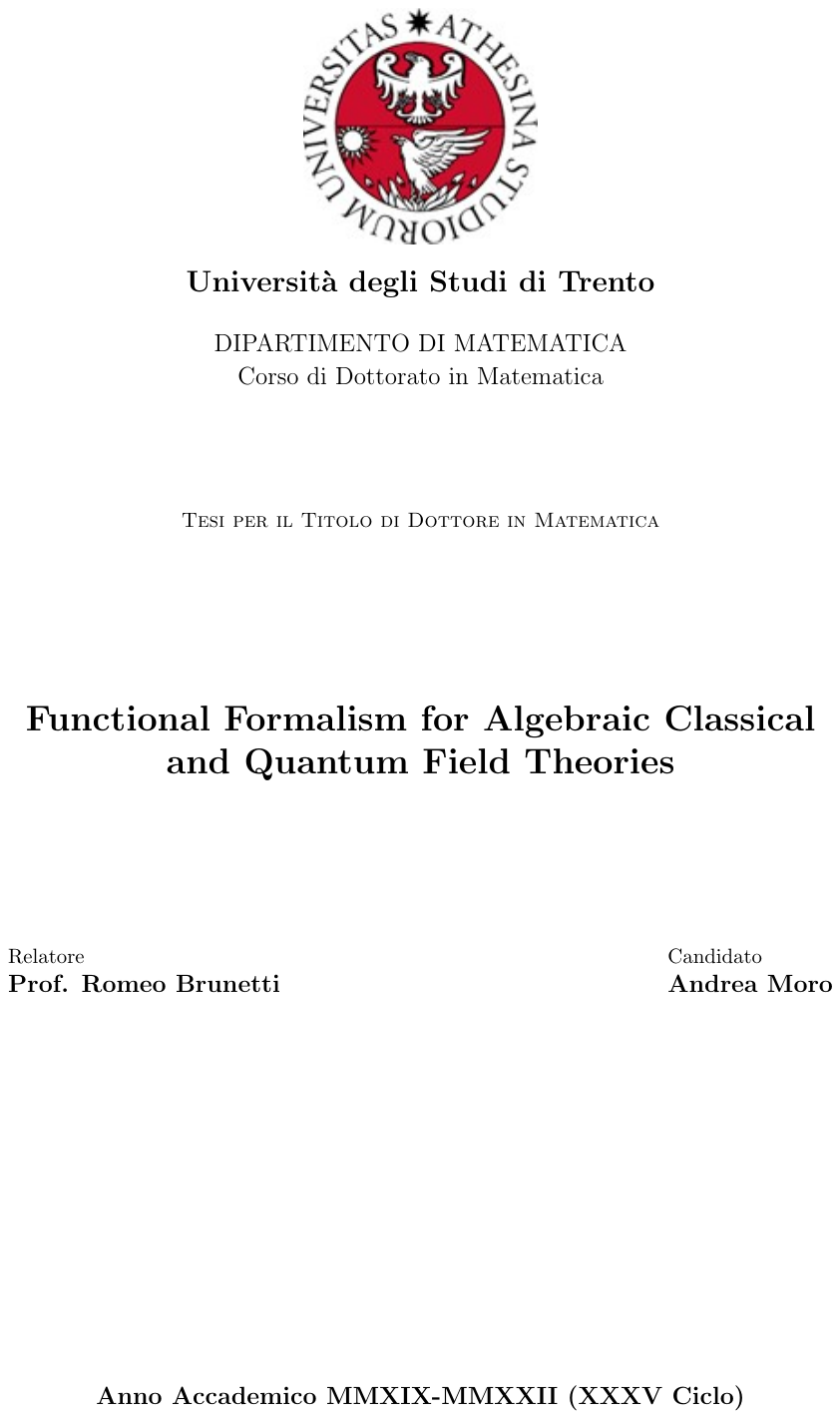}
%   QUOTE (OPTIONAL)
%\input{Include/Frontmatter/Quote}
\newpage
%   ABSTRACT
\afterpage{\null\thispagestyle{empty}\clearpage} % Blank page
\thispagestyle{plain}			% Supress header 
\setlength{\parskip}{0pt plus 1.0pt}
\section*{Abstract}
In the first part of this thesis we study the generalization of the recent algebraic approach to classical field theory by proposing a more general setting based on the manifold of smooth sections of a non-trivial fiber bundle. Central is the notion of observables/functionals over such sections, \textit{i.e.} appropriate smooth functions on them. The kinematic will be further specified by means of Peierls brackets, which in turn are defined via the causal propagators of linearized field equations. In the second part we implement deformation quantization of the algebras obtained above in the simpler setting of scalar field theory. Wick powers and time ordered products for quantum field theories in curved spacetimes are defined by giving a set of axioms which, when implemented, defines uniquely, up to some classifiable ambiguities, the aforementioned quantities. Those ambiguities are known to be tightly restrained by locality, covariance and other regularity conditions. One of the additional constraints used was to require continuous and analytic dependence on the metric and coupling parameters. It was recently shown that this rather strong requirement could be weakened, in the case of Wick powers, to the so-called parametrized microlocal spectrum condition. We therefore show the existence of Wick powers satisfying the above condition and extend this axiom to time ordered products, while reestablishing the usual uniqueness and existence results in light of the new constraint.

% KEYWORDS (MAXIMUM 10 WORDS)
\vfill

\thispagestyle{empty}
\mbox{}

\begin{singlespace}
 \tableofcontents 	
 % \addcontentsline{toc}{chapter}{\listfigurename}
 % \listoffigures
 % \addcontentsline{toc}{chapter}{\listtablename}
 % \listoftables
% \renewcommand{\lstlistlistingname}{List of Code}
% \addcontentsline{toc}{chapter}{\lstlistlistingname} 
% \lstlistoflistings

%   NOMENCLATURE AND ACRONYMS
 \printnomenclature
 \printglossaries

\end{singlespace}

\afterpage{\null\thispagestyle{empty}\clearpage} % Blank page

%   MAIN MATTER
\mainmatter	  % Begin normal, numeric (1,2,3...) page numbering
\clearpage
% ------ set page style fancy with the follow 
\pagestyle{fancy} 
\renewcommand{\chaptermark}[1]{\markright{\chaptername\ \thechapter.\ #1}{}}
\renewcommand{\sectionmark}[1]{\markright{\thesection.\ #1}}
\lhead{} 
\chead{}                   
\rhead{\slshape \rightmark} 
%\fancyfoot[LE,RO]{\thepage} % aggiunto dopo
% \fancyhead{} % clear all header fields
% \fancyhead[RO,LE]{\textbf{The performance of new graduates}}
\fancyfoot{} % clear all footer fields
\fancyfoot[LE,RO]{\thepage}
\fancyfoot[LO,RE]{Andrea Moro}

% \fancyfoot[]{}
% \lfoot{Andrea Moro}
% \cfoot{\thepage} 
% \rfoot{}          
\renewcommand{\headrulewidth}{0.4pt} 
\renewcommand{\footrulewidth}{0.4pt} 

% \pagestyle{fancy}
% \fancyhead{}
% \fancyhead[LE]{\textit{Wave Maps and the Algebraic Approach to Classical Field Theory}}
% \fancyhead[RO]{\textit{Andrea Moro}}
% \fancyfoot{}
% \fancyfoot[CE,CO]{\thepage}

%   CHAPTERS

\addcontentsline{toc}{chapter}{Introduction}

\chapter*{Introduction}

\thispagestyle{plain}
% \nomenclature{$M_\infty$}{Free Stream Mach number}
% \newacronym{cfd}{CFD}{Computational Fluid Dynamics}

\section*{Classical algebraic approach to field theories: state of the art}

The mathematical treatments of classical field theories find inspiration and draw ideas from two main sources: the Hamiltonian and Lagrangian formalisms of classical mechanics. If one intends also to include relativistic phenomena then there remain essentially only two rigorous frameworks, both emphasizing the geometric viewpoint: the multisymplectic approach (see \cite{gotay1998momentum}, \cite{gotay2004momentum}, \cite{forger2005covariant}) and another related to the formal theory of partial differential equations (see \cite{fati}, \cite{krupka}). They have several points in common and there is now a highly developed formalism leading to rigorous calculus of variations. The concept of \textit{dynamics} plays a distinguished role here: usually it is assigned by some PDEs called \textit{equations of motion}. It has to be regarded as a lucky coincidence that equations of motions in physical theories can be obtained via the principle of \textit{least action} once a certain a priori object, the Lagrangian, has been fixed. In this geometrical setting the notion of solutions is implemented as a submanifold of an appropriate jet bundle. We stress that this formalism carefully avoids treating the intrinsic infinite dimensional degrees of freedom of the configuration spaces.\\
On the other hand, there exists another treatment of classical mechanics that emphasises more the algebraic and the analytic structures and is intrinsically infinite dimensional, which is named after the pioneering works of von Neumann (\cite{neumann1932operatorenmethode}) and Koopman (\cite{koopman1931hamiltonian}) and works directly in Hilbert spaces. We could use another independent mathematical viewpoint which, combining algebraic geometry and quantum mechanics, naturally leads to an infinite dimensional setting: one can show that
\noindent
\begin{center}
``\textit{if $M$ is a second countable differentiable manifold then the ring of morphisms $\mathrm{Hom}\left( C^{\infty}(M,\mathbb{R}),\mathbb{R}\right)$ is canonically isomorphic to $M$}''\footnote{This result can be found \textit{e.g.} in Problem 1-C pp. 11-12 of \cite{milnor} and will be therefore called \textit{"Milnor's exercise"}. It was originally proved in \cite{pursell1953algebraic}. A more recent proof can be found in Lemma 35.8 and Corollary 35.9 pp. 300-301 in \cite{kolar2013natural}. We also remark that a consequence of the above result (see \textit{e.g.} in Corollary 35.9 pp. 301 in \cite{kolar2013natural}) is the following: for any two manifolds $M$, $N$ there is a bijective mapping 
$$
C^{\infty}(M,N) \to \mathrm{Hom}\big( C^{\infty}(M,\mathbb R) , C^{\infty}(N,\mathbb R)\big), \quad f\to f^*
$$
where $f^*(g)= g\circ f$ is the pull-back mapping.} ;
\end{center}
\noindent
this implies that properties of the manifold $M$ might be encoded into the algebra of smooth functions $C^{\infty}(M,\mathbb{R})$. These ideas converged in the last decade when another treatment was developed: it also drew inspirations from perturbative quantum field theories in the algebraic fashion (see \textit{e.g.} \cite{BDF09}). This approach is close in spirit to the von Neumann-Koopman formalism and emphasises more the observables point of view by dealing directly with the configuration space as an infinite dimensional space.\\
It is then natural to endow the infinite dimensional space of field configurations with a manifold structure and later define observables as smooth functions therein. If one is willing to generalise the setting of \cite{koopman1931hamiltonian, neumann1932operatorenmethode} to field theories finds immediately an insurmountable difficulty, namely, a result by Eells and Elworthy (see \cite{ely}, \cite{ely2}) constrains a configuration space, viewed as a second countable Hilbert manifold, to be smoothly embedded into its ambient space, \textit{i.e.}, it is just an open subset of the ambient Hilbert space. If, instead of Hilbert manifolds, one allows for Banach manifolds other difficulties emerge: by a result of \cite{omori}, if $G$ is a Lie Group with Banach manifold structure and acts transitively, effectively (\textit{i.e.} the only element of the isotropy group of $G$ is the identity) and smoothly on a compact finite dimensional manifold $M$, then $G$ has to be finite dimensional. This is problematic since the action of $\mathrm{Diff}(M)$ on $M$ is transitive and effective; however, the above result cannot be smooth in the Banach smooth structure. Hence, we need to bypass those facts of life and find a clever replacement. The solution is offered regarding the configuration space as an infinite dimensional manifold modelled over locally convex spaces (see \textit{e.g.} \cite{michor1980manifolds, convenient, hamilton1979inverse}). Here, one relies heavily on the clarifications given in the last thirty years about the most appropriate calculus on such spaces (see, \textit{e.g.}, \cite{mb,convenient}).\\
Based on the ideas exposed above, and using as inputs also some crucial notions belonging to microlocal analysis, the authors of \cite{acftstructure} described the case of scalar field theories on globally hyperbolic spacetimes. Observables (or equivalently functionals) are implemented as spacetime compactly supported smooth functions over the infinite dimensional manifold of scalar field configurations endowed with Bastiani smooth calculus and are classified with respect to regularity conditions typical of microlocal analysis. Instead of introducing dynamics by variation of an action functional, the latter is equivalently introduced by differentiation, in the smooth structure described above, of \textit{microlocal functionals}. The old variational setting is essentially recovered using the characterization of microlocal functionals \cite[Proposition 2.2]{acftstructure}. It is important to notice that observables treated henceforth are \textit{off-shell}, \textit{i.e.} no extra dynamical condition is imposed on them; this is especially relevant in the context of quantum field theory where often quantum fields are studied on-shell. In this thesis we avoid working on shell focusing on the more general off-shell approach.\\
Once a dynamic, that is a family of microlocal functionalsinducing the same normally hyperbolic linearized equation, is chosen, the authors introduce a Poisson structure on the space of functionals by the Peierls bracket (c.f \cite[Definition 3.5]{acftstructure}). It can be showed the the Poisson bracket of two microlocal functionals is well defined however, it fails to be microlocal, therefore the space of microlocal functionals is enlarged to the bigger space of microcausal functionals (\textit{c.f.} \cite[Definition 3.6]{acftstructure}) by optimally constraining the singularity structure of such observables. It is then shown (see \cite[Proposition 3.4 and Corollary 3.3]{acftstructure}) that the Poisson bracket of two microcausal functionals is well defined, closed and satisfies the Jacobi identities. \\
The space of microcausal functionals has many properties: it can be given the topology of a nuclear locally convex space (see \cite[Theorem 4.1 and the discussion in Remark 4.3]{acftstructure}), has the $C^{\infty}$-ring structure\footnote{An algebra $\mathcal{A}$ has the $C^{\infty}$-ring structure if, given any $a_1,\ldots,a_n\in \mathcal{A}$, $f\in C^{\infty}(\mathbb K^n,\mathbb K)$, there are mappings $\rho_f:\times^n\mathcal{A}\to \mathcal{A}$ such that if $g\in C^{\infty}(\mathbb K^m,\mathbb K)$, $f_i\in C^{\infty}(\mathbb K^n,\mathbb K)$ with $i=1,\ldots,m$ then $$\rho_h\big(\rho_{f_1}(a_1,\ldots,a_n),\ldots,\rho_{f_m}(a_1,\ldots,a_n)\big)=\rho_{h\circ (f_1,\ldots, f_m)}(a_1,\ldots,a_n).$$The field $\mathbb K$ can either be $\mathbb R$ or $\mathbb C$.} (\cite[Theorem 4.2]{acftstructure}), by $(i)-(iii)$ in \cite[Proposiotion 4.1]{acftstructure} one can generate partitions of unity with microcausal functionals on the configuration space and, by $(iv)-(v)$ in the same proposition, an analogue of \textit{Milnor's exercise} (see \textit{e.g.} Problem C-1 pp. 11 in \cite{milnor}) in the infinite dimensional setting holds.\\
On-shell observables are finally defined as elements of the quotient of the space of microcausal functionals (\textit{i.e.} off-shell observables) by the ideal (with respect to the pointwise product) of functionals vanishing on solutions of the dynamic. An alternative is to define this ideal as the set of functionals which are kinematical derivations of the (family) microlocal action inducing the dynamic. The latter is also an ideal for the Poisson structure (\cite[Proposition 4.2]{acftstructure}) but it is an open problem to determine the relation between those two definitions.\\

\section*{Functional approach to quantum field theories: state of the art}
Algebraic quantum field theory is an axiomatic, mathematically rigorous framework to study quantum field theories originating from \cite{haag1964algebraic}. The idea is to assign a suitable $*$-algebra with product $\star$ to the physical system under consideration, then once a state on this algebra has been assigned, one can use the GNS reconstruction theorem to recover the standard quantum mechanical approach. The above $*$-algebra must also be endowed with an additional product: the time-ordered product $\cdot_T$. The latter is easily defined for regular functionals, and can be extended to microlocal functionals, however it cannot be defined for general microcausal functionals. The $\star$ product is the deformation quantization product, the other defines the $S$-matrix which is crucial in perturbation theories to treat interactions. The definition of the $\cdot_T$ product is plagued by various problems such as the $IR$ and $UV$ divergences as well as the need to work in formal series of the quantum deformation parameter. The latter is generally tolerated, the first is cured by assuming the interactions have compact support whereas to cure the \textit{UV-divergence} problem requires some work. For instance, a general framework based on the the ideas of Bogoliubov, Parasiuk, Hepp, Zimmermann (see for details \cite{bogoliubov1957multiplication, hepp1966proof, zimmermann1968power,zlmmbrmann2000convergence}) which form the backbone of the \textit{BPHZ-renormalization scheme} aims at eliminating divergences at each perturbative order in $\hbar$ by expressing time-ordered products in terms of Feynman diagrams and then subtracting ad-hoc counterterms. Later on, in \cite{epstein1973role}, Epstein and Glaser implemented a renormalization scheme with extra physically relevant constraints (\textit{e.g.} unitarity and causality) which are shown to be conserved at each renormalization step. These renormalization scheme, \textit{i.e.} the subtractions of terms, is not unique and was shown in \cite{popineau2016pedagogical} that the \textit{renormalization freedom} is tightly constrained by the requirements made in the Epstein-Glaser framework; moreover, the different renormalization terms were precisely those of the St\"uckelberg-Petermann renormalization group (see \cite{stueckelberg1953normalisasion}). \\
The approach described above uses many features of Minkowski spacetime, moreover renormalization is carried out in Fock spaces. In generalizing this approach to curved spacetimes, one runs into a set of problems: first of all translational invariance (more generally Poincaré covariance) is generally lost, no vacuum state is present, it is not possible to use Fourier transform methods to regularize Feynman diagrams and there is no general connection between the Euclidean and Lorentzian theory such as Osterwalder-Schrader theory (see \cite{osterwalder1973axioms} for details). Those problems where given a solution by using ideas from Epstein-Glaser, the introduction of microanalytical techniques that substituted the Fourier transform with the analysis of singularities by means of the wave front set. In particular, Radzikowski showed in \cite{radzikowski1996micro} that microlocal analysis allows a key characterization of the spectrum condition in terms of wave front sets. Building on these results Brunetti, Fredenhagen and collaborators produced a framework, (see \cite{brunetti1996microlocal,BF97}) to study quantum field theory in curved (globally hyperbolic) spacetimes which did not require a vacuum state or any Fourier transform methods. This framework was reminiscent of the one provided by \cite{haag1964algebraic} for Minkowski spacetime, which stated a set of axioms describing the properties of a $C^*$-\textit{algebra of observables} associated to each spacetime region. Building on this approach, but dropping the rather strong requirement of $C^*$-algebra in favor of the weaker $*$-algebra of formal power series in $\hbar$, the authors of \cite{BF97, BDF09} derived what is now called \textit{perturbative algebraic quantum field theory} (pAQFT for short). The next step was to replace Poincaré covariance: this was accomplished by introducing the principle of locality and covariance for quantum field theories (\cite{BFV03,hollands2001local}) which was then extended to time-ordered products in \cite{hollands2001local,hollands2002existence,hollands2005conservation}. The locally covariant formulation of quantum field theories is set up in categorical terms by requiring that quantum fields to be natural transformations; however, at its core, this formulation implies that quantum processes ought to be localizable in spacetime and that the setting of the theory is not a just a fixed spacetime, but a family of spacetimes isometrically embedded into one another. %The notion of Wick power was crucial for the purpose of defining the quantum $*$ algebra of observables, and furthermore for studying time-ordered products (\textit{e.g.} in \cite{hollands2001local, hollands2002existence}); however,The introduction of functional formalism in pAQFT in \cite{BDF09}, relying on deformation quantization of the Poisson algebra of classical observables, proved an extremely important tool that made possible to study the theory \textit{off-shell} and to show that finite renormalization ambiguities do form a group in the sense of St\"uckelberg and Petermann. Moreover, it was shown that the formalism (\textit{e.g.} the construction of Wick powers) could be made independent of the non-covariant choice of a Hadamard state. The renormalization scheme put forward in \cite{hollands2001local,hollands2002existence} relies on a rather strong analyticity requirement: that Wick powers and time-ordered products should in fact vary analytically under analytical variations of the spacetime metric and the other parameters of the theory. This condition was weakened, for Wick powers, in \cite{khavkine2015algebraic, khavkine2019wick} to the so-called \textit{parameterized microlocal spectrum condition} which replaced analyticity with weak regularity (\textit{i.e.} a jointly smooth variation of the parameters produces a jointly smooth variation of Wick powers).\\
In Chapters \ref{chapter_Wick} and \ref{chapter_TO} we employ the off-shell functional formalism of \cite{BDF09,acftstructure} to the end of completely removing the analyticity requirements from the renormalization scheme of \cite{hollands2001local, hollands2002existence}. More precisely, we shall generalize the microlocal spectrum condition to time-ordered products while reestablishing the existence and uniqueness results. To start the Epstein-Glaser induction step for perturbative renormalization it is fundamental to show the existence of Wick powers satisfying the weak regularity requirement of \cite{khavkine2016analytic}.\\

\section*{Structure of the thesis and main results}

%Here we briefly summarize the structure of the thesis.
\subsection*{Chapter 1: Preliminaries}

In this chapter we collect all the notions that will be employed in the subsequent chapters. In particular, we shall focus on some differential geometry (\Cref{section_Lorentzian}, \Cref{section_fiber_bundles}) and on locally convex spaces and two notions of calculus therein (\Cref{section_LCS&calculus}, \Cref{section_Bastiani}, \Cref{section_convenient}).

\subsection*{Chapter 2: Non-trivial bundles and algebraic classical field theory} 

The aim of this chapter is to extend the classical algebraic approach of \cite{acftstructure} from scalar field to fields of generic non-trivial bundles. In \Cref{section_topology_map_manifold} and \ref{section_map_manifold} we introduce the topology and the infinite dimensional manifold structure for spaces of mappings. In \Cref{section_observables} we describe the notions of observables/functionals as compactly \textit{space-time supported} (\Cref{def_1_func_support}), smooth (\Cref{def_1_func_differentiability}) functions on the manifold $\Gamma^{\infty}(M\leftarrow B)$. Classes of functionals are presented (\Cref{def_1_func_classes}) and characterized in more intuitive ways (\Cref{prop_1_additivity} and \ref{porop_1_muloc_charachterization}). Then we study a general class of action functionals which admits normally hyperbolic linearized differential equation and define, in \ref{def_1_Peierls}, the Peierls bracket, to make the letter a closed operation, in \Cref{def_1_WF_mucaus} and \Cref{thm_1_mucaus_1}, we enlarge the space of microlocal functionals to microcausal functionals: $\mathcal{F}_{\mu c}(B,\mathcal{U},g)$. Finally we show that the Peierls bracket satisfies the Jacobi identity (\Cref{thm_1_jacobi}). Finally, we characterize $\mathcal{F}_{\mu c}(B,\mathcal{U},g)$ with the structure of a nuclear locally convex space (\Cref{thm_1_mucaus_top}), show that it has the $C^{\infty}$ ring property (\Cref{prop_1_C-infty_ring}) and other additional properties in \Cref{prop_1_mu_caus_top_prop}. In \Cref{lemma_1_regular_density}, we show, for the case of scalar field theories, the density of regular functionals into microcausal functionals.

\subsection*{Chapter 3: Functional formalism in quantum field theory: Wick powers} 

The aim of this chapter is to generalize the definition of Wick powers in \cite{khavkine2016analytic} to off-shell Wick powers and show their existence. Specifically, in \Cref{section_natural_bndl&background_geometries} we introduce the notion of natural bundle together with the example of the bundle of background geometries which we use throughout the rest of the thesis. Crucial for the characterization results of \Cref{thm_2_Moretti_Kavhkine} are \Cref{proposition_2_covariant_coordinates} and \Cref{thm_2_covariant_identity}. In \Cref{section_Hadamard_parametrix} we derive the local expression of the Hadamard parametrix. In \Cref{section_Wick_powers} we introduce the algebraic functional formalism for quantum field theories. This leads to the creation of the abstract algebra of microcausal functionals \eqref{eq_2_abstract_muc_algebra} which possess a deformation quantization product $\star$. We define off-shell Wick powers in \Cref{def_2_Wick_quantum_powers}, along the lines of \cite{khavkine2016analytic} we show their characterization result in \Cref{thm_2_Moretti_Kavhkine}. Finally, using the equivalence of weak regularity and convenient smoothness in Proposition \ref{prop_2_equivalence_conv-smooth_w-regularity} we show existence of Wick powers (\Cref{thm_2_existence_Wick_polynomials}).

\subsection*{Chapter 4: Functional formalism in quantum field theory: Time ordered products} 

The aim of this chapter is to extend the \textit{parameterized microlocal spectrum condition} (i.e. condition $(iv)$ in \cite{khavkine2016analytic}) to time-ordered products, then show existence and uniqueness results with the new definition given. We start by collecting in \Cref{section_dang} the results concerning the extensions of distributions to submanifolds. Then in \Cref{section_TO} we present the definition of time-ordered products \textit{c.f.} \Cref{def_3_TO_products}. Our definition does use the functional formalism (\textit{e.g.} to prescribe the singularity structure) and is off-shell. In \Cref{section_TO_existence}, we present the inductive construction to establish the existence of time-ordered products. This relies on the fact that renormalization is equivalent to extending certain distributions, and culminates with \Cref{thm_3_TO_renormalization_1} and \ref{thm_3_TO_renormalization_2}. Finally in \Cref{thm_2_uniqueness_TO}, as in the case of Wick powers, we characterize how two families of time-ordered products may differ, and in Corollary \ref{coro_2_TO_uniqueness} we further specialize this result.

\chapter{Preliminaries}
\thispagestyle{plain}

% \nomenclature{$M_\infty$}{Free Stream Mach number}
% \newacronym{cfd}{CFD}{Computational Fluid Dynamics}
We describe the content of this chapter:\\

In \Cref{section_Lorentzian} we shall recall some topological properties of a Lorentzian manifold related to the causality concept. Also a notion of Cauchy hypersurface (Definition \ref{def_1_Cauchy_hypersurface}) is introduced together with the notion (Definition \ref{def_1_global_hyperbolicity}) and properties (\Cref{thm_1_global_hyperbolicity}) of globally hyperbolic manifolds.\\

In \Cref{section_fiber_bundles} we introduce the notion of bundle (Definition \ref{def_0_fiber_bundles}), which has a geometric relevance on its own, but, on top of that, gives Physics a toolkit to place physical entities such as \textit{fields} (see \eqref{eq_0_space_of_smooth_sections} defining the space of fields). We will also specify certain special classes of bundles and introduce connections. Then we describe the Peetre-Slov\'ak Theorem (\Cref{thm_A_Peetre_Slovak}) which is a characterization of finite order differential operators between sections of bundles. The relevance of this theorem lies in its applications: we use it in Proposition \ref{porop_1_muloc_charachterization} to characterize microlocal functionals as \textit{local integrals of Lagrangians}, and in \Cref{thm_2_Moretti_Kavhkine}, \Cref{thm_2_uniqueness_TO} to characterize families of Wick powers and Time-ordered products. Among finite order differential operators we shall pay close attention to second order globally hyperbolic operators \textit{i.e.} wave equations (\textit{c.f.} Definition \ref{def_1_global_hyperbolicity}) and how, in this special case, we can derive a strong framework for solving wave equations geometrically (\textit{c.f.} Theorem \ref{thm_1_properties_of_Green_functions}).\\

In \Cref{section_LCS&calculus} we recall same basics facts about locally convex spaces and detail the constructions of the topological vector spaces structures on the spaces of smooth functions $C^{\infty}(M)$, $C^{\infty}_c(M)$, generalizing those to spaces of sections of vector bundles $\Gamma^{\infty}(M\leftarrow E)$, $\Gamma^{\infty}_c(M\leftarrow E)$ and their duals: the spaces of distributions. \textit{En passant}, we also briefly recall the notion of wave front set for distributions and how this can be used to compose integral kernels with singularities. Handling wave front sets is of paramount importance throughout the thesis and constitutes an essential part of the microanalytical formalism for quantum field theories.
Finally we will describe two possible notions of calculus on locally convex spaces. The first one is Bastiani calculus which is a generalization of calculus in Banach spaces (for details, we refer to \cite{mb,michal1938differential} and \cite{hamilton1979inverse} for the case of Fréchet spaces). On the one hand this calculus is very natural for it enjoys many of the properties of calculus: being smooth is stable under compositions, the chain rule holds and one can also define integrals. On the other hand, Bastiani calculus does not posses the following Cartesian closedness property: if $E, \ F, \ G$ are locally convex spaces, then
\begin{equation}\label{eq_0_cartesian_closedness}
    C^{\infty}\big( E, C^{\infty}(F,G)\big) \simeq C^{\infty}\big( E\times F,G\big) .
\end{equation}
Beyond being a somewhat natural categorical request for a calculus in locally convex spaces, it has great practical utility in our case due to the joint smoothness requirement that enters the definition of Wick powers (see $(iv)$ in Definition \ref{def_2_Wick_quantum_powers}). A suitable candidate satisfying \eqref{eq_0_cartesian_closedness}, called \textit{convenient calculus}, is provided in \cite{convenient}: it satisfies most of the natural properties of calculus (\textit{e.g.} smoothness is stable under compositions, the chain rule holds and integrals of curves are well defined), however, it has quite a drawback: smooth mappings may fail to be continuous (see Proposition 2.2 and $\S$2 of \cite{glockner2005discontinuous}). In the particular case of Fréchet spaces however, convenient calculus and Bastiani calculus are equivalent (see $(iv)$ Proposition \ref{prop_A_prop_of_c_infty_top}).

\section{Selected topics in Lorentzian geometry}\label{section_Lorentzian}

Let $(M,g)$ be a Lorentzian manifold. At each $x\in M$ we can choose coordinates for which $g_x=\mathrm{diag}(-1,1, \ldots,1)=\eta$.A non-zero tangent vector $v_x \in T_x M$ is called \textit{timelike} if $g_x(v_x,v_x) < 0$, \textit{lightlike} or \textit{null} if $g_x(v_x,v_x) = 0$, \textit{spacelike} if $g_x(v_x,v_x) > 0$.\\

At each point of $M$ we can define a \textit{direction of time}, more precisely, we seek a way to give a consistent notion of future-directed and past-directed for any timelike tangent vector on $TM$. Since the latter set is the interior of a cone with two disconnected components, we can consistently pick one of those two cones and call it \textit{future-directed} if there is a global timelike vector field $X$. From now on we shall assume that whenever we work with a Lorentzian manifold, it will have a time orientation, i.e. a global timelike vector field $X$ defining a future orientation.\\

If $x$, $y \in M$, we write $x \ll y$ if there exist a future directed, timelike smooth curve $\gamma$ joining $x$ to $y$, $x \leq y$ if either $x=y$ or there is a future directed, causal smooth curve $\gamma$ joining $x$ to $y$; in the latter case we write $x<y$. The sets $I^{\pm}(x)= \lbrace y \in M : x \ll y \rbrace$, $J^{\pm}(x)= \lbrace y \in M : y \leq x \rbrace$ are respectively called the \textit{chronological future/past} of $x$, \textit{causal future/past} of $x$. Analogously, if $U \subset M $ is open, we denote by $I^{\pm}(U)$ (resp. $J^{\pm}(U)$) the sets $\cup_{x\in U}I^{\pm}(x)$ (resp. $\cup_{x\in U}J^{\pm}(x)$). Note that the sets $J^{\pm}(U)$ are in general nor open neither closed, for example, consider spacetime $(\mathbb{R}^2 \backslash [ -2,2 ] \times 1, \mathrm{diag}(-1,1) )$ and $J^+((2,0))$. However, one can show that the sets $I^{\pm}(x)$ are always open; therefore, taking intersections of the form $I^{+}(x) \cap I^{-}(y)$ we get a family of open sets, which we call \textit{chronological open diamonds} by the shape they have on Minkowski spacetime. Such a family defines a topology on $M$, called the \textit{Alexandrov topology}. By construction (those sets are already open in $M$ with its standard topology) the Alexandrov topology is coarser than the original one. In the same spirit one can construct a family of \textit{causal diamonds} by $$	J(x,y) \doteq J^{+}(x) \cap J^{-}(y).$$If $A \subseteq M$, is called \textit{future (past) compact} if $J^+_M(x) \cap A$ ($J^-_M(x) \cap A$) is compact $\forall x \in M$.\\

Next we present the necessary tools to define globally hyperbolic spacetimes, \textit{i.e.} those spacetimes where the Cauchy problem can be defined in a meaningful way. If $A$ is a subset of $M$, we call it \textit{achronal} if every timelike curve intersects it at most once, \textit{acausal} if every causal curve intersects $A$ at most once. Any achronal subset is acausal, while the opposite is not always true, for example, in any closed lightlike cone on Minkowski spacetime, which is achronal, a light ray passing through the vertex and any other point of the boundary is a causal curve intersecting $A$ infinitely many times.

\begin{definition}\label{def_1_Cauchy_hypersurface}
A subset $\Sigma$ of $M$ is called Cauchy hypersurface if every inextendible timelike curve on $M$ meets $\Sigma$ in exactly one point.
\end{definition}

If $M$ admits a Cauchy surface $\Sigma $, given any $y \in M$ and any maximally extended timelike curve $\gamma$ through $y$, then $\gamma$ will intersect $\Sigma$ at some point $x$, and we will have either $x=y$ or $x \ll y$, $y \ll x$. As a result, we can write $M$ as the disjoint union $M=I^+(\Sigma) \bigsqcup \Sigma \bigsqcup I^-(\Sigma)$. Intuitively, a Cauchy hypersurface is the locus of points on the manifold where all the past is known, (in the sense that past particles and light rays will hit the surface), and the future is predictable given some notion of evolution.

\begin{definition}\label{def_1_causal_spacetime}
Let $M$ be a time-oriented Lorentzian manifold, then
\begin{itemize}
\item[$(i)$] $M$ is called \textit{causal} if there are not closed causal curve on $M$,
\item[$(ii)$] $U \subseteq M$ open is called \textit{causally convex} if any causal curve connecting points of $U$ is contained in $U$,
\item[$(iii)$] $M$ is \textit{strongly causal at} $x\in M$ if every open neighborhood of $x$ contains an open causally convex neighborhood. If this condition holds in very point, then $M$ is called strongly causal.
\end{itemize}
\end{definition}

\begin{theorem}{\textbf{(Kronheimer, Penrose)}}
A time oriented Lorentz manifold $(M,g)$ is strongly causal if and only if the Alexandrov topology coincides with the original topology on $M$.
\end{theorem}

\begin{definition}\label{def_1_global_hyperbolicity}
A time-oriented Lorentzian manifold $(M,g)$ is called \textit{globally hyperbolic} if
\begin{itemize}
\item[$(i)$] $M$ is causal,
\item[$(ii)$] $\forall x, \ y \in M$ all diamonds $J_M(x,y)$ are compact.
\end{itemize}
\end{definition}

\begin{definition}\label{def_1_time_function}
Let $(M,g)$ time-oriented Lorentzian manifold and $t :M \rightarrow \mathbb{R} $ a continuous map, then $t$ is called a
\begin{itemize}
\item[$(i)$] \textit{time function} if it is strictly increasing along all future directed causal curves,
\item[$(ii)$] \textit{temporal function} if it is smooth with $\mathrm{grad}(t)$ is future directed and timelike,
\item[$(iii)$] \textit{Cauchy time function} if it is a time function whose level sets are Cauchy hypersurfaces,
\item[$(iv)$] \textit{Cauchy temporal function} if it is a temporal function whose level sets are Cauchy hypersurfaces.
\end{itemize}
\end{definition}

We remark how computing the flow of the gradient of a Cauchy temporal function $t$ gives us a diffeomorphism between hypersurfaces, moreover, since timelike curves intersects a given level set of $t$ (which is a Cauchy hypersurface) exactly once then we create a diffeomorphism $$M \simeq t(M) \times \Sigma_{t_0} $$for some reference time $t_0 \in t(M)$ and $\Sigma_{t_0}=t^{-1}(t_0)$. This gives us a very strong necessary topological condition for $M$ to admit a Cauchy temporal function, and as we will immediately see for global hyperbolicity.
\begin{theorem}\label{thm_1_global_hyperbolicity}
Let $(M,g)$ be a connected time oriented Lorentz manifold, then the following are equivalent
\begin{itemize}
\item[$(i)$] $(M,g)$ is globally hyperbolic,
\item[$(ii)$] There exist a topological Cauchy hypersurface,
\item[$(iii)$] There exist a smooth spacelike Cauchy hypersurface.
\end{itemize}
In this case there is a Cauchy temporal function $t$ and $(M,g)$ is isometrically diffeomorphic to the product manifold $\left( \mathbb{R} \times \Sigma, \beta dt^2-g_t \right)$ where $\beta \in C^{\infty}(\mathbb{R} \times \Sigma)$ is a positive function and $g_t \in \Gamma^{\infty}(S^2T^{*}\Sigma)$ is a Riemannian metric on $\Sigma$ depending smoothly on $t$.
Moreover, each level set $ \Sigma_t= \lbrace (t,\sigma) \in \mathbb{R} \times \Sigma \rbrace \subseteq M $ of the temporal function $t$ is a smooth spacelike Cauchy hypersurface.
\end{theorem}
The equivalence of $(i)$ and $(ii)$ was shown in \cite{geroch1970domain}, whereas the remaining claims in \cite{bernal2003smooth, bernal2005smoothness}.

\section{Fiber bundles}\label{section_fiber_bundles}

In this section we introduce the notion of bundle (Definition \ref{def_0_fiber_bundles}), which has a geometric importance of its own, but, on top of that, gives physics a toolkit to place physical entities such as \textit{fields} (see \eqref{eq_0_space_of_smooth_sections} defining the space of fields). We will also specify certain special classes of bundles and introduce connections.
\begin{definition}\label{def_0_fiber_bundles}
    A \textit{fiber bundle} is a quadruple $(B,\pi,M,F)$, where $B$, $M$, $F$ are smooth manifold called respectively the \textit{bundle}, the \textit{base} and the \textit{typical fiber}, such that:
\begin{enumerate}
    \item[$(i)$] $\pi:B \rightarrow M$ is a smooth surjective submersion;
    \item[$(ii)$] there exists an open covering of the base manifold $M$, $\lbrace U_{\alpha} \rbrace_{\alpha \in A}$ admitting, for each $\alpha \in A$, diffeomorphisms $t_{\alpha}: \pi^{-1}(U_{\alpha})\rightarrow U_{\alpha}\times F$, called trivializations, which are fiber respecting \textit{i.e.} $\mathrm{pr}_1 \circ t_{\alpha} = \pi|_{\pi^{-1}(U_{\alpha})}$.
\end{enumerate}
Given $U_{\alpha}$, $U_{\beta}$ subsets of $M$, we can define the transition functions $g_{\alpha\beta}: U_{\alpha\beta}\times F \to F :(x,y) \mapsto g_{\alpha\beta}(x,y)\doteq \mathrm{pr}_2 \circ t_{\alpha}\circ t_{\beta}^{-1}(x,y)$.
\end{definition}

We remark that $\mathrm{pr}_1 \circ t_{\alpha} = \pi|_{\pi^{-1}(U_{\alpha})}$ implies $\mathrm{pr}_1 = \pi \circ t_{\alpha}^{-1}$, thus $t_{\alpha}(\pi^{-1}(x))\simeq F$ for all $x\in M$, moreover, by the very definition of trivializations, for any $x\in M$ the transition functions $g_{\alpha\beta}$ evaluated at $x$ are elements of $\mathrm{Diff}(F)$. Furthermore, the mappings $g_{\alpha\beta}$ do satisfy cocycle\footnote{More precisely they satisfy a $\mathrm{Diff}(F)$ valued \u Cech $2$-cocycle. For the definition of the latter, let $\mathcal{F}$ be a (pre)sheaf of $R$-modules over a topological space $M$ and $\{U_{\alpha}\}_{\alpha \in A}$ a cover of $M$. Setting $U_{\alpha_0, \ldots, \alpha_p} \doteq U_{\alpha_0}\cap \cdots \cap U_{\alpha_p}$ we define the \textit{\u Cech $p$-cochain} as
$$
\check{C}^p(M,\mathcal{F})\doteq \prod_{\alpha_0,\ldots,\alpha_p \in A^{p+1}} \mathcal{F}(U_{\alpha_0, \ldots, \alpha_p}).$$In addition, consider the mapping
$$
\delta_{\mathcal{F}}^p : \check{C}^p(M,\mathcal{F}) \to \check{C}^{p+1}(M,\mathcal{F}), \quad \delta_{\mathcal{F}}^p=\sum_{i=1}^{p+1}(-1)^p \mathcal{F}(\rho^{\alpha_i}_{\alpha_0,\ldots,\alpha_{p+1}})$$where $\rho^{\alpha_i}_{\alpha_0,\ldots,\alpha_{p+1}}$ is the restriction mapping from $U_{\alpha_0,\ldots\widehat{\alpha_i},\ldots,\alpha_{p+1}}$ to $U_{\alpha_0,\ldots,\alpha_{p+1}}$. For example, if $p=1$, then
$$
\delta^1_{\mathcal{F}}(f)_{\alpha\beta\gamma}= f_{\beta\gamma}|_{U_{\alpha\beta\gamma}} -f_{\alpha\gamma}|_{U_{\alpha\beta\gamma}} + f_{\alpha\beta}|_{U_{\alpha\beta\gamma}}.$$An element $f\in \check{C}^p(M,\mathcal{F})$ is called $p$-\textit{cocycle} if $\delta^p_{\mathcal{F}}(f)=0$. In our case, the sheaf is $C^{\infty}\big(\cdot,\mathrm{Diff}(F)\big)$, the group operation is, instead of addition, the composition of elements of $\mathrm{Diff}(F)$ and the resulting \u Cech $2$ cocycles are elements $g_{\alpha\beta}:U_{\alpha\beta} \to \mathrm{Diff}(F)$ satisfying
$$
g_{\alpha\beta}|_{U_{\alpha\beta\gamma}} \circ (g_{\gamma\beta})^{-1}|_{U_{\alpha\beta\gamma}} \circ g_{\gamma\alpha}|_{U_{\alpha\beta\gamma}} = id_{\mathrm{Diff}(F)}.$$} relations, that is for any $x \in U_{\alpha}\cap U_{\beta}\cap U_{\gamma} \equiv U_{\alpha\beta\gamma} $, by substituting the definition of $g_{\cdots}(x,\cdot)$, we have
$$
g_{\alpha\beta}(x,\cdot)\circ g_{\beta\gamma}(x,\cdot) \circ g_{\gamma\alpha}(x,\cdot) = id_{F}.
$$
Usually via trivialization it is possible to construct charts of $B$ via those of $M$ and $F$. We call those \textit{fibered coordinates} and denote them by $(x^{\mu},y^i)$ with the understanding that Greek indices denote the base coordinates and Latin indices the standard fiber coordinates. The product of two manifolds is always a bundle which is called \textit{trivial}, whereas general bundles are not trivial. For notational simplicity, we also denote by $B|_x$ the set $\pi^{-1}(x), \ x\in M$. \\

Given two fiber bundles $(B_i,\pi_i,M_i,F_i)$, $i=1,2$, we define a \textit{fibered morphism} as a pair $(\Psi,\psi)$, where $\Psi:B_1 \rightarrow B_2$, $\psi:M_1 \rightarrow M_2$ are smooth mappings, such that $\pi_2 \circ \Psi= \psi \circ \pi_1$, sometimes $\psi$ is referred to as the \textit{base projection} of $\Psi$.
\begin{proposition}\label{prop_0_construction_of_bundles}
\noindent
\begin{itemize}
\item[$(i)$] Let $\pi:B\to M$ be a proper surjective submersion and $M$ be connected, then $B$ is a fiber bundle with projection $\pi$ and base $M$.
\item[$(ii)$] Let $M$, $F$ be manifolds, $\{U_{\alpha}\}_{\alpha\in A}$ an open covering of $M$ and $g_{\alpha\beta}:U_{\alpha\beta}\to \mathrm{Diff}(F)$ be a cocycle. Then there exists a unique (modulo isomorphisms) fiber bundle $B$ with base $M$, fiber $F$ and $g_{\alpha\beta}$ as transition functions.
\end{itemize}
\end{proposition}
We remark how Proposition \ref{prop_0_construction_of_bundles} gives us two methods to construct bundles: in the first, one considers a fibered manifold\footnote{A fibered manifold is a generalization of fiber bundles in which however different fibers are not required to be diffeomorphic.} and then adds extra hypothesis ($M$ is connected and $\pi$ is proper) to force all fibers to be diffeomorphic; in the second, one construct the bundle out of the base and fiber as a quotient set with equivalence relations given by the transition functions. Because of $(i)$ in Prop \ref{prop_0_construction_of_bundles} we can denote fiber bundles by either $(B,\pi,M)$ or $\pi:B \to M$ understanding that the typical fiber can be taken, modulo diffeomorphism, as $\pi^{-1}(x)$ for some $x\in M$. In the sequel we shall prefer the latter notations unless we need to specify the typical fiber.\\

In the sequel we will often use {pullback bundles}: given a bundle $\pi:B\to M$, a manifold $N$ and a smooth mapping $\psi:N\to M$, the \textit{pullback bundle} $\varphi^*B$ is defined via the commutative diagram
\begin{center}
\begin{tikzpicture}[every node/.style={midway}]\matrix[column sep={10em,between origins}, row sep={3em}] at (0,0){ \node(A){$\psi^*B$} ; & \node(B) {$B$};\\ \node(C) {$N$}; & \node(D) {$M$};\\};\draw[->] (A) -- (B) node[anchor=south] {$\psi^*\pi$};\draw[->] (A) -- (C) node[anchor=east] {$\pi^*\psi$};\draw[->] (B) -- (D) node[anchor=west] {$\pi$};\draw[->] (C) -- (D) node[anchor=north] {$\psi$};
\end{tikzpicture}
\end{center}
as the set of points $\{(y,b)\in N\times B:\ \psi(y)=\pi(b)\}$.\\

We now list some relevant classes of bundles, which are related to special choices of the typical fiber $F$. First we introduce {vector bundles}, the latter are particular fiber bundles whose standard fiber is a vector space.

\begin{definition}\label{def_0_vector_bundles}
    A bundle $(B,\pi,M,V)$ is a \emph{vector bundle} if the standard fiber $V$ is a vector space and the transition functions $g_{\alpha\beta}$ are $\mathrm{Gl}(V)$-valued for each $\alpha, \beta$.
\end{definition}

We denote fibered coordinates on vector bundles by $\lbrace x^{\mu}, v^i\rbrace$. Given $(E,\pi,M,V)$, $(F,\rho,M,W)$ vector bundles over the same base manifold, it is possible to construct a third vector bundle $(E\otimes F,\pi\otimes \rho,M,V\otimes W)$ called the tensor product bundle whose standard fiber is the tensor product of the standard fibers of the starting bundles.\\

As any other manifolds, $B$ admits vector fields $X\in \mathfrak{X}(B)$, we will often use a particular kind of vector fields: they are called \textit{vertical vector fields} and are defined as the set $\mathfrak{X}_{\mathrm{vert}}(B)=\lbrace X \in \Gamma^{\infty}(TB \rightarrow B) : T \pi(X)=0 \rbrace$, where $\pi: B \rightarrow M$ is the bundle projection and $T\pi:TB \to TM$ denotes the its tangential lift. We will denote by $\Phi^{X}_t : B \rightarrow B$ the flow of any vector field on $B$, and assume in the rest of this work that the parameter $t$ varies in an appropriate interval which has been maximally extended. Note that if $X \in \mathfrak{X}_{\mathrm{vert}}(B)$, then $\Phi^{X}_t$ is a fibered morphism whose base projection $\phi$ is the identity over $M$. Vertical vector fields can be seen as sections of the \textit{vertical vector bundle}, $(VB\doteq \mathrm{ker}(T\pi),\tau_V,B)$ which is easily seen to carry a vector bundle structure over $B$.

\begin{definition}\label{def_0_principal_bundles}
    A bundle $(P,\pi,M,G)$ is a principal bundle if the standard fiber is a Lie group and the transition functions $g_{\alpha\beta}$ acts on $G$ as left translations \textit{i.e.} $g_{\alpha\beta}(x):G\to G : h \mapsto g_{\alpha\beta}(x) \cdot h$ for each $x\in U_{\alpha\beta}$ and each $h\in G$.
\end{definition}

With such a trivialization we can decompose points $p \in P$ as $(x,g) \in U_{\alpha}\times G$, defining $r_h:P\rightarrow P, \ p = (x,g) \mapsto p \cdot h =(x,g \cdot h)$ yields a smooth global right action of $G$ on $P$ called the \textit{principal right action}. Similarly to Proposition \ref{prop_0_construction_of_bundles}, one has an alternative definition of principal bundles:

\begin{lemma}\label{lemma_0_principal_bundles}
    Let $\pi:P\to M$ be a surjective submersion, $G$ a Lie group acting freely on $P$ on the right such that the orbits of the action are exactly the fibers $\pi^{-1}(x)$. Then $P\to M$ is a principal bundle with structure group $G$.
\end{lemma}

If $(Pi,\pi_i,M_i,G_i)$, $i=1,2$ are principal bundles, a principal $\theta$-morphism between them is a triplet $(\Psi,\psi, \theta)$, where $\Psi:P_1 \rightarrow P_2,\ \psi:M_1\to M_2$ is a fiber bundle morphism and $\theta:G_1 \rightarrow G_2$ a homomorphism of Lie groups, satisfying
$$
r_{\theta(h)}\circ \Psi = \Psi \circ r_h \space\ \forall h \in G.
$$
When $G_1=G_2 \equiv G$ and $\theta =id_G$ then we simply call it principal morphism. An intuitive example of principal bundle is constructed as follows: let $G$ be a Lie group acting freely on a manifold $M$, then $M/G$ is a well defined manifold and, by Lemma \ref{lemma_0_principal_bundles}, $(M,\pi,M/G,G)$ is a principal bundle.\\

We denote by 
\begin{equation}\label{eq_0_space_of_smooth_sections}
	\Gamma^{\infty}(M\leftarrow  B)=\lbrace \varphi:M \rightarrow B, \space\ \mathrm{smooth} \space\ : \pi \circ \varphi =id_M \rbrace
\end{equation}

the space of \textit{sections} of the bundle. In physical terminology the latter are also-called \textit{field configurations} or simply \textit{fields}. We shall later on put a topology and then an infinite dimensional manifold structure on $\Gamma^{\infty}(M\leftarrow B)$. We stress that for a generic bundle the space of global smooth sections might be empty, for example when considering nontrivial principal bundles, therefore we assume that whenever we invoke the presence of fields, $\Gamma^{\infty}(M\leftarrow  B)$ is not empty. This is always the case, for example, when dealing with vector or trivial bundles.\\

Another notion of interest, necessary for calculus of variations, are jet bundles. Heuristically they geometrically formalize PDEs. For general references see \cite{kolar2013natural} chapter IV section $12$ or \cite{jet}. Rather than giving the most general definition, we simply recall the bundle case. Given any fiber bundle $(B,\pi,M,F)$, two sections, $\varphi_1$, $\varphi_2$ are $k$th-order equivalent in $x \in M$, which we write $\varphi_1 \sim_{x}^k \varphi_2$, if for all $f \in C^{\infty}(B)$, $\gamma \in C^{\infty}(\mathbb{R},M)$ having $\gamma(0)=x$, the Taylor expansion at $0$ of order $k$ of $f\circ \varphi_1 \circ \gamma$ and $f\circ \varphi_2 \circ \gamma$ coincide. The relation $ \sim_{x}^k $ becomes then an equivalence relation and we denote by $j^k_x\varphi$ the equivalence class with respect to $\varphi$. Setting $J^k_xB \doteq \Gamma^{\infty}_x(M\leftarrow B)/\sim_{x}^k$, where $\Gamma^{\infty}_x(M\leftarrow B)$ are the germs of local sections of $B$ defined on a neighborhood of $x$. The $k$th order \textit{jet bundle} is 
$$
	J^kB \doteq \bigsqcup_{x\in M} J^k_xB. 
$$
The latter inherits the structure of a fiber bundle with base either $M$, $B$ or any $J^lB$ with $l<k$. If $\lbrace x^{\mu},y^j \rbrace$ are fibered coordinates on $B$, we induce fibered coordinates $ \lbrace x^{\mu},y^j,j^j_{\mu},\ldots,y^j_{\mu_1 \ldots \mu_k} \rbrace$ where Greek indices are symmetric. The latter coordinates embody the geometric notion of PDEs. The differential $d$ of the manifold $J^kB$ can be the split into a horizontal $d_h$ and a vertical $d_v$ component in $J^{k+1}B$, this defines a double complex for the differential forms of $\Omega^{\bullet}(J^kB)$. In particular the horizontal differential can be written in coordinates as 
$$
    d_{\mu} = \frac{\partial}{\partial x^{\mu}} + y^i_{\mu}\frac{\partial}{\partial y^i} + \ldots + y^i_{\mu \lambda_1\cdots \lambda_k} \frac{\partial}{\partial y^i_{\lambda_1\cdots \lambda_k}}.
$$
Given the family $\lbrace (J^rB,\pi^r) \rbrace_{r \in \mathbb{N}}$ with $\pi^r:J^rB\rightarrow M$ we can calculate its inverse limit $(J^{\infty}B,\pi^{\infty})$ called the \textit{infinite jet bundle} over $M$, it can be seen as a fiber bundle where the standard fiber $\mathbb{R}^{\infty}$ is a Fréchet topological vector space. Its sections denoted by $j^{\infty}\varphi$ are called \textit{infinite jet prolongations}.\\

A \textit{principal connection} on a principal fiber bundle $(P,\pi,M,G)$, is a mapping $\Gamma: TP \rightarrow VP$ such that $\Gamma \circ \Gamma = \Gamma$, $\mathrm{Im}(\Gamma)=VP$, $T_p r_g \circ \Gamma=\Gamma \circ T_p r_g$ for each $p \in P$ and $g\in G$. Usually for applications we write the connection in components with
$$
	\omega = dx^{\mu} \otimes \left( \partial_{\mu} - \omega^A_{\mu}(x) r_A \right)
$$
where $\lbrace r_A \rbrace_{A=1,\ldots,\mathrm{dim}(G)}$ is a basis of the Lie algebra $\mathfrak{g}$ of $G$ made by right invariant vector fields. It is possible to show that any principal fiber bundle admits a principal connection (see \cite[Lemma 11.3]{kolar2013natural}). Using the jet bundle formalism, one can define the so-called bundle of connection as $\mathrm{Conn}(P)\doteq J^1P / G$ where the quotient is taken over the equivalence relation induced by the jet prolongation of order one of the principal right action. Choosing a connection amounts to selecting a section of the latter bundle over $M$. Given a connection $\omega$ one can define a notion of covariant derivative of a section of the principal bundle, $\rho$, along a vector field $\xi \in \mathfrak{X}(M)$ as follows:
$$
	\nabla_{\xi}\rho \doteq T \rho (\xi) - \hat{\omega}(\xi) \circ \rho
$$
where $\hat{\omega}$ is the horizontal lift of vector fields of $M$ to horizontal bundle $HP \doteq \mathrm{ker}(\omega) \subset TP$. This latter application is uniquely determined once a connection $\omega$ is chosen.\\

\subsection{Peetre-Slov\'ak's Theorem}\label{section_Peetre_slovak}

A \textit{differential operator} is a mapping between sections of smooth bundles $P:\Gamma^{\infty}(M \leftarrow B)\to \Gamma^{\infty}(M \leftarrow C)$. Generally speaking, the mapping $P$ can be as complicated as possible, the most simple example of differential operator we can think of is the one induced by a fibered morphism 
% \begin{center}
% \begin{tikzcd}
% 		 & J^rB \arrow[r,"\Psi"] \arrow[d,"\pi^k"]  & C \arrow[d, "\rho"] \\
% 		 & M\arrow[r,"\psi"] & N \arrow[l, bend left=50, "\eta" ]
% \end{tikzcd}
% \end{center}	
\begin{center}
\begin{tikzcd}
		 & J^rB \arrow[r,"\Psi"] \arrow[d,"\pi^k"]  & C \arrow[d, "\rho"] \\
		 & M\arrow[r,"\psi"] & M
\end{tikzcd}
\end{center}
Given a section $\varphi$ of $B\to M$ we can manufacture a differential operator $P$ by requiring 
\begin{equation}\label{eq_A_finite_order_diff_op}
    P[\varphi](x)= \Psi \circ j^r\varphi (x) .
\end{equation}
Whereas general differential operators might be very complicated, this particular one is quite nice since it is induced at the bundle level; it is clear however that not all differential operators can have this form. We therefore seek an answer to the question of what conditions one needs to impose on $P$ to guarantee that it takes the form \eqref{eq_A_finite_order_diff_op}. % A necessary condition is readily apparent: $P$ ought to be $\eta$-\textit{local} (or \textit{local} in case $\eta=id_M$), that is $P[\varphi](y)$ depends only on the germ of $\varphi $ at $x=\eta(y)$. The Peetre-Slov\'ak theorem answers this important question. Since we do not need the extra generality we shall assume henceforth that $N\equiv M$ so that $\eta=\psi \equiv \mathrm{id}_M$.\\
A necessary condition is readily apparent: $P$ ought to be \textit{local}, that is $P[\varphi](x)$ depends only on the germ of $\varphi $ at $x$. The Peetre-Slov\'ak theorem gives sufficient conditions for differential operators $P$ to have the form \eqref{eq_A_finite_order_diff_op}.

\begin{definition}
    We say that $P:\Gamma^{\infty}(M \leftarrow B)\to \Gamma^{\infty}(M \leftarrow C)$ is a differential operator of \textit{globally bounded order} if there is $r\in \mathbb{N}$ and a fibered morphism $(\Psi,\mathrm{id}_M):J^rB\to C$ such that $P[\varphi](x)= \Psi\circ j^r\varphi (x)$; furthermore $P$ is a differential operator of \textit{locally bounded order} if, given any $x\in M$ and $\varphi_0\in \Gamma^{\infty}(M \leftarrow B)$, there are a compact neighborhood $K \ni x$, an open neighborhood $\mathcal{U}=\{\varphi : j^r\varphi(K)\subset V\subset J^rB \} $ of $\varphi_0$ in the $CO^{\infty}$-topology \footnote{The $CO^{\infty}$ topology on $C^{\infty}(M,N)$, where $M,\ N$ are differentiable manifolds, is the topology for which $j^{\infty}:C^{\infty}(M,N) \to \big( C(M,J^{\infty}(M,N), \tau_{CO}\big)$ is an embedding where $\tau_{CO}$ is the compact open topology generated by subsets of the form \eqref{eq_1_CO_open}. Then since $\Gamma^{\infty}(M \leftarrow B)\subset C^{\infty}(M,B)$ we can endow it with the subset $CO^{\infty}$-topology.} and a smooth mapping $\Psi: V \to C$ such that $P[\varphi](x)=\Psi \circ j^r_x \varphi$ for all $\varphi\in \mathcal{U}$.
\end{definition}

Before stating the theorem, let us mention the other condition entering as hypothesis besides locality. Consider a jointly smooth family of mappings $\Phi:\mathbb{R}\times M \to B$ such that
\begin{itemize}
    \item for each $t\in \mathbb{R}$ fixed, $\Phi_t:M\to B$ is a smooth section;
    \item there is a compact subset $H \subset M$ such that for all $x \notin H$, $\Phi_t(x)$ is constant in $t$.
\end{itemize}
We call such a family of mappings a one parameter \textit{compactly supported variation}. We say that a differential operator $P:\Gamma^{\infty}(M \leftarrow B)\to \Gamma^{\infty}(M \leftarrow C)$ is \textit{weakly regular} if given any compactly supported variation $\Phi$, the mapping $P[\Phi]:\mathbb{R}\times M \to C$, defined by $P[\Phi](t,x)=P[\Phi_t](x)$, is again a compactly supported variation.

\begin{theorem}[Peetre-Slov\'ak's Theorem]\label{thm_A_Peetre_Slovak}
    Let $B\to M$, $C\to M$ be bundles and $P:\Gamma^{\infty}(M \leftarrow B)\to \Gamma^{\infty}(M \leftarrow C)$ be a local, weakly regular mapping, then $P$ is of locally bounded order.
\end{theorem}

The proof of this theorem can be found in \textit{e.g.} \cite[$\S$ 19.7]{kolar2013natural}, \cite{slovak1988peetre} or \cite{navarro2014peetre}. To get more out of this result, \textit{i.e.} globally bounded order of the operator one needs extra hypothesis, some of which are discussed in \cite{zajtz1999nonlinear}.

\subsection{Differential operators}\label{section_differential_operators}

Lastly, we further study differential operators of global bounded order between sections of vector bundles (Definition \ref{def_1_differential_operator}) among those we are interested in normally hyperbolic operators (Definition \ref{def_1_normal_hyp}). Those are special since are almost invertible: they admit retarded and advanced Green operators (\textit{c.f.} in \Cref{thm_1_properties_of_Green_functions}).\\

Let $E \to M$ and $ F \rightarrow M$ be vector bundles over the same $n$ dimensional manifold $M$.

\begin{definition}\label{def_1_differential_operator}
A linear mapping $D: \Gamma^{\infty}(M\leftarrow E) \rightarrow \Gamma^{\infty}(M \leftarrow F)$ is a differential operator of order $k$ if:
\begin{itemize}
\item[$(i)$] for any open subset $U\subset M$, $D$ can be restricted to a linear map $D_U: \Gamma^{\infty}(U\leftarrow E|_U) \rightarrow \Gamma^{\infty}(U\leftarrow F|_U)$ such that $ \forall \sigma \in \Gamma^{\infty}(U\leftarrow E)$
$$
	D_U \left( \sigma_U \right)= \left( D\sigma \right)_U;
$$
\item[$(ii)$] given any local chart $(U,x)$ there exists smooth coefficients $D^{i_1, \ldots, i_r , \beta}_{\alpha}$ totally symmetric in the Latin indices $i_{\cdot}$, such that
$$
	Ds|_U = \sum_{0 \leq r\leq k} D^{\mu_1, \ldots, \mu_r , j}_{k} \frac{\partial^{r}s^{k}}{\partial x^{\mu_1} \ldots \partial x^{\mu_r} }f_{j},
$$
where $e_{k}\in \Gamma^{\infty}(U\leftarrow E|_U)$, $f_{j}\in \Gamma^{\infty}(U\leftarrow F|_U)$ are the local base sections of the two vector bundles.
Denote by $\mathrm{DiffOp}^k(E,F)$ the set of $k$-th order differential operators.
\end{itemize}
\end{definition}

If $D \in \mathrm{DiffOp}^k(E,F)$ then the \textit{principal symbol}
\begin{equation}\label{eq_1_principal_symbol}
     \sigma_k(D)=D^{\mu_1, \ldots, \mu_k , j}_{i} \frac{\partial}{\partial x^{\mu_1} }\odot \ldots \odot \frac{\partial}{\partial x^{\mu_k}} \otimes e^{i} \otimes f_{j} 
\end{equation}
is a globally defined smooth section of $\Gamma^{\infty}(S^k TM \otimes E^{*} \otimes F)$ called the leading symbol or principal symbol of $D$.
% \begin{lemma}
% Let $D \in \mathrm{DiffOp}^k(E,F)$ then it is a continuous operator in the $\mathcal{C}^{\infty}$ topology.
% \end{lemma}
% \begin{proof}
% Fix some $K \subseteq M$ compact and $l \in \mathbb{N}$, then consider
% \begin{align*}
% 	p_{U,x,K,l,{f_{j}}}(Ds) &= \sup_{\substack {p \in K  \\ | I | \leq l \\ j}} \left| \frac{\partial^{| I | }}{\partial x^{I}} \sum_{0 \leq r\leq k} \frac{1}{r!} D^{\mu_1, \ldots, \mu_r , j}_{k} \frac{\partial^{r}s^{k}}{\partial x^{\mu_1} \ldots \partial x^{\mu_r} f_{j}} \right|  \\
% 	& \leq c \sum_{r} \sup_{\substack {p \in K  \\ | I | \leq l \\ k,j}} \left| \frac{\partial^{| I | }}{\partial x^{I}} D^{\mu_1, \ldots, \mu_r , j}_{k} \right|  \sup_{\substack {p \in K  \\ | J | \leq l \\ k}} \left| \frac{\partial^{| J | }}{\partial x^{J}}  \frac{\partial^{r}s^{k}}{\partial x^{\mu_1} \ldots \partial x^{\mu_r} f_{j}} \right| \\
% 	& \leq c' p_{U,x,K,l+k,{e_{k}}}(s).
% \end{align*}
% This ensure continuity. 
% \end{proof}

\begin{lemma}
Let $\sigma\in \Gamma^{\infty}(M\leftarrow E) $ and $\omega \in \Gamma^{\infty}_c(M\leftarrow E^{*} \otimes \Lambda_m (M))$, define a pairing 
\begin{equation}\label{eq_1_pairing} 
    (\sigma,\omega) \mapsto \langle s, \omega \rangle= \int_M \omega (\sigma).
\end{equation}
This pairing is bilinear and non-degenerate. Furthermore $\langle \sigma,f \omega \rangle=\langle f \sigma, \omega \rangle$
\end{lemma}
\begin{proof}
Suppose that $\langle \sigma, \omega \rangle=0$ for all $\omega\in \Gamma^{\infty}_c(M\leftarrow E^{*} \otimes \Lambda_m (M))$ if $\sigma(x)\neq 0$ for some $x$, then $\sigma\neq 0$ for a neighborhood $U$ of $x$ as well. Let $\psi$ be a positive bump function $\psi$ supported in $U$, if $\theta$ is a positive density with $\theta|_x\neq 0$, $\langle \sigma, \psi \theta \rangle > 0$. Thus $\sigma=0$. Similarly if $\langle \sigma, \omega \rangle=0$ for all $\sigma\in \Gamma^{\infty}(M\leftarrow E) $, then $\omega=0$. The other claim is straightforward.
\end{proof}

\begin{proposition}\label{prop_1_adjoint_differential_operator}
    Let $D \in \mathrm{DiffOp}^k(E,F)$, then its restriction $D: \Gamma^{\infty}_c(M\leftarrow E) \rightarrow \Gamma^{\infty}_c(M\leftarrow F)$ admits a unique adjoint representation with respect to the pairing defined above. Moreover, the adjoint 
    $$
        D^{T}: \Gamma^{\infty}(M\leftarrow F^{*} \otimes \Lambda_m (M )) \rightarrow \Gamma^{\infty}(M\leftarrow E^{*} \otimes  \Lambda_m (M))
    $$ 
    a differential operator of order $k$.
\end{proposition}

\begin{proof}
Fix a partition of unity $\lbrace \chi_i \rbrace_{i\in I} $ subordinate to a covering of charts $\lbrace (U_i, x^{\mu})\rbrace_{i\in I} $, then
$$ 
    D\sigma |_{U_i} = \sum_{0 \leq r\leq k} \frac{1}{r!} D^{\mu_1, \ldots, \mu_r , j}_{k} |_{U_i} \frac{\partial^{r}\sigma^{k}}{\partial x^{\mu_1} \ldots \partial x^{\mu_r} }f_{j}.
$$ 
Consider an element of $\Gamma^{\infty}(F^{*} \otimes \Lambda_m (M))$ of the form $\omega |_{U_i}= \omega_{j} \epsilon_{\mu_1 \ldots \mu_n} f^{j} \otimes dx^{\mu_1} \wedge \ldots \wedge dx^{\mu_n}$, where $ \epsilon_{\mu_1 \ldots \mu_n}$ is the well known Levi-Civita tensor and $\omega_{j}\in C^{\infty}(U)$; then 
$$
	\langle D\sigma, \omega \rangle= \int_M \omega(D\sigma) = \sum_i \int_{U_i} \chi_i \omega(D\sigma)
$$
using the partition of unity we can write in coordinates the above integral and then integrate by parts, obtaining
$$
\begin{aligned}
	&\sum_{\substack{i \in I\\ 0 \leq r\leq k}} \int_{U_i} \frac{1}{r!} \chi_i \omega_{j}   D^{\mu_1, \ldots, \mu_r , j}_{k} \frac{\partial^{r}\sigma^{k}}{\partial x^{\mu_1} \ldots \partial x^{\mu_r} } \epsilon_{\mu_1 \ldots \mu_n} dx^{\mu_1} \wedge \ldots \wedge dx^{\mu_n}\\
	& \quad = \sum_{\substack{i\in I \\ 0 \leq r\leq k}} \int_{U_i} \frac{(-1)^r}{r!}  \sigma^{k}  \frac{\partial^{r} \left( \chi_i \omega_{j}   D^{\mu_1, \ldots, \mu_r , j}_{k} \right)}{\partial x^{\mu_1} \ldots \partial x^{\mu_r} } \epsilon_{\mu_1 \ldots \mu_n} dx^{\mu_1} \wedge \ldots \wedge dx^{\mu_n},
\end{aligned}
$$
where we wrote the density measure of $M$ induced by the coordinates $\{x^{\mu}\}$ using Levi-Civita symbol $\epsilon_{\mu_1\ldots \mu_n}$. Clearly 
$$
\begin{aligned}
    {\mu}_{(i)} &\doteq \sum_{\substack{ 0 \leq r\leq k}}  \frac{(-1)^r}{r!}\frac{\partial^{r} }{\partial x^{\mu_1} \ldots \partial x^{\mu_r} } \left( \chi_i \omega_{j}   D^{\mu_1, \ldots, \mu_r , }_{k} \right) \mu_{} e^{k} \epsilon_{\mu_1 \ldots \mu_m} dx^{\mu_1} \wedge \ldots \wedge dx^{\mu_m}\\
    &  = \sum_{\substack{ 0 \leq r\leq k}}  \frac{(-1)^r}{r!}\frac{\partial^{r} }{\partial x^{\mu_1} \ldots \partial x^{\mu_r} } \left( \chi_i \omega_{j}   D^{\mu_1, \ldots, \mu_r , }_{k} \right) \mu_{} e^{k} \epsilon_{\mu_1 \ldots \mu_m} dx^{\mu_1} \wedge \ldots \wedge dx^{\mu_m}
\end{aligned}
$$
is a section of $\Gamma^{\infty}_{c}(U_i\leftarrow (E^{*} \otimes  \Lambda_n (M))|_{U_i} )$, furthermore, it can be checked directly that $\sum_i\widetilde{\mu}_i\in \Gamma^{\infty}_{c}(M\leftarrow E^{*} \otimes  \Lambda_n (M) )$. Setting
\begin{equation}\label{eq_1_adjoint_diff_operator}
    D^{\dagger}\omega = \sum_i {\omega}_{(i)},
\end{equation}
we get
$$
	\langle D\sigma, \omega \rangle = \langle \sigma, D^{\dagger} \omega \rangle.
$$
It is clear from the local form that $D^{\dagger}$ is of order $k$. 
\end{proof}

\begin{definition}\label{def_1_normal_hyp}
Let $E\to M$ be a vector bundle with base a Lorentzian manifold, a second order differential operator $D:\Gamma^{\infty}(M\leftarrow E) \rightarrow \Gamma^{\infty}(M\leftarrow E)$ is called \textit{normally hyperbolic} if its principal symbol can be written as 
$$
	\sigma_2(D) =\frac{1}{2} g^{-1}\otimes \mathrm{id}_E
$$
where $g$ is the Lorentzian metric on $M$.
\end{definition}

\begin{lemma}\label{lemma_0_wave_eq_form_of_normally_hyp_op}
    Let $D$ be a normally hyperbolic differential operator, then one can always choose an appropriate fiber connection on $E$ such that for all $\sigma\in \Gamma^{\infty}(M\leftarrow E)$,
    $$
        D\sigma= \Big(g^{\mu\nu} \nabla_{\mu\nu}\sigma^i+ B_i^j\sigma^i\Big)e_j,
    $$
    where $\nabla$ is the covariant derivative induced from the connection.
\end{lemma}
\begin{proof}
    Fix some local chart $(x^{\mu},e^i)$ on $E$, then we can represent $D\sigma$ as 
$$
	\Big(\frac{1}{2}D^{\mu \nu j }_i \partial_{\mu \nu}\sigma^i + D^{\mu j}_i\partial_{\mu}\sigma^i + D_i^j\sigma^i \Big)e_j.
$$
Any connection on $E$ induces a covariant derivative $\nabla$: given $X\in \mathfrak{X}(M), \ \sigma \in \Gamma^{\infty}(M\leftarrow E)$ we have $\nabla_X\sigma=X^{\mu}\nabla_{\mu}\sigma^ie_i$ where $\nabla_{\mu}s^i=\partial_{\mu}s^i +\Gamma^i_{j \mu}s^j$. Using a torsionless connection on $M$ (\textit{e.g.} the Levi-Civita connection of the spacetime metric $g$) it is possible to calculate 
$$
	\nabla_{\mu}\nabla_{\nu}s^i=\partial_{\mu \nu}s^i+\Gamma^i_{j \mu} \nabla_{\nu}s^j- \Gamma^{\alpha}_{\mu \nu}\nabla_{\alpha}s^i+ \partial_{\mu}\Gamma^i_{j \nu}s^j + \Gamma^i_{j \nu} \nabla_{\mu}s^j- \Gamma^i_{j \nu} \Gamma^j_{k \mu}s^k.
$$
Substituting into the expression for $D\sigma$, we find 
$$
\begin{aligned}
    D\sigma=\Big(\frac{1}{2} D^{\mu \nu j}_{i}\nabla_{\mu}\nabla_{\nu}\sigma^i + \Big[\frac{1}{2}D^{\alpha\beta j}_{i}\Gamma^{\mu}_{\alpha\beta}+D^{\mu j}_i- \Gamma^i_{k\nu}D^{\mu\nu j}_i\Big]\nabla_{\mu}\sigma^i +\Big[\Gamma^l_{k\mu}\Gamma^k_{i\nu}D^{\mu\nu j}_l- D^{\mu j}_k\Gamma^k_{i\mu}+ D^{\mu}_i\Big]\sigma^i \Big)e_j.
\end{aligned}
$$
Since $D^{\mu\nu j}_i=g^{\mu\nu} \delta^j_i$, we can set
$$
    \Gamma^j_{k\mu}\doteq g_{\mu\nu}\Big(\frac{1}{2}D^{\alpha\beta j}_k + D^{\nu j}_k\Big)
$$
It can be shown, rather laboriously, that the above equation defines the Christoffel symbols of a connection on $E$. 
\end{proof}

Having developed the notion of Cauchy hypersurface in Definition \ref{def_1_Cauchy_hypersurface}, we can give a precise formulation of the Cauchy problem. One the one hand, a Cauchy surface on a globally hyperbolic spacetime will secure the well-posedness for any evolution type problem on $M$, on the other, to talk about evolution of quantities (that is fields) in $M$ we need some dynamic: fortunately enough, we have discussed differential operators. As usual we assume that $(M,g)$ is a Lorentzian, time-oriented globally hyperbolic manifold and $\Sigma$ is a spacelike Cauchy hypersurface. Since at any point $x \in \Sigma$, $T_x \Sigma \subseteq T_xM$ is spacelike, there is a unique unitary, timelike, future directed vector $ \mathfrak{n}_x \in T_x M$, that is 
$$
	g(\mathfrak{n}_x,v)=0 \space\  \space\ \forall  \space\ v \in T_x\Sigma,
$$
$$
	g(\mathfrak{n}_x,\mathfrak{n}_x)=1.
$$
The vector field $\mathfrak{n}: \Sigma \ni x \mapsto \mathfrak{n}_x $ is called the \textit{future directed normal vector field} of $\Sigma$. \\

Let $D: \Gamma^{\infty}(M\leftarrow E) \rightarrow \Gamma^{\infty}(M\leftarrow E) $ be a normally hyperbolic differential operator. We consider the wave equation
$$
	D\sigma = g^{\mu\nu}\nabla_{\mu\nu} \sigma + B\sigma =\theta 
$$
$\sigma, \theta \in \Gamma^{\infty}(M\leftarrow E) $ and the related Cauchy problem
\begin{align*}
	D\sigma &=\theta,\\
	\sigma|_{\Sigma}&=\sigma_0 \in \Gamma^{\infty}(M\leftarrow E |_{\Sigma}),\\
	\nabla_{\mathfrak{n}}\sigma |_{\Sigma} &= \sigma_1 \in \Gamma^{\infty}(M\leftarrow E|_{\Sigma}).
\end{align*}

We try to solve this problem by constructing distributional solution $F_x \in \Gamma^{- \infty}(M\leftarrow E) \otimes E^{*}_x \otimes  \Lambda_n (M)_x$ such that
$$
	D F_x=\delta_x
$$ 
where $\delta_x$ is viewed as a $E^{*}_x \otimes \Lambda_m (M)_p$-valued distributional section of $E$, that is if $\mu \in \Gamma_0^{\infty}(E^{*} \otimes \Lambda_m (M))$, then $\delta_x(\mu)= \mu(x) \in E^{*}_x \otimes  \Lambda_n (M)_x$. If $F_p$ satisfies $D F_x=\delta_x$ is called \textit{fundamental solution of} $D$ at $x\in M$. If, moreover, there exists fundamental solutions $F^{\pm}_x $ satisfying
$$
	\mathrm{supp} (F^{\pm}_x) \subseteq J^{\pm}(x),
$$
then $F^{\pm}_x$ is called \textit{advanced} or \textit{retarded fundamental solution} of $D$ at $x\in M$.

\begin{definition}\label{def_1_Green_functions}
Let $(M,g)$ be a globally hyperbolic Lorentzian manifold, $D$ a normally hyperbolic differential operator, then a continuous linear map 
$$
    G^{\pm}:\Gamma_c^{\infty}(M\leftarrow E) \rightarrow \Gamma^{\infty}(M\leftarrow E) 
$$
satisfying
\begin{itemize}
\item[$(i)$] $D\circ G^{\pm}=\mathrm{id}_{\Gamma_c^{\infty}(M\leftarrow E)}$,
\item[$(ii)$] $G^{\pm} \circ \left. D \right|_{\Gamma_c^{\infty}(M\leftarrow E)}=\mathrm{id}_{\Gamma_c^{\infty}(M\leftarrow E)}$,
\item[$(iii)$] $\mathrm{supp}(G^{\pm}(\sigma)) \subseteq \overline{ J^{\pm}(\mathrm{supp}(\sigma))}$ for all $\sigma \in \Gamma_c^{\infty}(M\leftarrow E)$,
\end{itemize}
is called \textit{advanced/retarded Green operator} for $D$.
\end{definition}

\begin{theorem}\label{thm_1_existence_adv/ret_fundamental_solutions}
Let $(M,g)$ be globally hyperbolic, and $D$ in $\mathrm{DiffOp}^2(E)$ normally hyperbolic. For every $x \in M$ there is a unique global advanced and retarded fundamental solution 
$$
    F^{\pm}: \Gamma^{\infty}_c(M\leftarrow E^{*}) \ni \mu \mapsto F^{\pm}( \mu) \in \Gamma^{\infty}(M\leftarrow E^{*})
$$
of $D^{\dagger}$ at $x$, that is 
$$ 
    D^{\dagger} F^{\pm}( \mu) = \mu 
$$
is a continuous mapping in the appropriate topologies.
\end{theorem}

% In the above representation it is understood that we pair (via tensor product) sections in $\Gamma^{\infty}(M\leftarrow E^{*})$ with the smooth measure $\sqrt{| g |} d\sigma=\sqrt{| g |} dx^1\wedge\ldots\wedge dx^m$ to form a section of $\Gamma^{\infty}_{c}(M\leftarrow E^{*} \otimes \Lambda_m(M ))$.

\begin{theorem}
Let $(M,g)$ be globally hyperbolic, $D \in \mathrm{DiffOp}^2(E)$ be normally hyperbolic. Assume that $\lbrace F^{\pm}_x \rbrace_{x \in M}$ is a family of global advanced and retarded fundamental solutions of $D^{\dagger}$, with the properties of Theorem \ref{thm_1_existence_adv/ret_fundamental_solutions}, then given $\sigma \in \Gamma^{\infty}_c(M\leftarrow E)$ define, using the pairing \eqref{eq_1_pairing}, 
$$ 
    \mu \left( G^{\pm} \sigma \right)=  F^{\mp} (\mu)( \sigma ).
$$
The mappings $G^{\pm}:\Gamma_c^{\infty}(M\leftarrow E) \rightarrow \Gamma^{\infty}(M\leftarrow E)$ defined above are Green operators for $D$. 
%\begin{itemize}
%\item[$(i)$]
% \item[$(ii)$] Assume that $\lbrace G^{\pm}_x\rbrace_{x \in M}$ are a family of global advanced and retarded Green operators of $D$, with the properties of Theorem C.4, then for some $\phi \in \Gamma^{\infty}_{c}(M\leftarrow E^{*})$ define $$ F_{M}^{\mp} (\phi) u=\phi   \left( G_{M}^{\pm} u \right), $$then the $F_{M}^{\pm}(p)$ defined  above are Green operators for $D^T$.
%\end{itemize}
\end{theorem}
\begin{proof}
From
$$
	\mu \left( D  G^{\pm} \sigma \right)= D^{\dagger} (\mu) \left(  G^{\pm} \sigma \right) = F^{\mp} (D^{\dagger} \mu) \sigma = \mu (\sigma)
$$
and 
$$
	\mu\left(   G^{\pm} D\sigma\right)= F^{\mp} (\mu) \left(  D \sigma \right) = D^{\dagger} F^{\mp} (\mu) \sigma = \mu (\sigma)
$$
we immediately get the first two properties of Green operators. The continuity follows from the fact that $F^{\mp}$ is continuous and the above calculations. We then look at supports which explains the flip from $\pm$ to $\mp$. Fix a section $\mu \in \Gamma^{\infty}_{c}(M\leftarrow E^{*}\otimes  \Lambda_n (M))$, then  $\mathrm{supp}(F^{\mp}(\mu)) \subseteq J^{\mp}(\mathrm{supp}(\mu))$, so $F^{\mp} (\mu)( \sigma )$ vanish if $ J^{\mp}(\mathrm{supp}(\mu)) \cap \mathrm{supp}(\sigma)=\emptyset$ which is equivalent to say $ J^{\pm}(\mathrm{supp}(\sigma)) \cap \mathrm{supp}(\mu)=\emptyset$. Thus, if $x \notin J^{\pm}(\mathrm{supp}(\sigma))$ and we choose some $\mu$ with  $ J^{\pm}(\mathrm{supp}(\sigma)) \cap \mathrm{supp}(\mu)=\emptyset $ then $\mu \left( G^{\pm} \sigma \right)=  F^{\mp} (\mu)( \sigma )=0$, so $ G^{\pm} \sigma=0$, which proves that $\mathrm{supp}(G^{\pm} \sigma) \subseteq  J^{\pm} (\mathrm{supp}(\sigma))$. The other statement is analogous.  
\end{proof}

% \begin{theorem}
% Let $(M,g)$ be globally hyperbolic, and $D$ in $\mathrm{DiffOp}^2(E)$ normally hyperbolic. Suppose also that $G_{M}^{\pm}$ are the advanced and retarded Green operators for $D$.
% \begin{itemize}
% \item[$(i)$] The Green operators $G_{M}^{\pm}: \Gamma_c^{\infty}(M\leftarrow E) \rightarrow \Gamma^{\infty}(M\leftarrow E)$ can be extended by continuity to generalized operators (which we denote with the same symbol) satisfying 
% $$DG_{M}^{\pm} = id _{\Gamma_c^{-\infty}(M\leftarrow E)},$$ $$G_{M}^{\pm}\left.D\right|_{\Gamma_c^{-\infty}(M\leftarrow E) }= id _{\Gamma_c^{-\infty}(M\leftarrow E)}.$$
% \item[(ii)] For every $v \in \Gamma_c^{-\infty}(M\leftarrow E)$ and spacelike Cauchy hypersurface $i:\Sigma : \rightarrow M$ with $\mathrm{supp}(v) \subseteq I^+_M(\Sigma)$ and all $u_0$, $\dot{u}_0$ $\in \Gamma_c^{\infty}(i^{*}E)$ there is a unique $u_+ \in \Gamma^{-\infty}(M\leftarrow E)$ with
% \begin{gather*}
% 	Du_+=v ,\\
% 	\mathrm{supp}(u_+) \subseteq J_M(\mathrm{supp}(u_0) \cup \mathrm{supp}(\dot{u_0})) \cup J_M(\mathrm{supp}(v)),\\
% 	\mathrm{singsupp}(u_+) \subseteq J_M(\mathrm{supp}(v))	,\\
% 	i^*u_+=u_0 \space\ \space\ \mathrm{and} \space\ \space\ i^* \nabla_{\mathfrak{n}}u_+ = \dot{u}_0.
% \end{gather*}
% tThe section $u_+$ also depends weak* continuously on $v$ and continuously on $u_0$, $\dot{u}_0$.
% \item[(iii)] In analogy the case with $\mathrm{supp}(v) \subseteq I^-_M(\Sigma)$.
% \end{itemize} 
% \end{theorem}

The final result sums up everything there is to know for Green functions:

\begin{theorem} \label{thm_1_properties_of_Green_functions}
Let $E\to M$ be a vector bundle with base a globally hyperbolic Lorentzian manifold and let $D \in \mathrm{DiffOp}^2(E)$ be a normally hyperbolic differential operator. Then $D$ admits global Green operators $G_M^{\pm}: \Gamma^{\infty}_c(M\leftarrow E) \rightarrow \Gamma^{\infty}(M\leftarrow E)$ and their causal propagator $G = G^+-G^-$, satisfying the following properties:
\begin{itemize}
\item[$(i)$] \textbf{Continuity.} $G^{\pm}$, $G$ are continuous linear operators where $\Gamma^{\infty}_c(M\leftarrow E)$ has the LF-topology described in \ref{ex_A_Frechet_structure_comp_supp_smooth_funct} and $\Gamma^{\infty}(M\leftarrow E) $ the Fréchet topology described in \ref{ex_A_smooth_function_frechet_space}; and admit a continuous and linear extension to the spaces $\Gamma^{-\infty}_{\pm}(M\leftarrow E)$ topological dual to $\Gamma^{\infty}_{\mp}(M\leftarrow E)=\lbrace \sigma \in \Gamma^{\infty}(M\leftarrow E) :\space\ \forall x\in M \space\ \mathrm{supp}(\vec{u})\cap J^{\mp}(x) $ is compact $\rbrace$. 
\item[$(ii)$] \textbf{Support Properties.} 
$$
	\mathrm{supp}(G^{\pm} u) \subset J^{\pm}(\mathrm{supp}(u)
$$
for all $u \in \Gamma^{-\infty}_{\pm}(M\leftarrow E)$.
\item[$(iii)$] \textbf{Cauchy problem.} For every $v \in \Gamma_c^{-\infty}(M\leftarrow E)$ and spacelike Cauchy hypersurface $i:\Sigma \hookrightarrow M$ with $\mathrm{supp}(v) \subseteq I^{\pm}(\Sigma)$ and all $u_0$, $\dot{u}_0$ $\in \Gamma_c^{\infty}(M\leftarrow i^{*}E)$ there is a unique $u_{\pm} \in \Gamma^{-\infty}(M\leftarrow E)$ with
\begin{gather*}
	Du_{\pm}=v ,\\
	\mathrm{supp}(u_{\pm}) \subseteq J^{\pm}\big(\mathrm{supp}(u_0) \cup \mathrm{supp}(\dot{u_0})\big) \cup J^{\pm}(\mathrm{supp}(v)).
\end{gather*}
The section $u_{\pm}$ also depends continuously on $v$, $u_0$ and $\dot{u}_0$.
\item[$(iv)$] \textbf{Wave Front Sets.} 
\begin{align*}
	\mathrm{WF}(G^{\pm}) &= \{ (x,x;\xi,-\xi)\in \dot{T}^{*}(M \times M) \, \, \mathrm{with} \ (x;\xi) \in \dot{T}^{*}M  \} \\ &\quad \cup \{ (x_1,x_2;\xi_1,\xi_2) \in \dot{T}^{*}(M\times M) \, \mathrm{with} \, (x_1,x_2;\xi_1,-\xi_2) \in \mathrm{BiCh}^{g}_{\gamma} \}, \\
	\mathrm{WF}(G_M) &= \{ (x_1,x_2;\xi_1,\xi_2)\in \dot{T}^{*}(M\times M)\, \, \mathrm{with} \,\,  (x_1,x_2;\xi_1,-\xi_2) \in \mathrm{BiCh}^{g}_{\gamma} \};
\end{align*}
where $ \mathrm{BiCh}^{g}_{\gamma}$ is the bicharacteristic strip of the lightlike geodesic $\gamma$, \textit{i.e.} the set of points $(x_1,x_2;\xi_1,\xi_2)$ such that there is an interval $[0,\Lambda] \subset \mathbb{R}$ for which $(x_1,\xi_1)=(\gamma(0), g^{\flat}\dot{\gamma}(0))$ and $(x_2,-\xi_2)=(\gamma(\Lambda), g^{\flat}\dot{\gamma}(\Lambda))$.
\item[$(v)$] \textbf{Propagation of Singularities.} Given $u \in \Gamma^{-\infty}_{\pm}(M \leftarrow E)$, $(x,\xi) \in \mathrm{WF}(G^{\pm}(u))$ if either $(x,\xi) \in \mathrm{WF}(u)$ or there is a lightlike geodesic $\gamma$ and some $(y;\eta) \in \mathrm{WF}(u)$ such that $(x,y;\xi,-\eta) \in \mathrm{BiCh}^{g}_{\gamma}$. Similarly, $(x,\xi) \in \mathrm{WF}(G(u))$ if there is a lightlike geodesic $\gamma$ and some $(y,\eta) \in \mathrm{\mathrm{WF}}(u)$ such that $(x,y;\xi,-\eta) \in \mathrm{BiCh}^{g}_{\gamma}$.
\end{itemize}
\end{theorem}

We remark that in the above theorem the notation $\Gamma^{-\infty}(M\leftarrow E)$ denotes distributional sections of the vector bundle $E$, \textit{i.e.} continuous linear mappings $\Gamma^{\infty}_c(M\leftarrow E) \rightarrow \mathbb{C} $, where the first space is endowed with the usual limit Fréchet topology. We also recall that the wave front set at $x\in M$ of a distributional section ${u}$ of a vector bundle $E$ of rank $k$ is calculated as follows: fixing a trivialization $(U_{\alpha},t_{\alpha})$ on $E$, then locally ${u}$ is represented by $k$ distributions ${u}^i\in \mathcal{D}'(U_{\alpha})$, each of which will have its own wave front set. 
Then we set
\begin{equation}\label{eq_1_def_WF_sections}
\mathrm{WF}({u})\doteq \bigcup_{i=1}^k \mathrm{WF}\big({u}^i\big).
\end{equation}
It is possible to show that choosing a different trivialization in the vector bundle give rise to a smooth vertical fibered morphism which does not alter the wave front set, as a result \eqref{eq_1_def_WF_sections} is independent of the trivialization chosen and hence well defined.

\section{Locally convex vector spaces and calculus on them}\label{section_LCS&calculus}

\subsection{Locally convex spaces}

%--------------------- DO A GENERAL ROUNDUP ON LCS, BASIC DEFINITIONS, EXAMPLES SUCH AS c^{\INFTY}, \GAMMA^{\INFTY} AND THE LIKES OF THEM--------------------------------------------------------------
Generally speaking a \textit{locally convex space} (LCS) $E$ is a \textit{topological vector space} (TVS), \textit{i.e.}, a vector space $E$ with a topology compatible with vector addition and scalar multiplication, such that the basis of neighborhoods of $0\in E$ is made by convex subsets. The reason for this definition is motivated by the fact that it would be desirable to have, as neighborhoods of $0$, the closest things possible to balls. Convexity is therefore the condition allowing this construction; however, it should be noted that in general TVSs there is no guarantee of finding such a basis. It can be shown (\textit{e.g.} in \cite[Proposition 7.6]{treves2016topological}) that a LCS is equivalently described by a family of \textit{seminorms} generating the basis of the topology. A seminorm is a mapping $p:E \to \mathbb{R}$ continuous in the vector space topology of $E$ satisfying
\begin{itemize}
    \item[$(i)$] $p(x)\geq 0 $ for all $x\in E$;
    \item[$(ii)$] $p(x+y)\leq p(x)+p(y)$ for all $x, \ y \in E$;
    \item[$(iii)$] $p(\alpha x)=|\alpha|p(x)$ for all $x\in E, \ \alpha\in \mathbb{R}$.
\end{itemize}
Most of the times we do not describe locally convex spaces by giving a TVS structure admitting a base of convex subsets for the topology, instead, we give a family of seminorms and then induce a topology from this family checking that the latter is compatible with the vector space structure. Clearly each Banach space is also a locally convex one, the next best things are \textit{Fréchet spaces}. Those are complete, metrizable locally convex spaces. Before giving some relevant examples, we describe two more properties that locally convex spaces can have. A LCS $E$ is \textit{barreled} if every absorbing, balanced, closed, convex subset $U$\footnote{A subset $U$ of $E$ is called: absorbent if given any $x\in E$, there is $\lambda\in \mathbb R$ such that $x\in \lambda U$; balanced if $x\in U \Rightarrow -x \in U$.} is a neighborhood of $0\in E$. It is a \textit{Montel} space if it is Hausdorff, barreled and if every closed bounded subset is compact. Montel spaces are very fascinating: indeed, the Banach-Alaoglu theorem states that, in Banach spaces, closed and bounded subsets of the strong dual are never compact, but only weakly compact and therefore weakly convergent sequences in dual are generally not strongly convergent; however, in Montel spaces closed bounded subsets of the strong dual are compact (\cite[Proposition 34.6]{treves2016topological}) and therefore weak and strong convergence in dual of Montel spaces are equivalent. The second property we wish to introduce is \textit{nuclearity}. We can introduce this property by trying to give a topology to the space $E\otimes F$ given locally convex spaces $E, \ F$. There are essentially two ways to do this: on the one hand, we can regard $E\times F \to E\otimes F$ as a quotient space and endow it with the quotient topology, we denote it by $E\otimes_{\pi} F$; on the other hand, if we denote by $E'_{\sigma}$ the dual of $E$ with the topology of pointwise convergence, then one can show (see for instance \cite[Proposition 42.4]{treves2016topological}) that $E\otimes F$ is isomorphic to the space $\mathfrak{B}(E'_{\sigma}, F'_{\sigma})$ of continuous bilinear mappings $E'_{\sigma}\times F'_{\sigma}\to \mathbb{R}$. Then one can induce a topology on $E\otimes F$ by inducing, on $\mathfrak{B}(E'_{\sigma}, F'_{\sigma})$, the topology of uniform convergence on products of equicontinuous\footnote{A subset $H$ of the space of linear mappings $\mathfrak L (E,\mathbb R)$ is equicontinuous if for each neighborhood $V$ of $0\in \mathbb R$, there is a neighborhood $U$ of $0\in E$ such that for all $u\in H$, $u(U)\subset V$.} set of $E'$ and $F'$, we then denote this topology by $E\otimes_{\epsilon} F$. In general the two are not equivalent and the $\pi$-topology is finer, however, when doing tensor product of locally convex spaces it is desirable to have a unique topology and not choose between the two. We therefore call $E$ \textit{nuclear} if for any other locally convex space $F$ the completions of the tensor products $E\otimes_{\pi} F$ and $E\otimes_{\epsilon} F$ are topologically isomorphic. For additional details on nuclear spaces see \cite[Chapters 50, 51]{treves2016topological} as well as \cite{nlcs}.\\

A very important example of Fréchet space, which we use throughout this thesis, is the space $C^{\infty}(M)$ of smooth functions on a differentiable manifold $M$. To differentiate from other topologies one can give the latter space (see for instance the $WO^{\infty}$-topology in \Cref{section_topology_map_manifold}), we shall denote it by $\mathcal{E}(M)$ when endowed with its Fréchet structure.

\begin{example}[Fréchet space structure of $\mathcal{E}(M)$]\label{ex_A_smooth_function_frechet_space}
\end{example}
    We start by considering the two families of seminorms
    \begin{enumerate}
        \item[(A)] given an atlas $\{(U_i,u_i)\}_{i\in \mathbb I}$ for $M$ and $K\subset M$ compact, consider the seminorms
        \begin{equation}
            p_{K,n,U_i,u_i}: f \mapsto \sup_{i \in I} \sup_{x\in K\cap U_i, \ j\leq n} |\partial^j (f\circ u_i^{-1})(x)|;
        \end{equation}
        \item[(B)] given any (Riemannian) metric on $M$, denote by $\nabla$ its covariant derivative, then consider the seminorms
        \begin{equation}
            p_{K,n,\nabla}: f \mapsto \sup_{x\in K, \ j\leq n} |\nabla^j f(x)|,
        \end{equation}
        with $K\subset M$ compact subset.
    \end{enumerate}
    That those families are indeed seminorms can be inferred by employing basic estimates. The seminorms of (A) are clearly well defined since only finitely many $U_i$ intersect each $K$, the drawback is that they depend, a priori, on the differential structure chosen. Those of (B), instead, do not depend on the smooth structure but on the choice of some connection to carry out the covariant differentiation.
    \begin{lemma}\label{lemma_A_Frechet_equiv_seminorms}
        The family of seminorms in (A)(respectively (B)) do not depend on the smooth structure (respectively on the connection) chosen. Moreover the two families are equivalent, \textit{i.e.} there are positive real constants $c_1$, $c_2$ for which
        \begin{equation}
            c_1 p_{K,n,\nabla}(f)\leq p_{K,n,U_i,u_i}(f) \leq c_2 p_{K,n,\nabla}(f)
        \end{equation}
        for all $f \in C^{\infty}(M)$.
    \end{lemma}
    We stress that the main consequence of Lemma \ref{lemma_A_Frechet_equiv_seminorms} is that on $C^{\infty}(M)$ the topologies induced by the families of seminorms (A) or (B) do coincide and therefore do not depend on the atlas or the connection chosen to perform calculations but only on the compact subset and the differentiation order. We shall henceforth denote $p_{n,K}$ the seminorms on $\mathcal{D}(M)$.
    \begin{proof}
        That two families with two different atlases in (A) are equivalent follows directly from the fact that in the intersection $U_i\cap \widetilde{U}_j$ the transition mapping $u_i \circ \widetilde{u}_j^{-1}$ is a smooth diffeomorphism and is therefore bounded in $K\cap U_i\cap \widetilde{U}_j$. That two different connections induces equivalent families in (B) follows from induction:
        \begin{itemize}
            \item if $n=1$, then $\nabla f\equiv \partial f \equiv \widetilde{\nabla} f$;
            \item if $n$ is arbitrary then we can write 
            $$
                \nabla^nf = \partial (\nabla^{n-1}f) + \sum \Gamma(\nabla^{n-1}f)  
            $$
            where the sum is linear in the Christoffel symbols $\Gamma $ of the connection and features $n$ terms. Then
            $$
            \begin{aligned}
                \nabla^nf-\widetilde{\nabla}^nf&= \partial \big(\nabla^{n-1}f- \widetilde{\nabla}^{n-1}f\big)+ \sum (\Gamma-\widetilde{\Gamma})\nabla^{n-1}f + \sum \widetilde{\Gamma} \big(\nabla^{n-1}f- \widetilde{\nabla}^{n-1}f\big) 
            \end{aligned}
            $$
            by the induction hypothesis one can estimate the coefficients $\big(\nabla^{n-1}f- \widetilde{\nabla}^{n-1}f\big) $ appropriately, $(\Gamma-\widetilde{\Gamma})$ are the smooth coefficients of a $(1,2)$ tensor which can be estimated by a constant in the compact $K$.
        \end{itemize}
        Finally, let $\{\psi_i\}$ be a partition of unity subordinated to the atlas $\{(U_i,u_i)\}$; recall that in a chart $\nabla^n f = \partial^n f + \sum q(\partial,\Gamma)f$ where $q$ is a polynomial of order at most $n-1$, then we can estimate 
        $$
        \begin{aligned}
            p_{K,n,\nabla}(f)= & \sup_{j\leq n,K} \big|\nabla^j\big(\sum_i\psi_i f\big)\big|\leq \sup_{j\leq n, K}\sum_{i}\big|\nabla^j(\psi_i f)\big| \\ 
            \leq & \sup_{j\leq n, K}  \sum_i\Big| \big( \partial^j(\psi_if\circ u_i^{-1})\big)+ \sum_{0\leq k<j} \Gamma^{j-k}\partial^k(\psi_if\circ u_i^{-1}) \big)\Big|\\
            \leq & \sum_i p_{K,n,U_i,u_i} (\psi_i f) +\sum_{i, \ 0\leq k \leq j \leq n} C(\Gamma, j-k) p_{K,j,U_i,u_i}(\psi_if)\\
            \leq & \sum_i C_i(\Gamma)\big( \sup_{0\leq j\leq n } p_{K,j,U_i,u_i}(\psi_if)\big) \leq c_1p_{K,n,U_i,u_i}(f).
        \end{aligned}
        $$
        The other estimate follows similarly if we notice that $\nabla^nf = \partial (\nabla^{n-1}f) + \sum \Gamma(\nabla^{n-1}f)$ so that we can estimate $p_{K,1,U_i,u_i}(\nabla^{n-1}f)\leq C p_{K,n,\nabla}(f)$, iterating this procedure and we arrive at
        $$
            p_{K,n,U_i,u_i}(f) \leq c_2 p_{K,n,\nabla}(f).
        $$
    \end{proof}
    \begin{remark*}
        To check that the topology induced by any of the two seminorms in (A) or (B) is compatible with the vector space structure of $C^{\infty}(M)$, we can check that it satisfies conditions (3.1)-(3.5) in \cite[Theorem 3.1]{treves2016topological}. (3.1) holds since the zero function belongs to any convex ball $\{p_{n,K}\leq \epsilon\}$, (3.2) follows from property $(ii)$ of the seminorm, (3.3) from $\lambda \{ p_{n,K}\leq \epsilon\}=\{p_{n,K}\leq \lambda \epsilon\}$ for any $\lambda \in \mathbb{R}$ and finally, by (iii) any ball $\{p_{n,K}\leq 1\}$ is absorbing and balanced thus satisfying conditions $(3.4)$ and (3.5).

    \end{remark*}

That $C^{\infty}(M)$ is metrizable follows from that fact than since $M$ is locally compact and second countable, it admits an exhaustion by compact subsets $K_l$, then $p_{n,K_l}$ is a countable basis of the above locally compact topology. 

Finally completeness of $C^{\infty}(M)$ can be obtained by showing:

\begin{lemma}\label{lemma_0_CO_top_is_Frechet}
    The topology of $\mathcal{E}(M)$ described above coincides with the topology of uniform convergence with all derivatives on compact subsets.
\end{lemma}
\begin{proof}
If $f_n \to f$ in $\mathcal{E}(M)$, then each neighborhood $\{g\in C^{\infty}(M):\ p_{n,K}(g-f)\leq \epsilon\}$ of $f$ must contain each $f_n$ for $n $ big enough, therefore in the compact subset $K$ the functions $f_n$ and $f$ differ uniformly (together with all derivatives) by at most $\epsilon$, thus $f_n\to f$ converges uniformly in $K$ with all its derivatives. Viceversa, if $j^nf_n\to j^nf$ converges uniformly in $K$, then, for $n$ sufficiently big, $f_n \in B_{1/k}(f)\{g\in \mathcal{D}(M):\ p_{n,K}(g-f)\leq 1/k \}$.
\end{proof}

Notice also that by \cite[Theorem 51.5 together with its Corollary]{treves2016topological} $C^{\infty}(M)$ is nuclear.

\begin{example}[LF space structure of $\mathcal{D}(M)$]\label{ex_A_Frechet_structure_comp_supp_smooth_funct}
\end{example}
Another space of significative interest is the space $C^{\infty}_c(M)$ of compactly supported functions on a manifold $M$.\footnote{Here we assume that the manifold is not compact, otherwise $C^{\infty}_c(M)=C^{\infty}(M)$ can be endowed with the Fréchet space structure of Example \ref{ex_A_smooth_function_frechet_space}.} We topologize it as follows: first, note that for any compact subset $K\subset M$ the space $C^{\infty}_K(M)$ of smooth functions supported inside $K$ can be given a Fréchet space structure with seminorms $\{p_{n,K}\}_{n\in \mathbb{N}}$ analogously to Example \ref{ex_A_smooth_function_frechet_space}; secondly, given any exhaustion $\{K_n\}_{n\in \mathbb{N}}$ by compact subsets of $M$, we can see $C^{\infty}_c(M)$ as the strict countable inductive limit 
$$
    C^{\infty}_c(M)=\lim_{\substack{\longrightarrow \\ n\in \mathbb{N}}} C^{\infty}_{K_n}(M)
$$
and endow the former space with the final topology with respect to the mappings $i_n: C^{\infty}_{K_n}(M) \hookrightarrow C^{\infty}_{c}(M)$. In particular, a neighborhood $U$ of $0\in C^{\infty}_c(M)$ will be open if and only if $U\cap C^{\infty}_{K_n}(M) $ is an open neighborhood of $0$ for each $n\in \mathbb{N}$. Any such space is called a \textit{limit Fréchet} space (or LF-space for short) and by \cite[Theorem 13.1]{treves2016topological} it is complete, moreover (see for instance \cite[Proposition 3.39]{osborne2014locally}) the final topology on $C^{\infty}_c(M)$ does not depend on the exhaustion chosen. Although $C^{\infty}_c(M)$ is not Banach nor Fréchet, one can show that it is Montel and nuclear (see \cite[Proposition 34.4 and Theorem 51.5 together with its Corollary]{treves2016topological}). We will denote $\mathcal{D}(M)$ the space $C^{\infty}_c(M)$ with the final topology induced by mappings $ C^{\infty}_{K_n}(M) \hookrightarrow C^{\infty}_c(M)$ .

We stress that both examples can be straightforwardly generalized to encompass the spaces of smooth sections of vector bundles. Therefore, if $E\to M $ is a vector bundle, $\Gamma^{\infty}(M \leftarrow E)$ can be given the structure of a nuclear Fréchet space, whereas $\Gamma^{\infty}_c(M \leftarrow E)$ becomes a nuclear, Montel, LF-space. 

\subsection{Distributions}\label{section_distributions}
%-------------- DISTRIBUTIONS ------------------

% We present other examples of locally convex spaces: the spaces of distributions on manifolds. Those are the space of linear continuous mappings from $\mathcal{D}(M) \to \mathbb{R}$ (respectively $C^{\infty}(M) \to \mathbb{R}$). Finally we shall extend these to the sections of vector bundles.\\

A \textit{distribution} $u$ on $M$ is a linear continuous functional 
$$ 
    u:\mathcal{D}(M) \rightarrow \mathbb{R} 
$$
where $\mathcal{D}(M)$ is endowed with the LF-space structure. This implies that an element $u$ of the algebraic dual is a distribution if and only if 
$$
    u:C^{\infty}_K(M) \rightarrow \mathbb{R}
$$
is continuous for each compact $K\subset M$, or, equivalently, that for each $K\subset M$ compact there is $n\in \mathbb{N}$, a positive constant $c$ such that
$$
	|u(f)|\leq c \sup_{j\leq n} p_{j,K}(f), \ \ \  f \in C^{\infty}_K(M).
$$
We denote by $\mathcal{D}'(M)$ the spaces of distributions. We stress that since each $C^{\infty}_K(M)$ is a Fréchet space, then continuity of $u$ can also be tested by sequential continuity. Any smooth function $g \in C^{\infty}(M)$ induces a distribution by choosing a measure $d\mu$ on $M$ and defining
$$
    f\mapsto \int_M g(x)f(x)d\mu(x).
$$

This last example inspires us to define distributional sections of vector bundles. Let $E \to M$ be a vector bundle and $E'\to M$ its dual bundle here we can choose two different definitions of distributional sections: either
\begin{itemize}
    \item we require that $\Gamma^{\infty}(M\leftarrow E)$ embeds into the space of distributional sections,
\end{itemize}
or
\begin{itemize}
    \item we define it as the space of continuous linear mappings $\Gamma^{\infty}_c(M\leftarrow E)\to \mathbb{R}$.
\end{itemize}
In the second case we have that $\Gamma^{\infty}\big(M\leftarrow E'\otimes \Lambda_m(M)\big)$ embeds into the space of distributions instead of $\Gamma^{\infty}(M\leftarrow E)$. This is true since, given $\mu\in\Gamma^{\infty}\big(M\leftarrow E'\otimes \Lambda_m(M)\big)$, it is natural to define
$$
    X \in \Gamma^{\infty}(M\leftarrow E) \mapsto \int_M \langle X, \mu\rangle (x)\in \mathbb{R}.
$$

\begin{definition}\label{def_A_distributional_sections}
Let $E \rightarrow M$ be a vector bundle, a distributional section is a continuous linear map 
$$ 
    s: \Gamma^{\infty}_{c}(M\leftarrow E) \rightarrow \mathbb{R}.
$$
We denote by $ \Gamma^{-\infty}_c(M\leftarrow E)$ the space of distributional sections.
\end{definition}

As in the previous case, since $\Gamma^{\infty}_{c}(M\leftarrow E)$ is a LF space, continuity of $ s: \Gamma^{\infty}_{c}(M\leftarrow E) \rightarrow \mathbb{R}$ is to be tested on each $\Gamma^{\infty}_{K}(M\leftarrow E)$, thus $s$ is continuous if and only if for each compact $K \subset M$, there is a positive constant $c$ and $n\in \mathbb{N}$ such that 
$$ 
    | s(X) |\leq c \sup_{j\leq n} p_{j,K}(X).
$$
An important notion with distributions is their support:

\begin{definition}\label{def_A_distr._section_support}
Let $U \subseteq M$ be open and $s \in  \Gamma^{-\infty}(M\leftarrow E) $, the \textit{restriction of} $s$ \textit{to} $U$ is the distribution
$$ 
   s|_{U}(X)= s(X), \ \ X \in \Gamma^{\infty}_{c}(U\leftarrow E|_U).
$$
The \textit{support} of $s$ is the set
$$
    \mathrm{supp}(s)= \bigcap_{ \substack {A \subset M \mathrm{closed} \\ \left. s\right|_{M \backslash A}=0}}A.
$$
\end{definition}

We remark that playing with partitions of unity it is possible to endow $ \Gamma^{-\infty}(M\leftarrow E) $ with the structure of a \textit{fine sheaf}. This in particular implies the principle of localization: a distributional section is the zero section if and only if for every point $x\in M$ there is an open neighborhood $U \ni x$ such that $s|_U=0$. If we denote by $\Gamma^{-\infty}_c(M \leftarrow E)$ the set of compactly supported distributions, then one can show that it is isomorphic to the space of continuous linear mappings $s :\Gamma^{\infty}(M\leftarrow E) \to \mathbb{R}$, \textit{i.e.} the dual of the Fréchet space $\Gamma^{\infty}(M\leftarrow E)$.\\

Next we briefly recall the definition and properties of the wave front set that will be used throughout this thesis. Let $U\subset M$ be an open subset. We will henceforth denote, by $\dot T^*U$ the cotangent bundle minus the graph of the zero sections.  
% \begin{lemma}
% Let $u \in \mathcal{D}'(U)$, then $u$ is smooth if and only if for every $N \in \mathbb{N}$ there is a constant $c(N)$ such that $$ \mid \hat{u}(\xi) \mid \leq c(N) (1+\mid \xi \mid)^{-N} .$$ 
% \end{lemma}
% \begin{proof}
% Suppose that $u$ is smooth then by a well known theorem $\hat{u} \in \mathcal{S}(\mathbb{R}^n)$, which implies the estimate. On the other hand suppose $\mid \hat{u}(\xi) \mid \leq c(N) (1+\mid \xi \mid)^{-N} $ then consider $$u(x)=\frac{1}{(\pi)^{m \backslash 2}}\int \hat{u} e^{i \langle x,  \xi\rangle},$$by hypothesis this is well defined, we have then to prove smoothness. $$\frac{\partial^k}{\partial x^{{\alpha}}} u(x)=\frac{1}{(\pi)^{m \backslash 2}}\int \hat{u} \frac{\partial^k}{\partial x^{{\alpha}}}e^{i \langle x,  \xi\rangle}$$for all polynomials $k \in \mathbb{N}$ and all multi-index ${\alpha}$ of magnitude $k$, so derivatives of $u$ are well defined, hence $u$ is smooth.
% \end{proof}

\begin{definition}\label{def_0_WF_set_of_distributions}
Let $u \in \mathcal{D}'(U)$, the wavefront set of $u$, $WF(u)$, is the complement in $\dot T^{*}U $ of the set of points $(x_0,\xi_0) \in T^{*}U$ such that there are :
\begin{itemize}
\item[$(i)$] $\chi \in \mathcal{D}(U)$  with $\chi(x_0)=1$,
\item[$(ii)$] an open conic neighborhood ${V} \subset \mathbb R^n $ of $\xi_0$, with 
$$ 
    \sup_{\xi \in {V}} (1+  |\xi|)^{N}\left| \widehat{\chi u}(\xi) \right| < \infty
$$
for all $N \in \mathbb{N}$ and all $\xi \in {V}$.
\end{itemize} 
\end{definition}

Intuitively, Definition \ref{def_0_WF_set_of_distributions} tells us not only where the distribution $u$ is singular in $U$, but also, via the Fourier transform, in which direction on the cotangent bundle $T^*U$ $u$ fails to be smooth.

\begin{proposition}\label{prop_0_WF_properties_1}
Let $u \in  \mathcal{D}'(U)$.
\begin{itemize}
%\item[$(i)$] The projection of $\mathrm{WF}(u)$ onto the first component is $\Sigma(u)$, that is $$\Sigma(u)= \left\lbrace x \in U \space\  \mathrm{ such } \space\ \mathrm{ that } \space\ \exists \xi \space\  \mathrm{ with } \space\ (x,\xi) \in \mathrm{WF}(u) \right\rbrace.$$
\item[$(i)$] If $\varphi \in C^{\infty}(U)$, then $$\mathrm{WF}(\varphi u) \subset \mathrm{WF}(u).$$
\item[$(ii)$] Let now $\alpha$ be a multi-index, then $$\mathrm{WF}(\partial^{\alpha} u) \subset \mathrm{WF}(u).$$
\item[$(iii)$] If $D$ is a differential operator with smooth coefficients then $$\mathrm{WF}(Du) \subset \mathrm{WF}(u).$$
\end{itemize}
\end{proposition}

Let us make some examples of calculations of wave fronts. We will calculate the wave fronts of the following distributions:
\begin{align*}
    \delta_0 &: C^{\infty}_c(\mathbb{R}) \rightarrow \mathbb{R}, \quad f \mapsto f(0),\\ 
    v_{\epsilon}=\lim_{\epsilon \rightarrow 0} \frac{1}{x+i\epsilon} & : C^{\infty}_c(\mathbb{R}) \rightarrow \mathbb{R}, \quad f \mapsto \lim_{\epsilon \rightarrow 0} \int_{\mathbb{R}}\frac{f(x)}{x+i \epsilon} \mathrm{d}x, \\
    \delta_{\Delta} &: C^{\infty}_c(\mathbb{R}^n) \rightarrow \mathbb{R}, \quad f \mapsto \int_{\mathbb{R}}f(x,\ldots,x) dx,
\end{align*}
where the notation $\Delta$ refers to the small diagonal $\{(x_1,\ldots,x_n)\in \mathbb R^n: x_1=\cdots=x_n\} $ of $\mathbb R^n$, and $\delta_{\Delta}$ is the \textit{diagonal delta}. For the Dirac delta we have that
$$
	\widehat{ \chi\delta_0}(\xi)= \frac{1}{\sqrt{2 \pi}}\langle \delta_0, \chi e^{i \xi x} \rangle =  \frac{\chi(0)}{\sqrt{2 \pi}},
$$
then either $\mathrm{supp}(\chi) \ni 0$ in which case $\widehat{\chi \delta_0}(\xi)$ is not rapidly decreasing, or $0 \notin \mathrm{supp}(\chi) $ which has $\widehat{\chi \delta_0}(\xi)\equiv 0$. So $\mathrm{WF}(\delta_0)=\{(0,\xi)\in \dot T^*\mathbb R\}$.\\

For the second distribution observe that it is smooth in all of $\mathbb{R}$ minus the origin, also
$$
	\widehat{ v_{\epsilon} }(\xi)= \frac{1}{\sqrt{2 \pi}} \lim_{\epsilon \rightarrow 0} \int_{\mathbb{R}}\frac{ e^{i\xi x}}{x+i \epsilon} \mathrm{d}x.
$$
Evaluating via the residue theorem yields
$$
	\widehat{v_{\epsilon} }(\xi)= -{\sqrt{2 \pi}}i \lim_{\epsilon \rightarrow 0} \theta(\xi)e^{-\xi \epsilon}= -{\sqrt{2 \pi}} i \theta(\xi),
$$
where $\theta$ is the Heaviside distribution. As a result, if $f$ is smooth and supported in a neighborhood of the origin
$$
	\widehat{f v_{\epsilon} }(\xi)= \int_{\mathbb{R}}\hat{f}(\eta)\hat{v}_{\epsilon}(\xi-\eta) \mathrm{d}\eta=-i \sqrt{2 \pi}\int_{-\infty}^{\xi} \hat{f}(\eta)\mathrm{d}\eta,
$$
if $\xi<0$, then $|\xi|<|\eta|$ and
$$
    \left| \int_{-\infty}^{\xi} \hat{f}(\eta)\mathrm{d}\eta\right|= \left| \int_{-\infty}^{\xi} \hat{f}(\eta)\frac{(1+|\eta|)^{N}}{(1+|\eta|)^{N}}\mathrm{d}\eta\right|  \leq (1+|\xi|)^{-N} \left| \int_{-\infty}^{\xi} \hat{f}(\eta){(1+|\eta|)^{N}}\mathrm{d}\eta\right|  \leq c_{f} (1+|\xi|)^{-N} 
$$
if $\xi>0$, the above estimate is not possible and we just have $\left| \int_{-\infty}^{\xi} \hat{f}(\eta)\mathrm{d}\eta\right|  \leq c'_f$. Thus $\mathrm{WF}(v_{\epsilon})=\lbrace(0,\xi)\in \dot T^*\mathbb R,\  \xi>0 \rbrace $.\\

Finally, consider the diagonal delta $\delta_{\Delta}$, let $f \in C^{\infty}_c(\mathbb{R}^n)$
$$
	\widehat{f \delta_{\Delta}}(\xi_1,\ldots,\xi_m)= \int_{\mathbb{R}}f(x,\ldots,x)e^{i(\xi_1+\ldots+\xi_m)x}.
$$
We immediately note that when $f$ is supported outside the diagonal the integral is identically zero, hence rapidly decreasing. On the other hand if $\mathrm{supp}(f)\cap \Delta_{n}\mathbb R \neq \emptyset$ we can, provided $\xi_1+\ldots+\xi_n\neq 0 $, integrate by parts, getting
$$
(1+|\xi|)^N \left|\widehat{f \delta_{\Delta}}(\xi_1,\ldots,\xi_n) \right| \leq \frac{(1+|\xi|)^N}{(\xi_1+\ldots+\xi_n)^M}\int_{\mathbb{R}}\left|\partial_x^Mf(x,\ldots,x) \right|dx.
$$ 
Choosing an appropriate $M \in \mathbb{N}$ makes the right hand side bounded by a constant. Finally, if $\xi_1+\ldots+\xi_n= 0 $ along the diagonal $\widehat{f \delta_{\Delta}}(\xi_1,\ldots,\xi_n)$ is a constant depending on $f$, which is not rapidly decreasing. To sum up we have $\mathrm{WF}(\delta_{\Delta})=\lbrace(x,\ldots,x;\xi_1,\ldots,\xi_n)\in \dot T^{*}\mathbb{R}^n : \xi_1+\ldots+\xi_n=0 \rbrace$.\\

With the notion of wave front set, the space of distributions $\mathcal{D}'(U)$ can be endowed with a topology which is finer than the usual dual topology of $\mathcal{D}(U)$. Let $\Upsilon \subset T^*U$ be a cone (\textit{i.e.} $\Upsilon_x \subset T^*_xU$ is a closed cone for any $x\in U$) and define $\mathcal{D}_{\Upsilon}'(U)$ as the space of distributions whose wave front set is contained in $\Upsilon$, then we induce a topology via the seminorms
\begin{equation}
   \begin{cases}
    p_{f}(u)=|u(f)|,\\
    p_{q,\chi,\Upsilon}(u)=\sup_{\xi \notin \Upsilon}(1+|\xi|)^q |\widehat{\chi u} (\xi)|,
\end{cases} 
\end{equation}
where $f, \chi \in \mathcal{D}(U)$, $q\in \mathbb N$. The latter is called \textit{H\"ormander topology}, it is finer then the subspace topology $\mathcal{D}_{\Upsilon}'(U)\subset \mathcal{D}'(U)$ and equal when $\Upsilon=T^*U$.\\

\begin{theorem}\label{thm_0_WF_pullback}
Let $\psi: U_1 \rightarrow U_2$ be a smooth function and $\Upsilon$ be a closed cone in $\dot T^{*}U_2$ with 
$$
    N^*{\psi}\doteq \left\lbrace(\psi(x), \xi) \in \dot T^{*}U_2 : \ x \in U_1 \  \psi^{*} \xi=0  \right\rbrace \cap \Upsilon = \emptyset,
$$
then there exist a unique sequentially continuous extension $ \mathcal{D}_{\Upsilon}(U_2)  \rightarrow \mathcal{D}_{\psi^{*}\Upsilon}(U_1) $ of $\psi^*: \mathcal{E}(U_2)\to \mathcal{E}(U_1)$.
\end{theorem}

This result is of paramount importance for defining the wave front set for distributions on manifolds. Notice that if $\psi:U_1 \rightarrow U_2$ is a diffeomorphism, $N^*{\psi}$ is trivial, and by \Cref{thm_0_WF_pullback} $\mathrm{WF}(\psi^{*}u)=\psi^{*}\mathrm{WF}(u)$. Next we consider a $n$-dimensional manifold $M$ with charts $\{(U_{\alpha},\phi_{\alpha})\}_{\alpha\in A}$; given $u\in\mathcal{D}'(M)$, we can define its localization along charts $u_{\alpha}=(\phi_{\alpha}^{-1})^{*}u=(\phi_{\alpha})_*u$ and if $x\in U_{\alpha}$, set
\begin{equation}\label{eq_0_WF_manifolds}
	\mathrm{WF}_{x, \alpha}(u) \doteq \phi_{\alpha}^*\mathrm{WF}((\phi_{\alpha})_*u).
\end{equation}
Direct application of \Cref{thm_0_WF_pullback} shows that \ref{eq_0_WF_manifolds} is independent from the chart used.\\

A particular use of the wave front set is the product of distributions. For instance let $u$, $v$ $\in \mathcal{D}'(U)$, and consider the distribution $u \otimes v \in \mathcal{D}'(U \times U)$ defined by $$ u \otimes v (\phi, \psi) = u(\phi) v(\psi).$$

\begin{lemma}\label{lemma_8.2.10}
Let $u$, $v\in \mathcal{D}'(U)$, then
$$
    \mathrm{WF}(u \otimes v) \subset \mathrm{WF}(u) \times \mathrm{WF}(v) \cup \mathrm{WF}(u) \times \mathrm{supp}(v) \cup  \mathrm{supp}(u) \times \mathrm{WF}(v).
$$
Moreover if $i:U \rightarrow U : x \mapsto(x,x)$ is the diagonal immersion, and $\{(x,\xi)\in \mathrm{WF}(u)\}\cap \{(x,-\xi)\in \mathrm{WF}(v)\}=\emptyset$, the product $u\cdot v=  i^*(u\otimes v) $ is well defined and
$$
\mathrm{WF}(uv) \subset \mathrm{WF}(u) \cup \mathrm{WF}(v) \cup (\mathrm{WF}(u)+  \mathrm{WF}(v)),
$$
where $\mathrm{WF}(u)+  \mathrm{WF}(v)= \left\lbrace (x,\xi+ \eta): (x,\xi)\in \mathrm{WF}(u) , (x,\eta)\in \mathrm{WF}(v)\right\rbrace$.
\end{lemma}

\begin{corollary}
If $i: \Sigma \rightarrow U$ is an embedding, then $u\in \mathcal{D}'(U)$ can be restricted to $\Sigma$ if $N^{*}\Sigma=\lbrace (x,\xi)\in T^{*}U: x \in \Sigma, \xi(v)=0 \ \forall v \in T_x \Sigma\rbrace \cap \mathrm{WF}(u)= \emptyset$.
\end{corollary}

We conclude by mentioning some results about integral kernels and their wave front sets. For the proof of those statements we refer to \cite{hormanderI}.

\begin{theorem}\label{thm_8.2.12}
Let $V\subset \mathbb{R}^n$, $U\subset \mathbb{R}^n$ be open subsets, and $K \in \mathcal{D}'(V \times U)$, then for every $\phi \in C^{\infty}_c(U)$, $K(\cdot, \phi) \in \mathcal{D}'(V)$ has 
$$
    \mathrm{WF}(K(\cdot, \phi)) \subset \lbrace (x,\xi) \in T^{*}V : (x,y;\xi,0) \in \mathrm{WF}(K) \space\ \mathrm{ and } \space\ y \in \mathrm{supp}(\phi) \rbrace
$$
\end{theorem}

For some $K\in \mathcal{D}'(V \times U)$ we shall denote $\mathrm{WF}(K)_V \doteq \lbrace (x,\xi)\in T^{*}V :(x,y;\xi,0) \in \mathrm{WF}(K) \rbrace $ and similarly $\mathrm{WF}(K)_U \doteq \lbrace (y,\eta)\in T^{*}U :(x,y;0,\eta) \in \mathrm{WF}(K) \rbrace $. Let additionally $\mathrm{WF}'(K)$ be $\lbrace (x,y;\xi,-\eta)\in T^{*}(V\times U) :(x,y;\xi,\eta) \in \mathrm{WF}(K) \rbrace $. 

% \begin{theorem}
% Let $V\subset \mathbb{R}^n$, $U\subset \mathbb{R}^n$ be open subsets, and $K \in \mathcal{D}'(V \times U)$, suppose that $v \in \mathcal{E}'(U)$ has $\mathrm{WF}(v) \cap \mathrm{WF}'(K)_U = \emptyset$. Then it is possible to uniquely define $K(\cdot, v) \in \mathcal{D}'(V)$ such that 
% $$
%     \mathrm{WF}(Kv)=\lbrace (x,\xi)\in T^{*}V : (x,y;\xi,\eta) \in \mathrm{WF}'(K) \,  (y,\eta) \in \mathrm{WF}(v) \rbrace \cup \mathrm{WF}(K)_V.
% $$
% \end{theorem}

\begin{theorem}\label{thm_8.2.14}
Let $V\subset \mathbb{R}^n$, $U\subset \mathbb{R}^n$, $W\subset \mathbb{R}^p$ be open subsets, $K_1 \in \mathcal{D}'(V \times U)$, $K_2 \in \mathcal{D}'(U \times W)$. Suppose that $\mathrm{WF}(K_2)\cap \mathrm{WF}'(K_1) =\emptyset$ and that the map $\mathrm{pr}_2: \mathrm{supp}(K_2) \rightarrow W$ is proper, then there is a unique way of defining the composition distribution $K_1 \circ K_2 \in \mathcal{D}'(V \times W)$ such that
$$
\begin{aligned}
\mathrm{WF}'(K_1 \circ K_2 ) \subset & \left( V \times \lbrace0\rbrace \times \mathrm{WF}'(K_2)_W \right) \cup \left( \mathrm{WF}(K_1)_V \times W \times \lbrace0\rbrace\right) \\ & \cup \left\lbrace (x,z;\xi,\zeta) \in T^{*}(V \times W) : \exists (y,\eta) \in T^{*}U \right. \\ & \space\ \left. :  (x,y;\xi,-\eta)\in \mathrm{WF}(K_1) ,(y,z;\eta,-\zeta)\in \mathrm{WF}(K_2)\right\rbrace.
\end{aligned}
$$
\end{theorem}

Similarly, we can define the notion of wave front set for distributions on vector bundles. Let $\pi:E\to M$ be a vector bundle and let $\{(\pi^{^-1}(U_{\alpha}),t_{\alpha})\}_{\alpha}$ be a family of trivializations of $E$. If $s\in \Gamma^{-\infty}_c(M\leftarrow E)$ is a distribution, $(t_{\alpha})_*s=(s^1,\ldots,s^k)$ where $k$ is the dimension of the fiber of $E$ and each $s^i\in \mathcal{D}'(M)$. Then we set
\begin{equation}\label{eq_WF_distributional_sections}
    \mathrm{WF}(s)=\cup_{i=1}^k \mathrm{WF}(s^i).
\end{equation}
By \Cref{thm_0_WF_pullback} we can show that the above definition does not depend on the trivialization chosen. Moreover we can straightforwardly generalize Lemma \ref{lemma_8.2.10}, \Cref{thm_8.2.12} and \Cref{thm_8.2.14} to distributional sections of $\Gamma^{-\infty}_c(M\leftarrow E)$.

\subsection{Bastiani calculus on locally convex spaces}\label{section_Bastiani}

The first notion of calculus in locally convex spaces we introduce is the so-called Bastiani calculus, its origin can be traced back to \cite{mb,michal1938differential}. We shall use it as a notion of calculus on the topological spaces of mapping we will use in \Cref{chapter_classical}. Here we introduce the basic definitions and properties.\\

In the sequel we will use complete locally convex spaces and denote them by capital letters $X$, $Y$, $Z$. 

\begin{definition}
    Let $U\subset X$ be an open subset, a mapping $P:U\subset X \to Y$ is Bastiani differentiable (or $P\in C^1_B(U,Y) $) if the following conditions hold:
    \begin{itemize}
        \item[$(i)$] $ \lim_{t\to 0} \frac{1}{t}\big(P(x+tv)-P(x)\big)=dP[x](v)$ exists for all $x\in U$, $v\in X$, giving rise to a mapping $dP:U\times X \to Y$ linear in the second entry; 
        \item[$(ii)$] The mapping $dP:U\times X \to Y, \ (x,v) \to dP[x](v)$ is jointly continuous. 
    \end{itemize}
\end{definition}

%One might think of giving condition $(ii)$ in a slightly different manner, for example requiring that $dP:U \to \mathfrak{L}(X,Y)$ valued in the space of linear mappings is continuous for some suitably chosen topology in the target space.% Unfortunately one could manufacture examples where the mapping $dP:U\times X \to Y$ is continuous, whereas $dP :U \to \mathfrak{L}(X,Y)$ is not. For example consider

 A technically important result is the Riemann integral for continuous curves $\gamma:[a,b]\subset \mathbb{R} \to X$. 

\begin{theorem}\label{thm_A_Riemann_curve_integral}
    Let $\gamma:[a,b]\subset \mathbb{R} \to X$ be a continuous curve, then there exists a unique object $\int_a^b \gamma(t)dt\in X$ such that
    \begin{itemize}
        \item[$(i)$] for every continuous linear mapping $l:X\to \mathbb{R}$
        $$
            l\bigg(\int_a^b\gamma(t)dt\bigg) =\int_a^bl(\gamma(t))dt;
        $$
        \item[$(ii)$] for every seminorm $p_i$ on $X$, 
        $$
        p_i\bigg(\int_a^b\gamma(t)dt\bigg) \leq \int_a^bp_i(\gamma(t))dt;
        $$
        \item[$(iii)$] for all continuous curves $\gamma$, $\beta$, $\int_a^b \big(\gamma(t)+ \beta(t)\big)dt=\int_a^b \gamma(t)dt+\int_a^b \beta(t)\big)dt$;  
        \item[$(iv)$] for all $\lambda\in \mathbb{R}$, $\int_a^b \big(\lambda \gamma(t)\big)dt=\lambda\bigg(\int_a^b \gamma(t)dt\bigg)$;
        \item[$(v)$] for all $a\leq c \leq b$, $\int_a^b\gamma(t) dt=\int_a^c\gamma(t) dt  +\int_c^b\gamma(t) dt$. 
    \end{itemize}
\end{theorem}
The proof of this result is quite standard (see \textit{e.g.} \cite[Theorem 2.2.1]{hamilton1979inverse}), it essentially follows from defined the integral as a Riemann summation in the interval $[a,b]$ whenever $\gamma$ is a piecewise straight continuous curves. Finally, noting that the latter space is dense in $C([a,b];X)$, we can extend the integral by continuity and get uniqueness by the Hahn-Banach theorem.

\begin{lemma}
    Let $U\subset X$, $V\subset Y$ be open subsets, let $P:U\to Y$ and $Q:V\to Z$ be Bastiani differentiable mappings such that $P(U) \subset V$, then $Q\circ P : U \to Z$ is Bastiani differentiable and $d(Q\circ P)[x](v)= dQ[P(x)]\big(dP[x](v)\big)$ for each $x\in U$, $v\in X$. 
\end{lemma}
\begin{proof}
    We claim that $P\in C^1_B(U; Y)$ if and only if there is a continuous mapping $L:U\times U \times X \to Y$ linear in the third entry such that 
    $$
        P(x_1)-P(x_2)= L[x_1,x_2](x_1-x_2).
    $$
    In particular we have $L[x,x](v)=dP[x](v)$. The necessity condition follows by considering the smooth curve $\gamma(t)=x_1+t(x_2-x_1)$ with $t\in [0,1]$, then $dP[\gamma(t)](v)$ is a smooth curve for each $v\in X$, by \ref{thm_A_Riemann_curve_integral}, define
    $$
        L[x_1,x_2](v)\doteq \int_0^1 dP[\gamma(t)](v)dt.
    $$
    It is clear that $L$ is continuous and linear in the third entry, moreover since $\frac{d}{dt}P(\gamma(t))=dP[\gamma(t)](x_1-x_2)$ we get $L[x_1,x_2](x_1-x_2)= dP[x_1](x_1-x_2)$. For the sufficiency condition just note that $\frac{1}{t}\big(P(x+tv)-P(x)\big)=L(x, x+tv)[v]$, so taking the limit we get our claim. Next suppose that 
    $$
        P(x_1)-P(x_2)= L[x_1,x_2](x_1-x_2),
    $$
    $$
        Q(y_1)-Q(y_2)= M[y_1,y_2](y_1-y_2).
    $$
    Then $Q(P(x+tv))-Q(P(x))=M[P(x+tv),P(x)]\big(P(x+tv)-P(x)\big)=t M[P(x+tv),P(x)]\big(L[x+tv,x](v)\big)$ thus dividing by $t$ and taking the limit yields $d(Q\circ P)[x](v) \equiv M[P(x),P(x)]\big(L[x,x](v)\big)= dQ[P(x)]\big(dP[x](v)\big)$.
\end{proof} 

\begin{definition}\label{def_A_Bastiani_smooth_map}
    A mapping $f:U\subset X \to Y$ is $k$ times Bastiani differentiable if $d^{k-1}f:U\times X \cdots \times X \to Y$ is Bastiani differentiable. The $k$th derivative of $f$ at $x$ is defined by recursion
    \begin{equation}\label{eq_A_Bastiani_k_der}
        d^{k}f[x](v_1,\ldots,v_k) \doteq \lim_{t\to 0} \frac{1}{t}\big( d^{k-1}f[x+tv_k](v_1,\ldots,v_{k-1})-d^{k-1}f[x+tv_k](v_1,\ldots,v_{k-1})\big)  .
    \end{equation}
    Finally, we denote by $C^k_B(U,Y)$ the set of $k$ times Bastiani differentiable functions $f:U \to Y$. 
\end{definition}

Computing explicitly $d\big( d^{k-1}f\big)$ in $(x,v_1,\ldots,v_{k-1})$, by linearity we get
$$
    \begin{aligned}
         & d\big(d^{k-1}f\big)[x,v_1,\ldots,v_{k-1}](v_k,w_1,\ldots,w_{k-1})= \sum_{j=1}^{k-1} d^{k-1}f[x](v_1,\ldots, \widehat{v_j}, w_j, \ldots, v_{k-1})\\ & \qquad + \   \lim_{t\to 0} \frac{1}{t}\big( d^{k-1}f[x+tv_k](v_1,\ldots,v_{k-1})-d^{k-1}f[x+tv_k](v_1,\ldots,v_{k-1})\big) ,
    \end{aligned}
$$
we can then see that $d\big( d^{k-1}f\big)$ is Bastiani differentiable if and only if the limit of \eqref{eq_A_Bastiani_k_der} exists and is a continuous mapping $U\times X^k \to Y$. We can thus state 
\begin{lemma}
    A mapping $f : U \subset X \to Y$ is $k$ Bastiani differentiable if and only if for each $0\leq j \leq k$ all the derivative mappings $d^{j}f:U \times X^j \to Y$ exists and are jointly continuous. 
\end{lemma}

\subsection{Convenient calculus on locally convex spaces}\label{section_convenient}

We introduce another calculus on locally convex spaces, called \textit{convenient calculus}, for details see \cite{convenient} and references therein. The main difference with Bastiani calculus is that smooth mappings need not be continuous, however, this is compensated by \Cref{thm_A_cart_closedness}, which asserts that $C^{\infty}(U,C^{\infty}(V,W))=C^{\infty}(U\times V,W)$ for any open subsets $U,\ V,\ W$ of locally convex spaces $X,\ Y,\ Z$. This is quite a strong property which will be very useful in \Cref{chapter_Wick}.

\begin{definition}
    Let $X$ be a locally convex space, a continuous curve $\gamma:(a,b)\subset \mathbb{R}\to X$ is differentiable is the limit
    $$
        \gamma'(s)\doteq\lim_{t\to 0} \frac{\gamma(s+t)-\gamma(s)}{t}
    $$
    exists and defines a continuous curve $\gamma':(a,b)\to X$. We say that $\gamma$ is $k$ times differentiable if $\big(\gamma^{(k-1)}\big)'$ exists and is a continuous curve. If $\gamma$ is differentiable for any $k \in \mathbb{N}$, then we say it is smooth.
\end{definition}

This notion of curve induces a topology on the space $X$ which following \cite{kriegl1997convenient} we call the $c^{\infty}$-topology. It is constructed as follows: it is the final topology on $X$ with respect to all smooth curves $\gamma\in C^{\infty}(\mathbb{R},X)$, that is the finest topology for which all smooth curves $\gamma:\mathbb{R} \to X$ becomes continuous. Its open subsets will be called $c^{\infty}$-open. 

Recall that a sequence $\{x_n\}\subset X$ is \textit{Mackey-Cauchy} if there exists a sequence $\{\mu_{n,m}\}\subset \mathbb{R}$ which diverges to infinity for which $\mu_{n,m}(x_n-x_m)\in B$ for a bounded subset $B \subset X$ and for all $n,m \in \mathbb{N}$. We say that $X$ is \textit{Mackey complete} if any Mackey-Cauchy sequence converges. Incidentally, we remark that by Lemma 2.2 pp.15 in \cite{kriegl1997convenient}, for a locally convex space $X$, every \textit{Mackey-Cauchy net} converges in $X$ if and only if every \textit{Mackey-Cauchy sequence} converges in $X$; therefore even if the topology of $X$ is generated by uncountably many seminorms, it is enough to check Mackey completeness for sequences.

\begin{definition}\label{def_A_convenient_vector_space}
A locally convex vector space $X$ is said to be \textit{convenient} or $c^{\infty}$\textit{-complete} if any of the following equivalent conditions hold:
\begin{itemize}
\item[$(i)$] $X$ is Mackey complete;
\item[$(ii)$] any curve $\gamma:\mathbb{R} \to X$ such that $l\circ \gamma:\mathbb{R} \to \mathbb{R}$ is smooth for all $l\in X'$, is itself smooth;
\item[$(iii)$] $X$ is $c^{\infty}$ closed.
\end{itemize}
\end{definition}

We observe that $(i)$ Definition \ref{def_A_convenient_vector_space} implies that $c^{\infty}$ completeness is a, quite mild, bornological\footnote{The bornology of a locally convex vector space $X$ is the family of its bounded subsets, \textit{i.e.} those subsets $B\subset X$ such that for all neighborhood $U$ of $0$ there is $\lambda \in \mathbb R$ which $B\subset \lambda U$. Saying that a convenient smoothness is \textit{bornological} refers to the fact that it depends on the bounded subsets of the topology. Therefore refining the topology while keeping the same bounded subsets does not alter the set of smooth curves $C^{\infty}(X)$.} condition; whereas $(ii)$ implies that whether a curve $\gamma$ is smooth can be tested by means of elements of the topological dual of $X$. This suggests a strong relationship between bornology and smoothness of curves. We shall expand this connection below when considering other properties of this topology.\\

The space $C^{\infty}(\mathbb{R},X)$ of smooth curves can be endowed with a locally convex space structure. The addition and multiplication by scalar are essentially inherited by those on $X$, the topology is the topology of uniform convergence on compact subsets of $\mathbb{R}$ in each derivative separately. If $J\subset \mathbb{R}$ is compact, then seminorms of this space are given by 
\begin{equation}\label{eq_A_top_unif_conv_on_smooth_curves}
    p_{i,J,n}(\gamma)=\sup_{t\in J, k\leq n} p_i\big(\gamma^{(k)}(t)\big)
\end{equation}
where $p_i$ is any seminorm in $X$. Equivalently this topology can be described as the initial topology with respect to linear mappings $C^{\infty}(\mathbb{R},X)\to l^{\infty}(J,X): $ . Notice that a set $\mathcal{B}$ in $ C^{\infty}(\mathbb{R },X)$ is bounded if and only if for all $\gamma\in \mathcal{B}$, $p_{i,J,n}(\gamma)\leq C $. For any continuous linear functional $l\in X'$, $l(\mathcal{B})\subset \mathbb{R}$ is bounded whenever $\mathcal{B}$ is. Since $l(\gamma^{(k)})=(l\circ\gamma)^{(k)}$, then we can state that $\mathcal{B}\subset C^{\infty}(\mathbb{R},X)$ is bounded if and only if $l(\mathcal{B})$ is bounded in $\mathbb{R}$ for each $l\in X'$.

\begin{proposition}\label{prop_A_Boman_general}
    Let $f:\mathbb{R}^2 \to X$ a mapping into a locally convex space (not necessarily $c^{\infty}$-complete) then the following assertions are equivalent:
    \begin{itemize}
        \item[$(i)$] $f\circ \gamma$ is smooth whenever $\gamma\in C^{\infty}(\mathbb{R},\mathbb{R}^2)$; 
        \item[$(ii)$] all iterated directional derivatives $df[p](v)$ exists and are locally bounded;
        \item[$(iii)$] all iterated partial derivatives of $f$ exists and are locally bounded;
        \item[$(iv)$] $f^{\vee}: \mathbb{R} \to C^{\infty}(\mathbb{R},X)$ is a smooth curve.
    \end{itemize}
\end{proposition}

Prior to writing the proof, let us cite an important result in \cite{boman1967differentiability} known as Boman Theorem:
\begin{theorem}[Boman]\label{thm_A_Boman}
    Let $f:\mathbb{R}^2\to \mathbb{R}$ a mapping, then the following conditions are equivalent:
    \begin{itemize}
    \item[$(i)$] all iterated partial derivatives exists and are continuous;
    \item[$(ii)$] all iterated partial derivatives exists and are locally bounded;
    \item[$(iii)$] for each $v\in \mathbb{R}^2$ the iterated directional derivatives $d^nf[x](v)=\frac{d^n}{dt^n}\big\vert_0f(x+tv)$ exists and are locally bounded with respect to $x$;
    \item[$(iv)$] for all smooth curves $\gamma:\mathbb{R}\to \mathbb{R}^2$ the composition $f\circ \gamma$ is smooth.
    \end{itemize}
\end{theorem}

\begin{proof}[Proof of Proposition \ref{prop_A_Boman_general}]
    We shall prove this result for convenient vector spaces, the general proof can be found in \cite[pp. 29]{kriegl1997convenient}. The equivalence of $(i)$, $(ii)$, $(iii)$ is essentially trivial: clearly $(i) \Rightarrow (ii)$, also $(ii) \Rightarrow (iii)$ by testing in the cases where $v$ is any vector of the standard basis of $\mathbb{R}^2$, finally $(ii) \Rightarrow (i)$ since $(f\circ \gamma)^{(k)}$ can be always expressed as a finite combination of terms of the form $\partial_{i_1}\partial_{i_p}f(\gamma(t)) \gamma^{(k_1)}(t)\cdots \gamma^{(k_p)}(t) $ with $i_j=1,2$ and $\sum k_j=p$. To complete the proof, note that by $(ii)$ in Definition \ref{def_A_convenient_vector_space} $f^{\vee}\in C^{\infty}(\mathbb{R},X)$ if and only if $l\circ f^{\vee}=(l\circ f)^{\vee} \in C^{\infty}(\mathbb{R},\mathbb{R})$. Applying \Cref{thm_A_Boman} establishes the equivalence of $(iv)$ and $(iii)$. 
\end{proof}

\begin{definition}\label{def_A_convenient_smooth_map}
    A mapping $f:U\subset X \to Y $, where $U$ is $c^{\infty}$-open is \textit{smooth} if for any $\gamma \in C^{\infty}(\mathbb{R},U)$, $f\circ \gamma \in C^{\infty}(\mathbb{R},Y)$. We denote by $C^{\infty}(U,Y)$ the set of smooth mappings between $U$ and $Y$.
\end{definition}

Comparing Definition \ref{def_A_convenient_smooth_map} with Definition \ref{def_A_Bastiani_smooth_map} one immediately sees that each Bastiani smooth mapping is conveniently smooth as well. The other implication does not generally hold true (see for instance \cite[$\S2$ and Proposition 2.2]{glockner2005discontinuous} for a counterexample). 

The space $C^{\infty}(U,Y)$ is a vector space, we can endow it with a topology as follows: if $\gamma\in C^{\infty}(\mathbb{R},U)$, then $\gamma^{*}:C^{\infty}(U,Y) \to C^{\infty}(\mathbb{R},Y)$, therefore we give $C^{\infty}(U,Y)$ the initial topology with respect to the mappings $\gamma^{*}$. We claim that 
\begin{equation}\label{eq_A_char_smooth_conv_functions}
\begin{aligned}
    C^{\infty}(U,Y)&= \lim_{\substack{\longleftarrow\\ \gamma\in C^{\infty}(\mathbb{R},U)}} C^{\infty}(\mathbb{R},Y) \\ &= \bigg\{ \{f_{\gamma}\}_{\gamma}\in \prod_{\gamma \in C^{\infty}(\mathbb{R},U)}C^{\infty}(\mathbb{R},Y): f_{\gamma}\circ g =f_{\gamma\circ g} \ \forall g\in C^{\infty}(\mathbb{R},\mathbb{R})\bigg\}.
\end{aligned}
\end{equation}
Notice that any mapping $f_{\gamma}$ such that $f_{\gamma}\circ g =f_{\gamma\circ g}$ for all reparametrization $ g\in C^{\infty}(\mathbb{R},\mathbb{R})$ gives rise to a mapping $f:U\to Y$ by setting $f(x)=f_{\gamma_x}(t)$ where $\gamma_x:t\mapsto x$ is the constant curve, if $g$ is any reparametrization, then $f_{\gamma_x}\circ g =f_{\gamma_x\circ g}$ but $\gamma\circ g (t)\equiv x$ thus $f$ is not altered; finally $f$ is smooth since $f\circ \gamma \equiv f_{\gamma}\in  C^{\infty}(\mathbb{R},Y)$. On the other hand any smooth function $f:U\to Y$ gives rise to $\{f_{\gamma}\}_{\gamma}$ by setting $f_{\gamma}=f\circ \gamma$, then $f_{\gamma}\circ g= f\circ \gamma\circ g  =f_{\gamma\circ g}$.

\begin{theorem}\label{thm_A_cart_closedness}
    Let $U_i$ be $c^{\infty}$-open subsets of locally convex spaces $X_i$ for $i=1,2$ not necessarily Mackey complete. A mapping $f:U_1\times U_2 \to Y$ is smooth if and only if the mapping $f^{\vee}:U_1\to C^{\infty}(U_2,Y)$ is smooth.
\end{theorem}

\begin{proof}
    Suppose that $f^{\vee}:U_1\to C^{\infty}(U_2,Y)$ is smooth then, by Definition \ref{def_A_convenient_smooth_map}, this is equivalent to: $f^{\vee}\circ \gamma_1:\mathbb{R}\to C^{\infty}(U_2,Y)$ is smooth for any $\gamma_1\in C^{\infty}(\mathbb{R},U_1)$. By \eqref{eq_A_char_smooth_conv_functions}, the above statement is equivalent to $\gamma_2^*(f^{\vee}\circ \gamma_1):\mathbb{R} \to C^{\infty}(\mathbb{R},Y)$ is smooth for any curve $\gamma_2\in C^{\infty}(\mathbb{R},U_2)$. Finally, by Proposition \ref{prop_A_Boman_general}, the last assertion is equivalent to
    $$
        \big( \gamma_2^*(f^{\vee}\circ \gamma_1)\big)^{\wedge} : \mathbb{R}^2 \to Y \ \mathrm{is } \ \mathrm{smooth} \ \forall \gamma_1\in C^{\infty}(\mathbb{R},U_1),\ \gamma_2\in C^{\infty}(\mathbb{R},U_2).
    $$
    By construction, $\big( \gamma_2^*(f^{\vee}\circ \gamma_1)\big)^{\wedge} =f\circ (\gamma_1\times \gamma_2)$, however, by Definition \ref{def_A_convenient_smooth_map} this is once again equivalent to state that $f:U_1\times U_2 \to Y$ is smooth.
\end{proof}

A consequence of \ref{thm_A_cart_closedness} is that convenient calculus enjoys the basic properties of calculus that is
\begin{itemize}
    \item the chain rule holds, in particular let $f:U\subset X\to Y$ and $g_V\subset Y \to Z$ be conveniently smooth mapping with $f(U)\subset V$, then $d(g\circ f)[x](v)=dg[f(x)](df[x](v))$. Moreover, if we set $B(X;Y)$ the space of bounded (hence conveniently smooth) linear mappings equipped with its natural topology, the application $d:C^{\infty}(U,Y) \to C^{\infty}\big(U,B(X;Y)\big) $ is bounded and smooth;
    \item integration of smooth curves can be performed as in \Cref{thm_A_Riemann_curve_integral}.
\end{itemize}

In the proposition below we shall list some of the properties of the $c^{\infty}$-topology. As we stressed above, for any locally convex space $X$ the $c^{\infty}$ topology is a refinement of the initial locally convex topology, in general however $c^{\infty}X\equiv (X,\tau_{c^{\infty}})$ fails to be a topological vector space.

\begin{proposition}\label{prop_A_prop_of_c_infty_top}
    Let $X$ be a locally convex space, the bornologification of $X$ is the finest locally convex topology on $X$ possessing the same bounded subsets as the initial topology. We denote by $X_{b}$ this locally convex space. Then
    \begin{itemize}
        \item[$(i)$] $U\subset X_b$ is open if and only if is $c^{\infty}$-open, an absolutely convex subset $U\subset X_b$ is a $0$-neighborhood if and only if it is so for the $c^{\infty}$-topology.
        \item[$(ii)$] $id:(X,\tau_{c^{\infty}})\to X_b$ is continuous, hence the $c^{\infty}$-topology is finer then the bornologification of the initial topology on $X$.  
        \item[$(iii)$] If $X$ is metrizable, then $c^{\infty}X=X_b=X$. In particular if $X$ is a Fréchet vector space then it is a convenient vector space.
        \item[$(iv)$] If $X$ is a Fréchet space, $U\subset X$ an open subset and $Y$ a locally convex space, then $C^{\infty}(U,Y)=C^{\infty}_B(U,Y)$.
    \end{itemize}
\end{proposition}

 Some comments are due. Condition $(i)$ implies that $c^{\infty}X$ possesses all the convex subsets generating the topology of $X_b$ but might fail to be locally convex due to the appearance of \textit{extra} open subsets. As a consequence of this we have $(ii)$. $(iii)$ establishes the important fact that when $X$ is metrizable, then the $c^{\infty}$-topology coincides with the initial topology, in which case when $X$ is also complete (that is Fréchet) then it must be Mackey complete as well and thus convenient. Finally using $(iii)$ and Theorem 1, pp. 71 of \cite{frolicher2006smooth} we get the key result that convenient smoothness and Bastiani smoothness do coincide on Fréchet spaces thus implying $(iv)$.

\begin{proposition}\label{prop_A_conv_structures_on_C_infty(M)}
    Let $M$ be a smooth finite dimensional manifold. The space $C^{\infty}(M)$ of all smooth functions on $M$ is a convenient vector space with respect to the following bornologically isomorphic descriptions:
\begin{itemize}
    \item[$(a)$] the initial topology with respect to mappings
    $$
        \gamma^{*}:C^{\infty}(M) \to C^{\infty}(\mathbb{R})
    $$
    where $\gamma\in C^{\infty}(\mathbb{R},M)$ and $C^{\infty}(\mathbb{R})$ has the usual Fréchet space structure.
    \item[$(b)$] the initial topology with respect to mappings 
    $$
        (u_{\alpha}^{-1})^{*}:C^{\infty}(M) \to C^{\infty}(\mathbb{R}^n,\mathbb{R})
    $$
    where $(U_{\alpha},u_{\alpha})_{\alpha}$ is the atlas of $M$ and $C^{\infty}(\mathbb{R}^n,\mathbb{R})$ is endowed with the usual Fréchet space structure; 
    \item[$(c)$] the initial structure with respect to mappings
    $$
        j^k:C^{\infty}(M)\to C\big(M,J^k(M,\mathbb{R})\big)
    $$
    with $k\in \mathbb{N}$ and $C\big(M,J^k(M,\mathbb{R})\big)$ endowed with the compact open topology.
\end{itemize}
\end{proposition}

\begin{proof}
    We observe that the topology of uniform convergence on $C^{\infty}(\mathbb{R}^d,\mathbb{R})$ described in \eqref{eq_A_char_smooth_conv_functions} does coincide with the standard Fréchet space topology for all $d\in \mathbb{N}$. $(b)$ is the standard Fréchet space structure introduced in Example \ref{ex_A_smooth_function_frechet_space}, moreover, by Lemma \ref{lemma_0_CO_top_is_Frechet}, the locally convex structure in $(c)$ coincides with that of $(b)$. By $(iii)$ and $(iv)$ of Proposition \ref{prop_A_prop_of_c_infty_top}, $(b), \ (c)$ induce the same convenient structure. To show that the convenient structure of $(a)$ is equivalent to that of $(b)$ is to show that $id:C^{\infty}(M)_{(a)}\to C^{\infty}(M)_{(b)}$ is bounded. By construction  
    $$
        C^{\infty}(M)_{(a)}= \bigg\{ \{f_{\gamma}\}_{\gamma}\in \prod_{\gamma \in C^{\infty}(\mathbb{R},M)}C^{\infty}(\mathbb{R},\mathbb{R}): f_{\gamma}\circ \kappa =f_{\gamma\circ \kappa} \ \forall \kappa\in C^{\infty}(\mathbb{R},\mathbb{R})\bigg\}.
    $$
    is a closed subset of $\prod_{\gamma \in C^{\infty}(\mathbb{R},M)}C^{\infty}(\mathbb{R},\mathbb{R})$ with the Tychonoff topology. A subset $\mathcal{B}$ is bounded therein if for all smooth curves $\gamma:\mathbb{R}\to M$, $\{f\circ \gamma\in C^{\infty}(\mathbb{R},\mathbb{R}): f\in \mathcal{B}\} $ is bounded, that is, by \eqref{eq_A_top_unif_conv_on_smooth_curves}, if for all compacts $J\subset \mathbb{R}$, all $k\in \mathbb{N}$, there is a constant $C>0$ such that
    \begin{equation}\label{eq_A_bdd_proof}
        \sup_{t\in J} |(f\circ \gamma)^{(k)}(t)|\leq C.
    \end{equation}
    We can suppose that $\gamma(J)\subset U_{\alpha}$, if not we split $J$ into (finite) smaller compact intervals and then repeat the following argument. If the $\mathcal{B}$ was not bounded for $(b)$, there would be $f$, a sequence $\{x_n\}\subset U_{\alpha}$ such that $\sup_{n, \ j\leq k}|\nabla^jf(x_n)|=\infty$, then we construct a smooth curve $\gamma:[0,1] \to U_{\alpha}$ such that $\gamma(1/n)=x_n$, by assumption $\sup_{t\in [0,1]} |(f\circ \gamma)^{(k)}(t)|\leq C$ reaching a contradiction. To construct $\gamma$, simply take a piecewise straight curve joining all elements of the sequence, for example 
    % $$
    %     \gamma(t)=
    %             \begin{cases}
    %                 \ \ldots \ &\mathrm{if} \ t> 1/n,\\
    %                 \ n(1-(n+1)t)x_n+ (n+1)(nt-1)x_{n+1}  \ &\mathrm{if} \ t\in [1/(n+1),1/n],\\
    %                 \ \ldots \ & \mathrm{if} \   t<1/(n+1),                     
    %             \end{cases}
    % $$
    $$
        \gamma(t)= n(1-(n+1)t)x_n+ (n+1)(nt-1)x_{n+1}  \ \mathrm{if} \ t\in [1/(n+1),1/n]
    $$
    and then mollify appropriately on the edges. On the other hand if $\mathcal{B}$ is bounded in the usual Fréchet space structure, then for each compact $K\subset M$, $k\in \mathbb{N}$ there is $C>0$ such that
    $$
        f\in \mathcal{B} \ \Rightarrow \sup_{\substack{x\in K\\ j\leq k}} |\nabla^jf(x)|\leq C.
    $$
    If $\gamma:\mathbb{R}\to M$ is any smooth curve and $J\subset \mathbb{R}$ a compact interval, we can assume that $\gamma(J)\subset K$, then $(f\circ \gamma)^{(k)}$ is a polynomial in $\nabla^jf$ and $\gamma^{(l)}$, each of which is bounded in $K$ by the above constant, thus \eqref{eq_A_bdd_proof} is valid.
\end{proof}

\begin{proposition}\label{prop_A_conv_sections_of_vector_bndl}
    Let $\pi:E\to M)$ be a smooth finite dimensional vector bundle, then
    \begin{itemize}
        \item[$(i)$] a curve $\gamma:\mathbb{R}\to \Gamma^{\infty}(M \leftarrow E ) $ is smooth if and only if $\gamma^{\wedge}:\mathbb{R}\times M \to E$ is smooth;
        \item[$(ii)$] a curve $\gamma:\mathbb{R}\to \Gamma^{\infty}_c(M \leftarrow E ) $ is smooth if and only if $\gamma^{\wedge}:\mathbb{R}\times M \to E$ satisfies the following property: for each compact interval $I\subset \mathbb{R}$ there is a compact subset $K\subset M$ such that $\gamma^{\wedge}(t,x)$ is constant in $t\in I$ for each $x\in M\backslash K$.
    \end{itemize}
\end{proposition}

\begin{proof}
We remark that $\Gamma^{\infty}(M \leftarrow E )$ and $\Gamma^{\infty}_c(M \leftarrow E )$ are endowed respectively with the Fréchet topology and the LF-topology. The smoothness of $\gamma^{\wedge}:\mathbb{R}\times M \to E$ is equivalently tested on the respective charts of $M$, $E$, therefore we can assume that $M$, $E$ are open subsets $U$, $V$, respectively of $\mathbb{R}^n $, $\mathbb{R}^k $. In particular, those are $c^{\infty}$-open subsets, therefore we can apply Theorem \ref{thm_A_cart_closedness} and get that $\gamma^{\wedge}:\mathbb{R}\times M \to E$ is smooth if and only if $\gamma:\mathbb{R}\to C^{\infty}(U_{\alpha},V_{\alpha})$ with $U_{\alpha}$, $V_{\alpha}$ varying among open charts neighborhoods. Then we conclude, by a reasoning similar to the one employed in the proof of Proposition \ref{prop_A_conv_structures_on_C_infty(M)}, that this is equivalent to the convenient smoothness of $\gamma:\mathbb{R}\to \Gamma^{\infty}(M \leftarrow E ) $.

    For $(ii)$ we observe that $\Gamma^{\infty}_c(M \leftarrow E )$ is a strict direct limit of Fréchet spaces, thus $\gamma:\mathbb{R}\to \Gamma^{\infty}_c(M \leftarrow E ) $ is smooth if and only if it factors locally to a smooth curve $\gamma:I\subset \mathbb{R}\to \Gamma^{\infty}_K(M \leftarrow E ) $ for some $K\subset M$. Notice however, that $\Gamma^{\infty}_K(M \leftarrow E )\simeq \Gamma^{\infty}(K \leftarrow \pi^{-1}(K) )$, thus we can apply point $(i)$ and conclude. 
\end{proof}
\chapter{The algebraic approach to classical field theory}\label{chapter_classical}

\thispagestyle{plain}

% \nomenclature{$M_\infty$}{Free Stream Mach number}
% \newacronym{cfd}{CFD}{Computational Fluid Dynamics}

In this chapter we generalize the results of \cite{acftstructure} to the more complicated situation in which fields are sections of fibre bundles. At first sight the idea looks straightforward to implement, however it contains some not trivial subtleties whose treatment needs a certain degree of care. Indeed, in our general setting, images of the fields do not take value in vector spaces and moreover, the global configuration space $\Gamma^{\infty}(M\leftarrow B)$ completely lacks any vector space structure, thus admitting only a manifold structure. This forces us to generalize many notions like the support of functionals, or their central notion of locality/additivity, over configuration space, which can be given in two different formulations, one global that uses the notion of relative support already used in \cite{gravbrunetti} and a local one that uses the notion of charts over configuration space seen as a infinite dimensional manifold. It is gratifying that both notions give equivalent results, as shown \textit{e.g.}\ in Proposition~\ref{prop_1_additivity}. This added generality does not spoil the existence of a Poisson algebra structure (\Cref{thm_1_peierls_closedness}, \Cref{thm_1_jacobi}), the $C^{\infty}$-ring structure (Proposition \ref{prop_1_C-infty_ring}), or the existence of partitions of unities (Proposition \ref{prop_1_mu_caus_top_prop}) for microcausal functionals. On the contrary, other properties, such as the characterization of microlocal functionals (Proposition \ref{porop_1_muloc_charachterization}) are valid only in the chart neighborhood in which are derived, we argue, using the variational sequence of \cite{krupka}, that if we do not make extra assumptions the latter result is not globally\footnote{In the sense of the infinite dimensional manifold structure of $\Gamma^{\infty}(M\leftarrow B)$} extensible.\\

This chapter is organized as follows: \Cref{section_topology_map_manifold} and \Cref{section_map_manifold} are devoted to introduce the infinite dimensional geometric formalism which will be the starting point of our analysis. In particular, we will endow $\Gamma^{\infty}(M\leftarrow B)$ with a topology (\textit{c.f.} Definition \ref{def_1_refined_whitney_top}) and an infinite dimensional manifold structure where the smooth structure (\textit{c.f.} \Cref{thm_1_mfd_mappings}) is defined with respect to Bastiani calculus developed in \Cref{section_Bastiani}. We remark that this is not the only choice: for instance one could use \cite[Theorem 42.1]{convenient} to create a convenient smooth structure for this space, the construction is identical, however the topology on the configuration space is finer however, since conveniently smooth mappings are not continuous in general (see \cite{glockner2005discontinuous}). This last feature is quite problematic since we will handle the space of smooth functions over the manifold $\Gamma^{\infty}(M\leftarrow B)$ and we wish to regard those as \textit{smooth and continuous} functions, this motivate us to pick the Bastiani smooth manifold structure. 

\Cref{section_observables} focuses on the definition of observables, their support, and the introduction of various classes of observables depending on their regularity. In particular two of this class admit a \textit{ultralocal} characterization, \textit{i.e.} local in the sense of the manifold structure of the space of sections. In the end we introduce the notion of \textit{generalized Lagrangian}, essentially showing that each Lagrangian in the standard geometric approach is a Lagrangian in the algebraic approach as well. We then discuss how linearized field equations are derived from generalized Lagrangians. 

In \Cref{section_peierls_bracket} we show the existence of the causal propagator which in turn is used in Definition \ref{def_1_Peierls} to define the Poisson bracket on the class of \textit{microlocal functionals}. Then we enlarge the domain of the bracket to the so-called \textit{microcausal functionals}, defined by requiring a specific form of the wave front set of their derivatives. Finally Proposition 
\ref{prop_1_Peierls_1}, Theorems \ref{thm_1_mucaus_1}, \ref{thm_1_peierls_closedness} and \ref{thm_1_jacobi} establish the Poisson $*$-algebra of microcausal functionals. 

Finally, we collect in \Cref{section_properties_of_muc_functionals} a series of results that culminate in Theorem \ref{thm_1_mucaus_top}, which establishes that microcausal functionals can be given the topology of a nuclear locally convex space. Furthermore Propositions \ref{prop_1_C-infty_ring} and \ref{prop_1_mu_caus_top_prop} give additional properties concerning this space and its topology. We conclude the section by defining the on-shell ideal with respect to the Lagrangian generating the Peierls bracket and the associated Poisson $*$-algebraic ideal. 

Eventually in \Cref{section_examples} we briefly show how to adapt the previous results to the case of wave maps and discuss the case of scalar field theories described in \cite{acftstructure}. The latter setting will be crucial for describing Wick powers and time-ordered products in Chapters \ref{chapter_Wick}, \ref{chapter_TO}.

\section{Topologies on the space of sections, manifolds of mappings}\label{section_topology_map_manifold}

Let $M$, $N$ be finite dimensional paracompact Hausdorff topological spaces, denote the space of continuous functions by $C(M,N)$. The \textit{compact open} topology $\tau_{CO}$ or CO-topology is the topology generated by a basis whose elements have the form 
\begin{equation}\label{eq_1_CO_open}
    N(K,V)=\{ \varphi \in C(M,N) : \varphi(K)\subset V\}, 
\end{equation}
where $K\subset M$ is a compact subset and $V\subset N $ is open. Roughly speaking, this topology controls the behaviour of functions only on small regions of $M$, whereas their behaviour "at infinity" is not specified.

\begin{lemma}\label{lemma_1_CO_is_Hausdorff}
    Let $M$, $N$ as described above, if $N$ is normal, then $\big( C(M,N), \tau_{CO}\big)$ is Hausdorff. 
\end{lemma}
\begin{proof}
Supposing $\varphi \neq \psi$, then at least $\varphi(x)\neq \psi(x) $ for some $x\in M$. By continuity of $\varphi$, $\psi$ there exists an open subset $U_x$ such that $\varphi(y)\neq \psi(y)$ for each $y\in \bar{U}_x$. Without loss of generality we can suppose that $\overline{U_x}$ is compact, then $\varphi(\overline{U_x})$, $\psi(\overline{U_x})$ are compact and therefore closed. Since $N$ is normal, there are disjoint open subsets $V_{\varphi}$, $V_{\psi}$ respectively containing $\varphi(\overline{U_x})$, $\psi(\overline{U_x})$, then
$$
    N\big(U_x,V_{\varphi} \big)\cap N\big(U_x,V_{\psi} \big)=\emptyset.
$$
\end{proof}

\begin{lemma}\label{lemma_1_CO_mertizability}
Let $M$, $N$ be topological spaces, if $N$ is a complete metric space then $\big( C(M,N), \tau_{CO}\big)$ is a complete metric space as well.
\end{lemma}

If $M$ is compact and $(N,d_N)$ is metric, then a neighborhood of $\varphi$ in the compact-open topology can be given as 
$$
    B_{\epsilon}(\varphi)\doteq\{ \psi \in C(M,N): d_N\big(\varphi(x),\psi(x)\big)<\epsilon(x) \forall x\in M\big\},
$$
where $\epsilon:M \to \mathbb{R}_+$ is a continuous function.

\begin{lemma}
Let $M$, $N$, $L$ be topological spaces with $M$ locally compact and Hausdorff, then 
$$
    C(M\times N, L) \simeq C\big( M, C(N,L)\big)
$$
where $C(N,L)$ posses the CO-topology.
\end{lemma}

Given $\varphi\in C(M,N)$, let $G_{\varphi}:M\to M\times N$ be the graph mapping associated to $\varphi$, set $\mathrm{Im}(G_{\varphi})\equiv \mathrm{gh}(\varphi)=\{(x,\varphi(x))\in M\times N: x\in M\}$. 

\begin{definition}\label{def_1_WO-top}
    The \textit{wholly open topology} $\tau_{\mathrm{WO}}$ or $\mathrm{WO}$-topology on $C(M,N)$ is generated by a subbasis of open subsets of the form 
    \begin{equation}\label{eq_1_WO-open}
        W(V)=\{ \varphi \in C(M,N): \varphi(M)\subset V\},        
    \end{equation}
    where $V\subseteq N$ is open.
\end{definition}

Note that the $\mathrm{WO}$-topology is not Hausdorff, for it cannot separate surjective functions.

\begin{definition}\label{def_1_WO^0-top}
    The \textit{graph topology} $\tau_{\mathrm{WO}^0}$ or ${\mathrm{WO}}^0$-topology on $C^{\infty}(M,N)$ is the one induced by requiring 
    $$
        G:C(M,N) \ni \varphi \mapsto G_{\varphi} \in \big( C(M,M\times N),\tau_{\mathrm{WO}}\big)
    $$
    to be an embedding.
\end{definition}

By Definition \ref{def_1_WO-top} the open subbasis of $C(M,M\times N)$ is given by subsets of the form
$$
    W(\widetilde{V}) = \{ f\in C(M,M\times N): f(M)\subset \widetilde{V}\}
$$
with $\widetilde{V}\subset M\times  N$ open subsets. When $f=G_{\varphi}$ for some $\varphi\in C(M,N)$, then the trace topology on the subset $G\big(C(M,N)\big)$ is generated by a subbasis of elements $W(\widetilde{V})$ where $\widetilde{V}=M\times V$ with $V\subset N$ open subset. Clearly $G$ is an injective mapping, and bijective onto its image. Therefore a subbasis for the $\mathrm{WO}^0$-topology is given by 
\begin{equation}\label{eq_1_WO^0-open}
    W(\widetilde{V})= \{ \varphi \in C(M,N) : G_{\varphi}\subset M\times V\}
\end{equation}
\begin{lemma}
    The $\mathrm{WO}^0$-topology is finer then the CO-topology and is therefore Hausdorff. 
\end{lemma}
\begin{proof}
We show that $\mathrm{id}_{C(M,N)}: \big( C(M,N),\tau_{\mathrm{WO}^0}\big) \to \big( C(M,N),\tau_{\mathrm{CO}}\big)$ is continuous. Let $N(K,V)$ be an open subset as in \eqref{eq_1_CO_open}, $U_1$, $U_2$ be a cover of $M$ such that $K\subset U_1$ and $U_2 =M\backslash K$. Consider the open subset
$$
    W(U_1\times V \cup U_2 \times N)=\{ \varphi \in C(M,N) : G_{\varphi}\subset U_1\times V \cup U_2 \times N\};
$$
the former is a $\mathrm{WO}^0$-open subset, which is however equal to $N(K,V)$.
\end{proof}

We stress that when the manifold $M$ is compact, then the graph and compact-open topology are equivalent, whereas in full generality the former is finer. The main difference from the compact open topology is that the graph topology does control the behaviour of a mapping over the whole space, while the former was limited to a compact region.

\begin{lemma}\label{lemma_1_WO^0-metric-subbasis}
    Let $M$ be paracompact and $(N,d)$ be a metric space, then a basis of neighborhood of $\varphi\in C(M,N)$ for the $\mathrm{WO}^0$-topology is given by 
    \begin{equation}
        W_{\varphi}(\epsilon)=\big\{ \psi \in C(M,N) : d\big(\varphi(x),\psi(x)\big)<\epsilon(x) \ \forall x\in M\big\},
    \end{equation}
    where $\epsilon:M \to \mathbb{R}_+$ is continuous.
\end{lemma}

\begin{proposition}\label{prop_1_WO^0-convergence}
    Let $M$ be paracompact and $(N,d)$ be a metric space then, for any sequence $\{\varphi_n\}\subset C(M,N)$, the following are equivalent:
    \begin{itemize}
        \item[$(i)$] $\varphi_n \to \varphi$ in the $\mathrm{WO}^0$-topology;
        \item[$(ii)$] there exists a compact set $K\in M$ such that $\varphi_n\big\vert_{M\backslash K}\equiv \varphi\big|_{M\backslash K}$ for each $n\in \mathbb{N}$, and $\varphi_n \to \varphi$ uniformly on $K$.
    \end{itemize}
\end{proposition}

Notice that, due to Proposition \ref{prop_1_WO^0-convergence}, the space $C(M,E)$ with $E$ vector space is not a topological vector space, in particular the multiplication mapping cannot be continuous since if $\lambda\in \mathbb R $ goes to $0$, then $\lambda\cdot f \not\to 0$ unless $f=0$ outside some compact subset of $M$.

\begin{proof}
Suppose $\varphi_n \to \varphi $ in the $\mathrm{WO}^0$-topology, however, for all $K\subset M$ compact, either $\varphi_n \not\rightarrow \varphi$ uniformly over $K$ or there is $x\in M\backslash K$ such that $\varphi_n(x) \neq \varphi(x)$. In the first case $\varphi_n \not\rightarrow \varphi$ in the CO-topology as well, contradicting the initial hypothesis. In the second case, let $\{K_n\}$ be an exhaustion of compact subsets of $M$, then for each $n\in \mathbb{N}$ there is $x_n \in M\backslash K_n$ having $\varphi_n(x_n)\neq \varphi(x_n)$. Set $0<\epsilon_n =\sup_{K_n}d(\varphi_n(x),\varphi(x))$. For each $n$ and consider the sequence of open neighborhoods of $\varphi$, $W_{\varphi}(\epsilon_n)$ as per Lemma \ref{lemma_1_WO^0-metric-subbasis}, by construction $\varphi_n \notin W_{\varphi}(\epsilon_n)$ which contradicts the convergence hypothesis. On the other hand let $\epsilon_n:M \to \mathbb{R}_+$ be the constant functions with $\epsilon_n= \sup_{x\in K}d(\varphi_n(x),\varphi(x))$, then $W_{\varphi}(\epsilon_{n_0})\cap \{ \varphi_n\}=\{\varphi\}_{n>n_0}$ by uniform convergence over $K$. Implying $\varphi_n\to \varphi$ in $\tau_{\mathrm{WO}^0}$.
\end{proof}

\begin{corollary} \label{coro_1_WO^0-curves}
    Let $M$ and $N$ as in Proposition \ref{prop_1_WO^0-convergence} and $\gamma:I\subset \mathbb{R}\to \big(C(M,N),\tau_{\mathrm{WO}^0}\big)$ be a continuous mapping with $I$ compact. Then there exists a compact $K\subset M$ such that 
    $$
        \gamma(t) : x \in  M \to N
    $$
    is constant in $M\backslash K$ for each $t\in I$.
\end{corollary}

\begin{proof}
We argue by contradiction, let $K_n$ be an exhaustion of compact subsets of $M$, then for each $n\in \mathbb{N}$ there is some $t_n\in I$, and some $x_n\in M\backslash K_n$ such that $\gamma(t_n)[x_n]\neq \gamma(t)[x_n]$ for at least a $t\in I$. Since $\{t_n\}$ is a sequence on a compact space we may assume, eventually passing to a subsequence, that $t_n\to t_0\in I$, by construction, $\{x_n\}$ does not admit a cluster point in $M$. Finally, by continuity, $t_n \to t_0 \implies \gamma(t_n)\to \gamma(t)$ in the $\mathrm{WO}^0$-topology, by Proposition \ref{prop_1_WO^0-convergence} there has to be a compact subset $K$ such that $\gamma(t_n)\equiv \gamma(t)$ outside $K$ thus the sequence $\{x_n\}$ admits a cluster point.
\end{proof}

From now on we assume that $M$, $N$ are smooth $m$, $n$ dimensional manifolds respectively, consider the $k$th order jet bundle $J^k(M,N)$. Recall as well the mappings $\alpha:J^k(M,N) \to M$, $\beta:J^k(M,N) \to N$.

\begin{lemma}\label{lemma_1_j^k-CO}
    Given $M$, $N$ smooth manifolds, the mapping 
    $$
        j^k: C^{\infty}(M,N) \to \big( C(M,J^k(M,N)),\tau_{CO}\big)
    $$
    is injective and has closed image.
\end{lemma}

\begin{proof}
    Injectivity follows from the fact that $\beta \circ j^k=\mathrm{id}_{C^{\infty}(M,N)}$. Since $J^k(M,N)$ is a manifold, it is metrizable as well, thus, by Lemma \ref{lemma_1_CO_mertizability}, $C(M,J^k(M,N)$ is metrizable in the CO-topology. We can conclude if we show that given $\{j^k\varphi_n\}\subset j^k\big(C^{\infty}(M,N)\big)$ with $j^k\varphi_n \to \psi$ uniformly, $\psi=j^k\varphi$. If $\{(U_i,u_i)\}$, $\{(V_j,v_j)\}$ are a family of charts of $M$ and $N$ respectively, it is enough to show the claim locally for each $g_n= v_j\circ \varphi_n \circ u_i^{-1}$. It is clear that the claim holds for $k=0$, since by uniform convergence $j^0g_n=g_n\to g $. If $k=1$, we can assume that the open sets $U_i$ are convex, then setting $h=\lim_{n\to \infty}dg_n $ we show that $h= dg $. 
    $$
    \begin{aligned}
        g(x+y) &= \lim_{n\to \infty} \bigg( g_n(x)+ \int_0^1\langle dg_n(x+ty),y\rangle dt \bigg)\\ 
        & = g_0(x)+\int_0^1 \langle h(x+ty),y\rangle dt
    \end{aligned}
    $$
    therefore $dg=h$. Iterating this argument for all orders up to $k$ yields the desired result.
\end{proof}

In view of Lemma \ref{lemma_1_j^k-CO}, we can state the following definition.

\begin{definition}
    The $\mathrm{CO}^k$-\textit{topology} on $C^r(M,N)$ for $0\geq k\geq r\geq \infty $ is the topology induced by requiring that 
    $$
        j^k:C^r(M,N) \to \Big(C\big( M, J^k(M,N) \big),\tau_{\mathrm{CO}} \Big)
    $$
    is a topological embedding.
\end{definition}

\begin{corollary}
    The $\mathrm{CO}^k$-topology has the following properties:
    \begin{itemize}
        \item[$(i)$] $\big(C^k(M,N),\mathrm{CO}^k \big)$ is a complete metric space, moreover it coincides with the standard Fréchet vector space structure;
        \item[$(ii)$] if $k>r$, $\big(C^r(M,N),\mathrm{CO}^k \big)$ is metrizable but not complete. Its completion is $C^k(M,N)$.
    \end{itemize}
\end{corollary}

\begin{definition}\label{def_1_WO^k_topology}
    The \textit{Whitney} $\mathrm{C}^k$-\textit{topology}, or $\mathrm{WO}^k$-topology, on $C^r(M,N)$ for $0\leq k\leq r \leq \infty$ is the topology induced by requiring
    $$
        j^k:C^r(M,N) \to \Big( C\big(M,J^k(M,N)\big),\tau_{\mathrm{WO}^0}\Big)
    $$
    to be topological embedding.
\end{definition}

\begin{proposition}\label{prop_1_WO^k_properties}
    The $\mathrm{WO}^k$-topology on $C^r(M,N)$ enjoys the following properties:
    \begin{itemize}
        \item[$(i)$] A subbasis of open subsets of the topology have the form 
        \begin{equation}\label{eq_1_WO^k-open}
            W(\widetilde{U})=\{ \varphi \in C^r(M,N): j^k\varphi(M)\subset \widetilde{U}\}
        \end{equation}
        where $\widetilde{U}\subset J^k(M,N)$ is an open subset.
        \item[$(ii)$] If $d_k$ is a metric on $J^k(M,N)$, then a basis of neighborhoods for the $\mathrm{WO}^k$-topology of $\varphi\in C^r(M,N)$ is
        $$
            N_{\varphi}^k(\epsilon)=\{ \psi\in C^r(M,N): d_k(j^k_x\psi,j^k_x\varphi) < \epsilon(x)\}
        $$
        where $\epsilon\in C(M,\mathbb{R}_+)$.
        \item[$(iii)$] The sequence $\{\varphi_n\}\subset C^k(M,N)$ converges to $\varphi$ in the $\mathrm{WO}^k$-topology if and only if there is a compact subset $K\subset M$ such that $\varphi_n\equiv \varphi $ in $M\backslash K$ and $j^k\varphi_n \to j^k\varphi$ uniformly over $K$. 
        \item[$(iv)$] If $I\subset \mathbb{R}$ is compact and $\gamma:I \to \big( C^r(M,N), \tau_{\mathrm{WO}^k} \big)$ is continuous, then there is a compact subset such that 
        $$
            \mathrm{ev}_x\gamma: I\ni t \mapsto \gamma(t)[x]
        $$
        is constant for all $x\in M\backslash K$.
%        \item[5)] $\big( C^r(M,N), \mathrm{WO}^k\big)$ is a Baire space. 
        \item[$(v)$] $\mathrm{WO}^{\infty}$ on $C^{\infty}(M,N) $ is the projective limit topology of all $\mathrm{WO}^k$-topologies for $0\leq k\leq \infty$.
        \item[$(vi)$] A basis of open neighborhood of the $\mathrm{WO}^{\infty}$ topology on $C^{\infty}(M,N)$ consists of open subsets 
        \begin{equation}\label{eq_1_WO^infty-open}
            W(\widetilde{U})=\{ \varphi \in C^{\infty}(M,N): j^{\infty}\varphi(M)\subset \widetilde{U}\}
        \end{equation}
        where $\widetilde{U}\subset J^{\infty}(M,N)$ is open.
        \item[$(vii)$] If $\{K_n\}_n$ is an exhaustion of compact subsets of $M$, a basis for the $\mathrm{WO}^{\infty}$ topology on $C^{\infty}(M,N)$ consists of open subsets
        $$
            M(U,n)= \{ f \in C^{\infty}(M,N): j^nf(M\backslash K_n^{\mathrm{o}})\subset U_n\}
        $$
        where $U_n \subset J^n(M,N)$ are open.
    \end{itemize}
\end{proposition}
\begin{proof}
We claim that on the image of $j^k$ in $C(M, J^k(M,N))$ in the $\mathrm{WO}^0$ and $\mathrm{WO}$ topology coincide. Indeed 
$$
    G_{j^k\varphi} (M)\subset M\times J^{k}(M,N)=\{ (x,j^k_x\varphi):x\in M)\}\simeq J^k(M,N),
$$
where the last is a topological embedding, therefore open subsets of $J^k(M,N)$ and $G_{j^k(C^r(M,N))}\subset  M \times J^k(M,N)$ coincide. 
As a result we obtain \eqref{eq_1_WO^k-open} by combining the above result with \eqref{eq_1_WO-open}. $(ii)$ follows by combining $(i)$ with Lemma \ref{lemma_1_WO^0-metric-subbasis}. Using that $\varphi_n\to \varphi $ in $\mathrm{WO}^k$ if and only if $j^k\varphi_n\to j^k\varphi$ in $\mathrm{WO}^0$ over $C(M,J^k(M,N))$ in conjunction with Proposition \ref{prop_1_WO^0-convergence} we get $(iii)$. Similarly $(iv)$ is obtained by combining Corollary \ref{coro_1_WO^0-curves} with the above argument. The argument for $(v)$ and $(vi)$ is the following: the topology on $J^{\infty}(M,N)$ is the coarsest such that each $\pi^{\infty}_k: J^{\infty}(M,N) \to J^k(M,N)$ is continuous. Note that we have an embedding
$$
    J^{\infty}(M,N)\simeq M\times_MJ^{\infty}(M,N) \hookrightarrow M\times J^{\infty}(M,N).
$$
Therefore we construct the following commutative diagram
\begin{center}
\begin{tikzcd}
		 & J^{\infty}(M,N) \arrow[r,hook] \arrow[d,"\pi^{\infty}_k"]  & M\times J^{\infty}(M,N) \arrow[d, "\mathrm{id}_M\times\pi^{\infty}_k"] \\
		 & J^{k}(M,N) \arrow[r,hook] & M\times J^k(M,N),
\end{tikzcd}
\end{center}
where the horizontal mapping are embeddings. Then a subbasis of $C\big( M,J^{\infty}(M,N)\big)$ for the $\mathrm{WO}^0$-topology is 
$$
    W(U_{\infty})=\{ j^{\infty}\varphi \in C^{\infty}(M,N): j^{\infty}(M)\subset U_{\infty}\}
$$
with $U_{\infty}\subset J^{\infty}(M,N)$ open. Finally we show $(vii)$. Let $U_n\subset J^n(M,N)$ be open subsets, then each
$$
    M(U)= \{ f \in C^{\infty}(M,N): \ \forall n\in \mathbb N,\  j^{\infty}f(M\backslash K_n^{\mathrm{o}})\subset (\pi^{\infty}_n)^{-1}U_n\};
$$
is an open subset of the $\mathrm{WO}^{\infty}$ topology. Setting $V_n=(\pi^{\infty}_0)^{-1}(U_0)\cap \cdots \cap (\pi^{\infty}_n)^{-1}(U_n)$ we have that
$$
    \{ f \in C^{\infty}(M,N):  \ \forall n\in \mathbb N,\ j^{\infty}f(K_{n+1}\backslash K_n^{\mathrm{o}})\subset V_n\}= M(U).
$$
The inclusion $\supset$ is clear, for the other, observe that in each region $K_{n+1}\backslash K_n^{\mathrm{o}}$ we have the requirement $ j^{\infty}f(K_{n+1}\backslash K_n^{\mathrm{o}})\subset V_n$ for all $n$ which is way stronger then the corresponding $ j^{\infty}f(M\backslash K_n^{\mathrm{o}})\subset (\pi^{\infty}_n)^{-1}U_n$ for all $n$. Since $J^{\infty}(M,N)$ is a fiber bundle with finite dimensional base and Fréchet space fiber and has the coarsest topology making each $\pi^{\infty}_k$ continuous, we may write
$$
    M(\widetilde{V})=\{ f \in C^{\infty}(M,N): \ \forall n\in \mathbb N,\   j^{\infty}f(K_{n+1}\backslash K_n^{\mathrm{o}})\subset \widetilde{V}_n\}.
$$
where each $\widetilde{V}_n\subset J^{\infty}(M,N)$ is open. We claim that $M(\widetilde{V})$ generates a topology equivalent to the $\mathrm{WO}^{\infty}$ topology. The latter's open subsets posses the form \eqref{eq_1_WO^infty-open}, is thus clear that $M(\widetilde{V})\subset W(\cup_n\widetilde{V}_n)$ thus making the former topology finer then the latter. To see the converse observe that $j^{\infty}f(K_{n+1}\backslash K_n^{\mathrm{o}})$ is a compact subset of a metric space for each $n$, thus there is some $\epsilon_n>0 $ for which the open subset $\{j^{\infty}_xg \in \mathbb{R}^{\infty} :d(j^{\infty}_xf,j^{\infty}_xg)<\epsilon_n \ \forall x\in K_{n+1}\backslash K_n^{\mathrm{o}} \} \subset \widetilde{V}_n$. Let then $\epsilon\in C(M)$ be a continuous function such that $\epsilon(x)<\epsilon_n$ for all $x\in K_{n+1}\backslash K_n^{\mathrm{o}}$, then $N^{\infty}_{f}(\epsilon)$ is an open subset of the Whitney topology which is contained in $M(\widetilde{V})$.
\end{proof}

As a consequence of Proposition \ref{prop_1_WO^0-convergence}, when $N$ is metrizable, and $M$ paracompact and second countable, if $\lbrace f_{n}\rbrace_{n \in \mathbb{N}}$ is a sequence we can characterize convergence in the following way:
\begin{itemize}
\item[$(i)$] $f_n \rightarrow f $ in the $\mathrm{WO}^{\infty}-$topology ,
\item[$(ii)$] $\forall n'\in \mathbb{N} \space\ \exists K_{n'} \subset M$ compact such that if $n\geq n'$ then $\left. f_n \right|_{M \backslash K_{n'}}=\left. f \right|_{M \backslash K_{n'}}$ and $\left. f_n \right|_{K_n'} \rightarrow \left. f \right|_{K_n'}$ uniformly with all its derivatives.
\end{itemize}

This fact has important implications, for if we consider the finite dimensional vector bundle $E \to M$, the vector space $\Gamma^{\infty}(M\leftarrow E)$ will not be a topological vector space due to the failure of continuity for the multiplication by scalar. This can be readily seen from condition $(ii)$ above: if for instance we had $\sigma\in \Gamma^{\infty}(M\leftarrow E)$, $\mathbb{R} \ni \epsilon_n \to 0$ and $\epsilon_n\sigma \to 0$, then each $\epsilon_n\sigma$ must possess compact support, thus $\sigma$ itself ought to be compactly supported. As a consequence we get the following result:

\begin{theorem}\label{thm_1_Gamma_c_TVS}
    Let $(E,\pi,M)$ be a finite dimensional vector bundle, then $\Gamma^{\infty}_c(M\leftarrow E)\subset \Gamma^{\infty}(M\leftarrow E)$, equipped with trace of the Whitney topology on $C^{\infty}(M,E)$. Then $\Gamma^{\infty}_c(M\leftarrow E)$ is the maximal locally convex space contained in $\Gamma^{\infty}(M\leftarrow E)$. Moreover the trace topology coincides with the (natural) final topology induced by the projective limit
\begin{equation}
	\lim_{\substack{\longrightarrow\\K\subset M}}\Gamma^{\infty}_K(M\leftarrow E) =\Gamma^{\infty}_c(M\leftarrow E).
\end{equation}
Consequently, $\Gamma^{\infty}_c(M\leftarrow E) $ is a complete, nuclear and Lindel\"of space, hence paracompact and normal. In particular, for each open cover $\mathcal{U}_i$ of $\Gamma^{\infty}_c(M\leftarrow E)$, there are Bastiani smooth bump functions $\rho_i:\Gamma^{\infty}_c(M\leftarrow E)\to \mathbb{R}$ each of which has $\mathrm{supp}(\rho_i)\subset \mathcal{U}_i$, satisfying
    $$
        \sum_i\rho_i(\sigma)=1 
    $$
    for each $\sigma\in \Gamma^{\infty}_c(M\leftarrow E)$.
\end{theorem}

From a topological standpoint, Theorem \ref{thm_1_Gamma_c_TVS} implies that $\Gamma^{\infty}_c(M\leftarrow E)\subset \Gamma^{\infty}(M\leftarrow E)$ equipped with the Whitney topology is the maximal topological vector subspace. Thus if we want to give a topological manifold structure to spaces such as $\Gamma^{\infty}(M\leftarrow E)$ (or $C^{\infty}(M,N)$) this forces us to use $\Gamma^{\infty}_c(M\leftarrow E)$ as the topological vector space on which to model the manifold (see Definition \ref{def_1_infinite_dim_mfd}). We would therefore like to have charts of the form $(\sigma_0+\Gamma^{\infty}_c(M\leftarrow E),u_{\sigma_0})$ where $u_{\sigma_0}(\sigma)=\sigma-\sigma_0$. The undesirable fact is that $\sigma_0+\Gamma^{\infty}_c(M\leftarrow E)$ would then become a closed subset. To remedy this problem we refine the Whitney topology \textit{just enough} to make the above subsets open. To wit consider the following equivalence class: given $M$, $N$ smooth finite dimensional manifolds, set $f \sim g$ if $\mathrm{supp}_f(g)= \overline{\lbrace x \in M : f(x) \neq g(x)\rbrace}\subset M$ is compact. 

\begin{definition}\label{def_1_refined_whitney_top}
    The refined Whitney topology, or refined $\mathrm{WO}^{\infty}$ topology, is the coarsest topology on $C^{\infty}(M,N)$ which is finer than the $\mathrm{WO}^{\infty}$-topology and for which the sets $\mathcal{U}_f=\lbrace g \in C^{\infty}(M,N): g \sim f \rbrace$ are open.
\end{definition}

The refined Whitney topology has the same converging sequences and smooth curves as the Whitney topology since the proofs of \ref{prop_1_WO^0-convergence} and Corollary \ref{coro_1_WO^0-curves} remains essentially valid. The reason is that the refinement we imposed on the topology was made by adding \textit{big} open subsets \textit{i.e.} the trace topology on subspaces of the form $\Gamma^{\infty}_c(M\leftarrow E)$ is not altered, thus the properties in questions remain valid. Clearly $\Gamma^{\infty}_c(M\leftarrow E) \subset\Gamma^{\infty}(M\leftarrow E)$ will then become open and moreover $\Gamma^{\infty}(M\leftarrow E)$ becomes a topological affine space with model topological vector space $\Gamma^{\infty}_c(M\leftarrow E)$. Notice also that the space $C^{\infty}(M,N)$ will no longer be a Baire space for example, if $N=\mathbb{R}$ and $K_n$ is an exhaustion of compact subsets, then $\cup_n C^{\infty}_{K_n}(M)=C^{\infty}_c(M)$, however $C^{\infty}_{K_n}(M)\subset C^{\infty}_{c}(M)$ is not dense for all $n\in \mathbb{N}$.

\begin{proposition}\label{prop_1_continuity_of_push_forward}
    Let $M$, $M'$, $N$, $N'$ be smooth finite dimensional manifolds,
    \begin{itemize}
        \item[$(i)$] if $f:M'\to M$ is a proper smooth mapping, then $f^*:C^{\infty}(M,N) \to C^{\infty}(M',N) $ is continuous in both the Whitney and refined Whitney topology;
        \item[$(ii)$] if $h:N\to N'$ is a smooth mapping, then $h_*:C^{\infty}(M,N) \to C^{\infty}(M,N') $ is continuous in both the Whitney and refined Whitney topology.
    \end{itemize}
\end{proposition}

\begin{proof}
    The proof of $(i)$ can be directly obtained by using Proposition 7.3 in \cite{michor1980manifolds} while keeping in mind that $f^*(g)=g\circ f$. For $(ii)$ consider a Whitney open subset $M'(U)=\{ f\in C^{\infty}(M,N') : \ j^nf(M\backslash K_n)\}\subset U_n \subset J^{n}(M,N') \forall n \in \mathbb N \}$ in $C^{\infty}(M,N')$. The mapping $j^nh : J^n(M,N) \to J^n(M,N')$ is smooth and hence continuous, thus set $V_n=(j^nh)^{-1}(U_n)$. Then $(h_*)^{-1}(M'(U))=M(V)$ which implies continuity of $h_*$ in the Whitney topology. For the refined Whitney topology one notes that if $f\sim f'$, then $h\circ f \sim h \circ f'$ as well, thus $h_*\mathcal{U}_f=\mathcal{U}_{h_*f}$, which in turn implies that $h_*$ remains continuous even when refining the topology. 
\end{proof}

% \begin{proposition}\label{prop_1_smooth_Bastiani_curve}
%     The mapping $C^{\infty}\big(M,\big(C^{\infty}(N,P),WO^{\infty}\big)\big) \to C^{\infty}\big(M\times N,P\big)$ is valued in the set of smooth functions $f:M\times N \to P$ such that for each compact $K\subset M$ there is a compact $L\subset N$ such that $f(x,y) $ is constant for all $x\in K$, $y\in N\backslash L$. If $N$ is not compact, the latter is not an isomorphism.
% \end{proposition}
% \begin{proof}
%     \textcolor{red}{proove the result!!!!!!!!!!!!!!!!!!!!!!}
% \end{proof}

\begin{theorem}\label{thm_A_Bastiani smooth_pushforward}
    Let $(E_i,\pi_i,M)$, $i=1,2$ be finite dimensional vector bundles, suppose that $\alpha: U\subset E_1\to E_2 $ is a smooth fibered morphism projecting to the identity of $M$ and let $\sigma_0\in \Gamma^{\infty}_c(M\leftarrow E_1)$ having $\sigma_0(M)\subset U$, $\alpha(\sigma_0)\in \Gamma^{\infty}_c(M\leftarrow E_2)$. Then the mapping $\alpha_{*}:\mathcal{U}=\{\sigma\in \Gamma^{\infty}_c(M\leftarrow E_1): \sigma(M)\subset U\} \to \Gamma^{\infty}_c(M\leftarrow E_2)$ is a Bastiani smooth mapping, moreover, if $d_v\alpha:VU\to E_2$ is the vertical derivative of $\alpha$, we have $d(\alpha_*)=(d_v\alpha)_*$.
\end{theorem}

We remark that given any connection on the vector bundle $E_1$ it induces a splitting $TE_1=HE_1\oplus VE_1$ into horizontal and vertical vector bundle, the latter is of course independent from the connection chosen. Thus if we use local fibered coordinates on $E_1$ and $E_2$ induced by local frames $e_i$, $f_i$ respectively and study $\alpha $ in a neighborhood of $p\in E_1$, $\alpha(p)= \sum \alpha^i(x,y)f_i$, then $d_v\alpha:V_pE_1\to E_2, \ (p;\sigma) =\sum \frac{\partial \alpha^i(x,y)}{\partial y^j}\sigma^jf_i \in E_2|_{x}$.

\begin{proof}
    Its clear that if $U$ is open, $\mathcal{U}=\{\sigma\in \Gamma^{\infty}_c(M\leftarrow E_1): \sigma(M)\subset U\}$ is open in $\Gamma^{\infty}_c(M\leftarrow E_1)$ with the Whitney topology as well by $(vi)$ of Proposition \ref{prop_1_WO^k_properties} taking $\widetilde{\mathcal{U}}=(\pi^{\infty})^{-1}(M\times U)$. Next we show that $\alpha_*$ is Bastiani differentiable. This is equivalent to show that $d(\alpha_*):\mathcal{U}\times \Gamma^{\infty}_c(M\leftarrow E_1)\to \Gamma^{\infty}_c(M\leftarrow E_2)$ exists and is continuous in the Whitney topology (see Definition \ref{def_1_WO^k_topology}). We thus claim that  
    \begin{equation}
        \lim_{t\to 0} \frac{\alpha_*(\sigma+t\sigma')-\alpha_*(\sigma)}{t} = (d_v\alpha)_*(\sigma;\sigma')
    \end{equation}
    in the Whitney topology. We start by showing that for any neighborhood $U$ of $x_0\in M$, $\frac{\alpha_*(\sigma+t\sigma')-\alpha_*(\sigma)}{t}$ converges uniformly to $(d_v\alpha)_*(\sigma;\sigma')(x)$ in $U$. This is a local problem and we can thus study it using local coordinates. Notice that if $U$ lies outside the support of $\sigma'$ then the claim is trivial. By an abuse of notation we set $\sigma(x)=(x,\sigma(x))$ and likewise for $\sigma'$. Then by Taylor theorem we have
    $$
        \begin{aligned}
            \big(\alpha_*(\sigma+t\sigma')-\alpha_*(\sigma)\big)^i(x)&= \alpha^i(x,\sigma(x)+t\sigma'(x))-\alpha^i(x,\sigma(x)) = t\partial_j\alpha^i(x,\sigma(x))\sigma'^j(x)\\ & \quad +t^2\int_0^1(1-\lambda)\partial_{jk}\alpha^i(x,\sigma(x)+t\lambda\sigma'(x))\sigma'^j(x)\sigma'^k(x)d\lambda,
        \end{aligned}
    $$
    We can then estimate in $U$
    $$
    \begin{aligned}
        \bigg|\frac{\big(\alpha_*(\sigma+t\sigma')-\alpha_*(\sigma)\big)^i(x)}{t}-\big((d_v\alpha)_*(\sigma;\sigma')\big)^i(x) 
    \end{aligned}\bigg| \leq |t| C_{\alpha,\sigma,\sigma',U}
    $$
    for each $x\in U$, establishing uniform convergence. Moreover, since $d_v\alpha: VU\subset VE_1 \to E_2$ remains a smooth fibered morphism, the mapping $(d_v\alpha)_*:\mathcal{U}\times \Gamma_c^{\infty}(M\leftarrow E_1)\to \Gamma_c^{\infty}(M\leftarrow E_2)$ is continuous by Proposition \ref{prop_1_WO^k_properties}. Moreover, if $d^k_v\alpha : \otimes^k_MV_pU$ is the mapping locally defined by
    $$
        d^k_v\alpha[p]: \otimes^k V_pE_1\ni (s_1,\ldots,s_k) \mapsto
        \partial_{j_1 \ldots j_k}\alpha^i(p) s_1^{j_1}\cdots s_k^{j_K},
    $$
    a similar argument to the one above shows that for each $x\in U$ the mapping $d^{(k-1)}(\alpha_*)[\sigma+t\sigma'](\sigma_1,\ldots,\sigma_{k-1})(x)-d^{(k-1)}(\alpha_*)[\sigma](\sigma_1,\ldots,\sigma_{k-1})(x)$ converges uniformly to $(d^k_v\alpha )_*(\sigma;\sigma',\sigma_1,\ldots,\sigma_{k-1})(x)$. Then again, Proposition \ref{prop_1_continuity_of_push_forward} implies the continuity of $d^{(k)}(\alpha_*):\mathcal{U}\times \Gamma_c^{\infty}(M\leftarrow E_1) \cdots \times \Gamma_c^{\infty}(M\leftarrow E_1)\to \Gamma_c^{\infty}(M\leftarrow E_2)$. Finally, by $(iii)$ in Proposition \ref{prop_1_WO^k_properties}, this shows that the mapping $\alpha_*:\mathcal{U} \to \Gamma_c^{\infty}(M\leftarrow E_2)$ is Bastiani smooth.
\end{proof}

\section{Manifolds of mappings}\label{section_map_manifold}

We begin with the definition of infinite dimensional bundles. As we mentioned earlier we shall choose to model those on locally convex spaces in view of the results by \cite{ely}, \cite{ely2}, \cite{omori}.

\begin{definition}\label{def_1_infinite_dim_mfd}
    Let $\mathscr{M}$ be a Hausdorff topological space, we say that $\mathscr{M}$ admits a Bastiani smooth manifold structure if 
    \begin{itemize}
        \item[$(i)$] there is a family $\{(\mathcal{U}_i, u_i, E_i)\}_{i\in I}$ where $\{\mathcal{U}_i\}_{i\in I}$ is an open cover of $\mathscr{M}$, $\{E_i\}_{i\in I}$ is a family of complete locally convex spaces and $u_i:\mathcal{U}_i \to E_i$ a family of homeomorphisms onto the open subsets $u_i(\mathcal{U}_i)\subseteq E_i$;
        \item[$(ii)$] for all $i,j \in I$, having $\mathcal{U}_{ij}=\mathcal{U}_i\cap \mathcal{U}_j \neq \emptyset$, the mapping 
        $$
            u_{ij}=u_j\circ u^{-1}_i: u_i(\mathcal{U}_{ij})\subseteq E_i \to u_j(\mathcal{U}_{ij})\subseteq E_j
        $$
        and its inverse $u_{ji}=u_i\circ u^{-1}_j$ are Bastiani smooth. 
    \end{itemize}
    We then call \textit{charts} elements of the family $\{(\mathcal{U}_i, u_i, E_i)\}_{i\in I}$.
\end{definition}

It follows from condition $(ii)$ above that the locally convex spaces $E_i$ linearly isomorphic. A subset $\mathscr{N}\subset \mathscr{M}$ of a differentiable manifold is called a \textit{splitting submanifold} of $\mathscr{M}$ if for each $p\in \mathscr{N}$, there are charts $(\mathcal{U},u,E)$ of $\mathscr{M}$ such that $u(p)=0\in E$ and $U(\mathcal{U}\cap \mathscr{N})=u(\mathcal{U})\cap F$, where $F$ is a closed vector subspace of $E$ for which $E=F\oplus F^c$. The collection of charts $\{(\mathcal{U}_i\cap \mathscr{N}, u_i\vert_{\mathcal{U}_i\cap \mathscr{N}},F_i)\}$ then makes $\mathscr{N}$ a manifold itself as per Definition \ref{def_1_infinite_dim_mfd}. A weaker notion of submanifold requires that $F$ is just a closed subspace of $E$, in this case we say that $\mathscr{N}$ is a \textit{non-splitting submanifold}. 

Next we define the tangent bundle. Let $\mathscr{M}$ be a Bastiani smooth manifold with atlas $\lbrace(\mathcal{U}_{i},u_{i},E_{i})\rbrace$. A \textit{tangent vector} is an equivalence class of elements $(p,v,U_{i},u_{i},E_{i})$, with $p\in \mathscr{M}$ and $v \in E_{i}$, where $(p,v,\mathcal{U}_{i},u_{i},E_{i})$ and $(p',w,U_{j},u_{j},E_{j})$ are equivalent if $p=p'$ and $d(u_{ij}[u_{i}(p)])(v)=w $. We denote by $T_p\mathscr{M}$ the set of all tangent vectors to $p$, moreover setting $T\mathscr{M} = \bigsqcup_{p \in \mathscr{M}} T_p\mathscr{M} $ we obtain the \textit{space of tangent vectors} of $\mathscr{M}$. It is easy to see that $T\mathscr{M}$ carries a natural structure of Bastiani smooth manifold. To wit, observe that we can always define a canonical projection $\tau : T\mathscr{M} \to \mathscr{M} $. For the family of charts set $\{ (\widetilde{\mathcal{U}}_{i}, \widetilde{u}_{i}, E_{i} \times E_{i}) \} $ where $( {\mathcal{U}}_{i}, {u}_{i}, E_{i}) $ is a chart of $\mathscr{M}$, $\widetilde{\mathcal{U}}_{i}=\tau^{-1}\big( \mathcal{U}_i\big)$, $\widetilde{u}_{i}: \tau^{-1}(\mathcal{U}_i)\ni \widetilde{p}  \mapsto (u_{i}(x),v) \in E_i\times E_i$.  

The topology on $T\mathscr{M}$ is the unique one making each $\widetilde{u}_{i}$ into a homeomorphism, also the transition mapping $\widetilde{u}_{ij}:(x,v) \mapsto (u_{ij}(y), du_{ij}[x](v))$ is Bastiani smooth since $u_{ij}$ is itself smooth in the first place. It is easily shown that $T\mathscr{M}$ is Hausdorff, thus $T\mathscr{M}$ is a differentiable manifold according to Definition \ref{def_1_infinite_dim_mfd}.\\

Next we give a manifold structure to $C^{\infty}(M,N)$ with $M, \ N$ smooth finite dimensional manifolds. We first recall that given any Riemannian $h$ on $N$ there exists the Riemannian exponential $\exp_y: U\subset T_yN \to N, w\mapsto \exp_y(w)$, where $\exp_y(w)$ is the value of the geodesic starting at $y$ with velocity $w$ at time $t=1$. Since $\exp_y(0)=y$, and $T_y\exp_y=id_{T_yN}$, $\exp_y$ is a local diffeomorphism, then we can define a local diffeomorphism $(\tau_N,\exp):\widetilde{U}\subset TN \to \mathcal{O }\subset N \times N : (y,w) \to (y, \exp_{y}(w))$ onto an open subset $\mathcal{O }$ of the diagonal of $N\times N$. 

\begin{theorem}\label{thm_1_mfd_mappings}
    Let $M$, $N$ be smooth finite dimensional manifolds, then $C^{\infty}(M,N)$ is a Bastiani smooth manifold according to Definition \ref{def_1_infinite_dim_mfd}, modelled on the nuclear locally convex space $\Gamma^{\infty}_c(M\leftarrow f^*TN)$.
\end{theorem}

\begin{proof}
    Let $f\in C^{\infty}(M,N)$, then define $\mathcal{U}_f$ to be the subset of all $g \in C^{\infty}(M,N)$ with compact support with respect to $f$, such that $(f,g)(M)\subset (\tau_N,\exp)\big(\widetilde{U}\big)$, then $\mathcal{U}_f$ is an open subset for example by $(vii)$ in Proposition \ref{prop_1_WO^k_properties}. Let then
    $$
        u_f: \mathcal{U}_f \ni g \mapsto u_f(g)\in \Gamma^{\infty}_c(M\leftarrow f^*TN)
    $$
    defined as follows:
    \begin{equation}\label{eq_1_trivial_gamma_local_chart}
        u_f(g)(x)=(\tau_N,\exp)^{-1} (f(x),g(x))\simeq \big(f(x),\exp^{-1}_{f(x)}(g(x))\big) .
    \end{equation}
    It is clear that $u_f(g)$ is a smooth mapping, its image is valued in $f^*TN$, moreover since $f\neq g$ only in a compact subset of $K\subset M$, then $\exp^{-1}_{f(x)}(g(x))\neq 0 $ if and only if $x \in K$. Thus $u_f$ is valued into $\big\{\vec{X} \in \Gamma^{\infty}_c(M\leftarrow f^*TN) : \vec{X}(M) \subset (\tau_N,\exp)^{-1}\big(f^*\widetilde{U}\big)\big\}$. Therefore the mapping $u_f$ becomes a homeomorphism between $\mathcal{U}_f$ with the trace of the refined Whitney topology and an open subset of $ \Gamma^{\infty}_c(M\leftarrow f^*TN) $ with the usual limit Fréchet topology. If $\mathcal{U}_{fg}= \mathcal{U}_f\cap \mathcal{U}_g\neq \emptyset$, then we can consider the transition mapping $u_{fg}\doteq u_{g}\circ u_{f}^{-1}:\Gamma^{\infty}_c(M\leftarrow f^*TN) \to  \Gamma^{\infty}_c(M\leftarrow g^*TN)$. This mapping can be constructed as the push forward of
    $$
        T_{f(x)}N \ni ( x,w) \mapsto \big(x, \exp_{g(x)}^{-1}\big(\exp_{f(x)}(w)\big) \big)\in T_{g(x)}N,
    $$
    which is a smooth global fibered isomorphism $f^*\widetilde U\subset f^*TN\to g^*\widetilde U\subset g^*TN$. Then by Theorem \ref{thm_A_Bastiani smooth_pushforward} we also have that $u_{fg}$ is a smooth mapping together with its inverse $u_{gf}$. Finally if one chooses a different metric $h'$ on $N$ inducing the exponential $\exp'$ and new charts $u'_f$, then again, the transition mapping $u'_{f}\circ u_f$ can be obtained as the push forward of the local fibered isomorphism
    $$
        f^*\widetilde U\subset f^*TN \ni ( x,w) \mapsto \big(x, (\exp'_{f(x)})^{-1}\big(\exp_{f(x)}(w)\big) \big)\in f^*\widetilde U'\subset f^*TN.
    $$
    which by Theorem \ref{thm_A_Bastiani smooth_pushforward} is smooth. Therefore the smooth structure on $C^{\infty}(M,N)$ does not depend on the choice of the exponential mapping.
\end{proof}

We remark that in \cite{michor1980manifolds}, the role of the mapping $(\tau_N,\exp)$ is played by the so-called \textit{local addition}, that is a mapping $A:TN\to N\times N$ which is a local diffeomorphism onto an open subset of the diagonal for which $A(0_y)=y$ for all $y\in N$. One can show Lemma 10.1 and 10.2 pp. 90 in \cite{michor1980manifolds} that $(\tau_N,\exp)$ is a local addition, then the proof of Theorem \ref{thm_1_mfd_mappings} can be repeated along the same lines without altering the result.\\

This concludes the description of the manifold structure for the space of section of a trivial bundle $M\times N \to M$. In the more general case of a non-trivial bundle $\pi:B\to M$, $\Gamma^{\infty}(M\leftarrow B)$ can be topologized as follows: first we give $C^{\infty}(M,B)$ the refined Whitney topology, then we note that $\varphi \in \Gamma^{\infty}(M\leftarrow B)  \subset C^{\infty}(M,B)$ if and only if $\pi_{*}(\varphi)=\pi\circ \varphi = \mathrm{id}_M$. By (ii) in Proposition \ref{prop_1_continuity_of_push_forward}, $\pi_{*}$ is continuous, so the equation $\pi_{*}(\cdot) =\mathrm{id}_M $ in $C^{\infty}(M,B)$ defines a closed subset in the refined Whitney topology. We wish to show that $\Gamma^{\infty}(M\leftarrow B)$ is a splitting submanifold of $C^{\infty}(M,B)$. First, notice that if $\varphi\in \Gamma^{\infty}(M\leftarrow B)\subset C^{\infty}(M,B)$, then $\mathcal{U}_{\varphi}\cap \Gamma^{\infty}(M\leftarrow B)$ is the set of all $\psi \in \Gamma^{\infty}(M\leftarrow B)$ such that $\psi$, $\varphi$ differ only on a compact subset of $M$. Secondly, observe that $TB$ can be split by the choice of a connection as $HB \oplus VB$, this induces the splitting $\Gamma^{\infty}_c(M\leftarrow \varphi^*TB)=\Gamma^{\infty}_c(M\leftarrow  \varphi^*HB)\oplus \Gamma^{\infty}_c(M\leftarrow \varphi^*VB)$ at the level of locally convex spaces. Finally, if $\varphi\in \Gamma^{\infty}(M\leftarrow B)$ and $(\mathcal{U}_{\varphi},u_{\varphi})$ is a chart of $C^{\infty}(M,B)$, then 
$$
    u_{\varphi}\vert_{\mathcal{U}_{\varphi}\cap \Gamma^{\infty}(M\leftarrow B)}: \mathcal{U}_{\varphi}\cap \Gamma^{\infty}(M\leftarrow B)\to \Gamma^{\infty}_c(M\leftarrow \varphi^*TB)
$$
with
\begin{equation}\label{eq_1_gamma_local_chart}
    u_{\varphi}(\psi)(x)=\big(x, \exp_{\varphi(x)}^{-1}(\psi(x))  \big).
\end{equation}
Notice that even though the image of $x\in M$ through $\psi$, $\varphi$, lies in the same fiber $\pi^{-1}(x)$ we are not guaranteed that $\exp_{\varphi(x)}^{-1}(\psi(x))\in V_{\varphi(x)}B$, for $\pi^{-1}(x)$ might fail to be totally geodesic for the Riemannian metric chosen on $B$; therefore, in general, the geodesic joining $\varphi(x)$ and $\psi(x)$ might travel outside $\pi^{-1}(x)$. If we were able to solve this issue and show that $\exp_{\varphi(x)}^{-1}(\psi(x))\in V_{\varphi(x)}B$, we could then proceed by noticing that although the splitting depends on the connection chosen, the vertical subbundle $VB = \mathrm{ker}(\tau_B:TB\to B)$ does not, therefore $\Gamma^{\infty}_c(M\leftarrow \varphi^*TB)=\Gamma^{\infty}_c(M\leftarrow  \varphi^*HB)\oplus \Gamma^{\infty}_c(M\leftarrow \varphi^*VB)$ depends on the connection chosen just on the horizontal part. We are thus left with solving the issue of $\pi^{-1}(x)$ being not totally geodesic. By \cite[Lemma 10.9]{michor1980manifolds} to each section $\varphi$ we can find a tubular neighborhood, \textit{i.e.} a vector bundle $(E_{\varphi},\widetilde \pi, \varphi(M))$ with $E_{\varphi}\subset B$ and $\widetilde \pi = \pi|_{E_{\varphi}}$. Moreover, by \cite[Lemma 10.6]{michor1980manifolds}, we can modify (trough smooth diffeomorphisms) the local addition $(\tau_B,\exp_B)$ defining the chart $(\mathcal{U}_{\varphi},u_{\varphi})$ of $C^{\infty}(M,B)$ to a local addition $(\tau_{E_{\varphi}}, \widetilde \exp)$ on $E_{\varphi}$ for which 
$$
    \widetilde \exp_{\varphi}^{-1}(\psi) \in \Gamma^{\infty}_c(M\leftarrow \varphi^*TE_{\varphi}).    
$$
By construction, if $\psi(x)\in E_{\varphi}$, for all $x\in M$
\begin{equation}\label{eq_1_tilde_exp}
    \widetilde \exp_{\varphi(x)}^{-1}(\psi(x)) \in T_{\varphi(x)}\big( E_{\varphi}|_{\varphi(x)} \big) \simeq V_{\varphi(x)}E_{\varphi} \simeq V_{\varphi(x)}B.
\end{equation}
 In the sequel we will write the charts of $\Gamma^{\infty}(M\leftarrow B)$ as $(\mathcal{U}_{\varphi},u_{\varphi})$ understanding that their are the slice charts induced above. We shall call them \textit{ultralocal charts}\footnote{The term ultralocal has been introduced in \cite{forger2004currents} to signify that the mapping $u_{\varphi}$ does just depend on the point values of the mappings $\psi$, $\varphi$ without dependence on higher derivatives of the two.} in order to differentiate them from the local chart of finite dimensional manifold that were mentioned before.\\

\begin{lemma}\label{lemma_A_locality&bastiani_implies_w-reg}
    Let $P:\Gamma^{\infty}(M \leftarrow B)\to \Gamma^{\infty}(M \leftarrow C)$ be a differential operator and $\Gamma^{\infty}(M \leftarrow B),\ \Gamma^{\infty}(M \leftarrow C)$ be endowed with the infinite dimensional structure described in \Cref{thm_1_mfd_mappings}; then if $P$ is smooth it is weakly regular\footnote{See \Cref{section_Peetre_slovak} for the definition of weakly regular mapping.}.
\end{lemma}

\begin{proof}
    Suppose that $\varphi_s$ is a compactly supported variation of $\varphi_0$. We suppose also that $s\in I\subset \mathbb{R}$ with $I$ compact, but the general case is a straightforward generalization. We claim that $s\mapsto \varphi_s$ is a smooth curve in $\Gamma^{\infty}(M \leftarrow B)$, then again by Lemma \ref{lemma_1_interpolation_of_sections}, we can assume that the image of this path lies in a chart $\mathcal{U}_{\varphi}$, therefore our claim is equivalent to smoothness of 
    $$
        s \mapsto u_{\varphi_0}(\varphi_s)\in \Gamma^{\infty}_c(M \leftarrow \varphi^*VB),
    $$
    where $u_{\varphi}$ is the chart mapping defined in \eqref{eq_1_gamma_local_chart}. If $K\subset M$ is the compact where $\varphi_s\neq \varphi_0$, then $u_{\varphi}(\varphi_s)\neq u_{\varphi}(\varphi_0)$ in $K$ as well. Then it is enough to test differentiability at each order in the Fréchet space $\Gamma^{\infty}_K(M \leftarrow \varphi^*VB)$. 
    \begin{itemize}
        \item If we see $u_{\varphi}$ as a mapping from a neighborhood of the diagonal $\mathcal{O}\subset B\times B $ to the tangent space of $B$, then it is smooth, thus by (ii) in Proposition \ref{prop_1_continuity_of_push_forward}, we conclude that $u_{\varphi}(\varphi_s)=(u_{\varphi})_{*}\circ\varphi_s$ is continuous.
        \item To show differentiability, note that $\varphi(\pi(\varphi_s(x)))=\varphi(x)$ for all $x\in M$, therefore $u_{\varphi}(\varphi_s)=\widetilde\exp_{\varphi}^{-1}(\varphi_s)$, then for all $x\in M$
        $$
            \frac{d}{ds}u_{\varphi}(\varphi_s)(x)= T_{\varphi_s(x)}\widetilde\exp_{\varphi(x)}^{-1}(\dot\varphi_s);
        $$
        thus the derivative $ \frac{d}{ds}u_{\varphi}(\varphi_s)$ exists and by (ii) in Proposition \ref{prop_1_continuity_of_push_forward} is continuous due to smoothness of $T_{\bullet}\widetilde\exp_{\bullet}^{-1}: T\mathcal{O} \to TTB$.
        \item Iterating this argument we have shown smoothness of $u_{\varphi_s}$. 
    \end{itemize}
    Since $P$ is smooth, then $P(\varphi_s)$ is a smooth curve in $\Gamma^{\infty}(M \leftarrow C)$, eventually shrinking $I$ we can 
    assume that $P(\varphi_s)\subset \mathcal{V}_{P(\varphi_0)}$ \textit{i.e.} it lies inside a chart $(\mathcal{V}_{P(\varphi_0)}, v_{P(\varphi_0)})$ of $\Gamma^{\infty}(M \leftarrow C)$. Then $v_{P(\varphi_0)}\big(P(\varphi_s)\big)\subset \Gamma^{\infty}_c(M \leftarrow \sigma^*VC)$ is smooth. By $(iii)$, $P(\varphi_s)\in C^{\infty}(M,C)$ is a smooth curve and there is a compact subset $K'$ such that $P(\varphi_s)\equiv P(\varphi_0)$ outside $K'$. Therefore, 
    $$
        v_{P(\varphi_0)}\big(P(\varphi_s)\big)\subset \Gamma^{\infty}_{K'}(M \leftarrow \sigma^*VC)
    $$
    the latter is a Fréchet space, therefore we apply (ii) of Proposition \ref{prop_A_conv_sections_of_vector_bndl} and conclude.
\end{proof}

The tangent space at each point $\varphi$ is $T_{\varphi}\Gamma^{\infty}(M\leftarrow B)\equiv \Gamma^{\infty}_c(M\leftarrow \varphi^{*}VB)$. The tangent bundle $(T\Gamma^{\infty}(M\leftarrow B),\tau_{\Gamma}, \Gamma^{\infty}(M\leftarrow B))$ is defined in analogy with the finite dimensional case, and carries a canonical infinite dimensional bundle structure with trivializations
$$
	t_{\varphi}: \tau^{-1}_{\Gamma}(\mathcal{U}_{\varphi})\rightarrow \mathcal{U}_{\varphi} \times \Gamma^{\infty}_c(M\leftarrow \varphi^{*}VB).
$$
As usual, we can identify points of $T\Gamma^{\infty}(M\leftarrow B)$ by elements $t^{-1}_{\varphi}\left(\varphi,\vec{X}_{\varphi}\right)$. With those trivializations a tangent vector to $\Gamma^{\infty}(M\leftarrow B)$, \textit{i.e.} an element of $T_{\varphi}\Gamma^{\infty}(M\leftarrow B)$, can equivalently be seen as a section of the vector bundle $\Gamma^{\infty}_c(M\leftarrow \varphi^{*}VB)$. When using the latter interpretation, we will write the section in local coordinates as $\vec{X}(x)=\vec{X}^i(x)\partial_i\big\vert_{\varphi(x)}$. %To avoid confusion between the two we adopt different notations: we shall write the arrow symbol over capital Roman letters, \textit{e.g.} $(\vec{X}, \vec{Y}, \vec{Z}, \dots)$, to denote elements in $T_{\varphi}\Gamma^{\infty}(M\leftarrow B)$; we shall use instead use $(\mathcal{X}, \mathcal{Y}, \mathcal{Z}, \dots) $ to denote sections of $\Gamma^{\infty}_c(M\leftarrow \varphi^{*}B)$. 
Finally we will use Roman letters, \textit{e.g.} $(\vec{s},\vec{u}, \dots)$ to denote elements of the topological dual space $\Gamma^{-\infty}_c(M\leftarrow \varphi^{*}VB)\equiv\big(\Gamma^{\infty}_c(M\leftarrow \varphi^{*}VB)\big)'$.
\begin{definition}\label{def_1_connection}
A connection over the (possibly infinite dimensional) bundle $(C,\pi,X)$ is a vector-valued one form $\Phi \in \Omega^1(C;VC)$ satisfying
	\begin{itemize}
	\item[$(i)$] $\mathrm{Im}(\Phi) = VC$,
	\item[$(ii)$] $\Phi \circ \Phi=\Phi$.
	\end{itemize}
\end{definition}
The mapping $\Phi $ represents the projection onto the vertical subbundle of $TC$. Given a connection $\Phi$ it is always possible to associate its canonical Christoffel form $\Gamma\doteq \mathrm{id}_{TC}-\Phi$ which will define the projection onto the space of horizontal vector fields. In our case, we consider $C=T\Gamma^{\infty}(M\leftarrow B)$, the latter has canonical trivialization
$$
	Tt_{\varphi}: \tau^{-1}_{T\Gamma}\circ \tau^{-1}_{\Gamma}(\mathcal{U}_{\varphi})\rightarrow \mathcal{U}_{\varphi} \times \Gamma^{\infty}_c(M\leftarrow \varphi^{*V}B) \times \Gamma^{\infty}_c(M\leftarrow \varphi^{*}VB) \times \Gamma^{\infty}_c(M\leftarrow \varphi^{*}VB).
$$
Therefore given $ Tt_{\varphi}^{-1}\left(\vec{Y}_{\varphi}, \vec{S}_{\vec{X}}\right) \in TT\Gamma^{\infty}(M\leftarrow B)$, we can write the connection locally as
$$
	t_{\varphi}^{*}\Phi\left(\vec{Y}_{\varphi}, \vec{S}_{\vec{X}}\right)=\left(\vec{0}_{\varphi}, \vec{S}_{\vec{X}}-\Gamma_{\varphi}(\vec{X},\vec{Y}) \right),
$$
where the Christoffel form
$$
	\Gamma_{\varphi}\equiv t_{\varphi}^{-1}\Gamma: \Gamma^{\infty}_c(M\leftarrow \varphi^{*}VB) \times \Gamma^{\infty}_c(M\leftarrow \varphi^{*}VB) \rightarrow\Gamma^{\infty}_c(M\leftarrow \varphi^{*}VB) 
$$
can be chosen to be linear in the first two entries. For additional details about connections see \cite{kriegl1997convenient} Chapter VI, section 37. Instead of using the abstract notion provided by Definition \ref{def_1_connection}, in the case of manifolds of mappings, there is a more intuitive way of generating a connection. For simplicity's sake we shall do the easier case of $C^{\infty}(M,N)$, since the generalization to general bundles is almost immediate. Let $\widetilde{\Gamma}$ be a connection on the finite dimensional manifold $TN$, then we induce a connection $\Phi$ on $TC^{\infty}(M,N)$ as follows: fix $f\in C^{\infty}(M,N)$, $\vec{X},\vec{Y}\in T_fC^{\infty}(M,N)\simeq \Gamma^{\infty}_c(M\leftarrow f^{*}TN)$, $\vec{S}\in T_{\vec{X}}T_fC^{\infty}(M,N)\simeq\Gamma^{\infty}_c(M\leftarrow f^{*}TN)$, then
\begin{equation}\label{eq_1_def_connection}
\big(t_{f}^{-1}\Phi\big)\big(f,\vec{X},\vec{Y},\vec{S}\big)\doteq \big( \vec{0}_f, \vec{S}-\Gamma_f(\vec{X},\vec{Y})\big)   
\end{equation}
where $\Gamma_f(\vec{X},\vec{Y})\in \Gamma^{\infty}_c(M\leftarrow f^{*}TN)$ is defined by 
$$
    \Gamma_f(\vec{X},\vec{Y})(x)\doteq \widetilde{\Gamma}_{jk}^i(f(x))\vec{X}^j(x)\vec{Y}^k(x) \partial_i\vert_{f(x)}.
$$
Equivalently we are setting $\Gamma_f = \widetilde{\Gamma}_*$, by Theorem \ref{thm_A_Bastiani smooth_pushforward}, the mappings $\Gamma_f$, $\Phi_f$ are Bastiani smooth, moreover they induce a connection $\Phi$. In the sequel we shall use \eqref{eq_1_conn} to induce a connection as in \eqref{eq_1_def_connection}.
%------------------------------- EXAMPLE OF CONNECTION AS IN HAMILTON ------------------
\section{Observables}\label{section_observables}

By a functional, we mean a smooth mapping 
$$
	F: \mathcal{U}\subset \Gamma^{\infty}(M\leftarrow B) \to \mathbb{R},
$$
where $\mathcal{U}$ is an open set in the $CO$-topology generated by \eqref{eq_1_CO_open}. Since smoothness is tested on ultralocal charts, a functional $F$ is smooth if and only if, given any ultralocal atlas $\lbrace \mathcal{U}_{\varphi}, u_{\varphi} \rbrace_{\varphi \in \mathcal{U}}$ its localization
\begin{equation}\label{eq_1_localization_of_F}
	F_{\varphi} \doteq F \circ u_{\varphi}^{-1} : \Gamma^{\infty}_c(M\leftarrow \varphi^*VB) \rightarrow \mathbb{R},
\end{equation}
is smooth in the sense of Definition \ref{def_A_Bastiani_smooth_map}. \\

The first notion we introduce is the spacetime support of a functional. The idea is to follow the definition of support given in \cite{acftstructure}, and account for the lack of linear structure on the fibers of the configuration bundle $B$.

\begin{definition}\label{def_1_func_support}
Let $F$ be a functional over $\mathcal{U}$, $CO$-open, then its support is the closure in $M$ of the subset $x\in M$ such that for all $V \subset M$ open neighborhood of $x$, there is $ \varphi \in \mathcal{U}$, $\vec{X}_{\varphi} \in \Gamma^{\infty}_c(M\leftarrow \varphi^{*}VB)$ having $\mathrm{supp}(\vec{X}_{\varphi})\subset V$, for which $F_{\varphi}(\vec{X}_{\varphi})=F_{\varphi}(0)$. The set of functionals over $\mathcal{U}$ with compact spacetime support will be denoted by $\mathcal{F}_c(B,\mathcal{U})$ and its elements called observables.
\end{definition}

Let us display some examples of functionals. Given $\alpha \in C^{\infty}( B,\mathbb{R})$, consider
\begin{equation}\label{eq_1_generic_func}
	F_{\alpha}:\Gamma^{\infty}(M\leftarrow B) \rightarrow \mathbb{R}: \varphi \mapsto F_{\alpha}(\varphi) \doteq  \left		     \lbrace \begin{array}{lr}
      \frac{1}{1+\sup_{M}(\alpha(\varphi))} & \alpha(\varphi) \space\ \text{bounded },\\
       0 & \text{otherwise }.
             \end{array} \right. 
\end{equation}
If $f\in C^{\infty}_c(M)$ and $\lambda \in \Omega_m(J^rB)$ define
\begin{equation} \label{eq_1_gen_Lag_1}
	\mathcal{L}_{f,\lambda} :\Gamma^{\infty}(M\leftarrow B) \rightarrow \mathbb{R}:  \varphi \mapsto \mathcal{L}_{f,\lambda}(\varphi) \doteq  \int_M f(x) j^r\varphi^*\lambda(x) \mathrm{d}\mu_g(x).
\end{equation}
On the other hand if $f, \lambda$ are as above and $\chi:\mathbb{R}\rightarrow \mathbb{R}$ with $0\leq \chi \leq 1$, $\chi(t)=1 \space\ \forall  \vert t \vert \leq 1/2$ and $\chi(t)=0 \space\ \forall  \vert t \vert \geq 1/2$ define
\begin{equation}\label{eq_1_reg_func}
	G_{f,\lambda,\chi}:\Gamma^{\infty}(M\leftarrow B) \rightarrow \mathbb{C}:  \varphi \mapsto G_{f,\lambda,\chi}(\varphi) \doteq e^{1-\chi\big( (\mathcal{L}_{f,\lambda}(\varphi))^2\big)}.
\end{equation}
\vspace{0.1cm}

We can endow $\mathcal{F}_c(B,\mathcal{U})$ with the following operations
\begin{equation}
    (F,G) \mapsto (F+G)(\varphi)\doteq F(\varphi)+G(\varphi);
\end{equation}
\begin{equation}
    (z \in \mathbb{C}, F) \mapsto (zF)(\varphi) \doteq z F(\varphi);
\end{equation}
\begin{equation}\label{eq_1_classical_algebra_product}
    (F,G) \mapsto (F \cdot G)(\varphi) \doteq F(\varphi)G(\varphi);
\end{equation}
\begin{equation}
    F \mapsto F^{*}, \ F^{*}(\varphi) \doteq \overline{F(\varphi)}\footnote{In adherence to standard clFT, we use real functionals, which makes involution a trivial operation; we remark though that one could repeat \textit{mutatis mutandis} everything with $\mathbb{R}$ replaced by $\mathbb{C}$, then involution is not trivial anymore.}.
\end{equation}

It can be shown that those operation preserve the compactness of the support, turning $\mathcal{F}_c(B,\mathcal{U})$ into a commutative \textit{*-algebra with unity} where the unit element is given by $\varphi \mapsto 1\in \mathbb{R}$. That involution and scalar multiplication are support preserving is trivial, to see that for multiplication and sum we use

\begin{lemma}
Let $F$, $G$ be functionals over $\mathcal{U} \subset \Gamma^{\infty}(M\leftarrow B)$ $CO$-open subset, then
\begin{itemize}
\item[$(i)$] $\mathrm{supp}(F+G)\subset \mathrm{supp}(F) \cup \mathrm{supp}(G)$,
\item[$(ii)$] $\mathrm{supp}(F\cdot G)\subset \mathrm{supp}(F) \cup \mathrm{supp}(G)$.
\end{itemize}
\end{lemma}

Before writing the proof we note that the more restrictive version of $(ii)$ with the intersection of domains does not hold in general, this can be checked by taking a constant functional $G(\varphi) \equiv c \space\ \forall \varphi \in \mathcal{U}$, then $\mathrm{supp}(G)=\emptyset$ while $\mathrm{supp}(F+G)$, $\mathrm{supp}(F\cdot G)=\mathrm{supp}(F)$.

\begin{proof}
Suppose that $x \notin \mathrm{supp}(F) \cup \mathrm{supp}(G) $, then there is an open neighborhood $V$ of $x$ such that for any $X \in \Gamma^{\infty}_c(M\leftarrow VB)$ with $\mathrm{supp}(X)\subset V$, and any $\varphi \in \mathcal{U}$ we have $(F+G)_{\varphi}(\vec{X}_{\varphi})=F_{\varphi}(\vec{X}_{\varphi})+G_{\varphi}(\vec{X}_{\varphi})=F_{\varphi}(0)+G_{\varphi}(0)$, so $x \notin \mathrm{supp}(F+G)$. The other follows analogously.
\end{proof}

Using the notion of Bastiani differentiability we can induce a related differentiability for functionals over $\Gamma^{\infty}(M\leftarrow B)$, in the same spirit as done for mappings between manifolds.

\begin{definition}\label{def_1_func_differentiability}
Let $\mathcal{U}$ be $CO$-open, a functional $F\in \mathcal{F}_c(B,\mathcal{U})$ is differentiable of order $k$ at $\varphi \in \mathcal{U}$ if for all $0 \leq j \leq k$ the functionals $d^jF_{\varphi}[0]:  \otimes^j\left(\Gamma^{\infty}_c(M\leftarrow \varphi^*VB)\right) \rightarrow \mathbb{R}: (\vec{X}_1,\ldots, \vec{X}_j) \mapsto d^jF_{\varphi}[0](\vec{X}_1,\ldots, \vec{X}_j)$ are linear and continuous with 
$$
\begin{aligned}
	d^jF_{\varphi}[u_{\varphi}(\varphi)](\vec{X}_1,\ldots, \vec{X}_j)\doteq &\left.\frac{d^j}{dt_1 \ldots dt_j}\right|_{t_1= \ldots=t_j=0}F_{\varphi}(t_1\vec{X}_1+ \cdots+ t_j \vec{X}_j) \\ &= \left\langle F^{(j)}_{\varphi}[0],\vec{X}_1 \otimes \cdots \otimes \vec{X}_j\right\rangle.
\end{aligned}
$$
If $F$ is differentiable of order $k$ at each $\varphi \in \mathcal{U}$ we say that $F$ is differentiable of order $k$ in $\mathcal{U}$. Whenever $F$ is differentiable of order $k$ in $\mathcal{U}$ for all $k \in \mathbb{N}$ we say that $F$ is smooth and denote the set of smooth functionals as $\mathcal{F}_0(B,\mathcal{U})$.
\end{definition}

We shall use the notation $F^{(j)}_{\varphi}[0]$ to emphasize the fact that since $F$ is Bastiani smooth, then $d^jF[\varphi]\equiv d^jF_{\varphi}[0] \in \Gamma^{-\infty}(M^j\leftarrow \boxtimes^j \varphi^*VB )$, thus by Schwartz theorem, we can represent the latter as an integral kernel $F^{(j)}$ with
$$
    \left\langle F^{(j)}_{\varphi}[0],\vec{X}_1 \otimes \ldots \otimes \vec{X}_j\right\rangle = \int_{M^j} f^{(j)}_{\varphi}[0]_{i_1 \cdots i_j}(x_1,\ldots,x_j) \vec X_1^{i_1}(x_1) \cdots \vec X_j^{i_j}(x_j) d\mu_g(x_1,\ldots,x_j) ,
$$
where $\vec X_p=\vec X_p^{i_p}\frac{\partial}{\partial y^{i_p}}\Big|_{y=\varphi(x)}$.   \\

When $F$ is smooth the condition of Definition \ref{def_1_func_differentiability} is independent from the chart we use to evaluate the B differential: suppose we take charts $(\mathcal{U}_{\varphi}, u_{\varphi})$, $(\mathcal{U}_{\psi}, u_{\psi})$ with $\varphi \in \mathcal{U}_{\psi}$, then by Faà di Bruno's formula
\begin{align}\label{eq_1_gamma_loc_change}
\begin{aligned}
	&d^jF_{\psi}[u_{\psi}(\varphi)](\vec{X}_1,\ldots, \vec{X}_j) 
	\\
	& \quad = \sum_{\pi \in \mathscr{P}(\{1,\ldots,j\})} F_{\varphi}^{|\pi|}[0]\left( d^{\vert I_1 \vert}u_{\varphi \psi}[u_{\psi}(\varphi)]\Big(\bigotimes_{i\in I_1}\vec{X}_{i}\Big), \ldots,  d^{\vert I_{|\pi|} \vert}u_{\varphi \psi}[u_{\psi}(\varphi)]\Big(\bigotimes_{i'\in I_{|\pi|}}\vec{X}_{i'}\Big) \right),
\end{aligned}
\end{align}
where $\pi$ is a partition of $\{1,\ldots,j\}$ into $|\pi|$ smaller subsets $I_1,\ldots,I_{|\pi|}$ and we denote by $u_{\varphi \psi}$ the transition function $u_{\varphi}\circ u_{\psi}^{-1}$. We immediately see that the right hand side is Bastiani smooth by the smoothness of the transition function, therefore the left hand side ought to be Bastiani smooth as well. Incidentally the same kind of reasoning shows Definition \ref{def_1_func_differentiability} is independent from the ultralocal atlas used for practical calculations. 

Although this is enough to ensure Bastiani differentiability, in the sequel we shall introduce a connection on the bundle $T\Gamma^{\infty}(M\leftarrow B) \rightarrow \Gamma^{\infty}(M\leftarrow B) $ 
so that \eqref{eq_1_gamma_loc_change} can be written as an equivalence between two single terms involving the covariant derivatives. In particular, as explained in \eqref{eq_1_def_connection}, we will choose a smooth connection %that is a $T\Gamma^{\infty}(M\leftarrow B)$-valued differential one form: $\Phi\in \Omega^1\big(\Gamma^{\infty}(M\leftarrow B);T\Gamma^{\infty}(M\leftarrow B)\big)$. Given a point $(\varphi,{X}) \in T\Gamma^{\infty}(M\leftarrow B)$ and an element in the fiber of $(\varphi,\vec{X})$, say  $T\tilde{u}_{\varphi}^{-1}(\vec{Y}_{\varphi}, \vec{S}_{\vec{X}}) \in TT\Gamma^{\infty}(M\leftarrow B) $, the action of the connection is defined by
% \begin{equation}\label{eq_1_conn}
% 	T\tilde{u}_{\varphi} \circ \Phi \circ T\tilde{u}_{\varphi}^{-1} (\vec{Y}_{\varphi}, \vec{S}_{\vec{X}}) \doteq \left( \vec{0}_{\varphi} , \vec{S}_{\vec{X}} - \Gamma_{\varphi}(\vec{X},\vec{Y})\right) ,
% \end{equation}
% where $\tilde{u}_{\varphi} \equiv Tu_{\varphi}: \left. T\Gamma^{\infty}(M\leftarrow  B) \right\vert_{\tilde{\mathcal{U}}_{\varphi}} \rightarrow \mathcal{U}_{\varphi} \times \Gamma^{\infty}_c(M\leftarrow \varphi^{*}VB)$ are the charts of the tangent bundle $T\Gamma^{\infty}(M\leftarrow B)$. One can always define this infinite dimensional connection as follows: 
$\widetilde{\Gamma}$ on the typical fiber $F$ of the bundle $B$, the latter will induce a linear connection $\varphi^*\widetilde{\Gamma}$ on the vector bundle $M\leftarrow \varphi^*VB$, and, in turn, a connection on $T\Gamma^{\infty}(M\leftarrow B)$
\begin{equation}\label{eq_1_conn}
\begin{aligned}
    \Gamma_{\varphi} :&  \Gamma^{\infty}_c(M\leftarrow \varphi^{*}VB)\times  \Gamma^{\infty}_c(M\leftarrow \varphi^{*}VB)
    \to  \Gamma^{\infty}_c(M\leftarrow \varphi^{*}VB)\\
    & (\vec{X},\vec{Y})\mapsto  \Gamma_{\varphi}(\vec{X}_{\varphi},\vec{Y}_{\varphi}),\\
    &\Gamma_{\varphi}(\vec{X},\vec{Y})(x)= \Gamma(\varphi(x))^i_{jk}\vec{X}^j(\varphi(x))\vec{Y}^k(\varphi(x))\partial_i\big|_{\varphi(x)};
\end{aligned}
\end{equation}
where $\vec{X}^j\partial_j\big|_{\varphi}$, $\vec{Y}_{\varphi}^k\partial_k\big|_{\varphi}$ are the expressions in local coordinates of $\vec{X}$, $\vec{Y}\in \Gamma^{\infty}_c(M\leftarrow \varphi^{*}VB)$. Armed with \eqref{eq_1_conn} we can define the notion of \textit{covariant differential} recursively setting
\begin{align}\label{eq_1_cov_der}
\begin{aligned}
	\nabla^{1}F_{\varphi}[0](\vec{X})  &\doteq dF_{\varphi}(\vec{X}) ,\\
	\nabla^{n}F_{\varphi}[0](\vec{X}_1,\ldots, \vec{X}_n)  &\doteq  F^{(n)}_{\varphi}(\vec{X}_1,\ldots, \vec{X}_n)  \\
	&\quad + \sum_{j=1}^n \frac{1}{n!} \sum_{\sigma \in \mathcal{P}(n)} \nabla^{n-1}F_{\varphi} (\Gamma_{\varphi} (\vec{X}_{\sigma(j)}, \vec{X}_{\sigma(n)}), \vec{X}_{\sigma(1)}, \ldots, \widehat{\vec{X}_{\sigma(j)}}, \ldots \vec{X}_{\sigma(n-1)}),
\end{aligned}
\end{align}
where $\mathcal{P}(n)$ denotes the set of permutations of $n$ elements. In this way we can extend properties of iterated derivatives, which are locally defined, globally. The price we pay is that, a priori, the property might depend on the connection chosen.

\begin{lemma}\label{lemma_1_support_property}
Let $\mathcal{U}$ be a locally convex, $CO$-open subset, and $F:\mathcal{U} \rightarrow \mathbb{R}$ a differentiable functional of order one, then
$$
	\mathrm{supp}(F)= \overline{\bigcup_{\varphi \in \mathcal{U}} \mathrm{supp}\left(F^{(1)}_{\varphi}[0]\right)},
$$
where $\mathrm{supp}\left(F^{(1)}_{\varphi}[0]\right)$ is to be understood as the distributional support defined in Definition \ref{def_A_distr._section_support}.
\end{lemma}
\begin{proof}
Suppose that $x \in \mathrm{supp}(F)$, then by definition for all open neighborhoods $V$ of $x$ there is $\varphi \in \mathcal{U}$ and $\vec{X}_{\varphi} \in \Gamma^{\infty}_c(M\leftarrow \varphi^{*}VB)$ with $\mathrm{supp}(\vec{X}) \subset V$ having $F_{\varphi}(\vec{X}_{\varphi}) \neq F_{\varphi}(0)$, using the convexity of $\mathcal{U}$ and the fundamental theorem of calculus we obtain that
$$
	F_{\varphi}( \vec{X}_{\varphi}) - F_{\varphi}(0)= \int_0^1 F_{\varphi}^{(1)}[\lambda  \vec{X}_{\varphi}]( \vec{X}_{\varphi}) \mathrm{d}\lambda \neq 0.
$$
Thus for at least for some $\lambda_0 \in (0,1)$, the integrand is not zero, setting $\psi=u_{\varphi}^{-1}(\lambda_0  \vec{X}_{\varphi} )$, we obtain 
$$
	dF_{\psi}[0]\left(d^1u_{\varphi \psi}[\lambda_0 \vec{X}_{\varphi}]( \vec{X}_{\varphi})\right)\neq 0.
$$
On the other hand if $x \in \mathrm{supp}\left( dF_{\varphi}[0] \right)$ for some $\varphi \in \mathcal{U}$, then there is $\vec{X}_{\varphi} \in \Gamma^{\infty}_c(M\leftarrow \varphi^*VB)$ having $\vec{X}_{\varphi}(x)\neq \vec{0}$ for which $dF_{\varphi}[0](\vec{X}_{\varphi}) \neq 0$, as a result, define
$$
	 F_{\varphi}(\epsilon \vec{X}_{\varphi}) = F_{\varphi}(0)+ \int_0^{\epsilon} F_{\varphi}^{(1)}[\lambda \vec{X}_{\varphi}]( \vec{X}_{\varphi}) \mathrm{d}\lambda
$$
having chosen $\epsilon$ small enough so that the integral is not vanishing.
\end{proof}

%Among smooth functionals we select certain classes:

\begin{definition}\label{def_1_func_classes}
Let $\mathcal{U}$ be $CO$-open. We select certain classes of $\mathcal{F}_c(B,\mathcal{U})$.
\begin{itemize}
    \item[$(i)$] \textbf{Regular Functionals}: the set of $F \in \mathcal{F}_c(B,\mathcal{U})$ such that for each $\varphi\in \mathcal{U}$, the integral kernel associated to $\nabla^{k}F_{\varphi}[0]$, 
    $$
	    \nabla^{k}F_{\varphi}[0](\vec{X}_1,\ldots,\vec{X}_k) = \int_{M^k} \nabla^{k} f_{\varphi}[0](x_1,\ldots,x_k)\vec{X}_1(x_1)\cdots\vec{X}_k(x_k) d\mu_{g}(x_1,\ldots,x_k)
    $$ 
    has $\nabla^{k} f_{\varphi}[0]\in \Gamma^{\infty}_c\big(M^k\leftarrow \boxtimes^k \big(\varphi^{*}VB' \big)\big)$, we denote this set by $\mathcal{F}_{\mathrm{reg}}(B,\mathcal{U})$;

    \item[$(ii)$] \textbf{Local Functionals}: the set of $F \in \mathcal{F}_c(B,\mathcal{U})$ such that for each $\varphi\in \mathcal{U}$, $\mathrm{supp}\big(\nabla^{(2)}F_{\varphi}[0]\big) \subset \triangle_2(M)$, the latter being the diagonal of $M\times M$, we denote this set by $\mathcal{F}_{loc}(B,\mathcal{U})$;
    
    \item[$(iii)$] \textbf{Microlocal Functionals}: the set of $F \in \mathcal{F}_{loc}(B,\mathcal{U})$ such that for each $\varphi\in \mathcal{U}$, the integral kernel associated to $\nabla^{1}F_{\varphi}[0]\equiv dF_{\varphi}[0]$ has $f^{(1)}_{\varphi}[0] \in \Gamma^{\infty}\big(M\leftarrow (\varphi^{*}VB)' \big)$, we denote this set by $\mathcal{F}_{\mu loc}(B,\mathcal{U})$.
\end{itemize}
\end{definition}

Using the Schwartz kernel theorem, we can equivalently define microlocal functionals by requiring $\{ F \in \mathcal{F}_{loc}(B,\mathcal{U}) \space\ \mathrm{:} \  \mathrm{WF}\big(F^{(1)}_{\varphi}[0]\big)=\emptyset \ \forall \varphi \in \mathcal{U}\}$. Other authors add also further requirements, for example in \cite{BDGR14}, microlocal functionals have the additional property that given any $\varphi\in \mathcal{U}$ there exists an open neighborhood $\mathcal{V}\ni \varphi$ in which $f^{(1)}_{\varphi'}[0]\in \Gamma^{\infty}\left(M\leftarrow (\varphi^{*}VB)' \right)$ depends on the $k$th order jet of $\varphi'$ for all $\varphi'\in \mathcal{V}$ and some $k\in \mathbb{N}$. We choose to give a somewhat more general description which however will turn out to be almost equivalent by Proposition \ref{porop_1_muloc_charachterization}. Finally we stress that the definition of local functionals together with Lemma \ref{lemma_1_support_property} shows that that $\mathrm{supp}\big(\nabla^{k}F_{\varphi}[0]\big) \subset \triangle_k(M)$ for each $k\in \mathbb{N}$.\\

As remarked earlier, writing differentials with a connection does yield a  definition which is independent from the ultralocal chart chosen to perform the calculations, however, we have to check that Definition \ref{def_1_func_classes} is independent from the chosen connection. 

\begin{lemma}
    Suppose that $\Phi$, $\widehat{\Phi}$ are two connections on $T\Gamma^{\infty}(M\leftarrow B)$. Then the definition of regular (resp. local, microlocal) functionals does not depend on the chosen connection. 
\end{lemma}
\begin{proof}
Denote by $\nabla$, $\widehat{\nabla}$ the covariant derivatives induced by $\Phi$, $\widehat{\Phi}$ respectively. If $F \in \mathcal{F}_c(B,\mathcal{U})$ is local with respect to the second connection,
$$
	\left( \nabla^{2} F_{\varphi} - \widehat{\nabla}^{2}F_{\varphi} \right) [0] (\vec{X}_1,\vec{X}_2) = dF_{\varphi}[0] \left( \Gamma_{\varphi}(\vec{X}_1,\vec{X}_2) - \widehat{\Gamma}_{\varphi}(\vec{X}_1,\vec{X}_2) \right).
$$
Due to linearity of the connection in both arguments, when the two sections $\vec{X}_1$, $\vec{X}_2$ have disjoint support the resulting vector field is identically zero, so that by linearity of $dF_{\varphi}[0](\cdot)$ the expression is zero and locality is preserved. As a result, since $\nabla^{1}F_{\varphi}[0]\equiv dF_{\varphi}[0]$, we immediately obtain that microlocality is independent as well. Regular functionals do not depend on the connection used to perform calculations either: this is easily seen by induction. If $k=1$ this is trivial since $\nabla^{1}F \equiv dF$, for arbitrary $k$ one simply notes that $\left( \nabla^{k} F_{\varphi} - \widehat{\nabla}^{k}F_{\varphi} \right) [0] (\ldots)$ depends on terms of order $l \leq k-1$ and applies the induction hypothesis. 
\end{proof}

We stress that in particular cases, such as when $B= M \times \mathbb{R}$, $TC^{\infty}(M) \equiv C^{\infty}(M) \times C^{\infty}_c(M)$, we are allowed to choose a trivial connection, in which case the differential and the covariant derivative coincide. It is also possible to formulate Definition \ref{def_1_func_classes} in terms of differentials instead of covariant derivatives, then the above argument can be used again to show that regular and local functionals do not depend on the choice of the chart.\\

When dealing with microlocal functionals we will often use the following notation ensuing from application of Schwartz integral kernel theorem:
\begin{equation}\label{eq_1_kernel_notation}
	dF_{\varphi}[0](\vec{X}_{\varphi})= \int_M f^{(1)}_{\varphi}[0](\vec{X}_{\varphi})(x) d\mu_g(x)= \int_M f^{(1)}_{\varphi}[0]_i(x)X_{\varphi}^i(x) d\mu_g(x),
\end{equation}
where repeated indices denotes summation of vector components as usual with Einstein notation and $X_{\varphi}^i(x) \in \varphi^*p^{-1}(x)$ denotes the component of the section along the typical fiber of the vector bundle $(\Gamma^{\infty}_c(M\leftarrow \varphi^*VB),\varphi^*p,M)$.\\

If we go back to the examples of functionals given earlier we find that \eqref{eq_1_generic_func} does not belong to any class, while \eqref{eq_1_reg_func} is a regular functional that however fails to be local. If $D\subset M$ is a compact subset and $\chi_D$ its characteristic function then 
$$
    \varphi \mapsto \mathcal{L}_{\chi_D,\lambda}(\varphi) \doteq  \int_M \chi_D(x)\lambda(j^r\varphi)(x)\mathrm{d}\mu_g(x).
$$
is a local functional which however, is not microlocal due to the possible singularities localized in the boundary of $D$. Finally we claim that \eqref{eq_1_gen_Lag_1} is a microlocal functional. To see it, let us consider a particular example where $r=1$,
$$
	\mathcal{L}_{f,\lambda}(\varphi)=\int_M f j^1 \varphi^{*}\lambda= \int_M f(x) \lambda(j^1\varphi)(x) d\mu_g(x)
$$
taking the first derivative and integrating by parts yields
\begin{equation}\label{eq_1_derivative_lag_wavemaps}
    d\mathcal{L}_{f,\lambda,\varphi}[0](\vec{X}_{\varphi})=\int_{M}f(x)\bigg\{\frac{\partial\lambda}{\partial y^i}-d_{\mu}\bigg( \frac{\partial \lambda}{\partial y^i_{\mu}} \bigg)   \bigg\}(x) X_{\varphi}^i(x) \mathrm{d}\mu_g(x),
\end{equation}
setting 
\begin{equation}
	\lambda_{f,\varphi}^{(1)}[0](x) \doteq f(x)\left\lbrace  \frac{\partial \lambda}{\partial y^i}-d_{\mu}\left( \frac{\partial \lambda}{\partial y^i_{\mu}} \right)   \right\rbrace(x) dy^i \wedge \mathrm{d}\mu_g(x),
\end{equation}
we see that the integral kernel of the first derivative in $\varphi$ of \eqref{eq_1_gen_Lag_1}, $\lambda_{f,\varphi}^{(1)}[0]$, belongs to $\Gamma^{\infty}\left(M\leftarrow (\varphi^{*}VB)' \right) $. For generic orders $r\neq 1$, multiple integration by parts will yield the desired result, for details on those calculations see \cite{fati} Chapter 6. This last example is important because it shows that functionals obtained by integration of pull-backs of $m$-forms $\lambda$ are microlocal. One could ask whether the converse can hold, \textit{i.e.} if all microlocal functionals have this form; the answer will be given in Proposition \ref{porop_1_muloc_charachterization}. We now give an equivalent characterization for local functionals.

\begin{definition} \label{def_1_additivity}
Let $\mathcal{U}$ be $CO$-open, a functional $F\in \mathcal{F}_c(B,\mathcal{U})$ is called:
\begin{itemize}
\item[$(i)$]$\varphi_0$-additive if for all $\varphi_j \in \mathcal{U}_{\varphi_0}\cap\, \mathcal{U}$ having $\mathrm{supp}_{\varphi_0}(\varphi_1) \cap \mathrm{supp}_{\varphi_0}(\varphi_{-1}) = \emptyset$, setting $\vec{X}_j=u_{\varphi_0}(\varphi_j)$, $j=1,-1$ and supposing that $\vec{X}_1+\vec{X}_{-1}\in u_{\varphi_0}(\mathcal{U}_{\varphi_0}\cap\mathcal{U})$, we have
\begin{equation}\label{eq_1_loc_additivity}
	F_{\varphi_0}(\vec{X}_1 {+} \vec{X}_{-1})= F_{\varphi_0}(\vec{X}_1) - F_{\varphi_0}(0) + F_{\varphi_0}(\vec{X}_{-1}).
\end{equation}
\item[$(ii)$] additive if for all $\varphi_j\in \mathcal{U}$, $j=1,0,-1$, with $\mathrm{supp}_{\varphi_0}(\varphi_1)\cap \mathrm{supp}_{\varphi_0}(\varphi_{-1})=\emptyset$, setting 
$$ \varphi=
\begin{cases}
	\varphi_1 \ \ \ \mathrm{in} \space\ \mathrm{supp}_{\varphi_0}(\varphi_{-1})^c\\
\varphi_{-1} \hspace{0.17cm}  \mathrm{in} \space\ \mathrm{supp}_{\varphi_0}(\varphi_1)^c
\end{cases}
$$
we have 
\begin{equation}\label{eq_1_glob_additivity}
	F(\varphi)=F(\varphi_1)+F(\varphi_0)-F(\varphi_{-1}).
\end{equation}
\end{itemize}
\end{definition}
We remark that $(ii)$ is equivalent to the definition of additivity present in \cite{gravbrunetti}. Before the proof of the equivalence of those two relations, we prove a technical lemma.

\begin{lemma} \label{lemma_1_interpolation_of_sections}
Let $\varphi_1$, $\varphi_0$, $\varphi_{-1} \in \Gamma^{\infty}(M\leftarrow B)$ have $\mathrm{supp}_{\varphi_0}(\varphi_1)\cap \mathrm{supp}_{\varphi_0}(\varphi_{-1})=\emptyset$, then there exist $n\in \mathbb{N}$, a finite family of sections  
\begin{equation}\label{eq_1_interpolation_0}
	\left\{\varphi_{(k,l)}\right\}_{k,l \in \{1,\ldots,n\}}
\end{equation}
for which the following conditions holds:
\begin{itemize}
\item[$(a)$] For each $k$, $l\in \mathbb{N}$
\begin{equation}\label{eq_1_intepolation_2}
	\left\{\varphi_{(k-1,l-1)},\varphi_{(k,l-1)},\varphi_{(k-1,l)},\varphi_{(k+1,l)},\varphi_{(k,l+1)},\varphi_{(k+1,l+1)} \right\}\in \mathcal{U}_{\varphi_{(k,l)}},
\end{equation}
\item[$(b)$]  Moreover for each $k$, $l\in \mathbb{N}$ we can define elements $\vec{X}_{k}$, $\vec{Y}_{l} \in \Gamma^{\infty}_c(M\leftarrow \cdot^{*}VB)$, where 
\begin{equation}\label{eq_+_vector}
	\vec{X}_{k} \doteq \widetilde\exp^{-1}_{\varphi_{(k-1,l)}}\left( \varphi_{(k,l)}\right),
\end{equation}
\begin{equation}\label{eq_-_vector}
	\vec{Y}_{l} \doteq \widetilde\exp^{-1}_{\varphi_{(k,l-1)}}\left( \varphi_{(k,l)} \right);
\end{equation}
whose exponential flows generate all the above sections:
% \begin{equation}\label{intepolation1}
% 	\varphi_{(k,l, m)} = \exp\big(\vec{X}_{(1,l,0)}\big)\circ \exp\big(\vec{Y}_{(0,l,1)}\big)\circ \cdots \circ \exp\big(\vec{X}_{(1,0,0)}\big)\circ \exp\big(\vec{Y}_{(0,0,1)}\big)\circ  \varphi_0,
% \end{equation}
$$
	\varphi_1=\widetilde\exp\big(\vec{X}_{n}\big)\circ\cdots \circ \widetilde\exp\big(\vec{X}_{1}\big) \circ \varphi_0\equiv \varphi_{(n,0)};
$$
$$
	\varphi_{-1}=\widetilde\exp\big(\vec{Y}_{n}\big)\circ\cdots \circ \widetilde\exp\big(\vec{Y}_{1}\big)\circ \varphi_0\equiv \varphi_{(0,n)};
$$
and
$$
    \varphi= \widetilde\exp\big(\vec{X}_{n}\big)\circ \cdots \circ \widetilde\exp\big(\vec{X}_{1}\big) \circ \widetilde\exp\big(\vec{Y}_{n}\big)\circ \cdots \circ \widetilde\exp\big(\vec{Y}_{1}\big)\circ  \varphi_0.
$$
\end{itemize}
\end{lemma}
\begin{proof}
Ideally we are taking $\varphi_0$ as a background section, then application of a number of exponential flows of the above fields will generate new sections interpolating between $\varphi_0$ and $\varphi,\varphi_1,\varphi_{-1}$, such that each section in the interpolation procedure has the adjacent sections in the same chart (as in \eqref{eq_1_intepolation_2}). This is, for a pair of generic sections, not trivial; however, due to the requirement of mutual compact support between sections, our case is special. Indeed, let $K$ be any compact containing $\mathrm{supp}_{\varphi_0}(\varphi_{1})\cup \mathrm{supp}_{\varphi_0}(\varphi_{-1}) $. Since $B$ is itself a paracompact manifold, it admits an exhaustion by compact subsets and a Riemannian metric compatible with the fibered structure. The exponential mapping of this metric will have a positive injective radius throughout any compact subset of $B$. Thus let $H$ be any compact subset of B containing the bounded subset 
$$
    \Big\{b\in \pi^{-1}(K)\subset B : \sup_{x\in K}d(\varphi_0(x),b) < 2\max\big(\sup_{x\in K} d(\varphi_0(x),\varphi_1(x)), \sup_{x\in K} d(\varphi_0(x),\varphi_{-1}(x))\big)\Big\},
$$
where $d$ is the distance induced by the metric chosen. Let $\delta>0$ be the injective radius of the metric on the compact $H$. If $r=\max\big(\sup_{x\in K} d(\varphi_0(x),\varphi_1(x)), \sup_{x\in K} d(\varphi_0(x),\varphi_{-1}(x))\big)$ there will be some finite $n \in \mathbb{N}$ such that $n\delta < r <(n+1)\delta$, and thus we can select a finite family of sections $\left\{\varphi_{(k,l)}\right\}_{k,l =1,\ldots,n}$ interpolating between $\varphi_0=\varphi_{(0,0)}$ and $\varphi_1=\varphi_{(n,0)}$, $\varphi_{-1}=\varphi_{(0,n)}$, $\varphi=\varphi_{(n,n)}$ such that 
$$
	(|k-k'|-1)\frac{\delta}{2}+(|l-l'|-1)\frac{\delta}{2}<\sup_{x\in K}d\left(\varphi_{(k,l)}(x),\varphi_{(k',l')}(x)\right)< (|k-k'|)\frac{\delta}{2}+(|l-l'|)\frac{\delta}{2},
$$
% $$
% 	0<\sup_{x\in K}d\left(\varphi_{(k,l)}(x),\varphi_{(k,l\pm 1)}(x)\right)< \frac{\delta}{2},
% $$
% $$
% 	k\frac{\delta}{2}<\sup_{x\in K_1}d_h\left(\varphi_{(1,0,-1)}^{(k,0,0)}(x),\varphi_0(x)\right)
% $$
% and
% $$
% 	\sup_{x\in K_1}d_h\left(\varphi_{(1,0,-1)}^{(k,0,0)}(x),\varphi_1(x)\right)< r-k\frac{\delta}{2}.
% $$
This property ensures that we are interpolating in the right direction, that is, as $k$ (resp. $l$) grows new sections are nearer to $\varphi_1$ (resp. $\varphi_{-1})$ and further away from $\varphi_0$. Eventually modifying $\exp$ to $\widetilde\exp$ as done in \eqref{eq_1_tilde_exp}, set
$$
    \vec{X}_{(k,l)} \doteq \widetilde\exp^{-1}_{\varphi_{(k-1,l)}}\left( \varphi_{(k,l)}\right),
$$
$$
	\vec{Y}_{(k,l)} \doteq \widetilde\exp^{-1}_{\varphi_{(k,l-1)}}\left( \varphi_{(k,l)} \right).
$$
We claim that those are the vector fields interpolating between sections. They are always well defined because, by construction, we choose adjacent sections to be separated by a distance where $\widetilde\exp$ is still a diffeomorphism. Due to the mutual disjoint support of $\varphi_1$ and $\varphi_{-1}$, we can identify $\vec{X}_{(k,l)}$ (resp. $\vec{Y}_{(k,l)}$) with each other $\vec{X}_{(k,l')}$ (resp. $\vec{Y}_{(k',l)}$), therefore it is justified to use one index to denote the vector fields as done in \eqref{eq_+_vector} and \eqref{eq_-_vector}. Moreover, for each $k$, $l\in \mathbb{N}$, we have 
$$
    \widetilde\exp\big(\vec{X}_{k}\big)\circ \widetilde\exp\big(\vec{Y}_{l}\big)=\widetilde\exp\big(\vec{Y}_{l}\big)\circ \widetilde\exp\big(\vec{X}_{k}\big);
$$
which provides a well defined section $\varphi$.
% Since relative supports of sections are mutually disjoint, as it is in our case, the exponential flows of $\vec{X}_k$ and $\vec{Y}_{m}$ commute for each $k$, $m$ thus allowing the identification \eqref{interp comm}. Additionally we get $\varphi=\varphi^{(n,0,n)}_{(1,0,-1)}\equiv\varphi^{(0,n,0)}_{(1,0,-1)}$, $\varphi_1=\varphi^{(n,0,0)}_{(1,0,-1)}$ and $\varphi_{-1}=\varphi^{(0,0,n)}_{(1,0,-1)}$.
\end{proof}

\begin{proposition} \label{prop_1_additivity}
Let $F \in \mathcal{F}_0(B,\mathcal{U})$ then the following statements are equivalent:
\begin{itemize}
\item[$(i)$] $F$ is additive; 
\item[$(ii)$] $F$ is $\varphi_0$-additive for all $\varphi_0 \in \mathcal{U}$;
\item[$(iii)$] $F \in \mathcal{F}_{\mathrm{loc}}(B,\mathcal{U)}$.
\end{itemize}
\end{proposition}
\begin{proof}
Let us start proving the equivalence between $(i)$ and $(ii)$.\\
%\begin{itemize}
%\item[$(i) \Rightarrow (ii)$]
$(i) \Rightarrow (ii) $ If $\varphi_j \in \mathcal{U}\cap\mathcal{U}_{\varphi_0}$ with $j=1,0,-1$ are as in $(ii)$ above, take $\vec{X}_j$ such that $u_{\varphi_0}^{-1}(\vec{X}_j)=\varphi_j$. Writing \eqref{eq_1_glob_additivity} in terms of $F_{\varphi_0}$ yields \eqref{eq_1_loc_additivity}.\\
$(ii) \Rightarrow (i) $ Let us take sections $\varphi_j$ with $j=1,0,-1$ such that $\mathrm{supp}_{\varphi_0}(\varphi_1)\cap \mathrm{supp}_{\varphi_0}(\varphi_{-1})=\emptyset$, then we calculate $F(\varphi)$ combining Lemma \ref{lemma_1_interpolation_of_sections} with $\varphi$-additivity for each section, yields
%To avoid unnecessary notational complications when using $(7)$ we omit the section lower index on summation, which is nonetheless apparent from the context, and write for example that $\varphi^{(i+1,k,l)}_{(1,0,-1)}=\psi^{(1,k,0)}_{(1,0,-1)}+\varphi^{(i,k,l)}_{(1,0,-1)}$, where as in the previous implication $\psi$ somehow represents the action of the flow $\mathrm{Fl}^{\vec{s}_1}_{t_{(1),k}}$, in ultralocal terms.
$$
\begin{aligned}
	F(\varphi)& = F_{\varphi_{(n-1,n-1)}}\big(\vec{X}_n+ \vec{Y}_n\big)=
	F_{\varphi_{(n-1,n-1)}}\big(\vec{X}_n\big)+F_{\varphi_{(n-1,n-1)}}\big( \vec{Y}_n\big)-F_{\varphi_{(n-1,n-1)}}(0) \\ 
    &= F\left(\varphi_{(n,n-2)}\right)+\textcolor{red}{F\left(\varphi_{(n-1,n-1)}\right)}-F\left(\varphi_{(n-1,n-2)}\right)	\\ 
    &\quad+ F\left(\varphi_{(n-2,n)}\right)+F\left(\varphi_{(n-1,n-1)}\right)-F\left(\varphi_{(n-2,n-1)}\right)-	\textcolor{red}{F\left(\varphi_{(n-1,n-1)}\right)}\\ 
     &= F\left(\varphi_{(n,n-2)}\right)-\textcolor{blue}{F\left(\varphi_{(n-1,n-2)}\right)} + F\left(\varphi_{(n-2,n)}\right)+\textcolor{green}{F\left(\varphi_{(n-2,n-1)}\right)}+\textcolor{blue}{F\left(\varphi_{(n-1,n-2)}\right)}\\
    &\quad-F\left(\varphi_{(n-2,n-2)}\right)-\textcolor{green}{F\left(\varphi_{(n-2,n-1)}\right)} \\
    &=  F\left(\varphi_{(n,n-2)}\right)+F\left(\varphi_{(n-2,n)}\right)-F\left(\varphi_{(n-2,n-2)}\right).
\end{aligned}
$$
%\end{itemize}
Repeating the above argument an extra $(n-2)$ times we arrive at
$$
	F(\varphi)= F\left(\varphi_{(n,0)}\right)+F\left(\varphi_{(0,n)}\right)-F\left(\varphi_{(0,0)}\right)\equiv F(\varphi_1)+F(\varphi_{-1})-F(\varphi_0) 	.
$$
We conclude proving that $(ii)$ and $(iii)$ are equivalent.\\
%\begin{itemize}
$(iii) \Rightarrow (ii)\ $Take $\varphi_j$, $\vec{X}_j \doteq u_{\varphi_0}(\varphi_j) $, with $j=1,0,-1$ as in $(i)$ Definition \ref{def_1_additivity}. Then 
$$
\begin{aligned}
	 & F_{\varphi_0}(\vec{X}_1 +\vec{X}_{-1})-F_{\varphi_0}(\vec{X}_1)+F_{\varphi_0}(0)-F_{\varphi_0}(\vec{X}_{-1})=
	 \int_{0}^1 \frac{d}{dt}\left( F_{\varphi_0}(\vec{X}_1+ t\vec{X}_{-1})- F_{\varphi_0}(t\vec{X}_{-1})\right)dt  \\& =
	 \int_{0}^1  \frac{d}{dt}  \left( \int_0^1 \frac{d}{dh}  F_{\varphi_0}(h\vec{X}_1 + t\vec{X}_{-1})dh\right)dt= \int_{0}^1  \int_{0}^1 d^2F{\varphi_0}[h\vec{X}_1+ t\vec{X}_{-1}](\vec{X}_1,\vec{X}_{-1})dhdt.
\end{aligned}
$$
By locality we have that $\mathrm{supp}\big(d^2F{\varphi_0}\big) \subset \triangle_2 M$, however, $\mathrm{supp}(\vec{X}_1) \cap \mathrm{supp}(\vec{X}_{-1})=\emptyset$ implying that the integrand on the right hand side of the above equation is identically zero.\\
$(ii) \Rightarrow (iii)\ $Fix any $\varphi_0 \in \mathcal{U}$, consider two vector fields $\vec{X}_1$, $\vec{X}_{-1}\in \Gamma^{\infty}\big(M\leftarrow  \varphi_0^{*}VB\big)$ such that $\mathrm{supp}(\vec{X}_1) \cap\mathrm{supp}(\vec{X}_{-1})=\emptyset$ and $\vec{X}_1+\vec{X}_{-1} \in u_{\varphi_0}(\mathcal{U}_{\varphi_0})$, let also $\varphi_{j}\doteq u_{\varphi_0}^{-1}(\vec{X}_j)$ for $j=1,-1$, then $\mathrm{supp}_{\varphi}(\varphi_1) \cap \mathrm{supp}_{\varphi}(\varphi_{-1})=\emptyset$. By direct computation we get 
$$
\begin{aligned}
	 F_{\varphi_0}^{(2)}[0](\vec{X}_1,\vec{X}_{-1}) & =\left.\frac{d^2}{dt_1dt_2} \right|_{t_1=t_2=0} F_{\varphi_0}(t_1\vec{X}_1 + t_2\vec{X}_{-1}) \\ & = \left.\frac{d^2}{dt_1dt_2} \right|_{t_1=t_2=0} \left( F_{\varphi_0}(t_1\vec{X}_1)-F_{\varphi_0}(0)+F_{\varphi_0}( t_2\vec{X}_{-1}) \right) \equiv 0,
\end{aligned}
$$
which proves locality.
%\end{itemize}
\end{proof}

As a result, we have shown that locality and additivity are consistent concepts in a broader generality than done in \cite{acftstructure}. Of course, additivity strongly relates to Bogoliubov's formula for S-matrices, therefore a priori we expect that whenever we can formulate the concept consistently in Definition \ref{def_1_additivity}, those must be equivalent formulations. We also mention that when the exponential map used to construct ultralocal charts is a global diffeomorphism, then additivity and $\varphi$ additivity becomes trivially equivalent since the chart can be enlarged to $\mathcal{V}_{\varphi}\equiv \{\psi\in \Gamma^{\infty}(M\leftarrow B): \mathrm{supp}_{\varphi}(\psi) \subset M \space\ \mathrm{is} \space\ \mathrm{compact} \}$.\\

The ultralocal notion of additivity \textit{i.e.} $(i)$ in Definition \ref{def_1_additivity} is independent from the chart used, in fact suppose that $F$ is $\varphi_0$-additive in $\lbrace \mathcal{U}_{\varphi_0},u_{\varphi_0} \rbrace$, take another chart $\lbrace \mathcal{U}'_{\varphi_0},u'_{\varphi_0} \rbrace$ such that $\mathcal{U}'_{\varphi_0} \cap \mathcal{U}_{\varphi}\neq \emptyset$, set $\vec{X}_j=u_{\varphi_0}(\varphi_j)$, $\vec{Y}_j=u_{\varphi_0}'(\varphi_j)$, for $j=1,-1$, we have\footnote{In the subsequent calculations we can assume, without loss of generality, that $\vec{Y}_1+\vec{Y}_{-1}\in u'_{\varphi_0}(\mathcal{U}'_{\varphi_0})$, for if this is not the case we can use an argument involving Lemma \ref{lemma_1_interpolation_of_sections} to make this expression meaningful.}
$$
\begin{aligned}
	F\circ u'^{-1}_{\varphi_0}(\vec{Y}_1+\vec{Y}_{-1}) & =F\circ u^{-1}_{\varphi_0}\circ u_{\varphi_0}\circ u'^{-1}_{\varphi_0} (\vec{Y}_1+\vec{Y}_{-1})=F\circ u^{-1}_{\varphi_0}(\vec{X}_1+\vec{X}_{-1})\\
	& = F\circ u^{-1}_{\varphi_0}(\vec{X}_1 ) - F\circ u^{-1}_{\varphi_0}(0) + F\circ u^{-1}_{\varphi_0}(\vec{X}_{-1})\\ 
	& = F\circ u'^{-1}_{\varphi_0}(\vec{Y}_1 ) - F\circ u'^{-1}_{\varphi_0}(0) + F\circ u'^{-1}_{\varphi_0}(\vec{Y}_{-1})
\end{aligned}
$$
where $u_{\varphi_0}\circ u'^{-1}_{\varphi_0} (\vec{Y}_1+\vec{Y}_{-1})=\vec{X}_1+\vec{X}_{-1}$ is due to the fact that the two vector fields have mutually disjoint supports. We then see that $\varphi_0$-additivity does not depend upon the chosen chart.\\

We remark that functionals are generally defined in $CO$-open subsets instead of more general Whitney open sets since we can always extend the domain to a $CO$-open subset. To wit, suppose $F:\mathcal{U} \to \mathbb{R}$ is a smooth functional with compact support, then consider the function $\chi\in C^{\infty}_c(M)$ having $0\leq \chi \leq 1$, $\chi\equiv 1$ inside $\mathrm{supp}(F)$ and, given any $\varphi_0 \in \mathcal{U}$, define
\begin{equation}\label{eq_1_CO_extension_mapping}
    i_{\chi}: \Gamma^{\infty}(M \leftarrow B) \to \mathcal{V}_{\varphi_0}, \quad \widetilde{\psi} \mapsto \psi.
\end{equation}
The mapping can be constructed using with the ideas of Lemma \ref{lemma_1_interpolation_of_sections}, indeed, starting with $\varphi_0$, we can modify the latter inside $K$ so that $\widetilde\exp(\Vec{X}_n)\circ \cdots \circ \widetilde\exp(\Vec{X}_1)\varphi_0|_K=\widetilde{\psi}|_K$, then setting $\psi=\widetilde\exp(\chi\Vec{X}_n)\circ \cdots \circ \widetilde\exp(\chi\Vec{X}_1)\varphi_0$, we have that $\psi=\widetilde{\psi}$ inside $\mathrm{supp}(F)$, $\psi=\varphi_0$ outside $K$. $i_{\chi}$ is a continuous and smooth mapping, and when $\mathcal{U}$ is a $\mathrm{WO}^{\infty}$ open neighborhood of $\varphi_0$, $i_{\chi}^{-1}(\mathcal{U})$ is a $CO$-open. Then we can seamlessly extend the functional $F$ to $\widetilde{F}: i_{\chi}^{-1}(\mathcal{U})\to \mathbb{R}$. The functional will remain smooth, and all its derivatives will not be affected by the cutoff function $\chi$.

We now give the characterization of microlocality; we will find that, contrary to additivity, the latter representation will be limited to a chart domain, in the sense that the functional can be represented as an integral provided we shrink its domain to a chart, this representation however will not be independent from the chart chosen.

We recall that a mapping $T:U\subset E_1\to E_2$ between locally convex spaces is \textit{locally bornological} if for any $e\in E_1$ there is a neighborhood $V\ni e$ contained in $U$ such that $T|_V$ maps bounded subsets of $V$ into bounded subsets of $E_2$. From this definition follows a technical result:

\begin{lemma}\label{lemma_1_loc_bornology}
    Let $\mathcal{U}\subset \Gamma^{\infty}(M\leftarrow B)$ be $CO$-open, then a smooth, compactly spacetime supported functional $F$ satisfies: $dF:\mathcal{U}_{\varphi} \to \Gamma^{\infty}_c(M\leftarrow \varphi^*VB'\otimes \Lambda_m(M))$ is Bastiani smooth if and only if it is locally bornological. 
\end{lemma}
\begin{proof}[Sketch of a proof]
The proof of this result when $B=M\times \mathbb R$ can be found in \cite[Lemma 2.6]{acftstructure}. Since we are allowing for a bit more generality (\textit{i.e.} we are considering distributional sections of the bundle $\varphi^*VB\to M$), we will just highlight the minor changes to the argument presented in the aforementioned Lemma 2.6. From Bastiani smoothness of $F$ we can see $F^{(1)}$ as a Bastiani smooth mapping $\mathcal{U}_{\varphi}\to \Gamma^{-\infty}_c(M\leftarrow \varphi^*VB)$, combining the support property of $F$ with the fact that it is microlocal, we obtain that $F^{(1)}$ can be viewed as a mapping $T:\mathcal{U}_{\varphi}\to \Gamma^{\infty}_K(M\leftarrow \varphi^*VB\otimes \Lambda_m(M))$ for some compact subset $K$ of $M$. To prove the lemma it is enough to show that $T$ is Bastiani smooth if and only if it is locally bornological. Necessity then follows from the fact that composing $T$ with (Bastiani smooth) chart mappings yields a Bastiani smooth, hence continuous, mapping $\Gamma^{\infty}_c(M\leftarrow \varphi^*VB)\to\Gamma^{\infty}_K(M\leftarrow \varphi^*VB\otimes \Lambda_m(M))$. Both spaces are semi-Montel, \textit{i.e.} every bounded subset is relatively compact. Thus let ${W}\subset \overline{{W}}\subset {V} \subset u_{\varphi}(\mathcal{U}_{\varphi})$, with the first two subsets bounded, then $T(\overline{W})$ is compact, hence bounded, due to continuity of $F$. As a result $T|_{V}$ is locally bornological. The sufficiency condition is guaranteed if, 
\begin{itemize}
    \item for a given compact subset $K\subset M$ and a finite cover of $K$ of the form ${\psi_{\alpha}^{-1}(Q)}$ where $Q\subset \mathbb{R}^n$ is the open $m$-cube and $\psi_{\alpha}$ are the charts of $M$;
    \item given a partition of unity ${f_{\alpha}}$ of the above cover, the induced conveniently smooth\footnote{We recall, as in Definition \ref{def_A_convenient_smooth_map}, that a mapping is conveniently smooth if it maps smooth curves in $\mathcal{U}_{\varphi}$ to smooth curves of $\mathcal{D}\big(E|_Q\big)$.} mappings $T_{\alpha}: \mathcal{U}_{\varphi} \to \mathcal{E}'\big(Q;V\big), \quad \varphi \mapsto (\psi_{\alpha})_*\big(f_{\alpha}T(\varphi)\big)$, where $V \simeq \mathbb{R}^d$ is the typical fiber of the bundle $\varphi^*VB'\otimes \Lambda_m(M) \to M$, can be equivalently seen as conveniently smooth mappings $T_{\alpha}: \mathcal{U}_{\varphi} \to \mathcal{D}\big(Q;V\big)$.
\end{itemize}

By hypothesis we know that $T_{\alpha}: \mathcal{U}_{\varphi} \to \mathcal{D}\big(Q;V\big)$ is locally bornological, to complete the proof we note that each projection mapping $\pi_i:V \to \mathbb{R}$, $i=1,\ldots ,n$, induces continuous and also bounded mappings $(\pi_i)_{*} : \mathcal{D}\big(Q;V\big) \to \mathcal{D}(Q)$. Thus, by claims (i), (ii) in the proof of Lemma 2.6 in \cite{acftstructure}, each $(\pi_i)_{*} T_{\alpha} :\mathcal{U}_{\varphi} \to \mathcal{D}(Q) $, which remains locally bornological, maps smooth curves of $\mathcal{U}_{\varphi}$ to smooth curves of $\mathcal{D}(Q)$ implying that $T_{\alpha}: \mathcal{U}_{\varphi} \to \mathcal{D}\big(Q;V\big)$ is convenient smooth as well for any $\alpha$.
\end{proof}

\begin{proposition}\label{porop_1_muloc_charachterization}
Let $\mathcal{U}\subset\Gamma^{\infty}(M\leftarrow B)$ be $CO$-open and $F \in \mathcal{F}_{\mu loc}(B,\mathcal{U})$ , then $f^{(1)}:\mathcal{U}\subset \Gamma^{\infty}(M\leftarrow B)\ni \varphi \mapsto f^{(1)}_{\varphi}[0]\in \Gamma^{\infty}_c(M\leftarrow \varphi^{*}VB'\otimes \Lambda_m(M))$ is locally bornological if and only if for each $\mathcal{U}_{\varphi_0}\subset \mathcal{U}$ there is a $m$-form $\lambda_{F,\varphi_0}\equiv \lambda_{F,0}$ with $\lambda_{F,0}(j^{\infty}\varphi)$ having compact support for all $\varphi \in \mathcal{U}_{\varphi_0}$ such that
\begin{equation}\label{eq_1_muloc_formula}
	F(\varphi)=F(\varphi_0)+ \int_M (j^r_x\varphi)^{*}\lambda_{F,0}.
\end{equation}
\end{proposition}
\begin{proof}
Suppose $ F(\varphi)=F(\varphi_0)+ \int_M (j^{r}\varphi)^{*}\lambda_{F,0}$ for all $\varphi\in \mathcal{U}_{\varphi_0}$, we evaluate $dF_{\varphi}[0]$ and find that its integral kernel may always be recast in the form
$$
	f^{(1)}_{\varphi}[0](x)=e_i[\lambda,\varphi_0](j^{2r}_x\varphi)dy^i\otimes d\mu_g(x)
$$
where $e_i[\lambda,F,\varphi_0]dy^i\otimes d\mu_g:\Gamma^{\infty}(M\leftarrow B) \to \Gamma^{\infty}_c(M\leftarrow \varphi^{*}VB'\otimes \Lambda_m(M))$ are the Euler-Lagrange equations associated to $\lambda_{F,\varphi_0}$ evaluated at some field configuration. Using the ultralocal differential structure of the source space, and keeping in mind that $e_i[\lambda,F,]: \mathcal{U}_{\varphi_0}\to \Gamma^{\infty}_c(M\leftarrow \varphi^{*}VB'\otimes \Lambda_m(M))$ is an operator of bounded order, we can apply Theorem \ref{thm_A_Bastiani smooth_pushforward} and by Lemma \ref{lemma_1_loc_bornology} to get that
$f^{(1)}$ is locally bornological.

Conversely suppose that $f^{(1)}$ is locally bornological, by Lemma \ref{lemma_1_loc_bornology} it is Bastiani smooth as a mapping as well. Fix $\varphi_0\in \mathcal{U}$ and call $\vec{X}=u_{\varphi_0}(\varphi)$, by microlocality combined with Schwartz kernel theorem
$$
	F(\varphi)-F(\varphi_0)=F_{\varphi_0}(\vec{X})-F_{\varphi_0}(0)= \int_0^1 dF_{\varphi_0}[t\vec{X}](\vec{X}) dt=\int_0^1 dt \int_M f^{(1)}_{\varphi_0}[t\vec{X}](\vec{X}).
$$
Applying the Fubini-Tonelli theorem to exchange the integrals in the above relation yields our candidate for $j^{r}\varphi^{*}\lambda_{F,0}$: the $m$-form $x \mapsto \theta[\varphi](x)\equiv \int_0^1 f^{(1)}_{\varphi_0}[t\vec{X}](\vec{X})(x)dt$. We have to show that this element depends at most on $j^{r}_x\varphi$. Notice that, a priori, $\theta[\varphi](x)$ might not depend on $j^r_x\varphi$, however we can say that if $\varphi_1$, $\varphi_2$ agree on any neighborhood $V$ of $x\in M$, then $\theta[\varphi_1]\vert_{V}=\theta[\varphi_2]\vert_{V}$. To see this, set $\vec{X}_1=u_{\varphi_0}(\varphi_1)$, $\vec{X}_2=u_{\varphi_0}(\varphi_2)$, by construction they agree in a suitably small neighborhood $V'$ of $x$, moreover 
$$
\begin{aligned}
	 \theta[\varphi_1](x)-\theta[\varphi_2](x)& =\int_0^1 dt \left(f^{(1)}_{\varphi_0}[t\vec{X}_{1}](\vec{X}_{1})(x)-f^{(1)}_{\varphi_0}[t\vec{X}_{2}](\vec{X}_{2})(x)\right)\\ 
	&=\int_0^1 dt \left(f^{(1)}_{\varphi_0}[t\vec{X}_{1}]_i(x)\vec{X}_{1}^i(x)-f^{(1)}_{\varphi_0}[t\vec{X}_{2}]_i(x)\vec{X}_{2}^i(x)\right)\\ 
	&= \int_0^1 dt \int_0^1 dh f^{(2)}_{\varphi_0}[t\vec{X}_{2}+th\vec{X}_{1}-th\vec{X}_{2}]_{ij}(x)(t\vec{X}_{1}-t\vec{X}_{2})^i(x) \vec{X}_1^j(x); 
\end{aligned}
$$
where in the last equality we used locality of $F$ and linearity of the derivative. The last line of the above equation identically vanishes in $V'$ due the support properties of $f^{(2)}_{\varphi_0}$ and the fact that $\vec{X}_1\vert_{V'}=\vec{X}_2\vert_{V'}$. Therefore $\theta[\varphi](x)\in \varphi^{*}VB'\otimes \Lambda_m(M)$ depends at most on $\mathrm{germ}_x(\varphi)$ with $\varphi\in \mathcal{U}$. 

We wish to apply \Cref{thm_A_Peetre_Slovak} %\footnote{See \textit{e.g.}\ Theorem 19.10 pp. 180 in \cite{kolar2013natural} for the precise statement of the theorem. Also notice that due to the fact that our functionals have compact spacetime support, it is enough to show a weaker regularity condition then the one employed in the original statement of the theorem, thus the regularity hypothesis therein can be weakened by testing it with a family of mappings with variation happening on compact set. More precisely an operator $D:C^{\infty}(M,X)\to C^{\infty}(M,Y)$ is weakly regular if, given any compactly supported variation, \textit{i.e.} a jointly smooth family of mapping $\Phi:\mathbb{R}\times M\to N: t \mapsto \varphi_t \in C^{\infty}(M,X)$ such that there is a compact $K\subset M$ with $\varphi_t\vert_{M\backslash K}$ constant in the variation parameter $t\in \mathbb{R}$; the family $D(\varphi_t)$ is again a compactly supported variation.} 
to $\theta$; the germ dependence hypothesis has been verified above, so one has to show that $\theta$ is also weakly regular, that is, if $\mathbb{R}\times M \ni (t,x ) \mapsto \varphi_t(x) \in B$ is compactly supported variation, then $(t,x)\mapsto\theta[\varphi_t](x)$ is again a compactly supported variation. $\theta$ is a compactly supported form, thus it maps compactly supported variations into compactly supported variations. Moreover it is Bastiani smooth since 
$$
    \varphi \mapsto F(\varphi_0)+ \int_M\theta[\varphi]d\mu_g(x)=F(\varphi)
$$
is an observable. Then we can apply Lemma \ref{lemma_A_locality&bastiani_implies_w-reg}.

% Suppose therefore that $\mathbb{R}\ni t\to \varphi_t \in \Gamma^{\infty}(M\leftarrow B)$ is the mapping described above, then $\varphi_t$ is a smooth curve for the smooth structure of $\Gamma^{\infty}(M\leftarrow B)$ (see Proposition \ref{prop_1_smooth_Bastiani_curve}). By Bastiani smoothness of $\theta:\varphi\mapsto \theta[\varphi]$,
% $$
% 	t\mapsto \theta[\varphi_t] \in \Gamma^{\infty}_c(M\leftarrow \varphi^{*}_tVB'\otimes \Lambda_m(M))
% $$
% is smooth. For each $t\in \mathbb{R}$, we can define a canonical fibered isomorphism $\varphi^{*}_tVB'\otimes \Lambda_m(M)\simeq \varphi_0^{*}VB'\otimes \Lambda_m(M) $ which, by Theorem \ref{thm_A_Bastiani smooth_pushforward}, induces a smooth diffeomorphism $\Gamma^{\infty}_c(M\leftarrow \varphi^{*}_tVB'\otimes \Lambda_m(M)) \simeq \Gamma^{\infty}_c(M\leftarrow \varphi^{*}VB'\otimes \Lambda_m(M)) $. As a result we have that $\theta$ is a smooth mapping
% $$
% 	t\mapsto \theta[\varphi_t] \in \Gamma^{\infty}_c(M\leftarrow \varphi^{*}_0VB'\otimes \Lambda_m(M)).
% $$
% We infer by $(ii)$ in Proposition \ref{prop_A_conv_sections_of_vector_bndl} and the fact that $\theta[\varphi_t]$ is a constant variation outside a compact subset of $M$, that $\theta[\varphi]: M\times \mathbb{R}: (t,x) \mapsto \theta[\varphi_t](x) \in B $ is smooth. Note that the topology limit-Fréchet topology on $\Gamma^{\infty}_c(M\leftarrow \varphi^{*}_0VB'\otimes \Lambda_m(M))$ makes it a topological embedding in $\Gamma^{\infty}(M\leftarrow \varphi^{*}_0VB'\otimes \Lambda_m(M))$ with the $\mathrm{WO}^{\infty}$ topology (see Theorem \ref{thm_1_Gamma_c_TVS}).
We can now apply the Peetre-Slovak theorem and deduce that for each neighborhood there exists $r=r(x,\varphi_0) \in \mathbb{N}$, an open neighborhood $U^r\subset J^rB$ of $j^r\varphi_0$ and a mapping $\lambda_{F,0}:J^{r}B\supset U^r \to \Gamma^{\infty}_c(M\leftarrow \varphi^{*}_0VB'\otimes \Lambda_m(M))$ such that $\lambda_{F,0}(j^r_x\varphi)=\theta[\varphi](x)$ for each $\varphi$ with $j^r\varphi\in U^r$. Due to compactness of $\mathrm{supp}(\theta)$ we can take the order $r$ to be independent from the point $x$ on $M$; then
$$
	F(\varphi)=F(\varphi_0)+\int_M \lambda_{F,0}(j^r_x\varphi).
$$
\end{proof}

One could also strengthen the hypothesis of Proposition \ref{porop_1_muloc_charachterization}, for example, by requiring that for every $k \in \mathbb{N} $, $R>0$ and every $\varphi_0 \in\mathcal{U}$, there is a positive constant $C$ for which
\begin{equation}\label{eq_1_lipschitz_bound}
    \big|\big|f_{\varphi_0}^{(1)}[\varphi_1]-f_{\varphi_0}^{(1)}[\varphi_1]\big|\big|_{K,k} \leq C \big|\big|u_{\varphi_0}(\varphi_1)-u_{\varphi_0}(\varphi_1) \big|\big|_{K,k+r} 
\end{equation}
whenever
$$
    \sup_{\substack{x\in K \\ j\leq k}} \Big\vert\nabla^j\big(u_{\varphi_0}(\varphi_1)-u_{\varphi_0}(\varphi_1) \big)(x) \Big\vert \equiv \big|\big|u_{\varphi_0}(\varphi_1)-u_{\varphi_0}(\varphi_1) \big|\big|_{K,k+r} < R.
$$
This condition implies that $f^{(1)}$ is locally bornological, moreover it is sufficient (see Lemma 1 together with (B) of Theorem 1 in \cite{zajtz1999nonlinear}) to imply that the order $r$ from Proposition \ref{porop_1_muloc_charachterization} is independent from the section $\varphi$ and thus globally constant.\\

As mentioned above, this characterization is limited to the ultralocal chart chosen: given charts $\lbrace \mathcal{U}_{\varphi_j},u_{\varphi_j} \rbrace$, $j=1,2$ such that $\varphi \in \mathcal{U}_{\varphi_1} \cap \mathcal{U}_{\varphi_2}$, let $\vec{X}_j = u_{\varphi_j}(\varphi)$ and suppose $F$ satisfies the hypothesis of Proposition \ref{porop_1_muloc_charachterization}, then according to \eqref{eq_1_muloc_formula}
$$
	F(\varphi)=  F(\varphi_1)+ \int_M (j^{r_1}\varphi)^{*}\lambda_{F,1}= F(\varphi_2)+ \int_M (j^{r_2}\varphi)^{*}\lambda_{F,2}.
$$
We can assume that $r_1=r_2\equiv r$, then using the same argument as in the proof of Proposition \ref{porop_1_muloc_charachterization},
$$
\begin{aligned}
	(j^{r}\varphi)^{*}\lambda_{F,1}(x) &=  \int_0^1 f_{\varphi_1}^{(1)}[t\vec{X}_1]\left(\vec{X}_1\right)(x) \mathrm{d}t \\ 
	& = \int_0^1 f_{\varphi_2}^{(1)}[u_{\varphi_1 \varphi_2}(t\vec{X}_1)]\left( u_{\varphi_1 \varphi_2}(\vec{X}_1)\right)(x) \mathrm{d}t	\\ 
	& = \int_0^1 f_{\varphi_2}^{(1)}[u_{\varphi_1 \varphi_2}(t\vec{X}_1)]\left(\vec{X}_2\right)(x) \mathrm{d}t,
\end{aligned}
$$
whereas
$$
	(j^{r}\varphi)^{*}\lambda_{F,2}(x)= \int_0^1 f_{\varphi_2}^{(1)}[t\vec{X}_2]\left(\vec{X}_2\right)(x) \mathrm{d}t,
$$
we therefore see that the lack of linearity of the transition mapping $u_{\varphi_1\varphi_2}$, namely $u_{\varphi_1 \varphi_2}(t\vec{X}_1)\neq t u_{\varphi_1 \varphi_2}(\vec{X}_1)= t\vec{X}_2 $, does not allow us to conclude $(j^{\infty}\varphi)^{*}\lambda_{F,1}(x)= (j^{\infty}\varphi)^{*}\lambda_{F,2}(x)$.\\

We give another argument that prevents ultralocal chart independence of the $m$ forms obtained from Proposition \ref{porop_1_muloc_charachterization}. This relies on the variational sequence\footnote{A complete exposition can be found in \cite{krupka}.}: a cohomological sequence of forms over $J^rB$ for some finite $r\in \mathbb{N}$, 

\begin{center}
\hspace{-1.45cm}
\begin{tikzcd}
		 & 0 \arrow[r]\arrow[d, phantom, ""{coordinate, name=Z}] & \mathbb{R}  \arrow[r,"E_0"] & \Omega^1(J^rB)/\sim \arrow[r,"E_1"]& \ldots \arrow[r] & \Omega^m(J^rB)/\sim \arrow[dllll,"E_m",rounded corners,to path={ -- ([xshift=2ex]\tikztostart.east)|- (Z) [near end]\tikztonodes
-| ([xshift=-2ex]\tikztotarget.west)-- (\tikztotarget)}] \\
   & \Omega^{m+1}(J^rB)/\sim \arrow[r,"E_{m+1}"] & \Omega^{m+2}_{h}(B)/\sim  \arrow[r]
		  & \ldots \arrow[r] &\Omega^{N}(B) \arrow[r,"E_{N}"]  &  0
\end{tikzcd}
\end{center}
where each element of the sequence is the quotient of the space of $p$-forms in $J^rB$ modulo some relation that cancel the exact forms (in the sense of the de-Rham differential on the manifold $J^rB$) and accounts for integration by parts when the order is greater then $m=\mathrm{dim}(M)$. In particular the $m$th differential $E_m$ is the operator which, given a horizontal $m$-form, calculates its \textit{Euler-Lagrange form} and the $(m+1)$th differential $E_{m+1}$ is the operator which associates to each Euler-Lagrange form its \textit{Helmholtz-Sonin form}. By the Poincar\'e lemma, if $\sigma\in \mathrm{ker}(E_{m+1})$, there exists a local chart $(V^r,\psi^r)$ in $J^rB$ and a horizontal $m$-form $\lambda \in \Omega^m(V^r)$ having $E_m\vert_{V^r}(\lambda)=\sigma\vert_{V^r}$. Establishing whether this condition holds globally is the heart of the inverse problem in calculus of variation and can be formulated as follows: given equations satisfying some condition (the associated Helmholtz-Sonin form vanishes) do they arise from the variation of some Lagrangian? The variational sequence implies that this is always the case whenever the $m$-th cohomology group vanishes, therefore giving a sufficient conditions whenever some topological obstruction is not present. 

Now, if Proposition \ref{porop_1_muloc_charachterization} could somehow reproduce \eqref{eq_1_muloc_formula} for each $\varphi \in \mathcal{U}$ with an integral over the same $m$-form $\lambda_F$, we would have found a way to circumvent the topological obstructions that ruin the exactness of the variational sequence. Furthermore in the derivation of $\lambda_{F,0}$ we did not even require that the associated Euler-Lagrange equations had vanishing Helmholtz-Sonin form, but instead a Bastiani smoothness requirement that, due to Proposition \ref{thm_A_Bastiani smooth_pushforward}, will always be met by integral functionals constructed from smooth geometric objects. It appears therefore that the two approaches bears some kind of duality: given a representative $F_{\varphi}^{(1)}[0] \in \Omega^{m+1}(J^rB)/ \hspace{-0.1cm}\sim$ one can, on one hand, give a \textit{ultralocal chart dependent} Lagrangian via Proposition \ref{porop_1_muloc_charachterization} \textit{i.e.} a global $m$-form on the bundle $J^rB$ which however describe the functional only when evaluated in a small neighborhood of a backgroung section $\varphi_0$; on the other, prioritize \textit{ultralocal chart independence}, therefore having a local $m$-form defined on the bundle $J^r(\pi^{-1}U)$ for some open subset $U$ of $M$, which however describe the functional for all sections of $\Gamma^{\infty}(U\leftarrow \pi^{-1}(U))$.\\

\begin{proposition}\label{prop_1_variational_sqn}
Let $\mathcal{U}\subset\Gamma^{\infty}(M\leftarrow B)$ be $CO$-open, $F \in \mathcal{F}_{\mu loc}(B,\mathcal{U})$ satisfying the hypothesis of Proposition \ref{porop_1_muloc_charachterization} and the bound \eqref{eq_1_lipschitz_bound}. Fix $\varphi \in \mathcal{U}_{\varphi_0}$ and suppose that
$$
	F(\varphi)=F(\varphi_0)+ \int_M (j^{r}\varphi)^{*}\lambda_{F,0}=F(\varphi_0)+ \int_M (j^{r}\varphi)^{*}\lambda'_{F,0}.
$$
then $\lambda_{F,0}-\lambda'_{F,0}=d_h\theta$ for some $\theta \in \Omega^{m-1}_{\mathrm{hor}}(J^rB)$ if and only if the $m$-th de Rham cohomology group $H_{\mathrm{dR}}^m(B)=0$. In particular the above condition is verified whenever $B$ is a vector bundle with finite dimensional fiber and $M$ is orientable non-compact and connected.
\end{proposition}
\begin{proof}
Using the notation introduced above for the variational sequence we have that $E_m(\lambda_{F,0})=E_m(\lambda'_{F,0})$, since each of the two expressions equals $f^{(1)}_{\varphi_0}[0]$; thus their difference is zero and $\lambda_{f,\psi}-\lambda'_{f,\psi}\in \Omega^m(J^rB)/\sim$. The latter cohomology group is isomorphic, by the abstract de Rham Theorem, to $H_{\mathrm{dR}}^m(B)$, therefore $\lambda_{F,0}-\lambda'_{F,0}=d_h\theta$ if and only if $H_{\mathrm{dR}}^m(B)=0$. When $B$ is a vector bundle over $M$ its de Rham cohomology group are isomorphic to those of $M$, which when orientable non-compact and connected, has $H^m_{\mathrm{dR}}(M)=0$. The latter claim can be established using Poincar\'e duality, \textit{i.e.} $H^m_{\mathrm{dR}}(M)\simeq H^0_{\mathrm{dR},c}(M)$, if $M$ is non-compact and connected, \textit{e.g.} when it is globally hyperbolic, there are no compactly supported functions with vanishing differential other then the zero function, so the $m$-th cohomology group is zero.
\end{proof}

We shall conclude this section by introducing generalized Lagrangians, which, as the name suggests, will be used to select a dynamic on $\Gamma^{\infty}(M\leftarrow B)$. We stress that unlike the usual notion of Lagrangian - either a horizontal $m$-form over $J^rB$ or a morphism $J^rB\to \Lambda_m(M)$ - this definition will allow us to bypass all problems of convergence of integrals of forms in noncompact manifolds.

\begin{definition}\label{def_1_gen_lag}
Let $\mathcal{U}\subset\Gamma^{\infty}(M\leftarrow B)$ be $CO$-open. A generalized Lagrangian $\mathcal{L}$ on $\mathcal{U}$ is a mapping $$\mathcal{L}:C^{\infty}_c(M) \rightarrow \mathcal{F}_{c}(B,\mathcal{U}),$$ such that 
\begin{itemize}
\item[$(i)$] $\mathrm{supp}(\mathcal{L}(f))\subseteq \mathrm{supp}(f)$ and $\mathcal{L}(f)$ is Bastiani smooth for all $f\in C^{\infty}_c(M)$, 
\item[$(ii)$] for each $f_1$, $f_2$, $f_3 \in C^{\infty}_c(M)$ with $\mathrm{supp}(f_1)\cap \mathrm{supp}(f_3)=\emptyset$, $$\mathcal{L}(f_1+f_2+f_3)=\mathcal{L}(f_1+f_2)-\mathcal{L}(f_2)+\mathcal{L}(f_2+f_3).$$
\end{itemize}
\end{definition}

Given the properties of the above Definition we immediately get:  

\begin{proposition} \label{prop_1_gen_lag_1}
Let $\mathcal{U}\subset\Gamma^{\infty}(M\leftarrow B)$ be $CO$-open, $\mathcal{L}$ a generalized Lagrangian on $\mathcal{U}$. Then
\begin{itemize}
\item[$(i)$] $\mathrm{supp}(\mathcal{L}(f+f_0)-\mathcal{L}(f_0)) \subseteq \mathrm{supp}(f)$ for all $f$, $f_0 \in C^{\infty}_c(M)$,
\item[$(ii)$] for all $f\in C^{\infty}_c(M)$, $\mathcal{L}(f)$ is a local functional.
\end{itemize}
\end{proposition}
\begin{proof}
For $(i)$ take $x \notin \mathrm{supp}(f)$ then we can find some $g$ compactly supported in some open set containing $x$ with $\mathrm{supp}(g)\cap \mathrm{supp}(f) = \emptyset$ and $g \equiv f_0$ in a neighborhood $V$ of $x$ then using $(i)$ Definition \ref{def_1_gen_lag}, $\mathcal{L}(f-g+f_0)-\mathcal{L}(f_0)=(\mathcal{L}(f+f_0)-\mathcal{L}(f_0))+(\mathcal{L}(f_0-g)-\mathcal{L}(f_))$ which gives $\mathcal{L}(f+f_0)-\mathcal{L}(f_0)=\mathcal{L}(f-g+f_0)-\mathcal{L}(f_0-g)$ whence $\mathrm{supp}(\mathcal{L}(f+f_0)-\mathcal{L}(f_0))=\mathrm{supp}(\mathcal{L}(f+f_0-g)-\mathcal{L}(f_0-g))\subset \mathrm{supp}(f+f_0-g)\cup \mathrm{supp}(f_0-g)$. Then $(f+f_0-g)(x)=0$ when $x \notin \mathrm{supp}(f)$ and $(f_0-g)(x)=0$ for each $x\in V$ therefore $x \notin \mathrm{supp}(\mathcal{L}(f+f_0)-\mathcal{L}(f_0))$.\\
Now we show that for each $\varphi \in \mathcal{U}$, $\mathrm{supp}(d^2F{\varphi}[0])\subset \triangle_2(M)$. 
Due to $(ii)$ Proposition \ref{prop_1_additivity}, is equivalent to show that $F_{\varphi_0}$ is $\varphi_0$-additive for each $\varphi_0\in \mathcal{U}$. Fix any such section of $B$, take $\varphi_j \in \mathcal{U}_{\varphi}$, $j=1,2$ such that all hypothesis of Definition \ref{def_1_additivity} $(i)$ are satisfied and, as usual, call $\vec{X}_j= u_{\varphi_0}^{-1}(\varphi_j)$. Let $g_j \in C^{\infty}_c(M)$, $j=1,2$ such that $g_j \equiv 1$ in an open neighborhood of $ \mathrm{supp}_{\varphi_0}(\varphi_j)$. Given any $f \in C^{\infty}_c(M)$ we set $f_j=g_j f$ for $i=1,2$ and $f_0=f-f_1-f_2$. We want to show that 
\begin{equation}\label{eq_1_gen_lag_additivity_1}
	\mathcal{L}(f)_{\varphi_0}(\vec{X}_1+\vec{X}_2)=\mathcal{L}(f)_{\varphi_0}(\vec{X}_1)-\mathcal{L}(f)_{\varphi_0}(0)+ \mathcal{L}(f)_{\varphi_0}(\vec{X}_2).
\end{equation}
Expanding $f$ in the l.h.s. and using $(ii)$ Definition \ref{def_1_gen_lag} yields
$$
\begin{aligned}
	\mathcal{L}(f)_{\varphi_0}(\vec{X}_1+\vec{X}_2)&= \mathcal{L}(f_1+f_2)_{\varphi_0}(\vec{X}_1+\vec{X}_2)-\mathcal{L}(f_2)_{\varphi_0}(\vec{X}_1+\vec{X}_2)+ \mathcal{L}(f_2+f_3)_{\varphi_0}(\vec{X}_1+\vec{X}_2).
\end{aligned}
$$
By construction $\mathrm{supp}(\mathcal{L}(f_1+f_0))\subset \mathrm{supp}(f_1+f_0)$ and $ \mathrm{supp}(f_1+f_0)\cap \mathrm{supp}(\vec{X}_3)= \emptyset$, therefore $\mathcal{L}(f_1+f_0)_{\varphi_0}(\vec{X}_1+\vec{X}_2)=\mathcal{L}(f_1+f_0)_{\varphi_0}(\vec{X}_1)$. Repeating this argument for the other terms gives us
\begin{equation}\label{eq_1_gen_lag_additivity_2}
	\mathcal{L}(f)_{\varphi_0}(\vec{X}_1+\vec{X}_2)=\mathcal{L}(f_1+f_0)_{\varphi_0}(\vec{X}_1)-\mathcal{L}(f_0)_{\varphi_0}(0)+\mathcal{L}(f_0+f_2)_{\varphi_0}(\vec{X}_2).
\end{equation}
On the other hand expanding $f$ in the r.h.s. of \eqref{eq_1_gen_lag_additivity_1} and using a similar argument we obtain 
\begin{align*}
	\mathcal{L}(f)_{\varphi_0}(\vec{X}_1)&= \mathcal{L}(f_1+f_0)_{\varphi_0}(\vec{X}_1)-\mathcal{L}(f_0)_{\varphi_0}(\vec{X}_1)+\mathcal{L}(f_0+f_2)_{\varphi_0}(\vec{X}_1) \\
	\equiv & \mathcal{L}(f_1+f_0)_{\varphi_0}(\vec{X}_1)-\mathcal{L}(f_0)_{\varphi_0}(\vec{X}_1)+\mathcal{L}(f_0+f_2)_{\varphi_0}(0) ,\\
	\mathcal{L}(f)_{\varphi_0}(0)&= \mathcal{L}(f_1+f_0)_{\varphi_0}(0)-\mathcal{L}(f_0)_{\varphi_0}(0)+\mathcal{L}(f_0+f_2)_{\varphi_0}(0), \\
	\mathcal{L}(f)_{\varphi_0}(\vec{X}_2)&= \mathcal{L}(f_1+f_0)_{\varphi_0}(\vec{X}_2)-\mathcal{L}(f_0)_{\varphi_0}(\vec{X}_2)+\mathcal{L}(f_0+f_2)_{\varphi_0}(\vec{X}_2)\\
	&\equiv  \mathcal{L}(f_1+f_0)_{\varphi_0}(\vec{X}_2)-\mathcal{L}(f_0)_{\varphi_0}(0)+\mathcal{L}(f_0+f_2)_{\varphi_0}(\vec{X}_2).
\end{align*}
Which combined as in the r.h.s. of \eqref{eq_1_gen_lag_additivity_1} yield the same expression of the l.h.s of \eqref{eq_1_gen_lag_additivity_2}. 
\end{proof}

Combining the linearity of $C^{\infty}_c(M)$, property $(ii)$ Definition \ref{def_1_gen_lag} and Proposition \ref{prop_1_gen_lag_1} we obtain that each generalized Lagrangian can be written as a suitable sum of arbitrarily small supported generalized Lagrangians. To see it, fix $\epsilon>0$ and consider $\mathcal{L}(f)$. By compactness $\mathrm{supp}(f)$ admits a finite open cover of balls, $\lbrace B_i \rbrace_{i \in I}$ of radius $\epsilon$ such that none of the open balls is completely contained in the union of the others. Let $\lbrace g_i \rbrace_{i \in I}$ be a partition of unity subordinate to the above cover of $\mathrm{supp}(f)$, set $f_i \doteq g_i \cdot f$. Then using $(ii)$ Definition \ref{def_1_gen_lag}
$$
	\mathcal{L}(f)=\mathcal{L}\left(\sum_i f_i \right)=\sum_{J\subset I} c_J \mathcal{L}\left(\sum_{j \in J} f_j \right),
$$
where $J\subset I$ contains the indices of all balls $B_i$ having non empty intersection with a fixed ball (the latter included), and $c_J= \pm 1$ are suitable coefficients determined by the application of $(ii)$ Definition \ref{def_1_gen_lag}. By construction each index $J$ has at most two elements and $\mathrm{supp}(\sum_{j \in J} f_j)$ is contained at most in a ball of radius $2\epsilon$. We have thus split $\mathcal{L}(f)$ as a sum of generalized Lagrangians with arbitrarily small supports.

\begin{definition}\label{def_1_euler_der}
Let $\mathcal{U}\subset\Gamma^{\infty}(M\leftarrow B)$ be $CO$-open, $\mathcal{L}$ a generalized Lagrangian on $\mathcal{U}$. The $k$-th Euler-Lagrange derivative of $\mathcal{L}$ in $\varphi \in \mathcal{U}$ along $ (\vec{X}_1,\ldots, \vec{X}_k) \in \Gamma^{\infty}_c(M\leftarrow \varphi^*VB)^k $ is 
\begin{equation}\label{eq_1_variational_der}
	\delta^{(k)} \mathcal{L}(1)_{\varphi}[0](\vec{X}_1,\ldots, \vec{X}_k)\doteq \left.\frac{d^k}{dt_1 \ldots dt_k}\right|_{t_1= \ldots=t_k=0} \mathcal{L}(f)_{\varphi}[0] (t_1 \vec{X}_1+ \ldots+ t_k\vec{X}_k)
\end{equation}
where $\left. f \right|_{K} \equiv 1$ on a suitable compact $K$ containing all compacts $\mathrm{supp}(\vec{X}_i)$.
\end{definition}

%One can see how the compact supports of the $\Gamma^{\infty}$-tangent vectors, \textit{i.e.} the sections of $\Gamma^{\infty}_c(M\leftarrow \varphi^*VB)^k $ allow us to perform an adiabatic limit and consider the cutoff function $f$ to be identically $1$ throughout $M$.

From now on we will assume that generalized Lagrangian used are microlocal, \textit{i.e.} $\mathcal{L}(f)\in \mathcal{F}_{\mu loc}(B,\mathcal{U})$ for each $f\in C^{\infty}_c(M)$; this means that the first Euler-Lagrange derivative can be written as
\begin{equation}\label{eq_1_euler_der}
	\delta^{(1)} \mathcal{L}(1)_{\varphi}[0](\vec{X})=\int_M E(\mathcal{L})_{\varphi}[0](\vec{X}),
\end{equation}
where by microlocality $E(\mathcal{L})_{\varphi}[0]\in\Gamma^{\infty}_c(M\leftarrow \varphi^{*}VB' \otimes \Lambda_m(M))$.\\

A generalized Lagrangian $\mathcal{L}$ is trivial whenever $\mathrm{supp}\big(\mathcal{L}(f)\big)\subset \mathrm{supp}(df)$ for each $f\in C^{\infty}_c(M)$. Triviality induces an equivalence relation on the space of generalized Lagrangians, namely two $\mathcal{L}_1$, $\mathcal{L}_{2}$ are equivalent whenever their difference is trivial. We can show that if two Lagrangians $\mathcal{L}_1$, $\mathcal{L}_2$ are equivalent then they end up producing the same first variation \eqref{eq_1_euler_der}. For instance suppose that $\mathcal{L}_1(f)-\mathcal{L}_2(f) =\Delta\mathcal{L}(f)$ with $\Delta\mathcal{L}(f)$ trivial generalized Lagrangian for each $f\in C^{\infty}_c(M)$. To evaluate $\delta^{(1)}\Delta\mathcal{L}(1)_{\varphi}[0](\vec{X})$ one has to choose some $f$ which is identically $1$ in a neighborhood of $\mathrm{supp}(\vec{X})$, however, by $(i)$ Definition \ref{def_1_gen_lag} $\mathrm{supp}\big(\Delta\mathcal{L}(f)\big)\subset \mathrm{supp}(df)\cap \mathrm{supp}(\vec{X})=\emptyset$, therefore by Lemma \ref{lemma_1_support_property} we obtain $E(\Delta\mathcal{L})_{\varphi}[0](\vec{X})=0$ and
$$
\begin{aligned}
     \delta^{(1)} \mathcal{L}_1(f)_{\varphi}[0](\vec{X}) &=  \delta^{(1)} \mathcal{L}_2(1)_{\varphi}[0](\vec{X})+\delta^{(1)} \Delta\mathcal{L}(1)_{\varphi}[0](\vec{X})\\
     &= \int_M E(\mathcal{L}_2)_{\varphi}[0](\vec{X})+\int_M E(\Delta\mathcal{L})_{\varphi}[0](\vec{X})\\
     &=\delta^{(1)} \mathcal{L}_2(1)_{\varphi}[0](\vec{X}).
\end{aligned}
$$
Finally we compare our generalized action functional with the \textit{standard} action which is generally used in classical field theory (see \textit{e.g.} \cite{fati}, \cite{krupka}). One generally introduce the \textit{standard geometric} Lagrangian, $\lambda$ of order $r$, as a bundle morphism
\begin{center}
\begin{tikzcd}
		 & J^rB \arrow[r,"\lambda"] \arrow[d,"\pi^r"]  & \Lambda_m(M) \arrow[d, "\rho"] \\
		 & M\arrow[r,equal] & M 
\end{tikzcd}
\end{center}	
between $(J^rB,\pi^r,M)$ and $(\Lambda_m(M),\rho, M, \vert\wedge^m T_x^*M \vert)$, where the latter is the vector bundle whose typical fiber is the vector space of weight one $m$-form densities. %In coordinates, setting $d\sigma(x)=dx^1 \wedge \ldots\wedge dx^m$, abusing the notation we can write $\lambda(j^r_xy)$ as $\lambda(j^r_xy)d\sigma(x)$. 
Two Lagrangian morphisms $\lambda_1$, $\lambda_2$ are equivalent whenever their difference is an exact form. Its associated \textit{standard geometric} action functional will therefore be 
\begin{equation}\label{eq_classical_lagrangian}
	\mathcal{A}_D(\varphi)= \int_M \chi_D(x) \lambda(j^r_x\varphi)
\end{equation}
where $\lambda$ an element of the equivalence class of Lagrangian morphisms, $D$ is a compact region of $M$ whose boundary $\partial D$ is an orientable $(m-1)$-manifold and $\chi_D$ its characteristic function. One could be tempted to draw a parallel with a generalized Lagrangian by considering the mapping
\begin{equation}\label{eq_classical_gen_lag}
	\chi_D \mapsto \mathcal{A}(\chi_D)= \int_M \chi_D(x) \lambda(j^r_x\varphi)\ .
\end{equation}
However \eqref{eq_classical_gen_lag} differs from Definition \ref{def_1_gen_lag} in the singular character of the cutoff function. Indeed the functional $\mathcal{A}_D \in \mathcal{F}_{loc}(B,\mathcal{U})$ for each choice of compact $D$ but it is \textit{never} microlocal, for the integral kernel of $\mathcal{A}(\chi_D)^{(1)}_{\varphi}[0]$ has always singularities localized in $\partial D$. This is a severe problem when attempting to calculate the Peierls bracket for local functionals, a way out is to extend this bracket to less regular functionals (see Definition \ref{def_1_WF_mucaus}) maintaining the closure of the operation (see Theorem \ref{thm_1_peierls_closedness}); however, we cannot outright extend the bracket to all local functionals. Therefore, in order to accommodate those less regular functionals such as \eqref{eq_classical_lagrangian}, one would need to place severe restrictions on the possible compact subsets $D$ which cut off possible integration divergences. This, however, is not consistent with the derivation of Euler-Lagrange equations by the usual variation technique where the latter are obtained by imposing requirements that ought to hold \textit{for each} $D\subset M$ compact. 

Of course, given a Lagrangian morphism $\lambda$ of order $r$ we can always define a generalized microlocal Lagrangian by a microlocal-valued distribution, \textit{i.e.} 

\begin{equation*}
    C^{\infty}_c(M)\times\mathcal{U} \ni (f,\varphi) \mapsto \mathcal{L}(f)(\varphi)= \int_M f(x) \lambda(j^r_x\varphi) \ .
\end{equation*}

When we calculate higher order derivatives we get
\begin{equation}\label{eq_1_k-th_derivative_gen_Lag}
	d^{k+1} \mathcal{L}(1)_{\varphi}[0](\vec{X}_1,\ldots, \vec{X}_{k+1})= \int_M \delta^{(k)} E(\mathcal{L})_{\varphi}[0](\vec{X}_1)(\vec{X}_2,\ldots, \vec{X}_{k+1}).
\end{equation}
In particular, we can view $\delta^{(1)} E(\mathcal{L})_{\varphi}[0]: \Gamma^{\infty}_c(M\leftarrow \varphi^{*}VB) \rightarrow  \Gamma^{\infty}_c\big(M\leftarrow \varphi^{*}VB' \otimes \Lambda_m(M)\big) $, and induce the linearized field equations around $\varphi$ which represents one of the ingredients for the construction of the Peierls bracket.

\section{The Peierls bracket}\label{section_peierls_bracket}

Heuristically speaking the Peierls bracket is a duality relating two observables, $F$, $G$, that accounts for the effect of the (antisymmetric) influence of $F$ on $G$ when the latter is perturbed around a solution of certain equations. We will define this quantity using the linearized field equations which can be constructed with the second derivative of a generalized Lagrangian, which with some additional hypothesis will turn out to be normally hyperbolic. We start by reviewing some basic notions from the theory of normally hyperbolic (NH) operators.\\

In the last section we have shown how for a microlocal generalized Lagrangian $\delta^{(1)} E(\mathcal{L})_{\varphi}[0]: \Gamma^{\infty}_c(M\leftarrow \varphi^{*}VB) \otimes \Gamma^{\infty}_c(M\leftarrow \varphi^{*}VB) \rightarrow \Lambda_m(M) $ one can define a linear operator
$$
\delta^{(1)} E(\mathcal{L})_{\varphi}[0]: \Gamma^{\infty}_c(M\leftarrow \varphi^{*}VB)  \rightarrow \Gamma^{\infty}_c\big(M\leftarrow \varphi^{*}VB'\otimes \Lambda_m(M) \big) 
$$
if we fix a metric $h$ on the standard fiber of $B$ and the Lorentzian metric $g$ of $M$ inducing the Hodge isomorphism $*_g$, we can define 
\begin{equation}\label{eq_1_diffop_1}
	D_{\varphi} \doteq (\varphi^*h)^{\sharp} \circ (\mathrm{id}_{\varphi^{*}VB' } \otimes *_g)\circ  \delta^{(1)} E(\mathcal{L})_{\varphi}[0]: \Gamma^{\infty}(M\leftarrow \varphi^{*}VB ) \rightarrow \Gamma^{\infty}(M\leftarrow \varphi^{*}VB).
\end{equation}
For each $\varphi$ we can see that $D_{\varphi}$ is a differential operator and determine the principal symbol. 

\begin{lemma}
    Suppose that $D_{\varphi}=$ is as in \eqref{eq_1_diffop_1}, then its principal symbol is independent form the section $\varphi $ chosen.
\end{lemma}

\begin{proof}
Indeed by \eqref{eq_1_gamma_loc_change},
\begin{align*}
	\delta^{(1)} E(\mathcal{L})_{\psi}[0](\vec{X}_1,\vec{X_2}) & = \delta^{(1)} E(\mathcal{L})_{\varphi}[0]\left(d^1u_{\varphi \psi}[u_{\psi}(\varphi)](\vec{X}_1),d^1u_{\varphi \psi}[u_{\psi}(\varphi)](\vec{X_2})\right) \\
	&\quad + E(\mathcal{L})_{\varphi}[0]\left(d^2u_{\varphi \psi}[u_{\psi}(\varphi)](\vec{X}_1,\vec{X_2})\right);
\end{align*}
while the second piece modifies the expression of the differential operator, it does not alter its principal symbol since the local form of $d^2u_{\varphi \psi}[u_{\psi}(\varphi)](\vec{X}_1,\vec{X_2})$ does yield extra derivatives. We therefore conclude that if we use a generalized Lagrangian $\mathcal{L}$ whose linearized equations differential operator, $D_\varphi$, is normally hyperbolic for some $\varphi_0 \in \mathcal{U}$, then it is normally hyperbolic (with the same principal symbol) for all $\varphi\in \mathcal{U}_{\varphi_0}$.
\end{proof}

Let us give a more specific example on how to calculate the principal symbol from a microlocal generalized Lagrangian. Recalling formula \eqref{eq_1_derivative_lag_wavemaps} with $\lambda:J^1(M\times N) \to \Lambda^m(M)$, the latter being a first order Lagrangian, we have
\begin{align*}
	d^2\mathcal{L}_{f,\lambda,\varphi}[0](\vec{X}_1,\vec{X}_2)&=\int_{M}f(x)\bigg\{ \frac{\partial^2 \lambda}{\partial y^i \partial y^j} \vec{X}^i_1 \vec{X}^j_2 + \frac{\partial^2 \lambda}{\partial y^i_{\mu} \partial y^j} d_{\mu}\big(\vec{X}^i_1\big) \vec{X}^j_2\\ 
	& \quad+ \frac{\partial^2 \lambda}{\partial y^i_{\mu} \partial y^j} \vec{X}^i_1 d_{\mu}\big(\vec{X}^j_2\big)+ \frac{\partial^2 \lambda}{\partial y^i_{\mu} \partial y^j_{\nu}} d_{\mu}\big(\vec{X}^i_1 \big) d_{\nu}\big(\vec{X}^j_2\big)  \bigg\}(x){d}\mu_g(x),
\end{align*}
where $d_{\mu}$ is the horizontal differential on jet bundles. The key ingredient for the principal symbol is the quantity $m^{\mu \nu}_{ij}\doteq \frac{\partial^2 \lambda}{\partial y^i_{\mu} \partial y^j_{\nu}}$. Applying the transformations to get the differential operator of linearized field equations, as in \eqref{eq_1_diffop_1}, to the above quantity yields principal symbol
\begin{equation}\label{eq_1_princ_symb}
		\sigma_2(D_{\varphi})=h^{ij}m^{\mu \nu}_{jk} \otimes \partial_{\mu}\vee \partial_{\nu} \otimes e_i\otimes e^j.
\end{equation}
In case this quantity satisfies the condition of Definition \ref{def_1_normal_hyp} we can conclude that the operator is normally hyperbolic. There are also other notions of hyperbolicity, for instance see \cite{christodoulou}, where the hyperbolicity condition is strictly weaker than the one employed here. From now on we shall assume that our microlocal Lagrangian produces always normally hyperbolic linearized equations. Then we can invoke the results of Theorem \ref{thm_1_properties_of_Green_functions}.\\

Summing up we created a way of associating to each $\varphi$ in the domain of $\mathcal{L}$, operators
\begin{equation}\label{eq_1_green_1}
	G_{\varphi}^{\pm} : \Gamma^{\infty}_c(M\leftarrow \varphi^{*}VB) \rightarrow \Gamma^{\infty}(M\leftarrow \varphi^{*}VB).
\end{equation}
By $(i)$ of Theorem \ref{thm_1_properties_of_Green_functions} and linearity, $G^{\pm}_{\varphi}$ is a smooth mapping, that is, $G^{\pm}_{\varphi}\in C^{\infty}\big( \Gamma^{\infty}_c(M\leftarrow \varphi^{*}VB) ,\Gamma^{\infty}(M\leftarrow \varphi^{*}VB)\big)$ for each $\varphi\in \mathcal{U}$. Given $\vec{s}\in \Gamma^{-\infty}_c(M\leftarrow \varphi^*VB)$ we can view 
$$
	G^{\pm}(\vec{s}):\mathcal{U}\ni \varphi \mapsto G^{\pm}_{\varphi}(\vec{s}) \in \Gamma^{\infty}(M\leftarrow \varphi^{*}VB).
$$
We ask whether this map is Bastiani smooth, in particular, we seek to evaluate 
\begin{equation}\label{eq_1_der_green_1}
	\lim_{t\rightarrow 0}\frac{G^{\pm}_{u_{\varphi}^{-1}(t\vec{X})}(\vec{s})-G^{\pm}_{\varphi}(\vec{s})}{t}.
\end{equation}

\begin{lemma}\label{lemma_1_deriv_green}
Let $\gamma:\mathbb{R}\to \mathcal{U}\subset \Gamma^{\infty}(M\leftarrow B)$ be a smooth curve, then for each fixed $\vec{X}\in \Gamma_c^{\infty}(\varphi^{*}VB)$ the mapping $t\mapsto G^{\pm}_{\gamma(t)}(\vec{X})\in \Gamma^{\infty}(M\leftarrow \varphi^{*}VB)$ is smooth. In particular we have
\begin{equation}\label{eq_1_der_G+-}
	dG^{\pm }_{\varphi}(\vec{X})=\lim_{t\rightarrow 0}\frac{1}{t}\left( G^{\pm}_{u_{\varphi}^{-1}(t\vec{X})}-G^{\pm}_{\varphi} \right) = -G^{\pm}_{\varphi} \circ D_{\varphi}^{(1)}(\vec{X})\circ G^{\pm}_{\varphi},
\end{equation}
where $\mathcal{U}\ni\varphi\mapsto D_{\varphi}(\vec{X})\in \Gamma^{\infty}(M\leftarrow \varphi^{*}VB)$ is the mapping induced by \eqref{eq_1_diffop_1}.
\end{lemma}
\begin{proof}
	We just show the claim for the retarded propagator since for the advanced one the result follows in complete analogy. Instead of a distributional section $\vec s$ it suffices to show the claim for $\vec Y$. Then we evaluate
$$
	\lim_{t\to 0} \frac{1}{t}\Big({G}^+_{\gamma(t)}(\vec{Y})-\mathcal{G}_{\gamma(0)}^+(\vec{Y})\Big)
$$
In the following argument we will omit the evaluation at $\vec{s}$ from the notation. The differential operator $D(\vec{Y}): \mathcal{U}\ni\varphi\mapsto D_{\varphi}(\vec{Y})\in \Gamma^{\infty}(M\leftarrow \varphi^{*}VB)$ is smooth when composed with $\gamma$, therefore we consider 
$$
\begin{aligned}
	 \lim_{t\to 0} \frac{1}{t}\Big(D_{\gamma(t)}{G}^+_{\gamma(t)}-D_{\gamma(t)}{G}_{\gamma(0)}^+\Big)(\vec{Y})
	&=  \lim_{t\to 0} \frac{1}{t}\Big(\mathrm{id}_{\Gamma^{\infty}(M\leftarrow \gamma(0)^{*}VB)}-D_{\gamma(t)}{G}_{\gamma(0)}^+\Big) (\vec{Y})\\ 
	&=\lim_{t\to 0} \frac{1}{t}\Big(D_{\gamma(0)}{G}^+_{\gamma(0)}-D_{\gamma(t)}{G}_{\gamma(0)}^+\Big)(\vec{Y})\\
	&=\lim_{t\to 0} \frac{1}{t}\Big(D_{\gamma(0)}-D_{\gamma(t)}\Big)\big({G}_{\gamma(0)}^+(\vec{Y})\big)\\
	&=- D_{\gamma(0)}^{(1)}(\dot{\gamma}(0))\circ G^{+}_{\gamma(0)}\big(\vec{Y}\big).
\end{aligned}
$$
since $\gamma$ is a smooth curve in $\Gamma^{\infty}(M\leftarrow B)$, given any interval $[-\epsilon,\epsilon]$ with $\epsilon>0$, there is a compact subset $K_{\epsilon}$ of $M$ for which $\gamma(t)(x)$ is constant in $t$ on $M\backslash K_{\epsilon}$, then differential operator $D_{\gamma(0)}^{(1)}\neq 0$ only inside $K_{\epsilon}$, therefore the quantity $ \big(D_{\gamma(0)}^{(1)}(\dot{\gamma}(0)) \circ G^{\pm}_{\gamma(0)}\big)(\vec{Y}) $ has compact support  for any $\vec{Y}\in \Gamma^{\infty}(M\leftarrow \gamma(0)^{*}VB)$ and we can write
$$
	\lim_{t\to 0} \frac{1}{t}\Big({G}^+_{\gamma(t)}-\mathcal{G}_{\gamma(0)}^+\Big)=- G^{+}_{\gamma(0)} \circ D_{\gamma(0)}^{(1)}(\dot{\gamma}(0))\circ G^{+}_{\gamma(0)}.
$$
Using the above relation one can show that all iterated derivatives of ${G}_{\varphi}^+$ exists, thus showing smoothness.
\end{proof}
% By $(i)$ of Theorem \ref{thm_1_properties_of_Green_functions}, \eqref{eq_1_der_G+-} can be extended to cases where $\vec{X}\in \Gamma^{-\infty}_c(\varphi^*VB)$ is a distributional section.\\
Similarly for the causal propagator we find 
\begin{align}\label{eq_1_der_G}
	dG_{\varphi}(\vec{X})\doteq \lim_{t\rightarrow 0}\frac{1}{t}\left( G_{u_{\varphi}^{-1}(t\vec{X})}-G_{\varphi} \right) = -G_{\varphi} \circ D_{\varphi}^{(1)}(\vec{X})\circ G^{+}_{\varphi}-G^{-}_{\varphi} \circ D_{\varphi}^{(1)}(\vec{X})\circ G_{\varphi}.
\end{align}
Given the Green's functions $G_{\varphi}^{\pm}$, set 
\begin{equation}\label{eq_1_Gvar_+-}
	\mathcal{G}^{\pm}_{\varphi} \doteq G^{\pm}_{\varphi} \circ (\varphi^{*}h)^{\sharp} \circ (id_{\left(\varphi^{*}VB\right)'} \otimes *_g) : \Gamma^{\infty}_c(M\leftarrow \left(\varphi^{*}VB\right)' \times  \Lambda_m(M)) \rightarrow \Gamma^{\infty}(M\leftarrow \varphi^{*}VB),
\end{equation}
\begin{equation}\label{eq_1_Gvar}
	\mathcal{G}_{\varphi} \doteq G_{\varphi} \circ (\varphi^{*}h)^{\sharp} \circ (id_{\left(\varphi^{*}VB\right)'} \otimes *_g) : \Gamma^{\infty}_c(M\leftarrow \left(\varphi^{*}VB\right)' \times  \Lambda_m(M)) \rightarrow \Gamma^{\infty}(M\leftarrow \varphi^{*}VB).
\end{equation}

Note how, up to this point, we used some fiberwise metric $h$ in \eqref{eq_1_diffop_1} in order to have a proper differential operator for the subsequent steps. As a consequence the resulting operator $D_{\varphi}(h)$ does depend on the metric chosen and so do its retarded and advanced Green's operators $G_{\varphi}^{\pm}(h) $ with their counterparts $\mathcal{G}_{\varphi}^{\pm}(h) $. From the definition of Green's operators we have 
$$
	\left\lbrace \hspace{-0.4cm} \begin{array}{ll}
		&D_{\varphi}(h) \circ G_{\varphi}^{\pm}(h) = \mathrm{id}_{\Gamma^{\infty}_c(M\leftarrow \varphi^{*}VB)}\ ,\\
		& G_{\varphi}^{\pm}(h) \circ \left. D_{\varphi}(h) \right\vert_{\Gamma^{\infty}_c(M\leftarrow \varphi^{*}VB)}=\mathrm{id}_{\Gamma^{\infty}_c(M\leftarrow \varphi^{*}VB)}\ .
	\end{array}\right. 
$$
The latter is equivalent to 
$$
	\left\lbrace \hspace{-0.4cm} \begin{array}{ll}
		& \delta^{(1)}E(\mathcal{L})_{\varphi}[0]  \circ \mathcal{G}_{\varphi}^{\pm}(h)  = \mathrm{id}_{\Gamma^{\infty}_c(M\leftarrow \varphi^{*}VB' \otimes \Lambda_m(M))}\ ,\\
		& \mathcal{G}_{\varphi}^{\pm}(h) \circ  \left. \delta^{(1)}E(\mathcal{L})_{\varphi} )[0] \right\vert_{\Gamma^{\infty}_c(M\leftarrow \varphi^{*}VB)}=\mathrm{id}_{\Gamma^{\infty}_c(M\leftarrow \varphi^{*}VB)}\ .
	\end{array} \right. 
$$
Using the notation of \cite{bar}, the family of operators $\lbrace \mathcal{G}^{\pm}_{\varphi}(h) \rbrace$ defines a family of Green-hyperbolic type operators with respect to the differential operator $D_{\varphi}$ of the linearized equations at $\varphi$. Finally, using Theorem 3.8 in \cite{bar}, we get uniqueness for the advanced and retarded propagators, which in turn results in the independence of the Riemannian metric $h$ used before. The idea behind the proof is as follows: one would like to both extend the domain of $\mathcal{G}^{\pm}_{\varphi}(h)$ and reduce the target space to the same suitable space, once this is done, each propagator becomes the inverse of the linearized equations, then using uniqueness of the inverse we conclude. It turns out that the extension to the spaces $\Gamma^{\infty}_{\pm}(M\leftarrow \varphi^{*}VB'\otimes \Lambda_m(M))$ of future/past compact smooth sections does the job. 

%\textcolor{red}{proprietà degli operatori green hyp.}

\begin{lemma}\label{lemma_1_duality_green}
Let $g$ a Lorentzian metric on $M$ and $D: \Gamma^{\infty}(M\leftarrow E) \rightarrow \Gamma^{\infty}(M\leftarrow E)$ a linear partial differential operator. Then $D$ is self adjoint with respect to the pairing$\footnote{We require that at least one of the entries has compact support.}$ given by  
$$ 
	\left\langle \vec{s},\vec{t}\right.  \left. \right\rangle = \int_M (id_{E'} \otimes *_g \circ h^{\flat} (\vec{t} \space\ )) \vec{s} = \int_M h^{\flat} (\vec{t}\space\ ) \vec{s} \space\ d\mu_g 
$$
if and only if its kernel ${D}(x,y)$ is symmetric. Moreover if $D$ is normally hyperbolic, then $G_M^{+}$ and $G_M^{-}$ are each the adjoint of the other in the common domain.
\end{lemma}
\begin{proof}
The equivalent condition follows essentially from following chain of equivalences
$$
\begin{aligned}
	\left\langle D\vec{s},\vec{t}\space\ \right\rangle &=\int_M h^{\flat} (\vec{t})(x) D\vec{s}(x) \space\ d\mu_g(x)= \int_M h^{\flat} (\vec{t})(x)  h^{\sharp} \circ (id_{E'} \otimes *_g) \mathcal{D}(x,\vec{s})  \space\ d\mu_g(x) \\ &= \int_{M^2}  \mathcal{D}(x,y)\vec{t}(x)\vec{s}(y) d\mu_g(x) d\mu_g(y)
\end{aligned} .
$$
Suppose now $D$ is self adjoint, then 
$$
%\begin{aligned}
	\left\langle \vec{s}, G_M^{-} \vec{t} \space\ \right\rangle =\left\langle D G_M^{+} \vec{s}, G_M^{-} \vec{t}\space\ \right\rangle =\left\langle  G_M^{+}\vec{s},D G_M^{-}\vec{t}\space\ \right\rangle  =\left\langle G_M^{+}\vec{s},\vec{t} \space\ \right\rangle
%\end{aligned} 
$$
whence the desired adjoint properties of $G_M^{+}$ and $G_M^{-}$.
\end{proof}

For future convenience, we calculate the functional derivatives of $\mathcal{G}^{\pm}_{\varphi}$ and $\mathcal{G}_{\varphi}$, which are clearly smooth by combining Lemma \ref{lemma_1_deriv_green} with \eqref{eq_1_Gvar_+-} and \eqref{eq_1_Gvar}, whence 
\begin{align}\label{eq_1_der_Gvar+-}
\begin{aligned}
	d^{k}\mathcal{G}_{\varphi}^{\pm} (\vec{X}_1,\ldots, \vec{X}_k)&= \sum_{l=1}^{k}(-1)^l \sum_{\substack{(I_1,\ldots,I_l) \\ \in \mathcal{P}(1,\ldots, k)} } \bigg( \bigcirc_{i=1}^{l} \mathcal{G}_{\varphi}^{\pm} \circ \delta^{(|I_{\sigma(i)}| +1)}E(L)_{\varphi}[0]\big( \vec{X}_{I_{i}}\big) \bigg) \circ \mathcal{G}_{\varphi}^{\pm},
\end{aligned}
\end{align}
\begin{equation}\label{eq_1_der_Gvar}
\begin{aligned}
	d^{k}\mathcal{G}_{\varphi} (\vec{X}_1,\ldots, \vec{X}_k) =  \sum_{l=1}^{k}(-1)^l & \sum_{\substack{(I_1,\ldots,I_l) \in \mathcal{P}(1,\ldots, k)} } \sum_{m=0}^{l}\bigg( \bigcirc_{i=1}^{m} \mathcal{G}^-_{\varphi}\circ \delta^{(|I_i| +1)}E(L)\big(\vec{X}_{I_i}\big) \bigg) \\
 &\circ \mathcal{G}_{\varphi}\circ \bigg( \bigcirc_{i=m+1}^{l} \delta^{(|I_i| +1)}E(L)_{\varphi}\big(\vec{X}_{I_i}\big)\circ \mathcal{G}^+_{\varphi}  \bigg),
\end{aligned}
\end{equation}
where $(I_1,\ldots,I_l)$ is partition of the set $\lbrace 1,\ldots,k \rbrace $, and $\vec{X}_I=\otimes_{i\in I}\vec{X}_i$. The main takeaway from \eqref{eq_1_der_Gvar} is the pattern of the composition of propagators and derivatives of $E(L)$, that is first the $\mathcal{G}_{\varphi}^-$'s, then a single $\mathcal{G}$ and at the end some $\mathcal{G}_{\varphi}^+$'s intertwined by derivatives of $E(\mathcal{L})$. These will be key to some later proofs. We are now in a position to introduce the Peierls bracket:

\begin{definition}\label{def_1_Peierls}
Let $\mathcal{U}\subset\Gamma^{\infty}(M\leftarrow B)$ be $CO$-open, and $F$, $G \in \mathcal{F}_{\mu loc}(B,\mathcal{U})$. Fix a generalized microlocal Lagrangian $\mathcal{L}$ whose linearized equations induce a normally hyperbolic operator. The retarded and advanced products $\textsf{R}_{\mathcal{L}}(F,G)$, $\textsf{A}_{\mathcal{L}}(F,G)$ are functionals defined by 
\begin{equation}
	\textsf{R}_{\mathcal{L}}(F,G)(\varphi) \doteq \left\langle dF_{\varphi}[0], \mathcal{G}_{\varphi}^{+} dG_{\varphi}[0]\right\rangle,
\end{equation}
\begin{equation}
	\textsf{A}_{\mathcal{L}}(F,G)(\varphi) \doteq \left\langle dF_{\varphi}[0], \mathcal{G}_{\varphi}^{-} dG_{\varphi}[0]\right\rangle,
\end{equation}
while the Peierls bracket of $F$ and $G$ is 
\begin{align}\label{eq_1_Peierls}
\left\lbrace F, G \right\rbrace_{\mathcal{L}} \doteq \textsf{R}_{\mathcal{L}}(F,G)-\textsf{A}_{\mathcal{L}}(F,G).
\end{align}
\end{definition}
We recall that for a microlocal functional $F$, by \eqref{eq_1_kernel_notation}
$$
	dF_{\varphi}[0](\vec{X})=\int_M f^{(1)}_{\varphi}[0]_i(x)X^i(x)d\mu_g(x),
$$
therefore we can write $\left\lbrace F, G \right\rbrace_{\mathcal{L}}(\varphi)$ as
\begin{equation}\label{eq_1_Peierls_kernel}
	\int_{M^2} f^{(1)}_{\varphi}[0]_i(x) \mathcal{G}_{\varphi}^{ij}(x,y) g^{(1)}_{\varphi }[0]_j(y) d\mu_g(x,y)
\end{equation}
where repeated indices as usual follows the Einstein notation. This implies clearly that Definition \ref{def_1_Peierls} is well posed. Moreover as a consequence of Lemma \ref{lemma_1_duality_green} we see that the Peierls bracket of $F$ and $G$ can also equivalently viewed as $\textsf{R}_{\mathcal{L}}(F,G)-\textsf{R}_{\mathcal{L}}(G,F)=\textsf{A}_{\mathcal{L}}(G,F)-\textsf{A}_{\mathcal{L}}(F,G)$. 

We begin our analysis of the Peierls bracket by listing the support properties of the functionals defined in Definition \ref{def_1_Peierls}. 

\begin{proposition}\label{prop_1_Peierls_1}
Let $\mathcal{U}$, $F$, $G$ be as in the above definition, then the retarded, advanced products and Peierls bracket are Bastiani smooth with the following support properties:
\begin{align}
	\mathrm{supp} \left( \textsf{R}_{\mathcal{L}}(F,G) \right) \subset & J^+(\mathrm{supp}(F)) \cap  J^-(\mathrm{supp}(G)),\\
	\mathrm{supp} \left( \textsf{A}_{\mathcal{L}}(F,G) \right) \subset & J^+(\mathrm{supp}(G)) \cap  J^-(\mathrm{supp}(F)),	 
\end{align}
which combined yields
\begin{equation}\label{eq_1_supp_peierls}
	\mathrm{supp} \left( \left\lbrace F,G \right\rbrace_{\mathcal{L}} \right) \subset   \left( J^+(\mathrm{supp}(F)) \cup J^-(\mathrm{supp}(F)) \right) \cap \left( J^+(\mathrm{supp}(G)) \cup J^-(\mathrm{supp}(G)) \right).
\end{equation}
\end{proposition}

\begin{proof}
By definition the support properties of $\mathcal{G}_M^{\pm}$ and $G_M^{\pm}$ are analogue, so combining these properties with  $\textsf{R}_{\mathcal{L}}(F,G)= \frac{1}{2} \textsf{R}_{\mathcal{L}}(F,G) + \frac{1}{2}  \textsf{A}_{\mathcal{L}}(G,F)$ yields the desired result. We now turn to the smoothness. We calculate the $k$-th derivative of $ \textsf{R}_{\mathcal{L}}$. By the chain rule, taking $\mathcal{P}(1,\ldots,k)$ the set of permutations of $\lbrace 1,\ldots , k \rbrace$, we can write 
\begin{align}\label{eq_1_der_retarded}
\begin{aligned}
	 &d^k\textsf{R}_{\mathcal{L}} (F,G)_{\varphi}[0](\vec{X}_{1},\ldots, \vec{X}_k)\\
	 &\quad = \sum_{(J_1,J_2,J_3)\subset P_k} \left\langle F^{(|J_1|+1)}_{\varphi}[0](\otimes_{j_1 \in J_1}\vec{X}_{j_1}) \right., 
	 \left.d^{(|J_2|)}\mathcal{G}_{\varphi}^{+}(\otimes_{j_2 \in J_2}\vec{X}_{j_2}) G^{(|J_3|+1)}_{\varphi}[0] (\otimes_{j_3 \in J_3}\vec{X}_{j_3}) \right\rangle,
\end{aligned}
\end{align}
and similarly
\begin{align}\label{eq_1_der_advanced}
\begin{aligned}
	&d^k\textsf{A}_{\mathcal{L}} (F,G)_{\varphi}(\vec{X}_{1},\ldots, \vec{X}_k)\\ 
	&\quad = \sum_{(J_1,J_2,J_3)\subset P_k} \left\langle F^{(|J_1|+1)}[\varphi](\otimes_{j_1 \in J_1}\vec{X}_{j_1})  \right., \left. d^{(|J_2|)}\mathcal{G}_{\varphi}^{-}(\otimes_{j_2 \in J_2}\vec{X}_{j_2}) G^{(|J_3|+1)}_{\varphi}[0] (\otimes_{j_3 \in J_3}\vec{X}_{j_3}) \right\rangle.
\end{aligned}
\end{align}
To see that the pairing in the derivatives of the  advanced, retarded products are well defined, we use the kernel notation \eqref{eq_1_kernel_notation}, therefore we write the integral kernel of $\textsf{R}_{\mathcal{L}}(F,G)$, which by a little abuse of notation we call $\textsf{R}_{\mathcal{L}}(F,G)(x,y)$ for $x$,$y \in M$. It is
$$
	\textsf{R}_{\mathcal{L}}(F,G)(x,y)= f^{(1)}_{\varphi}[0]_i(x) \left(\mathcal{G}^{+}_{\varphi}\right)^{ij}(x,y) g^{(1)}_{\varphi }[0]_j(y).
$$
Using this notation, we can write the integral kernel $d^k\textsf{R}_{\mathcal{L}} (F,G)_{\varphi}[0](\vec{X}_{1},\ldots, \vec{X}_k)(x,y)$ in \eqref{eq_1_der_retarded} as a sum of terms with two possible contributions:\\
$1)$ [$J_2=\emptyset$]
$$
	 f^{(p+1)}_{\varphi}[0]_i(x,\vec{X}_1, \ldots, \vec{X}_p) (\mathcal{G}^{+}_{\varphi})^{ij}(x,y) g^{(q+1)}_{\varphi}[0]_j(y,\vec{X}_{q+1} \ldots \vec{X}_{p+q}),
$$
where $p+q=k$. Due to smoothness of the functionals, this is well defined and continuous, so this part yields a Bastiani smooth functional;\\
$2)$ $J_2 \neq \emptyset$] 
$$
\begin{aligned}
	&\quad\int_{M^{k-2}} f^{(|J_1|+1)}_{\varphi}[0]_i\big(x,\vec{X}_{J_1}\big) \left(\mathcal{G}^{+}_{\varphi}\right)^{ij_1}(x,z_1)  \delta^{(|I_1|+1)}  E(\mathcal{L})_{\varphi}[0]_{j_1 j_2} \big(z_1,z_2,\vec{X}_{I_1}\big) \\
	&\qquad\times\left(\mathcal{G}^{+}_{\varphi}\right)^{j_2j_3}(z_2,z_3) \delta^{(|I_2|+1)} E( \mathcal{L})_{\varphi}[0]_{j_3 j_4}\big(z_3,z_4,\vec{X}_{I_2}\big) \cdots \delta^{(|I_l|+2)} E(\mathcal{L})_{\varphi}[0]_{j_{2l-1} j_{2l}}\big(z_{2l-1},z_{2l},\vec{X}_{I_l} \big) \\
	&\qquad\times\left(\mathcal{G}^{+}_{\varphi}\right)^{j_{2l}j}(z_{2l},y)   g^{(|J_3|+1)}_{\varphi}[0]_j(y,\vec{X}_{p+k_1+\ldots+ k_l+1}, \ldots, \vec{X}_{J_3})d\mu_g(z_1,\ldots ,z_{2l}) ,
\end{aligned}
$$
where $I_1\cup\ldots\cup I_l=J_2$. Again due to the Bastiani smoothness of all functionals involved in the above formula, we conclude that this piece too exists and is continuous. Hence as a whole $\textsf{R}_{\mathcal{L}}(F,G)^{(k)}_{\varphi}$. Repeating the above calculations for $\textsf{A}_{\mathcal{L}} $ amounts to substituting  each $+$ with $-$, resulting in Bastiani smoothness for the advanced product. Finally since $\left\lbrace F, G \right\rbrace_{\mathcal{L}} = \textsf{R}_{\mathcal{L}}(F,G)-\textsf{A}_{\mathcal{L}}(F,G)$ we conclude that it is smooth as well.
\end{proof}

We have seen that the Peierls bracket is well defined for microlocal functionals, we stress however that the image under the Peierls bracket of microlocal functionals fails to be microlocal, it is therefore necessary to broaden the functional domain of this bracket. An idea is to use the full potential of microlocal analysis, and use wave front sets to define pairings. First though we make explicit the ``good" subset of $T^{*}M$, that is, those subsets in which the wavefront can be localized.

\begin{definition}\label{def_1_WF_mucaus}
Let $(M,g)$ be a Lorentzian spacetime, define $\Upsilon_k(g) \subset T^{*}M^k$ as follows:
\begin{align}\label{eq_1_WF_mu_caus}
\begin{aligned}
	\Upsilon_k(g)\doteq \Big\lbrace (x_1,\ldots,x_k,\xi_1,\ldots,\xi_k)& \in T^{*}M^k \backslash 0 :\\
	&(x_1,\ldots,x_k,\xi_1,\ldots,\xi_k) \notin \overline{V}^+_{k}(x_1,\ldots,x_k) \cup \overline{V}^-_{k}(x_1,\ldots,x_k) \Big\rbrace
\end{aligned}
\end{align}
where 
$$
	\overline{V}^{\pm}_k(x_1,\ldots,x_k) = \prod_{j=1}^k \overline{V}^{\pm}(x_j). 
$$
If $\mathcal{U}\subset \Gamma^{\infty}(M\leftarrow B)$ is open we say that a functional $F:\mathcal{U} \rightarrow \mathbb{R}$ with compact support is \textit{microcausal} with respect to the Lorentz metric $g$ in $\varphi$ if $\mathrm{WF}(d^kF_{\varphi}[0]) \cap \Upsilon_k(g)=\emptyset$ for all $k \in \mathbb{N}$. We say that $F$ is microcausal with respect to $g$ in $\mathcal{U}$ if $F$ is microcausal for all $\varphi \in \mathcal{U}$. We denote the set of microcausal functionals in $\mathcal{U}$ by $\mathcal{F}_{\mu c}(B,\mathcal{U},g)$.
\end{definition}
% We remark how Definition \eqref{def_1_WF_mucaus} can equivalently be given by requiring 
% $$
% 	\mathrm{WF} \left(\nabla^{(k)}F_{\varphi}[0]\right) \subset \Upsilon_k(g).
% $$
One can show by induction, using \eqref{eq_1_cov_der}, that the two definitions are equivalent. The case $k=1$ is trivial, while the case with arbitrary $k$ follows from:

\begin{lemma}\label{lemma_1_cov_WF_equivalence}
	Suppose that for microcausal functional $F$ there is a given symmetric linear connection having $\mathrm{WF}\left(\nabla^{n-1}F_{\varphi}[0]\right) \cap \Upsilon_{n-1}(g)=\emptyset$, then 
$$
	\mathrm{WF}\left(\nabla^{n}F_{\varphi}[0]\right) \cap \Upsilon_n(g) =\emptyset.
$$
\end{lemma}
\begin{proof}
From \eqref{eq_1_cov_der} we have 
\begin{align*}
	&\nabla^{n}F_{\varphi}[0](\vec{X}_1,\ldots, \vec{X}_n) \\ &\doteq d^nF_{\varphi}[0](\vec{X}_1,\ldots, \vec{X}_n) 
	 + \sum_{j=1}^n \frac{1}{n!} \sum_{\sigma \in \mathcal{P}(n)} \nabla^{n-1}F_{\varphi} [0](\Gamma_{\varphi} (\vec{X}_{\sigma(j)}, \vec{X}_{\sigma(n)}), \vec{X}_{\sigma(1)}, \ldots, \widehat{\vec{X}_{\sigma(j)}}, \ldots \vec{X}_{\sigma(n-1)})\ .
\end{align*}
Assume that $\nabla^{n-1}F_{\varphi}[0]$ is microcausal. Since $F$ is microcausal as well, it is sufficient to show microcausality holds for the other terms in the sum. Due to symmetry of the connection, we can simply study the wave front set of a single term such as 
\begin{equation}\label{eq_1_kernel_comp}
	\nabla^{n-1}F_{\varphi} (\Gamma_{\varphi} (\vec{X}_{j}, \vec{X}_{n}), \vec{X}_{1}, \ldots, \widehat{\vec{X}_{j}}, \ldots \vec{X}_{n-1}).
\end{equation}
The idea is to apply Theorem 8.2.14 in \cite{hormanderI}. Recall that a connection $\Gamma_{\varphi}$ can be seen as a mapping $\Gamma^{\infty}_c(M\leftarrow \varphi^*VB)\times \Gamma^{\infty}_c(M\leftarrow \varphi^*VB)\rightarrow\Gamma^{\infty}_c(M\leftarrow \varphi^*VB)$ with associated integral kernel $\Gamma[\varphi](x,y,z)$ defined by
$$
	\otimes^3 \Gamma^{\infty}_c(M\leftarrow \varphi^*VB) \rightarrow \mathbb{R}: (\vec{X},\vec{Y},\vec{Z}) \mapsto \int_{M^3} h_{kl}(\varphi(x))\Gamma[\varphi]^l_{ij}(x,y,z)\vec{X}^i(x)\vec{Y}^j(y) \vec{Z}^k(z) d\mu_g(x,y,z)
$$
where $h$ is an auxiliary Riemannian metric on the fiber of the bundle $B$ which is to be regarded as a tool for calculations. We can estimate the wave front set of $\Gamma[\varphi]^l_{ij}(x,y,z)$ by using the support properties of the connection coefficients $\Gamma_{\varphi}$ and obtain $\Gamma[\varphi]^l_{ij}(x,y,z)=\Gamma^l_{ij}(\varphi(x))\delta(x,y,z)$, with $\Gamma^l_{ij}(\varphi(x))$ Christoffel coefficients of a connection on the typical fiber of $B$; thus
$$
	\mathrm{WF}\left(\Gamma[\varphi]\right)= \lbrace (x,y,z,\xi,\eta,\zeta) \in T^*M^3\backslash 0: x=y=z, \ \xi+\eta+\zeta=0 \rbrace.
$$
Composition of the two integral kernels in \eqref{eq_1_kernel_comp} is well defined provided $\mathrm{WF}'(\nabla^{n-1}F_{\varphi}[0])_{M}\cap \mathrm{WF}\left(\Gamma[\varphi]\right)_{M}= \emptyset$ and that the projection map $: \triangle_{3}M \rightarrow M $ is proper. The former is a consequence of $\mathrm{WF}\left(\Gamma[\varphi]\right)_{M}= \emptyset$, the latter is a trivial statement for the diagonal embedding. Then we can apply Theorem 8.2.14, and estimate
\begin{align*}
	\mathrm{WF}\left( \nabla^{n-1}F_{\varphi} \circ \Gamma_{\varphi} \right) \subset & \Big\{ (x_1,\ldots,x_n,\xi_1,\ldots,\xi_n) \in T^*M^n : \exists (y,\eta) : \  (x_j,x_n,y,\xi_j,\xi_n,-\eta) \in  \mathrm{WF}\left(\Gamma[\varphi]\right)\ ,\\ 
	& \quad (y,x_1,\ldots, \widehat{x_j},\ldots ,x_{n-1},\eta,\xi_1,\ldots,\widehat{\xi}_j,\ldots,\xi_{n-1})\in \mathrm{WF} \left(\nabla^{n-1}F_{\varphi}[0] \right) \Big\} \\
	& \bigcup
\Big\{ (x_1,\ldots,x_n,\xi_1,\ldots,\xi_n) \in T^*M^n : x_j=x_n , \space\ \xi_j = \xi_n=0\ , \\ 
	&\qquad  (y,x_1,\ldots, \widehat{x_j},\ldots ,x_{n-1},0,\xi_1,\ldots,\widehat{\xi}_j,\ldots,\xi_{n-1})\in \mathrm{WF} \left(\nabla^{n-1}F_{\varphi}[0] \right) \Big\} \\ 
	& \bigcup \Big\{ (x_1,\ldots,x_n,0,\ldots,0,\xi_j,0,\ldots,0,\xi_n) \in T^*M^n :  (x_j,x_n,y,\xi_j,\xi_n,0) \in  \mathrm{WF}\left(\Gamma_{\varphi}\right)\ , \\ 
	&\qquad (y,x_1,\ldots, \widehat{x_j},\ldots ,x_{n-1},\eta,0,\ldots,0)\in \mathrm{WF} \left(\nabla^{n-1}F_{\varphi}[0] \right) \Big\}\\
	& = \Pi_1\cup \Pi_2\cup \Pi_3.
\end{align*}
If by contradiction, we had that the $\nabla^{n-1}F_{\varphi} \circ \Gamma_{\varphi}$ was not microcausal, then there would be elements of its wavefront set for which all $\xi_1,\ldots,\xi_n$ are, say, future pointing. In this case those must belong to $\Pi_1$, but then $\eta$ is future pointing as well by the form of $\mathrm{WF}(\Gamma_{\varphi})$, so that $\nabla^{n-1}F_{\varphi}$ is not microcausal, contradicting our initial assumption.
\end{proof}

One can also show that microcausality does not depend upon the connection chosen by computing
\begin{align*}
	\nabla^{n}F_{\varphi}[0](\vec{X}_1,\ldots, \vec{X}_n) & - \widetilde{\nabla}^{n}F_{\varphi}[0]\big(\vec{X}_1,\ldots, \vec{X}_n\big)  \\
	&=  \sum_{j=1}^n \frac{1}{n!} \sum_{\sigma \in \mathcal{P}(n)} \nabla^{n-1}F_{\varphi} [0]\Big(\Gamma_{\varphi} \big(\vec{X}_{\sigma(j)}, \vec{X}_{\sigma(n)}\big), \vec{X}_{\sigma(1)}, \ldots, \widehat{\vec{X}_{\sigma(j)}}, \ldots \vec{X}_{\sigma(n-1)}\Big) \\ 
	& \quad- \sum_{j=1}^n \frac{1}{n!} \sum_{\sigma \in \mathcal{P}(n)} \widetilde{\nabla}^{n-1}F_{\varphi} [0]\Big(\widetilde{\Gamma}_{\varphi} \big(\vec{X}_{\sigma(j)}, \vec{X}_{\sigma(n)}\big), \vec{X}_{\sigma(1)}, \ldots, \widehat{\vec{X}_{\sigma(j)}}, \ldots \vec{X}_{\sigma(n-1)}\Big);
\end{align*}
and then combining induction with Lemma \ref{lemma_1_cov_WF_equivalence} to get an empty wave front set for the terms on right hand side of the above equation. Another consequence of Lemma \ref{lemma_1_cov_WF_equivalence} is that microcausality of a functional does not depend on the ultralocal charts used to perform the derivatives. We immediately have the inclusion $\mathcal{F}_{reg}(B,\mathcal{U}) \subset \mathcal{F}_{\mu c}(B,\mathcal{U},g)$.

\begin{proposition}\label{prop_1_muloc_into_mucaus}
Let $\mathcal{U}\subset\Gamma^{\infty}(M\leftarrow B)$ be $CO$-open, then if $F\in \mathcal{F}_{\mu loc}(B,\mathcal{U})$, $\mathrm{WF}\big(F^{(k)}_{\varphi}[0]\big)$ is conormal to $ \triangle_k (M)$ \textit{i.e.} $\mathrm{WF}\left(F^{(k)}_{\varphi}[0]\right)\subset \lbrace(x,\ldots,x,\xi_1,\ldots,\xi_k): \xi_1+\ldots+\xi_k=0 \rbrace$ for all $k \geq 2$ and $\varphi \in \mathcal{U}$. Therefore $\mathcal{F}_{\mu loc}(B,\mathcal{U}) \subset \mathcal{F}_{\mu c}(B,\mathcal{U}, g)$. 
\end{proposition}

\begin{proof}
Since the wavefront set is a local property independent from the chart, we fix any $\mathcal{U}_{\varphi}$ and calculate it there. Note that the first derivative results in a smooth functional, then we take the $k$th derivative with $k \geq 2$. Going through the calculations, we get that $F^{(k)}_{\varphi}[0]\big(\vec{X}_1,\ldots,\vec{X}_k\big)$ defines an integral kernel of the form 
$$
	\int_{M^k} f^{(k)}_{\varphi}[0]_{i_1 \cdots i_k}(x_1)\delta(x_1,\ldots,x_k)\vec{X}_1^{i_1}(x_1)\cdots\vec{X}_k^{i_k}(x_k)d\mu_g(x_1,\ldots,x_k)\ .
$$
where $f^{(k)}_{\varphi}[0]_{i_1 \cdots i_k}$ is some smooth function for each indices $i_1,\ldots, i_k$. The calculation of the wavefront of such an integral kernel is equivalent to the calculation of the wave front of the diagonal delta, resulting in a subset of the conormal bundle to the diagonal map image. Therefore 
$$
	\mathrm{WF}(F^{(k)}_{\varphi}[0])= N^{*} \triangle_k (M)=\bigg\{(x_1,\ldots,x_k,\xi_1,\ldots,\xi_k)\in T^{*}M^k\backslash 0 : \ x_1=\cdots=x_k; \  \sum_{j=1}^k \xi_j=0\bigg\}	\ .
$$
In addition if $(x,\ldots,x,\xi_1,\ldots,\xi_k)$ is in $\mathrm{WF}\big(F^{(k)}_{\varphi}[0]\big)$ and has, say, the first $k-1$ covectors in $\overline{V}_{k-1}^{+}(x,\ldots,x)$, by conormality $\xi_k=-(\xi_1+\ldots+\xi_{k-1})$ and we see that $\xi_k \in \overline{V}^{-}(x)$, whence microlocality implies microcausality.
\end{proof}

\begin{theorem}\label{thm_1_mucaus_1}
Let $\mathcal{U}\subset\Gamma^{\infty}(M\leftarrow B)$ be $CO$-open and $\mathcal{L}$ a generalized microlocal Lagrangian with normally hyperbolic linearized equations. Then the Peierls bracket associated to $\mathcal{L}$ extends to $\mathcal{F}_{\mu c}(B,\mathcal{U}, g)$, has the same support property of Proposition \ref{prop_1_Peierls_1} and depends only locally on $\mathcal{L}$, that is, for all $F$, $G\in \mathcal{F}_{\mu c}(B,\mathcal{U}, g) $, $\lbrace F,G \rbrace_{\mathcal{L}}$ is unaffected by perturbations of $\mathcal{L}$ outside the right hand side of \eqref{eq_1_supp_peierls}. The same locality property holds for the retarded and advanced products.
\end{theorem}

\begin{proof}
Clearly $\lbrace F,G \rbrace_{\mathcal{L}}$ is well defined, in fact since $\mathrm{WF}(G^{(1)}_{\varphi}[0])$ is spacelike, and $\mathcal{G}_{\varphi} $, according to Theorem \ref{thm_1_properties_of_Green_functions}, propagates only lightlike singularities along lightlike geodesics, then $\mathcal{G}_{\varphi} dG_{\varphi}[0]$ must be smooth, giving a well defined pairing. As for support properties the proof can be carried on analogously to the proof of Proposition \ref{prop_1_Peierls_1}.

We now study the local behavior of the bracket. Suppose $\mathcal{L}_1$ and $\mathcal{L}_2$ are generalized Lagrangians, such that for some fixed $\varphi \in \mathcal{U}$, $\delta^{(1)}E(\mathcal{L}_1)_{\varphi}[0]$ and $\delta^{(1)}E(\mathcal{L}_2){\varphi}[0]$ differ only in a region outside 
\begin{equation}\label{eq_1_supp_peierls_2}
\mathcal{O}\doteq \left( J^+(\mathrm{supp}(F)) \cup J^-(\mathrm{supp}(F)) \right) \cap \left( J^+(\mathrm{supp}(G)) \cup J^-(\mathrm{supp}(G)) \right).
\end{equation}
By the support properties of retarded and advanced propagators of Proposition \ref{prop_1_Peierls_1} we have 
$$
\left\langle dF_{\varphi}[0], (\mathcal{G}^+_{\varphi,\mathcal{L}_1}-\mathcal{G}^+_{\varphi,\mathcal{L}_2})dG_{\varphi}[0]\right\rangle=0,
$$
as well as
$$
\left\langle dF_{\varphi}[0], (-\mathcal{G}^-_{\varphi,\mathcal{L}_1}+\mathcal{G}^-_{\varphi,\mathcal{L}_2})dG_{\varphi}[0]\right\rangle=0.
$$
Taking the sum of the two we find 
$$
	\lbrace F,G \rbrace_{\mathcal{L}_1}-\lbrace F,G \rbrace_{\mathcal{L}_2}=0.
$$
\end{proof}

\begin{theorem}\label{thm_1_peierls_closedness}
Let $\mathcal{U} \subset \Gamma^{\infty}(M\leftarrow B)$ $CO$-open and $\mathcal{L}$ a generalized Lagrangian. If $F$, $G\in \mathcal{F}_{\mu c}(B,\mathcal{U},g) $ we have that $\lbrace F,G \rbrace_{\mathcal{L}}\in \mathcal{F}_{\mu c}(B, \mathcal{U},g) $ as well.
\end{theorem}

\begin{proof}
%Suppose that we can take $\delta^{k}\mathcal{L}(1)[\varphi]$ with compact support, if not by the above theorem we restrict the operator by subtracting a suitable piece, so as to make the operator supported in the set $\mathcal{O}$ defined in $2.26$. This is achieved by replacing $E(\mathcal{L})[\varphi]$ with a cutoff version $$E'(\mathcal{L})[\varphi_0](\varphi)+f\left(E(\mathcal{L})[\varphi]-E'(\mathcal{L})[\varphi_0](\varphi)\right)$$ where $f\equiv 1$ in a neighborhood of $\mathcal{O}$. Then no modifications of $\lbrace F,G \rbrace_{\mathcal{L}}$ occur for all $\varphi_0$, $\varphi \in \mathcal{U}_{\varphi_0}$. 
By Faà di Bruno's formula,
\begin{equation}\label{eq_1_k-th_derivative_Peierls}
\begin{aligned}
	&d^k\left\lbrace F,G \right\rbrace_{\mathcal{L} ,\varphi}[0](\vec{X}_{1},\ldots, \vec{X}_k)\\
	&\qquad=\sum_{(J_1,J_2,J_3)\subset P_k} \left\langle F^{(|J_1|+1)}_{\varphi }[0](\otimes_{j_1 \in J_1}\vec{X}_{j_1}) \right.,  \left.d^{(|J_2|)}\mathcal{G}_{\varphi}(\otimes_{j_2 \in J_2}\vec{X}_{j_2}) G^{(|J_3|+1)}_{\varphi}[0] (\otimes_{j_3 \in J_3}\vec{X}_{j_3}) \right\rangle.
\end{aligned}
\end{equation}
while by \eqref{eq_1_der_Gvar},
\begin{align}\label{eq_1_k-th_derivative_G_var}
\begin{aligned}
	d^{|J_2|}\mathcal{G}_{\varphi} (\vec{X}_1,\ldots, \vec{X}_k) =\sum_{l=1}^{k}(-1)^l & \sum_{\substack{(I_1,\ldots,I_l) \\ \in \mathcal{P}(J_2)} } \sum_{p=0}^{l}\bigg( \bigcirc_{i=1}^{p} \mathcal{G}^-_{\varphi}\circ \delta^{(|I_i| +1)}E(\mathcal{L})\big(\vec{X}_{I_i}\big) \bigg) \\ 
    &\circ \mathcal{G}_{\varphi}\circ \bigg( \bigcirc_{i=p+1}^{l} \delta^{(|I_i| +1)}E(\mathcal{L})_{\varphi}\big(\vec{X}_{I_i}\big)\circ \mathcal{G}^+_{\varphi}  \bigg),
\end{aligned}
\end{align}
where $\bigcirc_{i=1}^p$ stands for composition of mappings indexed by $i$ from $1$ to $p$. For the rest of the proof, we will use the integral notation we used in \eqref{eq_1_kernel_notation} and in the proof of Proposition \ref{prop_1_Peierls_1}. Recall that, by \eqref{eq_1_k-th_derivative_gen_Lag}, the mapping $\delta^{(n)}E(\mathcal{L})_{\varphi}[0]$ has associated a compactly supported integral kernel $\mathcal{L}(1)^{(n+1)}_{\varphi}[0](x,z_{1},\ldots, z_{n})$ and its wave front is in $N^{*}\triangle_{n+1}(M)$ by Proposition \ref{prop_1_muloc_into_mucaus}. Then again, we have two general cases:\\
$1)$ $J_2=\emptyset$.\\
Then, letting $|J_1|=p$, $|J_3|=q=k-p$, the typical term has the form
\begin{equation}\label{eq_1_kernel_k-th_derivative_Paierl_1}
\begin{aligned}
	d^k\left\lbrace F,G \right\rbrace_{\varphi}[0](z_1,\ldots,z_k)=\int_{M^{2}} f^{(p+1)}_{\varphi}[0]_i(x,z_1,\ldots,z_{p}) , \mathcal{G}_{\varphi}^{ij}(x,y) g^{(q+1)}_{\varphi}[0]_j(y,z_{p+1},\ldots,z_k)d\mu_g(x,y).
\end{aligned}
\end{equation}
Suppose by contradiction that there is some $(x_1,\ldots,x_k,\xi_1,\ldots,\xi_k)\in \mathrm{WF}( \left\lbrace F,G \right\rbrace_{\varphi}[0])$ has $(\xi_1,\ldots,\xi_{k})\in \overline{V}^+_{k}(x_1,\ldots,x_k)$ (the argument works similarly for $(\xi_1,\ldots,\xi_{k})\in \overline{V}^-_{k}(x_1,\ldots,x_k)$). Using twice Theorem 8.2.14 in \cite{hormanderI} in the above pairing yields
$$
\begin{aligned}
	&\mathrm{WF}\Big(\big\lbrace F,G \big\rbrace^{(k)}_{\varphi}[0]\Big)\\ &\qquad\subseteq  \Big\lbrace (z_1,\ldots,z_k,\xi_1,\ldots,\xi_k)  :\  \exists (y,\eta)\in T^{*}M 	(x,z_1,\ldots,z_p,-\eta,\xi_1,\ldots,\xi_p) \in \mathrm{WF}(F^{(p+1)}_{\varphi}[0]_i), \\
	&\qquad\qquad (x,z_{p+1},\ldots,z_{k},\eta,\xi_{p+1},\ldots,\xi_k) \in \mathrm{WF}\big(\mathcal{G}_{\varphi}^{ij}G^{(q+1)}_{\varphi}[0]_j\big) \Big\rbrace \\ 
	 &\qquad\subset \Big\lbrace (z_1,\ldots,z_k,\xi_1,\ldots,\xi_k):  \ \exists (x,\eta),(y,\zeta)\in T^*M : (x,z_1,\ldots,z_p,-\eta,\xi_1,\ldots, \xi_p) \in \mathrm{WF}\big(F^{(p+1)}_{\varphi}[0]_i\big) \\ 
	&\qquad\qquad (x,y,\eta,-\zeta)\in \mathrm{WF}(\mathcal{G}_{\varphi}^{ij}), \ (y,z_{p+1},\ldots,z_{k},\zeta,\xi_{p+1},\ldots, \xi_k) \in \mathrm{WF}(G^{(q+1)}_{\varphi}[0]_j) \Big\rbrace.
\end{aligned}
$$
So if $(z_1,\ldots,z_k,\xi_1,\ldots,\xi_k)\in \mathrm{WF}\big(\left\lbrace F,G \right\rbrace^{(k)}_{\varphi}\big)$, then $\exists$ $(x,\eta), (y,\zeta) \in T^{*}M$ such that 
$$
	\left\lbrace \begin{array}{ll}
		(x,z_1,\ldots,z_p,-\eta,\xi_1,\ldots,\xi_p)&  \in \mathrm{WF}(F^{(p+1)}_{\varphi}[0]_i))\\
		(x,y,\eta,-\zeta)  &\in \mathrm{WF}(\mathcal{G}_{\varphi}^{ij})\\
		(y,z_{p+1},\ldots,z_k,\zeta,\xi_{p+1},\ldots,\xi_k) &\in \mathrm{WF}(G^{(q+1)}_{\varphi}[0]_j)\ .
	\end{array} \right. 
$$
Now by Theorem \ref{thm_1_properties_of_Green_functions}, $\mathrm{WF}(\mathcal{G}_{\varphi}^{ij})$ contains pairs of lightlike covectors with opposite time orientation, therefore in case $\eta\in \overline{V}^{+}(x)$ (resp. $\eta\in \overline{V}^{-}(x)$), then $\zeta\in \overline{V}^{+}(y)$ (resp. $\zeta\in \overline{V}^{-}(y)$) in which case $	\mathrm{WF}\big(G^{(q+1)}_{\varphi}[0]_j\big)$ (resp. $\mathrm{WF}\big(F^{(p+1)}_{\varphi}[0]_i\big)$) does violate the microcausality condition of Definition \ref{def_1_WF_mucaus}.\\
$2)$ $J_2 \neq \emptyset$. \\
Again let $|J_1|=p$, $|J_3|=k-q$, set also, referring to \eqref{eq_1_k-th_derivative_G_var}, $|I_j|=k_j$ for $j=1,\ldots,l$ so that $|J_2|=k_1+\cdots+k_l$. Combining \eqref{eq_1_k-th_derivative_Peierls} with \eqref{eq_1_k-th_derivative_G_var} with the integral kernel notation we get
\begin{equation}\label{eq_1_kernel_k-th_derivative_Paierl_2}
\begin{aligned}
	 \{ F,G \}^{(k)}_{\varphi}[0](z_1,\ldots,z_k)&=\int_{M^{k}} f^{(p+1)}_{\varphi}[0]_i(x,z_1,\ldots,z_p) \mathcal{G}^{- \space\ i j_1}_{\varphi}(x,x_1) d^{(k_1+2)} \mathcal{L}_{\varphi}[0]_{j_1 i_1} (x_1,y_1,z_{I_1})\\ & \quad
	 \mathcal{G}^{- \space\ i_1 j_2}_{\varphi}(y_1,x_2)   \cdots  \mathcal{G}^{- \space\ i_{m-1} j_{m}}_{\varphi}(y_{m-1},x_{m}) d^{(k_m+2)}\mathcal{L}_{\varphi}[0]_{j_{m} i_{m}} (x_{m},y_{m},z_{I_m})\\ &  \quad
	 \mathcal{G}_{\varphi}^{i_{m} j_{m+1}}(y_{m},x_{m+1}) 
	 d^{(k_{m+1}+2)} \mathcal{L}_{\varphi}[0]_{j_{m+1} i_{m+1}} (x_{m+1},y_{m+1},z_{I_{m+1}})  \\ & \quad
	 \mathcal{G}^{+ \ i_{m+1} j_{m+2}}_{\varphi}(y_{m+1},z_{m+2})  \ldots  d^{(k_l+2)} \mathcal{L}_{\varphi}[0]_{j_{l} i_{l}} (x_{l},y_{l},z_{I_l},)\mathcal{G}^{+\ i_{l} j}(y_{l} ,y) \\
	 & g^{(k-q+1)}_{\varphi}[0]_j(y,z_{q+1},\ldots,z_k)d\mu_g(x,x_1,y_1,\ldots,x_{l},y_{l},y).
\end{aligned}
\end{equation}
Combining Theorem 8.2.14 in \cite{hormanderI}, Theorem \ref{thm_1_properties_of_Green_functions} and Proposition \ref{prop_1_muloc_into_mucaus} we can estimate the wave front set of the integral kernel of $\{ F,G \}^{(k)}_{\varphi}[0]$ as all elements $(z_1,\ldots,z_k,\xi_1,\ldots,\xi_k)\in T^{*}M^k$ for which there are $(x,\eta)$, $(x_1,\eta_1)$, $\ldots,(x_l\eta_l)$, $(y_1,\zeta_1)$, $\ldots , (y_l,\zeta_l)$ $(y,\zeta) \in T^{*}M $ such that 
$$
	\left\lbrace \begin{array}{ll}
		(x,z_1,\ldots,z_p,-\eta,\xi_1,\ldots,\xi_p) & \in \mathrm{WF}\big(f^{(p+1)}_{\varphi}[0]_i\big),\\
		(x,x_{1},\eta,-\eta_{1}) &\in \mathrm{WF}(\mathcal{G}^{-\ ij_1}_{\varphi}),\\
		(x_{1},y_{1},z_{|I_1|},\eta_{1},-\zeta_{1},\xi_{I_1}) &\in \mathrm{WF}\big(d^{(k_{1}+2)}\mathcal{L}_{\varphi}[0]_{i_1j_1}\big),\\
		\vdots & \vdots\\
		(y_{m-1},x_{m},\zeta_{m-1},-\eta_{m}) &\in \mathrm{WF}(\mathcal{G}^{-\ i_{m-1}j_m}_{\varphi}),\\
		(x_{m},y_{m},z_{I_m},\eta_{m},-\zeta_{m},\xi_{I_m}) &\in \mathrm{WF}\big(d^{(k_{m}+2)}\mathcal{L}_{\varphi}[0]_{i_mj_m}\big),\\
		(y_{m},x_{m+1},\zeta_{m},-\eta_{m+1})  &\in \mathrm{WF}(\mathcal{G}_{\varphi}^{i_mj_{m+1}}),\\
		(x_{m+1},y_{m+1},z_{I_{m+1}},\eta_{m+1},-\zeta_{m+1},\xi_{I_{m+1}}) &\in \mathrm{WF}\big(d^{(k_{m+1}+2)}\mathcal{L}_{\varphi}[0]_{i_{m+1}j_{m+1}}\big),\\
		(y_{m+1},x_{m+2},\zeta_{m+1},-\eta_{m+2}) &\in \mathrm{WF}(\mathcal{G}^{+ \ i_{m+1}j_{m+2}}_{\varphi}),\\
		\vdots & \vdots\\
		(y_{l-1},x_{l},\zeta_{l-1},-\eta_{l}) &\in \mathrm{WF}(\mathcal{G}^{+\ i_{l-1}j_l}_{\varphi}),\\
		(x_{l},y_{l},z_{I_l},\eta_{l},-\zeta_{l},\xi_{I_l}) &\in \mathrm{WF}\big(d^{(k_{l}+2)}\mathcal{L}_{\varphi}[0]_{i_lj_l}\big),\\
		(y_{l},y,\zeta_{l},-\zeta) &\in \mathrm{WF}(\mathcal{G}^{+\ i_lj}_{\varphi}),\\
		(y,z_{k-q+1},\ldots,z_k,\zeta,\xi_{k-q+1},\ldots,\xi_k) & \in \mathrm{WF}\big(g^{(k-q+1)}_{\varphi}[0]_j\big).
	\end{array} \right. 
$$
Suppose by contradiction, as above, that $(z_1,\ldots,z_{k},\xi_1,\ldots,\xi_k)\in \mathrm{WF}\big( \lbrace F,G \rbrace_{\varphi}^{(k)}\big)$ has $(\xi_1,\ldots,\xi_k)\in \overline{V}^+_k(z_1,\ldots,z_k)$ $\big($resp. $(\xi_1,\ldots,\xi_k)\in \overline{V}^-_{k}(z_1,\ldots,z_k)\big)$. Then $\zeta_m$ and $\eta_{m+1}$ are both either lightlike future directed, or lightlike past directed. In the first case, propagation of singularities implies that $\zeta$ is lightlike future directed, contradicting microcausality of $G^{(k-q+1)}_{\varphi}[0]$ (resp. $\zeta$ is lightlike past directed, contradicting the microlocality of $G^{(k-q+1)}_{\varphi}[0]$); in the second case, propagation of singularities implies that $\eta$ is lightlike past directed, contradicting microcausality of $F^{(p+1)}_{\varphi}[0]$ (resp. $\eta$ is lightlike future directed, contradicting the microlocality of $F^{(p+1)}_{\varphi}[0]$). We remark that in the wave front set of $\{F,G\}_{\varphi}^{(k)}[0]$ is the (finite) union under all possible choices of indices for all wave front sets of the form \eqref{eq_1_kernel_k-th_derivative_Paierl_1} or \eqref{eq_1_kernel_k-th_derivative_Paierl_2}, each of which is however microcausal, implying that their finite union will be microcausal as well.
\end{proof}

%An important important feature of this proof is the presence of the causal propagator $\mathcal{G}$ at some point in the composition pattern of the integral kernels. This implies that its wavefront covectors lacks the diagonal part of $\mathrm{WF}(\mathcal{G}^{\pm})$ which in some sense is responsible for a "stationary" propagation of singularities, therefore resulting in the appearance of spacelike covectors which will generate no final contradiction in the argument. Also note that spacetime compactness of derivatives of $E(\mathcal{L})$ enable us to compose operators when the wave front set analysis is carried out.

\begin{theorem}\label{thm_1_jacobi}
The mapping $(F,G) \mapsto \lbrace F,G \rbrace_{\mathcal{L}}$ defines a Lie bracket on $\mathcal{F}_{\mu c}(B,\mathcal{U},g) $, for all $\varphi \in \mathcal{U}$ $CO$-open. %In particular $\left(\mathcal{F}_{\mu caus}(B,\mathcal{U},g), \lbrace \cdot,\cdot \rbrace_{\mathcal{L}}\right)$ is a Poisson $*$-algebra.
\end{theorem}

\begin{proof}
Bilinearity and antisymmetry are clear from Definition \ref{def_1_Peierls}, while Theorem \ref{thm_1_peierls_closedness} ensures the closure of the bracket operation. We are thus left with the Jacobi identity:
$$
\lbrace F,\lbrace G,H \rbrace_{\mathcal{L}} \rbrace_{\mathcal{L}}+\lbrace G,\lbrace H,F \rbrace_{\mathcal{L}} \rbrace_{\mathcal{L}}+\lbrace H,\lbrace F,G \rbrace_{\mathcal{L}} \rbrace_{\mathcal{L}}=0.
$$
Using the integral kernel notation as in the above proof, we have 
$$
\begin{aligned}
	& \lbrace F,\lbrace G,H \rbrace_{\mathcal{L}} \rbrace_{\mathcal{L}}(\varphi) = \int_{M^{2}}  f^{(1)}_{\varphi}[0]_i(x)\mathcal{G}_{\varphi}^{ij}(x,y)\lbrace G,H \rbrace^{(1)}_{\mathcal{L} \space\ \varphi}[0]_j(y)d\mu_g(x,y) \\ &=
	\int_{M^{4}}f^{(1)}_{\varphi}[0]_i(x)\mathcal{G}_{\varphi}^{ij}(x,y)
	\Big( g^{(2)}_{\varphi}[0]_{jk}(y,z) \mathcal{G}_{\varphi}^{kl}(z,w) h^{(1)}_{\varphi}[0]_l(w) \\
 &\qquad + g^{(1)}_{\varphi}[0]_k(z) \mathcal{G}^{kl}_{\varphi}(z,w) h^{(2)}_{\varphi}[0]_{jl}(y,w) \Big) d\mu_g(x,y,z,w)   \\ & 
	-\int_{M^{6}} f^{(1)}_{\varphi}[0]_i(x)\mathcal{G}_{\varphi}^{ij}(x,y) \Big(d^{(3)}\mathcal{L}_{\varphi}[0]_{jj_1i_1}(y,y_1,x_1,) \mathcal{G}_{\varphi}^{- \space\ k j_1}(z,y_{1}) g^{(1)}_{\varphi}[0]_{k}(z) \mathcal{G}_{\varphi}^{i_1 l}(x_{1},w)h^{(1)}_{\varphi}[0]_{l}(w) \\ & \ \ \ \ \ \ \ \ \  
	+d^{(3)}\mathcal{L}_{\varphi}[0]_{j j_{1} i_1}(y_1,x_1,y) \mathcal{G}_{\varphi}^{k j_1}(z,y_1) g^{(1)}_{\varphi}[0]_{k}(z) \mathcal{G}^{+ \space\ i_1 l}_{\varphi}(x_1,w)h^{(1)}_{\varphi}[0]_l(w) \Big)d\mu_g(x,y,z,w,x_1,y_1).
\end{aligned}
$$
Summing over cyclic permutations of the first two terms yields 
$$
\begin{aligned}
	\int_{M^{4}} & \Big( \textcolor{blue}{ f^{(1)}_{\varphi}[0]_i(x)\mathcal{G}_{\varphi}^{ij}(x,y) g^{(2)}_{\varphi}[0]_{jk}(y,z) \mathcal{G}_{\varphi}^{kl}(z,w) h^{(1)}_{\varphi}[0]_l(w)}   \\ & 
	 +\textcolor{red}{f^{(1)}_{\varphi}[0]_i(x)\mathcal{G}_{\varphi}^{ij}(x,y)g^{(1)}_{\varphi}[0]_k(z) \mathcal{G}^{kl}_{\varphi}(z,w) h^{(2)}_{\varphi}[0]_{jl}(y,w)}  \\ & 
	 +\textcolor{red}{g^{(1)}_{\varphi}[0]_i(x)\mathcal{G}_{\varphi}^{ij}(x,y) h^{(2)}_{\varphi}[0]_{jk}(y,z) \mathcal{G}_{\varphi}^{kl}(z,w) f^{(1)}_{\varphi}[0]_l(w) } \\ & 
	 +\textcolor{green}{g^{(1)}_{\varphi}[0]_i(x)\mathcal{G}_{\varphi}^{ij}(x,y)h^{(1)}_{\varphi}[0]_k(z) \mathcal{G}^{kl}_{\varphi}(z,w) f^{(2)}{\varphi}[0]_{jl}(y,w)}  \\ &
	 +\textcolor{green}{h^{(1)}_{\varphi}[0]_i(x)\mathcal{G}_{\varphi}^{ij}(x,y) f^{(2)}{\varphi}[0]_{jk}(y,z) \mathcal{G}_{\varphi}^{kl}(z,w) g^{(1)}_{\varphi}[0]_l(w) } \\ & 
	 +\textcolor{blue}{h^{(1)}_{\varphi}[0]_i(x)\mathcal{G}_{\varphi}^{ij}(x,y)f^{(1)}_{\varphi}[0]_k(z) \mathcal{G}^{kl}_{\varphi}(z,w) g^{(2)}_{\varphi}[0]_{jl}(y,w)}\Big)d\mu_g(x,y,z,w)\\
	 &=0,
\end{aligned}
$$
while for the other two,
$$
\begin{aligned}
& \int_{M^{6}}d\mu_g(x,y,z,w,x_1,y_1)\\
	&\Big(\textcolor{blue}{f^{(1)}_{\varphi}[0]_i(x)\mathcal{G}_{\varphi}^{ij}(x,y) d^{(3)}\mathcal{L}_{\varphi}[0]_{jj_1x_1}(y,y_1,i_1,) \mathcal{G}_{\varphi}^{- \ k j_1}(z,y_{1}) g^{(1)}_{\varphi}[0]_{k}(z) \mathcal{G}_{\varphi}^{i_1 l}(x_{1},w)h^{(1)}_{\varphi}[0]_{l}(w)} \\ & 
	+\textcolor{red}{f^{(1)}_{\varphi}[0]_i(x)\mathcal{G}_{\varphi}^{ij}(x,y) d^{(3)}\mathcal{L}_{\varphi}[0]_{j j_{1} i_1}(y,y_1,x_1) \mathcal{G}_{\varphi}^{k j_1}(z,y_1) g^{(1)}_{\varphi}[0]_{k}(z) \mathcal{G}^{+ \space\ i_1 l}_{\varphi}(x_1,w)h^{(1)}_{\varphi}[0]_l(w)} \\ & 
	+\textcolor{red}{g^{(1)}_{\varphi}[0]_i(x)\mathcal{G}_{\varphi}^{ij}(x,y) d^{(3)}\mathcal{L}_{\varphi}[0]_{jj_1i_1}(y,y_1,x_1,) \mathcal{G}_{\varphi}^{- \ k j_1}(z,y_{1}) h^{(1)}_{\varphi}[0]_{k}(z) \mathcal{G}_{\varphi}^{i_1 l}(x_{1},w)f^{(1)}_{\varphi}[0]_{l}(w)} \\ & 
	+\textcolor{green}{g^{(1)}_{\varphi}[0]_i(x)\mathcal{G}_{\varphi}^{ij}(x,y) d^{(3)}\mathcal{L}_{\varphi}[0]_{j j_{1} i_1}(y,y_1,x_1) \mathcal{G}_{\varphi}^{k j_1}(z,y_1) h^{(1)}_{\varphi}[0]_{k}(z) \mathcal{G}^{+ \ i_1 l}_{\varphi}(x_1,w)f^{(1)}_{\varphi}[0]_l(w) }\\ & 
	+\textcolor{green}{h^{(1)}_{\varphi}[0]_i(x)\mathcal{G}_{\varphi}^{ij}(x,y) d^{(3)}\mathcal{L}_{\varphi}[0]_{jj_1i_1}(y,y_1,x_1,) \mathcal{G}_{\varphi}^{- \ k j_1}(z,y_{1}) f^{(1)}_{\varphi}[0]_{k}(z) \mathcal{G}_{\varphi}^{i_1 l}(x_{1},w)g^{(1)}_{\varphi}[0]_{l}(w)} \\ & 
	+\textcolor{blue}{h^{(1)}_{\varphi}[0]_i(x)\mathcal{G}_{\varphi}^{ij}(x,y) d^{(3)}\mathcal{L}_{\varphi}[0]_{j j_{1} i_1}(y,y_1,x_1) \mathcal{G}_{\varphi}^{k j_1}(z,y_1) f^{(1)}_{\varphi}[0]_{k}(z) \mathcal{G}^{+ \ i_1 l}_{\varphi}(x_1,w)g^{(1)}_{\varphi}[0]_l(w)} \Big)\\
	&=0.
\end{aligned}
$$
To make the simplifications we used the antisymmetry of the integral kernel $\mathcal{G}_{\varphi}(x,y)$, the adjoint relation between the propagators $\mathcal{G}^+(x,y)=\mathcal{G}^-(y,x)$ (see Lemma \ref{lemma_1_duality_green}) and $\mathcal{G}_{\varphi}^{ij}=\mathcal{G}_{\varphi}^{ji}$.
\end{proof}

\section{Structure of the space of microcausal functionals}\label{section_properties_of_muc_functionals}

The first point of emphasis is to give a topology to $\mathcal{F}_{\mu c}(B,\mathcal{U},g)$. We shall proceed step by step refining our starting definitions to better grasp the reasoning behind the choice of topology we will be giving $\mathcal{F}_{\mu c}(B,\mathcal{U},g)$.

The simplest guess, as well as the weakest, on $\mathcal{F}_{\mu c}(B,\mathcal{U},g)$ is the locally convex topology that corresponds to the initial topology induced by the mappings
$$
 F \rightarrow F(\varphi) \in \mathbb{R}.
$$
To account for smooth functionals we try the initial topology with respect to mappings
\begin{align*}
	&F \rightarrow F(\varphi) \in \mathbb{R},\\
	&F \rightarrow \nabla^{k}F_{\varphi}[0] \in \Gamma^{-\infty}_c\left(M^k\leftarrow\boxtimes^k \left(\varphi^{*}VB\right)\right).
\end{align*}
This time we are leaving out all information on the wave front set which plays a role in defining microcausal functionals. To remedy we would like to set up the H\"ormander topology on the spaces $\Gamma^{-\infty}_{ \Upsilon_{k,g}}\left(M^k\leftarrow\boxtimes^k \left(\varphi^{*}VB\right)\right)$, however this is not immediately possible since $\Upsilon_{k,g}$ are open cones, and the H\"ormander topology is given to closed ones, therefore we need the following result, whose proof can be found in Lemma 4.1 in \cite{acftstructure},

\begin{lemma}\label{lemma_1_cones}
Given the open cone $\Upsilon_k(g)$ it is always possible to find a sequence of closed cones $\lbrace \mathcal{V}_{m}(k) \subset T^{*}M^k \rbrace_{m \in \mathbb{N}}$ such that $\mathcal{V}_{m}(k) \subset \mathrm{Int}(\mathcal{V}_{m+1}(k) )$ and $\cup_{m \in \mathbb{N}}\mathcal{V}_{m}(k)= \Upsilon_k(g)$ for all $k \geq 1$. 
\end{lemma}

Then we can write
\begin{align}
	\Gamma^{-\infty}_{c\ \Upsilon_k(g)}\left(M^k\leftarrow\boxtimes_k \varphi^{*}VB\right)  = \underrightarrow{\lim_{m\in \mathbb{N}}} \Gamma^{-\infty}_{c\ \mathcal{V}_{m}(k)}\left(M^k\leftarrow\boxtimes_k \varphi^{*}VB\right).
\end{align}
By construction of the direct limit we have mappings 
$$
	\Gamma^{-\infty}_{c\ \mathcal{V}_{m}(k)}\left(M^k\leftarrow\boxtimes^k \left(\varphi^{*}VB\right)\right)  \rightarrow \Gamma^{-\infty}_{c\ \Upsilon_k(g)}\left(M^k\leftarrow\boxtimes^k \left(\varphi^{*}VB\right)\right)
$$
where the source space has $\mathcal{V}_{m}(k)\subset T^{*}M^k$ as a closed cone, so it can be given the H\"ormander topology. In particular when we are dealing with standard compactly supported distributions its topology can be defined, see the remark after Theorem 18.1.28 in \cite{hormanderIII}, to be the initial topology with respect to the mappings
\begin{align*}
	&F \rightarrow F(\varphi) \in \mathbb{C},\\
	&F \rightarrow P F \in \Gamma^{\infty}_{c}\left(M^k\leftarrow \boxtimes^k \left(\varphi^{*}VB\right)\right)
\end{align*}  
where $\varphi$ is any smooth section of $B$ and $P$ any properly supported pseudo-differential operator of order zero on the vector bundle $\boxtimes^k \left(\varphi^{*}VB\right) \rightarrow M^k$ such that $\mathrm{WF}(P) \cap \mathcal{V}_m(k)=\emptyset$. Using that Definition 18.1.32 \cite{hormanderIII}, Theorem 18.1.16 \cite{hormanderIII} and Theorem 8.2.13 in \cite{hormanderI} can be generalized to the vector bundle case provided we use the notion \eqref{eq_1_def_WF_sections} for the wave front set of vector valued distributions; we can argue as in Corollary 4.1 of \cite{acftstructure} that each $\Gamma^{-\infty}_{c\ \mathcal{V}_{m}(k)}\left(M^k\leftarrow\boxtimes^k\left(\varphi^{*}VB\right)\right)$ becomes a Hausdorff topological space.
By Theorem \ref{thm_1_Gamma_c_TVS} since the base manifold $M$ is separable and the fibers are finite dimensional vector spaces, hence nuclear$\footnote{We say that a Hausdorff locally convex space E is nuclear if given any other locally convex space F we have $E\otimes_{\pi}F \simeq E\otimes_{\epsilon}F$, where the two are the tensor product space respectively endowed with the quotient topology and with the canonical topology associated to the space of continuous bilinear mappings $:E'_{\sigma}\times F'_{\sigma}\rightarrow \mathbb{R}$ equipped with the topology of uniform convergence on products of equicontinuous subsets of $E'$ and $F'$. For more detail see either \cite{nlcs} or \cite{treves2016topological}.}$ and Fr\'echet, we find that $\Gamma^{\infty}_{c}\left(M^k\leftarrow \boxtimes^k \left(\varphi^{*}VB\right)\right) $ is a nuclear limit-Fréchet space, it is Hausdorff, thus each
$$
	\Gamma^{-\infty}_{c\ \mathcal{V}_{m}(k)}\left(M^k\leftarrow\boxtimes^k \left(\varphi^{*}VB\right)\right)
$$ 
is nuclear as well. Finally by Porposition 50.1 pp. 514 in \cite{treves2016topological} the direct limit topology on 
$$
    \Gamma^{-\infty}_{c\ \Upsilon_{k,g}} \left(M^k\leftarrow\boxtimes^k(\varphi^{*}VB)\right)
$$ 
is nuclear for all $k$ (and also Hausdorff). We have therefore proved:

\begin{lemma}\label{lemma_1_mucaus_top_1}
The direct limit topology on $\Gamma^{-\infty}_{c\ \Upsilon_{k,g}}\left(M^k\leftarrow\boxtimes^k \left(\varphi^{*}VB\right)\right)$ induced as a direct limit topology of the spaces $\Gamma^{-\infty}_{c\ \mathcal{V}_{m}(k)}\left(M^k\leftarrow\boxtimes^k \left(\varphi^{*}VB\right)\right)$ with the H\"ormander topology is a Hausdorff nuclear space.
\end{lemma}

Finally we can induce on $\mathcal{F}_{\mu c}(B,\mathcal{U},g)$ a topology by 

\begin{theorem}\label{thm_1_mucaus_top}
Given the set $\mathcal{F}_{\mu c}(B,\mathcal{U},g)$, consider the mappings 
\begin{equation}\label{eq_1_point_seminorm}
	\mathcal{F}_{\mu c}(B,\mathcal{U},g) \ni F \mapsto F(\varphi) \in \mathbb{R},
\end{equation}
\begin{equation}\label{eq_1_hormander_seminorm}
	\mathcal{F}_{\mu c}(B,\mathcal{U},g) \ni F  \mapsto \nabla^{k}F_{\varphi}[0] \in \Gamma^{-\infty}_{c\ \Upsilon_{k,g}}\left(M^k\leftarrow\boxtimes^k \left(\varphi^{*}VB\right)\right),
\end{equation}
and the related initial topology on $\mathcal{F}_{\mu c}(B,\mathcal{U},g)$. Then $\mathcal{F}_{\mu c}(B,\mathcal{U},g)$ is a nuclear locally convex topological space with a Poisson *-algebra with respect to the Peierls bracket of some microlocal generalized Lagrangian $\mathcal{L}$.
\end{theorem}

\begin{proof}
The nuclearity follows from the stability of nuclear spaces under projective limit topology, see Proposition 50.1 pp. 514 in \cite{treves2016topological}, so using the nuclearity of both $\Gamma^{-\infty}_{c\ \Upsilon_{k,g}}\left(M^k\leftarrow\boxtimes^k \left(\varphi^{*}VB\right)\right)$ (via Lemma \ref{lemma_1_mucaus_top_1}) and $\mathbb{R}$ (trivially) we have our claim. The Peierls bracket is well defined by Theorem \ref{thm_1_peierls_closedness} and satisfies the Jacobi identity due to Theorem \ref{thm_1_jacobi}, so we only have to show the Leibniz rule for the bracket, that is 
$$
	\lbrace F,GH \rbrace_{\mathcal{L}}= G\lbrace F,H \rbrace_{\mathcal{L}} + \lbrace F,G \rbrace_{\mathcal{L}}H.
$$
However this follows once we show that the product $F,G \mapsto F\cdot G$ with $(F \cdot G) (\varphi)= F(\varphi)G(\varphi)$ is closed in $\mathcal{F}_{\mu c}(B,\mathcal{U},g)$ and then use $d(F\cdot G)_{\varphi}[0]= dF_{\varphi}[0]G(\varphi)+ F(\varphi)dG_{\varphi}[0]$. The latter is a consequence of the definition of derivation (\textit{i.e.} the standard Leibniz rule), so we are left with showing the former: we compute 
$$
	d^k(F\cdot G)_{\varphi}[0](\vec{X}_1, \ldots,\vec{X}_k)= \sum_{\sigma \in \mathcal{P}(1,\ldots ,k)}\sum_{l=0}^k d^lF_{\varphi}[0](\vec{X}_{\sigma(1)}, \ldots,\vec{X}_{\sigma(l)}) d^{k-l}G_{\varphi}[0](\vec{X}_{\sigma(k-l+1)}, \ldots,\vec{X}_{\sigma(k)}),
$$
where $\mathcal{P}(1,\ldots ,k)$ is the set of permutations of $\lbrace 1,\ldots,k \rbrace$. For each of those terms using Theorem 8.2.9 in \cite{hormanderI} we have 
$$
\begin{aligned}
	\mathrm{WF}(F^{(l)}_{\varphi}[0]G^{(k-l)}_{\varphi}[0]) & \subset  \mathrm{WF}(F^{(l)}_{\varphi}[0]) \times \mathrm{WF}(G^{(k-l)}_{\varphi}[0]) \\ 
	&\quad  \bigcup \mathrm{WF}(G^{(k-l)}_{\varphi}[0]) \times \left( \mathrm{supp}(G^{(k-l)}_{\varphi}[0])\times \lbrace {0} \rbrace \right) \\ 
	&\quad \bigcup\left( \mathrm{supp}(F^{(l)}_{\varphi}[0])\times \lbrace {0} \rbrace \right)  \times  \mathrm{WF}(G^{(k-l)}_{\varphi}[0]),
\end{aligned}
$$
and therefore microcausality is met.
\end{proof}

Note that closed linear subspaces of $\mathcal{F}_{\mu c}(B,\mathcal{U},g)$ are nuclear as well (see Proposition 50.1 in \cite{treves2016topological}), so $\mathcal{F}_{\mu loc}(B,\mathcal{U},g)$ is a Hausdorff nuclear space. The space $\mathcal{F}_{\mu c}(B,\mathcal{U},g)$ can be given a structure of a $C^{\infty}$-ring, more precisely

\begin{proposition}\label{prop_1_C-infty_ring}
If $F_1,\ldots,F_n$ $\in \mathcal{F}_{\mu c}(B,\mathcal{U},g)$ and $\psi \in V \subset \mathbb{R}^n \rightarrow \mathbb{R}$ is smooth, then $\psi (F_1,\ldots,F_n) \in \mathcal{F}_{\mu c}(B,\mathcal{U},g)$ and 
$$
	\mathrm{supp}(\psi (F_1,\ldots,F_n)) \subset \bigcup_{i=1}^n \mathrm{supp}(F_i). 
$$
\end{proposition}

\begin{proof}
First we check the support properties. Suppose that $x \notin \cup_{i=1}^n \mathrm{supp}(F_i)$, we can find an open neighborhood $V$ of $x$ for which given any $\varphi \in \mathcal{U}$ and any $\vec{X} \in \Gamma^{\infty}_c(M\leftarrow \varphi^*VB)$ having $\mathrm{supp}(\vec{X}) \subset V$ implies $(F_i\circ u_{\varphi})(t\vec{X} )=(F_i\circ u_{\varphi})(0)$ for all $t$ in a suitable neighborhood of $0\in \mathbb{R}$. Then $ \psi\big((F_1\circ u_{\varphi})(t\vec{X}),\ldots,(F_n\circ u_{\varphi})(t\vec{X})\big)= \psi\big((F_1\circ u_{\varphi})(0),\ldots,(F_n\circ u_{\varphi})(0)\big)$ as well giving $x\notin \mathrm{supp}(\psi \circ (F_1,\ldots,F_n))$. We immediately see that the composition is Bastiani smooth, so consider its $k$th derivative expressed via Faà di Bruno's formula:
$$\small{
\begin{aligned}
	& d^k\psi(F_1,\ldots,F_1)_{\varphi}[0](\vec{X}_1,\ldots,\vec{X}_k)\\ & =  \sum_{\substack{(J_1,\ldots,J_n)\\ \in \mathcal{P}(1,\ldots,k)}} 
	\frac{\partial^k \psi (F_1(\varphi),\ldots,F_n(\varphi))}{\partial z^{J_1+\ldots+J_n}}\left( F^{(\vert J_1 \vert)}_{1 \space\ \varphi}[0] (\vec{X}_{j_{1,1}} ,\ldots ,\vec{X}_{j_{\vert J_1 \vert,1}})\cdot  \ldots \cdot    F^{(\vert J_n \vert)}_{n \space\ \varphi}[0](\vec{X}_{j_{1,n}},\ldots,\vec{X}_{j_{\vert J_n \vert,n}}) \right)  \\ 
\end{aligned}}
$$
where $ \mathcal{P}(1,\ldots,k)$ denotes the set of partitions of $\{1,\ldots , k\}$ and $z\in \mathbb R^n$. Since $\psi$ is smooth the only contribution to the wavefront set of the composition is the product of functional derivative in the above sum, for which Theorem 8.2.9 in \cite{hormanderI} gives 
$$
\begin{aligned}
	\mathrm{WF} &\left( F_1^{(\vert J_1 \vert)},\ldots, F_n^{(\vert J_n \vert)}\right)\\
  \subset & \quad  \mathrm{WF}\left( F_{1 \space\ \varphi}^{(\vert J_1 \vert)}\right) \times \ldots \times \mathrm{WF}\left(F_{n \space\ \varphi}^{(\vert J_n \vert)}\right)   \\ 
	&\quad\bigcup \mathrm{supp}\left(F_{1 \space\ \varphi}^{(\vert J_1 \vert)}\right)\times \lbrace \vec{0} \rbrace^{\vert J_1\vert} \times  \mathrm{WF}\left(F_{2 \space\ \varphi}^{(\vert J_2 \vert)}\right) \times \ldots \times \mathrm{WF}\left(F_{n \space\ \varphi}^{(\vert J_n \vert)}\right) \\ 
	&\quad \ldots \\ 
	&\quad\bigcup \mathrm{WF}\left(F_{1 \space\ \varphi}^{(\vert J_1 \vert)}\right) \times \ldots \times \mathrm{WF}\left(F_{n-1 \space\ \varphi}^{(\vert J_{n-1} \vert)}\right) \times \mathrm{supp}\left(F_{n \space\ \varphi}^{(\vert J_n \vert)}\right)\times \lbrace \vec{0} \rbrace^{\vert J_n\vert} \\
	&\quad \ldots \\ 
	& \quad\bigcup\mathrm{supp}\left(F_{1 \space\ \varphi}^{(\vert J_1 \vert)}\right)\times \lbrace \vec{0} \rbrace^{\vert J_1\vert} \times \ldots \times \mathrm{supp}\left(F_{n-1 \space\ \varphi}^{(\vert J_{n-1} \vert)}\right)\times \lbrace \vec{0} \rbrace^{\vert J_{n-1}\vert} \times \mathrm{WF}\left(F_{n \space\ \varphi}^{(\vert J_n \vert)}\right).
\end{aligned}
$$
We clearly have that if an element of $\mathrm{\mathrm{WF}} \left( F_1^{(\vert J_1 \vert)},\ldots, F_n^{(\vert J_n \vert)}\right)$ was contained in either $\overline{V}^k_{+,g}$ or $\overline{V}^k_{-,g}$ then at least one of the initial functional cannot be microcausal.
\end{proof}

Going through the same calculation for the proof of Proposition \ref{prop_1_C-infty_ring} we get the expression for the Peierls bracket of this composition:
\begin{equation}
\lbrace \psi(F_1,\ldots,F_n),G \rbrace_{\mathcal{L}}=\sum_{j=1}^n\left( \frac{\partial \psi}{\partial z^j}(F_1,\ldots,F_n)\lbrace F_j,G \rbrace_{\mathcal{L}} \right).
\end{equation}

With the topology of Theorem \ref{thm_1_mucaus_top} the space of microcausal functionals lacks sequential continuity. Consider as an example the simpler case where $B=M\times \mathbb{R}$, then choose $F:\varphi \in \mathcal{U} \mapsto \int_M \varphi(x)\omega$ for some smooth compactly supported $m$-form $\omega$ over $M$, also let $\lbrace f_n\rbrace$ be a sequence of smooth functions $:\mathbb{R}\rightarrow [0,1]$ supported in $[-2,2]$ and converging pointwise to the characteristic function of $[-1,1]$, $\chi_{[-1,1]}$, then sequences of derivatives of $f_n$ all converge punctually to the zero function on $\mathbb{R}$. If we define $F_n(\varphi)=f_n \circ F(\varphi)$, then pointwise $F_n(\varphi)\rightarrow \chi_{[-1,1]} \circ F(\varphi)$ but 
$$
	F^{(k)}_{n \space\ \varphi}[0](\psi_1,\ldots,\psi_k)=f^{(k)}_n(F(\varphi))\int_M \psi_1(x)\omega(x) \ldots \int_M\psi_k(x) \omega(x)
$$
converges pointwise to the zero functional in the microcausal topology, however $\chi_{[-1,1]} \circ F(\cdot)$ is not even continuous, and definitely not microcausal. Thus the completion of this topology might be somewhat uncontrolled.\\
\begin{remark*}
    
To define a topology that is both nuclear and has a well behaving sequential completion, we induce, as in \cite{acftstructure}, the \textit{strong convenient} topology. We start with the following observations: 
\begin{itemize}
    \item let $\mathcal{U}$ be a $CO$-open subset of $\Gamma^{\infty}(M\leftarrow B)$, then the smooth curves $\gamma:I\subseteq \mathbb{R} \to \mathcal{U}$ are precisely the conveniently smooth curves from $I$ to $\mathcal{U}$;
    \item by construction each $\mathcal{U}$ can be decomposed as the union of CO-open subsets $\sqcup_{\varphi\in \mathcal{U}} \mathcal{U}\cap\mathcal{V}_{\varphi} $, each of which is topologically isomorphic to an open subset of $\Gamma^{\infty}_c(M\leftarrow \varphi^*VB)$ with the LF topology;
    \item combining Corollary \ref{coro_1_WO^0-curves} we see that each smooth curve $\gamma:I\to \mathcal{U}$ is just valued in some $\mathcal{V}_{\varphi}$ for some $\varphi\in \mathcal{U}$. 
    \item if $\varphi_0\in \mathcal{U}$ is fixed, we can can consider the ultralocal chart representation of any functional $F\in \mathcal{F}_c(\mathcal{U},B)$ as $F_{\varphi_0}:\mathcal{U}_0\equiv u_{\varphi_0}(\mathcal{U}\cap \mathcal{U}_{\varphi_0})\to \mathbb{R}$. Using the fact that the functional has compact support we can extend it to $\widetilde{F}_{\varphi_0}=F_{\varphi_0}\circ i_{\chi}:\widetilde{\mathcal{U}}_{0}\subset \Gamma^{\infty}(M\leftarrow \varphi_0^*VB) \to \mathbb{R}$ where $\chi \in C^{\infty}_c(M)$ has $\chi|_{\mathrm{supp(F)}}\equiv 1$, $i_{\chi}(\vec X )=\chi \vec X$ and $\widetilde{\mathcal{U}}_{0}=i_{\chi}^{-1}(\mathcal{U}_0)$ is $CO$-open.
    \item if we equip $\Gamma^{\infty}(M\leftarrow \varphi_0^*VB)$ with the $CO$-open topology, then it becomes a Fréchet space {(the compact-open topology is equivalent to the topology of uniform convergence on compact subsets)}, thus, on $\widetilde{\mathcal{U}}_{0}$, the notions of Bastiani smoothness and convenient smoothness do coincide by Proposition \ref{prop_A_prop_of_c_infty_top}.
\end{itemize}

Next we note that, by \eqref{eq_A_char_smooth_conv_functions}, that we can write
$$
    C^{\infty}(\widetilde{\mathcal{U}}_{0},\mathbb{R})=\lim_{\substack{\longleftarrow \\ \widetilde{\gamma}\in C^{\infty}(\mathbb{R}, \widetilde{\mathcal{U}}_0)}}C^{\infty}(\mathbb{R},\mathbb{R})=\bigg\{ \{\widetilde{F}\}_{\widetilde{\gamma}} \in \prod_{\widetilde{\gamma}\in C^{\infty}(\mathbb{R},\widetilde{\mathcal{U}}_0)} C^{\infty}(\mathbb{R},\mathbb{R}) : \widetilde{F}_{\widetilde{\gamma}}\circ \kappa =\widetilde{F}_{\widetilde{\gamma}\circ \kappa}\bigg\}
$$
where the inverse limit is taken with respect to the pre-order $\gamma\leq \gamma'$ if and only if there is $\kappa \in C^{\infty}(\mathbb{R},\mathbb{R})$ with $\gamma=\gamma'\circ \kappa$. Moreover we induce on $C^{\infty}(\widetilde{\mathcal{U}}_0,\mathbb{R})$ the initial topology from the Fréchet space topology on $C^{\infty}(\mathbb{R},\mathbb{R})$ through the pullbacks $\gamma^{*}$. This is a nuclear and sequentially complete topology (see the discussion in remark 4.3 pp. 55 on \cite{acftstructure}); finally, since $\mathcal{F}_c(B,\mathcal{U}_0)$ is closed in $C^{\infty}(\widetilde{\mathcal{U}}_0,\mathbb{R})$, nuclearity and completeness are inherited in the quotient topology. This space is even a locally convex topological vector space with the seminorms
\begin{equation}\label{eq_1_strong_conv_seminorms_bundle}
    \sup_{\substack{t\in[a,b]\subset \mathbb{R}\\ \gamma\in C^{\infty}(\mathbb{R},\mathcal{U}_0)\\(\Vec{X}_1,\ldots,\Vec{X}_k)\in \mathcal{B}\subset \boxtimes^k \Gamma^{\infty}(M\leftarrow \varphi^*VB)}} \big\vert \nabla^{(k)}F[\gamma(t)](\Vec{X}_1,\ldots,\Vec{X}_k)\big\vert,
\end{equation}
where $\mathcal{B}\subset \boxtimes^k \Gamma^{\infty}(M\leftarrow \varphi^*VB)$ is a closed, bounded subset. Notice that due to $F$ having compact support, then we can evaluate $ \nabla^{(k)}F[\gamma(t)](\Vec{X}_1,\ldots,\Vec{X}_k)$ by taking a properly chosen cutoff $\chi\in C^{\infty}_c(M)$ and calculating $\nabla^{(k)}F[\gamma(t)](\chi\Vec{X}_1,\ldots,\chi\Vec{X}_k)$.

We can then induce a topology $\tau_{\mathrm{sc}}$, which we call the \textit{strong convenient} topology, in $\mathcal{F}_{\mu c}(B,\mathcal{U},g)$ with \eqref{eq_1_strong_conv_seminorms_bundle} in place of \eqref{eq_1_point_seminorm}, together with the seminorms \eqref{eq_1_hormander_seminorm}. Then $\tau_{\mathrm{sc}}$ will enjoy the following properties:

\begin{itemize}
    \item[(a)] it remains a nuclear locally convex space topology;
    \item[(b)] will have a well controlled and nuclear\footnote{See Proposition 5.3.1 in \cite{nlcs}.} completion, \textit{i.e.} its completion amounts to the topology induced from the completion of the spaces $\Gamma^{-\infty}_{c\ \Upsilon_{k,g}}(M\leftarrow \varphi^*VB)$;
    \item[(c)] the Poisson *-algebra and $C^\infty$-ring operations are continuous and remain such when passing to the completion described in (b), thanks to the result of \cite{brouder2014continuity}.
\end{itemize}

\end{remark*}
Before the next result, let us recall some notions from \cite{kriegl1997convenient}. A topological space $(X,\tau)$ is Lindel\"of if given any open cover of $X$ there is a countable open subcover, it is separable if it admits a countable dense subset, and it is second countable if it admits a countable basis for the topology.\\
Let $X$ be a Hausdorff locally convex topological space, possibly infinite dimensional, and take $S\subset C(X,\mathbb{R})$, a subalgebra. We say that $X$ is $S$-\textit{normal} if $\forall A_0,A_1$ closed disjoint subsets of $X$ there is some $f \in S$ such that $f|_{A_i}=i$, while we say it is $S$-\textit{regular} if for any neighborhood $U$ of a point $x$ there exists a function $f\in S $ such that $f(x)=1$ and $\mathrm{supp}(f)\subset U$. A $S$-\textit{partition of unity} is a family $\lbrace \psi_j \rbrace_{j \in J}$ of mappings $S\ni \psi_j:X \rightarrow \mathbb{R}$ with
\begin{itemize}
\item[$(i)$] $\psi_j(x) \geq 0$ for all $j\in J$ and $x\in X$;
\item[$(ii)$] the set $\lbrace \mathrm{supp}(\psi_j): j\in J \rbrace$ is a locally finite covering of $X$,
\item[$(iii)$] $\sum_{j \in J} \psi_j(x)=1$ for all $x \in X$. 
\end{itemize}
When $X$ admits such partition we say it is $S$-\textit{paracompact}.

\begin{proposition}\label{prop_1_mu_caus_top_prop}
The following facts hold true:
\begin{itemize}
\item[$(i)$] Given any $\mathcal{U}\subset \Gamma^{\infty}(M\leftarrow B)$ $CO$-open and any $\varphi_0 \in \mathcal{U}$ there is some $F \in \mathcal{F}_{\mu c}(B,\mathcal{U},g)$ such that $F(\varphi_0)=1$, $0 \leq F|_{\mathcal{U}}\leq 1$ and $F|_{\Gamma^{\infty}(M\leftarrow B) \backslash \mathcal{U}_{\varphi_0}}=0$, \textit{i.e.} $\mathcal{U}$ is $ \mathcal{F}_{\mu caus}(B,\varphi_{0},g)$-regular.
\item[$(ii)$] Any $\mathcal{U}\subset \Gamma^{\infty}(M\leftarrow B)$ $CO$-open admits locally finite partitions of unity belonging to $\mathcal{F}_{\mu c}(B,\mathcal{U},g)$.
\item[$(iii)$] Given any $\mathcal{U}\subset \Gamma^{\infty}(M\leftarrow B)$ $CO$-open, the algebra $\mathcal{F}_{\mu c}(B,\mathcal{U},g)$ separates the points of $\mathcal{U}$, that is if $\varphi_1\neq \varphi_2 $ there is a microcausal functional $F$ that has $F(\varphi_1)\neq F(\varphi_2)$. 
%\item[(iv)] Given any $\mathcal{U}\subset \Gamma^{\infty}(B)$ CO-open, then any unital $*$-morphism $H:\mathcal{F}_{\mu caus}(B,\varphi_{\alpha},g,\mathcal{U})\rightarrow \mathbb{C}$ is given by the evaluation functional at a unique $\varphi \in \mathcal{U}$.
%\item[(v)] Let $\mathcal{U}$, $\mathcal{W}\subset \Gamma^{\infty}(B)$ be CO-open subsets, then any continuous unital $*$-morphism $\alpha:\mathcal{F}_{\mu caus}(B,\varphi_{\alpha},g,\mathcal{U}) \rightarrow \mathcal{F}_{\mu caus}(B,\varphi_{\alpha},g,\mathcal{W})$, is the pullback of a unique smooth map $\alpha^{*}:\mathcal{W} \rightarrow \mathcal{U}$.
\end{itemize}
\end{proposition}

\begin{proof}
To show $(i)$ take the chart $(\mathcal{U}_{\varphi_0},u_{\varphi_0})$ and consider the open subset $\mathcal{U} \cap \mathcal{U}_{\varphi_0}$, fix some compact $K \subset M$ and some $\omega \in \Gamma^{\infty}_c(M\leftarrow \varphi^{*}VB'\otimes \Lambda_m(M))$ with $\mathrm{supp}(\omega) \subset K$, then we can define a functional 
$$
	G_{\omega}:\mathcal{U}\cap \mathcal{U}_{\varphi_0} \ni \varphi \mapsto G_{\omega}(\varphi)=\int_M \omega(u_{\varphi_{0}}(\varphi)). 
$$
Denote now by $G$ the functional with $G(\varphi)\doteq G_{\omega}\circ u_{\varphi_0}^{-1}(u_{\varphi_0}(\varphi))$. By construction $G(\varphi_0)=0$. Let now $\mathcal{W}=\lbrace \varphi \in \mathcal{U}\cap \mathcal{U}_{0} : G(\varphi)<\epsilon^2 \rbrace$ for some constant $\epsilon$, then if $\chi:\mathbb{R} \rightarrow \mathbb{R}$ is a smooth function supported in $[-1,1]$ with $0 \leq \chi|_{[-1,1]} \leq 1$ and $\chi|_{[-1/2,1/2]} \equiv 1$, consider the new functional $F=\chi \circ (\frac{1}{\epsilon^2}G)$, since $G$ is microlocal (and thus microcausal by Proposition \ref{prop_1_muloc_into_mucaus}) and $\chi$ is smooth, by Propositions \ref{prop_1_C-infty_ring} $F$ is microcausal. Outside $\mathcal{W}$, $F$ is identically zero so we can smoothly extend it to zero over the rest of $\mathcal{U}$ to a new functional which we denote always by $F$ that has the required properties. We first show that $(ii)$ holds for $U_{\varphi}$. Using the chart $(U_{\varphi},u_{\varphi})$, we can identify $U_{\varphi}$ with an open subset of $\Gamma^{\infty}_c(M\leftarrow \varphi^{*}VB)$. If we show that $U_{\varphi}$ is Lindel\"of and is $ \mathcal{F}_{\mu c}(B,\varphi_{0},g)$-regular, then we can conclude via Theorem 16.10 pp. 171 of \cite{kriegl1997convenient}. $ \mathcal{F}_{\mu c}(B,\mathcal{U},g)$-regularity was point $(i)$ while the Lindel\"of property follows from Theorem \ref{thm_1_Gamma_c_TVS}. Now we observe that any $\mathcal{U}$ can be obtained as the disjoint union of subsets $\mathcal{V}_{\varphi_0} \doteq \lbrace \psi \in \mathcal{U}: \mathrm{supp}_{\varphi_0}(\psi) \ \mathrm{is} \ \mathrm{compact}\rbrace$. Each of this is Lindel\"of and metrizable by Theorem \ref{thm_1_Gamma_c_TVS}, so given the open cover $\lbrace \mathcal{U}_{\varphi}\rbrace_{\varphi \in \mathcal{V}_{\varphi_0}}$, we can extract a locally finite subcover where each elements admits a partition of unity and then construct a partition of unity for the whole $\mathcal{V}_{\varphi_0}$. The fact that $\mathcal{U}= \sqcup \mathcal{V}_{\varphi_0}$ implies that the final partition of unity is the union of all others. Finally for $(iii)$ just take $\mathcal{U}_{\varphi_1}$, $\mathcal{U}_{\varphi_2}$ and $F$ as in $(i)$ constructed as follows: if $\varphi_2 \in \mathcal{U}_1$ we choose $\epsilon< G(\varphi_2)$ for which $F(\varphi_1)\neq F(\varphi_2)$, if not then any $\epsilon>0$ does the job.
\end{proof}

\begin{definition}\label{def_1_dyn_ideal}
Let $\mathcal{U} \subset \Gamma^{\infty}(M\leftarrow B)$ be $CO$-open and $\mathcal{L}$ a generalized microlocal Lagrangian. We define the on-shell ideal associated to $\mathcal{L}$ as the subspace $\mathcal{I}_{\mathcal{L}}(B,\mathcal{U},g) \subset \mathcal{F}_{\mu c}(B,\mathcal{U},g)$ whose microcausal functionals are of the form 
\begin{equation}\label{eq_1_ideal_gen}
	F(\varphi)=\vec{X}_{\varphi}\left( E(\mathcal{L})_{\varphi}[0]\right)
\end{equation}
with $X:\mathcal{U}  \rightarrow T\mathcal{U}:\varphi\mapsto (\varphi,\vec{X}_{\varphi})$ is a smooth vector field.
\end{definition}

With our usual integral kernel notation we can also write \eqref{eq_1_ideal_gen} as
\begin{equation}\label{eq_1_ideal_gen kernel}
	F(\varphi)=\int_{M}\vec{X}_{\varphi}^i(x) E(\mathcal{L})_{\varphi}[0]_i(x)d\mu_g(x).
\end{equation}
We stress that functionals of the form \eqref{eq_1_ideal_gen} are those which can be seen as the derivation of the Euler-Lagrange derivative by kinematical vector fields over $\mathcal{U}\subset \Gamma^{\infty}(M\leftarrow B)$.

\begin{proposition}\label{prop_1_Poisson_ideal}
$\mathcal{I}_{\mathcal{L}}(B,\mathcal{U},g) $ is a Poisson $*$-ideal of $ \mathcal{F}_{\mu c}(B,\mathcal{U},g)$.
\end{proposition}

\begin{proof}
Clearly is $F \in \mathcal{I}_{\mathcal{L}}(B,\mathcal{U},g)$ then also $\overline{F}$ is. If $G \in  \mathcal{F}_{\mu c}(B,\mathcal{U},g)$, $G\cdot F(\varphi)=G(\varphi)X(\varphi)\left( E(\mathcal{L})_{\varphi}[0]\right)$ but then $X'= G\cdot X \in \mathfrak{X}(T\mathcal{U})$ as well, then $G \cdot F$ is in the ideal and is associated to the new vector field $X'$. Finally we have to show that if $F \in \mathcal{I}_{\mathcal{L}}(B,\mathcal{U},g)$ then also $\lbrace F,G \rbrace_{\mathcal{L}}\in \mathcal{I}_{\mathcal{L}}(B,\mathcal{U},g)$. Fix $\varphi \in \mathcal{U}$, $\vec{Y}_{\varphi}\in T_{\varphi}\mathcal{U}$; by the chain rule
$$
	dF_{\varphi}[0](\vec{Y}_{\varphi})= \int_M \left[ \vec{X}^{(1)}_{\varphi}[0]^i(\vec{Y}_{\varphi})\left(E(\mathcal{L})_{\varphi}[0]_i\right) +\vec{X}_{\varphi}^i \left(E^{(1)}(\mathcal{L})_{\varphi}[0]_i(\vec{Y}_{\varphi})\right)\right]d\mu_g
$$
and 
\begin{align*}
	 \lbrace F,G \rbrace_{\mathcal{L}}(\varphi)&= \left\langle dF_{\varphi}[0],\mathcal{G}_{\varphi} dG_{\varphi}[0] \right\rangle\\ &= 
	\int_M\left[ \vec{X}^{(1)}_{\varphi}[0]^{i}_{j}\left(\mathcal{G}_{\varphi}^{jk} g^{(1)}_{\varphi}[0]_k \right)\left(E(\mathcal{L})_{\varphi}[0]_i \right)+\vec{X}_{\varphi}^i \left( E^{(1)}(\mathcal{L})_{\varphi}[0]_{ij}\left(\mathcal{G}_{\varphi}^{jk} g^{(1)}_{\varphi}[0]_k\right)\right)\right]d\mu_g \\ &= 
	\int_M \left\{ \left[ \vec{X}^{(1)}_{\varphi}[0]^{i}_{j}\left(\mathcal{G}_{\varphi}^{jk} g^{(1)}_{\varphi}[0]_k \right) \right]\left(E(\mathcal{L})_{\varphi}[0]_i\right)\right\}d\mu_g
\end{align*}
where we used that $\mathcal{G}_{\varphi}$ associates to its argument a solution of the linearized equations. Defining $\varphi\mapsto \vec{Z}_{\varphi}=\vec{X}^{(1)}_{\varphi}[0]^{i}_j\left(\mathcal{G}_{\varphi}^{jk} g^{(1)}_{\varphi}[0]_k \right)\partial_i \in \Gamma^{\infty}_c(M \leftarrow \varphi^*VB)$ yields a smooth mapping (by smoothness of $X$, the functional $G$ and the propagator $\mathcal{G}_{\varphi}$) defining the desired vector field.
\end{proof}

\begin{definition}\label{def_1_on_shell_ideal}
Let $\mathcal{U} \subset \Gamma^{\infty}(M\leftarrow B)$ be $CO$-open and $\mathcal{L}$ a generalized microlocal Lagrangian. We define the on-shell algebra on $\mathcal{U}$ associated to $\mathcal{L}$ as the quotient
\begin{equation}
	  \mathcal{F}_{\mathcal{L}}(B,\mathcal{U},g) \doteq \mathcal{F}_{\mu c}(B,\mathcal{U},g)/ \mathcal{I}_{\mathcal{L}}(B,\mathcal{U},g).
\end{equation}
\end{definition}

This accounts for the algebra of observable once the condition $E(\mathcal{L})_{\varphi}[0]=0$ has been imposed on $\mathcal{U}$.

\section{Examples}\label{section_examples}

\subsection{Wave maps}

Finally we introduce, as an example of physical theory, wave maps. The configuration bundle, is $C=M\times N$, where $M$ is an $m$ dimensional Lorentzian manifold and $N$ an $n$ dimensional manifold equipped with a Riemannian metric $h$. The space of sections is canonically isomorphic to $C^{\infty}(M,N)$, the latter possess a differentiable structure induced by the atlas $\big(\mathcal{U}_{\varphi}, u_{\varphi}, \Gamma^{\infty}_c(M\leftarrow \varphi^{*}TN) \big) $, where $u_{\varphi}$ has the exact same form \eqref{eq_1_trivial_gamma_local_chart}, with the only difference being that the sections are $N$-valued mappings, thus $\exp$ can be taken as the exponential function induced by a Riemannian metric $h$ on $N$. The generalized Lagrangian for wave maps is 
\begin{equation}\label{eq_1_lag_wave_maps}
	\mathcal{L}_{\mathrm{WM}}(f)(\varphi)=\frac{1}{2}\int_M f(x)  \mathrm{Trace}(g^{-1} \circ (\varphi^*h))(x)d\mu_g(x);
\end{equation}
obtained by integration of the standard geometric Lagrangian $\lambda= \frac{1}{2}g^{\mu\nu}h_{ij}(\varphi)\varphi^i_{\mu}\varphi^j_{\nu}d\mu_g$ smeared with a test function $f\in C^{\infty}_c(M)$. Computing the first functional derivative, as per \eqref{eq_1_euler_der}, we get the associated E-L equations, which written in jets coordinates reads
\begin{equation}\label{eq_1_wavemaps_EL_eq}
	h_{ij}g^{\mu \nu} \left( \varphi^i_{\mu \nu} +\lbrace h\rbrace^{i}_{kl} \varphi^k_{\mu}\varphi^l_{\nu}-\lbrace g\rbrace^{\lambda}_{\mu \nu}\varphi^i_{\lambda} \right)=0,
\end{equation}
where we denoted by $\lbrace h\rbrace$, $\lbrace g\rbrace$ the coefficients of the linear connection associated to $h$ and $g$ respectively. Computation of the second derivative of \eqref{eq_1_lag_wave_maps} yields
\begin{equation}
\begin{aligned}
	\delta^{(1)}E(\mathcal{L}_{WM})_{\varphi}[0] &: \Gamma^{\infty}_c(M\leftarrow \varphi^{*}TN) \times \Gamma^{\infty}_c(M\leftarrow \varphi^{*}TN ) \rightarrow \mathbb{R} \\ (\vec{X},\vec{Y}) \mapsto & \int_M\frac{1}{2}\left[g^{\mu\nu}(x)h_{ij}(\varphi(x)) \nabla_{\mu}{X}^i(x) \nabla_{\nu}{Y}^j(x)+A^{\mu}_{ij}(\varphi(x))\left(\nabla_{\mu}{X}^i(x){Y}^j(x)+\nabla_{\mu}{Y}^i(x){X}^j(x)\right)\right.\\
	& \space\ \space\ \space\ \space\ \space\ + \left. B_{ij}(\varphi(x)){X}^i(x){Y}^j(x) \right]d\mu_g(x).
\end{aligned}
\end{equation}
Where we choose $f\equiv 1$ in a neighborhood of $\mathrm{supp}(\vec{X})\cup \mathrm{supp}(\vec{X})$ as done before. One can show that the coefficients $A^{\mu}_{ij}$ always vanish and
\begin{equation*}
    \begin{split}
    &\delta^{(1)}E(\mathcal{L}_{WM})_{\varphi}[0] (\vec{X},\vec{Y})\\
    &\qquad =\int_M \frac{1}{2} \big( g^{\mu\nu}(x)h_{ij}(\varphi(x))\nabla_{\mu} {X}^i(x) \nabla_{\nu}Y^j(x)+ R^k_{ilj}(\varphi(x))p^{\alpha}_k(\varphi(x))\varphi^l_{\alpha}(x){X}^i(x)Y^j(x)\big)d\mu_g(x)
\end{split}
\end{equation*}
where $R$ are the components of the Riemann tensor of the Riemannian metric $h$, and $p^{\alpha}_k\doteq \frac{\partial \lambda}{\partial y^{k}_{\alpha}}$ is the conjugate momenta of the Lagrangian $\lambda$. It is therefore evident that the induced differential operator $D_{\varphi}$ can be expressed locally as
\begin{equation}
	D_{\varphi}(\vec{X})(x) =\big(g^{\mu\nu}(x)h_{ij}(\varphi(x))\nabla_{\mu \nu} \vec{X}^i(x) + R^k_{ilj}(\varphi(x))p^{\alpha}_k(\varphi(x))\varphi^l_{\alpha}(x){X}^i(x)\big) dy^j\big|_{\varphi(x)}\ .
\end{equation}
Its principal symbol is clearly
$$
	\sigma_2(D_{\varphi})= \frac{1}{2} g^{\mu\nu}\frac{\partial}{\partial x^{\nu}}\frac{\partial}{\partial x^{\mu}}\otimes id_{\varphi^{*}TN}. 
$$
Theorem \ref{thm_1_properties_of_Green_functions} then ensures the existence of the advanced and retarded propagators for Wave Maps $\mathcal{G}^{\pm}_{WM}[\varphi]$. Their difference defines the causal propagator and consequently the Peierls bracket as in Definition \ref{def_1_Peierls}. The results of Sections 2.3 and 2.4 do apply to wave maps: it is therefore possible to obtain a $*$-Poisson algebra generated by microcausal functionals $\mathcal{F}_{\mu{c}}(M\times N,g)$ which enjoys all the properties collected throughout Section 2.4.

\subsection{Scalar field theories}\label{subsection_scalar_field_theories}

We finish this chapter by translating some of the results obtained above in case where $B=M\times \mathbb{R}$. We will use those results in the following chapters. We notice at first that the manifold stricture of $C^{\infty}(M,\mathbb{R})$ is generated by charts $\{(\mathcal{U}_{\varphi},u_{\varphi},C^{\infty}_c(M))\}$ where 
\begin{equation}
    \begin{aligned}
        & \mathcal{U}_{\varphi}=\varphi + C^{\infty}_c(M);\\
        & u_{\varphi}(\psi)= \psi-\varphi.
    \end{aligned}
\end{equation}
Notice that the standard Euclidean metric on $\mathbb{R}$ has a globally defined exponential $\exp:T\mathbb{R} \to \mathbb{R}\times \mathbb{R}$, then $\mathcal{U}_{\varphi}$ becomes the largest possible subset isomorphic to a copy of the modelling vector space. This \textit{enlargement} of the chart open set yields a ultralocal chart independent result for the characterization of microlocal functionals in Proposition \ref{porop_1_muloc_charachterization}, therefore if $F:C^{\infty}(M,\mathbb R)$ is microlocal and $F^{(1)}$ is locally bornological, then 
\begin{equation}\label{eq_mulocal_representation_scalar_case}
    F(\varphi)=F(\varphi_0)= \int_M j^r_x\varphi^*\theta_{F,\varphi_0}
\end{equation}
for all $\varphi$ having $\varphi-\varphi_0\in C^{\infty}_c(M)$. We also remark that $C^{\infty}(M,\mathbb R)$ can be seen as a disjoint union as
$$
    C^{\infty}(M,\mathbb R)=\bigsqcup_{\varphi\in C^{\infty}(M,\mathbb R)/\sim} \varphi + C^{\infty}_c(M) 
$$
where $\sim$ is the equivalence relation defined by $\varphi_1\sim \varphi_2$ if and only if their difference $\varphi_1-\varphi_2$ is compactly supported. In this sense, \eqref{eq_mulocal_representation_scalar_case} is a ultralocal chart independent object. Finally, due to Proposition \ref{prop_1_variational_sqn}, we also get that the
$$
    F(\varphi)=\int_M j^{r}_x\varphi^*\theta_F
$$
for a globally defined $\theta\in \Omega^n(M)$.\\

In the sequel the dynamics will be that given by the Klein-Gordon Lagrangian:
\begin{equation}\label{eq_KG_lagrangian}
	\mathcal{L}_f:C^{\infty}(M,\mathbb R)\ni \varphi\mapsto \int_M f(x)\big( g^{\mu\nu}\nabla_{\mu}\varphi \nabla_{\nu}\varphi +m^2\varphi^2+\kappa R(g)\varphi^2\big)(x) d\mu_g(x).
\end{equation}
The latter generates the Klein-Gordon equations
\begin{equation}\label{eq_KG}
	P(\varphi)\equiv d\mathcal{L}_1[\varphi](x)=\big( g^{\mu\nu}\nabla_{\mu}\varphi\nabla_{\nu}\varphi+m^2\varphi+\kappa R(g)\big)\varphi.
\end{equation}
In particular the constants $m$, $\kappa$ are called the \textit{mass} and the \textit{coupling constant} of the theory. We also mention a special kind of functionals, \textit{i.e.} those generating Wick powers, we denote them by
\begin{equation}\label{eq_1_k-power_functional}
	\phi^k_{(M,g)}(f):C^{\infty}(M,\mathbb R)\ni\varphi\mapsto \int_M f(x)\varphi^k(x)d\mu_g(x).
\end{equation}
It it easy to show that both \eqref{eq_KG_lagrangian}, \eqref{eq_1_k-power_functional} are microlocal and are generalized Lagrangians according to Definition \ref{def_1_gen_lag}.\\

Finally, let us show some other results for the scalar case which we will extensively use in the following parts. In the scalar case the strong convenient topology is generated by seminorms
\begin{equation}\label{eq_1_strong_conv_seminorms_1}
	p_{[a,b],\gamma,\mathcal{B}}(F)\doteq \sup_{\substack{t\in[a,b]\subset \mathbb{R}\\\gamma\in C^{\infty}(\mathbb{R},\mathcal{U})\\(\psi_1,\ldots,\psi_k)\in \mathcal{B}}}\left| d^kF[\gamma(t)](\psi_1,\ldots,\psi_k)\right| ,
\end{equation}
\begin{equation}\label{eq_1_hormander_seminorm_scalar_case}
	p_{\varphi,\Upsilon_{k,g},q,\chi}(F) \doteq ||F^{(k)[\varphi]}||_{q,\chi,V} ,
\end{equation}
where $\gamma: \mathbb{R}\rightarrow \mathcal{U}$ is a smooth curve valued in a $c^{\infty}$-open subset of $C^{\infty}(M)$, $\mathcal{B}$ is a closed and bounded subset in $C^{\infty}(M)^k\simeq C^{\infty}(M^k)$, $q\in \mathbb N$, $\chi \in C^{\infty}_c(M^k) $, $V\subset T^*M^k$ with $\Upsilon_{k,g}\cap \mathrm{supp}(\chi)\times V = \emptyset$ are H\"ormander seminorms in the direct limit topology of Lemma \ref{lemma_1_mucaus_top_1}. Then we have the following technical result:

\begin{lemma}\label{lemma_1_regular_density}
    $\mathcal{F}_{reg}(M,g)\subset \mathcal{F}_{\mu c}(M,g)$ is sequentially dense in the strong convenient topology.
\end{lemma}
% \begin{lemma}\label{lemma_1_regular_density}
% 	Let $F\in \mathcal{F}_{\mu c}(M,g)$, then there exists a sequence $\{F_n\}\subset \mathcal{F}_{reg}(M,g)$ such that $F_n \rightarrow F$ in the strong convenient topology of $\mathcal{F}_{\mu c}(M,g)$.
% \end{lemma}

\begin{proof}
Let $F\in \mathcal{F}_{\mu c}(M,g)$, we show that there exists a sequence $\{F_n\}\subset \mathcal{F}_{reg}(M,g)$ such that $F_n \rightarrow F$ in the strong convenient topology of $\mathcal{F}_{\mu c}(M,g)$. First we search for a sequence of mollifiers strongly converging to the identity mapping. We shall use \cite[Theorem 12 pp.68]{de2012differentiable}, which guarantees the existence of continuous linear mappings $S_n: C^{\infty}(M)\rightarrow C^{\infty}(M)$ such that $S_n \rightarrow id_{C^{\infty}(M)}$ uniformly over bounded subsets of $C^{\infty}(M)$. Suppose that $F\in \mathcal{F}_{\mu c}(M,h)$ has support inside a compact $K\subset M$, and let $\chi\in C^{\infty}_c(M)$ with $\chi\vert_K=1$. Set $F_n(\varphi)=F(S_n(\chi_n\varphi))$, since each $S_n$ enlarges the support at most in a $1/n$ neighborhood of $K$, $F(S_n(\chi_n\varphi))=F(S_n(\varphi))$; by the Schwartz kernel theorem we can represent $S_n(\chi\varphi)$ as 
$$
	S_n(\chi\varphi)(x)= \int_MS_n(x,y) \chi(y)\varphi(y) d\mu_g(y),
$$
where $S_n \in C^{\infty}(M^2)$ for each $n\in \mathbb{N}$. To simplify the notation, we shall write $S_n(\varphi)$ in place of $ S_n(\chi\varphi ) $. We claim that $\mathrm{WF}(d^kF_n[\varphi])=\emptyset$ $\forall n,k \in \mathbb{N}$. To wit, by linearity of each $S_n$, $dS_n[\varphi] (\psi)= S_n(\psi)$, $d^kS_n[\varphi]=0$ $\forall k \geq 2$; thus the only non trivial term in Faà di Bruno's formula is
$$
	d^kF_n[\varphi](\psi_1,\ldots, \psi_k)= d^kF[S_n(\varphi)](dS_n[\varphi](\psi_1),\ldots, dS_n[\varphi](\psi_k))= d^kF[S_n(\varphi)](S_n(\psi_1),\ldots, S_n(\psi_k)).
$$
Its associated kernel, $d^kF_n[S_n(\varphi)](x_1,\ldots,x_k)$, will have the form
$$
	\int_{M^k}d^kF[S_n(\varphi)](y_1,\ldots,y_k)S_n(x_1,y_1)\cdots S_n(x_k,y_k) d\mu_g(y_1,\ldots,y_k).
$$
By smoothness of each $S_n$, singularities of $d^kF[S_n(\varphi)]$ are suppressed and $d^kF_n[S_n(\varphi)](x_1,\ldots,x_k)$ becomes smooth, therefore $F_n \in \mathcal{F}_{reg}(M)$. Finally we have to show convergence in the \textit{strong convenient topology}, which is generated by the family of seminorms in \eqref{eq_1_strong_conv_seminorms_1} and \eqref{eq_1_hormander_seminorm_scalar_case}. Since we are working with regular functionals, the second seminorms are trivial and all we have to check is convergence with respect to the first. In practice, chosen any seminorm and $\epsilon>0$, we have to show that there is some $n_0$ for which $n>n_0$ implies 
$$
	p_{\gamma,[a,b],B ,k }(F_n-F)<\epsilon.
$$
Thus we estimate
$$
\begin{aligned}
	p_{\gamma,[a,b],B ,k }(F_n-F)  &=\sup_{\substack{t\in [a,b]\\(\psi_1,\ldots,\psi_k)\in B}} \left|d^kF[S_n(\gamma(t))](S_n(\psi_1),\ldots, S_n(\psi_k))-d^kF[\gamma(t)](\psi_1,\ldots, \psi_k) \right|\\
	\leq & \sup_{\substack{t\in [a,b]\\(\psi_1,\ldots,\psi_k)\in B}}\left|d^kF[S_n(\gamma(t))](S_n(\psi_1),\ldots, S_n(\psi_k))-d^kF[S_n(\gamma(t)](\psi_1,\ldots, \psi_k) \right|\\
	&\quad +\sup_{\substack{t\in [a,b]\\(\psi_1,\ldots,\psi_k)\in B}} \left|d^kF[S_n(\gamma(t))](\psi_1,\ldots, \psi_k)-d^kF[\gamma(t)](\psi_1,\ldots, \psi_k) \right|\\
	& \quad\leq  \sup_{\substack{t\in [a,b]\\(\psi_1,\ldots,\psi_k)\in B}} \left|d^kF[S_n(\gamma(t))](S_n(\psi_1)-\psi_1,\ldots, S_n(\psi_k)-\psi_k) \right|\\ 
	&\quad +\sup_{\substack{t\in [a,b]\\(\psi_1,\ldots,\psi_k)\in B}}  \left| p_{S_n(\gamma(t)),\psi_1,\ldots,\psi_k}(F)-p_{\gamma(t),\psi_1,\ldots, \psi_k}(F)\right|.
\end{aligned}
$$
$F: C^{\infty}(M)\rightarrow \mathbb{R}$ is convenient smooth, hence Bastiani smooth due to the Fréchet space topology in $C^{\infty}(M)$; as a result all its derivatives are jointly continuous, and we can further estimate
$$
\begin{aligned}
	& \sup_{\substack{t\in [a,b]\\(\psi_1,\ldots,\psi_k)\in B}}\left\{C q_{K,l}(S_n(\gamma(t))) \prod_{i=1}^kq_{K,l}(S_n(\psi_i)-\psi_i))\right\} \\ &+\sup_{\substack{t\in [a,b]\\(\psi_1,\ldots,\psi_k)\in B}}\left\{ C' (q_{K,l}(S_n(\gamma(t))-q_{K,l}(\gamma(t))) \prod_{i=1}^kq_{K,l}(\psi_i))\right\}\\
	\leq & C \sup_{t\in [a,b]}\left\{ q_{K,l}(S_n(\gamma(t)))\right\} \sup_{(\psi_1,\ldots,\psi_k)\in B}\left\{\prod_{i=1}^kq_{K,l}(S_n(\psi_i)-\psi_i))\right\} \\& +C'\sup_{t\in [a,b]}\left\{  q_{K,l}(S_n(\gamma(t))-\gamma(t))\right\} \sup_{(\psi_1,\ldots,\psi_k)\in B}\left\{\prod_{i=1}^kq_{K,l}(\psi_i))\right\}\\ 
\end{aligned}
$$
where $q_{K,l}$ are the seminorms of the Fréchet topology on $C^{\infty}(M)$. The terms $\sup_{t\in [a,b]}\left\{ q_{K,l}(S_n(\gamma(t)))\right\}$, $\sup_{(\psi_1,\ldots,\psi_k)\in B}\left\{\prod_{i=1}^kq_{K,l}(\psi_i))\right\}$ are bounded by constants $D$ and $D'$ respectively, due to continuity of the seminorms; moreover, since $S_n$ converges uniformly on bounded subsets of $C^{\infty}(M)$, there is $n_0\in \mathbb{N}$, having $\sup_{(\psi_1,\ldots,\psi_k)\in B}\left\{q_{K,l}(S_n(\psi_i)-\psi_i))\right\}< \frac{\epsilon}{2CD}$ and $\sup_{\psi\in B' \supset \gamma([a,b])}\left\{  q_{K,l}(S_n(\psi)-\psi)\right\}< \frac{\epsilon}{2C'D'}$, finally establishing our claim. Notice that those estimates do not depend on the cut-off $\chi$ for choosing its support big enough, we can directly assume $\chi \equiv 1$.
\end{proof}

\chapter{ Functional formalism in quantum field theory: Wick powers}\label{chapter_Wick}
\thispagestyle{plain}

In this chapter we are transitioning from classical free field theories to quantum free field theories. First we review the definitions and properties of natural bundles: they are particular types of bundles which, broadly speaking, have a canonical way for lifting local diffeomorphisms of the base manifold to local automorphisms of the bundle. This notion is very common for one of the first bundles encountered in differential geometry is the tangent bundle, which has the aforementioned property, in particular given a local diffeomorphism we can associate its Jacobian matrix, and this is all is needed to create the trivializations of the tangent space. We stress this lifting of diffeomorphism (as well as the tangent mapping) has functorial origin, we will then see (\Cref{thm_0_natural_functor_characterization}, \Cref{thm_0_natural_bundle_characterization}), that natural bundles can be equivalently described by a functor or through a geometrical construction. The reason we are interested in such objects is that they provide a framework for handling the geometric parameters of the theory which, in the case of Klein-Gordon theory in curved spacetime, will result in defining the natural bundle of background geometries (see \eqref{eq_2_H(M)_natural_bundle} and Definition \ref{def_2_background_geometry}) whose sections $h=(g,m,\kappa)$ represents the metric, the mass and the coupling constant to gravity. The importance of such bundles is crucial for classical field theories for it gives a \textit{natural} way of treating covariant functions and forms (for example see \Cref{thm_2_covariant_identity}); however, it can be used in the quantum setting as well and will provide the technical framework for classifying ambiguities in the definition of Wick powers (\textit{c.f.} \Cref{thm_2_Moretti_Kavhkine}).\\

The starting point for the quantum theory are the classical results obtained in \Cref{chapter_classical}, then we apply deformation quantization to the Poisson algebra of microcausal functionals $(\mathcal{F}_{\mu c},\{ \ , \ \}_{(M,h)})$ ($h=(g,m,\kappa)$ is a background geometry), which consists in deforming the classical product of functionals to a $\star$-product. This is however, not directly possible for microcausal functionals; to find a way around we restrict the algebra to regular functionals and define there the $\star$ product. Introducing the notion of Hadamard state and using the characterization \cite[Theorem 5.1]{radzikowski1996micro} we show that it is possible to define a $\star_H$-product on $\mathcal{F}_{\mu c}$, where $H$ is a Hadamard parametrix. The issue here is that this $\star$ product is not uniquely defined for any two parametrix $H,\ H'$ differ by a smooth function $d\in C^{\infty}(M\times M)$. Using the density of regular functionals into microcausal, we extend the regular $\star$ product to a $\star$ product in the abstract algebra of microcausal functionals $\mathfrak A_{\mu c}$ with the property that whenever $H$ is the symmetric part of a Hadamard state, we have an isomorphism of $\star$ algebras
$$
    \alpha_H :(\mathfrak A_{\mu c}, \star) \to (\mathcal{F}_{\mu c},\star_H).
$$

The next step is the definition of Wick powers (see Definition \ref{def_2_Wick_quantum_powers}). As mentioned earlier, in \cite{hollands2001local}, the characterization of uniqueness has been established using the rather strong analytical dependence of Wick powers on the background geometry. This condition was weakened in \cite{khavkine2016analytic} to a milder one which was enough for using the Peetre-Slov\'ak Theorem and classify ambiguities in the definition of Wick powers. For the sake of thoroughness, we will state and prove the results obtained in \cite{khavkine2016analytic, khavkine2019wick} for scalar field theories. We stress that, differently from \cite{hollands2001local, khavkine2016analytic}, we are not using on-shell algebras, since functional formalism is general enough to avoid working in this special condition. Several of the requirements of Wick powers were therefore generalized in Definition \ref{def_2_Wick_quantum_powers}.\\

\section{Natural bundles}\label{section_natural_bndl&background_geometries}

\begin{definition}\label{def_0_natural_bundles}
	A natural bundle is a covariant functor, $\mathfrak{N}: \mathfrak{Man} \rightarrow \mathfrak{FBndl}$, from the category of manifold and local diffeomorphisms to the category of fiber bundles and local fibered morphisms, satisfying the following properties:
	\begin{itemize}
	\item[$(i)$] $\mathfrak{N}(M) \equiv (B,\pi, M,F)$ is a fiber bundle over $M$;
	\item[$(ii)$] if $i: U \hookrightarrow M$ is the inclusion of a submanifold, then $\mathfrak{N}(i):\mathfrak{N}(U) \rightarrow \mathfrak{N}(M)$ is the inclusion of the sub-bundle $\mathfrak{N}(U) \equiv (\pi^{-1}(U),\pi \mid_{U},U,F)$ into the fiber bundle $\mathfrak{N}(M)$;
	\item[$(iii)$] If $\Sigma$ is a parameter manifold and $\phi: \Sigma \times M \rightarrow N$ is a smooth family of local diffeomorphisms depending on the parameter $s\in \Sigma$, then $\mathfrak{N}(\phi):\Sigma \times \mathfrak{N}(M) \rightarrow \mathfrak{N}(N)$ is smooth.
	\end{itemize}
\end{definition}

It is now easy to see that the tangent bundle $(TM,\tau,M,\mathbb{R}^n)$ is an example of natural bundle, in fact $T$ is easily seen to be a covariant functor, $(i)$ is clear, $(ii)$ follows once we realize that $i:U \hookrightarrow$ is a local diffeomorphism that is locally the identity, therefore $Ti:TU \rightarrow TM$ will be locally the identity as well, therefore forming a sub-bundle. Finally to show $(iii)$ note that the fiber mapping $T \phi_s :TM \rightarrow TN$, where $s\in \Sigma $ is a parameter, will involve products and sums of coefficients of the Jacobian of $\phi_s$, $J^{\mu}_{\nu}(s)$, which is jointly smooth.\\

We now attempt to give a classification of the possible types of natural bundles. In order to do so, we will take, at first sight, a somewhat strange path which, in the end, will prove of great practical utility. To start consider open sets $U$, $V$ of $\mathbb{R}^n$ and $\mathbb{R}^m$ respectively each of which contains the origin, let 
$$
	J^k_0(U,V)_0 \doteq \lbrace \alpha \in C^{\infty}(U ,V) \space\ : \space\ \alpha (0)=0 \rbrace / \sim
$$
$$
	J^k_0(\mathbb{R}^n,\mathbb{R}^m)_0 \doteq \lbrace \alpha \in C^{\infty}(\mathbb{R}^n ,\mathbb{R}^m) \space\ : \space\ \alpha (0)=0 \rbrace / \sim
$$
where the equivalence relation is given by $\alpha \sim \alpha'$ if and only if $j^k_0\alpha=j^k_0\alpha'$. Using local coordinates $(U,x^i)$ and $(V,y^a)$ we can write
$$
	j^k_0 \alpha = \alpha^a|_0+ \frac{\partial \alpha^a}{\partial x^i}\Big|_0x^i+ \ldots+ \frac{1}{k!} \frac{\partial^k \alpha^a}{\partial x^{i_1} \ldots \partial x^{i_k}}\Big|_0x^{i_1}\ldots x^{i_k}\equiv \alpha^a_i x^i+ \ldots+ \frac{1}{k!} \alpha^a_{i_1 \ldots i_k} x^{i_1}\ldots x^{i_k}
$$
therefore we can make this space a $n\left(\binom{n+k}{n} -1\right)$-dimansional differentiable manifold structure whose coordinates are of the form $\left(\alpha^a_i,\ldots\alpha^a_{i_1 \ldots i_k}\right)$. Moreover, provided $s<k$, we have smooth projections 
$$
	\pi^k_s : J^k_0(U,V)_0 \rightarrow J^s_0(U,V)_0 : (\alpha^a_i,\ldots\alpha^a_{i_1 \ldots i_k}) \mapsto (\alpha^a_i,\ldots\alpha^a_{i_1 \ldots i_s})
$$
Another important property of those spaces is \textit{covariance} for the jet functor $J^k_0(\cdot,\cdot)_0$: let $\alpha \in J^k_0(U,V)_0$, $\beta \in J^k_0(V,W)_0$ with $U$, $V$, $W$ open subsets of $\mathbb{R}^n$, $\mathbb{R}^m$, $\mathbb{R}^p$ respectively, then $\beta\circ \alpha: \mathbb{R}^n \rightarrow \mathbb{R}^p$ is well defined and the image of $0\in \mathbb{R}^n$ is $0\in \mathbb{R}^p$, moreover by Faà di Bruno's formula $[\beta\circ \alpha] \in J^k_0(U,W)_0$. As an example, for $k=2$ we have that locally $[\beta \circ \alpha]$ has the form
$$
    (\beta^a_j\alpha^j_i, \beta^a_{ij}\alpha^i_k \alpha^j_l + \beta^a_i\alpha^i_{kl}).
$$
Suppose that $m=n$, and that instead of generic mappings we only consider those that are invertible on a neighborhood of the origin. We shall denote this space as 
$$
	GL^k(n,\mathbb{R}) \doteq J^k_0(U,V)^{(\mathrm{inv})}_0.
$$
Not only the latter space is a manifold but it is a Lie group as well: the inversion map is the one associating the local inverse, while the multiplication map is given by the composition of mappings; both operations are made by sums and products of monomials, and, as a consequence, are smooth. $GL^k(n,\mathbb{R})$ is then usually referred as the $k$th order jet group of $GL(n,\mathbb{R})$, since, when $k=1$, we can readily identify the former with with the general linear group of $\mathbb R^n$. \\

Now we go one step further, let $e:U \subset \mathbb{R}^n \rightarrow M$ be a local diffeomorphism in a neighborhood of the origin $0 \in \mathbb{R}^n$, calculating its $k$th order jet extension in the origin yields
$$
	(j^k_0e)^{\mu} = e^{\mu}(0)+e^{\mu}_i x^i + \ldots + e^{\mu}_{i_1 \ldots i_k}x^{i_1} \ldots x^{i_k}.
$$
Similarly to what we did before, let $J^k_0(U,M)^{(\mathrm{inv})}$ the $k$th order jet space of invertible mappings with the structure of a fiber bundle over $M$ of dimension $n \binom{n+k}{n}$, which we denote by $L^k(M)$. Fibered coordinates are of the form $(e^{\mu}, e^{\mu}_i\ldots , e^{\mu}_{i_1 \ldots i_k}) $, and coordinates transformations on $M$ with $x'^{\mu}=x'^{\mu}(x)$, induce transformations on $L^k(M)$ of the form
\begin{equation}\label{eq_0_natural_coordinates_transformations}
  \left\{\begin{array}{@{}l@{}}
    e'^{\mu}  \equiv x'^{\mu}(x), \\
    e'^{\mu}_i = J^{\mu}_{\alpha}(x)e^{\alpha}_i, \\
    e'^{\mu}_{ij} = J^{\mu}_{\alpha}(x)e^{\alpha}_{ij}+ J^{\mu}_{\alpha \beta}(x)e^{\alpha}_i e^{\beta}_j,  \\
    \ldots \\
	e'^{\mu}_{i_1 \ldots i_k}x^{i_1} = J^{\mu}_{\alpha} e^{\alpha}_{i_1 \ldots i_k}x^{i_1} + \ldots + J^{\mu}_{\alpha_1 \ldots \alpha_k} e^{\alpha_1}_{i_1} \ldots e^{\alpha_k}_{i_k}. 
  \end{array}\right.\,
\end{equation}

\begin{lemma}\label{lemma_0_principal_bndl_frame_k}
    $L^k(M)\to M$ is a principal fiber bundle, called the $k$th order frame bundle, with structure group $GL^k(n,\mathbb{R})$.
\end{lemma}
\begin{proof}
    We shall verify that the trivializations induced by \eqref{eq_0_natural_coordinates_transformations} define trivializations as per Definition \ref{def_0_principal_bundles}. The projection $\pi:L^k(M)\rightarrow M: j^k_0e \mapsto e(0)$ is a well defined smooth mapping, consider some atlas $\lbrace U_{a}, x^{\mu} \rbrace_{a\in A}$ of $M$, and define trivializations
\begin{equation}\label{eq_0_L^k(M)_trivialization}
	t_{(a)}: L^k(M) \rightarrow M \times GL^k(n,\mathbb{R}) : j^k_0e \mapsto (e(0), j^k_0 \alpha) ,
\end{equation}
where $j^k_0 \alpha \in GL^k(n,\mathbb{R})$ is the unique element in the equivalent class whose $k$th order jet expansion in $0\in \mathbb{R}^n$ is
$$
	\alpha(0)= e^{\mu}_i x^i + \ldots + e^{\mu}_{i_1 \ldots i_k}x^{i_1} \ldots x^{i_k}+ R_{k+1}(x).
$$
Since the mappings $t_{(a)}$ involve only polynomial terms, they are smooth. Transition mappings associates to trivializations $t_{(a)}$ and $t_{(b)}$ takes the form 
\begin{equation}\label{eq_0_L^k(M)_transition_mappings}
\begin{aligned}
    t_{(ab)} :&U_{\alpha \beta}\times GL^k(n,\mathbb{R}) \rightarrow U_{\alpha \beta}\times GL^k(n,\mathbb{R}) \\
    & j^k_0(x \circ e) \equiv (e(0), j^k_0 \alpha)\mapsto j^k_0(x' \circ e) \equiv ((x' \circ e)(0), j^k_0(J \circ \alpha)),
\end{aligned}
\end{equation}
which specified in coordinates reads $(e^{\mu}, e^{\mu}_i\ldots , e^{\mu}_{i_1 \ldots i_k})\mapsto (e'^{\mu}, e'^{\mu}_i\ldots , e'^{\mu}_{i_1 \ldots i_k})$ and the relation between the two is given by \eqref{eq_0_natural_coordinates_transformations}. In other words we are identifying the transition mapping as a left %$\footnote{The fact that this action is a left one follows essentially from the chain rule applied to a composition of three elements.}$ 
action of the element $j^kJ \in GL^k(n,\mathbb R)$, given by the Jacobian its subsequent derivatives, on the element $j^k_0 \alpha$ of the typical fiber. All of \eqref{eq_0_L^k(M)_trivialization}, \eqref{eq_0_L^k(M)_transition_mappings} are smooth in the respective domains since are polynomials in the respective variables.
\end{proof}

Since all principal bundles come equipped with a global right action, we explicitly write down the one of $GL^k(n,\mathbb R)$ on $L^k(M)$:
\begin{equation}\label{eq_0_L^k(M)_principal_R_action}
	L^k(M) \times GL^k(n,\mathbb R) \rightarrow L^k(M): (j^k_0e,j^k_0 \beta) \mapsto j^k_0(e \circ \beta),
\end{equation}
which locally, for $k=2$, reads 
$$
  (e^{\mu},e^{\mu}_a,e^{\mu}_{ab}, \alpha^{a}_i,\alpha^a_{ij})\mapsto (e^{\mu}, e^{\mu}_a \alpha^a_i, e^{\mu}_{ab} \alpha^a_i \alpha^b_j+ e^{\mu}_a \alpha^a_{ij}).
$$

The reason why $L^k(M)$ is called the $k$th order frame bundle is that, when $k=1$, it becomes the usual frame bundle\footnote{The frame bundle of $M$, denoted $L(M)$ is the set $\cup_{x\in M}\{(e_x)\}$, where $\{e_x\}$ is the set of basis of the tangent space $T_xM$. If a set of coordinates $(U_a,\{x^{\mu}\})$ are chosen on $M$, one can induce trivializations on $L(M)$ by setting 
$$
    t_a(e_x)=(x^{\mu},(e_x)^{\mu}_a),
$$
where $(e_x)^{\mu}_a\in GL(n,\mathbb{R})$ is the matrix associated to the change of basis $\partial_{\mu}\to (e_x)_a$.}. It is customary to introduce the tangent bundle $TM$ of $M$ as as the space of kinematic tangent vectors of $M$, it can, however, also be seen as the bundle associated to the frame bundle $L(M)$ with the left action $\lambda: GL(n,\mathbb{R})\times \mathbb{R}^n \rightarrow  \mathbb{R}^n$ that it is induced by left multiplication of a matrix on a vector. Define a bundle 
$$
	L(M) \times_{\lambda} \mathbb{R}^n \doteq L(M) \times \mathbb{R}^n / \sim_{\lambda},
$$
%In this setup we are thinking of vector components being written in column, then $\lambda_A(v)=A\cdot v$ in coordinates reads $\lambda_A(v)^i=A^i_j v^j$ with the upper index identifying rows and the lower one columns., 
where the equivalence relation is given by $(p,v)\sim (p',v')$ if there is some $A \in GL(n,\mathbb{R})$ such that $p'=R_A(p)$ and $v'=\lambda_A (v)$. Coordinates are of the form $(x^{\mu},v^{\mu})$, and if we change them on $M$ via $x'=x'(x)$ we have induced a transformation $v'^{\mu}=J^{\mu}_{\nu}v^{\nu}$. By (ii) in Proposition \ref{prop_0_construction_of_bundles} $L(M) \times_{\lambda} \mathbb{R}^n\simeq TM$.
% This can be brought a bit further, consider for instance the left actions
% \begin{align}
% 	\lambda_1 &: GL^k(m) \times J^k_0(\mathbb{R}^m,\mathbb{R}^n) \rightarrow J^k_0(\mathbb{R}^m,\mathbb{R}^n): (j^k_0 \alpha, j^k_0f)\mapsto j^k_0(f \circ \alpha^{-1}) ,\\
% 	\lambda_2 &: GL^k(m) \times J^k_0(\mathbb{R}^n,\mathbb{R}^m) \rightarrow J^k_0(\mathbb{R}^n,\mathbb{R}^m): (j^k_0 \alpha, j^k_0g)\mapsto j^k_0(\alpha \circ g),
% \end{align}
% then
% \begin{align*}
% 	& L^k(M) \times_{\lambda_1} J^k_0(\mathbb{R}^m,\mathbb{R}^n) \ni \left[ (x, j^k_0\alpha,j^k_0f) \right]_{\lambda_1}=\left(x,j^k_0(f \circ \alpha^{-1})\right),\\
% 	& L^k(M) \times_{\lambda_2} J^k_0(\mathbb{R}^n,\mathbb{R}^m) \ni \left[ (x, j^k_0\alpha,j^k_0g) \right]_{\lambda_2}=\left(x,j^k_0(\alpha \circ g)\right).
% \end{align*}
% Choosing local coordinates $x^{\mu}$ on $M$, we can give to the two bundles coordinates $(x^{\mu}, f^i_{\mu}, f^i_{\mu \nu},\ldots, f^i_{\mu_1, \ldots, \mu_k}) $ and $(x^{\mu}, g^{\mu}_i, g^{\mu }_{ij},\ldots, f^{\mu}_{i_i \ldots i_k}) $, respectively, while changing coordinates $x'=x'(x)$ in the base manifold, induces a coordinate change via the expressions in $(7)$ and $(8)$, where by $\alpha^{-1}$ we mean the element locally written $(\bar{J}^{\alpha}_{\mu}, \bar{J}^{\alpha}_{\mu \nu}, \ldots)$ identifying with $\bar{J}$ the inverse matrix to the Jacobian $J$ of the coordinate change in $M$. Usually the two are called the covelocity and velocity bundles, and when $n=1$ they reduce to the high order tangent and cotangent bundles.\\

The case of the tangent bundle presented above points to the emergence of a certain pattern: bundles associated with frame bundles are natural. In order to prove it we shall need some intermediate results.\\

\begin{lemma}\label{lemma_0_natural_equivalence_1}
	$L^k: M \mapsto L^k(M)$ is a functor from the category of manifolds and local diffeomorphisms to the category of fiber bundles and local fibered morphisms.
\end{lemma}
\begin{proof}
From Lemma \ref{lemma_0_principal_bndl_frame_k}, we have that for each manifold $M$ (resp. $U\subset M$), $L^k(M)$ (resp. $L^k(U)\subset L^k(M)$) is a principal fiber bundle. Suppose that $\psi \in \mathrm{Diff}_{{loc}}(M)$, then there are open subsets $U$,$V$ of $M$ such that $\psi: U \rightarrow V$ is a diffeomorphism. Set 
$$
    L^k(\varphi):L^k(U) \rightarrow L^k(V):j^k_0 e \mapsto j^k_0(\psi \circ e).
$$ 
By functoriality and covariance of the jet prolongations $j^k_0(\varphi \circ e) = j^k_{e(0)}\varphi \circ j^k_0e$, yielding a commutative diagram:
\begin{center}
\begin{tikzpicture}[every node/.style={midway}]
\matrix[column sep={10em,between origins},
        row sep={3em}] at (0,0)
{ \node(A)   {$L^k(U)$}  ; & \node(B) {$L^k(V)$};\\
  \node(C) {$U$};  & \node(D) {$V$};\\};
\draw[->] (A) -- (B) node[anchor=south] {$L^k(\varphi) $};
\draw[->] (C) -- (A) node[anchor=east]  {$j^k_0e$};
\draw[->] (D) -- (B) node[anchor=west]  {$j^k_0(\varphi \circ e)$};
\draw[->] (C) -- (D) node[anchor=north]  {$\varphi$};
\end{tikzpicture}
\end{center}
In addition, $L^k(\psi)\left( R_{j^k_0\alpha}j^k_0(\psi \circ e)\right)=j^k_0(\psi \circ e \circ \alpha )= R_{j^k_0 \alpha}\left(L^k(\psi)j^k_0e\right)$, implying that $L^k(\psi)$ is even a local principal isomorphism.
\end{proof}

Given a natural bundle functor $\mathfrak{N}$, we say that it has \textit{finite order} $k$ if, given any $M \in \mathrm{Ob}(\mathfrak{Man})$, $\phi,\psi \in \mathrm{Diff}_{\mathrm{loc}}(M) $ with $j^k_x\phi = j^k_x \psi$ in some $x \in M$, we have
$$
	\mathfrak{N}(\phi)(\pi^{-1}(x)) \simeq \mathfrak{N}(\psi)(\pi^{-1}(x))
$$
where as usual $\pi$ is the projection of the fiber bundle $\mathfrak{N}(M)$.

\begin{theorem}\label{thm_0_natural_bundle_characterization}
Let $\mathfrak{N}: \mathfrak{Man}\rightarrow \mathfrak{FBndl}$ be a natural functor of finite order $k$, then for all $M \in \mathrm{Ob}(\mathfrak{Man})$ we have a canonical isomorphism
$$
    \mathfrak{N}(M) \simeq L^k(M) \times_{\lambda} S
$$
 where $S \doteq \pi_{M}^{-1}(0)$ is the standard fiber of the natural bundle $\mathfrak{N}(\mathbb{R}^n)$.
\end{theorem}
\begin{proof}
Denote by $\pi_M:\mathfrak{N}(M)\to M$ the bundle projection. We start by writing down the action of the group $GL^k(n)$ on $S$. Consider two $n$-dimensional manifolds $M$, $N$ and the mapping 
$$
    J^k(M,N)^{(\mathrm{inv})} \times_M \mathfrak{N}(M)\rightarrow \mathfrak{N}(N): \big[j^k_x f, b\big] \mapsto \mathfrak{N}(f)(b)  ,
$$
where as usual the space $J^k(M,N)^{(\mathrm{inv})}$ denotes the space of $k$th order jets of local diffeomorphisms and $\big[j^k_x f, b\big]$ any representative of the equivalence relation $\big[j^k_x f, b\big]\sim_M \big[j^k_{x'} f', b'\big]$ if $x=x',\ \pi_M(b)=\pi_M(b')$. From here on we work locally on the chart $U_{a}\subseteq M$ and assume $M=N$, $U_a=\mathbb{R}^n$; in which case the above mapping takes the form
$$
\begin{aligned}
    J^k(U_a,U_a)^{(\mathrm{inv})} \times_{U_a} \mathfrak{N}(U_a)\equiv J^k(U_a,U_a)^{(\mathrm{inv})} \times S \rightarrow  \mathfrak{N}(U_a): (j^k_x \alpha, y) \mapsto (x,\mathfrak{N}(\alpha)(y)) \equiv (x, \tilde{\alpha}(y)),   
\end{aligned}
$$
where we are trivializing points $b$ of $\mathfrak{N}(U_a)$ as $(x,y)$. Choosing in particular $x=0$, and composing with the fiber projection yields a mapping
\begin{equation}\label{eq_0_natural_left_action}
    \lambda: GL^k(n,\mathbb{R}) \times S \rightarrow S : (j^k_0 \alpha , y) \rightarrow \tilde{\alpha}(y).
\end{equation}

This is a left action since $\lambda_{j^k_0 \alpha}\circ \lambda_{ j^k_0 \beta }(y) = \mathfrak{N}(\alpha)\circ \mathfrak{N}(\beta)(y)=\mathfrak{N}(\alpha \circ \beta)(y)=\lambda_{\alpha \circ \beta}(y)$. Then consider the mapping
\begin{equation}\label{eq_0_associated_natural_isomorphism}
    L^k(M)\times_{\lambda}S \ni ( j^k_xe, y) \mapsto (x,\mathfrak{N}(\alpha)(y)) \in \mathfrak{N}(M),
\end{equation}
where $j^k_xe=(x,j^k_0\alpha)$ for some $j^k_0 \alpha \in GL^k(n,\mathbb{R})$. This is well defined for if we take any other $\alpha'$ with $j^k_0\alpha'=j^k_0\alpha$, by finiteness of the order of $\mathfrak{N}$, $\mathfrak{N}(\alpha)=\mathfrak{N}(\alpha')$. Finally this is an isomorphism since choosing $j^k_x\overline{e}=(x,j^k_0\alpha^{-1})$, and using covariance of the jet and natural functor, we show that \eqref{eq_0_associated_natural_isomorphism} is an isomorphism of bundles projecting over $id_M$.
\end{proof}

Theorem \ref{thm_0_natural_bundle_characterization} can be framed into a more categorical setting: consider therefore the mapping $\lambda:GL^k(m) \times S \rightarrow S$ canonically defined by \eqref{eq_0_natural_left_action}, then we induce a covariant functor
\begin{equation}
\begin{aligned}
	N: & \mathfrak{PB}(GL^k(m)) \rightarrow \mathfrak{FBndl} \\
	  & P \mapsto P\times_{\lambda}S\\
	  & \Psi \mapsto L(\Psi) \equiv [\Psi,id_S]_{\lambda}
\end{aligned}
\end{equation}
where $\mathfrak{PB}(G)$ is the category of principal bundles with group $G$ and local principal bundle isomorphisms, and the action of $L(\Psi)$ on a point $[p,y]_{\lambda} \in P\times_{\lambda} S$ is understood as $[\Psi(p),y]_{\lambda}$. Combining this with Lemma \ref{lemma_0_principal_bndl_frame_k} we see that we can construct a functor $\mathfrak{L} \doteq N\circ L^k: \mathfrak{Man} \rightarrow \mathfrak{FBndl}$.

\begin{theorem}\label{thm_0_natural_functor_characterization}
    There is a bijective correspondence between the set of $k$th order natural bundles on $n$-dimensional manifolds and the associated $S$-bundle to a principal $GL^k(n,\mathbb{R})$-bundle. Moreover this correspondence takes the form of natural transformation between functors. 
\end{theorem}
\begin{proof}
Let $\mathfrak{N}$ a natural bundle functor, fixed $M \in \mathrm{Ob}(\mathfrak{Man})$, we can consider (up to diffeomorphisms) the standard fiber $S \simeq \pi^{-1}(x)$, where $\mathfrak{N}(M)=(B, \pi, M, S)$ is the natural bundle in question. Call $q_{j^k_0e}:L^k(M)\times_{\lambda}S \mid_{j^k_0e} \rightarrow \mathfrak{N}(M)|_{x=e(0)}$ the isomorphism introduced in Theorem \ref{thm_0_natural_bundle_characterization}, and set
$$
	\chi_M: \mathfrak{L}(M) \rightarrow \mathfrak{N}(M):(j^k_0e,y) \mapsto \mathfrak{N}(e)(y).
$$
Suppose that $\psi \in \mathrm{Diff}_{(\mathrm{loc})}(M)$, then 
$$
	\mathfrak{N}(\psi)\circ \chi_M([j^k_0e,y]_{\lambda})= \mathfrak{N}(\psi)\circ \mathfrak{N}(e)(y)=\mathfrak{N}(\psi \circ e)(y)= \chi_N([\psi \circ e, y]_{\lambda})= \chi_N \circ \mathfrak{L}(\psi) ([j^k_0e,y]_{\lambda}).
$$
Equivalently we showed that the following diagram is commutative:
\begin{center}
\begin{tikzpicture}[every node/.style={midway}]
\matrix[column sep={10em,between origins},
        row sep={3em}] at (0,0)
{ \node(A)   {$\mathfrak{L}(M)$}  ; & \node(B) {$\mathfrak{N}(M)$};\\
  \node(C) {$\mathfrak{L}(N)$};  & \node(D) {$\mathfrak{N}(N)$};                 \\};
\draw[->] (A) -- (B) node[anchor=south] {$\chi_M $};
\draw[->] (A) -- (C) node[anchor=east]  {$\mathfrak{L}(\psi)$};
\draw[->] (B) -- (D) node[anchor=west]  {$\mathfrak{N}(\psi)$};
\draw[->] (C) -- (D) node[anchor=north]  {$\chi_N$};
\end{tikzpicture}
\end{center}
Therefore $\chi:\mathfrak{L}\Rightarrow \mathfrak{N}$ is a natural transformation.
\end{proof}

A worthy example of a natural bundle is the bundle of linear connections $\mathrm{Conn}(M)$ of $M$. The latter is defined as the associated bundle to $L^2(M)$ via the action
$$
	\lambda: GL^2(n) \times \mathbb{R}^{n^3} \rightarrow \mathbb{R}^{n^3} : (J^{\alpha}_{\beta },J^{\alpha}_{\beta \gamma}, \Gamma^{\alpha}_{\beta \gamma}) \mapsto J^{\alpha}_{\beta} \left( \Gamma^{\beta}_{\mu \nu} \bar{J}^{\mu}_{\gamma} \bar{J}^{\nu}_{\delta}+ \bar{J}^{\beta}_{\gamma \delta} \right)
$$
where as before $J$ denotes an element of $GL(n)$ and $\bar{J}$ its inverse. Another example is 
\begin{equation}\label{eq_2_H(M)_natural_bundle}
H(M)=\left(HM\doteq \mathrm{Lor}(M)\times_M (M\times\mathbb{R}^2), \pi,M ,S^{2}_{(1,n-1)}(\mathbb{R}^n)\times \mathbb{R}^2 \right),
\end{equation}
where $S^{2}_{(1,n-1)}(\mathbb{R}^m)$ denotes the space of signature $(1,n-1)$ symmetric matrices. Coordinates will be given by fibered local charts $(x^{\alpha},g_{\mu \nu}, m,\kappa)$. We stress that since the coefficients of the metric are symmetric then not all of them are independent, therefore when we refer to $\{g_{\mu\nu}\}$ as coordinates we are really considering just $\{g_{\mu\nu}\}_{1\leq \mu < \nu \leq n}$. Physical scaling on $H(M)$ is then defined as follows: let $\lambda \in \mathbb{R}_+$, then we define a fibered isomorphism projecting onto the identity of $M$ by
\begin{equation}\label{eq_2_physical_scaling}
\begin{aligned}
    S_{\lambda}: & HM \to HM, \quad %: (x,h) \mapsto S_{\lambda} (x,h) \\
    &(x,g_{\mu \nu}, m,\kappa) \mapsto (x, \lambda^{-2} g_{\mu \nu}, \lambda^2m, \kappa).
\end{aligned}
\end{equation}
Here the coordinates $(m,\kappa)$ will later represent the mass and coupling constant on the Klein-Gordon equation on $M$, for now however we take them to just be coordinates of the fiber. We also note that $S_{\lambda}$ induces as well a mapping $:\Gamma^{\infty}(M \leftarrow HM) \rightarrow \Gamma^{\infty}(M \leftarrow HM):h\mapsto h_{\lambda} $ which we denote by the same symbol.

\begin{definition}\label{def_2_natural_bundle_with_scaling}
A natural bundle equipped with \textit{physical scaling} is a natural functor 
$$
    H: \mathfrak{Man} \rightarrow \mathfrak{FBndl},
$$
from the category of manifolds and local diffeomorphisms to the category of fiber bundles and local fiber isomorphisms, with the additional property that for any $M\in Ob(\mathfrak{Man})$ there exists a corresponding physical scaling $S_{\lambda}$ on $HM$ such that for any $\psi \in \mathrm{Diff}_{loc}(M,M')$ and $\lambda \in \mathbb{R}_+$ then 	
$$
	S'_{\lambda} \circ H(\psi)=H(\psi) \circ S_{\lambda}.
$$
\end{definition}

Clearly \eqref{eq_2_physical_scaling} is an example of physical scaling according to Definition \ref{def_2_natural_bundle_with_scaling}.

\begin{definition}\label{def_2_background_geometry}
Let $ H: \mathfrak{Man} \rightarrow \mathfrak{Bndl}$ be the natural scaling bundle described in \eqref{eq_2_H(M)_natural_bundle}, a \textit{background geometry} is a pair $(M,h)$ where $h \in \Gamma^{\infty}(M\leftarrow HM)$ and $(M,g)$ is a time orientable, globally hyperbolic spacetime. Let then $\mathfrak{BckG}$ be the category of background geometry, whose objects are time orientable globally hyperbolic spacetimes and whose morphisms are time respecting, causality preserving smooth isometric embeddings. 
\end{definition}

Note that in the above definition any such arrow $\chi: M\rightarrow M'$ of $\mathfrak{BckG}$, is in particular a local diffeomorphism, so that $H (\chi) $ is defined according to Definition \ref{def_0_natural_bundles}; moreover, the transformation induced on the object $(M',h')$ is given by $\left(M,h=\chi^{*}h'\right)$, where the section $\chi^{*}h'$ is defined as the mapping $H \left(\chi\vert_{H(\chi(M))}\right)^{-1} \circ h' \circ \chi$ guaranteeing the commutativity of the diagram 
\begin{center}
\begin{tikzpicture}[every node/.style={midway}]
\matrix[column sep={10em,between origins},
        row sep={3em}] at (0,0)
{ \node(A)   {$HM$}  ; & \node(B) {$HM'$};\\
  \node(C) {$M$};  & \node(D) {$M'$};                 \\};
\draw[->] (A) -- (B) node[anchor=south] {$H(\chi)$};
\draw[->] (C) -- (A) node[anchor=east]  {$\chi^*h'$};
\draw[->] (D) -- (B) node[anchor=west]  {$h'$};
\draw[->] (C) -- (D) node[anchor=north]  {$\chi$};
\end{tikzpicture}
\end{center} 
We stress that the physical scaling introduced above is consistent with the pullback of sections introduced above, due to commutativity with $H(\chi)$.\\

\subsection{Covariance identity for natural bundles}\label{section_covariance_identity}

In this section, we will present an important result (\Cref{thm_2_covariant_identity}), which states that if a smooth function $f$ from the $r$th order jet prolongation of a natural bundle $J^r(\mathfrak N(M))$ to $\mathbb R$ is diffeomorphism covariant, then it must satisfy some strict conditions: the coordinate dependence is limited to all those coordinates of $J^r(\mathfrak N(M))$ which are tensor-like quantities, furthermore, $f$ can only be a function of all the scalar quantities constructed out of the tensor-like coordinates mentioned earlier. This result, which we will prove for the natural bundle \eqref{eq_2_H(M)_natural_bundle}, will be crucial for the proof of \Cref{thm_2_Moretti_Kavhkine}, which essentially classifies the ambiguities of the definition of Wick powers as scalars constructed out of the matric and the coupling constants.\\  

A useful consequence of the formalism of natural bundles is the possibility of characterizing $\mathrm{Diff}(M)$-invariant functions on the $r$th jet order $J^r\mathfrak N (M)$ for some natural bundle $\mathfrak N(M)$ by means of the natural formalism developed in \cite{fati, kolar2013natural}. Instead of analyzing the general case we are going to study the case $\mathfrak N (M)=HM$ defined in \eqref{eq_2_H(M)_natural_bundle}. There we assign coordinates $(x^{\mu},g_{\mu\nu},m^2,\kappa)$, thus the $r$th order jet prolongation has canonical coordinates 
$$
	(x^{\mu},g_{\mu\nu},\ldots,g_{\mu\nu,\lambda_1\ldots,\lambda_r},m^2,\ldots,m^2_{\lambda_1\ldots,\lambda_r},\kappa,\ldots,\kappa_{\lambda_1\ldots,\lambda_r}).
$$
to derive the key result of covariance it is necessary to choose suitable coordinates in $J^rHM$ other than the above. 
\begin{proposition}\label{proposition_2_covariant_coordinates}
Let $J^rHM$ be the natural bundle described above then, for each $r\in \mathbb R$, there is a coordinates transformation
\begin{center}
\begin{tikzcd}
    (x^{\mu},g_{\mu\nu,\Lambda},m^2_{\Lambda},\kappa_{\Lambda}) \arrow[d] \\
	 (x^{\mu},g_{\mu\nu},d_{(\Lambda}S^{\alpha}_{\beta\mu)}, \nabla_{(\Lambda}R^{\alpha}_{\beta\mu)\nu},\nabla_{(\Lambda)}m^2,\nabla_{(\Lambda)}\kappa)
\end{tikzcd}
\end{center}
where $\Lambda=(\lambda_1,\ldots,\lambda_k)$ is a multi-index with $k=|\Lambda|\leq r-2$, round brackets between indices denotes symmetrization, the coefficients $S$ and $R$ are respectively the coefficients of the Levi-Civita connection and the Riemann tensor induced by $g$; finally, $d_{\mu}(\cdot) dx^{\mu}$ denotes the horizontal differential on $j^rHM$.
\end{proposition}

\begin{proof}
We shall divide this argument into four steps: we shall analyse the cases $r=1, 2, 3$ and then use those to formulate and induction hypothesis. We shall also restrict to the simpler case $HM=\mathrm{Lor}(M)$ the bundle of Lorentzian metrics on $M$ and come back to the general case at the end. 

\textbf{Step 1.}  We start with coordinates $(x^{\mu},g_{\mu\nu},d_{\alpha}g_{\mu\nu})$ on $J^1\mathrm{Lor}(M)$. Recall that the metric $g$ induces its Christoffel symbols, $S^{\alpha}_{\mu\nu}=\frac{1}{2}g^{\alpha \beta}(-d_{\beta}g_{\mu \nu}+d_{\mu}g_{\nu \beta}+d_{\nu}g_{\beta\mu})$, then 
$$
\begin{aligned}
	2g_{\alpha\beta}S^{\alpha}_{\mu\nu}+ 2g_{\alpha\mu}S^{\alpha}_{\beta\nu} &= \textcolor{red}{d_{\mu}g_{\nu\beta}}+ d_{\nu}g_{\beta\mu}-\textcolor{blue}{d_{\beta}g_{\mu\nu}}\\
	& \quad + \textcolor{blue}{d_{\beta}g_{\nu\mu}}+d_{\nu}g_{\mu\beta}-\textcolor{red}{d_{\mu}g_{\beta\nu}}\\
	&= 2d_{\nu}g_{\mu\beta}.
\end{aligned}
$$
There are in total $n^2\frac{n+1}{2}$ independent components of $d_{\nu}g_{\mu\beta}$ and of $S^{\alpha}_{\mu\nu}$, as a result
\begin{equation}\label{eq_2_covariant_coord_transformation_r=1}
\begin{cases}
    2d_{\nu}g_{\mu\beta}=2g_{\alpha\beta}S^{\alpha}_{\mu\nu}+ 2g_{\alpha\mu}S^{\alpha}_{\beta\nu}\\
    S^{\alpha}_{\mu\nu}=\frac{1}{2}g^{\alpha \beta}(-d_{\beta}g_{\mu \nu}+d_{\mu}g_{\nu \beta}+d_{\nu}g_{\beta\mu})
\end{cases}
\end{equation}
are one the inverse to each other and produce a well defined change of coordinates.

\textbf{Step 2.}  We start with coordinates $(x^{\mu},g_{\mu\nu},d_{\alpha}g_{\mu\nu},d_{\alpha\beta}g_{\mu\nu})$ on $J^2\mathrm{Lor}(M)$, by \eqref{eq_2_covariant_coord_transformation_r=1}, we can as well assume we have initial coordinates $(x^{\mu},g_{\mu\nu},S^{\alpha}_{\beta\nu},d_{\alpha\beta}g_{\mu\nu})$. Recall that the Riemann tensor components are expressed as $R^{\alpha}_{\beta \mu\nu}=d_{\mu}S^{\alpha}_{\beta \nu}+S^{\alpha}_{\sigma\mu}S^{\sigma}_{\beta\nu}-d_{\nu}S^{\alpha}_{\beta \mu}-S^{\alpha}_{\sigma\nu}S^{\sigma}_{\beta\mu}$. Next define
\begin{equation}\label{eq_2_S_1st_derivative_1}
	S^{\alpha}_{\beta\mu\nu}\doteq d_{(\mu}S^{\alpha}_{\beta\nu)}.
\end{equation}
we stress that round brackets means that all indices between them are symmetrized. 
$$
\begin{aligned}
    3S^{\alpha}_{\beta\mu\nu}-2R^{\alpha}_{(\beta\nu)\mu}&=d_{\mu}S^{\alpha}_{\beta\nu}+\textcolor{red}{d_{\beta}S^{\alpha}_{\nu\mu}}+\textcolor{blue}{d_{\nu}S^{\alpha}_{\mu\beta}}\\ & \quad -\textcolor{blue}{d_{\nu}S^{\alpha}_{\beta\mu}}+d_{\mu}S^{\alpha}_{\beta\nu}-\textcolor{red}{d_{\beta}S^{\alpha}_{\nu\mu}}+d_{\mu}S^{\alpha}_{\beta\nu} + O(j^1g)\\
    &= 3d_{\mu}S^{\alpha}_{\beta\nu} + O(j^1g)\\
\end{aligned}
$$
where we denoted by $O(j^1g)$ in the above formula terms which depends at most on $j^1g$ and for which \eqref{eq_2_covariant_coord_transformation_r=1} already gives the sought transformation. Summing up we have  
\begin{equation}\label{eq_2_S_1st_derivative_2}
    d_{\mu}S^{\alpha}_{\beta\nu}=S^{\alpha}_{\beta\mu\nu}-\frac{2}{3}R^{\alpha}_{(\beta\nu)\mu}+\ldots
\end{equation}
where $\ldots$ refers to terms living in the first order jet bundle. Next we compute
$$
\begin{aligned}
	d_{\alpha\beta}g_{\mu\nu} & = d_{\beta}\left(g_{\sigma\alpha}S^{\sigma}_{\mu \nu}+g_{\sigma\mu}S^{\sigma}_{\nu\alpha}\right)= g_{\sigma\alpha} d_{\beta}S^{\sigma}_{\mu \nu}+g_{\sigma\mu} d_{\beta}S^{\sigma}_{\nu\alpha}+O(j^1g)\\
	& =  g_{\sigma\alpha}\left(S^{\sigma}_{\mu\beta\nu}-\frac{2}{3}R^{\sigma}_{(\mu\nu)\beta} \right) +  g_{\sigma\mu}\left(S^{\sigma}_{\alpha\beta\nu}-\frac{2}{3}R^{\sigma}_{(\alpha\nu)\beta} \right)  O(j^1g)\\
	&= 2 g_{\sigma(\alpha}S^{\sigma}_{\mu)\beta\nu} -\frac{1}{3} \left( \textcolor{red}{R_{\alpha\mu\nu\beta}}+R_{\alpha\nu\mu\beta}+\textcolor{red}{R_{\mu\alpha\nu\beta}}+R_{\mu\nu\alpha\beta }\right)+O(j^1g)\\
	&= 2 g_{\sigma(\alpha}S^{\sigma}_{\mu)\beta\nu} -\frac{1}{3} \left( R_{\alpha\nu\mu\beta}+R_{\alpha\beta\mu\nu }\right)+O(j^1g).
\end{aligned}
$$
Finally, $d_{\alpha\beta}g_{\mu\nu}$ has $n^2\left(\frac{n+1}{2}\right)^2$ independent components, while $dS$ and the Riemann tensor $R$ have respectively $n\binom{n+2}{3}$ and $\frac{n^2}{12}(n+1)^2$ independent components. It follows easily that $n^2\left(\frac{n+1}{2}\right)^2=n\binom{n+2}{3}+\frac{n^2}{12}(n+1)^2$, therefore we conclude that the mapping 
$$
(x^{\mu},g_{\mu\nu},d_{\alpha}g_{\mu\nu},d_{\alpha\beta}g_{\mu\nu})\mapsto(x^{\mu},g_{\mu\nu},S^{\alpha}_{\mu\nu},S^{\alpha}_{\mu\beta\nu},R^{\alpha}_{\beta\mu\nu}),
$$
where up to order $1$ are defined by \eqref{eq_2_covariant_coord_transformation_r=2} and
\begin{equation}\label{eq_2_covariant_coord_transformation_r=2}
\begin{cases}
    d_{\alpha\beta}g_{\mu\nu}=2 g_{\sigma(\alpha}S^{\sigma}_{\mu)\beta\nu} -\frac{1}{3} \left( R_{\alpha\nu\mu\beta}+R_{\alpha\beta\mu\nu }\right)+O(j^1g),\\
    S^{\alpha}_{\beta\mu\nu}\doteq d_{(\mu}S^{\alpha}_{\beta\nu)},\quad R^{\alpha}_{\beta \mu\nu}=d_{\mu}S^{\alpha}_{\beta \nu}+S^{\alpha}_{\sigma\mu}S^{\sigma}_{\beta\nu}-d_{\nu}S^{\alpha}_{\beta \mu}-S^{\alpha}_{\sigma\nu}S^{\sigma}_{\beta\mu}
\end{cases}
\end{equation}
is a well defined coordinate transformation.

\textbf{Step 3.} Then again, by \eqref{eq_2_covariant_coord_transformation_r=1}, \eqref{eq_2_covariant_coord_transformation_r=2} we can assume to have initial coordinates 
$$
    (x^{\mu},g_{\mu\nu},S^{\alpha}_{\mu\nu},S^{\alpha}_{\mu\beta\nu},R^{\alpha}_{\beta\mu\nu},d_{\alpha\beta\gamma}g_{\mu\nu})
$$
on $J^3\mathrm{Lor}(M)$. From \eqref{eq_2_covariant_coord_transformation_r=2},
$$
\begin{aligned}
	d_{\alpha\beta\gamma}g_{\mu\nu}&= 2d_{\gamma}S_{(\alpha\mu)\beta\nu}-\frac{1}{3}\left( d_{\gamma}R_{\alpha\nu\mu\beta}+d_{\gamma}R_{\mu\nu\alpha\beta} \right)+ O(j^2g)\\
	&= 2d_{\gamma}S_{(\alpha\mu)\beta\nu}-\frac{1}{3}\left( \nabla_{\gamma}R_{\alpha\nu\mu\beta}+\nabla_{\gamma}R_{\mu\nu\alpha\beta} \right)+ O(j^2g).
\end{aligned}
$$
Then define $S^{\alpha}_{\gamma\mu\beta\nu}\doteq d_{(\gamma}S_{\alpha\mu\beta\nu)}$, which we can express as
$$
\begin{aligned}
	4S^{\alpha}_{\gamma\mu\beta\nu}& = d_{\gamma}S^{\alpha}_{\mu\beta\nu}+d_{\mu}S^{\alpha}_{\beta\nu\gamma}+d_{\beta}S^{\alpha}_{\nu\gamma\mu}+d_{\nu}S^{\alpha}_{\gamma\mu\beta}\\
	&= d_{\gamma}S^{\alpha}_{\mu\beta\nu}+\frac{1}{3} \left( d_{\mu\nu}S^{\alpha}_{\beta\gamma}+d_{\mu\beta}S^{\alpha}_{\gamma\mu}+d_{\mu\gamma}S^{\alpha}_{\nu\beta}\right)\\
	&\quad+\frac{1}{3} \left( d_{\beta\gamma}S^{\alpha}_{\nu\mu}+d_{\beta\nu}S^{\alpha}_{\mu\gamma}+d_{\beta\mu}S^{\alpha}_{\gamma\nu}\right)\\
	&\quad +\frac{1}{3} \left( d_{\nu\mu}S^{\alpha}_{\beta\gamma}+d_{\nu\beta}S^{\alpha}_{\gamma\mu}+d_{\nu\gamma}S^{\alpha}_{\nu\beta}\right)\\
	&= d_{\gamma}S^{\alpha}_{\mu\beta\nu}+\frac{1}{3}d_{\gamma} \left(d_{\nu}S^{\alpha}_{\nu\beta} +d_{\mu}S^{\alpha}_{\nu\beta}+d_{\beta}S^{\alpha}_{\nu\mu}\right)\\ &\quad+\frac{2}{3} \left( d_{\mu\nu}S^{\alpha}_{\beta\gamma}+d_{\mu\beta}S^{\alpha}_{\gamma\mu}+d_{\beta\nu}S^{\alpha}_{\mu\gamma}\right)\\
	&= 2d_{\gamma}S^{\alpha}_{\mu\beta\nu}+\frac{2}{3} \left( d_{\mu\nu}S^{\alpha}_{\beta\gamma}+d_{\mu\beta}S^{\alpha}_{\nu\gamma}+d_{\beta\nu}S^{\alpha}_{\mu\gamma}\right).\\
\end{aligned}
$$
From the definition of the Riemann tensor we get  $d_{\gamma}R^{\alpha}_{\nu\mu\beta}=d_{\gamma\mu}S^{\alpha}_{\nu\beta}-d_{\gamma\beta}S^{\alpha}_{\nu\mu}+O(j^2g)$, which plugged in the above expression yields
$$
\begin{aligned}
	4S^{\alpha}_{\gamma\mu\beta\nu}&= 2d_{\gamma}S^{\alpha}_{\mu\beta\nu}+\frac{2}{3} \left( d_{\mu\gamma}S^{\alpha}_{\beta\nu}-d_{\mu}R^{\alpha}_{\beta\gamma\nu}+d_{\beta\gamma}S^{\alpha}_{\nu\beta}-d_{\beta}R^{\alpha}_{\nu\gamma\mu}+d_{\gamma\nu}S^{\alpha}_{\mu\beta}+d_{\nu}R^{\alpha}_{\mu\beta\gamma}\right)+O(j^2g)\\
	&=4d_{\gamma}S^{\alpha}_{\mu\beta\nu}+\frac{2}{3} \left(d_{\mu}R^{\alpha}_{\beta\nu\gamma}+d_{\beta}R^{\alpha}_{\nu\mu\gamma}+d_{\nu}R^{\alpha}_{\mu\beta\gamma}\right)+O(j^2g)\\
	&=4d_{\gamma}S^{\alpha}_{\mu\beta\nu}+\frac{2}{3} \left(\nabla_{\mu}R^{\alpha}_{\beta\nu\gamma}+\nabla_{\beta}R^{\alpha}_{\nu\mu\gamma}+\nabla_{\nu}R^{\alpha}_{\mu\beta\gamma}\right)+O(j^2g),
\end{aligned}
$$
applying Bianchi identities simplifies the last term, and we get
\begin{equation}\label{eq_2_S_1st_derivative_3}
S^{\alpha}_{\gamma\mu\beta\nu}=d_{\gamma}S^{\alpha}_{\mu\beta\nu}+O(j^2g).
\end{equation} 
Counting the independent components of $d_{\alpha\beta\gamma}g_{\mu\nu}$, $S^{\alpha}_{\gamma\mu\beta\nu}$ and $\nabla_{\gamma}R^{\alpha}$ yields respectively $\binom{n+1}{2}\binom{n+2}{3}$,  $n\binom{n+3}{4}$ and $\frac{n+2}{2}\frac{n^2}{12}(n+1)^2$. Summing the last two numebers gives the first, therefore the mapping
$$
	(x^{\mu},g_{\mu\nu},d_{\alpha}g_{\mu\nu},d_{\alpha\beta}g_{\mu\nu},d_{\alpha\beta\gamma}g_{\mu\nu})\mapsto(x^{\mu},g_{\mu\nu},S^{\alpha}_{\mu\nu},S^{\alpha}_{\mu\beta\nu},S^{\alpha}_{\gamma\mu\beta\nu},R^{\alpha}_{\beta\mu\nu},\nabla_{\gamma}R^{\alpha}_{\beta\mu\nu})
$$
is a proper change of coordinates, with
\begin{equation}\label{eq_2_covariant_coord_transformation_r=3}
\begin{cases}
d_{\alpha\beta\gamma}g_{\mu\nu}=2g_{\sigma(\alpha}S^{\sigma}_{\beta)\gamma\mu\nu}-\frac{1}{3}\left(\nabla_{\gamma}R_{\alpha\nu\mu\beta}+\nabla_{\gamma}R_{\mu\nu\alpha\beta}\right)+O(j^2g).\\
S^{\alpha}_{\gamma\mu\beta\nu}=d_{\gamma}S^{\alpha}_{\mu\beta\nu}+O(j^2g),\quad d_{\gamma}R_{\alpha}{\nu\mu\beta} =\nabla_{\gamma}R_{\alpha}{\nu\mu\beta}+ O(j^2g)
\end{cases}
\end{equation}

\textbf{Step 4.}  

By induction, we assume that up to $r-1 \geq 3$, we have a one to one coordinate transformation
\begin{equation}\label{eq_2_covariant_coord_transformation_r=r-1}
\begin{cases}
d_{\alpha\beta\Lambda}g_{\mu\nu}=2g_{\sigma(\alpha}S^{\sigma}_{\beta)\Lambda\mu\nu}-\frac{1}{3}\left(\nabla_{\Lambda}R_{\alpha\nu\mu\beta}+\nabla_{\Lambda}R_{\mu\nu\alpha\beta}\right)+O(j^{r-2}g).\\
d_{\beta}S^{\alpha}_{\Lambda\mu\nu}=S^{\alpha}_{\beta\Lambda\mu\nu}+O(j^{r-2}g),\quad d_{\Lambda}R^{\alpha}_{\nu\mu\beta} =\nabla_{\Lambda}R^{\alpha}_{\nu\mu\beta}+ O(j^{r-2}g)
\end{cases}
\end{equation}
where $|\Lambda|=r-3$. Using the induction hypothesis $d_{\beta}S^{\alpha}_{\Lambda\mu\nu}=S^{\alpha}_{\beta\Lambda\mu\nu}+O(j^{r-2}g)$ and repeating a similar argument to the one made for the case $r=3$, one can show that 
\begin{equation}\label{eq_2_S_1st_derivative_r}
	S^{\alpha}_{\mu_0\ldots\mu_{r}}\doteq d_{(\mu_0}S^{\alpha}_{\mu_1\ldots\mu_r)}= d_{\mu_0}S^{\alpha}_{\mu_1\ldots\mu_r}+O(j^{r-1}g),
\end{equation}
Where similarly to what we did before, we denote by $O(j^{r-1}g)$ for a polynomial function of the coordinates depending up to the $r-1$th order derivatives of the metric. Then we have
$$
d_{\lambda_1\ldots\lambda{r}}g_{\mu\nu}=2g_{\sigma(\lambda_1}S^{\sigma}_{\mu)\nu\lambda_2\ldots\lambda_r}+\frac{2}{3} \nabla_{\lambda_3}\nabla_{\lambda_4\ldots\lambda_{r}}R_{\nu(\lambda_1\mu)\lambda_2}+O(j^{r-1}g),
$$
we can recast the above equation $\nabla_{\lambda_3}\nabla_{\lambda_4\ldots\lambda_{r}}R_{\nu(\lambda_1\mu)\lambda_2}=\nabla_{\lambda_3\ldots\lambda_{r}}R_{\nu(\lambda_1\mu)\lambda_2}+O(j^{r-2}g)$ using the identity $[\nabla_{\alpha},\nabla_{\beta}]=R^{\cdot}_{\cdot\alpha\beta}$
% $$
% \nabla_{\alpha}\nabla_{\beta_1,\ldots,\beta_q}=\nabla_{\alpha\beta\gamma}+\frac{1}{6}(2[\nabla_{\alpha},\nabla_{\beta}]\nabla_{\gamma}+2[\nabla_{\alpha},\nabla_{\gamma}]\nabla_{\beta}+ \nabla_{\beta}[\nabla_{\gamma},\nabla_{\alpha}]+ \nabla_{\gamma}[\nabla_{\beta},\nabla_{\alpha}]).
% $$
we get for $|\Lambda|=r-2$, 
\begin{equation}
d_{\alpha\beta\Lambda}g_{\mu\nu}=2g_{\sigma(\alpha}S^{\sigma}_{\beta)\Lambda\mu\nu}-\frac{1}{3}\left(\nabla_{\Lambda}R_{\alpha\nu\mu\beta}+\nabla_{\Lambda}R_{\mu\nu\alpha\beta}\right)+O(j^{r-1}g)
\end{equation}
defines a proper coordinate transformation for the highest order derivatives of the metric. \\

In our case where two scalar functions are present, we notice that each symmetrized derivative, \textit{e.g.} $d_{\Lambda}\kappa$, can be written as $\nabla_{\Lambda}\kappa$, using the fact that $\nabla_{\lambda}\kappa=d_{\lambda}\kappa$ and for all $\Lambda, \ \Lambda'$, $|\Lambda|=k\geq 2$, $|\Lambda'|=k-1$, we have $\nabla_{\Lambda}\kappa-d_{\Lambda}\kappa = O(\nabla_{\Lambda'}\kappa,j^{|\Lambda|-1}g)$.
\end{proof}

% In conclusion we can use local coordinates $\left(x^{\mu},g_{\mu\nu},S^{\beta}_{\mu\nu},S^{\beta}_{\alpha\mu\nu},R^{\beta}_{\alpha\mu\nu},\ldots,S^{\beta}_{\alpha_1\ldots\alpha_r\mu},\nabla_{\alpha_3\ldots\alpha_r}R^{\beta}_{\alpha_1\mu\alpha_2}\right)$ on $J^r\mathrm{Lor}(M)$.

\begin{theorem}\label{thm_2_covariant_identity}
    If $f:J^rHM\rightarrow \mathbb{R}$ is a diffeomorphism covariant function, then using the coordinates introduced in Proposition \ref{proposition_2_covariant_coordinates} it satisfies the following relations:
\begin{itemize}
\item[$(i)$] $$ \frac{\partial f}{\partial x^{\sigma}}=0;$$ 
\item[$(ii)$] $$ \frac{\partial f}{\partial S^{\alpha}_{\mu\nu}}=\frac{\partial f}{\partial S^{\alpha}_{\beta\mu\nu}}=\ldots=\frac{\partial f}{\partial S^{\alpha}_{\lambda_1\ldots\lambda_r \mu}}=0;$$
\item[$(iii)$] furthermore, it holds 
$$
\begin{aligned}
    0&=2\frac{\partial f}{\partial g_{\sigma\beta}}f g_{\rho\beta} + \frac{\partial f}{\partial R^{\alpha}_{\beta\mu\nu}}\Big(-\delta^{\alpha}_{\rho}R^{\sigma}_{\beta\mu\nu}+\delta^{\sigma}_{\beta}R^{\alpha}_{\rho\mu\nu}+2\delta^{\sigma}_{\mu}R^{\alpha}_{\beta\rho\nu}\Big) +\ldots\\
    &\quad+\frac{\partial f}{\partial \nabla_{\Lambda}R^{\alpha}_{\beta\mu\nu}} \Big(-\delta^{\alpha}_{\sigma}\nabla_{\Lambda}R^{\sigma}_{\beta\nu\mu} +\delta^{\sigma}_{\beta}\nabla_{\Lambda}R^{\alpha}_{\rho\nu\mu}+2\delta^{\sigma}_{\nu}\nabla_{\Lambda}R^{\alpha}_{\beta\rho\mu} +(r-3)\delta^{\sigma}_{\lambda_1}\nabla_{\rho\Hat{\Lambda}}R^{\alpha}_{\beta\nu\mu}\Big)\\
    &\quad+\frac{\partial f}{\partial \nabla_{\theta}m^2} \Big(\delta^{\sigma}_{\theta}\nabla_{\rho}m^2\Big)+ \ldots+\frac{\partial f}{\partial \nabla_{\Theta}m^2} \Big(-r\delta^{\sigma}_{\theta_1}\nabla_{\rho \Hat{\Theta}}m^2\Big)\\
    &\quad+\frac{\partial f}{\partial \nabla_{\theta}\kappa} \Big(\delta^{\sigma}_{\theta}\nabla_{\rho}\kappa\Big)+ \ldots+\frac{\partial f}{\partial \nabla_{\Theta}\kappa}\Big(-r\delta^{\sigma}_{\theta_1}\nabla_{\rho \Hat{\Theta}}\kappa\Big)
\end{aligned}
$$
where $\Lambda=(\lambda_1,\ldots,\lambda_{r-3})$ is a multi-index of length $r-3$ and $\Hat{\Lambda}=(\lambda_2,\ldots,\lambda_{r-3})$.
\end{itemize} 
\end{theorem}

Condition $(iii)$ essentially states that the function $f$ is a function whose arguments are all covariant scalars that can be constructed out of the coordinated of Proposition \ref{proposition_2_covariant_coordinates}. These are always finite such objects.%, for example, when $r=1$ we have $g^{\mu\nu}g_{\mu\nu}\equiv n$, %\ \mathrm{det}(g)=\epsilon^{\mu_1 \cdots \mu_n}g^{\mu_1 1}\cdots g^{\mu_nn}$, $R=g^{\mu\nu}R^{\alpha}_{\mu\alpha\nu}$, $R^{\mu\nu}R_{\mu\nu}$, $R^{\alpha}_{\mu\beta\nu}R_{\alpha}^{\mu\beta\nu}$

\begin{proof}
Then again, for simplicity's sake, we study only the case where $f$ is a smooth function in $J^r\mathrm{Lor}(M)$ and briefly comment at the end how the adapt this argument in the general case.\\

Recall that since both the jet functor and the natural functor $M \to \mathrm{Lor}(M)$ are covariant, then $N(M)=J^r\mathrm{Lor}(M)$ is a natural bundle as well, thus, by Definition \ref{def_0_natural_bundles}, we can lift any $\chi\in \mathrm{Diff}(M)$ to an isomorphism $N\chi$ of $J^r\mathrm{Lor}(M)$. $f$ is covariant whether $j^r \mathrm{Lor}(\chi)^{*}f=f$ for all $\chi\in \mathrm{Diff}(M)$. If $\chi_s$ is a one parameter group of diffeomorphisms of $M$, by $(iii)$ in Definition \ref{def_0_natural_bundles} we can lift it to a one parameter isomorphism of $\mathrm{Lor}(M)$ and calling $X$, $\Hat{X}$ the respective infinitesimal generators, then
\begin{equation}\label{eq_2_covariance_identity}
	j^r \mathrm{Lor}(\chi)^{*}f=f \Leftrightarrow \mathscr{L}_{j^r\Hat{X}}f=0 \Leftrightarrow i_{X}d_Hf=i_{j^r\mathscr{L}_{\Hat{X}}g}d_Vf,
\end{equation}
where $\mathscr L$ denotes the Lie derivative. In coordinates, $\Hat{X}=X^{\mu}\partial_{\mu}+X_{\mu\nu}\partial^{\mu\nu}$, we can then lift it to a vector field in $\mathfrak{X}\left(J^r \mathrm{Lor}(M)\right)$ by
$$
\begin{aligned}
\Hat{X}&=X^{\mu}\partial_{\mu}+X_{\mu\nu}\partial^{\mu\nu}+X^{\beta}_{\mu\nu}\frac{\widetilde\partial}{\widetilde\partial S^{\beta}_{\mu\nu}}+\widetilde{X}^{\beta}_{\alpha\mu\nu}\frac{\widetilde\partial}{\widetilde\partial S^{\beta}_{\alpha\mu\nu}}+{X}^{\beta}_{\alpha\mu\nu}\frac{\partial}{\partial R^{\beta}_{\alpha\mu\nu}}+\ldots\\ &\quad+\widetilde{X}^{\beta}_{\lambda_1\ldots\lambda_r\mu}\frac{\widetilde\partial}{\widetilde\partial S^{\beta}_{\lambda_1\ldots\lambda_r\mu}}+{X}^{\beta}_{\lambda_1\ldots\lambda_r\mu}\frac{\partial}{\partial\nabla_{\lambda_3\ldots\lambda_r}R^{\beta}_{\lambda_1\mu\lambda_2}}.
\end{aligned}
$$
Calculations of the necessary Lie derivatives yields
$$
\begin{aligned}
&\mathscr{L}_{\Hat{X}}g_{\alpha\beta}=2\nabla_{(\alpha}X^{\cdot}_{\beta)}\\
&\mathscr{L}_{j^1\Hat{X}}S^{\alpha}_{\mu\nu}=\nabla_{\mu\nu}X^{\alpha}-R^{\alpha}_{(\mu\nu)\rho}X^{\rho}\\
&\mathscr{L}_{j^2\Hat{X}}S^{\alpha}_{\beta\mu\nu}=\nabla_{\beta\mu\nu}X^{\alpha}+O(\nabla X,X)\\
&\mathscr{L}_{j^2\Hat{X}}R^{\alpha}_{\beta\mu\nu}=\nabla_{\sigma}R^{\alpha}_{\beta\mu\nu}X^{\sigma}-\nabla_{\sigma}X^{\alpha}R^{\sigma}_{\beta\mu\nu}+\nabla_{\beta}X^{\sigma}R^{\alpha}_{\sigma\mu\nu}+\nabla_{\mu}X^{\sigma}R^{\alpha}_{\beta\sigma\nu}+\nabla_{\nu}X^{\sigma}R^{\alpha}_{\beta\mu\sigma}\\
&\ldots \\
&\mathscr{L}_{j^r\Hat{X}}S^{\alpha}_{\lambda_1\ldots\lambda_r\mu}=\nabla_{\lambda_1\ldots\lambda_r\mu}X^{\alpha}+O(\nabla^{r-1}X,\ldots,X)\\
&\mathscr{L}_{j^r\Hat{X}}\nabla_{\lambda_3\ldots,\lambda_r}R^{\alpha}_{\lambda_1\mu\lambda_2}=X^{\sigma}\nabla_{\sigma}\nabla_{\lambda_3\ldots,\lambda_r}R^{\alpha}_{\lambda_1\mu\lambda_2}-\nabla_{\lambda_3\ldots,\lambda_r}R^{\sigma}_{\lambda_1\mu\lambda_2}\nabla_{\sigma}X^{\alpha}+\nabla_{\lambda_3\ldots,\lambda_r}R^{\alpha}_{\sigma\mu\lambda_2}\nabla_{\lambda_1}X^{\sigma}\\ &\space\ \space\ \space\ \space\ \space\ \space\ \space\ \space\ \space\ \space\ \space\ \space\ \space\ \space\ \space\ \space\ \space\ \space\ \space\ \space\ \space\ \space\ \space\ \space\ \space\ \space\ \space\ \space\ \space\ + \ldots+\nabla_{\lambda_3\ldots\lambda_{r-1}\sigma}R^{\alpha}_{\lambda_1\mu\lambda_2}\nabla_{\lambda_r}X^{\sigma} +\nabla_{\lambda_3\ldots\lambda_{r }}R^{\alpha}_{\lambda_1\sigma\lambda_2}\nabla_{\mu}X^{\sigma} 
\end{aligned}
$$
Expanding \eqref{eq_2_covariance_identity} yields 
$$
\begin{aligned}
	&X^{\sigma}\left(\partial_{\sigma}f + \partial^{\alpha\beta}f \partial_{\sigma}g_{\alpha\beta}+\partial^{\mu\nu}_{\alpha}f\partial_{\sigma}S^{\alpha}_{\mu\nu}+\widetilde{\partial}_{\alpha}^{\beta\mu\nu}f \partial_{\sigma} S^{\alpha}_{\beta\mu\nu} +{\partial}_{\alpha}^{\beta\mu\nu}f\partial_{\sigma}R^{\alpha}_{\beta\mu\nu}+ \ldots \right.\\& \left.+\widetilde{\partial}_{\alpha}^{\lambda_1\ldots\lambda_r\mu}f \partial_{\sigma}S^{\alpha}_{\lambda_1\ldots\lambda_r\mu}+{\partial}_{\alpha}^{\lambda_1\ldots\lambda_r\mu}f \partial_{\sigma}\nabla_{\lambda_3\ldots\lambda_{r }}R^{\alpha}_{\lambda_1\mu\lambda_2}\right)=\\
	&2\nabla_{(\alpha}X^{\cdot}_{\beta)}\partial^{\alpha\beta}f+\partial^{\mu\nu}_{\alpha}f\left(\nabla_{\mu\nu}X^{\alpha}+o(X)\right)+\widetilde{\partial}_{\alpha}^{\beta\mu\nu}f\left(\nabla_{\beta\mu\nu}X^{\alpha}+O(\nabla X,X)\right)\\
	&+{\partial}_{\alpha}^{\beta\mu\nu}f\left(\nabla_{\sigma}R^{\alpha}_{\beta\mu\nu}X^{\sigma}-\nabla_{\sigma}X^{\alpha}R^{\sigma}_{\beta\mu\nu}+\nabla_{\beta}X^{\sigma}R^{\alpha}_{\sigma\mu\nu}+\nabla_{\mu}X^{\sigma}R^{\alpha}_{\beta\sigma\nu}+\nabla_{\nu}X^{\sigma}R^{\alpha}_{\beta\mu\sigma}\right)+\\
	&\widetilde{\partial}_{\alpha}^{\lambda_1\ldots\lambda_r\mu}f\left(\nabla_{\lambda_1\ldots\lambda_r\mu}X^{\alpha}+O(\nabla^{r-1}X,\ldots,X)\right)+{\partial}_{\alpha}^{\lambda_1\ldots\lambda_r\mu}f \\
	&\left(X^{\sigma}\nabla_{\sigma}\nabla_{\lambda_3\ldots,\lambda_r}R^{\alpha}_{\lambda_1\mu\lambda_2}-\nabla_{\lambda_3\ldots,\lambda_r}R^{\sigma}_{\lambda_1\mu\lambda_2}\nabla_{\sigma}X^{\alpha}+\nabla_{\lambda_3\ldots,\lambda_r}R^{\alpha}_{\sigma\mu\lambda_2}\nabla_{\lambda_1}X^{\sigma}\right.\\ &\left.+ \ldots+\nabla_{\lambda_3\ldots\lambda_{r-1}\sigma}R^{\alpha}_{\lambda_1\mu\lambda_2}\nabla_{\lambda_r}X^{\sigma} +\nabla_{\lambda_3\ldots\lambda_{r }}R^{\alpha}_{\lambda_1\sigma\lambda_2}\nabla_{\mu}X^{\sigma} \right).
\end{aligned}
$$
Upon rewriting the partial derivative as covariant differential, we find
$$
\begin{aligned}
&X^{\sigma}\left(\partial_{\sigma}f + 2\partial^{\alpha\beta}f S^{\rho}_{\alpha\sigma}g_{\rho\beta}+\partial^{\mu\nu}_{\alpha}f\partial_{\sigma}S^{\alpha}_{\mu\nu}+\widetilde{\partial}_{\alpha}^{\beta\mu\nu}f \left(S^{\alpha}_{\sigma\beta\mu\nu}+O(j^2g)\right) \right. \\ 
&\left.+{\partial}_{\alpha}^{\beta\mu\nu}f\left(\textcolor{red}{\nabla_{\sigma}R^{\alpha}_{\beta\mu\nu}}-S^{\alpha}_{\rho\sigma}R^{\rho}_{\beta\mu\nu}+S^{\rho}_{\beta\sigma}R^{\alpha}_{\rho\mu\nu}+2S^{\rho}_{\mu\sigma}R^{\alpha}_{\beta\rho\nu}\right)+ \ldots \right.\\
& \left.+\widetilde{\partial}_{\alpha}^{\lambda_1\ldots\lambda_r\mu}f \left(S^{\alpha}_{\sigma\lambda_1\ldots\lambda_r\mu}+O(j^{r})\right)+{\partial}_{\alpha}^{\lambda_1\ldots\lambda_r\mu}f \left(\textcolor{blue}{\nabla_{\sigma}\nabla_{\lambda_3\ldots\lambda_{r }}R^{\alpha}_{\lambda_1\mu\lambda_2}}-S^{\alpha}_{\rho\sigma}\nabla_{\lambda_3\ldots\lambda_{r }}R^{\sigma}_{\lambda_1\mu\lambda_2}\right. \right.\\ 
& \left.\left.  +S^{\rho}_{\lambda_1\sigma}\nabla_{\lambda_3\ldots\lambda_{r }}R^{\alpha}_{\rho\mu\lambda_2}+2S^{\rho}_{\mu\sigma}\nabla_{\lambda_3\ldots\lambda_{r}}R^{\alpha}_{\lambda_1\rho\lambda_2} +(r-3)S^{\rho}_{\lambda_3\sigma}\nabla_{\rho\ldots\lambda_{r}}R^{\alpha}_{\lambda_1\mu\lambda_2}\right) \right)=\\
&2\nabla_{(\alpha}X^{\cdot}_{\beta)}\partial^{\alpha\beta}f+\partial^{\mu\nu}_{\alpha}f\left(\nabla_{\mu\nu}X^{\alpha}+O(X)\right)+\widetilde{\partial}_{\alpha}^{\beta\mu\nu}f\left(\nabla_{\beta\mu\nu}X^{\alpha}+o(\nabla X,X)\right)\\
&+{\partial}_{\alpha}^{\beta\mu\nu}f\left(\textcolor{red}{\nabla_{\sigma}R^{\alpha}_{\beta\mu\nu}X^{\sigma}}-\nabla_{\sigma}X^{\alpha}R^{\sigma}_{\beta\mu\nu}+\nabla_{\beta}X^{\sigma}R^{\alpha}_{\sigma\mu\nu}+\nabla_{\mu}X^{\sigma}R^{\alpha}_{\beta\sigma\nu}+\nabla_{\nu}X^{\sigma}R^{\alpha}_{\beta\mu\sigma}\right)\\
&+\widetilde{\partial}_{\alpha}^{\lambda_1\ldots\lambda_r\mu}f\left(\nabla_{\lambda_1\ldots\lambda_r\mu}X^{\alpha}+O(\nabla^{r-1}X,\ldots,X)\right)+{\partial}_{\alpha}^{\lambda_1\ldots\lambda_r\mu}f \\
&\left(\textcolor{blue}{X^{\sigma}\nabla_{\sigma}\nabla_{\lambda_3\ldots,\lambda_r}R^{\alpha}_{\lambda_1\mu\lambda_2}}-\nabla_{\lambda_3\ldots,\lambda_r}R^{\sigma}_{\lambda_1\mu\lambda_2}\nabla_{\sigma}X^{\alpha}+\nabla_{\lambda_3\ldots,\lambda_r}R^{\alpha}_{\sigma\mu\lambda_2}\nabla_{\lambda_1}X^{\sigma}\right.\\ 
&\left.+ \ldots+\nabla_{\lambda_3\ldots\lambda_{r-1}\sigma}R^{\alpha}_{\lambda_1\mu\lambda_2}\nabla_{\lambda_r}X^{\sigma} +\nabla_{\lambda_3\ldots\lambda_{r }}R^{\alpha}_{\lambda_1\sigma\lambda_2}\nabla_{\mu}X^{\sigma} \right).
\end{aligned}
$$
Since the vector field $X \in \mathfrak{X}(M)$ is arbitrary, the above equation must hold for any order of covariant derivative in $X$ independently. As a result, we recursively find that $\widetilde{\partial}_{\alpha}^{\lambda_1\ldots\lambda_r\mu}f=\ldots=\widetilde{\partial}_{\alpha}^{\beta\mu\nu}f=\partial^{\mu\nu}_{\alpha}f=0$. Then collecting the $X^{\sigma}$ and $\nabla_{\sigma}X^{\rho}$ terms we get the system 
$$
\begin{cases}
	\begin{aligned}
		&0=  \partial_{\sigma}f + S^{\rho}_{\gamma\sigma} \Big\{ 2\partial^{\sigma\beta}f g_{\rho\beta} + {\partial}_{\alpha}^{\beta\mu\nu}f\left(-\delta^{\alpha}_{\rho}R^{\sigma}_{\beta\mu\nu}+\delta^{\sigma}_{\beta}R^{\alpha}_{\rho\mu\nu}+2\delta^{\sigma}_{\mu}R^{\alpha}_{\beta\rho\nu}\right) +\ldots\\
		&+{\partial}_{\alpha}^{\beta \mu \Lambda\nu}f \left(-\delta^{\alpha}_{\sigma}\nabla_{\Lambda}R^{\sigma}_{\beta\nu\mu} +\delta^{\sigma}_{\beta}\nabla_{\Lambda}R^{\alpha}_{\rho\nu\mu}+2\delta^{\sigma}_{\nu}\nabla_{\Lambda}R^{\alpha}_{\beta\rho\mu} +(r-3)\delta^{\sigma}_{\lambda_1}\nabla_{\rho\Hat{\Lambda}}R^{\alpha}_{\beta\nu\mu}\right)\Big\},
	\end{aligned}\\
	\begin{aligned}
		&0=2\partial^{\sigma\beta}f g_{\rho\beta} + {\partial}_{\alpha}^{\beta\mu\nu}f\left(-\delta^{\alpha}_{\rho}R^{\sigma}_{\beta\mu\nu}+\delta^{\sigma}_{\beta}R^{\alpha}_{\rho\mu\nu}+2\delta^{\sigma}_{\mu}R^{\alpha}_{\beta\rho\nu}\right) +\ldots\\
		&+{\partial}_{\alpha}^{\beta \mu \Lambda\nu}f \left(-\delta^{\alpha}_{\sigma}\nabla_{\Lambda}R^{\sigma}_{\beta\nu\mu} +\delta^{\sigma}_{\beta}\nabla_{\Lambda}R^{\alpha}_{\rho\nu\mu}+2\delta^{\sigma}_{\nu}\nabla_{\Lambda}R^{\alpha}_{\beta\rho\mu} +(r-3)\delta^{\sigma}_{\lambda_1}\nabla_{\rho\Hat{\Lambda}}R^{\alpha}_{\beta\nu\mu}\right).
	\end{aligned}
\end{cases}
$$
% $$
% \begin{cases}
% 	\begin{aligned}
% 		&0=  \partial_{\sigma}f + S^{\rho}_{\gamma\sigma} \left\{ 2\partial^{\sigma\beta}f g_{\rho\beta} + {\partial}_{\alpha}^{\beta\mu\nu}f\left(-\delta^{\alpha}_{\rho}R^{\sigma}_{\beta\mu\nu}+\delta^{\sigma}_{\beta}R^{\alpha}_{\rho\mu\nu}+2\delta^{\sigma}_{\mu}R^{\alpha}_{\beta\rho\nu}\right) +\ldots\right.\\
% 		&\left.+{\partial}_{\alpha}^{\lambda_1\ldots\lambda_r\mu}f \left(-\delta^{\alpha}_{\rho}\nabla_{\lambda_3\ldots\lambda_{r }}R^{\sigma}_{\lambda_1\mu\lambda_2} +\delta^{\sigma}_{\lambda_1}\nabla_{\lambda_3\ldots\lambda_{r }}R^{\alpha}_{\rho\mu\lambda_2}+2\delta^{\sigma}_{\mu}\nabla_{\lambda_3\ldots\lambda_{r}}R^{\alpha}_{\lambda_1\rho\lambda_2} +(r-3)\delta^{\sigma}_{\lambda_3}\nabla_{\rho\ldots\lambda_{r}}R^{\alpha}_{\lambda_1\mu\lambda_2}\right)\right\},
% 	\end{aligned}\\
% 	\begin{aligned}
% 		&0=2\partial^{\sigma\beta}f g_{\rho\beta} + {\partial}_{\alpha}^{\beta\mu\nu}f\left(-\delta^{\alpha}_{\rho}R^{\sigma}_{\beta\mu\nu}+\delta^{\sigma}_{\beta}R^{\alpha}_{\rho\mu\nu}+2\delta^{\sigma}_{\mu}R^{\alpha}_{\beta\rho\nu}\right) +\ldots\\
% 		&+{\partial}_{\alpha}^{\lambda_1\ldots\lambda_r\mu}f \left(-\delta^{\alpha}_{\rho}\nabla_{\lambda_3\ldots\lambda_{r }}R^{\sigma}_{\lambda_1\mu\lambda_2} +\delta^{\sigma}_{\lambda_1}\nabla_{\lambda_3\ldots\lambda_{r }}R^{\alpha}_{\rho\mu\lambda_2}+2\delta^{\sigma}_{\mu}\nabla_{\lambda_3\ldots\lambda_{r}}R^{\alpha}_{\lambda_1\rho\lambda_2} +(r-3)\delta^{\sigma}_{\lambda_3}\nabla_{\rho\ldots\lambda_{r}}R^{\alpha}_{\lambda_1\mu\lambda_2}\right).
% 	\end{aligned}
% \end{cases}
% $$
substituting the second equation into the first we get the initial claim
$$
\begin{aligned}
    0&=2\frac{\partial f}{\partial g_{\sigma\beta}}f g_{\rho\beta} + \frac{\partial f}{\partial R^{\alpha}_{\beta\mu\nu}}\Big(-\delta^{\alpha}_{\rho}R^{\sigma}_{\beta\mu\nu}+\delta^{\sigma}_{\beta}R^{\alpha}_{\rho\mu\nu}+2\delta^{\sigma}_{\mu}R^{\alpha}_{\beta\rho\nu}\Big) +\ldots\\
    &\quad+\frac{\partial f}{\partial \nabla_{\Lambda}R^{\alpha}_{\beta\mu\nu}} \Big(-\delta^{\alpha}_{\sigma}\nabla_{\Lambda}R^{\sigma}_{\beta\nu\mu} +\delta^{\sigma}_{\beta}\nabla_{\Lambda}R^{\alpha}_{\rho\nu\mu}+2\delta^{\sigma}_{\nu}\nabla_{\Lambda}R^{\alpha}_{\beta\rho\mu} +(r-3)\delta^{\sigma}_{\lambda_1}\nabla_{\rho\Hat{\Lambda}}R^{\alpha}_{\beta\nu\mu}\Big).
\end{aligned}
$$
To include the scalars $m^2, \ \kappa$ we simply note that being scalar coordinates, their covariant derivatives are tensors, therefore diff-covariant functions are allowed to depend on them, and the above equation gets modified by adding to the right hand side the term 
$$
\begin{aligned}
    &\frac{\partial f}{\partial \nabla_{\theta}m^2} \Big(\delta^{\sigma}_{\theta}\nabla_{\rho}m^2\Big)+ \ldots+\frac{\partial f}{\partial \nabla_{\Theta}m^2} \Big(-r\delta^{\sigma}_{\theta_1}\nabla_{\rho \Hat{\Theta}}m^2\Big)+ \frac{\partial f}{\partial \nabla_{\theta}\kappa} \Big(\delta^{\sigma}_{\theta}\nabla_{\rho}\kappa\Big)+ \ldots+\frac{\partial f}{\partial \nabla_{\Theta}\kappa}\Big(-r\delta^{\sigma}_{\theta_1}\nabla_{\rho \Hat{\Theta}}\kappa\Big)
\end{aligned}
$$
with $|\Theta|=r$.
\end{proof}

\section{Hadamard parametrices}\label{section_Hadamard_parametrix}

In this section we will review some of the geometry needed to write locally the Hadamard parametrix, the latter is the closest thing (from a microanalytical standpoint) to a solution of some partial differential equation: the parametrix is a distribution which is smoothened by the differential operator inducing the linear equations. In particular we are interested in wave equations such as \eqref{eq_KG}. The construction that we are describing can be readily generalized from the scalar to the vector case (\textit{e.g.} see \cite[$\S$ 6]{friedlander1975wave}, \cite[$\S$ A]{carfora2020ricci} for the Euclidean case). \\

Let $(M,g)$ be a time oriented Lorentzian manifold, given some point $x\in M$ consider the pseudo-Riemannian exponential mapping $\mathrm{exp}_x:V\subset T_xM \to M$ of the metric $g$. Since $T_0\mathrm{exp}_x=id_{T_xM}$, there will be a neighborhood $\Omega\ni x$ such that, possibly restricting the domain, $\mathrm{exp}_x:V\subset T_xM \to \Omega$ is a diffeomorphism. Alternatively, we can say that there are neighborhoods $\mathcal{O}\subset M^2$ of the diagonal $\Delta_2(M)$, $\widetilde{V}\subset TM$ of the graph of the zero section, such that $\mathrm{exp}:\widetilde{V}\subset TM \to \mathcal{O}$ is a diffeomorphism. The square geodesic distance function is then defined as
\begin{equation}\label{eq_2_square_geodesic_distance}
    \sigma:\mathcal{O}\to \mathbb{R} : (x,y) \mapsto \sigma(x,y)=\frac{1}{2} g(x)\big(\mathrm{exp}_x^{-1}(y),\mathrm{exp}_x^{-1}(y)\big)
\end{equation}

In this we shall adopt a particular notation: let $(x,y)\in \mathcal{O}$, $v=\mathrm{exp}^{-1}_x(y), \ w=\mathrm{exp}_y^{-1}(x)$, we denote by $(U_x,x^{\mu})$, $(U_y,x^{\alpha})$ the coordinates around $x$, $y$ respectively, and by $(V,\dot x^{\mu})$, $(V,\dot x^{\alpha})$ the coordinates of $T_xM$, $T_yM$. In short, Greek indices such as $\alpha, \beta, \gamma, \delta$ will denote coordinates around $y$ of $T_yM$, indices such as $\mu,\nu,\lambda,\rho$ will refer to coordinates around $x$ or $T_xM$. In those coordinates, set $(\mathrm{exp}_x(v))^{\alpha}=e^{\alpha}_x(v)$, $(\mathrm{exp}^{-1}_x(y))^{\mu}=\bar{e}^{\mu}_x(y)$, $(\mathrm{exp}_y(w))^{\mu}=e^{\mu}_y(w)$, $(\mathrm{exp}^{-1}_y(x))^{\alpha}=\bar{e}^{\alpha}_y(x)$. Since $\mathrm{exp}_x(\mathrm{exp}_x^{-1}(y))=y$, $\mathrm{exp}_x^{-1}(\mathrm{exp}_x(v))=v$, applying the tangent mapping to both sides of those equations yield the identities
$$
	T_{\mathrm{exp}^{-1}_x(y)}\mathrm{exp}_x \circ T_y\mathrm{exp}_x^{-1}=T_yid \space\ \Rightarrow \frac{\partial e_x^{\alpha}}{\partial \dot{x}^{\mu}}\frac{\partial \bar{e}_x^{\mu}}{\partial y^{\beta}}=\delta^{\alpha}_{\beta}; \ \frac{\partial e_y^{\mu}}{\partial \dot{y}^{\alpha}}\frac{\partial \bar{e}_x^{\alpha}}{\partial y^{\nu}}=\delta^{\mu}_{\nu},
$$

$$
	T_{\mathrm{exp}_x(v)}\mathrm{exp}^{-1}_x \circ T_v\mathrm{exp}_x=T_vid, \space\ \Rightarrow \frac{\partial \bar{e}_x^{\mu}}{\partial y^{\alpha}}\frac{\partial e_x^{\alpha}}{\partial \dot{x}^{\nu}}=\delta^{\mu}_{\nu}, \ \frac{\partial \bar{e}_y^{\alpha}}{\partial x^{\mu}}\frac{\partial e_y^{\mu}}{\partial \dot{y}^{\beta}}=\delta^{\alpha}_{\beta};
$$  

\begin{lemma}\label{lemma_properties_of_sigma}
    The square geodesic distance \eqref{eq_2_square_geodesic_distance} satisfies the following identity:
    \begin{equation}\label{eq_2_fundamental_identity_geo_dist}
        \sigma=\frac{1}{2}g^{\mu\nu}(x) \nabla_{\mu}\sigma\nabla_{\nu}\sigma=\frac{1}{2}g^{\alpha\beta}(y) \nabla_{\alpha}\sigma\nabla_{\beta}\sigma.
    \end{equation}
    Moreover, as $y\to x$, we have the identities
    \begin{align}\label{eq_2_sigma_derivative_coincidence_limits}
        & \sigma_{\alpha} \to 0;\\
        & \sigma_{\alpha\beta} \to g_{\alpha\beta}(x);\\
        & \sigma_{\alpha\beta \gamma} \to 0;\\
        & \sigma_{\alpha\beta \gamma \delta} \to \frac{1}{3}\big( R_{\alpha\gamma \beta \delta}(x) + R_{\alpha\delta \beta \gamma}(x)\big);
    \end{align} 
    where $\sigma_{\alpha}=\nabla_{\alpha}\sigma$, and $R$ the Riemann tensor of the metric $g$.
    Finally, if $y\in I^{+}(x)$ $($resp. $y\in I^{-}(x)$ $)$ then $\sigma_{\alpha}$ is past $($resp. future$)$ directed.
\end{lemma}
\begin{proof}
    First recall that Gauss lemma, which states that $\mathrm{exp}_x$ is a radial isometry, that is,
    $$
        g(y)\big(T_v\mathrm{exp}_x(v), T_v\mathrm{exp}_x(v')\big)= g(x)(v,v'),\  \forall v,v' \in T_xM, \ y=\mathrm{exp}_x(v);
    $$
    or, in coordinate notation,
\begin{equation}\label{eq_2_Gauss_lemma}
	g_{\alpha\beta}(y)  \frac{\partial e_x^{\alpha}}{\partial \dot{x}^{\mu}} \frac{\partial e_x^{\beta}}{\partial \dot{x}^{\nu}}= g_{\mu\nu}(x), \ g^{\alpha\beta}(y)  \frac{\partial \bar e_x^{\mu}}{\partial {y}^{\alpha}} \frac{\partial \bar e_x^{\nu}}{\partial {y}^{\beta}}= g^{\mu\nu}(x)
\end{equation}
Then,
$$
\begin{aligned}
    \frac{1}{2}g^{\alpha\beta}(y) \nabla_{\alpha}\sigma\nabla_{\beta}\sigma&=  \frac{1}{8}g^{\alpha\beta}(y)\nabla_{\alpha}\big(g_{\mu\nu}(x)\bar e^{\mu}_x(y)\bar e^{\nu}_x(y)\big)\big(g_{\lambda\rho}(x)\bar e^{\lambda}_x(y)\bar e^{\rho}_x(y)\big)\\
    &= \frac{1}{2}g^{\alpha\beta}(y)g_{\mu\nu}(x)\frac{\bar e^{\mu}_x(y)}{\partial y^{\alpha}}\bar e^{\nu}_x(y)g_{\lambda\rho}(x)\frac{\bar e^{\lambda}_x(y)}{\partial y^{\beta}}\bar e^{\rho}_x(y) \\
    &= \frac{1}{2} g^{\mu\lambda}(x)g_{\mu\nu}(x)g_{\lambda\rho}(x)e^{\nu}_x(y)\bar e^{\rho}_x(y)= \sigma(x,y).
\end{aligned}
$$
where in the third equality we used \eqref{eq_2_Gauss_lemma}. The identities \eqref{eq_2_sigma_derivative_coincidence_limits} follows from $\sigma \to 0$ as $y\to x$ and repeated differentiation of \eqref{eq_2_fundamental_identity_geo_dist}. For details see \cite[pp. 227-228]{dewitt1960radiation}. Finally, suppose $y\in I^+(x)$, then $\sigma(x,y)<0$, by \eqref{eq_2_fundamental_identity_geo_dist}, we must have that $\sigma_{\alpha}$ is timelike. Let $T$ be the vector field inducing the time orientation on $M$ and $\gamma_x(t)$ the integral curve of $T$ starting at $x$ with velocity $T(x)\in T_xM$, by construction the function $t\mapsto \sigma(x, \gamma_x(t))$ is negative and decreasing, thus 
$$
    0>\frac{d}{dt} \sigma(x, \gamma_x(t))= g^{\alpha\beta}(y)\big(\sigma_{\alpha}|_{\gamma_x(t)},T(\gamma_x(t))\big).
$$
\end{proof}

The \textit{geodesic parallel displacement} is a mapping 
\begin{equation}\label{eq_2_parallel_displacement}
    \bar{\delta} :\mathcal{O} \to \Omega_1(M)\boxtimes \Omega_1(M): (x,y) \mapsto \bar{\delta}_{\alpha\mu}dy^{\alpha}\boxtimes dx^{\mu},    
\end{equation}
satisfying the following requirements:
\begin{equation}\label{eq_2_parallel_displacement_prop_1}
    \lim_{y\to x} g^{\alpha\beta}(y) \bar \delta_{\beta\mu}
    (x,y)=\delta^{\alpha}_{\mu},
\end{equation}
\begin{equation}\label{eq_2_parallel_displacement_prop_2}
    g^{\mu\nu}(x)\sigma_{\mu}\nabla_{\nu}\bar\delta_{\alpha \lambda}=0,\ g^{\alpha\beta}(y)\sigma_{\alpha}\nabla_{\beta}\bar\delta_{\gamma \mu}=0.
\end{equation}

We stress that $\bar\delta$ is the mapping condition \eqref{eq_2_parallel_displacement_prop_2} is the equation of parallel transport, from $T_xM$ to $T_yM$, along the geodesic joining $x $ to $y$ whereas \eqref{eq_2_parallel_displacement_prop_1} is the initial condition. From those condition the existence and uniqueness of the geodesic parallel displacement can be derived.

\begin{lemma}\label{lemma_2_geodesic_displacement}
    The geodesic parallel displacement introduced in \eqref{eq_2_parallel_displacement} satisfies the following relations:
    \begin{itemize}
        \item[$(i)$] $\bar\delta^{\alpha}_{\ \mu}\sigma_{\alpha}=-\sigma_{\mu}$, $\bar\delta_{\alpha}^{\ \mu}\sigma_{\mu}=-\sigma_{\alpha}$;
        \item[$(ii)$] $\bar\delta^{\alpha\mu}\bar\delta_{\alpha \nu}=\delta^{\mu}_{\nu}$, $\bar\delta^{\alpha\mu}\bar\delta_{\beta \mu}=\delta^{\alpha}_{\beta}$;
        \item[$(iii)$] $\bar\delta^{\alpha}_{\ \mu}\bar\delta^{\beta}_{\ \nu}g_{\alpha\beta}(y)= g_{\mu\nu}(x)$, $\bar\delta^{\ \mu}_{\ \alpha}\bar\delta_{\beta}^{\ \nu}g_{\mu\nu}(x)= g_{\alpha\beta}(y)$.
    \end{itemize}
\end{lemma}

\begin{proof}
    $(i)$ follows from the fact that if $y=\exp_x(v)$, the geodesic joining the two points is $\gamma(t)=\exp_x(tv)$, calling $w=-\exp^{-1}_x(y)$, we have $\gamma(1-t)=\exp_y(tw)$, then we have
    \begin{equation}\label{eq_2_inverse_geodesic_trick}
        \frac{\partial e_x^{\alpha}}{\partial \dot x^{\mu}}\bar e_x^{\mu}=-\bar e_y^{\alpha}.
    \end{equation}
    Using the above equality and \eqref{eq_2_Gauss_lemma} we get $g_{\mu\nu}(x)\sigma_{\nu}=- \bar e_y^{\alpha}$. Then
    $$
        \bar\delta_{\alpha}^{\ \mu}\sigma_{\mu}= \frac{\partial e_x^{\alpha}}{\partial \dot x^{\mu}}\bar e_x^{\mu}=-\bar e_y^{\alpha}= -\sigma_{\alpha}.
    $$
    $(ii)$ states that the composition of parallel transport from $x $ to $y$ and back is the identity, and can be seen using \eqref{eq_2_parallel_displacement_prop_2} to show that for any $v\in T_xM$,
    $$
        \sigma^{\lambda}\nabla_{\lambda} \big( \bar\delta^{\alpha\mu}\bar\delta_{\alpha \nu} v^{\nu}\big)=0.
    $$
    Finally, $(iii)$ follows from the fact that along the geodesic $ \gamma(t)=e_x(t\bar e_x(y))$, each $v,v'\in T_xM$ is transported parallel to the geodesic and thus the scalar product of the two vector remain constant along $\gamma$. Therefore expanding
    $$
    0= \frac{d}{dt}\Big( g(\gamma(t))\big( \delta(v),\delta(v')\big)\Big)
    $$
    in coordinates we conclude.
\end{proof}

We now introduce three quantities that are of great importance to define the Hadamard parametrix:
\begin{equation}\label{eq_2_det_parallel_displacement_1}
    \bar\delta (x,y)= \mathrm{det}(\bar\delta^{\alpha}_{\ \mu}),
\end{equation}
Notice that by $(i)$ in Lemma \ref{lemma_2_geodesic_displacement}, we have $g^{\mu\nu}\bar\delta^{\alpha}_{\ \mu}\sigma_{\alpha}\bar\delta^{\beta}_{\ \nu}\sigma_{\beta}=g^{\mu\nu}\sigma_{\mu}\sigma_{\nu}=\sigma$, implying $g_{\mu\nu}=\bar\delta^{\alpha}_{\ \mu}g_{\alpha\beta}\bar\delta^{\beta}_{\ \nu}$. Taking the determinant we get
\begin{equation}\label{eq_2_det_parallel_displacement_2}
    \bar\delta (x,y)=g^{1/2}(x)g^{1/2}(y).
\end{equation}
The second quantity of interest is \textit{van Vleck determinant}: 
\begin{equation}\label{eq_2_Van_vleck_determinat_0} 
    D(x,y)=-\mathrm{det}(\sigma_{\alpha\mu}),
\end{equation}
it's easy to show, arguing with \eqref{eq_2_Gauss_lemma} and \eqref{eq_2_inverse_geodesic_trick}, that 
\begin{equation}\label{eq_2_Van_vleck_determinat_1} 
    -\sigma_{\mu\alpha}(x,y)= g_{\alpha\beta}(y)\frac{\partial \bar e_y^{\beta}}{\partial x^{\mu}}.
\end{equation}
Finally, define 
\begin{equation}\label{eq_2_Delta(x,y)_coefficient} 
    \Delta(x,y)= \frac{D(x,y)}{\bar\delta(x,y)}.
\end{equation}
\begin{lemma}\label{lemma_2_Delta_identity}
    The quantity $\Delta$ defined in \eqref{eq_2_Delta(x,y)_coefficient} satisfies
    $$
        g^{\alpha\beta} \nabla_{\alpha}\big(\Delta\sigma_{\beta} \big)=n \Delta,
    $$
    where $n$ is the dimension of $M$.
\end{lemma}

\begin{proof}
    By explicit calculation,
    $$
    \begin{aligned}
        g^{\alpha\beta} \nabla_{\alpha}\big(\Delta\sigma_{\nu} \big) &= g^{\alpha\beta} \nabla_{\alpha}\Delta\sigma_{\beta} +  g^{\alpha\beta} \Delta\sigma_{\alpha\beta} \\
        &= -\Delta g^{\alpha\beta}  \delta_{\epsilon \mu}\nabla_{\gamma}\bar\delta^{\epsilon \mu} \sigma_{\beta} + \Delta g^{\alpha\beta}  \sigma_{\epsilon \mu}\nabla_{\gamma}\sigma^{\epsilon \mu} \sigma_{\beta} +  g^{\alpha\beta} \Delta\sigma_{\alpha\beta} ,
    \end{aligned}
    $$
    the first term is zero by \eqref{eq_2_parallel_displacement_prop_2}, where as the rest can be written as 
    $$
        \Delta g^{\alpha\beta}\sigma^{\gamma\mu}(\nabla_{\alpha}\sigma_{\gamma\mu}\sigma_{\beta}+ \sigma_{\alpha\gamma}\sigma_{\mu \beta}),
    $$
    Notice that,
    $$
        \sigma_{\alpha\mu}=\nabla_{\alpha}\sigma_{\mu}=\nabla_{\alpha}(g^{\rho\nu}\sigma_{\mu\nu}\sigma_{\rho})=g^{\rho\nu}(\nabla_{\alpha}\sigma_{\mu\nu}+\sigma_{\alpha\rho}),
    $$
    which, using \eqref{eq_2_Van_vleck_determinat_1} to calculate $\sigma^{\gamma\mu}\sigma_{\gamma\mu}$, implies 
    $$
    g^{\alpha\beta} \nabla_{\alpha}\big(\Delta\sigma_{\nu} \big)=\Delta \sigma^{\gamma\mu}\sigma_{\gamma\mu}=n\Delta.
    $$
\end{proof}

We are going to use this identities to construct solutions for wave equations of the form $P=g^{\alpha\beta}\nabla_{\alpha}\nabla_{\beta}+ B$, where $B \in C^{\infty}(M)$. \textit{i.e.} equations defined by normally hyperbolic differential operators \textit{c.f.} Lemma \ref{lemma_0_wave_eq_form_of_normally_hyp_op}. We are, in particular, looking for a \textit{parametrix} for the equations, \textit{i.e.} a distribution $H\in \mathcal{D}'(\mathcal{O})$, such that if $x\in M$ is fixed, $\Omega\ni x$ is a geodesically convex neighborhood, 
\begin{equation}\label{eq_2_def_parametrix}
    P(y)H(x,y)-\delta_y\in C^{\infty}(\Omega).
\end{equation}
In the above expression $\delta_y:\mathcal{E}(\Omega)\ni f \mapsto f(y)\in  \mathbb{R}$ is Dirac's delta. Denote by $\mathrm{Par}(M,h)$ the set of parametrices with respect to the normally hyperbolic operator $P$.

\begin{theorem}\label{thm_2_existence_of_parametrices}
    Suppose that $(M,g)$ is a time oriented Lorentzian manifold let $d=\frac{n-2}{2}$, then given any wave equation of the form $P=\square_g + B$ with $B\in C^{\infty}(M)$, there exist local parametrices:
    \begin{itemize}
        \item[$(i)$] if $n$ is even, 
        $$
            H(x,y)=\beta_n \bigg[ \frac{U(x,y)}{(\sigma(x,y)+i\epsilon t(x,y)+\epsilon^2)^d}+V(x,y)\ln\Big(\frac{1}{\mu^2}(\sigma+i\epsilon t(x,y)+\epsilon^2)\Big)+ W(x,y)\bigg],
        $$
        \item[$(ii)$] if $n$ is odd,
        $$
            H(x,y)=\beta_n\frac{Z(x,y)}{(\sigma(x,y)+i\epsilon t(x,y)+\epsilon^2)^d},
        $$
    \end{itemize}
    where $\mu$ is a constant length scale, $t(x,y)\equiv t(x)-t(y)$ and $t$ is the Cauchy temporal function, $\epsilon\to 0$ the parameter used to evaluate the Hadamard principal value of the integral and the cut in the complex domain of $\ln$ is in the negative real axis. The coefficients $U=\sum_{j=1}^{d-1}U_j\sigma^j$, $V=\sum_{j\geq 0}V_j\sigma^j\rho(\sigma k_j)$, $W=\sum_{j\geq 0}W_j\sigma^j\rho(\sigma k_j)$, $Z=\sum_{j\geq 0}Z_j\sigma^j\rho(\sigma k_j)$ are recursively constructed out of the geometric data of the problem, $\rho \in C^{\infty}_c(\mathbb{R})$ satisfies $0\leq \rho \leq 1$, $\rho(z)\equiv 1$ for $|z|\leq 1/2$, $\rho(z)\equiv 0 $ for $|z| \geq 1$ and $\{k_j\}_{j\in \mathbb{N}}$ is a suitable sequence with $k_j \nearrow \infty$. Each function $U_j,V_j,W_j,Z_j \in C^{\infty}(\Omega \times \Omega)$ is inductively constructed out of the metric $g$ and the operator $P$ and $\beta_n$ are numerical coefficients depending on the dimension, moreover each series is convergent to a smooth function of $C^{\infty}(\Omega \times \Omega)$ for a suitable geodesically convex neighborhood $\Omega$ of $x\in M$. Moreover, given $H_1,\ H_2 \in \mathrm{Par}(M,h)$, $H_1-H_2\in C^{\infty}(\mathcal{O})$. 
\end{theorem}
\begin{proof}

For simplicity we will do the case $n=4$, however, similar arguments also apply to the other cases; for the specific details we point to \cite{garabedian1960partial}. For simplicity, we also set $\mu =1$, the general case being a straightforward generalization. We divide the proof in steps: first we define the coefficients $U,\ V,\ W$ by requiring that $PH=0$ at each order in $\sigma$, then we take care of the convergence of the newly determined infinite series by introducing some ad-hoc cutoff factor which ensure convergence in the Fréchet topology of $C^{\infty}(\mathcal{O})$ but weakens $PH=0$ into $PH\in C^{\infty}(\mathcal{O})$, finally, using the results in \cite[Appendix B]{kay1991theorems} we show that $H$ is a well defined distribution.

\textbf{Step 1.} Fixed $x\in M$ consider 
$$
\begin{aligned}
    P(y)H(x,y)&= -\sigma^{-2}g^{\alpha\beta}\big(2 \nabla_{\alpha}U-U\Delta^{-1}\nabla_{\alpha}\Delta  \big)\sigma_{\beta}\\
    &\quad+\sigma^{-1}\big[ 2V+ g^{\alpha\beta}\big( 2\nabla_{\alpha}V-V\Delta^{-1}\nabla_{\alpha}\Delta  \big)\sigma_{\beta}+PU\big]\\
    &\quad+ PV\ln(\sigma)+PW.
\end{aligned}
$$
setting
\begin{equation}\label{eq_2_def_U(x,y)}
    U(x,y)=\Delta^{1/2}(x,y)
\end{equation}
and using Lemma \ref{lemma_2_Delta_identity} we get $U^{-1}\nabla_{\alpha}U=\frac{1}{2}\Delta^{-1}\nabla_{\alpha}\Delta$, which substituted into the above expression for $P(y)H(x,y)$ yields, at each order in $\sigma$, the following identities
$$
\begin{aligned}
    V_0+g^{\alpha\beta}\big(\nabla_{\alpha}V_0-\frac{1}{2}V_0\Delta^{-1}\nabla_{\alpha}\Delta)\sigma_{\beta}&=-\frac{1}{2} PU\\
    (j+1) V_j + g^{\alpha\beta}\big(\nabla_{\alpha}V_j-\frac{1}{2}V_j\Delta^{-1}\nabla_{\alpha}\Delta\big) \sigma_{\beta}&=-\frac{1}{2j}PV_{j-1}\\
    (j+1)W_j+g^{\alpha\beta}\big(\nabla_{\alpha}W_j-\frac{1}{2}W_j\Delta^{-1}\nabla_{\alpha}\Delta\big) \sigma_{\beta}&= -\frac{1}{2j}P(W_{j-1}-V_{j-1})-V_n
\end{aligned}
$$
To solve the above system we use normal coordinates $\dot x^{\mu}=\bar e_x^{\mu}(y)$ around $x$, in those coordinates
\begin{equation}\label{eq_2_def_U(x,y)_normal_coord}
    U(x,\dot x)=\Big(\frac{g(x)}{g(\dot x)}\Big)^{1/4}
\end{equation}
then we get the system
$$
\begin{aligned}
    \dot x^{\mu}\partial_{\mu}\Big(\frac{V_0}{U}\Big) + \frac{V_0}{U}&=-\frac{1}{4}\frac{PU}{U},\\
    \dot x^{\mu}\partial_{\mu}\Big(\frac{V_j}{U}\Big) + (j+1)\frac{V_j}{U}&=-\frac{1}{4}\frac{PV_{j-1}}{U},\\
    \dot x^{\mu}\partial_{\mu}\Big(\frac{W_j}{U}\Big) + (j+1)\frac{W_j}{U}&=-\frac{1}{4}\frac{PW_{j-1}-PV_{j-1}}{U}-\frac{V_j}{U},\\
\end{aligned}
$$
which can be integrated as 
\begin{equation}\label{eq_2_def_V_0(x,y)_normal_coord}
    V_0(x,\dot x)=-\frac{U(x,\dot x)}{4}\int_0^1 \frac{PU(x,s\dot x^{\mu})}{U(x,s\dot x^{\mu})}ds.
\end{equation}
\begin{equation}\label{eq_2_def_V_j(x,y)_normal_coord}
	V_l(x,\dot x)= -\frac{U(x,\dot x)}{4l}\int_0^1\frac{s^lPV_{l-1}(x,s\dot x^{\mu})}{U(x,s\dot x^{\mu})}ds.
\end{equation}
\begin{equation}\label{eq_2_def_W_j(x,y)_normal_coord}
	W_l(x,\dot x)=\frac{U(x,\dot x)}{4l}\int_0^1 s^l \frac{PV_{l-1}(x,s\dot x^{\mu})-lPW_{l-1}(x,s\dot x^{\mu})}{U(x,s\dot x^{\mu})}ds -lU(x,\dot x)\int_0^1 s^l \frac{V_{l}(x,s\dot x^{\mu})}{U(x,s\dot x^{\mu})}ds.
\end{equation}

\textbf{Step 2.} Up to now the series defining the function $H$ is a formal series in $\sigma$, when the coefficients of the differential operator $P$ are analytic, then the series for $H$ does converge to the exact solution. In our case however we have smooth coefficients, so the convergence has to be forced by hand with a suitable cutoff. Let $\rho \in C^{\infty}_c(\mathbb{R})$ with $0\leq \rho \leq 1$, $\rho(t)\equiv 1$ in $|t|\leq 1/2$ and $\rho(t)\equiv 0 $ in $|t| \geq 1$. Suppose that $\{k_j\}_{j\in \mathbb{N}}$ is a suitable sequence with $k_j \rightarrow \infty$, then claim that 
\begin{equation}\label{eq_2_Borel_series_parametrix}
	H(x,y)\doteq \frac{U(x,y)}{\sigma(x,y)}+ \sum_{l\geq 0} \rho(\sigma(x,y)k_l)\left(V_l(x,y)\sigma(x,y)^l\ln(\sigma(x,y))+W_l(x,y)\sigma(x,y)^l\right)
\end{equation}
is a parametrix. We are going to show convergence in the Fréchet topology of $C^{\infty}(\mathcal{O})$ of the series in $\sigma$. First we study the case where $x\neq y$, when $k_j^{-1}\leq \sigma(x,y)$, $\rho=0$, thus, given any compact subset $K\subset M$, we can estimate
$$
	\sup_{K\times K} \left| \rho(\sigma(x,y)k_l)\left( V_l(x,y)\sigma(x,y)^l+W_l(x,y)\sigma(x,y)^l \right)\right|\leq \sup_{K\times K} \left(\frac{1}{k_l}\right)^l\left( |V_l(x,y)||\ln(1/k_l)|+|W_l(x,y)|\right) .
$$
Then to ensure convergence of the series of \eqref{eq_2_Borel_series_parametrix} we require that 
$$
\begin{cases}
	\sum_{l\geq0} \sup_{K \times K}  \left(\frac{1}{k_l}\right)^l |V_l(x,y)||\ln(1/k_l)| <\infty,\\
	\sum_{l\geq0} \sup_{K \times K}  \left(\frac{1}{k_l}\right)^l |W_l(x,y)| <\infty.
\end{cases}
$$
Denote by $\{k_{j,0}\}$ the sequence satisfying those estimates. Next, we study
$$
	\sup_{K\times K} \left| \nabla_{\mu} \left( \rho(\sigma(x,y)k_l)\left( V_l(x,y)\sigma(x,y)^l+W_l(x,y)\sigma(x,y)^l \right) \right)\right|.
$$
Then again, we will get some conditions the sequence $\{k_j\}$ has to satisfy to ensure that the above quantity is bounded. Denote by $\{k_{j,1}\}$ the sequence obtained. Repeating this at all orders give us sequences $\{k_{j,n}\}$, without loss of generality we can assume that $\forall j\in \mathbb{N}$, $k_{j,n}\geq k_{j,n-1}$; by a diagonal argument, define a sequence $k_{j}\doteq k_{j,j}$. With this particular choice of $\{k_j\}$, the series defining the parametrix converges uniformly with all its derivatives to a smooth function in $K\times K \backslash \Delta_2(K)$. We can then take an exhaustion of a compact subset of $M$, for each we determine the corresponding sequence, then a diagonal argument enables us to conclude that the series defining $H$ converges with all its derivatives uniformly in every compact set of $M$. 

\textbf{Step 3.} When $y\to x$, or $x, \ y$ are points along a lightlike geodesics, singularities begin to appear, applying \cite[Lemma B.1, Lemma B.2]{kay1991theorems}, we obtain that $H$ is a distribution defined through Hadamard principal value, \textit{i.e.} 
\begin{equation}\label{eq_2_Hadamard_distribution}
	H(x,y)= \frac{1}{4\pi^2}\frac{U(x,y)}{\sigma(x,y)+i\epsilon t(x,y)+\epsilon^2}+ \widetilde{V}(x,y)\ln(\sigma(x,y)+i\epsilon t(x,y)+\epsilon^2)+ \widetilde{W}(x,y);
\end{equation}
where $t$ is the Cauchy temporal function, $t(x,y)\equiv t(x)-t(y)$, $\epsilon\to 0$ the parameter used to evaluate the principal value of the integral.
\end{proof}

\section{Wick Powers}\label{section_Wick_powers}

From the classical results presented in Chapter \ref{chapter_classical} we would like to enter the quantum realm. We accomplish this via deformation quantization of the algebra of microcausal functionals. In this regard, we will assume that the microcausal algebra of Definition \ref{def_1_WF_mucaus} is maximally defined, that is the $CO$-open subset $\mathcal{U}$ representing the domain of functionals is the whole manifold $C^{\infty}(M,\mathbb{R})$. Given a background geometry $h=(g,m,\kappa)$ we will denote by
$$
    \Big(\mathcal{F}_{\mu c}(M,h),\{ \ , \ \}_{(M,h)}\Big)
$$
the Poisson algebra of microcausal functional where the background geometry $h=(g,m,\kappa)$ determines both the wave front set properties (\textit{c.f.} \eqref{eq_1_WF_mu_caus}) and the dynamic induced by the Klein-Gordon operator 
$$
    \square_g+ m^2+ \kappa R(g).
$$
A \textit{deformation quantization} of the algebra $\mathcal{F}_{\mu c}(M,h)$ is therefore an associative algebra with a $\star$ product of the form
$$
    F\star G= \sum_j \hbar^j \Pi_j(F,G)
$$
for some mappings $\Pi_j$ satisfying the consistency conditions
$$
\begin{aligned}
    \Pi_0(F,G) &= F\cdot G,\\
    \Pi_1(F,G)-\Pi_1(G,F)&= i\hbar \langle dF,\mathcal{G}_{(M,h)}dG \rangle + O(\hbar^2),
\end{aligned}
$$
where $\cdot$ is the product of the classical algebra defined in \eqref{eq_1_classical_algebra_product}, $\langle dF,\mathcal{G}_{(M,h)}dG \rangle$ and $\mathcal{G}_{(M,h)}$ are respectively the Peierls bracket and the causal propagator associated to the Klein Gordon equations on globally hyperbolic spacetime (\textit{c.f.} Definition \ref{def_1_Peierls}, Theorem \ref{thm_1_properties_of_Green_functions}). The price we pay to work with such algebras is that they are infinite series in the deformation parameter $\hbar$. 

\begin{lemma}\label{lemma_2_Weyl_regular_algebra}
    Consider the algebra of regular functionals $\mathcal{F}_{reg}(M,h)$ with the topology of strong convenient convergence generated by seminorms \eqref{eq_1_strong_conv_seminorms_1}. Then the star product defined, at each order of $\hbar$, by
    \begin{equation}\label{eq_2_regular_star_product}
	\big(F \star G\big) (\varphi)=F(\varphi)G(\varphi) +\sum_{j\geq 1}\frac{\hbar^j}{j!} \bigg(\frac{i}{2}\bigg)^j\Big\langle d^jF[\varphi], \otimes^j \mathcal{G}_{(M,h)} \big(d^jG[\varphi]\big) \Big\rangle ,
    \end{equation}
    is a deformation quantization product. Moreover, if we consider the subalgebra $\mathcal{W}(M,h)$ of $\mathcal{F}_{reg}(M,h)$ generated by imaginary exponentials of linear fields 
    $$
        \phi_{(M,h,f)}:\varphi \mapsto \int_M f \varphi d\mu_g,
    $$
    as $f$ varies in $\mathcal{D}(M)$, then $\mathcal{W}(M,h)$ is a Weyl subalgebra.
\end{lemma}
\begin{proof}
    That $\Pi_0(F,G) = F\cdot G$ is clear from \eqref{eq_2_regular_star_product}; for $ \Pi_1(F,G)-\Pi_1(G,F)= i\hbar  \langle F,G \rangle_{\mathcal{G}} +O(\hbar)$, we can suppose that $F, \ G$ are regular functionals \textit{i.e.} their deformation series in $\hbar$ are trivial. Then at $\hbar$
    $$
    \begin{aligned}
            \big(\Pi_1(F,G)-\Pi_1(G,F)\big) (\varphi)&= \frac{i}{2} \Big(\Big\langle dF[\varphi], \mathcal{G}_{(M,h)}  \big(dG[\varphi]\big) \Big\rangle  - \Big\langle dG[\varphi], \mathcal{G}_{(M,h)}  \big(dF[\varphi]\big) \Big\rangle \Big)+O(\hbar^2)\\
            &= i \hbar  \langle F,G \rangle_{\mathcal{G}} +O(\hbar).
    \end{aligned} 
    $$
    Next we notice that for the regular functional $ \exp\big( i\phi_{(M,h,f)}\big)$, we have
    $$
        d \exp\big(i \phi_{(M,h,f)}[\varphi](\psi)\big)= i\phi_{(M,h,f)}(\psi)\exp\big(i \phi_{(M,h,f)}(\varphi)\big)
    $$
    for all $\varphi\in C^{\infty}(M,\mathbb{R})$, $f,\ \psi\in \mathcal{D}(M)$. Therefore,
$$
\begin{aligned}
    \Big(\exp\big(i \phi_{(M,h,f_1)}\big)\star \exp\big(i \phi_{(M,h,f_2)}\big)\Big) (\varphi)&=\sum_{j\geq 0 } \frac{1}{j!} \bigg(\frac{-i\hbar}{2}\bigg)^j\big( \big\langle f_1, \mathcal{G}_{(M,h)} f_2 \big\rangle\big)^j \exp\big(i \phi_{(M,h,f_1+f_2)}\big)(\varphi) \\
    &= e^{\frac{-i\hbar}{2} \mathcal{G}_{(M,h)}(f_1,f_2)}\exp\big(i \phi_{(M,h,f_1+f_2)}\big)(\varphi).
\end{aligned}
$$
\end{proof}

Notice that if we try to force prescription \eqref{eq_2_regular_star_product} for microcausal functionals, then we run into problems: for example  
$$
\begin{aligned}
    \big(\phi^2_{(M,h,f_1)}\star \phi^2_{(M,h,f_2)}\big) (\varphi)= 
    \int_{M^2} d\mu_g(x,y) f_1(x)f_2(y)\big( & \varphi^2(x)\varphi^2(y)+\varphi(x)\varphi(y)\mathcal{G}_{(M,h)}(x,y) \\ &+ \mathcal{G}_{(M,h)}(x,y)^2 ).
\end{aligned}
$$
does possess the element $\mathcal{G}(x,y)^2$ which is not a well defined distribution. To remedy this we start noticing that given the operator 
% $$
%     \Gamma_{\mathcal{G}}: \mathcal{F}_{reg}(M,h) \otimes \mathcal{F}_{reg}(M,h) \to \mathcal{F}_{reg}(M,h) \otimes \mathcal{F}_{reg}(M,h),\ F \otimes G \mapsto \Gamma_{\mathcal{G}}(F,G)
% $$
% with
\begin{equation}\label{eq_2_def_Gamma_operator}
\begin{aligned}
    \Gamma_{\mathcal{G}}= & \hbar\int_{M^2}d\mu_g(x,y) \mathcal{G}_{{(M,h)}}(x,y) \frac{\partial}{\partial \varphi(x)}\otimes \frac{\partial}{\partial \varphi(y)},%\\ & \Gamma_{\mathcal{G}}(F,G)(\varphi)= \hbar \big\langle dF[\varphi],\mathcal{G}_{(M,h)}dG[\varphi]\big\rangle,
\end{aligned}
\end{equation}
we can write 
\begin{equation}
     F\star G= \mathcal{M} \circ e^{\frac{i}{2}\Gamma_{\mathcal{G}}}\big( F \otimes G\big).
\end{equation}
where $\mathcal{M}(F\otimes G)= F\cdot G$ and 
$$
    e^{\frac{i}{2}\Gamma_{\mathcal{G}}}=id_{\mathcal{F}_{reg}}\otimes id_{\mathcal{F}_{reg}} + \sum_j\bigg(\frac{i}{2}\bigg)^j\otimes^j\Gamma_{\mathcal{G}}. 
$$

\begin{definition}\label{def_2_hadamard_state_1}
Let $\omega$ be a state on the algebra $\mathcal{F}_{reg}(M,h)$, then it is called a Hadamard state if:
\begin{itemize}
    \item[$(i)$] its two point function $\omega_2:\mathcal{D}(M^2)\rightarrow \mathbb{R}: (f_1,f_2) \mapsto \omega\left(\phi_{(M,h)}(f_1),\phi_{(M,h)}(f_2)\right)$ satisfies the property
\begin{equation} \label{eq_WF_Hadamard_state}
    \mathrm{WF}(\omega_2)= \{(x,y;\xi,-\eta)\in \dot T^*M^2: \  (x,\xi)\sim_g(y,\eta),\ g^{\sharp}\xi\in V_g^+(x) \},
\end{equation}
where the notation $(x,\xi)\sim_g(y,\eta)$, introduced in Theorem \ref{thm_1_properties_of_Green_functions}, means that there is a lightlike geodesic $\gamma:[0,\Lambda] \subset \mathbb{R}\to M$ for which $(x,\xi)=(\gamma(0), g^{\flat}\dot{\gamma}(0))$ and $(y,\eta)=(\gamma(\Lambda), g^{\flat}\dot{\gamma}(\Lambda))$;
\item[$(ii)$] its truncated $n$-point functions $\omega^T_n$ defined implicitly by the relation\footnote{In the relation below $\pi \in \mathcal{P}(\{1,\ldots,n\})$ is a partition by $|\pi|$ elements $I_1, \ldots, I_{|\pi|}$ of $\{1,\ldots,n\}$ and $f_I= \otimes_{i\in I} f_i$. This requirement is consistent with Theorem 4.2 and Corollary 4.3 in \cite{sanders2010equivalence} and ensures that the $n$-point functions $\omega_2$ exhibit a wave front set satisfying the microlocal spectrum condition ($\mu$SC), that is $\mathrm{WF}(\omega_2) \subset \Gamma_n(M,h)$ where the latter is the set of points $(x_1,\ldots,x_q,\xi_1,\ldots,\xi_q)\in \dot{T}^{*}M^n$ with the following property: any point $x_i$ is connected at least to some other point $x_j$ via a lightlike geodesic $\gamma_{i\to j}$ and 
$$
	g^{\sharp}\xi_i= \sum_{e\in \{1,\ldots,q\}} \dot{\gamma}_{e\to i}(x_i) - \sum_{s\in \{1,\ldots,q\}}\dot{\gamma}_{i\to s}(x_i),
$$
where the first sum is taken on all future directed lightlike geodesics starting at some other point $x_e$ and ending at $x_i$, while the second is taken over all future directed lightlike geodesics starting at $x_i$ and ending at some other point $x_s$.}  
$$
    \omega_n(f_1,\ldots,f_n) = \sum_{\pi \in \mathcal{P}(\{1,\ldots,n\})} \prod_{I\in \pi} \omega^T_{|I|}(f_I), \quad f_i\in C^{\infty}_c(M) \ \ 1\leq i \leq n
$$
are smooth for all $n\neq 2$.
\end{itemize}
\end{definition}

If $x\in M$ is arbitrary and $\Omega$ is a geodesically convex neighborhood of $x$, \cite[Theorem 5.1]{radzikowski1996micro} proved that 
$$
    \omega_2\big|_{\Omega\times \Omega}(x,y)- H(x,y)-\frac{i}{2}\mathcal{G}_{(M,h)}(x,y) \in C^{\infty}(\Omega\times \Omega)
$$
where $H$ is the Hadamard parametrix defined in \eqref{eq_2_Hadamard_distribution} and $\omega_2$ the two point function of a Hadamard state. As a result, the quantity $H+\frac{i}{2}\mathcal{G}$ satisfies the microlocal condition \eqref{eq_WF_Hadamard_state}. \\

It is now possible, by direct application of \cite[Theorem 8.2.14]{hormanderI}, not only to define powers of $\omega=H+\frac{i}{2}\mathcal{G}$ as distributions over $M\times M$, but to define altogether the quantity
\begin{equation}\label{eq_2_star_G_H_product_microcausal}
    F \star_H G \doteq F(\varphi)G(\varphi) +\sum_{j\geq 1}\frac{\hbar^j}{j!} \bigg(\frac{i}{2}\bigg)^j\Big\langle d^jF[\varphi], \otimes^j \omega \big(d^jG[\varphi]\big) \Big\rangle ,
\end{equation}
given any $F,G \in \mathcal{F}_{\mu c }(M,h)$. We stress that the notation $\star_H$ refers to the fact that in $\omega=H+\frac{i}{2}\mathcal{G}$, the antisymmetric part $\mathcal{G}$ is unique for each background geometry, whereas the symmetric part $H$ in defined modulo a smooth function. We then claim the following:

\begin{lemma}\label{lemma_2_star_to_star_H_product}
    Suppose that $\Gamma_{\mathcal{G}}$ is the mapping introduced in \eqref{eq_2_def_Gamma_operator}, then for any $F,G \in \mathcal{F}_{reg}(M,h)$ the $\star_{\omega}$ product defined in \eqref{eq_2_star_G_H_product_microcausal} with $\omega= H + \frac{i}{2}\mathcal{G}_{(M,h)}$ can be written as
    $$
        F\star_H G= \alpha_H \Big( \alpha_H^{-1}(F) \star \alpha_H^{-1}(G)\Big)
    $$
    where
    $$
        \alpha_H(F)(\varphi) \doteq F(\varphi) + \sum_{j\geq 1}( i\hbar)^j \big\langle H^j , d^{2j}F[\varphi]\big\rangle
    $$
\end{lemma}
Before the prove, let us use a notation similar to \eqref{eq_2_def_Gamma_operator} and set
\begin{equation}\label{eq_2_def_Gamma_tilde_operator}
   \widetilde{\Gamma}_H : \mathcal{F}_{reg}(M,h) \to \mathcal{F}_{reg}(M,h), \quad F \to \widetilde{\Gamma}_H (F)  
\end{equation}
with
$$
    \widetilde{\Gamma}_H (F)(\varphi) = \big\langle H, d^2F[\varphi]\big\rangle.
$$
Since the starting $F$ is regular, $\widetilde{\Gamma}_H (F) $ will be well defined and regular as well. 
\begin{proof}
We notice that by the Leibniz rule,
$$
    d(F \cdot G)[\varphi]= ( dF[\varphi]\cdot G + F\cdot dG[\varphi]),
$$
then we can write 
$$
\begin{aligned}
    \widetilde{\Gamma}_H(F\cdot G)= \Big( \widetilde{\Gamma}_HF \cdot G + F\cdot \widetilde{\Gamma}_H G+ 2\big\langle dF, H  (dG) \big\rangle  \Big),
\end{aligned}
$$
or equivalently, using the multiplication mapping $\mathcal{M}$ and the operator $\Gamma$ defined in \eqref{eq_2_def_Gamma_operator},
$$
    \widetilde{\Gamma}_H \circ \mathcal{M} = \mathcal{M} \circ \Big( id\otimes \widetilde{\Gamma}_H + \widetilde{\Gamma}_H\otimes id + {\Gamma}_H \Big).
$$
Finally, we have
$$
\begin{aligned}
    \alpha_H \Big( \alpha_H^{-1}(F) \star \alpha_H^{-1}(G)\Big)&=e^{\frac{i}{2}\widetilde{\Gamma}_H}\Big( e^{-\widetilde{\Gamma}_{H}}(F)\star e^{-\widetilde{\Gamma}_{H}}(G)\Big)\\
    &=e^{\frac{i}{2}\widetilde{\Gamma}_H}\circ \mathcal{M}\Big( e^{\frac{i}{2}{\Gamma}_{\mathcal{G}} -id\otimes \widetilde{\Gamma}_H - \widetilde{\Gamma}_H\otimes id }(F\otimes G)\Big)\\
    &=e^{\frac{i}{2}\widetilde{\Gamma}_H}\circ \mathcal{M}\Big( e^{-id\otimes \widetilde{\Gamma}_H - \widetilde{\Gamma}_H\otimes id +\frac{i}{2}{\Gamma}_{\mathcal{G}}}(F\otimes G)\Big)\\
    &= \mathcal{M}\circ \Big( e^{id\otimes \widetilde{\Gamma}_H + \widetilde{\Gamma}_H\otimes id +{\Gamma}_{H}-id\otimes \widetilde{\Gamma}_H - \widetilde{\Gamma}_H\otimes id +\frac{i}{2}{\Gamma}_{\mathcal{G}}}(F\otimes G)\Big)\\
    &= \mathcal{M} \circ \Big( e^{{\Gamma}_{H} +\frac{i}{2}{\Gamma}_{\mathcal{G}}}(F\otimes G)\Big)\\
    &= F \star_H G
\end{aligned}
$$
where in the third equality we used that the various exponentials commute as can be inferred directly from their definitions \eqref{eq_2_def_Gamma_operator}, \eqref{eq_2_def_Gamma_tilde_operator}.
\end{proof}

Lemma \ref{lemma_2_star_to_star_H_product} can also be interpreted as follows: using that $\alpha_{-H}=e^{\widetilde{\Gamma}_{-H}} = e^{-\widetilde{\Gamma}_H}=\alpha_H^{-1}$, we have
\begin{equation}\label{eq_2_abstract_star_product}
        \big(\alpha_H^{-1}(F)\star \alpha_H^{-1}(G)\big)=\alpha_H^{-1}\big(F\star_{H}G\big).
\end{equation}
We notice that \eqref{eq_2_abstract_star_product} is well defined for $F, \ G\in \mathcal{F}_{reg}(M,h)$ since $F\in \mathcal{F}_{reg}(M,h) \Rightarrow \alpha_H^{\pm 1}(F)\equiv e^{\pm\widetilde{\Gamma}_H}(F)\in \mathcal{F}_{reg}(M,h)$; however, $F\star_{H}G$ is well defined for microcausal functionals as well, therefore an extension of the $\star$ products could be obtained if we were able to define $e^{\pm\widetilde{\Gamma}_H}(F)$ for microcausal functionals. Of course, the drawback is that any definition of the star product relying on \eqref{eq_2_abstract_star_product}, will depend on the symmetric part $H$ of the Hadamard state $\omega$, which is not uniquely fixed by the background geometry. We shall first address the latter issue and then come back to the former.\\

We denote by $\mathcal{F}_{reg}(M,h,H)$, $\mathcal{F}_{\mu c}(M,h,H)$ the algebra generated by $\mathcal{F}_{reg}(M,h)$, $\mathcal{F}_{\mu c}(M,h)$ with the product $\star_H$ defined in \eqref{eq_2_star_G_H_product_microcausal}, since both algebras are of the form $\bigoplus_{j}\hbar^j\mathcal{F}_{\cdot}(M,h)$ we can endow them with the topology of pointwise convergence, \textit{i.e.} a sequence of formal series $F_k$ converge to $F$ if and only if for each order $j$ of $\hbar$, the sequence of functional $(F_j)_k \to F_j$ in the strong convenient topology (\textit{c.f.} \eqref{eq_1_strong_conv_seminorms_1}). For notational convenience, we also denote by $\alpha_H$ the operator $e^{\Gamma_H}$.

\begin{lemma}\label{lemma_2_regular_algebra_isomorphism}
	Let $H$, $H'$ be Hadamard quasifree states, there is a canonical isomorphism
$$
	\alpha_{H,H'}: \mathcal{F}_{{reg}}(M,h,H')\rightarrow  \mathcal{F}_{{reg}}(M,h,H): \alpha_{H'}^{-1}(F)\mapsto \alpha_{H}^{-1}\big(   \widetilde{ F} \big) 
$$
with
\begin{equation}\label{eq_algebra_isomorphism}
	\begin{aligned}
		\widetilde{F}(\varphi) & =\big(e^{\widetilde{\Gamma}_d}F\big)(\varphi)
	\end{aligned}
\end{equation}
where $d=H'-H\in C^{\infty}(M\times M)$.
\end{lemma}
\begin{proof}
Using \eqref{eq_algebra_isomorphism}, and the fact that $e^{\widetilde{\Gamma}_H}e^{\widetilde{\Gamma}_d}= e^{\widetilde{\Gamma}_{H+d}}=e^{\widetilde{\Gamma}_{H'}}$ we have
\begin{align*}
	&e^{\widetilde{\Gamma}_{H'}}(F)= e^{\widetilde{\Gamma}_{H+d}}(F)=e^{\widetilde{\Gamma}_H}e^{\widetilde{\Gamma}_d}(F)=e^{\widetilde{\Gamma}_H}(\widetilde F).
\end{align*}
By the very same calculation it is possible to show that $\alpha_{H,H'}\circ \alpha_{H',H}=id$. Thus the mapping so constructed is a bijection. Finally we claim that
$$
	e^{\widetilde{\Gamma}_{H'-H}}\big(F \star_H G \big)= \Big(e^{\widetilde{\Gamma}_{H'-H}}(F)\Big)\star_{H'} \Big(e^{\widetilde{\Gamma}_{H'-H}}(G)\Big)=\widetilde F \star_{H'} \widetilde G.
$$
which would imply that $\alpha_{H,H'}$ is a $\star$-isomorphism. By Lemma \ref{lemma_2_star_to_star_H_product},
$$
	F\star_H G= e^{\widetilde{\Gamma}_H}\Big( e^{-\widetilde{\Gamma}_H}(F) \star e^{-\widetilde{\Gamma}_H}(G)\Big),
$$
therefore
$$
\begin{aligned}
    e^{\widetilde{\Gamma}_{H'-H}}\Big(e^{\widetilde{\Gamma}_H}(F) \star_H e^{\widetilde{\Gamma}_H}(G) \Big)
    &= e^{\widetilde{\Gamma}_{H'-H}}\Big(e^{\widetilde{\Gamma}_H}\circ \mathcal{M}\Big(e^{-\widetilde{\Gamma}_H\otimes id - id \otimes \widetilde{\Gamma}_H+ \Gamma_{\frac{i}{2}\mathcal{G}}}\big( F \otimes G \big)\Big)\Big)\\ 
    &= e^{\widetilde{\Gamma}_{H'}}\circ \mathcal{M}\Big(e^{-\widetilde{\Gamma}_H\otimes id - id \otimes \widetilde{\Gamma}_H+ \Gamma_{\frac{i}{2}\mathcal{G}}}\big( F \otimes G \big)\Big)\Big)\\ 
    &= \mathcal{M}\Big(e^{-\widetilde{\Gamma}_{H-H'}\otimes id - id \otimes \widetilde{\Gamma}_{H-H'}+ \Gamma_{H'-\frac{i}{2}\mathcal{G}}}\big( F \otimes G \big)\Big)\Big)\\ 
    &= \widetilde F \star_{H'} \widetilde G.
\end{aligned}
$$
\end{proof}
This helps addressing the issue of dependence on the Hadamard state chosen to perform the $\star_H$ product. In fact, let $\mathrm{Had}(M,h)$\footnote{From here we will identify elements $\omega$ of $\mathrm{Had}(M,h)$ with their symmetric parts due to the fact that the causal propagator $\mathcal{G}_{(M,h)}$ is uniquely defined. When we will be investigating Wick powers, we will also deliberately identify H with the Hadamard parametrix \eqref{eq_2_Hadamard_distribution} since microlocal functionals need just $H$ to be defined along the diagonal, in which case, by \cite[Theorem 5.1]{radzikowski1996micro} it coincides with the expression \eqref{eq_2_Hadamard_distribution}, therefore identifying $\mathrm{Had}(M,h)$ with $\mathrm{Par}(M,h)$.} be the collection of Hadamard states of $\mathcal{F}_{reg}(M,h)$, define the abstract algebra of Wick ordered regular functionals
\begin{equation}\label{eq_2_abstract_regular_algebra}
\mathfrak{A}_{reg}(M,h)= \bigsqcup_{\omega\in \mathrm{Had}(M,h)}  \alpha^{-1}_H\big(\mathcal{F}_{reg}(M,h,H)\big)/\sim ,
\end{equation}
where $H$ is the symmetric part of $\omega$ and $\sim$ is the equivalence relation induced by the isomorphism $\alpha_{H,H'}$ of Definition \ref{lemma_2_regular_algebra_isomorphism}. The latter algebra naturally inherits the $\star$ product 
\begin{equation}\label{eq_2_abstract_regular_algebra_product}
	\star : \mathfrak{A}_{reg}(M,h) \times \mathfrak{A}_{reg}(M,h)\to \mathfrak{A}_{reg}(M,h): (A,B)\mapsto  \alpha_H^{-1}\big( \alpha_H(A)\star_H\alpha_H(B)\big).
\end{equation}
We stress that $\alpha_H: (\mathfrak{A}_{reg}(M,h),\star) \to (\mathcal{F}_{reg}(M,h,H),\star_H)$ is an algebra isomorphism, therefore  $\mathfrak A_{reg}(M,h)\simeq \mathcal{F}_{reg}(M,h)$ inherits the topology of pointwise convergence previously described. 

The advantage of using the abstract $*$-algebra $\mathfrak A_{reg}(M,h)$ lies in the possibility of extending it to microcausal functionals in such a way that the $\star$-product is independent from the $H$ chosen for its construction at the functional algebra level (see \eqref{eq_2_star_G_H_product_microcausal}). Recall that, by \ref{lemma_1_regular_density}, $\mathcal{F}_{reg}(M,h)\subset \mathcal{F}_{\mu c}(M,h)$ is dense, we then define the abstract algebra of Wick ordered microcausal functional as the sequential completion of $\mathfrak{A}_{reg}(M,h)$, that is
\begin{equation}\label{eq_2_abstract_muc_algebra}
	 \mathfrak{A}_{\mu c}(M,h)\doteq \big\{\lim_{n\rightarrow \infty}\alpha^{-1}_H(F_n): \mathfrak{A}_{reg}(M,h)\simeq \mathcal{F}_{reg}(M,h) \supset \{F_n\} \to F \in \mathcal{F}_{\mu c}(M,h) \big\},
\end{equation}
with $\star$ product \eqref{eq_2_abstract_regular_algebra_product}. We shall denote by $\alpha^{-1}_H(F)$, $F\in  \mathcal{F}_{\mu c}(M,h)$, the elements of this algebra. Unfortunately those elements fail to remain microcausal functionals (or formal series thereof), as an example
\begin{equation}\label{eq_2_Wick_powers}
	\phi^k_{(M,h,f,H)}=\alpha_H^{-1}\big( \phi^k_{(M,h)}(f) \big)= \int_M f(x):\phi^k:_H(x)d\mu_g(x),
\end{equation}
where $:\phi^k:_H$ denotes the usual ordering prescription of Wick powers of fields, fails to be a functional since evaluation at a configuration $\varphi\in C^{\infty}(M)$ yields a divergent product of distributions. The case $k=2$, is emblematic of this since $:\phi^2:_H(x)=\int_M \delta(x,y)(\phi(x)\phi(y)-H(x,y))d\mu_g(y)$ is not a well defined distribution in $\mathcal{D}'(M)$. We can therefore interpret the quantity $:\phi^k:_H(x)$ as a \textit{abstract algebra} element.\\

A potential issue with the prescription in \eqref{eq_2_abstract_muc_algebra} is that Wick ordering might differ when we use different sequences converging to the same microcausal functional, namely given $F_k \rightarrow F \leftarrow F'_k$ does it implies $e^{\Gamma_H}(F_k-F'_k)\to 0$?
% \begin{lemma}\label{lemma_2_unique_conv}
%     Let $\{F_k\}_{k\in \mathbb{N}}$, $\{F'_k\}_{k\in \mathbb{N}}$ be sequences of regular functionals converging to some microcausal functional $F$. Then for each perturbative order $j\in \mathbb{N}$ in $\hbar$,  
%     $$
%     \lim_{k\rightarrow \infty}\Big(\alpha^{-1}_H\big((F_j)_k-(F'_j)_k\big) \Big) = 0.
%     $$
%     Therefore, $\lim_{k \to \infty} (\alpha_H(F_k)-\alpha_H(F'_k))=0$.
% \end{lemma}
\begin{lemma}\label{lemma_2_unique_conv}
    The mapping $\alpha_H^{-1}: \mathcal{F}_{reg}(M,h) \to \mathfrak A_{reg}(M,h) $ is sequentially continuous in the strong convenient topology.
\end{lemma}
\begin{proof}
For simplicity we consider a sequence $\{F_k\}\subset \mathcal{F}_{reg}(M,h)$ with $F_k\to 0$, then we have to show that at each perturbative order $j\in \mathbb N$, $(\alpha_H(F_k))_j=\langle d^{2j}F_k[\cdot], \otimes^j H \rangle \to 0$. %Then the argument below can be repeated for all other orders and, by definition of the topology of pointwise convergence, this suffices to show the claim. 
Similarly to what we did in the proof of Lemma \ref{lemma_1_regular_density}, if we consider the sequence $S_l$ of mollifiers (see \cite[Theorem 12 pp.68]{de2012differentiable}) strongly converging to the identity mapping and denote %$ F_{j,H}(\varphi)=\langle d^{2j}F[\varphi], \otimes^j H \rangle$, $F'_{j,H}(\varphi)=\langle d^{2j}F'[\varphi], \otimes^j H \rangle$, 
$ F_{j,\widetilde H,l}(\varphi)=\langle d^{2j}F[\varphi], \otimes^j \widetilde H_l \rangle$, % $ F'_{j,\widetilde H,l}(\varphi)=\langle d^{2j}F'[\varphi], \otimes^j \widetilde H_l \rangle$, 
where $\widetilde H_l = H\circ (S_l\otimes S_l) \in C^{\infty}(M\times M)$. By the same argument we used in the proof of Lemma \ref{lemma_1_regular_density},
$$
     F_{j,\widetilde H,l} \to (\alpha_H(F_k))_j. % F_{j,H}; \quad F'_{j,\widetilde H,l}\to  F'_{j,H}.
$$
Thus if we show that, as $l \to 0 $, $ F_{j,\widetilde H,l}\to 0$ for all $j\in \mathbb{N}$, we conclude. The latter is a more convenient expression since given any $[a,b]\subset \mathbb{R}$, $\gamma\in C^{\infty}(\mathbb{R},\mathcal{U})$ with $\mathcal{U}\subset C^{\infty}(M,\mathbb{R})$ $CO$-open, $\mathcal{B}\subset C^{\infty}(M)^l$ bounded subset, we see that $\widetilde H$ is trivially a bounded subset of $C^{\infty}(M^2)$, therefore
$$
\begin{aligned}
    \sup_{\substack{t\in[a,b]\subset \mathbb{R}\\ \gamma\in C^{\infty}(\mathbb{R},\mathcal{U})\\(\psi_1,\ldots,\psi_k)\in \mathcal{B}}}\Big|\big( d^kF_{j,\widetilde H,l}\big)[\gamma(t)](\psi_1,\ldots,\psi_k)\Big|&= \sup_{\substack{t\in[a,b]\subset \mathbb{R}\\ \gamma\in C^{\infty}(\mathbb{R},\mathcal{U})\\(\psi_1,\ldots,\psi_k)\in \mathcal{B}}}\Big|\big( d^{k+2j}F\big)[\gamma(t)](\psi_1,\ldots,\psi_k,(\widetilde H_l)^j)\Big| \\
    \leq & \sup_{\substack{t\in[a,b]\subset \mathbb{R}\\ \gamma\in C^{\infty}(\mathbb{R},\mathcal{U})\\(\psi_1,\ldots,\psi_k,h_1,\ldots,h_{2j})\in \mathcal{B}'}}\Big|\big( d^{k+2j}F\big)[\gamma(t)](\psi_1,\ldots,\psi_k,h_1,\ldots,h_{2j})\Big| 
\end{aligned}
$$
which is arbitrarily small.
\end{proof}
\begin{remark*}
Notice that by Lemma \ref{lemma_2_unique_conv}, it is easy to show that the limit performed in \eqref{eq_2_abstract_muc_algebra} is independent from the sequence chosen to perform the limit. Indeed let $\{F_k\}_{k\in \mathbb{N}}$, $\{F'_k\}_{k\in \mathbb{N}}$ be sequences of regular functionals converging to some microcausal functional $F$. Then for each perturbative order $j\in \mathbb{N}$ in $\hbar$,  
$$
\lim_{k\rightarrow \infty}\Big(\alpha^{-1}_H\big((F_j)_k-(F'_j)_k\big) \Big) = 0.
$$
Therefore, $\lim_{k \to \infty} (\alpha_H(F_k)-\alpha_H(F'_k))=0$.
\end{remark*}

Combining Lemma \ref{lemma_2_unique_conv} with Lemma \ref{lemma_2_regular_algebra_isomorphism} we find that the mapping $\alpha_{HH'}$ can be extended to an \textit{abstract} isomorphism of algebras
$$
	\mathcal{F}_{\mu c}(M,h,H')\to \mathcal{F}_{\mu c}(M,h,H).
$$
Furthermore, by the proof of Lemma \ref{lemma_2_unique_conv}, $\alpha_H$ extends by continuity to a $\star$-algebra isomorphism
$$
    \alpha_H : \big( \mathfrak{A}_{\mu c}(M,h),\star_H \big) \to \big( \mathcal{F}_{\mu c}(M,h),\star_H \big).
$$

In \cite{drago}, was given a slightly different, nonetheless equivalent, characterization of the algebra $\mathfrak{A}_{\mu c}(M,h)$: let $\mathrm{Had}(M,h)$ be the set of Hadamard states of the algebra $\mathfrak A_{reg}(M,h)$. We define the bundle of microcausal algebras over $\mathrm{Had}(M,h)$ as
$$
	B_{\mu c}(M,h)=\bigsqcup_{H\in \mathrm{Had}(M,h)} \big( \mathfrak{A}_{\mu c}(M,h,H),\star_H\big),
$$
then
\begin{equation}\label{eq_2_muc_equivariant_algebra}
	\mathfrak{A}_{\mu c} (M,h)= \big\{ A \in \Gamma(\mathrm{Had}(M,h) \leftarrow B_{\mu c}(M,h,H)) : \alpha_{H,H'} (A(H')) =A(H) \big\},
\end{equation}
therefore elements of the abstract algebra $\mathfrak{A}_{\mu c} (M,h)$ are $\alpha$-equivariant sections of the bundle $B_{\mu c}(M,h)$. The connection between \eqref{eq_2_muc_equivariant_algebra} and \eqref{eq_2_abstract_muc_algebra} is that $A(H)=\alpha^{-1}_H(F)$ for some $F\in \mathcal{F}_{\mu c}(M,h)$. Elements of \eqref{eq_2_muc_equivariant_algebra} posses a $\star$ product operation defined by
\begin{equation}\label{eq_2_muc_equivariant_algebra_prod}
	(A\star B)(H)= A(H)\star_H B(H).
\end{equation}
This is a well defined product due to Lemma \ref{lemma_2_regular_algebra_isomorphism}. %We stress that the $\star$ product of abstract algebra $\mathfrak{A}_{\mu c} (M,h)$ is actually independent from the parametrix chosen and is moreover equivalent to \eqref{eq_2_abstract_star_product}; granted, we can always represent the algebra by means of some parametrix $H$ to a more concrete form which is usually more suited for calculations, however, \eqref{eq_2_muc_equivariant_algebra}, \eqref{eq_2_muc_equivariant_algebra_prod} do eliminate the reliance of the $\star$ product in $\mathcal{F}_{\mu c}(M,h,H)$ on the choice of a parametrix.\\

Call $\mathcal{D}: \mathfrak{M} \rightarrow \mathfrak{LCS}$, the functor that associates to each manifold $M$ the locally convex space $\mathcal{D}(M)\equiv C^{\infty}_c(M)$, and to each causally convex embedding $\psi:M\rightarrow M'$ the mapping $\mathcal{D}(\psi): \mathcal{D}(M) \rightarrow \mathcal{D}(M'): f \mapsto f \circ \psi^{-1} $. We remark that the dual functor of $\mathcal{D}$ is the contravariant functor $\mathcal{E}: \mathfrak{M} \rightarrow \mathfrak{LCS}$, defined by $\mathcal{E}:M\mapsto \mathcal{E}(M)\equiv C^{\infty}(M): \psi \mapsto \mathcal{E}(\psi)$ where $\mathcal{E}(\psi): \mathcal{E}(M')\rightarrow \mathcal{E}(M):f'\mapsto \psi^{*}f'=f'\circ \psi$.

\begin{definition}\label{def_2_net_of_observables}
A functor 
$$
    \mathfrak{A}: \mathfrak{BckG} \rightarrow \mathfrak{Alg}
$$
from the category of background geometries to the category of *-algebras with unity and *-morphisms with the following property:\\	
\textbf{Scaling} if $S_{\lambda} :\Gamma^{\infty}(M\leftarrow HM) \rightarrow \Gamma^{\infty}(M\leftarrow HM) $ represents physical scaling (see \eqref{eq_2_physical_scaling} for the definition), then there is a $*$-isomorphism $\sigma_\lambda: \mathfrak{A}(M,h)\rightarrow \mathfrak{A}(M,S_{\lambda}(h))$. The assignment of a $*$-algebra $\mathfrak A(M,h)$ to each background geometry $(M,h)$ together with the $*$-morphism $\mathfrak A(\chi):\mathfrak{A}(M,h) \to \mathfrak{A}(M',h') $ to each morphism $\chi: (M,h) \to (M',h')$ creates a \textit{net of algebras of observables}.
\end{definition}

The latter property can be reformulated in categorical language as follows. Let $\Lambda: \mathfrak{BckG}\rightarrow \mathfrak{BckG}$ be the the functor implementing physical scaling on background geometries, then $\sigma : \Lambda \Rightarrow \mathfrak{A}$ is a natural transformation. Note that since $S_{\lambda}\circ S_{\lambda^{-1}}=id$, then each $\sigma_{\lambda}$ conveniently tuns out to be a $*$-isomorphism. Similarly, we can introduce another natural transformation $s:\Lambda\Rightarrow \mathfrak{D}$ defined by
$$
	s_{\lambda}: \mathcal{D}(M) \rightarrow \mathcal{D}(M) : f \mapsto \lambda^n f.
$$
Then again each $s_{\lambda}$ defines an isomorphism with inverse $s_{\lambda^{-1}}$.\\

\begin{definition}\label{def_2_covariant_quantum_field}
A \textit{locally covariant scalar quantum field} $\Phi$ is a natural transformation
$$
    \Phi : \mathcal{D} \Rightarrow \mathfrak{A}.
$$
such that, if we fix a background geometry $(M,h)$, then 
$$
    \Phi_{(M,h)}: D(M) \rightarrow \mathfrak{A}(M,h)
$$
is an algebra valued distribution. 
\end{definition}

In particular, Definition \ref{def_2_covariant_quantum_field}, entails that $\Phi_{(M,h)}$ satisfies the following commutative diagrams
\begin{center}
\begin{tikzcd}
	(M,h) \arrow[rrr, "\chi"] \arrow[ddd]  \arrow[dr,"\mathcal{D}"] &  &   & (M'h') \arrow[ddd] \arrow[dl,"\mathcal{D}"]\\
		 & \mathcal{D}(M)\arrow[d, "\Phi_{(M,h)}"] \arrow[r,"\mathcal{D}(\chi)"]  & \mathcal{D}(M')  \arrow[d, "\Phi_{(M',h')}"] &  \\
		 & \mathfrak{A}(M,h) \arrow[r,"\mathfrak{A}(\chi)"] & \mathfrak{A}(M',h')  &\\
	(M,h) \arrow[ur,"\mathfrak{A}"] \arrow[rrr, "\chi"] &  & & (M',h') \arrow[ul,"\mathfrak{A}"] 	 
\end{tikzcd}
\end{center}	
In particular if we consider the scaling transformation $S_{\lambda}$ as an element of $ \mathrm{Hom}(HM)$ to the above diagram the scaling transformation induces the following diagram:
\begin{center}
\begin{tikzcd}
	(M,h) \arrow[rrr, "\Lambda \equiv S_{\lambda}"] \arrow[ddd]  \arrow[dr,"\mathcal{D}"] &  &   & (M,h_{\lambda}) \arrow[ddd] \arrow[dl,"\mathcal{D}"]\\
		 & D(M)\arrow[d, "\Phi_{(M,h)}"]  \arrow[r,"s_{\lambda}"] & D(M)    \arrow[d, "\Phi_{(M,h_{\lambda})}"] &  \\
		 & \mathfrak{A}(M,h) \arrow[r,"\sigma_{\lambda}"]  & \mathfrak{A}(M,h_{\lambda})  &\\
	(M,h) \arrow[ur,"\mathfrak{A}"] \arrow[rrr, "\Lambda \equiv S_{\lambda}"] &  & & (M,h_{\lambda}) \arrow[ul,"\mathfrak{A}"] 	 
\end{tikzcd}
\end{center}
We see that using the above diagram is always possible to compare the scaled field $\Phi_{(M,h_\lambda)}$ to the unscaled field $\Phi$ in the algebra $\mathfrak{A}(M,h)$ by considering the new field $\sigma_{\lambda}^{-1}\circ \Phi_{(M,h)} \circ s_{\lambda}$.\\

An immediate consequence of Definition \ref{def_2_covariant_quantum_field} is that naturality of $\Phi$ implies the following conditions:
\begin{itemize}
\item \textbf{Locality} If $\chi:M \rightarrow M'$ is an inclusion, then $\mathfrak{A}(\chi)$ is injective.
\item \textbf{Covariance} If $\chi: (M,h) \rightarrow (M'.h')$ is a causality preserving, isometric embedding then $\mathfrak{A}(M,h) \subset \mathfrak{A}(M',h')$
\end{itemize}

\begin{proposition}\label{prop_net_obs}
    The assignment $\mathfrak{A}_{\mu c}: \mathfrak{BckG}\rightarrow \mathfrak{Alg}:(M,h)\rightarrow \mathfrak{A}_{\mu c}(M,h)$ is a functor as per Definition \ref{def_2_net_of_observables}. 
\end{proposition} 
\begin{proof}
The construction of the algebra relying on \eqref{eq_2_muc_equivariant_algebra} defines a mapping 
$$
	(M,h) \rightarrow \mathfrak{A}_{\mu c}(M,h),
$$
if $\chi: (M,h) \rightarrow  (M',h')$ is a causality preserving, isometric embedding, let $A\in \mathfrak{A}_{\mu c}(M,h)$, then for each $H\in \mathrm{Had}(M,h)$, $A(H)\in \mathcal{F}_{\mu c}(M,h,H)$ has compact support contained in $M$. Since $\chi$ is an isometry, for any Hadamard state $H'$ of $(M',h')$, $H=\chi^{*}H'$ is a Hadamard state for $(M,h)$. Therefore we define an element of  $A\in \mathfrak{A}_{\mu c}(M',h')$ by 
$$
	A'(H')(\varphi')= A(\chi^{*}H')(\chi^{*}\varphi')
$$
clearly $\mathrm{supp}(A'(H'))\subset \chi(M)$ for all $H' \in \mathrm{Had}(M',h')$, as a result we can identify $A(H)$ as an element of the bigger algebra $A_{\mu c}(M',h')$.
Next we tackle scaling. We have to construct an algebra isomorphism $\sigma_{\lambda}$. Consider
\begin{equation}\label{eq_scaling_varphi}
	\widetilde{S}_{\lambda}: M\times \mathbb{R}\rightarrow M\times \mathbb{R} :(x,\varphi) \mapsto (x,\varphi_{\lambda})=(x,\lambda^{\frac{n-2}{2}}\varphi).
\end{equation}
it is possible to show that the action
$$
	A_f(\mathcal{\varphi})=\int_Mf(x)j^1\varphi^*\mathbb{L}=\int_Mf(x)\big( g^{\mu\nu}\nabla_{\mu}\varphi \nabla_{\nu}\varphi +m^2\varphi^2+\kappa R(g)\varphi^2\big)(x) d\mu_g(x)
$$
is invariant under the transformations $h\mapsto S_{\lambda}h$, $\varphi \mapsto \widetilde{S}_{\lambda}(\varphi)$, that is
$$
	A_f[M,h](\mathcal{\varphi})=A_{f}[M,S_{\lambda}(h)](\widetilde{S}_{\lambda}(\varphi)).
$$
Defining
\begin{equation}\label{eq_scaling_phields}
	\sigma^{-1}_{\lambda}(\alpha^{-1}_{H} (F))= \alpha^{-1}_{H_{\lambda}} (F\circ \widetilde{S}_{\lambda}),
\end{equation}
$\sigma_{\lambda}$ becomes the sought isomorphism.
\end{proof}

%We are now in position to sort out some relevant classes of quantum fields: $c$-number fields and linear quantum fields.

Let $ \mathfrak{A}$ be a functor as per Definition \ref{prop_net_obs}.
\begin{itemize}
    \item A \textit{C-number} field $C$ is a locally covariant scalar quantum field satisfying the following property: for any background geometry $(M,h)$ and any $H\in \mathrm{Had}(M,h)$,
\begin{equation}\label{eq_C-number_field}
	ev_{\varphi}d\Big( \alpha_H\big( C_{(M,h)}(f)\big)\Big) =0, \ \ \forall \varphi \in C^{\infty}(M),\ \forall f \in \mathcal{D}(M).
\end{equation}
    \item Moreover, we say that a scalar quantum field $\Phi$ is \textit{linear} if for each background geometry $(M,h)$, $f \in \mathcal{D}(M)$ and $H\in \mathrm{Had}(M,h)$ the derivative field $\alpha_H^{-1}\big(d\alpha_H\big(\Phi_{(M,h,H)}(f)\big)\big)$ is a C-number.
\end{itemize}

It's easy to verify that the quantum field defined in \eqref{eq_2_Wick_powers} is a linear scalar quantum field for $k=1$. For notational sake let us denote by $\mathfrak{A}_{\mu loc}(M,h)$ the set 
\begin{equation}\label{eq_2_abstract_muloc_algebra}
    \Big\{\lim_{n\rightarrow \infty}\alpha^{-1}_H(F_n): \mathfrak{A}_{reg}(M,h)\simeq \mathcal{F}_{reg}(M,h) \supset \{F_n\} \to F \in \mathcal{F}_{\mu loc}(M,h) \Big\}.
\end{equation}
Of course, we have $\mathfrak{A}_{\mu loc}(M,h) \subset \mathfrak{A}_{\mu c}(M,h)$
\begin{definition}\label{def_2_Wick_quantum_powers}
Given the functor $\mathfrak{A}_{\mu c}$ on the category of background geometry $\mathfrak{BckG}$ and a locally covariant linear scalar quantum field $\Phi$, its Wick powers, $\lbrace {\Phi}^k \rbrace_{k \in \mathbb{N}}$, are locally covariant scalar quantum fields satisfying the following axioms:
\begin{itemize}
\item[$(i)$] \textbf{(Locality and covariance)}: each $\Phi^k $ is a natural transformation $\Phi^k: \mathcal{D} \Rightarrow \mathfrak{A}_{\mu c}$ such that for each background geometry $(M,h)$ we have $\Phi^0_{(M,h)}=1_{\mathfrak{A}}$ and $\Phi^1_{(M,h)}\equiv \Phi_{(M,h)}$;
\item[$(ii)$] \textbf{(Scaling)}: each $\Phi^k_{(M,h,H)}$ is almost homogeneous of degree $k\frac{n-2}{2}$ with respect to physical scaling, that is there exists some $l \in \mathbb{N}$ and some other locally covariant scalar quantum fields $\{\psi^j \}_{j=1,2,\ldots , l}$ scaling almost homogeneously with degree $k\frac{n-2}{2}$ such that
$$
	\sigma_{\lambda}\big(\Phi^k_{(M,h_{\lambda})}\big) =\lambda^{k\frac{n-2}{2}} \Phi^k_{(M,h)} + \lambda^{k\frac{n-2}{2}} \sum_{j=1}^l \ln^j(\lambda) \psi^j_{ (M,h)} \ ;
$$
\item[$(iii)$] \textbf{(Algebraic)}: given any background geometry $(M,h)$, each ${\Phi}^k$ satisfies the hermiticity condition
$$
	{\Phi}^k_{(M,h)}(f)^*={\Phi}_{(M,h)}^k(\bar{f}),
$$
and
$$
	d{\Phi}^k_{(M,h)}(f)(\psi) = k \Phi^{k-1}_{(M,h)}(\psi f),
$$
for all $\psi \in C^{\infty}(M)$ and all $ f \in \mathcal{D}(M)$;

\item[$(iv)$] \textbf{(Parameterized microlocal spectrum condition)}: given any compactly supported smooth variation of the background geometry $\mathbb{R}^d \ni s \mapsto h_s\in \Gamma^{\infty}(M\leftarrow HM)$ with $h_0=h$, any $H\in \mathrm{Had}(M,h)$; identify $ \mathfrak{A}_{\mu loc}(M,h_s)$ with $ \mathfrak{A}_{\mu loc}(M,h)$\footnote{We will be more precise in the proof of Theorem \ref{thm_2_existence_Wick_polynomials} on how to construct this isomorphism of off-shell microlocal algebras. We stress that in the on-shell case, in \cite[Definition 3.5]{khavkine2016analytic}, this isomorphism can be constructed even for the whole microcausal algebra (\textit{e.g.} see Lemma 4.1 in \cite{hollands2001local}. Then for any $f\in \mathcal{D}(M)$, $\alpha_H\big(\Phi^k_{(M,h_s)}(f)\big) \in \mathcal{F}_{loc}(M,h)$). This is however not true in our case for in general two algebras of microcausal functionals with respect to different metric are in general not isomorphic. For details see \cite[Remark 3.5]{drago2017generalised}.} and consider the distribution
$$
    \mathcal{D}(M\times \mathbb R^d) \ni f\otimes z \mapsto \Big\langle \alpha_H\big(\Phi^k_{(M,h_s)}\big)(\varphi), f\otimes z\Big\rangle=\int_{M\times \mathbb R^d}\Phi^k_{(M,h_s,H,\varphi)}(s,x)f(x)z(s)d\mu_{g}(x);
$$
then its integral kernel, $\Phi^k_{(M,h_s,H,\varphi)}(s,x)$, is jointly smooth in $(s,x)$ for each $\varphi \in C^{\infty}(M)$ and any $H\in \mathrm{Had}(M,h)$.
\end{itemize} 
\end{definition}
We remark that condition $(iii)$ above implies that Wick powers are local (in the sense of Lemma \ref{def_1_func_classes}) for if we calculate the second derivative in $(\psi_1,\psi_2)$ with disjoint support, then $\Phi^{k-2}_{(M,h)}(\psi_1\psi_2f)\equiv \Phi^{k-2}_{(M,h)}(0)=0$. Thus in $(iv)$ we are legitimized in assuming that $\Phi^k_{(M,h_s)}(f)\in \mathfrak{A}_{\mu loc}(M,h)$. We will make explicit  the identification $ \mathfrak{A}_{\mu loc}(M,h_s)\simeq \mathfrak{A}_{\mu loc}(M,h)$ in the proof of Theorem \ref{thm_2_existence_Wick_polynomials}. Finally, let us remark that condition $(iv)$ can be also framed by requiring that $\Phi^k_{(M,h_s,H,\varphi)}(s,x)$ is a smooth distribution in $\mathcal{D}'(M\times \mathbb{R}^d )$ that is an element of $C^{\infty}(M\times \mathbb{R}^d)$. The purpose of $(iv)$ is to replace analytic dependence of Wick powers of the metric $g$ from the parameter of the deformation $s$ originally stated in \cite{hollands2001local}. We also stress that the Definition \ref{def_2_Wick_quantum_powers} differs slightly from that given in \cite[Definition 2.2]{khavkine2016analytic} since our axiom holds off-shell.\\

\subsection{Uniqueness of Wick powers}\label{section_Wick_uniqueness}

Under the hypothesis of Definition \ref{def_2_Wick_quantum_powers}, one can precisely characterize how much two families of \textit{off-shell} Wick powers on curved spacetime are allowed to vary: in particular, we find below that their difference is tightly constrained by \eqref{eq_Moretti_Kavhkine_Wick}. 

\begin{theorem}[Theorem 3.1 \cite{khavkine2016analytic}]\label{thm_2_Moretti_Kavhkine}
	Let $\{ \widetilde{{\Phi}}^k\}_{k \in \mathbb{N}}$, $\{{{\Phi}}^k\}_{k \in \mathbb{N}}$ be two families of Wick powers with respect to the same linear scalar quantum field $\Phi$ in a net of algebras $\mathfrak{A}$ as in Definition \ref{def_2_Wick_quantum_powers}. The the difference between the two families can be parameterized as follows:
\begin{equation}\label{eq_Moretti_Kavhkine_Wick}
	\widetilde{{\Phi}}_{(M,h)}^k={{\Phi}}_{(M,h)}^k+\sum_{j=0}^{k-1}\binom{k}{j}C_{k-j}(h){\Phi}^j_{(M,h)}.
\end{equation}
Where $C_{k-j}(h):x \mapsto  C_{k-j}(h)(x)$ and 
$$
	C_{k-j}(h)(x)= C_{k-j}\left(x^{\mu},g_{\mu\nu},S^{\beta}_{\mu\nu},R^{\beta}_{\alpha\mu\nu},\ldots,\nabla_{\alpha_3\ldots\alpha_r}R^{\beta}_{\alpha_1\mu\alpha_2},m^2,\ldots,\nabla_{\alpha_1\ldots,\alpha_r}m^2,\kappa,\ldots,\nabla_{\alpha_1\ldots,\alpha_r}\kappa\right)
$$
the latter being functions of scalar polynomials constructed from its coordinates. 
\end{theorem}

\begin{lemma}[Lemma 3.2 \cite{khavkine2016analytic}]\label{lemma_2_Moretti_Kavhkine}
    In the same hypothesis of Theorem \ref{thm_2_Moretti_Kavhkine}, the smooth functions $x\mapsto C_{k-j}(h)(x)$ have image depending only on the germ of $h $ at $x\in M$ define weakly regular differential operators $\Gamma^{\infty}(M\leftarrow HM) \rightarrow C^{\infty}(M)$. Moreover, they satisfy the covariance identity $\chi^{*}C_{k-j}(h)(x)=C_{k-j}(\chi^{*}h)(x)$ for any $\chi \in \mathrm{Hom}(\mathscr{BckG})$, and each $C_{j}$ scales almost homogeneously with degree $j(n-2)/2$ under physical scaling.
\end{lemma}

We first recall\footnote{See \cite{khavkine2016analytic} Definition 2.2} that given any mapping $D:\Gamma^{\infty}(M\leftarrow B_1)\rightarrow\Gamma^{\infty}(M\leftarrow B_2)$ between sections of two fiber bundles, then $D$ is weakly regular if given any compactly supported variation $\sigma_s$\footnote{A compactly supported variation of $\sigma_0\in \Gamma^{\infty}(M \leftarrow B_1)$ is a smooth family of sections $\sigma_s\in \mathrm{pr}_2^{*}B_1 \rightarrow\mathbb{R}^d\times M$, with $\mathrm{pr}_2:\mathbb{R}^d\times M \rightarrow M:(s,x)\mapsto x$, such that there is a compact subset $K$ of $M$ for which $\sigma\vert_{\mathbb{R}^d\times M \backslash \mathrm{pr}_2^{-1}(K)}\equiv \sigma_0$ is constant along the $\mathbb{R}^d$ factor.} of $\sigma_0 \in \Gamma^{\infty}(M\leftarrow B_1)$, $D(\sigma_s) $ is a smooth compactly supported variation for $D(\sigma_0)$.

\begin{proof}
We start by showing Lemma \ref{lemma_2_Moretti_Kavhkine} using induction on $k$. When $k=1$, condition $(i)$ implies $\widetilde{\Phi}_{(M,h)}-\Phi_{(M,h)}=0$, therefore their difference can be expressed in terms of a smooth function, the zero section, which satisfies all conditions above. Then suppose that all counterterms \eqref{eq_Moretti_Kavhkine_Wick} have been calculated up to $C_{k-1}(h)$, and satisfy all the properties of Lemma \ref{lemma_2_Moretti_Kavhkine}. Then consider
\begin{equation}\label{eq_2_Psi_kernel_KM}
    \Psi_{(M,h,H,\varphi)}(x)=\widetilde{\Phi}_{(M,h,H,\varphi)}^k(x)-{{\Phi}}_{(M,h,H,\varphi)}^k(x)+\sum_{j=0}^{k-2}\binom{k}{j}C_{k-j}(h)(x){\Phi}^j_{(M,h,H,\varphi)}(x).
\end{equation}

Using the induction hypothesis we claim that $d\Psi_{(M,h,H,\varphi)}=0$ for all $\varphi \in C^{\infty}(M)$. 
$$
\begin{aligned}
	d\Psi_{(M,h,H,\varphi)}(x)&= k \left(\widetilde{\Phi}^{k-1}_{(M,h,H,\varphi)}(x)-\Phi^{k-1}_{(M,h,H,\varphi)}(x)\right)\\ 
	&-j \sum_{j=1}^{k-2}\binom{k}{j}C_{k-j}(h)(x)\Phi^{j-1}_{(M,h,H,\varphi)}(x)\\
	&= k \left(  \sum_{j=0}^{k-3}\binom{k-1}{j}C_{k-j-1}(h)(x)\Phi^{j-1}_{(M,h,H,\varphi)}(x)\right) \\ 
	&- \sum_{j=1}^{k-2}\binom{k}{j}C_{k-j}(h)(x)\Phi^{j-1}_{(M,h,H,\varphi)}(x)\\
	&= \sum_{j=1}^{k-2}\binom{k-1}{j}C_{k-j}(h)(x)\Phi^{j-1}_{(M,h,H,\varphi)}(x) \\
 &-\sum_{j=1}^{k-2}\binom{k}{j}C_{k-j}(h)(x)\Phi^{j-1}_{(M,h,H,\varphi)}(x)=0
\end{aligned}
$$
Therefore $\Psi_{(M,h,H)}$ ought to be the image through $\alpha_H^{-1}$ of a functional with empty support, with local and covariant dependence on the geometric data such that $\Psi_{(M,h,H)} (\varphi): \mathcal{D}(M) \to \mathbb R$ is a distribution; therefore we can identify it with a functional, independent from $\varphi$ and therefore from $H$, of the form
$$
	\Psi_{(M,h)}=\int_Mf(x) C(h)(x)d\mu_g(x) 1_{\mathfrak A_{\mu c}(M,h)}.
$$
By construction, $\Psi$ is a locally covariant field with the right scaling property, and, using the parameterized microlocal spectrum condition, we can represent $C$ as a differential operator $\Gamma^{\infty}(M\leftarrow HM)\ni h\mapsto C(h)\in C^{\infty}(M)$. Moreover, by condition $(i)$ we deduce the covariance relation $\chi^{*}C_{k-j}(h)(x)=C_{k-j}(\chi^{*}h)(x)$, thus taking smaller and smaller neighborhood $U$ of $x$ with their relative embeddings $\chi:U \to M$ implies that $ C_k(h)(x)$ just depends on the germ of $h$ at $x$. Finally, if $\mathbb{R}^d\ni s \mapsto h_s$ is a compactly supported variation of $h_0$, a second application of condition $(iv)$,implies that $C_k(h_s)(x)$ is a jointly smooth mapping in $(s,x)$.

Next we check the scaling property of $C_k(h)$. Both $\widetilde{\Phi}_{(M,h)}^k$ and $\Phi_{(M,h)}^k$ do scale almost homogeneously with degree $k(n-2)/2$ by $(ii)$, in addition $C_{j}(h)$ for $0\leq j<k$ scales almost homogeneously with degree $(k-j)(n-2)/2$ by induction hypothesis.
% that is there are $C$-numbers $D_{j,q}(h)$ such that
% $$
%     C_{j}(h_{\lambda})=\lambda^{j\frac{n-2}{2}} C_{j}(h) +\lambda^{j\frac{n-2}{2}} \Big( \sum_{q=1}^l\ln^q(\lambda)D_{j,q}(h)\Big).
% $$
Now, $\Psi_{(M,h)}=C_{k}(h)1_{\mathfrak{A}_{\mu c}(M,h)}$, moreover \eqref{eq_scaling_phields} implies $\sigma_{\lambda}\big(1_{\mathfrak{A}_{\mu c}(M,h_{\lambda})}\big)=1_{\mathfrak{A}_{\mu c}(M,h)}$; consequently, we must have $\sigma_{\lambda}\Psi_{(M,h_{\lambda})}=C_{k}(h_{\lambda})1_{\mathfrak{A}_{\mu c}(M,h)}$. On the other hand, scaling both sides in \eqref{eq_2_Psi_kernel_KM}, we get
$$
\begin{aligned}
     \sigma_{\lambda}\Psi_{(M,h_{\lambda})}&= \lambda^{k\frac{n-2}{2}}\Big( \Phi^k_{(M,h,f)} + \sum_{q=1} \ln^q(\lambda) \psi^q_{ (M,h,f)}\Big)+ \lambda^{k\frac{n-2}{2}} \Big(\widetilde\Phi^k_{(M,h,f)} + \sum_{q=1} \ln^q(\lambda) \widetilde{\psi}^q_{ (M,h,f)}\Big)\\
    &\quad+ \sum_{j=0}^{k-2}\binom{k}{j} \lambda^{(k-j)\frac{n-2}{2}} \Big(C_{k-j}(h)+  \sum_{q=1}\ln^q(\lambda)D_{k-j,q}(h)\Big) \lambda^{j\frac{n-2}{2}}\Big( \Phi^{j}_{(M,h,f)} +\sum_{q=1} \ln^q(\lambda) \psi^q_{ (M,h,f)}\Big).\\
\end{aligned}
$$
for some locally covariant quantum fields $(\psi^q_{k-j})_{(M,h)}$, $(\widetilde\psi^q_{k-j})_{(M,h)}$. Simplifying terms in the above expression and comparing with we arrive at
$$
\begin{aligned}
     C_{k}(h_{\lambda})1_{\mathfrak{A}_{\mu c}(M,h)} =\sigma_{\lambda}\Psi_{(M,h_{\lambda})}&=\lambda^{k\frac{n-2}{2}} C_{k}(h) 1_{\mathfrak{A}_{\mu c}(M,h)}+\lambda^{k\frac{n-2}{2}} \Big( \sum_{q=1}^l(\theta_{q,k})_{(M,h)}\Big).
\end{aligned}
$$
This forces the quantum fields $\theta_{q,k}$ to be $C$-number fields.

% $$
% \begin{aligned}
%     &S_{\lambda}\widetilde{\Phi}_{(M,h,f)}^k- S_{\lambda}\Phi_{(M,h,f)}^k+\sum_{j=0}^{k-2}\binom{k}{j} S_{\lambda}C_{k-j}(h) S_{\lambda}\Phi^j_{(M,h,f)}\\
%     &= \lambda^{k\frac{n-2}{2}}\Big( \Phi^k_{(M,h,f)} + \sum_{q=1}^l \ln^q(\lambda) \psi^q_{ (M,h,f)}\Big)+ \lambda^{k\frac{n-2}{2}} \Big(\widetilde\Phi^k_{(M,h,f)} + \sum_{q=1}^l \ln^q(\lambda) \widetilde{\psi}^q_{ (M,h,f)}\Big)\\
%     &+ \sum_{j=0}^{k-2}\binom{k}{j} \lambda^{(k-j)\frac{n-2}{2}} \Big(C_{k-j}(h)+  \sum_{q=1}^lD^l_{k-j}(h)\Big) \lambda^{j\frac{n-2}{2}}\Big( \Phi^k_{(M,h,f)} +\sum_{q=1}^l \ln^q(\lambda) \psi^q_{ (M,h,f)}\Big).\\
%     &=\lambda^{k\frac{n-2}{2}} \sum_{q=1}^l \Big( \ln^q(\lambda) \psi^q_{ (M,h,f)}+\ln^q(\lambda) \widetilde{\psi}^q_{ (M,h,f)} \sum_{j=0}^{k-2}\binom{k}{j}  \sum_{q=1}^lD^l_{k-j}(h)\Big) \Big( \Phi^k_{(M,h,f)} +\sum_{q=1}^l \ln^q(\lambda) \psi^q_{ (M,h,f)}\Big).\\
% \end{aligned}
% $$
\end{proof}

Next we show Theorem \ref{thm_2_Moretti_Kavhkine}. Instead of giving a full account of the proof we just sketch it and refer to \cite{khavkine2016analytic, khavkine2019wick} for the full proof.

\begin{proof}

%---------------- COOCRDINATES ----------------

The proof is divided into six steps. We start by applying Peetre-Slov\'ak's theorem to the mapping $C_k:\Gamma^{\infty}(M\leftarrow HM)\to C^{\infty}(M)$, this will imply that in a neighborhood $U$ of any $x\in M$ $C_k$ has locally bounded order $r$ whenever the image of $j^rh$ belongs to $ U\times V^r\subset J^rHM$. The rest of the proof aims at enlarging the open subset $V^r$ to contain the whole $\pi^{-1}(U)$. Once this is done, we use diffeomorphism covariance to characterize $C_k$ in the whole $M$ and the scaling properties to globally bound the order $r$ of $C_k$. For future convenience we recall that a key asset in the proof of Theorem \ref{thm_2_covariant_identity} is that we can classify coordinates in $J^rH$ as
\begin{equation}\label{eq_2_covariant_coordinates}
\begin{aligned}
    \big( & x^{\mu},g_{\mu\nu},S^{\beta}_{\mu\nu},S^{\beta}_{\alpha\mu\nu},R^{\beta}_{\alpha\mu\nu},\ldots,S^{\beta}_{\alpha_1\ldots\alpha_r\mu\nu},\nabla^{A}R^{\beta}_{\alpha\mu\nu}, m^2,\ldots,\nabla^{A}m^2,\kappa,\ldots,\nabla^{A}\kappa\big).
\end{aligned}    
\end{equation}
For future convenience we also highlights other coordinates: 
\begin{equation}\label{eq_2_scaled_coordinates}
\begin{aligned}
    \big( & x^{\mu},g,g^{-\frac{1}{n}}g_{\mu\nu},g^{\frac{1}{n}+\frac{1}{n}|A| }\nabla^{A}g^{\mu\nu},g^{-\frac{1}{n}}m^2,\kappa,\ldots,g^{-\frac{1}{n}}\nabla_{A}m^2,\nabla_{A}\kappa\big),
\end{aligned}
\end{equation}
where $A$ is a multi-index, of the $n(n+1)/2$ coordinates describing the metric tensor $g_{\mu\nu}$ we consider only $n(n+1)/2-1$ and the remaining one is expressed in terms of the determinant of $g$. The nice feature of those coordinates is that, except $g$, all others are invariant under physical scaling. Finally, combining \eqref{eq_2_covariant_coordinates} and \eqref{eq_2_scaled_coordinates} we get a third set of coordinates
\begin{equation}\label{eq_2_scaled_covarinat_coordinates}
\begin{aligned}
    \big( & x^{\mu},g,g^{-\frac{1}{n}}g_{\mu\nu},S^{\beta}_{\mu\nu}, S^{\beta}_{\alpha\mu\nu},R^{\beta}_{\alpha\mu\nu}, \ldots,S^{\beta}_{\alpha_1\ldots\alpha_r\mu\nu},g^{\frac{|A|}{n}}\nabla^{A}R^{\beta}_{\alpha\mu\nu}, \\
    & g^{-\frac{1}{n}}m^2,\ldots,g^{\frac{1+|A|}{n}}\nabla^{A}m^2,\kappa,\ldots,g^{\frac{|A|}{n}}\nabla^{A}\kappa\big),
\end{aligned}
\end{equation}
We shall employ those coordinates at various stages of the proof.
 
\textbf{Step 1.} In the proof of Lemma \ref{lemma_2_Moretti_Kavhkine}, we found that the coefficients $C_k:\Gamma^{\infty}(M\leftarrow HM) \rightarrow C^{\infty}(M):h \mapsto C_k(h)$ are weakly regular (due to $(iv)$ in Definition \ref{def_2_Wick_quantum_powers}) and depends only on the germ of $h$ (due to $(i)$ in Definition \ref{def_2_Wick_quantum_powers}), therefore we can apply Peetre-Slov\'ak's Theorem \ref{thm_A_Peetre_Slovak} and conclude that $C_k$ is a differential operator of locally bounded order. This means that for any point $x\in M$ and any section $h_0\in \Gamma^{\infty}(M\leftarrow HM) $ there is a compact neighborhood $U$ and a $CO$-open neighborhood $\mathcal{U}=\{ j^rh(U)\subset V^r\}$ of $h_0$, such that for all $j^rh \in \mathcal{U}$ (or equivalently all $j^r_xh \in U\times V^r\subset J^rHM$),
$$
    C(h)(x) = F_k(j^r_xh).
$$
To start we shall choose any $x\in M$ and the section $h_0=(g_0,0,0)$ for a generic Lorentzian metric $g_0$. We denote by $U_0$, $V^r_0$ the corresponding open subsets.

\textbf{Step 2.} In this step we want to further characterize the functional form of $F_k$ and enlarge the set $V^r_0$ using scaling invariance. Notice that to the fibered isomorphism $S_{\lambda}:HM \to HM$ of physical scaling, we can associate a vector field 
$$
    \rho = \frac{d}{d\lambda}\Big|_{\lambda=1} S_{\lambda}
$$
which in the coordinates \eqref{eq_2_scaled_coordinates} reads $\rho=-2ng\frac{\partial}{\partial g}$. One can show (\textit{c.f.} \cite[Lemma 2.3]{khavkine2016analytic}) that the operator $(\rho -k\frac{n-2}{2})^l: C^{\infty}(J^rHM)\to C^{\infty}(J^rHM)$ annihilates functions scaling almost homogeneously with degree $k\frac{n-2}{2}$ and order $l-1$. According to Lemma \ref{lemma_2_Moretti_Kavhkine}, each $C_k$ scales almost homogeneously with degree $k\frac{n-2}{2}$ and finite order $l$; therefore by \cite[Lemma 2.4]{khavkine2016analytic}, we get that there exists functions $H_j:V^r_0\to \mathbb{R}$, $0\leq j\leq l$ such that  

$$
	F_k(j^r_xh)=g^{-\frac{k(n-2)}{4n}}\sum_{j=0}^l \ln^j\left(g^{-\frac{1}{2n}}\right) H_j(x^{\mu},g^{-\frac{1}{n}}g_{\mu\nu},g^{\frac{1}{n}+\frac{1}{n}|A| }\nabla^{A}g^{\mu\nu},g^{-\frac{1}{n}}m^2,\kappa,\ldots,g^{-\frac{1}{n}}\nabla^{A}m^2,\nabla^{A}\kappa).
$$
We can notice that the dependence on the metric determinant $g$ has been factored out of $F_k$ introducing ad hoc powers of $g$ and $\ln(g)$ terms. We are then able to extend the open subset $V^r_0$ to $V^r_1=\mathbb{R}_+\times W^r_1$, by simply declaring that if $\lambda g$ falls outside $V_0^r$, then 
$$
    F_k(\ldots,\lambda g,\ldots)=(\lambda g)^{-\frac{k(n-2)}{4n}}\sum_{j=0}^l \ln^j\left((\lambda g)^{-\frac{1}{2n}}\right) H_j(\ldots)
$$
using that, other then $g$, the coordinates in \eqref{eq_2_scaled_coordinates} are invariant under physical scaling.

\textbf{Step 3.} In this step we further extend to $(\pi^r)^{-1}(U)$ the domain of $F_k$ using diffeomorphism invariance. Since $\mathrm{Diff}(M)$ acts transitively on $M$, we can repeat the previous two steps for each $x'\in U$, by naturality of the bundle $HM$ (see $(ii)$ Definition \ref{def_0_natural_bundles}), we can lift all such local diffeomorphism to local isomorphism of $HM$, then we are able to enlarge $V^r_1$ to a neighborhood $V^r_2$ of the $\mathrm{Diff}(M)$ orbit of $h_0=(g_0,0,0)$. Notice that by covariance of $C_k$, $\psi^*C_k(h)=C_k(\psi^*h)$, which implies that the order of $C_k$ as a differential operator remains the initial $r$. In $\pi^{-1}_r(U)\cap  U \times V_2^r$, we can use Theorem \ref{thm_2_covariant_identity} to characterize the diffeomorphism invariant nature of $F_k$ in coordinates \eqref{eq_2_covariant_coordinates}. In particular we conclude that 
$$
    F_k(j^r_xh)=g^{-\frac{k(n-2)}{4n}}\sum_{j=0}^l \ln^j\left(g^{-\frac{1}{2n}}\right) H_j,
$$
where in coordinates \eqref{eq_2_scaled_covarinat_coordinates}, 
$$
    H_j=G_j\big(g^{-\frac{1}{n}}g_{\mu\nu},R^{\beta}_{\alpha\mu\nu}, \ldots,g^{-\frac{|A|}{n}}\nabla^{A}R^{\beta}_{\alpha\mu\nu}, m^2,\ldots,g^{\frac{2+|A|}{n}}\nabla^{A}m^2,\kappa,\ldots,g^{\frac{|A|}{n}}\nabla^{A}\kappa\big).
$$
Then we can extend $V_2^r=\mathbb{R}_+\times W_2^r$ to a set $V^r_3=\mathbb{R}_+\times \eta(1,n-1)(U) \times \mathbb{R}^{a}\times Q^r_3$, where $\eta(1,n-1)(U)$ is the space of non degenerate symmetric tensor in $\mathbb{R}^n$ with signature $(1,n-1)$, and the other pieces project as follows:
\begin{center}
\begin{tikzcd}
	 U \arrow[d,"\pi^r"] &   \eta(1,n-1)(U) \arrow[d] & \mathbb{R}^{a} \arrow[d] &  Q^r_3 \arrow[d]\\
	(x^{\mu}) & (g,g^{-\frac{1}{n}}g_{\mu\nu})   & (S^{\alpha}_{\mu\nu},S^{\alpha}_{\beta\mu\nu},\ldots) & (R^{\beta}_{\alpha\mu\nu},g^{\frac{1}{n}}m^2,\kappa,\ldots)
\end{tikzcd}
\end{center}
Then again, we stress that since the functions $G_j$ do not depend explicitely on the non covariant coordinates of \eqref{eq_2_scaled_covarinat_coordinates}, we are allowed to extend maximally the domain of such coordinates to $\mathbb{R}^a$.

\textbf{Step 4.} We further simplify the expression of $F_k$ using invariance for local diffeomorphism of the form $\psi_{s}:(x^{\mu})\mapsto (s x^{\mu}) $. By $(ii)$ in Definition \ref{def_0_natural_bundles}, the latter local diffeomorphism can be lifted to a local fibered morphism of $HM$, whose Jacobian matrix is $s\delta^{\mu}_{\nu}$, and the action on coordinates \eqref{eq_2_scaled_covarinat_coordinates} is given by
\begin{equation}\label{eq_2_coordinate_scaling_action}
\begin{aligned}
    &g^{-\frac{1}{n}} g_{\mu\nu}\to g^{-\frac{1}{n}} g_{\mu\nu},\ g\to s^{2n}g, \  g^{\frac{1}{n}} \nabla^A m\to s^{2+|A|} g^{\frac{1+|A|}{n}} \nabla^A m, \\
    &\nabla^A\kappa \to s^{|A|}\nabla^A \kappa, g^{\frac{|A|}{n}}\nabla^A R^{\alpha\beta}_{\ \ \mu\nu} \to s^{|A|}g^{\frac{|A|}{n}}\nabla^A R^{\alpha\beta}_{\ \ \mu\nu}. 
\end{aligned}
\end{equation}
For $s\in [0,1]$ the mapping $\psi_s$ does not leave $U$, thus it ought to define an covariant transformation for $F_k$, \textit{i.e.}
$$
\begin{aligned}
    F_k(g,g^{-\frac{1}{n}}g_{\mu\nu},y,z)&=g^{-\frac{k(n-2)}{4n}}\sum_{j=0}^l \ln^j\left(g^{-\frac{1}{2n}}\right) G_j(y,\psi_s^*z)\\
    &=s^{-\frac{k(n-2)}{2}}g^{-\frac{k(n-2)}{4n}}\sum_{j=0}^l \ln^j\left(s^{-1}g^{-\frac{1}{2n}}\right) G_j(g^{-\frac{1}{n}}g_{\mu\nu},y,\psi_s^*z)    
\end{aligned}
$$
where $(y,z)$ identify the coordinates in \eqref{eq_2_scaled_covarinat_coordinates}, other than the metric determinant $g$ and $g^{-\frac{1}{n}}g_{\mu\nu}$, which have the property $\psi_s^*y=y$, $\psi_s^*z=s^{p}z$ for some appropriate positive number $p\in \mathbb{Q}$; confronting with \eqref{eq_2_coordinate_scaling_action} we see that $y=\kappa,$ and $z$ are the remaining ones. Then
$$
    s^{\frac{k(n-2)}{2}}F_k(g,g_{\mu\nu},\kappa,z) =g^{-\frac{k(n-2)}{4n}}\sum_{j=0}^l \ln^j\left(s^{-1}g^{-\frac{1}{2n}}\right) G_j(g^{-\frac{1}{n}}g_{\mu\nu},\kappa,\psi_s^*z).
$$
However, as $s\to 0$, the right hand side of the above equation diverges unless $G_j=0$ whenever $j\neq 0$. Thus we find
\begin{equation}\label{eq_2_step_4_form_of_F}
    F_k(j^r_xh)=g^{-\frac{k(n-2)}{4n}}G_0(g^{-\frac{1}{n}}g_{\mu\nu},\kappa,\psi_s^*z).
\end{equation}
We can still squeeze some more information out of $G_0$. If we take its Taylor expansion on the coordinates $z$, in $z=0$, we find
$$
	 F_k(j^r_xh)=g^{-\frac{k(n-2)}{4n}}G_0(g^{-\frac{1}{n}}g_{\mu\nu},\kappa,z)= \sum_{|I|<N}g^{-\frac{k(n-2)}{4n}}G_I(g^{-\frac{1}{n}}g_{\mu\nu},\kappa,0)z^I+o(|z|^N),
$$
where $I\in \mathbb{N}^{3r-1}$ is a multi-index, $z^I$ a polynomial in the coordinates $z$ to the power $I$, and $G_I=\frac{\partial G}{\partial z^I}$. Again, applying to the above equality the coordinate scaling \eqref{eq_2_coordinate_scaling_action}, and taking the limit as $s\to 0$, we can see that the right hand side diverges, or goes to 0, unless the cumulative scaling power $s^{\sum_{i\in I} p_i}$ of $z^I$ balances exactly that of $g^{-\frac{k(n-2)}{4n}}$, \textit{i.e.} $s^{k\frac{n-2}{2}}$. However, this is extremely consequential, since there are only finitely many $z^I$ with this property. As a consequence,
\begin{equation}\label{eq_2_step_4/5_form_of_F}
    F_k(j^r_xh)=\sum_{I}g^{-\frac{k(n-2)}{4n}}G_I(g^{-\frac{1}{n}}g_{\mu\nu},\kappa)z^I,
\end{equation}
We are now in a position to further enlarge the domain $ V^r_3$ to $ V^r_4=\mathbb{R} \times \eta(1,n-1)(U) \times \mathbb{R}^{a}\times P^r_4 \times \mathbb{R}^b$, where each component projects as
\begin{center}
\begin{tikzcd}
    \eta(1,n-1)(U) \arrow[d] & \mathbb{R}^{a} \arrow[d] &  P^r_4 \arrow[d] & \mathbb{R}^b \arrow[d] \\
	 (g,g^{-\frac{1}{n}}g_{\mu\nu})   & (S^{\alpha}_{\mu\nu},S^{\alpha}_{\beta\mu\nu},\ldots)& (\kappa) & (z) 
\end{tikzcd}
\end{center}

\textbf{Step 5.} In this part we will furthermore characterize the dependence on $g^{-\frac{1}{n}}g_{\mu\nu}$ for the functions $G_I$. \\

Note that being $HM$ a natural bundle we have a canonical action of $Gl^r(n,\mathbb{R})$ on $HM$, which however, in the coordinates on which \eqref{eq_2_step_4/5_form_of_F} depends, reduces to the one of $Gl(n,\mathbb{R})$. Covariance of $F_k$ then prescribes once again that the above action leaves unaltered $F_k$ itself, therefore for all $A\in Gl(n,\mathbb{R})$,
$$
	g^{-\frac{k(n-2)}{4n}} \sum_{I}G_I(g^{-\frac{1}{n}}g_{\mu\nu},\kappa)z^I=A^{\frac{k(n-2)}{2n}}g^{-\frac{k(n-2)}{4n}} \sum_{I}G_I(A^{\frac{2}{n}}g^{-\frac{1}{n}}A^{\alpha}_{\mu}g_{\alpha\beta}A^{\beta}_{\nu},\kappa)(\rho_A(z))^I.
$$
where $\rho_A$ is the $Gl(n,\mathbb{R})$ action\footnote{For example for the Riemann tensor $R^{\alpha\beta}_{\ \ \mu\nu}$, the action by the matrix $A$ is given by $A^{\alpha}_{\alpha'}A^{\beta}_{\beta'}\bar A^{\mu'}_{\mu}\bar A^{\nu'}_{\nu} R^{\alpha'\beta'}_{\quad \  \mu'\nu'}$, where $\bar A $ denotes the inverse matrix of $A$.} on the polynomials $z^I$. If we fix $\kappa$ and think of $G_I(\cdot,\kappa)$ as a mapping $:\eta(1,n-1)(U)\rightarrow \mathbb{R}[z]$, then $G_I$ must be equivariant with respect to the $Gl(n,\mathbb{R})$ action induced on $\eta(1,n-1)(U)$ and $ \mathbb{R}[z]$. By \cite[Lemma 2.8]{khavkine2016analytic}, we deduce that the coefficients $G_I$ must be a polynomial involving the components of the metric tensor $g_{\mu\nu}$, the Levi-Civita tensor $\epsilon_{\mu_1\cdots \mu_n}$ up to an overall multiple of the density $g=\mathrm{det}(g_{\mu\nu})$ which multiplied with the corresponding $z^I$ produces a scalar. Thus
\begin{equation}\label{eq_2_step_5_form_of_F}
     F_k(j^r_xh)=\sum_{I}g^{-\frac{k(n-2)}{4n}}\widetilde G_I(\kappa)P_I\big(g_{\mu\nu},\epsilon_{\mu_1\cdots \mu_n},\nabla^{A}R^{\beta}_{\alpha\mu\nu},\nabla^{A}m^2,\nabla^{A}\kappa\big)
\end{equation}
globally defined on $U \times V^r_5$ with $V^r_5=\eta(1,n-1)(U) \times \mathbb{R}^{a+b}\times P^r_4 $, with the $P^r_4$ an open subset of $0$ in the coordinate $\kappa$. 
% Note also that if $A\in O(1,n-1)$, then 
% $$
% 	A_I(-\rho_A(\eta_{\mu\nu}),\kappa)=A_I(-\eta_{\mu\nu},\kappa),
% $$therefore $A_I(-\eta_{\mu\nu},\kappa)$ must be of the form $$a_I(\kappa)P_I(\eta_{\mu\nu})\sum_{\sigma\in \mathcal{P}(n)}a_I(\kappa)c(\sigma)\eta_{\sigma(\mu_1)\sigma(\mu_2)}\cdots\eta_{\sigma(\mu_{n-1})\sigma(\mu_n)},$$as a result we have found that 

% $$
% 	F_k(j^r_xh)=\sum_{I}a_I(\kappa)P_I(g^{-\frac{1}{n}}g_{\mu\nu},g)Q_I(R^{\alpha\beta\mu\nu}, m^2,\ldots, \nabla^{\lambda_1\ldots\lambda_r}\kappa)
% $$for some polynomials $P_I$, $Q_I$ in the respective variables and finite $I$. 

\textbf{Step 6.} To conclude, we will extend the domain of $F_k$ on the whole $HM$. First notice that we can further extend the function $F_k$ in $\kappa$ by repeating Steps 1)-5) for neighborhoods of sufficiently many $(g_0,0,\kappa_0)$ and the form \eqref{eq_2_step_5_form_of_F} will not be affected since the coordinate $\kappa$ transforms trivially under diffeomorphisms and does not scale under physical scaling. Then we have extended $F_k$ to the whole $\pi^{-1}_r(U)$. By transitivity of the action of $\mathrm{Diff}(M)$ on $M$, we can extend this functional form of $F_k$ to the whole $HM$, in each chart $U_{\alpha}$ of $M$ we will have produced a local $(F_k)_{\alpha}:J^{r_{\alpha}}\to \mathbb{R}$, whose order $r_{\alpha}$ cannot be unbounded for if so, there will be dependence on a derivative of the Riemann tensor $\nabla^AR$, the mass $\nabla^Bm$ or the coupling constant $\nabla^C\kappa$, for which the physical scaling by $S_{\lambda}$ exceeds the amount $k\frac{n-2}{2}$ globally fixed by Lemma \ref{lemma_2_Moretti_Kavhkine}. Therefore, the bound on the jet order $r$ enable us to safely enlarge the domain of $F_k$ to $J^{r_0}HM$ for this maximal order $r_0$.
\end{proof}

\subsection{Existence of Wick powers}\label{section_Wick_existence}
 
In trying to show the existence of Wick powers, functional formalism provides us with the natural candidates for Wick powers: namely the elements $\phi^k_{(M,h)} : f \mapsto \phi^k_{(M,h)}(f)$ which evaluated at $H\in \mathrm{Had}(M,h)$ are defined by \eqref{eq_2_Wick_powers}. We will thereafter prove that they satisfy all conditions given in Definition \ref{def_2_Wick_quantum_powers}. Prior to this, we need some technical results which come from Cartesian closedness (\textit{c.f.} Theorem \ref{thm_A_cart_closedness}).

\begin{proposition}\label{prop_2_equivalence_conv-smooth_w-regularity}
    Let $\pi:B\rightarrow M$ be a fiber bundle over $M$. Then any mapping $D:\Gamma^{\infty}(M \leftarrow B)\rightarrow C^{\infty}(M):\sigma \mapsto D(\sigma)$, where $D(\sigma)(x)$ depends only on the germ of $\sigma$ at $x$, is weakly regular if and only if it is conveniently smooth.
\end{proposition}
%----------------- better show w-reg -> bastiani smooth -> conv. smooth; w-reg <- conv. smooth --------------------
\begin{proof}
We remark that $\Gamma^{\infty}(M \leftarrow B)$ is a smooth infinite dimensional manifold with the smooth structure described in  \cite[Theorem 42.1]{convenient} which has the same charts as the one described in \Cref{thm_1_mfd_mappings}, but a different topology. In particular in a neighborhood of any section $\sigma \in \Gamma^{\infty}(M \leftarrow B)$, there are charts $\{\mathcal{U}_{\sigma},u_{\sigma}\}$, with $u_{\sigma}:\mathcal{U}_{\sigma} \to \Gamma^{\infty}_c(M\leftarrow \sigma^{*}VB)$ where the target space has the canonical limit Fréchet vector space topology. Meanwhile recall that $C^{\infty}(M)\equiv \mathcal{E}(M)$ has the Fréchet space topology described in \Cref{ex_A_smooth_function_frechet_space}. As a result, the smoothness requirement can be tested in a generic chart $\mathcal{U}_{\sigma}$.

Suppose that $D$ is weakly regular, then $D$ is conveniently smooth if and only if for any $\sigma_0 \in \Gamma^{\infty}(M \leftarrow B)$, $\sigma\in \mathcal{U}_{\sigma}$,
$$
	D_{\sigma}=D \circ u_{\sigma_1}^{-1} : u_{\sigma}(\mathcal{U})\subset \Gamma^{\infty}_c(M\leftarrow \sigma^{*}VB) \rightarrow C^{\infty}(M,)
$$
is conveniently smooth. Equivalently, see Definition \ref{def_A_convenient_smooth_map}, given any smooth curve $\gamma: \mathbb{R} \rightarrow \Gamma^{\infty}_c(M\leftarrow \sigma^{*}VB)$ with $\gamma(0)=\sigma$, we have to show that
$$
	D_{\sigma}\circ {\gamma}:\mathbb{R}\ni t \mapsto D_{\sigma}(\gamma(t)) \in C^{\infty}(M)
$$
is smooth. By $(ii)$ Proposition \ref{prop_A_conv_sections_of_vector_bndl}, smoothness of $\gamma$ implies that for any compact interval $[a,b]\in \mathbb{R}$ there is a compact subset $K\in M$ such that $\gamma(t)$ is constant outside $K$ for each $t\in [a,b]$. Then since $D$ is weakly regular, there is a compact subset $K'\subset M$ outside of which $D_{\sigma}(\gamma(t))$ is constant for each $t\in[a,b]$ and jointly smooth in $(t,x)\in \mathbb R \times M$. Then by $(i)$ Proposition \ref{prop_A_conv_sections_of_vector_bndl} we get that $D_{\sigma}\circ {\gamma}$ is smooth as well. We remark that in this argument the locality of the operator is not required.

On the other hand, suppose that $D$ is conveniently smooth, consider any $d$-parametric compactly supported variation $\Sigma$ of $\sigma$, \textit{i.e.} a mapping $\Sigma: \mathbb{R}^d \times M \rightarrow B$ such that there exists a compact region $K\subset M$ for which $\Sigma(s,x)=\sigma(x)$ for all $(s,x) \in \mathbb{R}^d\times (M\backslash K)$. Without loss of generality we can assume that $\Sigma(s) \in \mathcal{U}_{\sigma}$ for each $s=(s^1,\ldots,s^d)$.

To show weak regularity of $D$ we have to prove that $D(\Sigma(s))\in C^{\infty}(\mathbb{R}^d\times M)$ and that $D(\Phi(s))(x)= D(\sigma)(x)$ whenever $(s,x)\in \mathbb{R}^d \times (M \backslash K')$ for some compact subset $K'$ of $M$. The second assertion follows immediately using the fact that $D$ evaluated at a point just depends on germs of the source section evaluated at that point, therefore we even get $D(\Sigma(s))(x)= D(\sigma)(x)$ outside the original compact subset $K$. For the other assertion, let $s_0\in \mathbb{R}^d$, we show joint smoothness in a neighborhood $I$ of $s_0$. Assume that the dimension of parameter space is $d=1$, then the mapping $ I \ni s \rightarrow u_{\sigma}(\Sigma(s))$ is a smooth curve in $\Gamma^{\infty}_c(M\leftarrow \sigma^{*}VB) $ provided $I$ is a small enough neighborhood of $s_0\in \mathbb{R}$, thus by \Cref{thm_A_cart_closedness}, Proposition \ref{prop_A_conv_structures_on_C_infty(M)} we have
$$
\begin{aligned}
	 &D(\Sigma): I_0 \rightarrow C^{\infty}(M,\mathbb{R}) \space\ \mathrm{is} \space\ \mathrm{smooth}\\
	\Leftrightarrow & D(\Sigma) : I_0 \rightarrow C^{\infty}(v_{\alpha}(V_{\alpha}),\mathbb{R})  \space\ \mathrm{is} \space\ \mathrm{smooth} \space\ \forall \ \{V_{\alpha},v_{\alpha}\} \space\ \mathrm{chart} \space\ \mathrm{of} \space\ M ,\\
	\Leftrightarrow  & D(\Sigma)^{\wedge} : v_{\alpha}(V_{\alpha})\times I_0 \rightarrow \mathbb{R} \space\ \mathrm{is} \space\ \mathrm{smooth} \space\ \forall \ \{V_{\alpha},v_{\alpha}\} \space\ \mathrm{chart} \space\ \mathrm{of} \space\ M.\\
	\Leftrightarrow  & D(\Sigma)^{\wedge} : M\times I_0 \rightarrow \mathbb{R} \space\ \mathrm{is} \space\ \mathrm{smooth}.
\end{aligned}
$$
By arbitrariness of $s_0$ we conclude that $D(\Sigma) : M\times \mathbb{R} \rightarrow \mathbb{R}$ is smooth as well. In case $d>1$ it suffices to iterate $d$ times this argument: at the $l$th iteration we fix the previous $d-l-1$ parameters and take a smooth curve $D(\Sigma): \mathbb{R}\to C^{\infty}(M\times\mathbb{R}^{l-1}) $ with domain the $1$th factor of $\mathbb{R}^d$ while keeping the remaining $d-l-1$ constant. Finally, by the above argument we infer that $D(\Sigma)^{\wedge}\in C^{\infty}(M\times \mathbb{R}^{l})$.
\end{proof}

We remark that the key feature of this proof is the Cartesian closedness of convenient calculus \textit{i.e.} Theorem 3.12 of \cite{convenient}, which essentially enable us to "move back and forth" the $M$ factor in the smoothness of mappings. This property is in general true for the so-called convenient calculus on vector spaces; however, in general, the $c^\infty$-topology is not compatible with the vector space structure of the underlying space being often too fine (see \cite[Corollary 4.6]{convenient}). Our case is special though: the $c^\infty$-topology coincides with the Fréchet vector space topology. We also stress that part of this result is very similar to Lemma \ref{lemma_A_locality&bastiani_implies_w-reg}, where we showed that Bastiani smoothness of the differential operator implies weak regularity; however, the opposite assertion is not generally valid due to the lack of cartesian closedness for Bastiani calculus.\\ 

We further remark that Proposition \ref{prop_2_equivalence_conv-smooth_w-regularity} can be generalized to vector bundles, using Proposition 30.1, Lemma 30.3, and Lemma 30.8 of \cite{convenient}.\\

%--------------- generalization to smooth bundles ------------------

%Therefore to generalize further condition $(iv)$ in Definition \ref{def_2_Wick_quantum_powers} one has to choose one of the two. Let then $H$ be a Hadamard state, then for any 

% Next we show existence of Wick powers satisfying the conditions above. The only non immediate property we want to show is the parameterized microlocal spectrum condition. First however, let us state a technical lemma:
Let us use two other auxiliary results. To better introduce them, we shall use the following notation: if $s\mapsto h_s=(g_s,m_s,\kappa_s)$ is the compactly supported variation of $h\equiv h_0=(g_0,m_0,\kappa_0)$ we denote the corresponding variation of the Klein-Gordon differential operator $P$ by 
\begin{equation}\label{eq_2_P_s}
    P_s=\sqrt{\frac{g_s}{g_0}}\big(g_s^{\mu\nu}\nabla_{\mu}(s)\nabla_{\nu}(s)+m^2_s+\kappa_sR(g_s)\big).
\end{equation}
Notice that $P_s$ differs from the normally hyperbolic operator $ P(h_s)=g_s^{\mu\nu}\nabla_{\mu}(s)\nabla_{\nu}(s)+m^2_s+\kappa_sR(g_s)$ by a factor $\sqrt{\frac{g_s}{g_0}}$, this extra factor will come in handy in the calculations below. Then we call $\mathcal{G}^{\pm}(M,h_s)$ the advanced retarded Green operators associated to $P(h_s)$, and set 
\begin{equation}\label{eq_2_G_s^+-}
    \mathcal{G}_s^{\pm}(x,y)= \sqrt{\frac{g_0(x)}{g_s(x)}}\sqrt{\frac{g_0(y)}{g_s(y)}}\mathcal{G}^{\pm}_{(M,h_s)}(x,y),
\end{equation}
by \cite[Theorem 3.8]{bar}, we see that $\mathcal{G}_s^{\pm}:C^{\infty}_c(M)\to C^{\infty}(M,\mathbb R)$ are Green-hyperbolic operator and, given any $f \in C^{\infty}_c(M)$,
$$
\begin{aligned}
    \int_M \mathcal{G}_s^{\pm}(x,y)P_sf(y)d\mu_{g_0}(y)=f(x),\\
    \int_M \mathcal{G}_s^{\pm}(x,y)P_sf(x)d\mu_{g_0}(x)=f(y).
\end{aligned}
$$
The reason for this notation is essentially to simplify subsequent formulas which otherwise will have factors $\sqrt{\frac{g_s}{g_0}}$. Notice that $P_0=P(h_0)$, $\mathcal{G}_0^{\pm}=\mathcal{G}^{\pm}_{(M,h_0)}$. Notice also that since $P_s\neq P_0$ only in a compact subset of $M$, then 
\begin{equation}\label{eq_2_id_G_s-G_0}
    \mathcal{G}_s^{\pm} -\mathcal{G}_0^{\pm} =- \mathcal{G}_s^{\pm} (P_s-P_0)\mathcal{G}_0^{\pm}  
\end{equation}

\begin{lemma}\label{lemma_2_propagator_smothness}
	Let $h_0$ be a background geometry and consider a smooth curve $h_s$, $ \mathbb{R}\ni s \mapsto h_s \in \Gamma^{\infty}(M\leftarrow HM)$. Then the associated space-time propagators $\mathbb{R}\ni s\mapsto\mathcal{G}^{\pm}_{(M,h_s)}(f) \in C^{\infty}(M)$ are smooth for each $f\in \mathcal{D}(M)$. 
\end{lemma}
\begin{proof}
	We just show the claim for the retarded propagator but for the advanced propagator we can repeat the same argument. We start by evaluating
$$
	\lim_{s\to 0} \frac{1}{s}\Big(\mathcal{G}_s^+(f)-\mathcal{G}_0^+(f) \Big),
$$
for some $f \in C^{\infty}_c(M)$. By construction the differential operator $P_s(x)=\sqrt{\frac{g_s}{g_0}}g_s^{\mu\nu}\nabla_{\mu}(s)\nabla_{\nu}(s)+m^2(s)+\kappa(s)R(g(s))$ is smooth in $s$, therefore we consider 
$$
\begin{aligned}
	&P_s(x) \lim_{s\to 0} \frac{1}{s}\Big(\mathcal{G}_s^+(x,f)-\mathcal{G}_0^+(x,f) \Big)=\lim_{s\to 0} \frac{1}{s}\Big(P_s(x)\mathcal{G}_s^+(x,f)-P_s(x)\mathcal{G}_0^+(x,f) \Big)\\
	& \quad = \lim_{s\to 0} \frac{1}{s}\Big(\delta(x,f)-P_s(x)\mathcal{G}_0^+(x,f) \Big)=\lim_{s\to 0} \frac{1}{s}\Big(P_0(x)\mathcal{G}_0^+(x,f)-P_s(x)\mathcal{G}_0^+(x,f) \Big)\\
	&\quad =\lim_{s\to 0} \frac{1}{s}\Big(P_0(x)-P_s(x)\Big)\mathcal{G}_0^+(x,y) =- \dot{P}_s(x)\mathcal{G}_0^+(x,f) 
\end{aligned}
$$
since $h_s\neq h_0$ in just a compact subset of $M$, the differential operator $\dot{P}_s(x)\equiv \frac{d}{ds}P_s(x)\neq 0$ only inside that compact, therefore the quantity $ \dot{P}_s(x)\mathcal{G}_0^+(x,f) $ has compact support  for any $f\in \mathcal{D}(M)$ and we can write
\begin{equation}\label{eq_derivative_G+}
\begin{aligned}
    \frac{d}{ds}\Big\vert_{s=0} \mathcal{G}_s(x,f)&=\lim_{s\to 0} \frac{1}{s}\Big(\mathcal{G}_s^+(x,f)-\mathcal{G}_0^+(x,f) \Big)\\
    &=- \int_M \mathcal{G}_0^+(x,z) \dot{P}_s(z)\mathcal{G}_0^+(z,f)d\mu_{g_0}(z). 
\end{aligned}
\end{equation}
Using \eqref{eq_derivative_G+} one can show that all iterated derivatives of $\mathcal{G}_0^+$ exists and are continuous by \Cref{thm_1_properties_of_Green_functions}, thus showing smoothness. Finally, since $\mathcal{G}_s(x,f) =\mathcal{G}_s^+(x,f) -\mathcal{G}_s^-(x,f) $, the causal propagator is smooth as well.
\end{proof}

\begin{proposition}\label{prop_2_properties_of_R}
    Using the notation introduced in \eqref{eq_2_P_s}, \eqref{eq_2_G_s^+-} define 
    \begin{equation}\label{eq_2_R_+-}
        R_{\pm} :C^{\infty}(M)\to C^{\infty}(M) : \psi \mapsto \psi - \mathcal{G}_s^{\pm} (P_s-P_0) \psi, 
    \end{equation}
    then the operators $R_{\pm}$ are continuous and satisfy:
    \begin{itemize}
        \item[$(i)$] $R_{\pm}$ is an isomorphism and the inverse is given by 
        \begin{equation}\label{eq_2_R_+-_inverse}
            R_{\pm}^{-1}= id_{C^{\infty}(M)}+\mathcal{G}_0^{\pm} (P_s-P_0); 
        \end{equation}
        \item[$(ii)$] $P_sR_{\pm}=P_0$.
    \end{itemize}
Moreover, if we denote by $R_{\pm}^{\dagger}:C_c^{\infty}(M)\to C^{\infty}_c(M)$ the adjoint of $R_{\pm}$ with respect to the paring
$$
    C^{\infty}(M)\otimes C^{\infty}_c(M) \to \mathbb R, \ (\psi,f)\to \int_M f\psi d\mu_{g_0},
$$
we have
    \begin{itemize}
        \item[$(iii)$] $R_{\pm}^{\dagger}= id_{C_c^{\infty}(M)}- (P_s-P_0)\mathcal{G}_s^{\pm}$, and $ (R^{\dagger}_{\pm})^{-1}=( R^{-1}_{\pm})^{\dagger}$;
        \item[$(iv)$] If $\mathcal{G}_s$, $\mathcal{G}_0$ are the causal propagator associated to $\mathcal{G}_s^{\pm}$, $\mathcal{G}_0^{\pm}$, then
        $$
            R_{\pm} \mathcal{G}_0 R^{\dagger}_{\pm}= \mathcal{G}_s;
        $$
        
    \end{itemize}
\end{proposition}
\begin{proof}
If $C^{\infty}(M)$ has the Fréchet space topology, continuity of $R_{\pm}$ follows directly from \Cref{thm_1_properties_of_Green_functions} and the continuity of the differential operator $P_s$. Let then $\psi\in C^{\infty}(M)$, $f\in C_c^{\infty}(M)$ be arbitrary functions. Using \eqref{eq_2_id_G_s-G_0}, 
$$
\begin{aligned}
    R_{\pm}R_{\pm}^{-1}(\psi)&= R_{\pm}(\psi- \mathcal{G}_0^{\pm}(P_s-P_0)\psi)\\
    &= \psi-  \mathcal{G}_s^{\pm}(P_s-P_0)\psi +\mathcal{G}_0^{\pm}(P_s-P_0)\psi- \mathcal{G}_s^{\pm} \big( (P_s-P_0)\mathcal{G}_0^{\pm}(P_s-P_0)\big)\psi\\
    &= \psi- \textcolor{blue}{\mathcal{G}_s^{\pm}(P_s-P_0)\psi} +\textcolor{red}{\mathcal{G}_0^{\pm}(P_s-P_0)\psi}+  \textcolor{blue}{\mathcal{G}_s^{\pm}(P_s-P_0)\psi}- \textcolor{red}{\mathcal{G}_0^{\pm}(P_s-P_0)\psi}\\
    &=\psi.
\end{aligned}
$$
A similar argument works for the inverse composition, thus we conclude that \eqref{eq_2_R_+-_inverse} in $(i)$ is actually the inverse of $R_{\pm}$. Directly from \eqref{eq_2_R_+-},
$$
    P_s R_{\pm} \psi = P_s \psi - P_s\mathcal{G}_s^{\pm}(P_s-P_0)\psi= P_s\psi-P_s\psi+P_0\psi=P_0\psi.
$$
to show $(iii)$ we have to verify that 
$$
    \int_{M^2} d\mu_{g_0}(x,y) R^{\dagger}_{\pm}(x,y)f(x)\psi(y)=\int_{M^2} d\mu_{g_0}(x,y) R_{\pm}(x,y)f(y)\psi(x).
$$
expanding the right hand side and using that $\mathcal{G}_{\cdot}^{\mp}$ is the adjoint of $\mathcal{G}_{\cdot}^{\pm}$ (see Lemma \ref{lemma_1_duality_green}) yields
$$
\begin{aligned}
    \int_{M^2} d\mu_{g_0}(x,y) R_{\pm}(x,y)&f(y)\psi(x)= \int_{M^2} d\mu_{g_0}(x,y)f(y) \big[ \delta(x,y) - \mathcal{G}_s^{\pm}(x,y)(P_s-P_0)\big]\psi(x)\\
    &= \int_{M^2} d\mu_{g_0}(y)f(y)\psi(y) - \int_{M^2} d\mu_{g_0}(x,y)f(y)\big[\mathcal{G}_s^{\pm}(x,y)(P_s-P_0)\psi(x)\big]\\
    &= \int_{M^2} d\mu_{g_0}(y)f(y)\psi(y) - \int_{M^2} d\mu_{g_0}(x,y)\mathcal{G}_s^{\mp}(y,x)f(y)(P_s-P_0)\psi(x)\\
    &= \int_{M^2} d\mu_{g_0}(y)f(y)\psi(y) - \int_{M^2} d\mu_{g_0}(x,y)\big[(P_s-P_0)\mathcal{G}_s^{\mp}(y,x)f(y)\big]\psi(x)\\
    &=  \int_{M^2} d\mu_{g_0}(x,y)\big[\delta(x,y)- (P_s-P_0)\mathcal{G}_s^{\mp}(y,x)f(y)\big]\psi(x)\\
\end{aligned}
$$
Incidentally, using the same argument on $R_{\pm}^{-1}$ shows the other part of $(iii)$. Finally, using repeatedly \eqref{eq_2_id_G_s-G_0}, 
$$
\begin{aligned}
    R_{+}\mathcal{G}_0R^{\dagger}_+ &= R_+ \big( \mathcal{G}_0 -  \mathcal{G}_0(P_s-P_0) \mathcal{G}_s^-\big)=R_+ \big( \mathcal{G}_0 -  \mathcal{G}_0^+(P_s-P_0) \mathcal{G}_s^-+\mathcal{G}_0^--\mathcal{G}_s^-\big)\\
    &= R_+ \big( \mathcal{G}_0^+ -\mathcal{G}_s^- -  \mathcal{G}_0^+(P_s-P_0) \mathcal{G}_s^-\big)   = \mathcal{G}_0^+ -\mathcal{G}_s^- -  \mathcal{G}_0^+(P_s-P_0) \mathcal{G}_s^- \\
    & \quad - \big[ \mathcal{G}_s^+(P_s-P_0)\mathcal{G}_0^+-\mathcal{G}_s^+(P_s-P_0)\mathcal{G}_s^--\mathcal{G}_s^+(P_s-P_0)\mathcal{G}_0^+(P_s-P_0)\mathcal{G}_s^-\big]\\
    &= \textcolor{red}{\mathcal{G}_0^+} -\mathcal{G}_s^- -  \textcolor{green}{\mathcal{G}_0^+(P_s-P_0) \mathcal{G}_s^-} -\textcolor{red}{\mathcal{G}_0^+}+ \mathcal{G}_s^+ +\textcolor{blue}{\mathcal{G}_s^+(P_s-P_0)\mathcal{G}_s^-} + \textcolor{green}{\mathcal{G}_0^+(P_s-P_0) \mathcal{G}_s^-}\\ &- \textcolor{blue}{\mathcal{G}_s^+(P_s-P_0)\mathcal{G}_s^-}=\mathcal{G}_s.
\end{aligned}
$$
\end{proof}

\begin{theorem}\label{thm_2_existence_Wick_polynomials}
    Let $(M,h)$ be a background geometry and $H\in \mathrm{Had}(M,h)$ be any Hadamard state associated to the data $h=(g,m,\kappa)$, then the quantum field $\phi^k_{(M,h,H)}(f)$ is the $k$th order Wick power according to Definition \ref{def_2_Wick_quantum_powers}.
\end{theorem}

\begin{proof}
We dived the proof in three steps: in the first we verify that $\phi^k_{(M,h,H)}$ are linear scalar quantum field when $k=1$ and satisfy property $(i),\ (ii), \ (iii)$ in Definition \ref{def_2_Wick_quantum_powers}; in the second we shall define an isomorphism $\tau_{(M,h_s\to h_0)} : \mathfrak A_{\mu loc}(M,h_s) \to  \mathfrak A_{\mu loc}(M,h_0)$ using the result of Proposition \ref{prop_2_properties_of_R}; whereas in the final part we will show that condition $(iv)$ in Definition \ref{def_2_Wick_quantum_powers} holds. 

\textbf{Step 1.}
Firstly, we notice that each $\phi^k_{(M,h,H)}: f \mapsto \phi^k_{(M,h,H)}(f)$ is an algebra valued distribution since it is linear in $f$ and continuous with respect to the topology of $\mathfrak A_{\mu c }(M,h)$ with the completion of the topology of $\hbar$ convergence in $\mathfrak A_{reg }(M,h)$. Moreover, $\phi$ is a linear scalar quantum field if and only if $d\phi_{(M,h,H)}$ is a $c$-number field. By direct inspection
$$
	d\phi_{(M,h,H)}(f)=\bigg(\int_Mf(x) d\mu_g(x)\bigg)1_{\mathfrak{A}_{\mu c}(M,h)}
$$
clearly yields a $c$-number, therefore it is a linear quantum field as claimed, thus we can proceed to check whether elements of the form \eqref{eq_2_Wick_powers} are effectively Wick powers. The locality and covariance property for Wick powers are equivalent to require that given any $\chi \in \mathrm{Hom}(\mathfrak{BckG})$ there is an algebraic morphism
$$
	\mathfrak{A}(\chi) \big[\phi_{(M,h,\chi^*H')}^k(f))=\phi^k_{(M',h',H')}(\chi_{*}f).
$$
To fulfill this condition we follow \cite[$\S$6]{BFV03} and substitute $H(x,y)-W_0(x,y)$ in place of the Hadamard parametrix $H(x,y)$ that we use to construct Wick polynomials. Here, $W_0(x,y)$ is the non geometrical, smooth term in the Hadamard expansion \eqref{eq_2_Hadamard_distribution}. Notice that, although $W_0$ is well defined only on a open neighborhood of the diagonal in $M^2$, the coincidence limit of $W_0$ does still give rise to a smooth function $M\ni x\mapsto W_0(x,x)$ by a partition of unity argument. the covariant version of Wick powers is then
\begin{equation}\label{eq_2_covariant_Wick_powers}
	\phi^k_{(M,h,H)}\doteq \alpha^{-1}_{H-W_0}(\phi^k_{(M,h)}).
\end{equation}
By Lemma \ref{lemma_2_regular_algebra_isomorphism}, this is still in the same algebra as the original power $\phi^k_{(M,h,H)}$. In the sequel we will abuse the notation and write directly $\alpha_H$ instead of $\alpha_{H-W_0}$. Explicit calculation of $\alpha^{-1}_{H}\big(\phi^k_{(M,h)}(f)\big)$ yields the abstract algebra element
$$
	\int_Mf(x) :\phi^k:_{H}(x) d\mu_g(x)
$$
where, as usual, $:\phi:_H(x)=\phi(x)$ and 
\begin{equation}\label{eq_Wick_power_expansion}
	:\phi^k:_{H}(x)=\delta(x,\ldots,x_k) \bigg( :\phi(x_1)\cdots \phi(x_{k-1}):_H\phi(x_k)-\sum_{i=1}^{k-1} :\phi(x_1)\cdots \widehat{\phi(x_i)}\cdots  \phi(x_{k-1}):_H \bigg).
\end{equation}
By the scaling behavior of each field $\varphi$ given by \eqref{eq_scaling_varphi}, using \eqref{eq_2_Hadamard_distribution}, we see that
$$
	H_{\lambda}(x,y)=\lambda^{n-2}H(x,y)+2\lambda^{n-2}\ln(\lambda)V(x,y)
$$
which implies that $\phi^k_{(M,h,H)}(f)$ scales almost homogeneously with degree $k\frac{n-2}{2}$ under physical scaling for any $f\in C^{\infty}_c(M)$.

The algebraic condition follows immediately by taking into account that 
$$
    \phi^k_{(M,h,H)}(f)^{*}=\alpha^{-1}_H\big(\phi^k_{(M,h)}(f)^{*}\big)=\alpha^{-1}_H\big(\phi^k_{(M,h)}(\bar{f})\big)
$$
and $d\phi^k_{(M,h,H)}(f)=\alpha^{-1}_H\big(d\phi^k_{(M,h)}(f)\big)$. 

\textbf{Step 2.} Suppose $\mathbb{R} \ni s \rightarrow h_s\in \Gamma^{\infty}(M\leftarrow HM)$ defines a compactly supported variation of the background geometry $(M,h)$. To each $h_s$ we can associate its differential operator $P_s= g_s^{\mu\nu}\nabla_{\mu}(s)\nabla_{\nu}(s)+m^2(s)+\kappa(s)R(g(s))$, which is still a Green hyperbolic operator according to the hypothesis and is moreover conveniently smooth as a mapping $\mathbb{R}\ni s\mapsto P_s \in C^{\infty}\big(C^{\infty}(M,\mathbb R),C^{\infty}(M,\mathbb R)\big)$. Moreover, to each $P_s$ we associate its Green operators $\mathcal{G}_s^{\pm}$, which thanks to Lemma \ref{lemma_2_propagator_smothness} are conveniently smooth as well. Given the background geometry $(M,h_s)$, we can associate a Hadamard two point function
$$
    \omega_{s}(x,y)=H_s(x,y)+ \frac{i}{2}\mathcal{G}_{(M,h_s)} 
$$
where $H_s$, $\mathcal{G}_{(M,h_s)} $ are respectively a Hadamard parametrix and the causal propagator canonically associated to the geometry $h_s=(g_s,m_s,\kappa_s)$. Setting 
\begin{equation}\label{eq_2_Hadamard_state_R^2}
    \widetilde {\omega}= \omega_s \circ\big(\big(R^{\dagger}_+\big)^{-1}\otimes\big(R^{\dagger}_+\big)^{-1}\big),
\end{equation}
we claim that \eqref{eq_2_Hadamard_state_R^2} is a Hadamard state for the background geometry $(M,h)$. From $(iv)$ in Proposition \ref{prop_2_properties_of_R}, setting $\mathcal{R}^{\dagger}_+=\sqrt{{g_0}/{g_s}}R^{\dagger}_+$, we see that 
$$
\mathcal{G}_{(M,h_0)}=\mathcal{G}_{(M,h_s)}\circ\Big((\mathcal{R}^{\dagger}_+)^{-1}\otimes(\mathcal{R}^{\dagger}_+)^{-1}\Big).
$$
We therefore have just to verify that $\widetilde H\equiv H_s\circ\Big(\sqrt{{g_0}/{g_s}}(R^{\dagger}_+)^{-1}\otimes \sqrt{{g_0}/{g_s}}(R^{\dagger}_+)^{-1}\Big)$ gives a well defined Hadamard parametrix. Notice that, by $(ii)$ in Proposition \ref{prop_2_properties_of_R}, using \eqref{eq_2_P_s} and \eqref{eq_2_G_s^+-} to handle the density quotients; for all $f_1,\ f_2 \in C^{\infty}_c(M)$
$$
\begin{aligned}
    \widetilde H(f_1, P(h_s)f_2) &=  W_{s,0}\big((\mathcal{R}^{\dagger}_+)^{-1}(f_1),f_2\big),\\
    \widetilde H( P(h_s)f_1,f_2) &=  W_{s,0}\big(f_1,(\mathcal{R}^{\dagger}_+)^{-1}(f_2)\big),\\   
\end{aligned}
$$
where $W_{s,0}$ is the smooth part of the Hadamard parametrix $H_s$ constructed with the background geometry $h_s$. We can estimate
$$
    \mathrm{WF}((\mathcal{R}^{\dagger}_+)^{-1})\subseteq \mathrm{WF}(\mathcal{G}_{(M,h_s)}^-),
$$
however, since $\mathrm{WF}(W_{s,0})=\emptyset$, by \cite[Theorem 8.2.14]{hormanderI} 
$$
    \mathrm{WF}\big(W_{s,0}\circ \big(id_{C^{\infty}_c(M)}\otimes (\mathcal{R}^{\dagger}_+)^{-1}\big) \big)=\emptyset, \quad \mathrm{WF}\big(W_{s,0}\circ \big((\mathcal{R}^{\dagger}_+)^{-1} \otimes id_{C^{\infty}_c(M)}\big) \big)=\emptyset.
$$
Applying \cite[Theorem 5.1]{radzikowski1996micro} establishes that 
$$
    \widetilde{\omega}= \widetilde H +\frac{i}{2} \mathcal{G}_{(M,h_0)}
$$
is the two point function of a well defined Hadamard state. We can thus define
$$
    \tau_{(M,h_s\to h_0)} : \mathfrak A_{\mu loc}(M,h_s) \to \mathfrak  A_{\mu loc}(M,h_0), \quad A \mapsto \tau_{(M,h_s\to h_0)}(A).
$$
where, given any parametrix $H$ related to the background geometry $(M,h_0)$, 
\begin{equation}\label{eq_2_tau_isomorphism}
    \tau_{(M,h_s\to h_0)}(A)(H)= A(H'_s)
\end{equation}
where $H'_s= H \circ\big(\mathcal{R}^{\dagger}_+\otimes \mathcal{R}^{\dagger}_+\big) $ is a Hadamard parametrix for $(M,h_s)$. Notice that selecting another parametrix $H'$ amounts to shift everything by a smooth function $d=H'-H$, which by Lemma \ref{lemma_2_regular_algebra_isomorphism} and \ref{lemma_2_unique_conv} does not affect our isomorphism.
%Notice that $A \in  \mathfrak A_{\mu loc}(M,h_s) $, has $A(H'_s)\in  \mathcal{F}_{\mu loc}(M,h_s)=\mathcal{F}_{\mu loc}(M,h) $ for any Hadamard parametrix $H'_s$ with respect to the background geometry $(M,h_s)$, therefore $\alpha_{\widetilde H-H_s}^{-1}(A)=\alpha_{\widetilde H-(H_s-H'_s)}^{-1}(F)=\alpha_{\widetilde H-d}^{-1}(F)\in  A_{\mu loc}(M,h_0)$ by Lemma \ref{lemma_2_regular_algebra_isomorphism} and our previous results on $\widetilde H$. 
The action of $\tau_{(M,h_s\to h_0)}$ on Wick powers (\textit{c.f.} equation \eqref{eq_2_Wick_powers}), is then given by
$$
    \tau_{(M,h_s\to h_0)}\Big(\alpha^{-1}_{H_s}\big(\phi^k_{(M,h_s)}\big) \Big)= \alpha^{-1}_{\widetilde H}\big(\phi^k_{(M,h_s)}\big) 
$$
\textbf{Step 3.}
Finally, if $h_s$ is a compactly supported variation of $h=h_0$, by an argument similar to the one in the proof of Lemma \ref{lemma_A_locality&bastiani_implies_w-reg}, then $s\mapsto h_s\in \Gamma^{\infty}(M\leftarrow HM)$ is a smooth mapping when the latter is equipped with the convenient smooth structure of \cite[THeorem 42.1]{convenient}. Therefore $s\mapsto H_s\in \mathcal{D}'(M\times M)$ is smooth as well for any Hadamard parametrix $H_s$ associated to $h_s$. By Lemma \ref{lemma_2_propagator_smothness} we deduce that the propagators are smooth as well, thus
$$
    s \mapsto \widetilde{\omega}= \omega_s \circ\big(\big(R^{\dagger}_+\big)^{-1}\otimes\big(R^{\dagger}_+\big)^{-1}\big)\in \mathcal{D}'(M\times M)
$$
is smooth as well. Moreover, fixed $s\in \mathbb R^d$, for any $H \in \mathrm{Had}(M,h)$ the distribution
$$
    \mathcal{D}'(M) \ni f\mapsto \alpha_H\Big(\phi^k_{(M,h,\widetilde H)}\Big)(f)(\varphi) \in \mathbb R
$$
has integral kernel $\alpha_H\Big(\phi^k_{(M,h,\widetilde H)}\Big)(\varphi)$ in $C^{\infty}(M)\subset \mathcal{D}'(M)$. Using the smoothness properties just established and the fact that $\widetilde H - H \in C^{\infty}(M\times M)$, %of $s\mapsto \mathcal{R}^{\dagger}_+$ and $s \mapsto H_s$, established in the previeous step, 
we get that 
\begin{equation}\label{eq_Wick_test_smooth}
	s\mapsto \alpha_H\Big(\alpha^{-1}_{\widetilde H}\big(\phi^k_{(M,h)}\big)\Big)(\varphi)= e^{\widetilde{\Gamma}_{H-\widetilde H}}\big(\phi^k_{(M,h)}\big)\big|_{\varphi} \in C^{\infty}(M)
\end{equation}
is smooth. We have therefore shown that the family of mappings $\Gamma^{\infty}(M\leftarrow HM)\ni h\to e^{\Gamma_{H-\widetilde H}}\big(\phi^k_{(M,h)}\big)\big|_{\varphi} \in C^{\infty}(M)$ maps smooth curves $s\mapsto h_s$ in $C^{\infty}\big(\mathbb R,\Gamma^{\infty}(M\leftarrow HM)\big)$\footnote{Notice that, by \cite[Remark 42.2]{convenient} the smooth curves $:\mathbb R \to \Gamma^{\infty}(M\leftarrow HM)$ are the same whether the latter space is endowed with the convenient smooth structure of \cite[Theorem 42.1]{convenient} or with the Bastiani smooth of \Cref{thm_1_mfd_mappings}. Therefore, using an argument similar to that employed in the proof of Lemma \ref{lemma_A_locality&bastiani_implies_w-reg}, we can show that all regular variation are smooth curves $:\mathbb R \to \Gamma^{\infty}(M\leftarrow HM)$.} into smooth curves $C^{\infty}\big(\mathbb R,C^{\infty}(M)\big)$ of the form \eqref{eq_Wick_test_smooth}. By Definition \ref{def_A_convenient_smooth_map} this is conveniently smooth and using Proposition \ref{prop_2_equivalence_conv-smooth_w-regularity} we conclude.
\end{proof}

\begin{remark*}
A key difference with the Euclidean case treated in \cite{drago} is that, in the Lorentzian case, the singularity structure of both $\mathcal{G}_s^{\pm}$ and the Hadamard parametrix $H_s$ does depend on $s$, in the sense that it has a singularity component $\{(x,x;\xi,-\xi)\}$ along the diagonal which is independent from the metric variation plus a part along bicharacteristics of lightlike geodesics which depends on the metric variation; in the Euclidean case, only the former is present, therefore one can conclude the above proof by the smooth dependence on the parameter $s$ and the fact that the singularity structure is independent from $s$ itself without needing cartesian closedness. In the Lorenzian case, this is not possible; for example, consider the family of distributions $s\cdot\delta_0 $ with $s\in \mathbb R$, $\delta_0\in \mathcal{D}'(\mathbb R)$ the Dirac's distribution, then this is a family of distributions which is smooth in $s$, but the wave front set is
$$
    \mathrm{WF}=\{(s,0;0,\xi)\in \dot T^*\mathbb R^2, s\neq 0 \}\cup \{(0,0,\zeta,\xi)\in \dot T^*\mathbb R^2 \}.
$$
The reason is that when $s=0$ the above distribution in $x$ is the zero distribution which is of course smooth, however, the moment $s$ changes, new $\delta$-like singularities appear at $x=0$. Therefore in the Lorentzian case, where the wave front set depends on $s$, it is not automatic to conclude that the integral kernel of \eqref{eq_Wick_test_smooth} is jointly smooth in $(s,x)$ and some extra argument is needed.
\end{remark*}

\chapter{ Functional formalism in quantum field theory: Time ordered products}\label{chapter_TO}
\thispagestyle{plain}

% \nomenclature{$M_\infty$}{Free Stream Mach number}
% \newacronym{cfd}{CFD}{Computational Fluid Dynamics}

In this chapter we continue our discussion of quantum free field theories by studying their \textit{time ordered products}. We can see it as a second product $\cdot_T$ on the algebra $\mathfrak A_{\mu loc}(M,h)\subset \mathfrak A_{\mu c}(M,h)$, this enable us to define the $S$-matrix which is crucial in perturbation theories to treat interactions. In particular, once a time ordered product has been obtained we can set 
$$
    S(V)=\exp_{\cdot_T}(V)= 1_{\mathfrak A_{\mu c}(M,h)}+\sum_{j\geq 1} \frac{1}{j!}\bigg(\frac{i}{\hbar}\bigg)^j \underbrace{{V\cdot_T \ldots \cdot_T V}}_{\hbox{$j$ times}}.
$$
In general, the series defining $S(V)$ is not convergent, however, when $V\in \mathfrak A_{reg}(M,h)$, we shall see that this formula still be interpreted as a formal power series in $\hbar$. The above formula is however not physically interesting unless $V\in \mathfrak A_{\mu loc}(M,h)$, then the latter can be thought as the interaction term in the Lagrangian, \textit{i.e.} 
$$
    \mathcal{L}=\underbrace{{\mathcal{L}_0}}_{\hbox{free part}}\ +\ \underbrace{{\alpha_H(V)}}_{\hbox{interacting part}}.
$$
This however causes an important issue: time ordered products as defined for $V\in \mathfrak A_{reg}(M,h)$ suffer from the so-called \textit{UV-divergences}, those are \textit{singularities} resulting from the impossibility of multiplying certain distributions supported in the diagonal. Notice that in principle there could be other divergences as well: namely \textit{IR-divergences} which are related to the possibility of the interaction being arbitrarily extended throughout spacetime; however, in our setting, functionals have compact support and this is therefore avoided. We remark however that when the $UV$-problem is solved, one can study more general interactions via the \textit{adiabatic limit} which consists in applying a cutoff function $f$ to the non-compact interactions and then study the behaviour of the $S$-matrix (or equivalently the time ordered products) as $f\to 1$.\\

As mentioned in the Introduction, we shall use a Epstein-Glaser renormalization scheme which gives us the possibility of constructing time ordered products at each order in $\hbar$ consistently with a series of physical constraints (see Definition \ref{def_3_TO_products}) on the renormalization scheme itself. We shall then use the local Wick expansion Lemma \ref{lemma_3_Wick_expansion} to characterize the renormalization of time ordered products into the problem of extending certain distributions subject to other requirements. The argument for extending those distributions essentially follows that of \cite{hollands2002existence}, with the main difference that instead of the analyticity requirement (see \cite[$\S$ 2, Conditions T4-T6]{hollands2002existence}), we use a generalized off-shell version of the parameterized microlocal condition introduced in \cite[Definition 3.5]{khavkine2016analytic}. The argument for the extension of the aforementioned distributions (\Cref{thm_3_TO_renormalization_1} and \Cref{thm_3_TO_renormalization_2}) does benefit from the recent results in \cite{D16} which are thoroughly analysed in \Cref{section_dang}. Finally, in \Cref{section_TO_uniqueness}, we will show, in \Cref{thm_2_uniqueness_TO}, that the newly introduced parameterized microlocal condition is consistent with the Main Theorem of renormalization (for reference see  \cite[Theorem 4.1]{BDF09} or \cite[Theorem 3.6.3]{dutsch2019classical}); moreover, in Corollary \ref{coro_2_TO_uniqueness} we obtain that the difference of two time ordered products satisfying the prescription Definition \ref{def_3_TO_products}, has the usual form 
$$
\begin{aligned}
    \phi^{k_1}_{(M,h)}\cdot_T\ldots\cdot_T\phi^{k_p}_{(M,h)}-\phi^{k_1}_{(M,h)}\cdot_{\widetilde T}\ldots\cdot_{\widetilde T}\phi^{k_p}_{(M,h)}= \sum_{J\leq K}\binom{K}{J}C^{K-J}\phi^{J}_{(M,h)},
\end{aligned}
$$
where $J,\ K$ are multiindices, $C_J: (M,h)\mapsto C_J[h] \in C^{\infty}(M)$ are coefficients, analogously to \Cref{thm_2_Moretti_Kavhkine}, such that each $C_J[h]$ is a polynomial constructed out of scalars made by all tensor-like objects of $j^rh$ for some finite order $r$, $\phi^{J}_{(M,h)}=\phi^{j_1}_{(M,h)}\cdot \ldots \cdot \phi^{j_p}_{(M,h)}$ is the classical product of Wick powers.\\

\section{The extension of distributions on manifolds}\label{section_dang}

As mentioned before, in the Section we follow \cite{D16}.\\

Let $M$ be a smooth $n$-dimensional manifold, $I\subset M$ a closed embedded submanifold and $t\in \mathcal{D}'_{\Gamma}(M\backslash I)$ a distribution. We seek an answer to the following problems:
\begin{itemize}
    \item can we construct an extension $\overline{t}\in \mathcal{D}'(M)$?
    \item do we have some form of control on the wave front set $\mathrm{WF}(\overline{t})$?
\end{itemize}

Both questions admit a positive answer with the proper assumptions on $t$. Notice that the extension problem is far from trivial: indeed, by Hahn-Banach one can always devise an \textit{extension} of $t$ as a linear mapping $\mathcal{D}(M)\to \mathbb{R}$, however, there is no information about its continuity. \\

Suppose $I$ is a closed $\iota$-dimensional submanifold, then there are slice charts $\{(U_{\alpha},(x^i,h^j)_{\alpha})\}$ such that points $x\in I\cap U_{\alpha}$ can be characterized by $\{h^j=0\}$. Below, $U\subset M$ will be an open subset of $M$. Denote by $\mathcal{I}(U)=\{ f \in \mathcal{E}(U) : \space f|_{U\cap I}=0 \}$ the ideal of smooth functions vanishing on $I$, we say that a vector field $X\in \mathfrak{X}(M)$ is an \textit{Euler vector field} if 
\begin{equation}\label{eq_3_Euler_vector_field}
    \forall f \in \mathcal{I}(U), \ X(f)-f^2 \in \mathcal{I}^2(U),
\end{equation}
where $\mathcal{I}^2(U)$ is the set $\{f \cdot g \in C^{\infty}(U) : \ f,g \in \mathcal{I}(U)\}$. One can easily show that Euler vector fields can be characterized locally as
$$
    X= h^j\frac{\partial}{\partial h^j} + h^j A^i_j(x,h)\frac{\partial}{\partial x^i} +  h^jh^k A^i_{jk}(x,h)\frac{\partial}{\partial h^i},
$$
with $A,\ B$ smooth coefficients. Ideally, Euler vector field flow transversally out of the region $I$; this is no coincidence since their flow can be used to define the scaling degree of distributions near $I$. For notation's sake we will denote 
\begin{equation}\label{eq_3_logarithmic_flow}
    F^X_{\lambda}=\mathrm{Fl}^X_{\ln{\lambda}},\ \ \lambda\in (0,1]
\end{equation}
the logarithmic flow.

A simple example of an Euler field with $M=\mathbb{R}^{n+\iota}$, $I=\{0\}\times \mathbb{R}^{\iota}$ is $X= h^j\frac{\partial}{\partial h^j}$. Then the associated logarithmic flow satisfies $(F^X_{\lambda})^*f(x,h)=f(x,\lambda h)$ for all $f\in \mathcal{E}(M)$. We therefore see how \eqref{eq_3_logarithmic_flow} is the natural candidate to induce the coordinate scaling in a geometrically consistent way. In the following Proposition we collect some properties about Euler vector fields.

\begin{proposition}\label{prop_3_properties_of_Euler_fields}
    The following assertion are true:
    \begin{enumerate}
        \item[(i)] Suppose $I \hookrightarrow M$, $I' \hookrightarrow M'$ be two closed embedded submanifolds. Let $\psi \in \mathrm{Diff}_{loc}(M,M')$ and consider open subsets $U, \ U'$ such that $U \simeq^{\psi} U'$ and $\psi(U\cap I)=U'\cap I'$, then if $X\in \mathfrak{X}(U)$ is Euler, $T\psi(X)\in \mathfrak X(U') $ is Euler as well.
        \item[(ii)] Let $X_1,\ X_2$ be Euler fields in a neighborhood of $x\in I$. Then there is a one parameter smooth family of diffeomorphism $\Psi_{\lambda}$ defined on a neighborhood of $x$ such that
        $$
            F^{X_2}_{\lambda}=F^{X_1}_{\lambda} \circ \Psi_{\lambda}.
        $$
        \item[(iii)] Let $X_1,\ X_2$ be Euler fields in a neighborhood of $x\in I$. Then there is a diffeomorphism $\psi$ defined in a neighborhood of $x$ such that 
        $$
            X_1=T\psi \circ X_2 \circ \psi^{-1}.
        $$
    \end{enumerate}
\end{proposition}

This result is very important since it states that the property of being Euler vector fields is invariant under the action of diffeomorphism and that given two Euler vector fields their coordinate scaling (and thus the vector field themselves) are diffeomorphic to each other. An open subset $U\subset M$ is said to be $X$-\textit{stable} if it is stable under the logarithmic flow of the Euler vector field $X$, that is $(F^X_{\lambda})^*U \subseteq U$.

\begin{definition}\label{def_3_w_homogeneous_distributions_1}
    Let $U\subset M$ be an $X$-stable open subset, denote by $E_s^X$ the set of distributions $t\in \mathcal{D}'(U)$ such that
    $$
        \forall f \in \mathcal{D}(U), \ \exists \ C>0 , \sup_{\lambda\in(0,1]}\big| \big\langle\lambda^{-s}(F^X_{\lambda})^*t,f \big\rangle \big|\leq C.
    $$
    Finally, given $x\in I$ denote by $E_{s,x}^X$, the set of distributions $t\in \mathcal{D}'(U)$ such that there is a neighborhood $U\ni x$ for which $t\in E_s^X$.
\end{definition}
If we represent the distribution $t$ as the integral kernel $t(x,h)$ in the slice coordinates $(x,h)$ then the action of the coordinate scaling $F^X_{\lambda}$ is given by $(F^X_{\lambda})^*t(x,h)\equiv t_{\lambda}(x,h)=t(x,\lambda h)$. The above condition states that elements of $E_{s,x}^X$ are those distributions $t$ for which $\lambda^{-s}t_{\lambda}$ is bounded in $\mathcal{D}'(V)$ for a suitable neighborhood $V$ of $x$. We can see that the value $s$ is the opposite of the Steinmann scaling degree defined in \cite{steinmann1971perturbation}.

Using $(ii)$ of Proposition \ref{prop_3_properties_of_Euler_fields}, we can show that the space $E^X_{s,x} $ does not depend on the Euler field chosen to perform the scaling. To wit, given Euler fields $X_1, \ X_2$, consider 
$$
    \lambda^{-s} (F^{X_2}_{\lambda})^*t =  \lambda^{-s} (F^{X_1}_{\lambda} \circ \Psi_{\lambda})^*t =  \lambda^{-s} (\Psi_{\lambda})^*(F^{X_1}_{\lambda} )^*t ,
$$
since $\Psi_{\lambda}$ is a jointly smooth diffeomorphism with inverse $\Psi_{-\lambda}$ we get that the right hand is bounded whenever paired with $f \in \mathcal{D}(U)$ if and only if the left hand side is. From now on, we shall omit the Euler vector field from the spaces of distributions in Definition \ref{def_3_w_homogeneous_distributions_1}. Finally, we say that $t \in E_{x,I}(U)$ if for each $x\in {U} \cap I$ there exists an Euler field $X$ with $t\in E^X_{s,x}$. 

\begin{proposition}\label{prop_3_w_hom_diff_invariance}
    The space $E_{s,I}(U)$ satisfies the following properties:
    \begin{enumerate}
        \item[(i)] suppose that $V_i$ is an open cover for $U\subset M$ and $t\in \mathcal{D}'(\cup_iV_i)$, then $t\in E_{s,I}(V_i)$ for all $i$ implies $t\in E_{s,I}(U)$;
        \item[(ii)] suppose that $I \hookrightarrow M$, $I' \hookrightarrow M'$ are two closed embedded submanifolds, $\psi \in\mathrm{Diff}_{loc}(M,M') $ such that $ \psi: U \to U'$ is a diffeomorphism for the open subsets $U, \ U'$ having $\psi(U\cap I)=U'\cap I'$. Then $\psi^*E_{s,I'}(U')=E_{s,I}(U)$. 
    \end{enumerate}
\end{proposition}

\begin{proof}
    The first assertion follows easily for if $x\in {U}\cap I$, then $x\in V_i$ for some $i$ (since $\{V_i\}$ is an open cover of $U$, $\overline{U}\subset \cup_iV_i$) and $t\in E_{s,I}(V_i)$. The latter implies that there is a neighborhood $x\in W_i \subset V_i\cap U$ such that $\lambda^{-s}t_{\lambda}$ is bounded in $\mathcal{D}'(W_i)$, but then $t\in E_{s,x}(U)$ as well. For the second, simply note that by (i) in Proposition \ref{prop_3_w_hom_diff_invariance} we can localize the problem and show the corresponding claim for points $I'\cap U' \ni x'= \psi(x), \ x \in U\cap I$. Suppose there is an Euler field $X$ for which $\lambda^{-s} (F^X_{\lambda})^*t$ is bounded in $\mathcal{D}'(V)$, then
    $$
    \begin{aligned}
        & \lambda^{-s} (F^X_{\lambda})^*t =  \lambda^{-s} (F^X_{\lambda})^* \psi_*\psi^*t\  \mathrm{is} \ \mathrm{bounded} \ \mathrm{in}  \ \mathcal{D}'(V) \\
        \Leftrightarrow & \ \lambda^{-s} \psi^*(F^X_{\lambda})^* \psi_*\psi^*t \mathrm{is} \ \mathrm{bounded } \ \mathrm{in} \ \mathcal{D}'(\psi^{-1}(V))\\
        \Leftrightarrow &\  \lambda^{-s} (F^{T\psi(X)}_{\lambda})\psi^*t \mathrm{is} \ \mathrm{bounded } \ \mathrm{in} \ \mathcal{D}'(\psi^{-1}(V)).
    \end{aligned}
    $$
    By $(iii)$ in Proposition \ref{prop_3_properties_of_Euler_fields}, $T\psi(X)$ is still Euler and we conclude.
\end{proof}

We can now start addressing the problem of the extension of distributions, since our goal is to control the wave front set of the extension, it seems natural to strengthen the requirement of \ref{def_3_w_homogeneous_distributions_1}, in particular to control boundedness of $\lambda^{-s}t_{\lambda}$ in the H\"ormander topology (see \eqref{eq_1_hormander_seminorm}) and hope that the extra requirement allow us to control the wave front of the extension. Let $U\subset M$ be open, $\Gamma$ be a cone in $\dot{T}^*U$\footnote{For here on we shall denote by $\dot{T}^*U$ the cotangent space of $U$ minus the graph of the null section.}.

\begin{definition}\label{def_3_w_homogeneous_distributions_2}
    Let $U,\ \Gamma\subset \dot{T}^*U$ as above, we say that $t\in \mathcal{D}_{\Gamma}'(U\backslash I)$ is \textit{weakly homogeneous of degree} $s$ if for all $f\in \mathcal{D}(U\backslash I),\ \chi\in \mathcal{D}(U\backslash I), \ q\in \mathbb{N}, \ \Upsilon\subset \mathbb{R}^n$ having $\Gamma\cap \mathrm{supp}(\chi)\times \Upsilon=\emptyset $, there is an Euler vector field $X$ and positive constants $c, C$ with
    \begin{equation}
        \begin{aligned}
            &\sup_{\lambda \in (0,1]} \left\vert \lambda^{-s} t_{\lambda}\right\vert\leq c, \ \sup_{\lambda \in (0,1]} \left\Vert \lambda^{-s} t_{\lambda}\right\Vert_{\chi,q,\Upsilon}\leq C
        \end{aligned}
    \end{equation}
    We shall denote $E_{s,I,\Gamma}(U\backslash I)$ the set of such distributions.
\end{definition}

Note that we have a projection $E_{s,I,\Gamma}(U\backslash I)\to E_{s,I}(U\backslash I)$, moreover as $\Gamma\to \dot{T}^*U$, then $E_{s,I,\Gamma}(U\backslash I)\to E_{s,I}(U\backslash I)$.\\

Suppose that $M=\mathbb{R}^{n}\equiv \mathbb{R}^{l+\iota}$, $I= \{0\}\times \mathbb{R}^{\iota}$, consider a bump function $\chi\in C^{\infty}(\mathbb{R}^{l+\iota}) $ which is $1$ in $[-1/2,1/2]^l\times I$ and identically zero outside $[-1,1]^l\times I$. Set 
$$
    \Psi_{\Lambda}(x,h) = 1- \chi(x,\Lambda h).
$$
Then, for all $t\in \mathcal{D}(U\backslash I)$, $\lim_{\Lambda\to \infty} \Psi_{\Lambda}\cdot t=t$ that is, for all $f\in \mathcal{D}(U\backslash I)$, $\lim_{\Lambda\to \infty} \langle\Psi_{\Lambda}\cdot t,f\rangle=\langle t,f\rangle$. Since it will be needed in the upcoming results we give a short proof of the former claim. Notice that we can write
$$
\begin{aligned}
    1- \chi(x,\Lambda h)&= 1-\chi(x, h) +\chi(x, h)-\chi(x,\Lambda h) =  \Psi_{1}(x,h)+ \int_{\Lambda}^1\frac{d}{d\lambda} \chi(x,\lambda h)d\lambda\\
    &= \Psi_{1}(x,h)+ \int_{\Lambda}^1h^j\frac{\partial}{\partial h^j}\chi(x,\lambda h)d\lambda= \Psi_{1}(x,h)- \int_{\Lambda^{-1}}^1(X\chi)(x,\lambda^{-1} h)\frac{d\lambda}{\lambda}
\end{aligned}
$$
where $X=h^j\frac{\partial}{\partial h^j}$ is an Euler vector field. Calling $X\chi=\psi$, $\epsilon=\Lambda^{-1}$, $\epsilon\to 0$ as $\Lambda\to \infty$; therefore 
$$
\begin{aligned}
    \lim_{\epsilon\to 0} \langle\big(1- \chi_{\Lambda}\big)t, f \rangle &= \lim_{\epsilon\to 0} \Big\langle \int_{\epsilon}^1 \frac{d\lambda}{\lambda} \psi_{\lambda^{-1}}t,f\Big\rangle+\big\langle\big(1- \chi\big)t, f \big\rangle \\
    &= \lim_{\epsilon\to 0} \int dxdh \int_{\epsilon}^1 \frac{d\lambda}{\lambda} \psi(x,\lambda^{-1}h)t(x,h)f(x,h)+\int dxdh \big(1- \chi(x, h)\big)t(x,h) f(x,h) \\
    &= \langle t, f\rangle;
\end{aligned}
$$
combining the above expression with the fact that $\chi_{\Lambda}\to 1$  as $\Lambda\to \infty$ in the Fréchet space $C^{\infty}(\mathbb{R}^{n})$, we get the initial claim.\\

We are now ready to state the first extension result. We shall see that under suitable conditions we can extend $t$ by $\lim_{\epsilon\to 0}(1-\chi_{\epsilon^{-1}})t$ and control its wave front set.

\begin{theorem}\label{thm_3_w_homogeneous_extension_1}
    Let $t\in E_{s,I,\Gamma}(U\backslash I)$ and suppose that $s+\iota>0$, $\Gamma\subset \dot{T}^*U$ is stable under the scaling induced by Euler fields of $I$. Then $\overline{t}=\lim_{\epsilon\to  0}(1-\chi_{\epsilon^{-1}})t\in \mathcal{D}'(U)$ and 
    $$
        \mathrm{WF}(\overline{t})=\mathrm{WF}({t})\cup N^{*}I\cup \Xi,
    $$
    where $N^*I$ is the conormal bundle to $I$ and 
    $$
        \Xi=\big\{(x,0,\xi,\eta)\in \dot{T}^*U : \ \exists (x,h,\xi,0) \in \Gamma\cap T^*U|_{\mathrm{supp}(\psi)} \big\}.
    $$
    Moreover, we have $\overline{t}\in E_{s,I,\Gamma\cup N^*I \cup \Xi}$.
\end{theorem}

\begin{proof}
We divide the proof in three steps: in the first we show that $\lim_{\epsilon\to  0}(1-\chi_{\epsilon^{-1}})t$ converges in $\mathcal{D}'(U)$; in the second, we show that the family $\{(1-\chi_{\epsilon^{-1}})t\}_{\epsilon}$ is bounded in $\mathcal{D}'_{\mathrm{WF}({t})\cup N^{*}I\cup \Xi}(U)$, then this entails the convergence of $(1-\chi_{\epsilon^{-1}})t$ in $\mathcal{D}'_{\mathrm{WF}({t})\cup N^{*}I\cup \Xi}(U)$ as $\epsilon\to 0$; finally, we show that the extension is in $ E_{s,I,\Gamma\cup N^*I \cup \Xi}$.
    \textbf{Step 1.} Since 
    $$
    \begin{aligned}
        \langle\big(1- \chi_{\epsilon^{-1}}\big)t, f \rangle &=  \int_{\epsilon}^1 \frac{d\lambda}{\lambda}\langle  \psi_{\lambda^{-1}}t,f\rangle+\big\langle\big(1- \chi\big)t, f \big\rangle\\
        &= \int dxdh \int_{\epsilon}^1 \frac{d\lambda}{\lambda} \psi(x,\lambda^{-1}h)t(x,h)f(x,h)+\big\langle\big(1- \chi\big)t, f \big\rangle\\
        &= \int dxdh \int_{\epsilon}^1 \frac{d\lambda}{\lambda} \lambda^{s+\iota} \psi(x,h)\lambda^{-s} t(x,\lambda h)f(x,\lambda h)+\big\langle\big(1- \chi\big)t, f \big\rangle
    \end{aligned}
    $$
    we see that the only term requiring attention is the first one for the other is a constant for any $f\in \mathcal{D}(U)$. Notice that since $s+\iota>0$ the term $\lambda^{s+\iota}$ is integrable, moreover, by assumption $\lambda^{-s} t(x,\lambda h)f(x,\lambda h)$ is bounded. Therefore the limit exists. To show that it is a distribution let $f\in \mathcal{D}(U)$ be supported inside the compact subset $K$, then
    $$
    \begin{aligned}
        & \Big|\int dxdh \int_{\epsilon}^1 \frac{d\lambda}{\lambda} \lambda^{s+\iota} \psi(x,h)\lambda^{-s} t(x,\lambda h)f(x,\lambda h) \Big|\\
        & \leq \Big|\int dxdh \int_{\epsilon}^1 \frac{d\lambda}{\lambda} \lambda^{s+\iota}\sup_{\lambda\in (0,1]} \psi(x,h)\lambda^{-s} t(x,\lambda h)f(x,\lambda h) \Big|,
    \end{aligned}
    $$
    however, $\langle \lambda^{-s} t_{\lambda} \psi f_{\lambda} \rangle\leq C_{K,\psi} p_{K,j}(f)$ by assumption, thus
    $$
        |\langle (\chi-\chi_{\epsilon^{-1}})t,f \rangle |\leq C_{K,\psi} p_{K,j}(f)\Big| \int_{\epsilon}^1 \frac{d\lambda}{\lambda} \lambda^{s+\iota}\Big| \leq \frac{1-\epsilon^{s+\iota}}{s+\iota}C_{K,\psi} p_{K,j}(f)
    $$
    stays bounded when $\epsilon \to 0$.
    \textbf{Step 2.} We already know that outside $I$ the family $(\chi-\chi_{\epsilon^{-1}})t$ is bounded in $\mathcal{D}'_{\mathrm{WF}(t)}(U\backslash I)$, in order to study the latter problem in $I$, we write 
    $$
    \begin{aligned}
        \langle\big(\chi- \chi_{\epsilon^{-1}}\big)t, f \rangle= \int dx'dh' \int dxdh \int_{0}^{\infty} 1_{[\epsilon,1]}(\lambda) \frac{d\lambda}{\lambda} \lambda^{s+\iota} \psi(x,h)\lambda^{-s} t(x,\lambda h)\delta(x-x',\lambda h-h')f(x', h'),
    \end{aligned}
    $$
    where $1_{[\epsilon,1]}$ is the characteristic function of the set $[\epsilon,1]\subset \mathbb{R}$ and $\delta $ the Dirac distribution. By \cite[Lemma 10.3]{D16}, we conclude that $A_{\epsilon}=1_{[\epsilon,1]}(\lambda) \frac{d\lambda}{\lambda} \lambda^{s+\iota} \psi(x,h)\lambda^{-s} t(x,\lambda h)$ is a well defined distribution whenever $\epsilon\in (0,1]$ and 
    $$
        \mathrm{WF}(A_{\epsilon}) \subset \{(\lambda,x,h;\zeta, \xi,\eta)\in \dot{T}^*(\mathbb{R}\times U): \ (x,h)\in \mathrm{supp}(\psi), \ (x,h;\xi,\eta)\in \mathrm{WF}(t)\cup 0\}.
    $$
    Moreover,
    $$
        \mathrm{WF}(\delta) \subset \{(\lambda,x,h,x',h';-\langle\eta,h\rangle, -\xi,-\lambda\eta,\xi,\eta)\in \dot{T}^*(\mathbb{R}\times U\times U):\  x'=x, \ h'=\lambda h, \  (\xi,\eta)\neq 0\}.
    $$
    therefore by \cite[Theorem 8.2.14]{hormanderI}
    $$
    \begin{aligned}
        \mathrm{WF}(A_{\epsilon}\delta) \subset  \{ & (x',h'; \xi',\eta')\in \dot{T}^*U: \ \exists (\lambda,x,h;-\zeta, -\xi,-\eta) \in \mathrm{WF}(A_{\epsilon}), \\ 
        & (\lambda,x,h,x', h';\zeta, \xi,-\eta,\xi',\eta')\in  \mathrm{WF}(\delta)\}\\
        & \cup \{ (x',h'; \xi',\eta')\in \dot{T}^*U: \ \exists (\lambda,x,h,x', h';0,0,0,\xi',\eta')\in  \mathrm{WF}(\delta)\}
    \end{aligned}
    $$
    By direct inspection, setting $h'=0$, the above set becomes 
    $$
    \begin{aligned}
        &\{(x',0,\xi',\eta')\in \dot{T}^*U : \ \exists (x',h',\xi',0) \in \Gamma\cap T^*U|_{\mathrm{supp}(\psi)}\} \cup \{(x',0,\xi',\eta')\in \dot{T}^*U : \xi'=0\} = \Xi\cup N^{*}I
    \end{aligned}
    $$
    \textbf{Step 3.} Consider the family $\mu^{-s}\overline{t}_{\mu}$, to complete the proof we have to show that it is bounded in $\mathcal{D}'_{\mathrm{WF}(t)\cup \Xi \cup N^*I}(U)$. Since $\overline{t}=\lim_{\epsilon\to 0} (1-\chi_{\epsilon^{-1}})t$, it suffice to study the family $\{\mu^{-s}\big((1-\chi_{\epsilon^{-1}})t\big)_{\mu}\}_{(\epsilon,\mu)\in (0,1]}$. In integral notation we can write
    $$
    \begin{aligned}
        \mu^{-s}\big\langle \big((1-\chi_{\epsilon^{-1}})t\big)_{\mu},f\big\rangle&=  \int_{\epsilon}^1  \frac{d\lambda}{\lambda} \big\langle  \mu^{-s}\psi_{\mu\lambda^{-1}}t_{\mu},f\big\rangle+\big\langle\mu^{-s}\big(1- \chi\big)t_{\mu}, f \big\rangle\\
        &=   \int_{\epsilon}^1 \frac{d\lambda}{\lambda} \mu^{-s-\iota} \big\langle \psi_{\lambda^{-1}}t,f_{1/\mu} \big\rangle+\big\langle\mu^{-s}\big(1- \chi\big)t_{\mu}, f \big\rangle\\
        &=   \int_{\epsilon}^1 \frac{d\lambda}{\lambda} \frac{\lambda^{s+\iota}}{\mu^{s+\iota}} \big\langle \lambda^{-s} \psi t_{\lambda},f_{\lambda/\mu} \big\rangle+\big\langle\mu^{-s}\big(1- \chi\big)t_{\mu}, f \big\rangle\\
         &=   \int_{\epsilon/\mu}^{1/\mu} \frac{d\lambda}{\lambda} {\lambda^{s+\iota}} \big\langle (\mu\lambda)^{-s} \psi t_{\lambda\mu},f_{\lambda} \big\rangle+\big\langle\mu^{-s}\big(1- \chi\big)t_{\mu}, f \big\rangle.
    \end{aligned}
    $$
    Of the two terms above, the second does not pose any problem since it is bounded by assumption and independent from $\epsilon$, however the first might for, even though $\big\langle (\mu\lambda)^{-s} \psi t_{\lambda\mu},f_{\lambda} \big\rangle$ is bounded by assumption, the lower bound of integration might not be. Therefore, suppose that $\mathrm{supp}(f)\subset B_R$ for a suitable ball of radius $R$ centered at the origin of $U\subset \mathbb{R}^n$, then $\mathrm{supp}(f)\subset \{ (|x|^2+|h|^2)^{1/2}<R \} \subset \{|h| < R \}$ and $\mathrm{supp}(f_{\lambda})\subset \{ (|x|^2+|\lambda h|^2)^{1/2}<R \} \subset \{|h| < R/\lambda \}$. Furthermore, we can take $\psi=X(\chi)$ with $X=h^j\partial_j$ and $\chi=1_x\otimes \widetilde{\chi} $ where $\widetilde{\chi}$ is a bump function in the $h$ variable localized near $\{h=0  \}$; then the projection on the $h$ factor of $\mathrm{supp}(\psi)$ will be contained in an interval $(a,b)$ with $0<a<b$. Thus the above integral can be written as
    $$
        \int_0^{\infty} \frac{d\lambda}{\lambda} 1_{[\epsilon/\mu, R/a]}(\lambda) {\lambda^{s+\iota}} \big\langle (\mu\lambda)^{-s} \psi t_{\lambda\mu},f_{\lambda} \big\rangle
    $$
    It's easy to see that the above quantity is bounded as $\epsilon\to 0$ when tested with $f\in \mathcal{D}(U)$. It remain to show that 
    $$
        (1+|\xi|+|\eta|)^q\Big|\int_0^{\infty} \frac{d\lambda}{\lambda} 1_{[\epsilon/\mu, R/a]}(\lambda) {\lambda^{s+\iota}} \big\langle (\mu\lambda)^{-s} \psi t_{\lambda\mu},e^{i(\xi x + \eta h)}\big\rangle\Big|
    $$
    is bounded for all $q\in \mathbb N, f\in C^{\infty}_c(M)$, $\xi,\ \eta \in \Upsilon$, $\Upsilon\cap  \mathrm{WF}(\overline{t})=\emptyset$. The above expression can be written as
    $$
    \begin{aligned}
        (1+|\xi|+|\eta|)^q\Big|\int_0^{\infty} d\lambda g(\lambda)  \mathcal{F}(f\psi (\mu\lambda)^{-s} t_{\lambda\mu})(\xi,\eta) \Big|
    \end{aligned}
    $$
    where $g(\lambda)=1_{[\epsilon/\mu, R/a]}(\lambda) {\lambda^{s+\iota-1}}$ is an integrable function and $\mathcal{F}$ denotes the Fourier transform. By assumption, since $\mathrm{{supp}}(f\psi)\times \Upsilon \cap \mathrm{WF}(\overline{t})=\emptyset$, the above family is bounded $\forall \mu \in (0,1]$ as $\epsilon\to 0$.
\end{proof}

This theorem is a generalization of the result obtained in \cite{BF97}, where the authors used extra hypothesis on top of those in Theorem \ref{thm_3_w_homogeneous_extension_1}. We shall obtain the same result in Theorem \ref{thm_3_w_homogeneous_extension_conormal}. We remark that this result is somehow optimal for constraining the wave front set of the unextended scaled distribution near the singular hypersurface: in fact it produces a well defined extension with control of the wave front set coming from boundedness of the scaled distribution in the H\"ormander topology. The other condition we assumed was that $s+\iota>0$, this combined with Theorem 5.2.3 in \cite{hormanderI} implies that the extension obtained is unique in $E_{s,I,\Gamma}(U)$. We can show however that we can nonetheless derive a similar result for arbitrary $s\in \mathbb{R}$, however uniqueness of the extension will be lost. To treat this more general case, suppose $-m-1<s+\iota\leq -m$ and let 
\begin{equation}\label{eq_3_Taylor_reminder}
    I_m: \mathcal{D}(U)\to \mathcal{D}(U) : f \mapsto \frac{1}{m!} \sum_{|\alpha|= m+1} h^{\alpha}\int_0^1(1-\mu)^m\frac{\partial^{\alpha}}{\partial h^{\alpha}} f(x,\mu h)d\mu.
\end{equation}
Notice that $I_m$ can be thought of as the projection associating to each $f$ a new function $I_m(f)$ which vanishes of order $m+1$ near $I$. We can now state the extension theorem.

\begin{theorem}[Theorem 4.7 in \cite{D16}]\label{thm_3_w_homogeneous_extension_2}
    Let $t\in E_{s,I,\Gamma}(U\backslash I)$ and suppose that there is $m\in \mathbb{N}$ with $-m-1<s+\iota\leq -m$, $\Gamma\subset \dot{T}^*U$ is stable under the scaling of Euler fields of $I$. Then $\overline{t}=\lim_{\epsilon\to  0}(1-\chi_{\epsilon^{-1}})t\in \mathcal{D}'(U)$ and 
    $$
        \mathrm{WF}(\overline{t})=\mathrm{WF}({t})\cup N^{*}I\cup \Xi.
    $$
    Moreover, we have $\overline{t}\in E_{s,I,\Gamma\cup N^*I \cup \Xi}$.
\end{theorem}

\begin{proof}
Then again we proceed by steps: in the first we study the convergence of 
\begin{equation}\label{eq_3_extension_decomposition_2}
     \int_{\epsilon}^1 \frac{d\lambda}{\lambda}\big\langle  \psi_{\lambda^{-1}}t,I_m(f)\big\rangle+\big\langle\big(1- \chi\big)t, f \big\rangle    
\end{equation}
as $\epsilon \to 0$, in the second we estimate the wave front of the Taylor remainder operator $I_m$, in the third we study boundedness of the limit distribution in $\mathcal{D}'_{\mathrm{WF}({t})\cup N^{*}I\cup \Xi}(U)$; finally we show that the extension is in $ E_{s,I,\Gamma\cup N^*I \cup \Xi}$.
    \textbf{Step 1.} Looking at \eqref{eq_3_extension_decomposition_2}, we see that the only term which might misbehave as $\epsilon \to 0$ is the first. Writing $I_m$ with the compact notation $\frac{1}{m!}\sum_{\alpha}h^{\alpha}R_{\alpha}$, then
    $$
    \begin{aligned}
        \int_{\epsilon}^1 \frac{d\lambda}{\lambda}\big\langle  \psi_{\lambda^{-1}}t,I_m(f)\big\rangle &= \int_{\epsilon}^1 \frac{d\lambda}{\lambda}\lambda^{s+\iota+m+1}\big\langle  \psi t_{\lambda},I_m(f)_{\lambda}\big\rangle\\
        &= \sum_{|\alpha|=m+1}\int_{\epsilon}^1 \frac{d\lambda}{\lambda}\lambda^{s+\iota}\big\langle  \psi t_{\lambda},(\lambda h)^{\alpha} (R_{\alpha}(f))_{\lambda}\big\rangle\\
        &= \sum_{|\alpha|=m+1}\int_{\epsilon}^1 \frac{d\lambda}{\lambda}\lambda^{s+\iota+m+1}\big\langle  \psi t_{\lambda}, h^{\alpha} (R_{\alpha}(f))_{\lambda}\big\rangle
    \end{aligned}
    $$
    If we repeat the same argument of \textbf{Step 1} in the proof of Theorem \ref{thm_3_w_homogeneous_extension_1}, we arrive at
    $$
        |\langle (\chi-\chi_{\epsilon^{-1}})t,I_m(f) \rangle |\leq \widetilde{C}_{K,\psi} p_{K,m+j}(f)\Big| \int_{\epsilon}^1 d\lambda \lambda^{s+\iota+m}\Big| \leq \frac{1-\epsilon^{s+\iota+m+1}}{s+\iota+m+1}\widetilde{C}_{K,\psi} p_{K,j+m}(f),
    $$    
    which is bounded as $\epsilon \to 0$.
    \textbf{Step 2.} Let us now look at the operator $I_m$, since we wish to apply the same arguments of the proof of Theorem \ref{thm_3_w_homogeneous_extension_1}, we write it as
    $$
    \begin{aligned}
        I_m(f)(x,h)&= \frac{1}{m!} \sum_{|\alpha|= m+1} h^{\alpha}\int dx'dh'\int_0^1(1-\mu)^m \partial'_{\alpha} \delta(x-x',h'-\lambda h)f(x', h')d\mu\\
        &= \frac{1}{m!} \sum_{|\alpha|= m+1} \int dx'dh'\widetilde{I}_m(x,h,x',h')f(x', h').
    \end{aligned}
    $$
    To study the wave front set of this object we observe that
    $$
    \begin{aligned}
        \mathcal{F}(g \widetilde{I}_m)(\xi,\eta,\xi',\eta')&= \int dxdhdx'dh' \int_0^1(1-\mu)^m \partial'_{\alpha} \delta(x-x',h'-\lambda h)g(x,h,x', h')e^{i(x\xi+h\eta+x'\xi'+h'\eta')}d\mu\\
        &=  \int dxdh \int_0^1(1-\mu)^m \partial_{\alpha} g(x,h,x,\lambda h)e^{i\big(x(\xi+\xi')+h(\eta+\lambda\eta')\big)}d\mu
    \end{aligned}
    $$
    From there we deduce that
    \begin{equation}
        \mathrm{WF}(\widetilde{I}_m)=\{(x,h,x',h'; -\xi',-\lambda\eta',\xi',\eta') \in \dot{T}^*(U\times U) : \ x=x',\ \lambda h=h', \ (\xi,\eta) \neq (0,0)\}.
    \end{equation}
    \textbf{Step 3.} We study the wave front set of $(\chi-\chi_{\epsilon^{-1}})t I_m(x,h) $ as $\epsilon\to 0$; in integral notation,
    $$
    \begin{aligned}
        (\chi-\chi_{\epsilon^{-1}})t I_m(x,h) = \frac{1}{m!} \sum_{|\alpha|= m+1} h^{\alpha} \int dx'dh' \int_{0}^1  & 1_{[\epsilon,1]}(\lambda)\lambda^{s+\iota+m} \lambda^{-s}t(x',\lambda h') \psi(x',h') \\ 
        & \lambda^{-m-1}\widetilde{I}_m(x,h,x',h') 
    \end{aligned}
    $$
    setting $B_{\epsilon}(\lambda,x',h')= 1_{[\epsilon,1]}(\lambda)\lambda^{s+\iota+m} \lambda^{-s}t(x',\lambda h') \psi(x',h') $, we have
    $$
        \mathrm{WF}(B_{\epsilon})\subset \big\{(\lambda,x',h';\zeta, \xi',\eta')\in \dot{T}^*([0,1]\times U): \ (x',h') \in \mathrm{supp}(\psi),\ (x',h';\zeta,\xi',\eta')\in \mathrm{WF}(t)\cup \{0\}\big\}.
    $$
    If we consider the multiplication mapping $(\lambda,x,h,x',h')\mapsto (x,h,x',\lambda h') $, we can estimate 
    $$
    \begin{aligned}
        \mathrm{WF}(\lambda^{-m-1}\widetilde{I}_m)\subset \big\{ & (\lambda,x,h,x',h'; \langle h',\eta' \rangle, -\xi',-\lambda\eta',\xi',\eta')   \in \dot{T}^*([0,1]\times U\times U) \\
        & : \ x=x',\ \lambda h=h', \ (\xi,\eta) \neq (0,0)\big\}.
    \end{aligned}
    $$
    Using \cite[Theorem 8.2.14]{hormanderI} to compose distributions, we get
    $$
    \begin{aligned}
        \mathrm{WF}(B_{\epsilon}\lambda^{-m-1}\widetilde{I}_m) \subset  \{ & (x',h'; \xi',\eta')\in \dot{T}^*U: \ \exists (\lambda,x,h;-\zeta, -\xi,-\eta) \in \mathrm{WF}(B_{\epsilon}), \\ 
        & (\lambda,x,h,x', h';\zeta, \xi,\eta,\xi',\eta')\in  \mathrm{WF}(\lambda^{-m-1}\widetilde{I}_m)\}\cup  \\
        & \{ (x',h'; \xi',\eta')\in \dot{T}^*U: \ \exists (\lambda,x,h,x', h';0,0,0,\xi',\eta')\in  \mathrm{WF}(\lambda^{-m-1}\widetilde{I}_m)\}
    \end{aligned}
    $$
    Then again, setting $h'=0$, the above set becomes 
    $$
    \begin{aligned}
        &\{(x',0,\xi',\eta')\in \dot{T}^*U : \ \exists (x',h',\xi',0) \in \Gamma\cap T^*U|_{\mathrm{supp}(\psi)}\} \cup \{(x',0,\xi',\eta')\in \dot{T}^*U : \xi'=0\} =\Xi\cup N^{*}I.
    \end{aligned}
    $$
    \textbf{Step 4.} Consider the family $\mu^{-s}\overline{t}_{\mu}$, to complete the proof we have to show that it is bounded in $\mathcal{D}'_{\mathrm{WF}(t)\cup \Xi \cup N^*I}$. Since $\overline{t}=\lim_{\epsilon\to 0} (1-\chi_{\epsilon^{-1}})t$, it suffice to study the family $\{\mu^{-s}\big((1-\chi_{\epsilon^{-1}})t\big)_{\mu}\}_{(\epsilon,\mu)\in (0,1]}$. In integral notation we can write
    $$
    \begin{aligned}
        \mu^{-s}\big\langle \big((1-\chi_{\epsilon^{-1}})t\big)_{\mu},f\big\rangle&=  \int_{\epsilon}^1  \frac{d\lambda}{\lambda} \big\langle  \mu^{-s}\psi_{\mu\lambda^{-1}}t_{\mu},I_m(f)\big\rangle+\big\langle\mu^{-s}\big(1- \chi\big)t_{\mu}, f \big\rangle\\
        &=   \int_{\epsilon}^1 \frac{d\lambda}{\lambda} \mu^{-s-\iota} \big\langle \psi_{\lambda^{-1}}t,I_m(f)_{1/\mu} \big\rangle+\big\langle\mu^{-s}\big(1- \chi\big)t_{\mu}, f \big\rangle\\
        &=   \int_{\epsilon}^1 \frac{d\lambda}{\lambda} \frac{\lambda^{s+\iota}}{\mu^{s+\iota}} \big\langle \lambda^{-s} \psi t_{\lambda},I_m(f)_{\lambda/\mu} \big\rangle+\big\langle\mu^{-s}\big(1- \chi\big)t_{\mu}, f \big\rangle\\
    \end{aligned}
    $$
    As in the proof of Theorem \ref{thm_3_w_homogeneous_extension_1}, the second term is bounded, so we focus on the first. We can write $f=P_m(f)+I_m(f)$ Taylor's formula, where $P_m$ is the operator defining the Taylor polynomial in $h$ up to order $m$ and $I_m$ the reminder defined in \eqref{eq_3_Taylor_reminder}. The test function $f$ will be supported inside $\{|h|\leq r\}$ for $r>0$ big enough, thus $f_{\lambda/\mu} $ is supported inside $\{|h|\leq \frac{r\mu}{\lambda}\}$, moreover $\psi$ is supported in $\{|h|>a\}$ for some small enough positive constant $a$. Then we must have $\frac{r\mu}{a}\geq \lambda$, and 
    $$
    \begin{aligned}
        \mu^{-s}\big\langle \big((1-\chi_{\epsilon^{-1}})t\big)_{\mu},f\big\rangle&= \langle I^{\mu}_1,f\rangle+\langle I^{\mu}_2,f\rangle;\\
        \langle I^{\mu}_1,f\rangle&=  \int_{\epsilon}^{\frac{r\mu}{a}} \frac{d\lambda}{\lambda} \frac{\lambda^{s+\iota}}{\mu^{s+\iota}} \big\langle \lambda^{-s} \psi t_{\lambda},I_m(f)_{\lambda/\mu} \big\rangle,\\
        \langle I^{\mu}_2,f\rangle&=  \int_{\frac{r\mu}{a}}^1 \frac{d\lambda}{\lambda} \frac{\lambda^{s+\iota}}{\mu^{s+\iota}} \big\langle \lambda^{-s} \psi t_{\lambda},I_m(f)_{\lambda/\mu} \big\rangle.
    \end{aligned}
    $$
    The advantage is that by writing $I_m(f)=f-P_m(f)$ in $I_2^{\mu}$, the $f$-part is identically zero by the above support properties. Moreover, $\langle I^{\mu}_1,f\rangle$ is bounded by the same argument used in \textbf{Step 3.} of the proof of Theorem \ref{thm_3_w_homogeneous_extension_2}. Expanding $P_m(f)_{\lambda/\mu}$ in $I_2^{\mu}$ we can write 
    $$
    \begin{aligned}
        \langle I^{\mu}_2,f\rangle &=  \int_{\frac{r\mu}{a}}^1 \frac{d\lambda}{\lambda} \frac{\lambda^{s+\iota}}{\mu^{s+\iota}} \big\langle \lambda^{-s} \psi t_{\lambda},-P_m(f)_{\lambda/\mu} \big\rangle\\
        I^{\mu}_2(x',h') &=\int  dxdh \sum_{|\alpha|\leq m+1} \frac{h^{\alpha}}{|\alpha|!}\frac{d\lambda}{\lambda} 1_{[\frac{r\mu}{a},1]}\frac{\lambda^{s+\iota+|\alpha|}}{\mu^{s+\iota+|\alpha|}}\psi(x,y)\lambda^{-s}t_{\lambda}(x,h)\partial'_{\alpha}\delta(x-x',h-\lambda h')
    \end{aligned}
    $$
    Notice that when $s+\iota+|\alpha|<0$, the integral 
    $$
        \int 1_{[\frac{r\mu}{a},1]}\frac{\lambda^{s+\iota+|\alpha|}}{\mu^{s+\iota+|\alpha|}}\frac{d\lambda}{\lambda} = \frac{1}{s+\iota+|\alpha|}\Big(\frac{1}{\mu^{s+\iota+|\alpha|}}-\frac{r^{s+\iota+|\alpha|}}{a^{s+\iota+|\alpha|}}\Big)
    $$
    is bounded as $\mu \to 0$. Now we can combine \cite[Lemma 10.3]{D16} with \textbf{Step 3.} in the proof of Theorem \ref{thm_3_w_homogeneous_extension_2}, to conclude that $I^{\mu}_2$ is bounded in $\mathcal{D}'_{\mathrm{WF}(t)\cup \Xi \cup N^*I}(U)$.
\end{proof}

We are now in a position to derive a more specific result which is equivalent to the extension result of \cite{BF97}. First we introduce the so-called \textit{conormal landing condition}. Let $U$ be an open neighborhood of $I$, a closed conic subset $\Gamma\subset \dot{T}^*(U\backslash I)$ satisfies the conormal landing condition if 
\begin{equation}\label{eq_3_conormal_landing_condition}
    \overline{\Gamma} \cap T^{*}U|_{I} \subseteq N^{*}I.
\end{equation}
The idea behind \eqref{eq_3_conormal_landing_condition} is to further constrain the wave front set of the extended distribution by getting rid of the $\Xi$ part in Theorems \ref{thm_3_w_homogeneous_extension_1}, \ref{thm_3_w_homogeneous_extension_2}. This condition prompts us to define a new set of distributions.

\begin{definition}\label{def_3_w_homogeneous_distributions_conormal}
    Let $\Gamma\subset \dot{T}^*(U\backslash I)$ be a cone satisfying \eqref{eq_3_conormal_landing_condition}, $x\in I$. Then $t\in E_{s,N^*I,x}^X$ if there is an Euler vector field $X$, an open neighborhood $U$ of $x$ stable under the logarithmic flow of $X$, for which the family $\{\lambda^{-s}t_{\lambda}\}_{\lambda\in (0,1]}$ is bounded in $\mathcal{D}'_{\Gamma}(U\backslash I)$ equipped with the H\"ormander topology. 
    \noindent
    Moreover, we say that $t\in E_{s,N^*I}^X(U)$ if $t\in E_{s,N^*I,x}^X$ for all $x\in I\cap U$.
\end{definition}

We remark that combining Propositions \ref{prop_3_properties_of_Euler_fields}, \ref{prop_3_w_hom_diff_invariance} we can show that Definition \ref{def_3_w_homogeneous_distributions_conormal} is independent from the Euler vector field chosen, moreover if $\psi:U\to U'$ is a diffeomorphism satisfying the hypothesis of $(ii)$ Proposition \ref{prop_3_w_hom_diff_invariance}, then $\psi^*E_{s,N^*I}(U)=E_{s,N^*I'}(U')$.

\begin{theorem}\label{thm_3_w_homogeneous_extension_conormal}
    Let $U$ be an open neighborhood of $I$, $t\in E_{s,N^*I}(U\backslash I)$. Then there exists an extension $\overline{t}\in \mathcal{D}'(U)$ having $ \mathrm{WF}(\overline{t})=\mathrm{WF}({t})\cup N^{*}I$.
    Moreover, we have $\overline{t}\in  E_{s',N^*I}(U)$ where $s'=s$ if $s+\iota\notin \mathbb{N}$, $s'<s$ otherwise. Finally if $s+\iota>0$ the extension is unique.
\end{theorem}
\begin{proof}
    The proof essentially mirrors those of Theorems \ref{thm_3_w_homogeneous_extension_1}, \ref{thm_3_w_homogeneous_extension_2} with the only difference the we have to show that condition \eqref{eq_3_conormal_landing_condition} effectively extends the wave front set of the unextended distribution $t$ by an amount $N^*I$. In practice we show that there are $V$ open neighborhood of $I$ and $\chi \in \mathcal{D}(U)$ such that
    \begin{itemize}
        \item if $(x,h;\xi,\eta )\in \mathrm{WF}(t)\cap \dot{T}^*V$ then $\eta \neq 0$;
        \item $\chi \in C^{\infty}_V(U)$ with $\chi=1$ near $I$.
    \end{itemize}
    Suppose there is no such $V$, let $K$ be any compact set\footnote{We are justified in taking a compact subset since the wave front set is evaluated by multiplying the distribution with a test function.} intersecting $I$, $V_n=\{|h|\leq 1/n\}$, by assumption there are $(x_n,h_n;\xi_n,0 )\in \mathrm{WF}(t)\cap \dot{T}^*(K_n\cap V_n)U$, since the latter is a distribution in a compact subset, we can extract a convergent subsequence $(x_n,h_n;\xi_n,0 ) \to (x,0,\xi,0)\in \mathrm{WF}(t)$, however, the conormal landing condition (\textit{c.f.} \eqref{eq_3_conormal_landing_condition}) implies that $\overline{\mathrm{WF}(t)}\cap \dot{T}^*(K_n\cap V_n)U\subseteq N^*I=\{(x,0,0,\eta)\}$. Finally we can easily construct $\chi$ as claimed above, and $\psi=-X(\chi)$ has support in $0<a\leq |h|\leq b$, but for no $(x,h)$ with $a\leq |h|\leq b$ there is $(x,h,\xi,0)\in \mathrm{WF}(t)$; thus we can choose an extension of $t$ with $\Xi=\emptyset$.
\end{proof}

Often in quantum field theory it 
is common to extend the product of, say, Feynman propagators, we therefore state the following result which helps in this regard.

\begin{theorem}\label{thm_3_renormalized_products}
    Suppose that $U$ is a neighborhood of $I$, $\Gamma_1,\ \Gamma_2 $ are cones in $\dot{T}^*(U\backslash I)$ satisfying the conormal landing condition and $\Gamma_1\cap -\Gamma_2=\emptyset$. There there exists a bilinear mapping 
    $$
        \mathcal{R}: \mathcal{D}'_{\Gamma_1}(U\backslash I)\times  \mathcal{D}'_{\Gamma_2}(U\backslash I) \to E_{s,N^*I},\ s<s_1+s_2
    $$
    such that 
    \begin{itemize}
        \item $\mathcal{R}(u_1u_2)=u_1u_2$ in $U\backslash I$,
        \item $\mathcal{R}(u_1u_2)\in \mathcal{D}'_{(\Gamma_1+\Gamma_2)\cup \Gamma_1\cup \Gamma_2 \cup N^*I}$.
    \end{itemize}
\end{theorem}
\begin{proof}
Setting $\Gamma=(\Gamma_1+\Gamma_2)\cup \Gamma_1\cup \Gamma_2 $, by \cite[Theorem 7.1]{brouder2014continuity} the product of distribution is hypocontinuous when defined, thus $u_1u_2\in \mathcal{D}'_{\Gamma}(U\backslash I)$. Then we can apply Theorem \ref{thm_3_w_homogeneous_extension_conormal} and, setting $\mathcal{R}(u_1u_2)=\overline{u_1u_2}$, we conclude. 
\end{proof}

We mention how the results given above do generalize to manifolds. It is well known (see for instance Theorems 2.2.1, 2.2.4 in \cite{hormanderI}) that distributions can be localized; in the case of the manifold $M$, $t\in \mathcal{D}'(U\backslash I)$ for some open subset $U\subset M$, we can localize $t$ in the slice charts, then since the spaces $E_{s,I,\Gamma}(U\backslash I),\ E_{s,I,\Gamma}(U\backslash I)$ are diffeomorphism invariant we can represent the localized distributions in open subsets of $\mathbb{R}^n$ and extend them by means of Theorems \ref{thm_3_w_homogeneous_extension_1}, \ref{thm_3_w_homogeneous_extension_2}, \ref{thm_3_w_homogeneous_extension_conormal}. Then we can use again diffeomorphism invariance to pull back the extended distributions and the gluing property (\textit{c.f.} $(i)$ in Proposition \ref{prop_3_w_hom_diff_invariance}) to obtain a global extended distribution $\overline{t} \in \mathcal{D}'(U)$. \\

We finish this section with an important example that will be used in the sequel for the existence of time ordered products.

\begin{lemma}\label{lemma_Feynman_conormal_landing}
    Let $(M,h)$ be a background geometry; then 
    \begin{itemize}
        \item[$(i)$] any power of the Feynman propagator $H_F[M,h]^k$, 
        seen as a distribution in $M^2\backslash\Delta_2(M)$, satisfies the conormal landing condition; therefore, any extension $\overline{H}_F[M,h]^k$ has
    $$
        \mathrm{WF}\big(\overline{H}_F[M,h]^k\big)\vert_{\Delta_2(M)}\subseteq N^*(\Delta_2(M)).
    $$
        \item[$(ii)$] if $h_s$ is a compactly supported variation of the background geometry (M,h) with parameter $s\in \mathbb{R}^d$ and $H_F[M,h_s]^k$ is the associated distribution in $\mathbb{R}^d\times(M^2\backslash\Delta_2(M))$, then it satisfies the conormal landing condition; therefore, any extension $\overline{H}_F[M,h_s]^k$ has
    $$
        \mathrm{WF}\big(\overline{H}_F[M,h_s]^k\big)\vert_{\mathbb{R}^d\times\Delta_2(M)}\subseteq N^*(\mathbb{R}^d\times\Delta_2(M)).
    $$
    \end{itemize}
\end{lemma}

\begin{proof}

For $(x,y)\in\mathcal{O}$ the Feynman propagator can always be expressed as 
\begin{equation}
    H_F(x,y)= H(x,y) + \frac{i}{2}(\mathcal{G}^+(x,y)+ \mathcal{G}^-(x,y)),
\end{equation}
where $H$ is the Hadamard parametrix (see \eqref{eq_2_Hadamard_distribution}) and $\mathcal{G}^{\pm}$ the retarded and advanced propagators. Equivalently, we can write
\begin{equation}
    H(x,y)= \frac{U(x,y)}{\sigma(x,y)+i\epsilon t(x,y)+\epsilon^2}+ {V}(x,y)\ln\bigg(\frac{\sigma(x,y)+i\epsilon t(x,y)+\epsilon^2}{\mu^2}\bigg)+{W}(x,y)
\end{equation}
where $t$ is the temporal function associated to the globally hyperbolic metric $g$ and $\mu$ an appropriate length scale. Finally, recall that its wave front set is given by
\begin{equation}
\begin{aligned}
    \Gamma_F\equiv\mathrm{WF}(H_F)&= \big\{(x,y;\xi,-\eta) \in \dot{T}^*M^2: (x,\xi)\sim (y,\eta), \ \xi\in \overline{V^{\pm}_g(x)} \ \mathrm{if} \ y\in J^{\pm}(x)   \big\}\\
    & \quad\cup \big\{(x,x;\xi,-\xi) \in \dot{T}^*M^2 : (x,\xi)\in \dot{T}^*M\big\}
\end{aligned}
\end{equation}
We stress that the second part of the wave front set is the one responsible for the impossibility of directly using powers of $H_F\vert{M^2\backslash \Delta_2(M)}$ as distributions. However, if we consider $H_F$ as a distribution on $M^2\backslash \Delta_2(M)$, then we can actually define powers of $H_F$ and 
\begin{equation}\label{eq_WF_powers_of_H_F}
\begin{aligned}
    &\mathrm{WF}\big((H_F\vert_{M^2\backslash \Delta_2(M)})^k\big) \subseteq \bigcup_{j=1}^k \Upsilon_j\times (M^{k-j}\times \{0\});\\
    & \Upsilon_j= \Big\{(x,y;\xi,-\eta) \in \dot{T}^*M^2: \xi=\sum_{i=1}^j\xi_i, \ \eta=\sum_{i=1}^j\eta_i, \ (x,\xi_i)\sim (y,\eta_i), \ \xi\in \bar{V^{\pm}_g(x)} \ \mathrm{if} \ y\in J^{\pm}(x)   \Big\} .
\end{aligned}
\end{equation}
Therefore, if $y\to x$, we get that 
$$
    \overline{\mathrm{WF}\big((H_F\vert_{M^2\backslash \Delta_2(M)})^k\big)}\vert_{\Delta_2(M)}\cap T^{*}\Delta_2(M)\subseteq N^*\Delta_2(M).
$$
Taking into account that $H_F[M,h]$ scales almost homogeneously with degree $n-2$, both under physical scaling of parameters and under transversal coordinate scaling, we can show that it is a weakly homogeneous distribution in $E_{n-2,\Gamma_F}(\mathcal{O})$. Thus we can apply \Cref{thm_3_renormalized_products}, which guarantees the existence of an extension to the diagonal of products of $H_F$ and bounds the wave front set of the extension along the diagonal by $\mathrm{WF}(H_F^k)\subset N^*\Delta_2(M)$.\\

To show $(ii)$, let $\mathbb{R}^d \ni s \to h_s$ be the compactly supported variation of $h=h_0$. Calling $\widetilde{H}_F$ the associated distribution in $\mathbb{R}^d\times \mathcal{O}$, we can estimate
\begin{equation}\label{eq_WF_H_F(s)_outside_diagonal} 
\begin{aligned}
    \mathrm{WF}&\big(\widetilde{H}_F\vert_{\mathbb{R}^d\times ( M^2  \backslash\Delta_2(M))}\big)\subset   \big\{(s,x,x;\zeta,\xi,-\xi) \in {T}^*(\mathbb{R}^d \times M^2): (x,\xi)\in \dot{T}^*M  \big\}\\
    &\quad \cup \big\{(s,x,y;\zeta,\xi,-\eta) \in \dot{T}^*(\mathbb{R}^d\times M^2): (x,\xi)\sim_{g_s} (y,\eta), \ \xi\in \overline{V^{\pm}_{g_s}(x)} \ \mathrm{if} \ y\in J^{\pm}_{g_s}(x)   \big\}\\
\end{aligned}
\end{equation}
we want to prove something more:
\begin{equation}\label{eq_WF_H_F(s)_inside_diagonal} 
\begin{aligned}
    \mathrm{WF}\big(\widetilde{H}_F\vert_{\mathbb{R}^d\times \Delta_2(M)}\big)= \big\{(s,x,x;0,\xi,-\xi) \in {T}^*(\mathbb{R}^d\times M^2): (x,\xi)\in \dot{T}^*M  \big\}.
\end{aligned}
\end{equation}
To do so, we use the notations of \Cref{section_Hadamard_parametrix} and study the problem in normal coordinates $(x,\dot{x})$, where if $y \in M$ is sufficiently close to $x$, we can write $\dot{x}=\bar{e}_x(y)$ where $\bar{e}_x : U\subset M \to V \subset T_xM$ is the Riemannian exponential of the metric $g_s$. Locally, we can write the geometric part of the Feynman propagator as
$$
    \widetilde{H}_F(s,x,\dot{x}) ={H}_F(s,x,e_x(\dot{x}))= \frac{U(s,x,\dot{x})}{\sigma(s,x,\dot{x})}+V(s,x,\dot{x})\ln(\sigma(s,x,\dot{x}))+W(s,x,\dot{x})
$$
where 
$$
\begin{aligned}
    &\sigma(s,x,\dot{x})= g_{\alpha\beta}(s,x)\dot{x}^{\alpha}\dot{x}^{\beta},\ U(s,x,\dot{x})= \bigg(\frac{\mathrm{det}(g(s))|_x}{\mathrm{det}(g(s))|_{\dot{x}}}\bigg)^{1/4},\\
    & V(s,x,\dot{x})= \sum_{j\geq 0} V_j(s,x,\dot{x})\sigma^j(s,x,\dot{x}), , \ V_{-1}(s,x,\dot{x}) \equiv  U(s,x,\dot{x})\\
    & V_{j+1}(s,x,\dot{x})=-\frac{1}{4} U(s,x,\dot{x})\int_0^1 v^{j+1}\frac{PV_j(s,x,v\dot{x})}{U(s,x,v\dot{x})}dv\\
    & W(s,x,\dot{x})= \sum_{j\geq 1} W_j(s,x,\dot{x})\sigma^j(s,x,\dot{x}),\ W_0(s,x,\dot{x}) \equiv  0\\
    & W_{j}(s,x,\dot x)=\frac{U(s,x,\dot x)}{4j}\int_0^1 v^j \frac{PV_{j-1}(s,x,v\dot x^{\mu})-jPW_{j-1}(s,x,v\dot x^{\mu})}{U(s,x,v\dot x^{\mu})}dv \\
    & \qquad \qquad  \quad \ \ -jU(s,x,\dot x)\int_0^1 v^j \frac{V_{j}(s,x,v\dot x^{\mu})}{U(s,x,v\dot x^{\mu})}dv.
\end{aligned}
$$

We claim that $\widetilde{H}_F(s,x,\lambda\dot{x})$ belongs to $E_{2,N^*(\mathbb{R}^d\times \Delta_2(M))}\big(\mathbb{R}^d\times (\mathcal{O}\backslash \Delta_2(M))\big)$. Writing this distribution in integral notation explicitly
$$
    \lambda^{-2} \frac{U(s,x,\lambda\dot{x})}{\sigma(s,x,\dot{x})}+V(s,x,\lambda\dot{x})\ln(   \lambda^2\sigma(s,x,\dot{x}))+W(s,x,\lambda\dot{x}).
$$
We see that each $\lambda$ inside $\sigma$ factors out, and the remaining terms depending on $\lambda$ are jointly smooth functions in all their variables including $\lambda\in [0,1]$. Thus $\lambda^2\langle F_{\lambda}\widetilde{H}_F, f\rangle$ is bounded for $\lambda\in [0,1]$. Next we have to show that
$$
    \sup_{\lambda \in (0,1]} \left\Vert \lambda^{2}\big(F(\lambda)^X\big)^*\widetilde{H}_F\right\Vert_{\chi,k,\Upsilon}\leq C
$$
for all $\chi\in \mathcal{D}\big(\mathbb{R}^d\times (M^2\backslash\Delta_2(M))\big), \ q\in \mathbb{N}, \ \Upsilon\subset \mathbb{R}^n, \ \Upsilon'\subset \mathbb{R}^n$ such that $\mathrm{WF}(\widetilde{H}_F)\cap \mathrm{supp}(\chi)\times \Upsilon =\emptyset$. Consider
$$
\begin{aligned}
    &\int dxd\dot x \lambda^2 \widetilde{H}_F(s,x,\lambda \dot x) \chi(s, x,\dot x) e^{i(\zeta s + \xi x+ \eta \dot x)}\\
    &=\int dxd\dot x \Big\{  \frac{U(s,x,\lambda\dot{x})}{\sigma(s,x,\dot{x})}+ \lambda^2\ln{\lambda}\Big( \sum_{j\geq 0} \lambda^{2j} V_j(s,x,\lambda\dot{x})\sigma^j(s,x,\dot{x})\Big) \\ 
    & \qquad \qquad+ \lambda^2 \Big( \sum_{j\geq 1} \lambda^{2j}W_j(s,x,\lambda\dot{x})\sigma^j(s,x,\dot{x}) \Big)\Big\} \chi(s, x,\dot x) e^{i(\zeta s + \xi x+ \eta \dot x)}\\
%    &= I_1^{\lambda} + I_2^{\lambda}+I_3^{\lambda}.
\end{aligned}
$$
Noticing that derivatives of $\dot{x}$ on $U$ can be written as
$$
    \frac{\partial}{\partial \dot{x}^{\mu}} U(s,x,\lambda\dot{x})= \frac{1}{2}  \frac{\partial}{\partial \dot{x}^{\mu}} \ln(g(s,\lambda\dot x))U(s,x,\lambda \dot x)=\lambda U(s,x,\lambda \dot x)  \frac{\Delta_{\mu}(s,\lambda \dot x) }{\Delta(s,\lambda \dot x)},
$$
where $\Delta_{\mu }=\nabla_{\mu}\Delta$. Thus $ \frac{\partial}{\partial \dot{x}^{\mu}} U(s,x,\lambda\dot{x})$ is of the form $\lambda A^{(1)}_{\mu}$ for a suitable function $A^{(1)}_{\mu}$ jointly smooth in $(s,\lambda,x,\dot{x})$. Taking repeatedly derivatives of $U$ in $\dot x$ yields  
\begin{equation}\label{eq_3_scaled_U_derivatives}
    \frac{\partial^l}{\partial \dot{x}^{\mu_1}\cdots \partial \dot{x}^{\mu_l}} U(s,x,\lambda\dot{x})= \lambda^l A^{(l)}_{\mu}(\lambda,s,x,\dot x).
\end{equation}
The $V_j, \ W_j$ factors can be treated analogously, for instance
$$
\begin{aligned}
     \lambda^{2{j}}\frac{\partial}{\partial \dot{x}^{\mu}} V_{j}(s,x,\lambda\dot x) &=   \frac{1}{2}  \frac{\partial}{\partial \dot{x}^{\mu}} \ln(g(s,\lambda\dot x)) V_j(s,x,\lambda\dot x)\\& \quad- \frac{1}{4} U(s,x,\lambda\dot{x}) \lambda^j\int_0^{\lambda}\bigg[v^{j}\frac{\partial}{\partial \dot{x}^{\mu}}\frac{PV_{j-1}(s,x,v\dot{x})}{U(s,x,v \dot{x})}\bigg]dv    \\
     &= \lambda B^{(1)}_{\mu}(\lambda,s,x,\dot x) U(s,x,\lambda\dot{x}) .
\end{aligned}
$$
where $ B^{(1)}_{\mu}$ is jointly smooth in its arguments. Then
\begin{equation}\label{eq_3_scaled_V_derivatives}
    \frac{\partial^l}{\partial \dot{x}^{\mu_1}\cdots \partial \dot{x}^{\mu_l}} V_j(s,x,\lambda\dot{x})= \lambda^l B^{(l)}_{\mu}(\lambda,s,x,\dot x).
\end{equation}
Applying a similar argument to $W_j$ yields 
\begin{equation}\label{eq_3_scaled_W_derivatives}
    \frac{\partial^l}{\partial \dot{x}^{\mu_1}\cdots \partial \dot{x}^{\mu_l}} W_j(s,x,\lambda\dot{x})= \lambda^l C^{(l)}_{\mu}(\lambda,s,x,\dot x).
\end{equation}
Using \eqref{eq_3_scaled_U_derivatives}, \eqref{eq_3_scaled_V_derivatives}, \eqref{eq_3_scaled_W_derivatives}, we can integrate by parts terms in 
$$
\begin{aligned}
    &\int dxd\dot x \Big\{  \frac{U(s,x,\lambda\dot{x})}{\sigma(s,x,\dot{x})}+ \lambda^2\ln{\lambda}\Big( \sum_{j\geq 0} \lambda^{2j} V_j(s,x,\lambda\dot{x})\sigma^j(s,x,\dot{x})\Big) \\ 
    & \qquad + \lambda^2 \Big( \sum_{j\geq 1} \lambda^{2j}W_j(s,x,\lambda\dot{x})\sigma^j(s,x,\dot{x}) \Big)\Big\} \chi(s, x,\dot x) e^{i(\zeta s + \xi x+ \eta \dot x)}
\end{aligned}
$$
so as to make 
$$
    \int dxd\dot x \lambda^2 (1+|\zeta|+|\xi|+|\eta|)^k\mathcal{F}\big(\widetilde{H}_F\chi\big)(\zeta,\xi,\eta) 
$$
bounded, for all $\lambda\in [0,1]$, whenever $(\zeta,\xi,\eta)\in \Upsilon$ and $\mathrm{supp}(\chi)\times \Upsilon \cap \mathrm{WF}(\widetilde H_F)=\emptyset$. Since the spacetime dimension is $4>2$ the extension of $\widetilde{H}_F|_{\mathbb{R}^d\times (M^2\backslash\Delta_2(M))} $, which is unique and coincides with $\widetilde{H}_F$ itself, must have, by Theorem \ref{thm_3_w_homogeneous_extension_1}, wave front contained in 
$$
	\mathrm{WF}(\widetilde{H}_F)\subset \mathrm{WF}\big(\widetilde{H}_F|_{\mathbb{R}^d\times (M^2\backslash\Delta_2(M))}\big) \cup N^{*}(\mathbb{R}^d\times \Delta_2(M)) \cup \Xi.
$$
In our case,
$$
	\Xi= \{ (s,x,0;\zeta,\xi,\eta) \in \dot{T}^{*}(M^2\backslash\Delta_2(M)) : (s,x,\dot{x};\zeta,\xi,0)\in \mathrm{WF}\big(\widetilde{H}_F|_{\mathbb{R}^d\times (M^2\backslash\Delta_2(M))}\big) \}.
$$
There is, however, no element $(s,x,\dot{x};\zeta,\xi,0)$ in $\mathrm{WF}\big(\widetilde{H}_F|_{\mathbb{R}^d\times (M^2\backslash\Delta_2(M))}\big)$, thus $\Xi=\emptyset$.
Finally it's easy to verify that products of $\widetilde H_F$ are well defined away from the diagonal, and satisfy the conormal landing condition (\textit{c.f.} \eqref{eq_3_conormal_landing_condition}); then by Theorem \ref{thm_3_renormalized_products} we can extend such products to $\mathbb{R}^d\times \Delta_2(M)$ so that they satisfy again \eqref{eq_3_conormal_landing_condition}. 
\end{proof}

\section{Time ordered products}\label{section_TO}

In this section we shall define time ordered products on the algebra of quantum fields $\big(\mathfrak{A}_{\mu c},\star\big)$. In Definition \ref{def_3_TO_products} we modify the definition given in \cite{hollands2002existence, hollands2005conservation} by removing the analytic requirement and substituting it with a generalization of the parameterized microlocal spectrum condition in \cite{khavkine2016analytic, khavkine2019wick}, moreover, instead of working with the algebras generated by Wick polynomials of quantum fields, we state our axioms for general elements of the algebra $\mathfrak A_{\mu c}$. In Theorem \ref{thm_3_TO_renormalization_1}, \ref{thm_3_TO_renormalization_2}, we modify the proof given originally in \cite{hollands2002existence} to account for the modified axioms. Finally, in Theorem \ref{thm_2_uniqueness_TO} and Corollary \ref{coro_2_TO_uniqueness} we show that the new axioms introduced produce the same characterization of time ordered products already known in the literature (\textit{c.f.} \cite{hollands2001local, BF97, dutsch2019classical}).\\

The time ordered product can be seen as another product $\cdot_T$, other than the $\star$ product, on the quantum algebra $\mathfrak{A}_{\mu c}(M,h)$. We can consider a simpler situation by restricting to the algebra of regular functionals $\mathcal{F}_{reg}(M,h)$. We require that 
\begin{equation}\label{eq_3_product_ansaz_1_TO}
F \cdot_T G=
\begin{cases}
    F \star G \quad \mathrm{supp}(F)>_{\Sigma} \mathrm{supp}(G),\\
    G \star F \quad \mathrm{supp}(G)>_{\Sigma} \mathrm{supp}(F);
\end{cases}    
\end{equation}

where with the notation $\mathrm{supp}(F)>_{\Sigma} \mathrm{supp}(G)$ we mean that there exists a Cauchy hypersurface separating the supports of $F$ and $G$ and $\mathrm{supp}(F)\subseteq J^{+}\big( \mathrm{supp}(G)\big)$. Similarly to what we did with the star product, we postulate
$$
    F \cdot_T G= \sum_j \hbar^j \Pi_j(F,G).
$$
In particular, if we take $f_1,f_2\in \mathcal{D}(M)$ with $\mathrm{supp}(f)|_{\Sigma} \mathrm{supp}(g)$, we obtain the consistency conditions
$$
\begin{aligned}
    \Pi_0\big(\phi_{(M,h)}(f_1),\phi_{(M,h)}(f_2)\big)&=\phi_{(M,h)}(f_1) \cdot \phi_{(M,h)}(f_2)\\
    \Pi_1\big(\phi_{(M,h)}(f_1),\phi_{(M,h)}(f_2)\big)&= \frac{i\hbar}{2}\mathcal{G}_{(M,h)}(f_1,f_2)=\frac{i\hbar}{2}\mathcal{G}_{(M,h)}^+(f_1,f_2)
\end{aligned}
$$
\begin{lemma}\label{lemma_3_T_product_regular_algebra}
    Consider the algebra of regular functionals $\mathcal{F}_{reg}(M,h)$ with the topology of strong convenient convergence generated by seminorms \eqref{eq_1_strong_conv_seminorms_1}. Then the time ordered product is defined, at each order of $\hbar$, by
    \begin{equation}\label{eq_2_regular_T_product}
	\big(F \cdot_T G\big) (\varphi)=F(\varphi)G(\varphi) +\sum_{j\geq 1}\frac{\hbar^j}{j!} \bigg(\frac{i}{2}\bigg)^j\Big\langle d^jF[\varphi], D_{(M,h)}^j \big(d^jG[\varphi]\big) \Big\rangle ,
    \end{equation}
    where $D_{(M,h)}=\mathcal{G}_{(M,h)}^++\mathcal{G}_{(M,h)}^-$ is called \textit{Dirac's propagator}. Alternatively, we can write 
    $$
        F\cdot_T G= e^{\frac{i}{2}\widetilde{\Gamma}_{D}}\Big( e^{-\frac{i}{2}\widetilde{\Gamma}_{D}}(F)\cdot e^{-\frac{i}{2}\widetilde{\Gamma}_{D}}(G)\Big).
    $$
    where
    $$
        \widetilde{\Gamma}_D=\frac{\hbar}{2}\int_{M^2}d\mu_g(x,y) D_{(M,h)}(x,y) \frac{\partial^2}{\partial \varphi(x)\partial \varphi(y)},\quad \widetilde{\Gamma}_{D}(F)(\varphi)= \hbar \big\langle d^2F[\varphi],D_{(M,h)}\big\rangle,
    $$
\end{lemma}
\
The proof is completely analogous to that of Lemmas \ref{lemma_2_Weyl_regular_algebra} and \ref{lemma_2_star_to_star_H_product} with $D_{(M,h)}$ substituted in place of $\mathcal{G}_{(M,h)}$.\\

In analogy with what we did in Section \ref{section_Wick_powers} we seek to extend this product to microlocal functionals, \textit{i.e.} those functional supported along the diagonal to which physically relevant functionals belong. Then again, powers of $D_{(M,h)}$ are not well defined, however this time powers of \textit{Feynman propagator} $H_F\doteq H+\frac{i}{2}D_{(M,h)}$ are not well defined ether due to the lingering diagonal part $\{(x,x;-\xi,\xi)\in \dot{ T}^*M^2 \}$ of the wave front set of $D_{(M,h)}$. Thus, to define the extension of $\cdot_T$  to the algebra $\mathfrak A_{\mu loc}(M,h)$, one needs a locally covariant prescription to extend the powers of $H_F$ seen as a distribution outside the diagonal $\Delta_2(M)$ of $M^2$. Finally, we note that the problem of extending distributions has not always a unique answer (see \textit{e.g.} Theorem \ref{thm_3_w_homogeneous_extension_conormal}), we therefore expect that the prescription for defining Time ordered products will not be unique and therefore would need to undergo some sort of characterization.\\ 

Motivated by this, we begin by giving an axiomatic definition with the properties that the time ordered product must satisfy. Our definition differs from that given in \cite{hollands2002existence} or \cite{hollands2005conservation} by the introduction of a different microlocal spectral condition. 

\begin{definition}{\textbf{Time ordered products}}\label{def_3_TO_products}
A family of elements $\{T_p \}_{p\in \mathbb{N}}$ such that for each $(M,h)\in \mathfrak{Bckg}$ defines linear mappings
\begin{equation}\label{eq_abstract_TO_product}
	T_p[M,h]: \otimes^p \mathfrak{A}_{\mu loc}(M,h)\rightarrow \mathfrak{A}_{\mu c}(M,h):(A_1,\ldots,A_p)\mapsto T_p[M,h](A_1,\ldots,A_p),
\end{equation}
with 
\begin{equation}\label{eq_3_H_TO_product}
	\alpha_H\Big(T_p[M,h](A_1,\ldots,A_p)\Big) \equiv  T^H_p\big( \alpha_H\big(A_1\big),\ldots,\alpha_H\big(A_p\big)\big);
\end{equation}
is called a time ordered product if it satisfies the following properties:
\begin{itemize}
\item[$(i)$]\textbf{Locality and Covariance}. Let $\chi:(M,h)\rightarrow (M',h')$ be a causality preserving isometric embedding, then for each $A_1,\ldots,A_p\in \mathfrak{A}_{\mu loc}(M,h)$, we have 
$$
	T_p[M',h']\big(\mathfrak{A}_{\chi}(A_1),\ldots,\mathfrak{A}_{\chi}(A_p)\big) =T_p[M,\chi^{*}h'](A_1,\ldots,A_p);
$$
\item[$(ii)$] \textbf{Causal factorization}. Let $A_1,\ldots,A_p\in \mathfrak{A}_{\mu loc}(M,h)$, suppose there exists a subset $I\subset \{1,\ldots,p\}$ such that for each $i\in I$ $\mathrm{supp}(A_i)\notin J^{-}_M\left(\mathrm{supp}(A_j)\right)$ for all $j\in I^c$, then
$$
	T_p[M,h](A_1,\ldots,A_p)=T_{|I|}\big(\otimes_{i\in I}A_i\big)\star T^H_{|I^c|}\big(\otimes_{j\in I^c}A_j\big);
$$ 
\item[$(iii)$] \textbf{Symmetry}. The time ordered product is symmetric under permutation of its arguments, that is, given $A_1,\ldots,A_p\in \mathfrak{A}_{\mu loc}(M,h)$, and any permutation $\sigma$ of the set $\{1,\ldots,p\}$, 
$$
	T_p[M,h]\big(A_{\sigma(1)},\ldots,A_{\sigma(p)}\big)=T_p[M,h](A_1,\ldots,A_p);
$$
\item[$(iv)$] \textbf{Initial values}. $T_0[M,h]=1_{\mathfrak{A}_{\mu c}(M,h)}$, $T_1[M,h]=id_{\mathfrak{A}_{\mu c}(M,h)}$.
\item[$(v)$] \textbf{Scaling}. If each $A_j\in \mathfrak{A}_{\mu loc}(M,h)$ scales almost homogeneously with degree $\rho_j$ under the rescaling of parameters $h\to S_{\lambda}h$, then $T_p[M,h](A_1,\ldots,A_p)\in \mathfrak{A}_{\mu c}(M,h)$ scales almost homogeneously with scaling degree $\sum_{j=1}^p \rho_j$ under physical scaling.
\item[$(vi)$] \textbf{Field independence}. Given $A_1,\ldots,A_p\in \mathfrak{A}_{\mu loc}(M,h)$ we have
$$
	dT_p[M,h](A_1,\ldots,A_p)[\varphi]=\sum_{j=1}^pT_p[M,h](A_1,\ldots , dA_i[\varphi],\ldots ,\ldots,A_p);
$$
\item[$(vii)$] \textbf{Unitarity}. If $A_1,\ldots,A_p\in \mathfrak{A}_{\mu loc}(M,h)$, then
$$
	T_p[M,h](A_1,\ldots,A_p)^{*}= \sum_{I_1,\ldots,I_q}(-1)^{p+q}T_{|I_1|}[M,h]\left(\prod_{i\in I_1}A_i\right)\star \cdots\star T_{|I_q|}[M,h]\left(\prod_{j\in I_q}A_j\right),
$$
where ${I_1,\ldots,I_q}$ is a partition of $\{1,\ldots,p\}$, and the sum is understood to be made on all possible such partitions;
\item[$(viii)$] \textbf{Action Ward identity}. Let $A_1,\ldots,A_p\in \mathfrak{A}_{\mu loc}(M,h)$ such that at least one, say $A_1$, has the form $A_1=\alpha_H^{-1}(F_1)$ where 
$$
	F_1(\varphi)=\int_M j^{r}\varphi^{*}(d \theta)
$$
for some $n-1$-form $\theta \in \Omega_{n-1}(J^r(M\times \mathbb{R}))$, then
$$
	T_p[M,h](A_1,\ldots,A_p)=0;
$$
\item[$(ix)$] \textbf{$\varphi$-Locality} Let $A_1,\ldots,A_p \in \mathfrak{A}_{\mu loc}(M,h)$, then denote by $A_i^{[k]}$ the Wick ordering of the $k$th order truncated Taylor series of $F_i\equiv \alpha_H(A_i)$ in $0\in C^{\infty}(M)$, we have 
$$
	T_p[M,h]\left(A_1,\ldots,A_p\right)= T_p[M,h]\left(A_1^{[k_1]},\ldots,A^{[k_p]}_p\right)+O\big(\hbar^{\lfloor \sum_ik_i/2 \rfloor}\big);
$$
\item[$(x)$] \textbf{Microlocal spectrum condition} ($\mu$SC). Let $\Gamma^T_q(M,h)\subset \dot{T}^{*}M^q$ be composed by elements $	(x_1,\ldots,x_q,\xi_1,\ldots,\xi_q)$ with the following property: any point $x_i$ is connected at least to some other point $x_j$ via a lightlike geodesic $\gamma_{i\to j}$ and 
$$
	g^{\sharp}\xi_i= \sum_{e\in \{1,\ldots,q\}} \dot{\gamma}_{e\to i}(x_i) - \sum_{s\in \{1,\ldots,q\}}\dot{\gamma}_{i\to s}(x_i),
$$
where the first sum is taken on all future directed lightlike geodesics starting at some other point $x_e$ and ending at $x_i$, while the second is taken over all future directed lightlike geodesics starting at $x_i$ and ending at some other point $x_s$. Then for each $A_1,\ldots,A_p\in \mathfrak{A}_{\mu loc}(M,h)$ and all $\varphi \in C^{\infty}(M)$, consider the integral kernel associated to $d^k\Big(\alpha_H\big(T_p[M,h](A_1,\ldots,A_p)\big)\Big)[\varphi]$, we require that 
\begin{equation}\label{eq_mulocal_spectrum_condition}
\mathrm{WF}\Big(d^kT^H_p\big(\alpha_H(A_1),\ldots,\alpha_H(A_p)\big)[\varphi]\Big) \subset \Gamma^T_{k}(M,h);
\end{equation}
\item[$(xi)$] \textbf{Parameterized microlocal spectrum condition (P$\mu$SC)}. Given any background geometry $(M,h)$ and any smooth compactly supported variation, $\mathbb{R}\ni s \mapsto h_s\in \Gamma^{\infty}(HM)$, we require that for every $\varphi\in C^{\infty}(M)$ and every $A_1,\ldots,A_p \in \mathfrak{A}_{\mu loc}(M,h)\simeq \mathfrak{A}_{\mu loc}(M,h_s)$, we have 
$$
	\begin{aligned}
		& \mathrm{WF}\Big(d^kT^{H_s}_p\big(\alpha_{H_s}(A_1),\ldots,\alpha_{H_s}(A_p)\big)[\varphi]\Big) \\ &\subset \big\{ (s,x_1,\ldots,x_k,\zeta,\xi_1,\ldots, \xi_k) : \  (x_1,\ldots,x_n,\xi_1,\ldots, \xi_p )\in \Gamma^T_{k}(M,h_s)\big\}.
	\end{aligned}
$$
\end{itemize}
\end{definition}

We remark that a key feature of Definition \ref{def_3_TO_products} is the fact that time ordered products have an abstract characterization given by \eqref{eq_abstract_TO_product} while their "functional" form given by \eqref{eq_3_H_TO_product}. The two can be related by the following commutative diagram:
\begin{center}
\begin{tikzcd}
		 & \mathfrak{A}_{\mu loc}(M,h) \arrow[d,"\otimes^{p}\alpha_H"] \arrow[r,"T_p"]  & \mathfrak{A}_{\mu c}(M,h)  \arrow[d, "\alpha_H"] \\
		 & \otimes^p\mathcal{F}_{\mu loc}(M,h,H) \arrow[r,"T_p^H"] & \mathcal{F}_{\mu c}(M,h,H) 
\end{tikzcd}
\end{center}	
In light of this correspondence, the field independence condition has to be imposed at the level of functionals and then pulled back to the algebra level by the map $\alpha_H^{-1}$; this translates to requiring
\begin{equation}
	dT_p(A_1,\ldots,A_p)[M,h][\varphi] \doteq \alpha_{H}^{-1} \Big( dT_p^{H}\big(\alpha_{H}^{-1}(A_1),\ldots,\alpha_{H}^{-1}(A_p) \big)[\varphi](\psi) \Big),
\end{equation}
which combined with the Field Independence property yields 
$$
	dT_p^{H}(F_1,\ldots,F_p)[\varphi](\psi)=\sum_{i=1}^p T_p^{H}(F_1,\ldots,dF_i[\cdot](\psi),\ldots,F_p)(\varphi).
$$
We stress that the above relation is well defined since if $F\in \mathcal{F}_{loc}(M,h,H)$ then $dF[\cdot](\psi)\in \mathcal{F}_{loc}(M,h,H)$ for each $\psi \in\mathcal{D}(M)$. Furthermore we stress that each $T_p(M,h)\in \mathfrak{A}_{\mu c}(M,h)$ is $\mathrm{Had}(M,h)$-equivariant according to \eqref{eq_2_abstract_muc_algebra}, therefore 
\begin{equation}\label{eq_3_time_ordered_had_equivariance}
    T^H_p(F_1,\ldots,F_p) = T^{H'}_p\big(\alpha_{H'-H}(F_1),\ldots,\alpha_{H'-H}(F_p)\big).
\end{equation}
\\

Therefore by field independence and Faà di Bruno's formula we can write the integral kernel of $d^kT^H_p\big(F_1,\ldots,F_p\big)[\varphi]$ as a sum of terms of the form
$$
	T^H_p\big(d^{k_i}F_1[\varphi],\ldots,d^{k_p}F_p[\varphi]\big)(x_1,\ldots,x_k);
$$
then, by locality of $F_i=\alpha_H(A_i)$, smearing with diagonal delta yields
\begin{equation}\label{eq_3_mulocal_spectrum_condition_2}
	\mathrm{WF}\bigg(T^H_p\big(d^{k_i}\alpha_H(A_1)[\varphi],\ldots,d^{k_p}\alpha_H(A_p)[\varphi]\big)\bigg) \subset \Gamma^T_{\beta(k,p)}(M,h)
\end{equation}
where $1\leq \beta(k_1,\ldots,k_p)\leq p$ is the number of $k_i$ which are different from zero.\\ 

Finally note that in the case $p=1$ Definition \ref{def_3_TO_products} reduces to Definition \ref{def_2_Wick_quantum_powers}, whose existence has already been established in Theorem \ref{thm_2_existence_Wick_polynomials}, whereas the case $p=2$, setting $T_2^H(F_1,F_2)=F_1\cdot_{T}F_2$ is consistent with the requirements of \eqref{eq_3_product_ansaz_1_TO}

\subsection{Existence of time ordered products}\label{section_TO_existence}

We turn to show the existence of time ordered products. The proof will proceed by steps and closely follows the one given in \cite{hollands2002existence}. We remark that the difference will be the absence of the condition on the analytic wave front set, given in $\S$2 of \cite{hollands2002existence}, replaced by the parameterized microlocal spectrum condition: $(xi)$ in Definition \ref{def_3_TO_products}.\\

The first step is to state and prove the so-called Wick expansion for time ordered products, which implies that, up to some order in $\hbar$ the time ordered products can be completely characterized by a sum of products of distributions and Wick powers. 

\begin{lemma}\label{lemma_3_Wick_expansion}
Given any time ordered product satisfying the algebraic and parameterized microlocal spectrum conditions and any quasifree Hadamard state $H$, there exists a geodesically convex set $\Omega\subset M$ and functionals $A_1,\ldots, A_p \in \mathfrak{A}_{\mu loc}(M,h)$ such that $\mathrm{supp}(A_i) \subset \Omega$ for each $i=1,\ldots,p$, then   
\begin{equation}\label{eq_3_Wick2}
\begin{aligned}
    T_p[M,h](A_1,\ldots,A_p)=\sum_{\substack{k\leq 2N-1\\j_1+\ldots+j_p= k}} \int_{M^p} &\hbar^{\lfloor |J|/2\rfloor}t_{J}[M,h]\left(f_1^{(j_1)},\ldots,f_p^{(j_p)} \right)(x_1,\ldots, x_{p})\\
    &\phi^{J}_{(M,h,H)}(x_1,\ldots,x_p)d\mu_g(x_1,\ldots,x_p) + O(\hbar^N)
\end{aligned}
\end{equation}
where $J=(j_1,\ldots,j_p)$, $\lfloor |J|/2\rfloor$ the integer part of $|J|/2$, each $t_{J}[M,h]\left(f_1^{(j_1)},\ldots,f_p^{(j_p)} \right)\in \mathcal{D}'(\Omega^{p})$ and the integral is to be intended in the abstract algebra sense. Moreover, those distribution satisfies additional conditions:
\begin{itemize}
\item[$(i)$] \textbf{Locality and Covariance.} If $\chi :M\rightarrow M'$ is a causally preserving isometric embedding, $\Omega\subset M$ is as above and $f\in \mathcal{D}(\Omega^{p})$, then we can assume $\Omega'=\chi(\Omega)$ to remain causally convex, and 
$$
	t_p[M',h'](\chi_{*}f)=t_p[M,\chi^{*}h'](f).
$$
\item[$(ii)$] \textbf{Scaling.} Each $t_{J}[M,h]\big(f_1^{(j_1)},\ldots,f_p^{(j_p)} \big)$ scales almost homogeneously with degree 
$$
	sd\Big(T_p[M,h]\big(A_1,\ldots,A_p\big)\Big)-\sum_{i=1}^{p}j_i\frac{n-2}{2}.
$$
\item[$(iii)$] \textbf{Microlocal Spectrum Condition.} The distributions $t_{J}[M,h]\big(f_1^{(j_1)},\ldots,f_p^{(j_p)} \big)$ belong to $\mathcal{D}'_{\Gamma^T_p(M,h)}(\Omega^{p})$.

\item[$(iv)$] \textbf{Parameterized Microlocal Spectrum Condition.} If $\mathbb{R}^d\ni s \mapsto h_s \in \Gamma^{\infty}(HM)$ is compactly supported variation of the background geometry $(M,h)$, then 
$$
\begin{aligned}
    &\mathrm{WF}\bigg(t_{J}[M,h_s]\big(f_1^{(j_1)},\ldots,f_p^{(j_p)} \big)\bigg)\\ &\subset \big\{ (s,x_1,\ldots,x_p,\zeta,\xi_1,\ldots, \xi_k) : \  (x_1,\ldots,x_p,\xi_1,\ldots, \xi_p )\in \Gamma^T_{p}(M,h_s)\big\}.
\end{aligned}
$$
\end{itemize}
\end{lemma}

\begin{proof}
Note that given any local functional $F\equiv \alpha_H(A) \in \mathcal{F}_{loc}(M,h,H)$ we can consider its Taylor series in $0$:
$$
\begin{aligned}
	F(\varphi) &= F(0)+\sum_{k\leq N} \frac{1}{k!}d^kF[0](\otimes^k \varphi)+R_N(F)\\
	&= F(0)+\sum_{k\leq N} \frac{1}{k!}\int_{M^k}f^{(k)}[0](x_1,\ldots,x_k)\varphi(x_1)\cdots \varphi(x_k)d\mu_g(x_1,\ldots,x_k) +R_N(F)  \end{aligned}
$$
where each compactly supported distribution $f^{(k)}[0](x_1,\ldots,x_k)$ has support along the diagonal. Using \cite[Theorem 5.2.3]{hormanderI}, by a slight abuse of notation we shall write $f^{(k)}[0](x_1,\ldots,x_k)=f^{(k)}[0](x_1)\delta(x_1,\ldots,x_k)$ where $f^{(k)}[0](x_1)\in C^{\infty}_c(M)$\footnote{The more general form of distributions $f^{(k)}[0](x_1,\ldots,x_k)$ does however contain derivatives, in which case $f^{(k)}[0](x_1,\ldots,x_k)=\sum f^{(k)}[0](x_1)\nabla^{j_1}(x_1)\cdots \nabla^{j_k}(x_k)\delta(x_1,\ldots,x_k)$, however using the Action Ward identity we can bring those derivative outside the time ordered product and later on put them back. The proof of the consistency of the Action Ward Identity axiom with the other renormalization conditions can be found in Proposition 3.1 pp. 21 of \cite{hollands2005conservation}}. Applying the Wick ordering operator $\alpha^{-1}_H$ yields
$$
	\alpha^{-1}_H(F)=F(0)+\sum_{k\leq N}\frac{1}{k!}\int_{M^k}f^{(k)}[0](x_1)\phi_{(M,h,H)}^k(x_1)d\mu_g(x_1) +R_N^H(F).
$$
Next we evaluate the Taylor expansion of the functional form of the time ordered product, taking into account field independence, we can replicate the procedure and Taylor-expand the microcausal functional $T^H_p\big(\alpha_H(A_1),\ldots,\alpha_H(A_p)\big)$ and subsequently applying $\alpha_H^{-1}$. Setting $F_i=\alpha_H(A_i)$, by Lemma 3.2 in \cite{BDF09} we can always assume that all their supports are small enough to be contained in a small neighborhood of the diagonal of $M^p$, otherwise we can use the causal factorization axiom and repeat the subsequent argument for some $T_{p'}$ having $p'<p$. For later convenience we shall assume that this neighborhood of the diagonal has the form $\Omega^p$ with $\Omega$ geodesically convex set for the metric $g$.
$$
\begin{aligned}
	& T^H_p(F_1,\ldots, F_p)(\varphi) \\
	&= \sum_{\substack{k\leq 2N-1\\j_1+\ldots+j_p= k}}\frac{\hbar^{\lfloor |J|/2\rfloor}}{J!}\int_{M^k}T^H_p\Big(F_1^{(j_1)},\ldots,F_p^{(j_p)} \Big)[0](x_1,\ldots,x_{k}) \varphi(x_1)\cdots\varphi(x_k) d\mu_g(x_1,\ldots,x_k)\\ 
	&\quad + \hbar^{N} R^H_{2N-1}\left(T^H_p(F_1,\ldots,F_p)\right)\\
	&= \sum_{\substack{k\leq 2N-1\\j_1+\ldots+j_p= k}}\frac{\hbar^{{\lfloor |J|/2\rfloor}}}{J!}\int_{\Omega^p}t_{J}[M,h](x_1,\ldots,x_p) \varphi^{j_1}(x_1)\cdots\varphi^{j_p}(x_p)d\mu_g(x_1,\ldots,x_p)\\ 
	&\quad + \hbar^{N} R^H_{2N-1}\left({T}^H_p(F_1,\ldots,F_p)\right).
\end{aligned}
$$
Where in the last step we combined the fact that the only nontrivial term of $d^{j_i}F_i[0]$ is, by the field independence and the fact that $T_p(0,A_2,\ldots,A_p)=0$, $f_i^{(j_i)}[0](x_1)\delta(x_1,\ldots,x_{j_i})$. Moreover, taking into account linearity and $\varphi$-locality of the time ordered product, the reminder in above expression is of higher order in $\hbar$. Applying $\alpha_H^{-1}$ to both sides of the above expansion, we arrive at
\begin{equation}\label{eq_t_coeff}
	\begin{aligned}
		& T_p[M,h](A_1,\ldots, A_p)(\varphi)\\
		&=\sum_{\substack{k\leq 2N-1\\j_1+\ldots+j_p= k}}\frac{\hbar^{\lfloor |J|/2\rfloor}}{J!}\int_{\Omega^k}t_{J}[M,h](x_1,\ldots, x_{p}) \phi^{J}_{(M,h,H)}(x_1,\ldots,x_p)d\mu_g(x_1,\ldots,x_p)\\ 
		& \quad+\hbar^N R_{2N}\big({T}_p(A_1,\ldots,A_p)\big),
	\end{aligned}
\end{equation}
where $\phi^{J}_{(M,h,H)}(x_1,\ldots,x_p)$ is the formal integral kernel of the $\mathfrak A_{\mu c}$-valued distribution associated to the multilocal functional $\phi^{j_1}_{(M,h)}(x_1)\cdots \phi^{j_p}_{(M,h)}(x_p)$. The coefficients 
$$
\begin{aligned}
    t_{J}[M,h](x_1,\ldots,x_p)\doteq \int & \frac{1}{J!} T^H_p\left(F_1^{(j_1)},\ldots,F_p^{(j_p)}\right)[0](x_1,X_{J_1},\ldots,x_p,X_{J_p})\\
    &\delta(X_{J_1})\cdots \delta(X_{J_p})d\mu_g(X_{J_1},\ldots,X_{J_p})
\end{aligned}
$$
will then be distributions in $\mathcal{D}'(\Omega^p)$. By \eqref{eq_3_time_ordered_had_equivariance} and the fact that the action of $\alpha_{H'-H}$ leaves each $t_J[M,h]$ invariant, we can assume that each $t_J[M,h]$ does not depend on the choice of the symmetric part $H$ of the Hadamard state. Furthermore, %by Theorem 8.4.12 in \cite{hormanderI}, condition $(x)$ in Definition \ref{def_3_TO_products} and the fact that $\mathrm{WF}(\delta_k)=N^*\Delta_k(M)$. Moreover,
by $(x)$ in Definition \ref{def_3_TO_products}, we have
$$
	\mathrm{WF}\big(t_{J}[M,h]\big) \subset \Gamma^T_{p}(M,h).
$$
A similar reasoning shows that if we consider a compactly supported variation of the background geometry $h_s$, we obtain
$$
\begin{aligned}
    \mathrm{WF}\bigg(t_{J}[M,h_s]\big(f_1^{(j_1)},\ldots,f_p^{(j_p)} \big)\bigg)\subset \big\{ (s,x_1,\ldots,x_p,\zeta,\xi_1,\ldots, \xi_p) : \  (x_1,\ldots,x_p,\xi_1,\ldots, \xi_p )\in \Gamma^T_{p}(M,h_s)\big\}.
\end{aligned}
$$
By locality and covariance of the time ordered products, given $\chi :M\rightarrow M'$ causally preserving isometric embedding, assuming that $\Omega'=\chi(\Omega)$ remains causally convex, we must have
$$
	t_{J}[M',h'](\chi_{*}f)=t_{J}[M,\chi^{*}h'](f),
$$
for all $f \in\mathcal{D}(\Omega^k)$. The scaling properties of $T_p$ and those of $\phi^{k_i}_{(M,h,H)}$ can be used to straightforwardly compute the scaling properties of $t_{J}[M,h]$ under physical scaling.
\end{proof}

The above lemma will also be used later to precisely frame the problem of Epstein-Glaser renormalization, which consists in taking a time ordered product inductively constructed up to the diagonal and then extending it to the diagonal itself. In the next step we shall show that this problem is equivalent to start extend the distributions $t_{J}[M,h]$ defined up to the diagonal of $\Omega^p$ such that conditions $(i)-(iv)$ in \ref{lemma_3_Wick_expansion} remain valid for the extensions. For this we use an inductive construction, on the order $p$ of the product, up to the small diagonal of $M^p$. The techniques which we will 
employ for this construction are already present in the known literature (\textit{c.f.} \cite{BF97, hollands2002existence, hollands2005conservation, BDF09}), we will therefore go through them rapidly to adapt the notation and be specific just on those results concerning directly the newly introduced parameterized microcausal spectrum condition. As mentioned before, we suppose by induction that time ordered products have been defined up to some order $p$ and construct the order $p+1$.

\begin{lemma}\label{lemma_causal_cover}
    Let $I\subset \{1,\ldots,p+1\} $ be nonempty, then $\{U_I\}_{I\subset \{1,\ldots,p+1\}}$, with
    $$
        U_I=\left\{ (x_1,\ldots,x_{p+1}) \in M^{p+1} : \forall i\in I \space\  x_{i}, \notin J^-(x_j) \ \forall j \in I^c\right\},
    $$
    is a cover of $M^{p+1}\backslash \Delta_{p+1}(M)$.
\end{lemma}

\begin{proof}
    Take any $(x_1,\ldots,x_{p+1})\in M^{p+1}\backslash \Delta_{p+1}$, then there will be at least two different points, say $x_1$, $x_2$ with $x_1\neq x_2$. Without loss of generality we can assume that $x_2 \notin J^-(x_1)$ then we have $\{1\} \subset I^c$, $\{2\} \subset I$. Next take $x_3$, either $x_3 \in J^-(x_1)$ or $x_3 \notin J^-(x_1)$. In the first case we add $\{3\}$ to $I^c$, in the second we add $\{3\}$ to $I$. Then we consider $x_4$, if $x_4 \notin J^-(x_j) \forall j \in  I^c $ we add $\{4\}$ to $I$, otherwise we add it to $I^c$. Iterating this procedure we arrive at $(x_1,\ldots,x_{p+1})\in U_I$ for the $I$ constructed. Note that the assumption $x_1\neq x_2$ implies $I \neq \emptyset$.
\end{proof}

Combining $\varphi$-locality with Lemma \ref{lemma_3_Wick_expansion} and the Action Ward Identity (condition $(viii)$ in Definition \ref{def_3_TO_products}), we will henceforth assume that all our local functionals are monomials in the field without derivatives. Using this we show that, up to some order in $\hbar$, the time ordered product $T_{p+1}$ can be uniquely constructed outside the diagonal of $M^{p+1}$ as a $\star$-product of lower order T-product using the causal factorization axiom. Employing \cite[Lemma 2.4]{acftstructure} and linearity of the time ordered products we can localize each $T(A_0,\ldots,A_p)$ inside a small enough neighborhood $\Omega^{p+1}$ of the diagonal in order to apply Lemma \ref{lemma_3_Wick_expansion}. As a result, modulo higher order terms in $\hbar$, we just have to compute the abstract integral kernels 
$$
	T_{p+1}[M,h]\big(\phi_{(M,h,H)}^{k_0},\ldots,\phi_{(M,h,H)}^{k_{p}}\big)(x_0,\ldots,x_{p})
$$
where the above expression is the algebra valued distribution 
$$
    (f_1,\ldots,f_p) \mapsto T_{p+1}[M,h]\big(\phi_{(M,h,H)}^{k_0}(f_1),\ldots,\phi_{(M,h,H)}^{k_{p}}(f_p)\big)
$$
with $(x_0,\ldots,x_{p}) \notin \Delta_{p+1}(\Omega)$. Notice that, by Lemma \ref{lemma_3_Wick_expansion} and property $(iii)$ in Definition \ref{def_2_Wick_quantum_powers}, the above relation is consistent with
$$
\begin{aligned}
	T^0_{p+1}[M,h]\big(\phi_{(M,h,H)}^{k_0}(x_0),\ldots,\phi_{(M,h,H)}^{k_{p}}(x_{p})\big)= \sum_{J\leq K} t_{J}^0[M,h](x_0,\ldots, x_{p})\phi^{j_0}_{(M,h,H)}(x_0)\cdots\phi^{j_{p}}_{(M,h,H)}(x_{p})
\end{aligned}
$$
with $t_{J}^0[M,h] \in \mathcal{D}'\big(\Omega^{p+1}\backslash\Delta_{p+1}(\Omega) \big)$ satisfying
\begin{itemize}
\item[$(i')$] \textbf{Locality and Covariance.} If $\chi :M\rightarrow M'$ is a causally preserving isometric embedding, $\Omega\subset M$ is as above and $f\in \mathcal{D}\big(\Omega^{p+1}\backslash \Delta_{p+1}(\Omega)\big)$, then we can assume $\Omega'=\chi(\Omega)$ to remain causally convex, and 
$$
	t^0_{J}[M',h'](\chi_{*}f)=t^0_{J}[M,\chi^{*}h'](f).
$$
\item[$(ii')$] \textbf{Scaling.} Each $t^0_{J}$ scales almost homogeneously with degree $	\sum_{i=0}^{p}(k_i-j_i)\frac{n-2}{2}$ under physical scaling of the background geometry, that is
$$
	S_{\lambda}t^0_{J}[M,h]=\lambda^{\sum_{i=0}^{p}(k_i-j_i)\frac{n-2}{2}}\big( t^0_{J}[M,h]+\sum_l \ln^l(\lambda) w^0_{l}[M,h]\big)
$$
where each $w^0_{l}[M,h]\in \mathcal{D}'\big(\Omega^{p+1}\backslash \Delta_{p+1}(\Omega)\big)$ scales almost homogeneously with respect to physical scaling.
\item[$(iii')$] \textbf{Microlocal Spectrum Condition.} The distributions $t^0_{J}[M,h]$ belong to $\mathcal{D}'\big(\Omega^{p+1}\backslash\Delta_{p+1}(\Omega) \big)$ and have 
\begin{equation}\label{eq_WF_t^0[M,h]}
    \mathrm{WF}\big(t^0_{J}[M,h]\big)\subset {\Gamma^{T}_{p+1}(M,h)}
\end{equation}
\item[$(iv')$] \textbf{Parameterized Microlocal Spectrum Condition.} If $\mathbb{R}^d\ni s \mapsto h_s \in \Gamma^{\infty}(HM)$ is any compactly supported variation of the background geometry, then 
$t^0_{J}[M,h_s]\in  \mathcal{D}'\big(\Omega^{p+1}\backslash\Delta_{p+1}(\Omega) \big)$ and have
\begin{equation}\label{eq_WF_t^0[M,h_s]}
\begin{aligned}
    &\mathrm{WF}\bigg(t^0_{J}[M,h_s]\big(f_1^{(j_1)},\ldots,f_p^{(j_p)} \big)\bigg)\\
    &\quad \subset \big\{ (s,x_1,\ldots,x_p,\zeta,\xi_1,\ldots, \xi_p) : \  (x_1,\ldots,x_p,\xi_1,\ldots, \xi_p )\in \Gamma^T_{p}(M,h_s)\big\}.
\end{aligned}
\end{equation}
\end{itemize} 

If $(x_0,\ldots,x_{p}) \in U_I\cap\Omega^{p+1}$, set
$$
\begin{aligned}
    	T_{p+1}[M,h]\big(\phi_{(M,h,H)}^{k_0}(x_0),\ldots,\phi_{(M,h,H)}^{k_{p}}(x_{p})\big)\doteq  & T_{|I|}[M,h]\bigg(\bigotimes_{i\in I} \phi_{(M,h,H)}^{k_i}(x_i) \bigg)\\
     &\star T_{|I^c|}[M,h]\bigg(\bigotimes_{j\in I^c} \phi_{(M,h,H)}^{k_j}(x_j) \bigg).
\end{aligned}
$$
For arbitrary points in $\Omega^{p+1}\backslash \Delta_{p+1}(\Omega)$, let $\psi_I$ be a partition of unity of $\Omega^{p+1}\backslash \Delta_{p+1}(\Omega)$ and let
\begin{equation}\label{eq_T_product_off_diag}
\begin{aligned}
	T^0_{p+1}[M,h]\big(\phi_{(M,h,H)}^{k_0}(x_0),\ldots,\phi_{(M,h,H)}^{k_{p}}&(x_{p})\big)\doteq \sum_I  \psi_I(x_0,\ldots,x_{p}) \\
	& T_{|I|}[M,h]\bigg(\bigotimes_{i\in I} \phi_{(M,h,H)}^{k_i}(x_i) \bigg)\star T_{|I^c|}[M,h]\bigg(\bigotimes_{j\in I^c} \phi_{(M,h,H)}^{k_j}(x_j) \bigg).
\end{aligned}
\end{equation}
one can show, for instance see \cite[Proposition 4.2, Theorem 4.3]{BF97}, that the above quantity is independent from the partition of unity used and moreover satisfies all requirements of Definition \ref{def_3_TO_products} as long as we are off diagonal. \eqref{eq_T_product_off_diag} is usually called the \textit{unrenormalized} time ordered product, meaning that the renormalization procedure is aimed at extending this prescription to the diagonal of $\Omega^{p+1}$. Therefore, renormalization of time ordered products is equivalently framed as extending the distributions $t_{J}^0[M,h]$, satisfying properties $(i')-(iv')$ above, to some distributions $t_{J}[M,h]$ which satisfies properties $i)-(iv)$ in Lemma \ref{lemma_3_Wick_expansion}. By the result of Section \ref{section_dang} one could proceed similarly to \cite{BF97} and extend the distribution $t_J[M,h]$ to the diagonal right away (in particular see Theorem 5.2 and 5.3 in \cite{BF97}), however if the scaling degree were too big, the extension would not be unique and would therefore fail to be local due to the presence of a fixed cutoff function localized in a neighborhood of the diagonal. To circumvent this problem we follow \cite{hollands2002existence}. The idea is to take full advantage of the microlocal spectrum condition and expand the generic unextended distribution $t^0_{J}[M,h](x,\cdot)$ in suitable truncated Taylor series of the metric deformation parameter, so that the remainder of this expansion can be extended directly (hence locally and covariantly) while the other terms can be written in relative coordinates about the point $x\in \Omega$ transforming the covariance requirement into an $SO(1,n-1)$-invariance requirement. \\

In the following arguments we shall write $t^0[M,h]$ in place of $t^0_{J}[M,h]$ meaning that the argument ought to be repeated for each $J=(j_1,\ldots,j_k)$ without substantial modifications. We start by noting that if we fix a point $x\equiv x_0\in \Omega$, then there exists normal coordinates $( \{\dot{x}_1^a\},\ldots,\{\dot{x}^a_p\})$ for the point $(x_1,\ldots, x_{p})\in \Omega^{p}$ by setting, for each $1\leq i\leq p$,
\begin{equation}\label{eq_normal_coordinates}
	\dot{x}_{i}^a\equiv \bar{e}_{x_0}(x_i)= -e^{a}_{\mu}(x_i)g^{\mu\nu}(x_i)\nabla^{(i)}_{\nu}\sigma(x_0,x_i),
\end{equation}
where $e^a_{\mu}$ is a vielbein, $\sigma$ the geodesic distance function between its arguments and the covariant derivative is carried out for the local coordinates $\{x_i^{\mu}\}$. We remark that inside any convex normal neighborhood of $x_0$, $\Omega$, the mapping defined above $(e_{x_0})^p:\Omega^p\ni(x_1,\ldots,x_{p}) \mapsto (\dot{x}_{1},\ldots,\dot{x}_p)\in V^p$ is a diffeomorphism into the open subset $V^p\subset T_{x_0}\Omega^p$. The presence of the vielbein 
is key in transforming the covariance requirement for the extension into $SO(1,n-1)$-invariance. Let now $C_{x_0}\doteq \{x_0\}\times (\Omega\backslash \{x_0\})^p$, by combining the microlocal spectrum condition with Theorem 8.2.4. in \cite{hormanderI}, we see that the distributions $t^0_{J}[M,h]$ admits a restriction to $C_{x_0}$, since $\Gamma^T(M,h)\cap N^{*}C_{x_0} =\emptyset$, $t^0_{J}[M,h](x_0,\cdot)$ which is smooth in the first variable. The coordinates $(x,\dot{x}_1,\ldots, \dot{x}_p) $ are suited to identify the diagonal of $M^{p+1}$ as the submanifold $(x,0,\ldots,0)$ setting up the extension problem for the distributions $t^0[M,h]$ in terms of $\S$ \ref{section_dang}.\\

Next, suppose that there is a compact subset $K$ contained in $V$ and let $\chi\in \mathcal{D}(V)$ be a cutoff function on $V$ with support contained in $K$, then consider the compactly supported vector field $X=\sum_iX_i=\sum_{i=1}^p\chi(\dot{x}_{i}^a)\dot{x}_{i}^a\partial/\partial x_i^a$ and denote its logarithmic flow by $F^{X}_{\rho}$. Define a compactly supported variation of $h$ by $h_{\rho}=(g_{\rho},m^2,\kappa)$, $g_{\rho}(\dot{x}_i)=\rho^{-2}\circ (F^{X_i}_{\ln(\rho)})^{*}g(x)$, we obtain $(g_{\rho})_{\mu\nu}(\dot{x}_i)=g_{\mu\nu}(\rho\dot{x}_i)$. We remark that this vector field is Euler according to \eqref{eq_3_Euler_vector_field}, and that $V^p$ is invariant under the scaling induced by each $X_i$, \textit{i.e.} $\lambda V^p\backslash \{0\} \subset V^p\backslash \{0\} $ for all $\lambda \in (0,1]$. Finally, consider $t^0[h_{\rho}](x,f)$ for some $f \in \mathcal{D}\big(V^p \backslash\{0\}\big)$, by condition $(iv)$ in Lemma \ref{lemma_3_Wick_expansion}, the latter is a smooth function $\mathbb{R} \times \Omega \rightarrow \mathbb{R}$. We thus define distributions on $C_{x_0}$ by 
\begin{equation}\label{eq_3_T_tau_development}
	\tau^0_k[M,h](x_0,\cdot) \doteq  \frac{d^k}{d\rho^k}\Big|_{\rho=0} t^0[M,h_{\rho}](x_0,\cdot).
\end{equation}
then 
\begin{align}\label{eq_3_T_development}
	t^0[M,h](x_0,\cdot) &= \sum_{k=0}^N \tau^0_{k}[M,h](x_0\cdot)+ r^0_N[M,h](x_0,\cdot),
\end{align}
\begin{align}\label{eq_3_r_term}
	r^0_N[M,h](x_0,\cdot) &= \frac{1}{N!}\int_0^1 (1-\rho)^N \frac{d^{N+1} t^0[M,h_{\rho}]}{d\rho^{N+1}}(x_0,\cdot) d\rho.
\end{align}

Those distributions enjoy the following properties (Theorem 4.1 in \cite{hollands2002existence}):

\begin{theorem}\label{thm_3_TO_renormalization_1}
Let $x\in \Omega$ be a fixed point, then
\begin{itemize}
% \leavevmode
% \makeatletter
% \@nobreaktrue
% \makeatother
    \item[$(i)$] The distributions $\tau^0_{k}[M,h](x,\cdot)$ and $r^0_N[M,h](x,\cdot)$ are well defined distribution on $V^{p}\backslash \{0\}$, are local and covariant in the sense that for each $\chi:M\rightarrow M'$ causally convex isometric embedding, $\chi^{*}\tau^0_k[M',h'](x,\cdot)=\tau^0_k[M,\chi^{*}h'](x,\cdot)$ and $\chi^{*}r^0_N[M',h'](x,\cdot)=r^0_N[M,\chi^{*}h'](x,\cdot)$. Their wave front set are contained in 
$$
\begin{aligned}
    \bigg\{&(\dot{x}_1,\ldots,\dot{x}_{p},\eta_1,\ldots,\eta_{p})\in \dot{T}^*(V^{p}\backslash \{0\}) \\
    &: \Big(x,e_{x}(\dot{x}_1),\ldots,e_x(\dot{x}_{p}),\xi-\sum_i \frac{\partial \bar{e}_x}{\partial \dot{x}_i}\eta_i ,\frac{\partial \bar{e}_x}{\partial \dot{x}_1}\eta_1,\ldots,\frac{\partial \bar{e}_x}{\partial \dot{x}_p}\eta_{p}\Big)\in \Gamma^T_{p+1}(M,h)\bigg\};
\end{aligned}
$$
    \item[$(ii)$] For any background geometry $(M,h)$ the integral kernel of each $\tau^0_k$ may be represented as
\begin{equation}\label{eq_tau_expression}
    \tau^0_k[M,h](x,\dot{x}_1,\ldots,\dot{x}_p)= \sum_rC^{a_1\ldots a_r}[h](x)e_x^{*}u^0_{a_1\ldots a_r}(\dot{x}_1,\ldots,\dot{x}_p) ,
\end{equation}
where each $C^{a_1\ldots a_r}[h](x)$ is a covariant tensor constructed out of the metric $g$, curvature tensor $R$ and their covariant derivative at $x$ up to some fixed order; $u^0_{a_1,\ldots,a_r}$ are tensor valued $SO(1,n-1)$-equivariant distributions defined everywhere except at the origin of $V^p$;
    \item[$(iii)$] $\tau^0_{k}[M,h]$ and $r^0_N[M,h]$ scale almost homogeneously, under physical scaling, with scaling degree equal to that of $t[M,h]$; we call this number $D$ in the sequel;
    \item[$(iv)$] the distributions $u^0_{a_1\ldots a_r}$ and $r^0_N[M,h](x,\cdot)$ belong respectively to $E_{k-D}( V^p\backslash \{0\})$ and $E_{N-D}(V^p\backslash \{0\})$; that is, given the logarithmic flow $F_{\lambda}^X$ of the Euler vector field $X= \sum_{j=1}^p \sum_a \dot{x}^a_j \frac{\partial}{\partial \dot{x}^a_j}$, there are positive constants $c, \ c'$ such that for all $f\in\mathcal{D}(V^p\backslash \{0\})$
$$
	\sup_{\lambda \in (0,1]} \left\vert \lambda^{k-D} \big(F^X_{\lambda}\big)^{*} u^0(f)\right\vert\leq c, \ \sup_{\lambda \in (0,1]} \left\vert \lambda^{N+1-D} \big(F^X_{\lambda}\big)^*r^0_N[M,h](x,f)\right\vert\leq c'
$$
moreover, $\forall \chi\in \mathcal{D}(V^p\backslash{0}), \ q\in \mathbb{N}, \ \Upsilon\subset \mathbb{R}^n, \ \Upsilon'\subset \mathbb{R}^n$ such that $\mathrm{WF}(u^0)\cap \mathrm{supp}(\chi)\times \Upsilon =\emptyset$, $\mathrm{WF}(r^0_N(x))\cap \mathrm{supp}(\chi)\times \Upsilon'=\emptyset$, there are constants $C,\ C' >0$ such that
$$
     \sup_{\lambda \in (0,1]} \left\Vert \lambda^{k-D} \big(F(\lambda)^X\big)^*u^0\right\Vert_{\chi,q,\Upsilon}\leq C, \ \sup_{\lambda \in (0,1]} \left\Vert \lambda^{N+1-D} \big(F(\lambda)^X\big)^*r^0_N[M,h](x,\cdot)\right\Vert_{\chi,q,\Upsilon'}\leq C',
$$
\end{itemize}
\end{theorem}

\begin{proof}
The arguments used in the proof are essentially those used in Theorem 4.1 of \cite{hollands2002existence} by means of which one can show items $(i)-(iii)$, we just remark that combining \eqref{eq_tau_expression} and the argument in the proof of Theorem 3.1 in \cite{khavkine2016analytic} directly obtain that the coefficients $C^{a_1\cdots a_r}[h](x)$ are of the form claimed in $(ii)$. We are therefore left with $(iv)$. Note that the action of the logarithmic flow $\big(F^X_{\lambda}\big)^*$ of the Euler vector field $X$, satisfies
$$
\begin{aligned}
    &\lambda^{-k} \sum C[h](x)e_x^*u^0(\lambda\dot{x}_1,\ldots,\lambda\dot{x}_p)= \lambda^{-k}\big(F^X_{\lambda}\big)^* \Big( \sum C[h](x)e_x^*u^0(\dot{x}_1,\ldots,\dot{x}_p)\Big)=\\
    &=\lambda^{-k}\big(F^X_{\lambda}\big)^* \Big( \tau^0_k[M,h](x,\cdot)\Big)= \frac{d^k}{d(s/\lambda)^k}\bigg|_{\rho=0}  t^0[M, ((s/\lambda)^{-2} \big(F^X_{s}\big)^*g, (s/\lambda)^2m^2,\kappa)](x,\cdot)\\
    &= \frac{d^k}{d\rho^k}\bigg|_{\rho=0}  t^0[M,(\lambda^{-2} g_s, \lambda^2m^2,\kappa)](x,\cdot) = \tau^0_k[M,S_{1/\lambda}h](x,\cdot)
\end{aligned}
$$
where we used the covariance of $t^0$ in the first and fourth equality to factor in and out the $\big(F^X_{\lambda}\big)^*$ term. Moreover, if $\lambda\in (0,1]$, $1/\lambda \in [1,\infty)$, thus by almost homogeneous scaling of $t^0[M,h]$ under physical scaling, $\tau^0_k[M,S_{1/\lambda}h](x,\cdot)= (1/\lambda)^D \big(\tau^0_k[M,h](x,\cdot) + \sum_i\ln^i(\lambda)w_i[M,h]\big)$, therefore we get that
\begin{equation}\label{eq_w-homogenuity_estimate_tau}
    \lambda^{-r}\sum C[h](x)e_x^*u^0(\lambda\dot{x}_1,\ldots,\lambda\dot{x}_p)= \lambda^{k-D-r} \big(\tau^0_k[M,h](x,\cdot) + \sum_i\ln^i(\lambda)w_i[M,h]\big).
\end{equation}
The quantity $\lambda^{-r} \langle\tau^0_k[M,h](x,\lambda\cdot), f \rangle$ is bounded for every $f\in \in\mathcal{D}(V^p\backslash \{0\})$ whenever $-r+k-D>0$. Similarly, when estimating the other seminorms we can take advantage of \eqref{eq_w-homogenuity_estimate_tau}: given any $\chi, \ l, \ V$ as above, we see that 
$$
\begin{aligned}
    &\left\Vert \lambda^{k-D} \big(F(\lambda)^X\big)^*u^0\right\Vert_{\chi,q,V}\leq C(h) \lambda^{-r} \Big\langle\chi\tau^0_k[M,h](x,\lambda\cdot), e^{i(\sum_j\dot{x}_j\eta_j)} \Big\rangle \\
    &= \lambda^{k-D-r} \Big( \Big\langle\chi \tau^0_k[M,h](x,\cdot) , e^{i(\sum_j\dot{x}_j\eta_j)} \Big\rangle +\sum\ln(\lambda) \Big\langle\chi w_i[M,h], e^{i(\sum_j\dot{x}_j\eta_j)} \Big\rangle\Big).
\end{aligned}
$$
Each of the terms inside pairings is rapidly decreasing along the direction in the cone $V$ by hypothesis and is independent of $\lambda$; then again whenever $-r+k-D>0$ the above quantity is bounded, thus we conclude that the scaled distributions $\tau^0_k[M,h](x,\cdot)$, and therefore each of the $u^0$ belong in $E_{k-D}( V^p\backslash \{0\})$. The argument for the proof for the distributions $r^0[M,h]$ is similar for boundedness in the H\"ormander topology follows from the following calculation:
$$
\begin{aligned}
    &N! r^0_N[M,h](x,\lambda\dot{x}_1,\ldots,\lambda\dot{x}_p)= \int_0^1 (1-\rho)^N \frac{d^{N+1}}{d\rho^{N+1}}\big(F_{\lambda}^X\big)^*t^0[M,h_{\rho}](x,\dot{x}_1,\ldots,\dot{x}_p)d\rho\\
    &=  \int_0^1 (1-\rho)^N \frac{d^{N+1}}{d\rho^{N+1}} t^0\big[M,\big(s^{-2} \big(F_{s}^X\big)^*g,s^2m^2,\kappa\big)\big](x,\dot{x}_1,\ldots,\dot{x}_p)d\rho   \\
    &=\lambda^{N} \int_0^{\lambda} (1-s/\lambda)^N  \frac{d^{N+1}}{d\rho^{N+1}} t^0[M,S_{1/\lambda}h_{\rho}](x,\dot{x}_1,\ldots,\dot{x}_p)d\rho\\
\end{aligned}
$$
finally, using the almost homogeneous scaling of $r^0_{N}$ with respect to the scaling of the background geometry, we obtain 
\begin{equation}\label{eq_w-homogenuity_estimate_r}
\begin{aligned}
    \lambda^{-r}r_{N+1}^0[M,h](x,\lambda\dot{x}_1,\ldots,\lambda\dot{x}_p)&= \lambda^{N-D-r} \Big( \int_0^{\lambda} (1-s/\lambda)^N  \frac{d^{N+1}}{d\rho^{N+1}} t^0[M,h_{\rho}](x,\dot{x}_1,\ldots,\dot{x}_p) \frac{1}{\lambda}d\rho \\ 
    &\quad+ \sum_j\ln(\lambda)^j\int_0^{\lambda} (1-s/\lambda)^N   \frac{d^{N+1}}{d\rho^{N+1}} w^0[M,h_{\rho}](x,\dot{x}_1,\ldots,\dot{x}_p) \frac{1}{\lambda}d\rho\Big).
\end{aligned}
\end{equation}
Notice that the integral contributions give rise to distributions in the variables $(\lambda,x, \dot{x}_1,\ldots,\dot{x}_p)$, which however are smooth in $\lambda$, therefore when calculating
$$
    \sup_{\lambda \in (0,1]} \left\vert \lambda^{N-D} \big(F^X_{\lambda}\big)^*r^0_N[M,h](x,f)\right\vert\leq c' ,\ \sup_{\lambda \in (0,1]} \left\Vert \lambda^{N-D} \big(F(\lambda)^X\big)^*r^0_N[M,h](x,\cdot)\right\Vert_{\chi,q,\Upsilon'}\leq C'
$$
we can estimate the $\lambda$ dependence with appropriate constants. Thus $r^0_{N}[M,h](x,\cdot)\in E_{N-D}(V^p\backslash\{0\})$. 
\end{proof}

This proof allow us to extend the distributions $u^0, \ r^0_{N+1}[M,h] \in \mathcal{D}'(V^p\backslash\{0\})$ to distributions in $\mathcal{D}'(V^p)$ according to \Cref{thm_3_w_homogeneous_extension_1} and \Cref{thm_3_w_homogeneous_extension_2}. In this case we need not a precise estimation of the wave front set of the extension, thus in $(iv)$ of \Cref{thm_3_TO_renormalization_1} we could have just controlled boundedness of $u^0$, $r^0[M,h](x)$ in the standard topology of $\mathcal{D}'(V^p\backslash 0)$, however, the arguments developed there will prove important for the next result.

\begin{theorem}\label{thm_3_TO_renormalization_2}
	Let $\tau^0_k[M,h](x,\cdot)$, $r^0_N[M,h](x,\cdot) \in \mathcal{D}' (V^p\backslash \{0\})$ be the distributions defined as in \eqref{eq_3_T_tau_development} and \eqref{eq_3_T_development}, then there exists extensions $\tau_k[M,h](x,\cdot)$, $r_N[M,h](x,\cdot)$ such that
	$$	
		t_{J}[M,h](x,\cdot) \doteq  \sum_{k=0}^N \tau_{k}[M,h](x,\cdot)+ r_N[M,h](x,\cdot)
	$$
	satisfies all requirements given in Lemma \ref{lemma_3_Wick_expansion}. Moreover, if $\mathbb{R}^d \ni s \mapsto h_s$ is a compactly supported variation of the background geometry, the distributions $\tau^0_k[M,h_s](x,\cdot)$, $r^0_N[M,h_s](x,\cdot)\in \mathcal{D}'\big(\mathbb{R}^d\times (V^p\backslash \{0\})\big)$ admits extensions $\tau_k[M,h_s](x,\cdot)$, $r_N[M,h_s](x,\cdot)$ such that
	$$	
		t_{J}[M,h_s](x,\cdot) \doteq  \sum_{k=0}^N \tau_{k}[M,h_s](x,\cdot)+ r_N[M,h_s](x,\cdot)
	$$
	with $t_{J}[M,h_0]=t_{J}[M,h]$ and
	$$
	\mathrm{WF}(t_{J}[M,h_s])\subset \{(s,x,x_1,\ldots,x_{p},\zeta,\xi_1,\ldots,\xi_{p+1}): \ (x_1,\ldots,x_{p+1},\xi_1,\ldots,\xi_{p+1}) \in \Gamma^T_{p+1}{(M,h_s)}\}.
	$$
\end{theorem}

\begin{proof}
We will proceed in three steps: first we extend the distributions $\tau^0_k[M,h]$, $r_N[M,h]$ and $\tau^0_k[M,h_s]$, $r_N[M,h_s]$ as distributions on $\mathcal{D}'(\Omega \times V^p)$ and $\mathcal{D}'(\mathbb R^d \Omega \times V^p)$ respectively, then in the remaining steps we establish the wave front set of those extensions.\\

\textbf{Step 1.} Each $\tau^0_k[M,h]$ can be extended to the diagonal by extending the distributions $u^0$. Indeed, by property $(iv)$ in Theorem \ref{thm_3_TO_renormalization_1}, the divergence degree of those distributions will be controlled, thus by Theorem \ref{thm_3_TO_renormalization_2} we can find an extension $u$ of $u^0$. Lemma 4.1 in \cite{hollands2002existence} ensures that $u$ can be chosen to be $SO(1,n-1)$-equivariant. As a result, we define the extension $\tau_k[M,h](x,x_1,\ldots,x_p)=\sum C^{a_1,\ldots,a_r}[h](x)e^*_x u_{a_1,\ldots,a_r}(\dot{x}_1,\ldots,\dot{x}_p)$, which with the appropriate choice of $u$ is local and covariant. If we chose $N$ big enough, \textit{i.e.} $N+1-D+n\cdot p>0$. By Theorem \ref{thm_3_w_homogeneous_extension_1} we can extend uniquely the distribution $r^0_N[M,h](x)$ directly by
\begin{equation}\label{eq_r^0_ext}
	r_N[M,h](x)\doteq \lim_{\epsilon\to  0}(1-\chi_{\epsilon^{-1}})r^0_N[M,h](x).
\end{equation}
In particular, uniqueness of the extension will imply locality and covariance. The scaling properties of the extended distributions $r_N[M,h]$ and each $\tau_k[M,h]$ remain the same as those of $t[M,h]$. The extension of the distributions $\tau^0_k[M,h_s]$, $r^0_N[M,h_s]$ works in the same way noticing the following facts
\begin{itemize}
    \item the expansion 
    \begin{align}\label{eq_3_T_development_s_variation}
	t^0[M,h_s](x,\cdot) &= \sum_{k=0}^N \tau^0_{k}[M,h_s](x_0\cdot)+ r^0_N[M,h_s](x,\cdot),
    \end{align}
    where 
    \begin{equation}\label{eq_3_T_tau_development_s_variation}
	\tau^0_k[M,h_s](x,\cdot) =\sum_rC^{a_1\ldots a_r}[h_s](x)e_x(s)^{*}u^0_{a_1\ldots a_r}  ,
    \end{equation}
    \begin{align}\label{eq_3_r_term_s_variation}
	r^0_N[M,h_s](x,\cdot) &= \frac{1}{N!}\int_0^1 (1-\rho)^N \frac{d^{N+1} t^0[M,(h_{s})_\rho]}{d\rho^{N+1}}(x,\cdot) d\rho,
\end{align}
    can be performed in analogy with \eqref{eq_3_T_development} due to the fact that $t^0[M,h_s](x,f)$ is jointly smooth in $(s,x)$ by the induction hypothesis and $(xi)$ in Definition \ref{def_3_TO_products};
    \item the coefficients $C^{a_1\ldots a_r}[h_s](x)$ are jointly smooth in $(s,x)$ and the distributions $r^0_N[M,h_s](x)$, $u^0_{a_1\ldots a_r} $ satisfy bounds 
    $$
	\sup_{\lambda \in (0,1]} \left\vert \lambda^{k-D} \big(F^X_{\lambda}\big)^{*} u^0(f)\right\vert\leq c, \ \sup_{\lambda \in (0,1]} \left\vert \lambda^{N+1-D} \big(F^X_{\lambda}\big)^*r^0_N[M,h_s](x,f)\right\vert\leq c'.
    $$
    for any $f\in \mathcal{D}(V^p\backslash 0)$, repeating the same arguments used for the proof of $(iv)$ in \Cref{thm_3_TO_renormalization_1}.
\end{itemize}
We remark that $r^0_N[M,h_s]|_{s=0}\equiv r^0_N[M,h]$, $\tau_k[M,h_s]|_{s=0}\equiv\tau_k[M,h]$.

\textbf{Step 2.} Next we look at the wave front sets of $r_N[M,h]$ and each $\tau_k[M,h]$. Since those are known off the diagonal we have to estimate just the newly added singularities in $\Delta_{p+1}(M)$. We claim that 
$$
\begin{aligned}
    &{\mathrm{WF}(\tau_k[M,h])}\vert_{\Delta_{p+1}(\Omega)} \subset N^{*}\Delta_{p+1}(\Omega), \\
    &{\mathrm{WF}(r_{N}[M,h])} \vert_{\Delta_{p+1}(\Omega)}\subset N^{*}\Delta_{p+1}(\Omega).
\end{aligned}
$$
Notice that this time we are actually thinking of $\tau_k[M,h] $ as a distribution in $\mathcal{D}'(\Omega^{p+1})$. By construction $\tau_k[M,h](x, x_1,\ldots,x_p)=\sum C^{a_1,\ldots,a_r}[h](x)e^*_{x} u_{a_1,\ldots,a_r}(\dot{x}_1,\ldots,\dot{x}_p)$, consider the (local) diffeomorphism $(x,x_i)\to (x, \bar{e}_{x}(x_i))$, if $u$ has a certain wave front set $\mathrm{WF}(u)$, by \cite[Theorem 8.4.2]{hormanderI} the pull-back distribution $C[h](x)e^*_{x} u(\dot{x}_1,\ldots,\dot{x}_p)$ will have wave front set given by
$$
\begin{aligned}
    \mathrm{WF}\big(C[h](x)e^*_{x} u\big)\subseteq \Big\{ & \Big(x,e_{x}(\dot{x}_1),\ldots,e_{x}(\dot{x}_p);\xi + \sum_{i=1}^p\frac{\partial \bar{e}_{x}(x_i)}{\partial x} \eta_i,\frac{\partial \bar{e}_{x}(x_1)}{\partial x_1}\eta_1,\ldots,\frac{\partial \bar{e}_{x}(x_p)}{\partial x_p}\eta_p\Big) \\ &
    \in \dot{T}^*\Omega^{p+1}: (\dot{x}_1,\ldots,\dot{x}_p;\eta_1\ldots,\eta_p)\in \mathrm{WF}\big(u\big)\Big\}
\end{aligned}
$$
By the microlocal spectrum condition, we know that $C[h]u$ is smooth in the $x$ variable, thus $\xi=0$. Moreover we have the identity
$$
    \frac{\partial \bar{e}^{\mu}_{x}(x_i)}{\partial x^{\nu}}= -\frac{\partial {e}_{x_i}^{\mu}(\dot{x}_i)}{\partial \dot{x_i}^{\alpha}}\frac{\partial \bar{e}_{x_i}^{\alpha}(x)}{\partial x^{\nu}}
$$
which evaluated in the coincidence limit $x_i \to x$ yields
$$
    \frac{\partial \bar{e}^{\mu}_{x}(x_i)}{\partial x^{\nu}}=-\delta^{\mu}_{\alpha} \frac{\partial \bar{e}^{\alpha}_{x}(x_i)}{\partial x_i^{\nu}}.
$$
As a result, we get 
\begin{equation}\label{eq_tau^0_conormal_landing}
    \overline{\mathrm{WF}(\tau_k[M,h])}\vert_{\Delta_{p+1}(\Omega)} \subseteq \Big\{ \Big(x,\ldots,x;- \sum_{i=1}^p\xi_i,\xi_1,\ldots,\xi_p\Big)
    \in \dot{T}^*\Omega^{p+1}: (x_1,\ldots,x_p;\xi_1,\ldots,\xi_p)\in \dot{T}^*\Omega^p\Big\}.
\end{equation}

It remains to show that ${\mathrm{WF}(r_N[M,h])} \vert_{\Delta_{p+1}(\Omega)} \subseteq N^{*}\Delta_{p+1}(\Omega)$. Here we proceed as follows: first we note that $r^0[M,h](x,\cdot)\in \mathcal{D}'(V^p\backslash\{0\})$ 
can be seen as a distribution $r^0[M,h]\in \mathcal{D}'\big(\Omega\times (V^p\backslash\{0\})\big)$ depending smoothly on the first factor. Notice that $r^0[M,h]\in E_{N-D}\big(\Omega\times (V^p\backslash\{0\})\big)$, for the vector field $X=\sum_{i=1}^p \dot{x}_i^a\frac{\partial}{\partial \dot{x}_i^a}$ is an Euler vector field for the submanifold $\Omega\times \{0\}\subset \Omega\times V^p$, therefore we can adapt the proof of Theorem \ref{thm_3_TO_renormalization_1} and get $r^0[M,h]\in E_{N-D}\big(\Omega\times (V^p\backslash\{0\})\big)$. We claim that $r^0[M,h]$ satisfies the conormal landing condition
$$
    \overline{\mathrm{WF}(r^0_{N}[M,h])}\cap T^*\big(\Omega\times V^p)\vert_{\Omega\times \{0\}}\subset N^{*}(\Omega\times \{0\}).
$$
If that was true, since
$$
    \overline{e}_{\cdot}^p:\Omega^{p+1}\ni(x,x_1,\ldots,x_{p}) \mapsto (x,\dot{x}_{1},\ldots,\dot{x}_p)=(x,\overline{e}_{x}(x_1),\ldots,\overline{e}_{x}(x_p))\in \Omega\times V^p
$$
is a diffeomorphism, $(ii)$ in Proposition \ref{prop_3_w_hom_diff_invariance} implies that  
$$
    E_{N-D, N^*(\Omega\times \{0\})}\big(\Omega\times (V^p\backslash\{0\})\big)\simeq E_{N-D, N^*(\Delta_{p+1}(\Omega))}\big(\Omega^{p+1}\backslash\Delta_{p+1}(\Omega))\big),
$$
thus by uniqueness of the extension combined with Theorem \ref{thm_3_w_homogeneous_extension_2} we obtain ${\mathrm{WF}(r_{N}[M,h])} \vert_{\Delta_{p+1}(\Omega)}\subset N^{*}\Delta_{p+1}(\Omega).$\\

We are thus left with showing that $r^0[M,h]$ satisfies the conormal landing condition. Instead of doing this directly, we show that $\overline{\mathrm{WF}t^0[M,h])} \vert_{\Delta_{p+1}(\Omega)} \subseteq N^{*}\Delta_{p+1}(\Omega)$ which together with \eqref{eq_tau^0_conormal_landing} implies our claim. To do so we proceed by induction over $p$. When $p=0$, we are in the case of Wick powers, where the dependence is smooth; so we assume that the condition holds for any $n<p$. By Theorem 4.3 in \cite{BF97} and equation (31) therein, in $ \Omega^{p+1}\backslash \Delta_{p+1}(\Omega)$ we can write 
\begin{equation}\label{eq_distributional_causal_factorization}
    t[M,h](x_0,x_1,\ldots,x_p)=\sum_{I} t^0_I(x_I)t^0_{I^c}(x_{I^c})\prod_{i\in I,j\in I^c} H_F^{a_{ij}}(x_i,x_j)    
\end{equation}
where $I\cup I_c=\{0,\ldots,p\}$ is any partition with $I\neq \emptyset,\{0,\ldots,p\} $ and $H_F^{a_{ij}}$ the $a_{ij}$th power of the Feynman propagator. We stress that the latter distribution is well defined since we are avoiding the diagonal. By (i) in Lemma \ref{lemma_Feynman_conormal_landing} we have that $\overline{\mathrm{WF}(H_F^{a_{ij}})}\vert_{\Delta_2(\Omega)} \subset N^*\Delta_2(\Omega)$, moreover by the induction hypothesis each $t_I$ satisfies
$$
    \overline{\mathrm{WF}t^0_I[M,h])} \vert_{\Delta_{|I|}(\Omega)} \subseteq N^{*}\Delta_{|I|}(\Omega).
$$
If we denote by $\mathrm{pr}_I$ the projection mapping $M^{p+1}\ni(x,\ldots,x_p)\mapsto x_I \in M^{|I|}$; then again, by \cite[Theorem 8.4.2]{hormanderI}, we have
\begin{equation}\label{eq_WF_distributional_causal_factorization}
    \mathrm{WF}(t^0[M,h])\vert_{U_I}\subset \Big\{ \mathrm{pr}_I^*\big( \mathrm{WF}(t^0_I[M,h]\big) + \mathrm{pr}_{I^c}^*\big( \mathrm{WF}(t^0_{I^c}[M,h]\big)+ \sum_{i\in I,j\in I^c} \mathrm{pr}_{ij}^*\big(\mathrm{WF}(H_F^{a_{ij}}) \big)\Big\}.
\end{equation}
Taking the closure of \eqref{eq_WF_distributional_causal_factorization} in ${T}^*U_I$ for each $I$, yields
$$
    \overline{\mathrm{WF}(t^0[M,h])}\vert_{\Delta_{p+1}(\Omega)}\subset N^*\Delta_{p+1}(\Omega).
$$
Consequently $r^0_N[M,h]$ satisfies the conormal landing condition and together with Theorem \ref{thm_3_w_homogeneous_extension_2} we obtain 
$$
    {\mathrm{WF}(r_N[M,h])}\vert_{\Delta_{p+1}(\Omega)}\subseteq N^{*}\Delta_{p+1}(\Omega).
$$
\textbf{Step 3.} Lastly, we are left with showing that the parameterized extensions $r_N[M,h_s]$, $\tau_k[M,h_s]$ described in the first step do satisfy the parameterized microlocal spectrum condition $(iv)$ in Lemma \ref{lemma_3_Wick_expansion}. We claim even more: if $h_s$ is a compactly supported variation with $s\in \mathbb{R}^d$, then
$$
\begin{aligned}
    &\overline{\mathrm{WF}(\tau_k^0[M,h_s])}\cap {T}^*(\mathbb{R}^d\times\Omega^{p+1})\vert_{\mathbb{R}^d\times\Delta_{p+1}(\Omega)} \subset N^{*}(\mathbb{R}^d\times\Delta_{p+1}(\Omega)), \\
    &\overline{\mathrm{WF}(r_N^0[M,h_s])}\cap {T}^*(\mathbb{R}^d\times\Omega^{p+1})\vert_{\mathbb{R}^d\times\Delta_{p+1}(\Omega)} \subset N^{*}(\mathbb{R}^d\times\Delta_{p+1}(\Omega)),.
\end{aligned}
$$
To show this claim we can repeat the previous argument noting that:
\begin{itemize}
    \item $\tau^0_k[M,h_s](x,\cdot)=\sum C^{a_1,\ldots,a_r}[h_s](x)e_x(s)^* u^0_{a_1,\ldots,a_r}(\cdot)$; therefore the dependence on $s\in \mathbb{R}^d$ is localized in the coefficient $C[h_s]$ (which is a scalar constructed out of tensors built out of $j^rh_s$ for some order $r\in\mathbb{N}$) and in the mapping $e_x(s): V\subset T_x \Omega^p \to \Omega_x\times \cdots \times \Omega_x : (\dot{x}_1,\ldots,\dot{x}_p)\mapsto (e_x(s)(\dot{x}_1),\ldots,e_x(s)(\dot{x}_p))$. The combination induces a diffeomorphism $\mathbb R^d\times \Omega \times \Omega^p \to \mathbb R^d\times \Omega \times V^p$, therefore studying the pullback of the wave front set of $\mathrm{WF}(u^0)$ as done for \eqref{eq_tau^0_conormal_landing}, yields
\begin{equation}\label{eq_tau^0_s_conormal_landing}
\begin{aligned}
    \overline{\mathrm{WF}(\tau_k^0[M,h_s])}\vert_{\mathbb{R}^d\times \Delta_{p+1}(\Omega)} \subseteq \Big\{ & \Big(s, x,\ldots,x;0,- \sum_{i=1}^p\xi_i,\xi_1,\ldots,\xi_p\Big)
    \in \dot{T}^*\Delta_{p+1}(\Omega)\\
    &: (x_1,\ldots,x_p;\xi_1,\ldots,\xi_p)\in \dot{T}^*\Omega^p\Big\}.    
\end{aligned}
\end{equation}
    \item Similarly, to establish that $r_N^0[M,h_s]$ satisfies the conormal landing condition, we apply the same induction argument used above, noting that the basic case $p=0$, is handled by Wick powers, and that for generic $p$ the causal factorization 
    \begin{equation}\label{eq_variational_distributional_causal_factorization}
      t[M,h_s](x,x_1,\ldots,x_p)=\sum_{I} t^0_I[M,h_s](x_I)t^0[M,h_s]_{I^c}(x_{I^c})\prod_{i\in I,j\in I^c} H_F[M,h_s]^{a_{ij}}(x_i,x_j),
    \end{equation}
    still holds. Combining the inductive hypothesis $\overline{\mathrm{WF}t^0_I[M,h_s])} \vert_{\Delta_{|I|}(\Omega)} \subseteq N^{*}\big(\mathbb{R}^d\times \Delta_{|I|}(\Omega)\big)$ and (ii) in Lemma \ref{lemma_Feynman_conormal_landing}, we can apply Theorem \ref{thm_3_w_homogeneous_extension_2} and obtain ${\mathrm{WF}(r_N[M,h_s])}\vert_{\mathbb{R}^d\times\Delta_{p+1}(\Omega)} \subseteq N^{*}\big(\mathbb{R}^d\times\Delta_{p+1}(\Omega)\big)$.
\end{itemize}
This concludes the proof for the existence.
\end{proof}

\subsection{Uniqueness of time ordered products}\label{section_TO_uniqueness}

To conclude this section we establish the uniqueness result for time ordered products. In the proof of the existence of such products we constructed uniquely the product up to the diagonal, therefore different prescriptions for, say the $p$th order product, ought to differ only in $\Delta_p(M)$; that is, they differ by a renormalization choice. Theorem \ref{thm_2_uniqueness_TO} and Corollary \ref{coro_2_TO_uniqueness} strictly characterize such choices.

\begin{theorem}\label{thm_2_uniqueness_TO}
Let $\mathfrak{A}_{\mu loc}(M,h)\subset \mathfrak{A}_{\mu c}(M,h)$ the algebra of microlocal functionals in the background geometry $(M,h)$. 
\begin{enumerate}
% \leavevmode
% \makeatletter
% \@nobreaktrue
% \makeatother
\item[$(a)$] Given two prescriptions $\{T_p\}_{p \in \mathbb{N}}$, $\{\widetilde{T}_p\}_{p \in \mathbb{N}}$ for time ordered products, there exists a family of natural transformations $\{Z_p:\otimes^p\mathfrak{A}_{ \mu loc}\Rightarrow \mathfrak{A}_{\mu loc}  \}_{p \in \mathbb{N}}$ inducing linear mappings $Z_p[M,h]:\otimes^p\mathfrak{A}_{\mu loc}(M,h) \rightarrow \mathfrak{A}_{\mu loc}(M,h) $, uniquely defined by 
\begin{equation}\label{eq_def_Z}
\begin{aligned}
	  \left(\frac{i}{\hbar}\right)^{p}\widetilde{T}_p[M,h](A_1,\ldots, A_p) = \sum_{\substack{\pi \in \mathcal{P}(\{1,\ldots,p\})}}\left(\frac{i}{\hbar}\right)^{|\pi|}T_{|\pi|}[M,h]\left(\otimes_{I \in \pi} Z_{|I|}[M,h](\otimes_{i\in I}A_i) \right)
\end{aligned}
\end{equation}
where $\pi$ is a partition of the set $\{1,\ldots,p\}$ into $|\pi|$ smaller subsets $I_1,\ldots, I_{|\pi|}$. Satisfying the following conditions:
\begin{itemize}
\item[$(i)$]\textbf{Covariance}. Let $\chi:(M,h)\rightarrow (M',h')$ be a causality preserving isometric embedding, then for each $A_1,\ldots,A_p \in \mathfrak{A}_{\mu loc}(M,h)$, we have 
$$
	Z_p[M,\chi^{*}h'](A_1,\ldots,A_p) = Z_p[M',h'](\mathfrak{A}_{\chi}(A_1),\ldots,\mathfrak{A}_{\chi}(A_p));
$$
\item[$(ii)$] \textbf{Support properties}. Given $A_i \in \mathfrak{A}_{\mu loc}(M,h)$ with $i=1,\ldots,p$, we have
$$
	\mathrm{supp}\left(Z_p[M,h](A_1,\ldots,A_p)\right)=\bigcap_{j=1}^p \mathrm{supp}(A_j);
$$ 
\item[$(iii)$] \textbf{Symmetry}. The mapping $Z_p$ is symmetric under permutation of its arguments, that is, given any permutation $\sigma$ of the set $\{1,\ldots,p\}$,
$$
	Z_p[M,h]\big(A_{\sigma(1)},\ldots,A_{\sigma(p)}\big)=Z_p[M,h](A_1,\ldots,A_p);
$$
\item[$(iv)$] \textbf{Initial value}. $Z_0[M,h]=0$ for all background geometries $(M,h)$;
\item[$(v)$] \textbf{Scaling}. If each $A_j\in \mathfrak{A}_{\mu loc}(M,h)$ scales almost homogeneously with degree $\rho_j$ under the rescaling of parameters $h\to S_{\lambda}h$, then $Z_p[M,h](A_1,\ldots,A_p)\in \mathfrak{A}_{\mu loc}(M,h)$ scales almost homogeneously with scaling degree $\sum_{j=1}^p \rho_j$;
\item[$(vi)$] \textbf{Field independence}. Given $A_1,\ldots,A_p \in \mathfrak{A}_{\mu loc}(M,h)$, we have
$$
	dZ_p[M,h](A_1,\ldots,A_p)[\varphi]=\sum_{j=1}^p Z_p[M,h](A_1,\ldots,dA_j[\varphi],\ldots,A_p);
$$
\item[$(vii)$] \textbf{Unitarity}. If $A_1,\ldots,A_p \in \mathfrak{A}_{\mu loc}(M,h)$, then
$$
	Z_p[M,h](A_1,\ldots,A_p)^{*}=Z_p[M,h](A_1^{*},\ldots,A_p^{*}).
$$
\item[$(viii)$] \textbf{Action Ward identity}. Let $A_1,\ldots,A_p\in \mathfrak{A}_{\mu loc}(M,h)$ such that at least one, say $A_1$, has the form $A_1=\alpha_H^{-1}(F_1)$ where 
$$
	F_1(\varphi)=\int_M j^{r}\varphi^{*}(d \theta)
$$
for some $n-1$-form $\theta \in \Omega_{n-1}(J^rM)$, then
$$
	Z_p(A_1,\ldots,A_p)=0;
$$
\item[$(ix)$] \textbf{$\varphi$-Locality} Let $A_1,\ldots,A_p \in \mathfrak{A}_{\mu loc}(M,h)$, denote by $A_i^{[k_1]}$ the Wick ordering of the $k$th order truncated Taylor series of $F_i=\alpha_H(A_i)$ in $0\in C^{\infty}(M)$, then 
$$
	Z_p[M,h]\big(A_1,\ldots,A_p\big)= Z_p[M,h]\left(A_1^{[k_1]},\ldots,A^{[k_p]}_p\right)+O\big(\hbar^{\lfloor p(k+1)/2\rfloor+1}\big);
$$
\item[$(x)$] \textbf{Microlocal spectrum condition} ($\mu$SC). Given any $A_1,\ldots,A_p\in \mathfrak{A}_{\mu loc}(M,h)$, consider the formal integral kernel associated to the functional derivative $d^k\Big(\alpha_H\big(Z_p[M,h](A_1,\ldots,A_p)\big)\Big)[\varphi]$, %by field independence and Faà di Bruno's formula we can write this quantity as a sum of terms of the form
% $$
% 	Z^H_p\big(d^{k_i}\alpha_H(A_1)[\varphi],\ldots,d^{k_p}\alpha_H(A_p)[\varphi]\big)(x_1,\ldots,x_k);
% $$
we require that for every $\varphi\in C^{\infty}(M)$,
\begin{equation}\label{eq_3_mulocal_spectrum_condition_Z}
	\mathrm{WF}\bigg(d^k\Big(\alpha_H\big(Z_p[M,h](A_1,\ldots,A_p)\big)\Big)[\varphi]\bigg) \subset N^{*}\Delta_{k}(M)
\end{equation}
\item[$(xi)$] \textbf{Parameterized microlocal spectrum condition} (P$\mu$SC). If $(M,h)$ is a background geometry and $ \mathbb{R}^d\ni s \mapsto h_s\in \Gamma^{\infty}(HM)$ any compactly supported variation of $h$; we require that for every $\varphi\in C^{\infty}(M)$ and every $A_1,\ldots,A_p \in \mathfrak{A}_{\mu loc}(M,h)$,
$$
	\begin{aligned}
		& \mathrm{WF}\bigg(d^k\Big(\alpha_H\big(Z_p[M,h](A_1,\ldots,A_p)\big)\Big)[\varphi]\bigg)  \\ &\subset \Big\{ (s,x_1,\ldots,x_k,\zeta,\xi_1,\ldots, \xi_k ): \ (x_1,\ldots,x_k,\xi_1,\ldots, \xi_k )\in N^{*}\Delta_{k}(M) \Big\}.
	\end{aligned}
$$
\end{itemize}
\item[$(b)$] If $\{T_p \}_{p \in\mathbb{N}}$ is a family satisfying all requirements of Definition \ref{def_3_TO_products}, and $\{Z_p\}_{p \in \mathbb{N}}$ another family satisfying properties $(i)$-$(xi)$ in $(a)$ above; given any $A_1,\ldots,A_p \in \mathfrak{A}_{\mu loc}(M,h)$, setting
$$
	\begin{aligned}
		\left(\frac{i}{\hbar}\right)^{p}\widetilde{T}_p[M,h](A_1,\ldots, A_p) \doteq \sum_{\substack{\pi \in \mathcal{P}(\{1,\ldots,p\})}}\left(\frac{i}{\hbar}\right)^{|\pi|}T_{|\pi|}[M,h]\left(\otimes_{I \in \pi} Z_{|I|}[M,h](\otimes_{i\in I}A_i) \right),
	\end{aligned}
$$
yields another prescription for a time ordered product according to Definition \ref{def_3_TO_products}.
\end{enumerate}
\end{theorem}
 
\begin{proof}
For notational sake we shall omit the $[M,h]$ form $T_p[M,h], Z_p[M,h]$. We start with $(a)$. We define the family $\{Z_p\}$ by induction over $p$. First we note that $(iv)$ in Definition \ref{def_3_TO_products} is consistent with 
$$
	\widetilde{T}_0(A)= 1_{\mathfrak{A}_{\mu c}}={T}_0(A).
$$
Next, suppose we have constructed $\{Z_k\}_{k<p}$, from \eqref{eq_def_Z} define
$$
	\begin{aligned}
		Z_p(A_1\ldots, A_p)\doteq &	\left(\frac{i}{\hbar}\right)^{p-1}\widetilde{T}_p(A_1,\ldots, A_p)-	\left(\frac{i}{\hbar}\right)^{p-1}T_p(A_1,\ldots, A_p)\\
		 & - \sum_{\substack{\pi \in \mathcal{P}(\{1,\ldots,p\})\\ |\pi| \neq n}}\left(\frac{i}{\hbar}\right)^{|\pi|-1}T_{|\pi|}\left(\bigotimes_{I \in \pi} Z_{|I|}\Big(\bigotimes_{i\in I}A_i\Big) \right).
	\end{aligned}
$$
The right hand side can always be constructed since it involves the mappings $\{Z_k\}_{k<p}$ and the two time ordered products. We stress that for $p=1$, \eqref{eq_def_Z} yields $Z_1=id_{\mathfrak{A}_{\mu loc}}$ which is consistent since $T_1(A)=A=\widetilde{T}_1(A)$ for any $A \in \mathfrak{A}_{\mu loc}(M,h)$. The element $Z_p$ constructed above will be local and covariant by $(i)$ in \ref{def_3_TO_products} and the inductive hypothesis applied to $\{Z_k\}_{k<p}$. This argument also shows that properties $(iii)$, $(v)$, $(vi)$, $(viii)$, $(ix)$ in $(a)$ Theorem \ref{thm_2_uniqueness_TO} hold. To show $(ii)$, suppose that at least a pair of elements $A_1$, $A_2$ have mutually disjoint supports, \textit{i.e.} $S_1=\mathrm{supp}(A_1)\cap\mathrm{supp}(A_2)=S_2=\emptyset$. Then there is $I\subset \{1,\ldots,p\}$ for which any support $S_i=\mathrm{supp}(A_i)$, $i\in I$ has $S_i \notin J^{-}(S_j)$ for all $j\in I^c$; by assumption, at the very least $\{1\}\subset I$, $\{2\}\subset I^c$. Combining the causal factorization axiom of time ordered products with the inductive assumptions we get
$$
\begin{aligned}
    &\left(\frac{i}{\hbar}\right)^{1-p} Z_p(A_i,\ldots, A_p) \\
    &= \widetilde{T}_{|I|}\Big(\bigotimes_{i\in I}A_i\Big)\star \widetilde{T}_{|I^c|}\Big(\bigotimes_{j\in I^c}A_j\Big)-{T}_{|I|}\Big(\bigotimes_{i\in I}A_i\Big)\star{T}_{|I^c|}\Big(\bigotimes_{j\in I^c}A_j\Big)\\ 
    & \quad- \sum_{\substack{\pi_1 \in \mathcal{P}(I), \space\ \pi_2\in \mathcal{P}(I_c)\\ |\pi_1| \neq 1 \&  |\pi_2| \neq 1}}\left(\frac{i}{\hbar}\right)^{|\pi_1|+|\pi_2|-p}T_{|\pi_1|}\left(\bigotimes_{I_1 \in \pi_1} Z_{|I_1|}\Big(\bigotimes_{i\in I_1}A_i\Big) \right)\star T_{|\pi_2|}\left(\bigotimes_{I_2 \in \pi_2} Z_{|I_2|}\Big(\bigotimes_{j\in I_2}A_j\Big) \right)\\
    &= \left( \sum_{\substack{\pi_1 \in \mathcal{P}(I)\\ }}\left(\frac{i}{\hbar}\right)^{|\pi_1|-|I|}T_{|\pi_1|}\left(\bigotimes_{I_1 \in \pi_1} Z_{|I_1|}\Big(\bigotimes_{i\in I_1}A_i\Big)\right) \right) \\
    & \quad\star \left( \sum_{\substack{\pi_2 \in \mathcal{P}(I^c)\\ }}\left(\frac{i}{\hbar}\right)^{|\pi_2|-|I^c|}T_{|\pi_2|}\left(\bigotimes_{I_2 \in \pi_2} Z_{|I_2|}\Big(\bigotimes_{j\in I_2}A_j\Big)\right) \right)-{T}_{|I|}\Big(\bigotimes_{i\in I}A_i\Big)\star {T}_{|I^c|}\Big(\bigotimes_{j\in I^c}A_j\Big)\\ 
    & \quad- \sum_{\substack{\pi_1 \in \mathcal{P}(I), \space\ \pi_2\in \mathcal{P}(I_c)\\ |\pi_1| \neq 1 \&  |\pi_2| \neq 1}}\left(\frac{i}{\hbar}\right)^{|\pi_1|+|\pi_2|-p}T_{|\pi_1|}\left(\bigotimes_{I_1 \in \pi_1} Z_{|I_1|}\Big(\bigotimes_{i\in I_1}A_i\Big) \right)\star T_{|\pi_2|}\left(\bigotimes_{I_2 \in \pi_2} Z_{|I_2|}\Big(\bigotimes_{j\in I_2}A_j\Big) \right).\\
\end{aligned}
$$
The above quantity is identically zero since the last two terms cancel the first one. The argument for $(x)$ and $(xi)$ is similar so we just give the latter: consider a compactly supported variation $h_s$ of the background geometry $h$,
% To show $(x)$ we simply note that $\mathrm{supp}({Z_p})\subset \Delta_p(M)$ implies that the functional $\alpha_H \circ Z_p \circ \otimes^p \alpha_H^{-1}$ is a local, to estimate its wave front we consider 
% $$
% 	d^k\big(\alpha_H\big(Z_p(A_1,\ldots,A_p)\big)\propto \sum_{k_1+\cdots+k_p=k}\alpha_H\big(Z_p(d^{k_i}A_1,\ldots,d^{k_p}A_p)\big)
% $$
%-------------------Precise argument for the wave front set-------------------
% then a number of $k_j$'s could be zero, in which case we note that the wave front of with, by \cite{hormanderI} Theorem 5.2.3 we obtain that, then $WF\Big(d^k\big(\alpha_H\big(Z_p(A_1,\ldots,A_p)\big)\big)[\varphi]\Big)\subset N^{*}\Delta_{\beta(k,p)}(M)$, if, on the other, $p\leq k$, then by field independence. Similarly we can show $(xi)$. 
since each $A_i$ is local the wave front set of $d^k\alpha_{H_s}(A_i)$ will be in $N^{*}(\mathbb{R}^d\times \Delta_k(M))$. Using \eqref{eq_def_Z},
$$
\begin{aligned}
	&\left(\frac{i}{\hbar}\right)^{1-p}d^k\big(\alpha_{H_s}\big(Z_p(A_1,\ldots,A_p)\big)\big)[\varphi](x_1,\ldots,x_k)=d^k\big(\alpha_{H_s}\big(\widetilde{T}_p(A_1,\ldots,A_p)\big)\big)[\varphi](x_1,\ldots,x_k)\\ 
	& \quad -\sum_{\substack{\pi \in \mathcal{P}(\{1,\ldots,p\})}}\bigg(\frac{i}{\hbar}\bigg)^{|\pi|-p}d^k\bigg( \alpha_{H_s} \bigg( T_{|\pi|}\bigg(\bigotimes_{I \in \pi} Z_{|I|}\bigg(\bigotimes_{i\in I}A_i\bigg) \bigg)\bigg)\bigg)[\varphi](x_1,\ldots,x_k)\\
	%&=d^k\big(\widetilde{T}^{H_s}_p(A_1,\ldots,A_p)\big)[\varphi](x_1,\ldots,x_k)\\
	%&-  \sum_{\substack{\pi\in \mathcal{P} \{ (1,\ldots,p) \} \\ \sigma\in  \mathcal{P}\{(1,\ldots,k)\}\\}}\bigg(\frac{i}{\hbar}\bigg)^{|\pi|-p} d^{|\sigT|}T_{|\pi|}^{H_s}\Bigg(\bigotimes_{j=1}^{|\pi|} d^{|J_j|}Z^{H_s}_{|I_j|}\bigg(\bigotimes_{i\in I}\alpha_{H_s}(A_i) \bigg)[\varphi](x_1,\ldots,x_{|J_j|}) \Bigg)
\end{aligned}
$$
By Faà di Bruno's formula and the field independence of time ordered products (property $(vi)$ in Definition \ref{def_3_TO_products}), we can write the terms involving $Z_{|I|}$'s on the right hand side of the above equation as $$
    d^{|\sigma|}T_{|\pi|}^{H_s}\Bigg(\bigotimes_{j=1}^{|\pi|} d^{|J_j|}Z^{H_s}_{|I_j|}\bigg(\bigotimes_{i\in I}\alpha_{H_s}\big(A_i\big) \bigg)[\varphi]\Bigg),
$$
where $I_1,\ldots,I_{|\pi|}\in \pi $ are the elements of the partition of $\{1,\ldots,p\}$, $\sum_i |J_i|=k$ and $T_p^H=\alpha_H\circ T_p \circ \otimes^p\alpha_H^{-1}$, $Z_p^H\doteq \alpha_H\circ Z_p \circ \otimes^p\alpha_H^{-1}$. Since the right hand side of the above equality is zero whenever $\mathrm{supp}(A_i)\cap \mathrm{supp}(A_j)=\emptyset$ for some $i,j \in \{1,\ldots,p\}$, we can estimate
$$
\begin{aligned}
	\mathrm{WF}\Big( d^k Z^{H_s}_p(\alpha_{H_s}(A_1),&\ldots,\alpha_{H_s}(A_p))[\varphi]\Big) \subset \mathrm{WF} \big(d^k\widetilde{T}^{H_s}_n(\alpha_{H_s}(A_1),\ldots,\alpha_{H_s}(A_p))\big)\Big\vert_{\mathbb{R}^d\times\Delta_k(M)} \\
	&\bigcup \ \mathrm{WF} \Bigg( d^{|\sigma|}T_{|\pi|}^{H_s}\Bigg(\bigotimes_{j=1}^{|\pi|} d^{|J_j|}Z^{H_s}_{|I_j|}\bigg(\bigotimes_{i\in I}\alpha_{H_s}\big(A_i\big) \bigg)[\varphi]\Bigg)\Bigg) \Bigg\vert_{\mathbb{R}^d\times\Delta_k(M)} 
\end{aligned}
$$
That the first gives the right wave front set when restricted comes readily from property $(xi)$ in Definition \ref{def_3_TO_products}, the other terms are composition of distributions
$$
	d^{|\sigma|}T_{|\pi|}^{H_s}\Bigg(\bigotimes_{j=1}^{|\pi|} d^{|J_j|}Z^{H_s}_{|I_j|}\bigg(\bigotimes_{i\in I}\alpha_{H_s}\big(A_i\big) \bigg)[\varphi]\Bigg)
$$
which are well defined for
$$
	\{(s,x_i,\zeta,\xi_i)\} \in \mathrm{WF}_{\mathrm{pr}_i}\Big(d^{|\sigma|}T_{|\pi|}^{H_s}\Big(G_1,\ldots,G_{|\pi|}\Big)[\varphi]\Big)= \emptyset, \ \ \forall \ i=1,\ldots,|\sigma|.
$$
Arguing using Theorem 8.2.14 in \cite{hormanderI} we can estimate
$$
\begin{aligned}
	& \mathrm{WF} \bigg( d^{|\sigma|}T_{|\pi|}^{H_s}\bigg(\bigotimes_{j=1}^{|\pi|} d^{|J_j|}Z^{H_s}_{|I_j|}\bigg(\bigotimes_{i\in I}\alpha_{H_s}\big(A_i\big) \bigg)[\varphi]\bigg)\bigg) \bigg\vert_{\mathbb{R}^d\times \Delta_{k}(M)}\\
	& \quad \subset   \{ (s,x_1,\ldots,x_{k},\zeta,\xi_1,\ldots, \xi_{k} ): \ (x_1,\ldots,x_{k},\xi_1,\ldots, \xi_{k} )\in N^{*}\Delta_{k}(M)\backslash\{0\} \}.
\end{aligned}
$$
Next we show $ii)$. Let $A_j\in \mathfrak{A}_{\mu loc}(M,h)$ for $j=1,\ldots,p$, define
\begin{equation}\label{eq_def_tilde_T}
	  \widetilde{T}_p(A_1,\ldots,A_p)=  \sum_{\substack{\pi \in \mathcal{P}(\{1,\ldots,p\})}}\bigg(\frac{i}{\hbar}\bigg)^{|\pi|-p}T_{|\pi|}\left(\otimes_{I \in \pi} Z_{|I|}(\otimes_{i\in I}A_i) \right),
\end{equation}
using that by hypothesis properties $1)$, $(iii)$, $(iv)$, $(v)$, $(vi)$, $(vii)$, $(viii)$, $(ix)$ hold in the right hand side of \eqref{eq_def_tilde_T}, then they trivially hold on the left hand side as well. For $(x)$, note that the term appearing to the right hand side of \eqref{eq_def_tilde_T} is of the form
$$
    d^{|\sigma|}T_{|\pi|}^{H}\bigg(\bigotimes_{j=1}^{|\pi|}d^{|J_{jl}|} Z^{H}_{|I_j|}\bigg(\bigotimes_{i\in I}\alpha_{H}\big(A_i\big) \bigg)[\varphi]\bigg)
$$
where $\sum_{j=1}^{|\pi|}\sum_{l=1}^{|\sigma|}|J_{jl}|=k$. By Theorem 8.2.14 in \cite{hormanderI} we have 
$$
\begin{aligned}
	& \mathrm{WF} \Big( d^kT^H_{|\pi|}\Big(\otimes_{I \in \pi} Z^H_{|I|}\big(\otimes_{i\in I}\alpha_H (A_i)\big) \Big)[\varphi]\Big) \subset \Big\{ (x_1,\ldots,x_k, \xi_1,\ldots,\xi_k) \in T^{*}M^k\backslash 0 : \\ & \qquad \exists (y_1,\ldots,y_{|\sigma|},\eta_1,\ldots,\eta_{|\sigma|})\in \mathrm{WF}(d^{|\sigma|}T^H_{|\pi|}),  \ (x_{J_{il}},y_l,\xi_{J_{jl}},\eta_l)\in \mathrm{WF}\big(Z^{H}_{|I_j|}\big(\bigotimes_{i\in I}\alpha_{H}\big(A_i\big) \big)[\varphi]\big)\Big\}\\
	&\qquad\subset \Gamma^T_k(M,h).
\end{aligned}
$$
A similar reasoning shows that $\widetilde{T}^H$ does indeed satisfy also the parameterized microlocal spectrum condition as well. Finally we show that causality holds. Suppose $ A_1,\ldots, A_p \in \mathfrak{A}_{\mu loc}(M,h)$ and $I\subset \{1,\ldots,p\}$ with the property that each $i\in I$ has $\mathrm{supp}(A_i) \notin J^{-}\big(\mathrm{supp}(A_j)\big) \ \forall j \in I^c$, then we evaluate 
$$
\begin{aligned}
	  &\widetilde{T}_p(A_1,\ldots, A_p)=\sum_{\substack{\pi \in \mathcal{P}(\{1,\ldots,p\})}}\bigg(\frac{i}{\hbar}\bigg)^{|\pi|-p}T_{|\pi|}\bigg(\bigotimes_{I \in \pi} Z_{|I|}\bigg(\bigotimes_{i\in I}A_i\bigg) \bigg)\\
	  &= \sum_{\substack{ \pi=\pi'\cup\pi''\\ \pi'\in \mathcal{P}(I), \pi'' \in \mathcal{P}(I^c)}}\bigg(\frac{i}{\hbar}\bigg)^{|\pi|-p}T_{|\pi|}\bigg(\bigotimes_{I' \in \pi'} Z_{|I'|}\bigg(\bigotimes_{i\in I'}A_i\bigg) \bigotimes_{I'' \in \pi''} Z_{|I''|}\bigg(\bigotimes_{j\in I''}A_j\bigg)\bigg)\\
	  &=\sum_{\substack{ \pi'\in \mathcal{P}(I)\\ \pi'' \in \mathcal{P}(I^c) }}\bigg(\frac{i}{\hbar}\bigg)^{|\pi'|-|I|}T_{|\pi'|}\bigg(\bigotimes_{I' \in \pi'} Z_{|I'|}\bigg(\bigotimes_{i\in I'}A_i\bigg)\bigg)\star  \bigg(\frac{i}{\hbar}\bigg)^{|\pi''|-|I^c|} T_{|\pi''|} \bigg(\bigotimes_{I'' \in \pi''} Z_{|I''|}\bigg(\bigotimes_{j\in I''}A_j\bigg)\bigg) \\
	  &= \widetilde{T}_{|I|}\bigg(\bigotimes_{i\in I}A_i\bigg)\star \widetilde{T}_{|I^c|}\bigg(\bigotimes_{j\in I^c}A_j\bigg)
\end{aligned}
$$
where, in the second equality, we used that each $Z_{k}$ is supported along the diagonal effectively canceling all terms with mixed indices from $I$ and $I^c$ and in the last we applied the causality axiom in \ref{def_3_TO_products}.
\end{proof}

\begin{corollary}\label{coro_2_TO_uniqueness}
In the hypothesis of $(a)$ Theorem \ref{thm_2_uniqueness_TO}, combining $\varphi$-locality with field independence, for a fixed perturbative order $\hbar$, the mapping $Z_p$ is completely determined by the action on Wick powers. Moreover 
\begin{equation}\label{eq_Z_induction}
	\begin{aligned}
		& Z_p[M,h]\big(\phi^{k_1}_{(M,h,H)},\ldots,\phi^{k_p}_{(M,h,H)}\big)(x_1,\ldots,x_p) \\
		&= \sum_{J\leq K}\binom{K}{J}C^{K-J}_{(M,h)}(x_1)\delta(x_1,\ldots,x_p)\phi^{J}_{(M,h,H)}(x_1,\ldots,x_p),
	\end{aligned}
\end{equation}
where $K=(k_1,\ldots,k_p)$, $J=(j_1,\ldots,j_p)$ are multi-indices, $\phi^{J}_{(M,h,H)}(x_1,\ldots,x_p)$ is the formal integral kernel of the $\mathfrak A_{\mu c}$-valued distribution associated to the multilocal functional $\phi^{j_1}_{(M,h)}(x_1)\cdots \phi^{j_p}_{(M,h)}(x_p)$. The coefficients $C^{K-J}_{(M,h)}(x)$ are polynomials in scalars constructed out all possible tensor-like object of $j^rh\in \Gamma^{\infty}(J^rHM)$ for some fixed order $r$. Finally, each $C^{K-J}_{(M,h)}$ scales homogeneously with degree $\sum_{i=1}^p(k_i-j_i)\frac{n-2}{2}$ under physical scaling defined in \eqref{eq_2_physical_scaling}.
\end{corollary}

\begin{proof}
By $\varphi$-locality, axiom $(ix)$ in Theorem \eqref{thm_2_uniqueness_TO}, arguing as in the proof of Lemma \ref{lemma_3_Wick_expansion} evaluating $Z_p(A_1,\ldots,A_p)$, up to some fixed order in $\hbar$, depends only on the evaluation of elements of the form $Z_p^{K}\big(\phi^{k_1}_{(M,h,H)},\ldots,\phi^{k_p}_{(M,h,H)}\big)$, where $K=(k_1,\ldots,k_p)$ is a multi-index. Again we can forget about algebra elements carrying derivatives using the Action-Ward identity (axiom $(viii)$ in Theorem \ref{thm_2_uniqueness_TO}).
% The integral kernel associated to $Z_p^{\ {K}}\big(\phi^{k_1}_{(M,h,H)},\ldots,\phi^{k_p}_{(M,h,H)}\big)$ can be represented as
% \begin{equation}\label{eq_Z_proof}
%  	Z_n^{K}(\phi^{k_1}_{(M,h,H)},\ldots,\phi^{k_p}_{(M,h,H)})(x_1,\ldots,x_p)= \sum_{J\leq K}\binom{K}{J}C^{K-J}_{(M,h)}(x_1)\delta(x_1,\ldots,x_p)\phi^{j_1}_{(M,h,H)}(x_1)\cdots \phi^{j_p}_{(M,h,H)}(x_p)
% \end{equation}
% where the coefficients $C^{K-J}_{(M,h)}(x)$ are covariant scalars build out of tensor elements of $j^rh$, $h \in \Gamma^{\infty}(M\leftarrow HM)$, scaling almost homogeneously with degree $\frac{n-2}{2}(\sum_{i=1}^pk_i-j_i)$. 
We stress that by $\phi^{j_i}_{(M,h,H)}(x_i)$, in \eqref{eq_Z_induction}, we mean the integral kernel associated to the algebra value distribution $\phi^{j_i}_{(M,h,H)}:\mathcal{D}(M)\to \mathfrak{A}_{\mu c}(M,h)$. The proof goes through double induction over the multi-index $K$ and the number $p$ of variables. First we assume that $Z^{J}_{q}$ have been defined for all $J\leq K$ and $q\leq p$, then \eqref{eq_def_Z} defines $Z^{K}_{p+1}$ in terms of the previous coefficients and the time ordered products $T$, $\widetilde{T}$. Fix a value of $p$ and increase the multi-index $K$ by $1$ in one of its components. Let 
\begin{equation}\label{eq_psi_induction}
\begin{aligned}
	\Psi_p^{{0}}(x_1,\ldots,x_p)\doteq & Z_p^{K}\big(\phi^{k_1}_{(M,h,H)},\ldots,\phi^{k_p}_{(M,h,H)}\big)(x_1,\ldots,x_p) \\
	 &- \sum_{J<K}\binom{K}{J}C^{K-J}_{(M,h)}(x_1)\delta(x_1,\ldots,x_p) \phi^{J}_{(M,h,H)}(x_1,\ldots, x_p),
\end{aligned}
\end{equation}
where $J< K$ if $j_{i_0}<k_{i_0}$ for at least one $i_0 \in \{1,\ldots,p\}$ and $j_i\leq k_i$ for all $i \in \{1,\ldots,p\} $, $i\neq i_0$. Consider therefore 
\begin{align*}
	&d\Psi_p^{{0}}[\varphi](x_1,\ldots,x_p,y) = dZ_p^{K}\big(\phi^{k_1}_{(M,h,H)},\ldots,\phi^{k_p}_{(M,h,H)}\big)[\varphi](x_1,\ldots,x_p,y)\\
	& \ \ -\sum_{J<K}\binom{K}{J}C^{K-J}_{(M,h)}(x_1)\delta(x_1,\ldots,x_p)d \Big(\phi^{J}_{(M,h,H)}(x_1,\ldots,x_p)\Big)[\varphi](y)\\
	&=\sum_{i=1}^p\bigg(k_i Z_p^{K-1_i}(\phi^{k_1}_{(M,h,H)},\ldots,\phi^{k_i-1}_{(M,h,H)},\ldots,\phi^{k_p}_{(M,h,H)})(x_1,\ldots,x_p,y) \\
	& \ \ - \sum_{J<K}j_i\binom{K}{J}C^{K-J}_{(M,h)}(x_1)\delta(x_1,\ldots,x_p)\phi^{J-1_i}_{(M,h,H)}(x_1,\ldots,x_p) \bigg)\delta(x_i,y).
\end{align*}
Taking into account \eqref{eq_psi_induction} and using the binomial identity $j\binom{k}{j}= k\binom{k-1}{j-1}$ for each index $i=1,\ldots,p$, we get that 
\begin{align*}
	&d\Psi_p^{{0}}[\varphi](x_1,\ldots,x_p,y)  \\
	&=\sum_{i=1}^p\Bigg(\sum_{J\leq K-1_i}k_i\binom{K-1_i}{J}C^{K-J-1_i}_{(M,h)}(x_1)-\sum_{J<K}j_i\binom{K}{J}C^{K-J}_{(M,h)}(x_1) \Bigg)\phi^{J-1_i}_{(M,h,H)}(x_1,\ldots,x_p)\delta(x_1,\ldots,x_p,y).\\
	&=0.
\end{align*}
Thus implying that $\Psi_p^{{0}}(x_1,\ldots,x_p)$ is a $c$-number. In particular given the support property of $Z_p$ we necessarily have (see Theorem 5.2.3 in \cite{hormanderI}) $\Psi_p^{{0}}[M,h](x_1,\ldots,x_p)=C_{K}[M,h](x_1)\delta(x_1,\ldots,x_p)1_{(M,h)}$. Confronting with \eqref{eq_psi_induction} $C_{K}[M,h_t](x_1)$ must scale homogeneously with degree $\frac{n-2}{2}(\sum_{i=1}^p k_i-j_i)$. Taking $f\in C^{\infty}(M^{p-1})$ with $f\vert_{\Delta_{p-1}(M)}\equiv 1$, we define
\begin{equation}\label{eq coeff}
	C_{K}[M,h](x)=\int_{M^{p-1}} f(x_2,\ldots,x_p)\Psi_p^{{0}}[M,h](x_1,\ldots,x_p)d\mu_g(x_1,\ldots,x_p).
\end{equation}
A standard propagation of singularity argument shows that the coefficient defined by \eqref{eq coeff} is smooth, that is $C_{K}[M,h]\in C^{\infty}(M)$. The mapping $\Gamma^{\infty}(HM) \ni h \mapsto C_{K}[M,h]\in C^{\infty}(M)$ is weakly regular since the right hand side of \eqref{eq_psi_induction} smeared with $f$ defines a weakly regular mapping $\Gamma^{\infty}(HM) \ni h \mapsto \Psi_p^0[h](\cdot,f) \in C^{\infty}(\mathbb{R}^d\times M)$ by axiom $(xi)$ in Theorem \ref{thm_2_uniqueness_TO}. Moreover, by locality and covariance (property $1)$ in Theorem \ref{thm_2_uniqueness_TO}), $\Psi_p^0[h](x,f) $ does depend at most on the germ of $h$ at $x$. We can now apply Theorem \ref{thm_2_Moretti_Kavhkine} to show that the coefficients $C_{K}[M,h](x)$ are polynomials in the scalars obtained from coefficients of $j^rh$, $h \in \Gamma^{\infty}(M\leftarrow HM)$, with finite jet order $r$ bounded globally and with homogeneous scaling having degree $\frac{n-2}{2}(\sum_{i=1}^p k_i-j_i)$.
\end{proof}

%   BACK MATTER
%   BIBLIOGRAPHY
\cleardoublepage
\addcontentsline{toc}{chapter}{Bibliography}
\bibliography{Report}

%   APPENDIX
% \cleardoublepage
% \appendix % to tell LaTeX that the following chapters are appendices
% \input{Include/Backmatter/Appendix/Appendix_A}

%   ACKNOWLEDGEMENTS
\cleardoublepage
\pagenumbering{gobble}
\thispagestyle{plain}			% Supress header
\section*{Acknowledgements}
I would like to thank my PhD supervisor Romeo Brunetti for guidance during the PhD and the choice of a project fitting my strengths and, more importantly, my tastes. Moreover, I would like to thank the whole Mathematical Physics group of Trento, and especially Nicolò Drago, for many discussions and clarifications regarding an uncountable amount of topics. I also express my gratitude to the referees for reading this thesis, especially Pedro Lauridsen Ribeiro, who also gave a clarifying explanation of the various topologies for microcausal functionals, suggested the structure of the proof of \Cref{lemma_1_regular_density} and offered many other valuable inputs regarding the classical project. Many thanks to Klaus Fredenhagen as well for inviting me to Hamburg, discussing the thesis projects at an earlier stage and giving many suggestions regarding the quantum project. Finally, this PhD experience would not have been such a lively and unforgettable experience without all the PhD students I have met in Trento, G\"ottingen and Leipzig.

\vspace{1.5cm}
\hfill
Andrea Moro, Trento, July 2023

\newpage				% Create empty back of side
\thispagestyle{empty}
\mbox{}

\end{document}